\documentclass[amssymb]{revtex4}

\usepackage{array}
\usepackage{scalerel}
\usepackage{tabularx}
\usepackage{multirow}
\usepackage{amsmath}
\usepackage{mathtools}
\usepackage{stmaryrd}
\usepackage{amsthm}
\usepackage{graphicx} 
\usepackage{amssymb}
\usepackage{mathrsfs}
\usepackage{subeqnarray}
\usepackage[all]{xy}
\usepackage{accents}
\usepackage{subfigure}

\usepackage[pdftex, bookmarks, colorlinks=true, plainpages = false, citecolor = red, urlcolor = black, filecolor = blue]{hyperref}
\usepackage[usenames,dvipsnames]{color}
\usepackage[shortlabels]{enumitem}
\setenumerate{listparindent=\parindent}

\usepackage[titletoc]{appendix}

\newcommand{\itemcolor}[1]{
  \renewcommand{\makelabel}[1]{\color{#1}\hfil ##1}}

\usepackage{stackengine}
\stackMath

\DeclareFontFamily{U}{mathb}{\hyphenchar\font45}
\DeclareFontShape{U}{mathb}{m}{n}{
      <5> <6> <7> <8> <9> <10> gen * mathb
      <10.95> mathb10 <12> <14.4> <17.28> <20.74> <24.88> mathb12
      }{}
\DeclareSymbolFont{mathb}{U}{mathb}{m}{n}
\DeclareMathSymbol{\curvearrowright}{3}{mathb}{'361}

\usepackage{calligra}
\usepackage{tikz-cd}
\usepackage{thm-restate}
\DeclareMathAlphabet{\mathcalligra}{T1}{calligra}{m}{n}
\DeclareFontShape{T1}{calligra}{m}{n}{<->s*[2.2]callig15}{}

\def\dlceil{\left\lceil\kern-4.75pt\left\lceil}
\def\drceil{\right\rceil\kern-4.75pt\right\rceil}

\usepackage{etoolbox}
\patchcmd{\appendices}{\quad}{. }{}{}

\newcommand{\Module}[2]{\smash{\superscr{#2 \;} \scalebox{1}{\rotatebox[origin=c]{-90}{$\circlearrowright$}} \; #1}}

\newcommand{\RModule}[2]{\smash{#1 \; \scalebox{1}{\rotatebox[origin=c]{90}{$\circlearrowleft$}} \, \superscr{\; #2}}}
\newcommand{\CModule}[2]{\smash{\subscr{#2 \;}\scalebox{1}{\rotatebox[origin=c]{-90}{$\circlearrowright$}} \; #1}}
\newcommand{\BIModule}[3]{\smash{\subscr{#3 \;}\superscr{#2 \;} \scalebox{1}{\rotatebox[origin=c]{-90}{$\circlearrowright$}} \; #1}}

\newcommand{\CRModule}[2]{\smash{#1 \; \scalebox{1}{\rotatebox[origin=c]{90}{$\circlearrowleft$}} \, \subscr{\; #2}}}

\newcommand{\alg}{\mathsf{A}}
\newcommand{\almostId}{i}

\newcommand{\algB}{\mathsf{A}}

\newcommand{\overbarStraight}[1]{\mkern 1.5mu\overline{\mkern-1.5mu#1\mkern-1.5mu}\mkern 1.5mu}
\newcommand{\overbarcal}[1]{\mkern 2.50mu\overline{\mkern-2.50mu#1\mkern-0.5mu}\mkern 0.5mu}
\newcommand{\overbarcalcal}[1]{\mkern 7.50mu\overline{\mkern-7.50mu#1\mkern-0.5mu}\mkern 0.5mu}
\newcommand{\overhat}[1]{\hstretch{1.75}{\hat{\hstretch{.57}{#1}}}}

\newcommand{\alphaBar}{\overbarcal{\alpha}}
\newcommand{\betaBar}{\overbarcal{\beta}}

\newcommand{\DeltaBar}{\overbarcal{\Delta}}

\newcommand{\EBar}{\overbarcal{E}}
\newcommand{\FBar}{\overbarcal{F}}
\newcommand{\KBar}{\overbarcal{K}}

\global\long\def\bZ{\mathbb{Z}}

\global\long\def\bZpos{\mathbb{Z}_{>0}}

\global\long\def\bZnn{\mathbb{Z}_{\geq 0}}

\global\long\def\bC{\mathbb{C}}
\global\long\def\tieOp{\wp}

\global\long\def\ii{\mathfrak{i}}

\newcommand{\rad}{\textnormal{rad}\,}
\newcommand{\Span}{\textnormal{span}\,}

\newcommand{\EndMod}[1]{\textnormal{End}_{\,#1}}

\newcommand{\HomMod}[1]{\textnormal{Hom}_{\,#1}}

\newcommand{\End}{\textnormal{End}\,}

\newcommand{\EndOp}{\textnormal{End}\superscr{\textnormal{op}}\,}

\newcommand{\Hom}{\textnormal{Hom}\,}

\newcommand{\LS}{\mathsf{L}}
\newcommand{\LSBar}{\overbarStraight{\LS}}

\newcommand{\Gen}{U}
\newcommand{\ValGen}{U}
\newcommand{\Lgen}{L}
\newcommand{\Rgen}{R}

\newcommand{\Quo}{\mathsf{Q}}
\newcommand{\QuoBar}{\overbarStraight{\Quo}}

\newcommand{\LP}{\mathsf{LP}}
\newcommand{\LPBar}{\overbarStraight{\LP}}

\newcommand{\ThetaNet}{\Theta}

\newcommand{\smin}{s_{\textnormal{min}}}
\newcommand{\smax}{s_{\textnormal{max}}}
\newcommand{\pmin}{\mathfrak{p}}

\newcommand{\TL}{\mathsf{TL}}

\newcommand{\PS}{\mathsf{K}}

\newcommand{\PD}{\PD\mathsf{D}}

\newcommand{\Wd}{\mathsf{M}}
\newcommand{\WdBar}{\overbarStraight{\Wd}}

\newcommand{\MTbas}{\theta}
\newcommand{\MTbasBar}{\overbarcal{\MTbas}}

\newcommand{\sing}{\mathcalligra{s}\,}
\newcommand{\singBar}{\overbarcal{\sing}}

\newcommand{\BarAction}{\left\bracevert\phantom{A}\hspace{-9pt}\right.}

\newcommand{\Sing}{w}
\newcommand{\SingBar}{\overbarStraight{\Sing}}

\newcommand{\im}{\textnormal{im}\,}

\newcommand{\cheque}{{\scaleobj{0.85}{\vee}}}

\newcommand{\SpecialPattern}{\mathsf{SP}}
\newcommand{\SpecialDiagram}{\mathsf{SD}}
\newcommand{\SpecialPatternBar}{\overbarStraight{\mathsf{SP}}}

\newcommand{\Uqsltwo}{{\mathsf{U}_q}}
\newcommand{\UqsltwoBar}{{\smash{\overbarStraight{\mathsf{U}}_q}}}

\newcommand{\UqsltwoPow}[1]{{\mathsf{U}^{\otimes #1}_q}}
\newcommand{\UqsltwoBarPow}[1]{{\smash{\overbarStraight{\mathsf{U}}}^{\otimes #1}_q}}

\newcommand{\EGen}{E}
\newcommand{\FGen}{F}
\newcommand{\KGen}{K}
\newcommand{\AntiEGen}{\smash{\overbarcal{\EGen}}}
\newcommand{\AntiFGen}{\smash{\overbarcal{\FGen}}}
\newcommand{\AntiKGen}{\smash{\overbarcal{\KGen}}}

\newcommand{\PlusMinusEGen}{\EGen_\pm}
\newcommand{\PlusMinusFGen}{\FGen_\pm}
\newcommand{\PlusMinusKGen}{\KGen_\pm}
\newcommand{\MinusPlusEGen}{\EGen_\mp}
\newcommand{\MinusPlusFGen}{\FGen_\mp}
\newcommand{\MinusPlusKGen}{\KGen_\mp}

\newcommand{\AntiPlusMinusEGen}{\AntiEGen_\pm}
\newcommand{\AntiPlusMinusFGen}{\AntiFGen_\pm}
\newcommand{\AntiPlusMinusKGen}{\AntiKGen_\pm}
\newcommand{\AntiMinusPlusEGen}{\AntiEGen_\mp}
\newcommand{\AntiMinusPlusFGen}{\AntiFGen_\mp}
\newcommand{\AntiMinusPlusKGen}{\AntiKGen_\mp}

\newcommand{\PlusMinusUqSLTwo}{\mathsf{U}_{q\pm}}
\newcommand{\MinusPlusUqSLTwo}{\mathsf{U}_{q\mp}}

\newcommand{\AntiPlusMinusUqSLTwo}{\smash{\overbarStraight{\mathsf{U}}}_{q\pm}}
\newcommand{\AntiMinusPlusUqSLTwo}{\smash{\overbarStraight{\mathsf{U}}}_{q\mp}}

\newcommand{\Usltwo}{\mathsf{U}}

\newcommand{\FundBasis}{\varepsilon}
\newcommand{\FundBasisBar}{\overbarcal{\FundBasis}}
\newcommand{\Basis}{e}
\newcommand{\BasisBar}{\overbarStraight{\Basis}}
\newcommand{\np}{d}

\newcommand{\OneVec}[1]{\vec{#1}}

\newcommand{\Summed}{n}

\newcommand{\sIndex}{s}
\newcommand{\HWsp}{\mathsf{H}}
\newcommand{\HWspBar}{\overbarStraight{\HWsp}}

\newcommand{\Ksp}{\mathsf{V}}
\newcommand{\KspBar}{\overbarStraight{\Ksp}}
\newcommand{\DefectSet}{\mathsf{E}}
\newcommand{\multii}{\varsigma}
\newcommand{\multiii}{\varpi}

\newcommand{\defect}[1]{\mathrm{def}_#1}

\newcommand{\VecSp}{\mathsf{V}}
\newcommand{\VecSpBar}{\overbarStraight{\VecSp}}

\newcommand{\Trep}{{\mathscr{I}}}

\newcommand{\TrepBar}{\overbarcalcal{\Trep}}

\newcommand{\fds}{{\underset{\vspace{1pt} \scaleobj{1.4}{\check{}}}{\multii}}}
\newcommand{\lds}{{\hat{\multii}}}

\newcommand{\HWvec}{u}
\newcommand{\HWvecBar}{\overbarcal{\HWvec}}

\newcommand{\HWvecMap}{\eta}
\newcommand{\HWvecMapBar}{\overbarcal{\eta}}

\newcommand{\super}[1]{^{\scaleobj{0.85}{(#1)}}}
\newcommand{\sub}[1]{_{\scaleobj{0.85}{(#1)}}}
\newcommand{\superscr}[1]{^{\scaleobj{0.85}{#1}}}
\newcommand{\subscr}[1]{_{\scaleobj{0.85}{#1}}}

\newcommand{\LSBiFormBar}[2]{(#1 \BarAction #2)}

\newcommand{\SPBiForm}[2]{(#1 \,|\, #2)}
\newcommand{\SPBiFormBig}[2]{\big(#1 \,\big|\, #2 \big)}

\newcommand{\SPBiFormNewInv}[2]{\big( \hspace*{-1mm}| #1 \,,\, #2 |\hspace*{-1mm}\big) }

\newcommand{\Dim}{D}

\newcommand{\WJProj}{P}
\newcommand{\WJproj}{\WJProj}
\newcommand{\WJProjHat}{\overhat{\WJProj}}
\newcommand{\WJEmb}{I}

\newcommand{\one}{\mathbf{1} \hspace*{-.25em} \textnormal{l}}

\newcommand{\id}{\textnormal{id}}
\global\long\def\Projection{\mathfrak{P}}
\global\long\def\Projectionhat{\smash{\overhat{\Projection}}}
\global\long\def\Embedding{\mathfrak{I}}

\global\long\def\ProjectionBar{\overbarStraight{\Projection}}
\global\long\def\ProjectionhatBar{\smash{\overhat{\ProjectionBar}}}
\global\long\def\EmbeddingBar{\overbarStraight{\Embedding}}

\newcommand{\ProjBox}{\;\vcenter{\hbox{\includegraphics[scale=0.275]{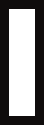}}}\;}

\newsavebox\CBox
\newcommand\hcancel[2][0.5pt]{%
  \ifmmode\sbox\CBox{$#2$}\else\sbox\CBox{#2}\fi%
  \makebox[0pt][l]{\usebox\CBox}%
  \rule[0.5\ht\CBox-#1/2]{\wd\CBox}{#1}}
  
\global\long\def\qbin#1#2{\left[\begin{array}{c}
	#1\\
	#2 
	\end{array}\right]}

\newcommand{\defects}{%
 \,  \raisebox{-.4ex}{%
    \scalebox{1.2}{%
      \rotatebox[origin=c]{270}{$\exists$}%
    }%
  }%
}

\newcommand{\defectsBar}{%
 \,  \raisebox{-.4ex}{%
    \scalebox{1.2}{%
      \rotatebox[origin=c]{90}{$\exists$}%
    }%
  }%
}

\newcommand{\defectsVal}{%
 \,  \raisebox{.4ex}{%
    \scalebox{1.1}{%
      \rotatebox[origin=c]{0}{$\mathbin{\stackanchor[-8.5pt]{\vee}{\perp}}$}%
    }%
  }%
}

\newcommand{\Mod}[1]{\ (\mathrm{mod}\ #1)}

\newcommand{\DPle}{\prec} 
\newcommand{\DPleq}{\preceq} 

\definecolor{amber}{rgb}{1.0, 0.75, 0.0}

\newcommand{\blue}{\textcolor{blue}}

\newcommand{\red}{\textcolor{red}}

\newcommand{\be}{\begin{equation}}
\newcommand{\ee}{\end{equation}}
\newcommand{\bea}{\begin{eqnarray}}
\newcommand{\eea}{\end{eqnarray}}
\newcommand{\bean}{\begin{eqnarray*}}
\newcommand{\eean}{\end{eqnarray*}}

\newcommand{\CCprojector}{\pi}
\newcommand{\CChatprojector}{\overhat{\CCprojector}}
\newcommand{\CCembedor}{\iota}

\newcommand{\CCprojectorBar}{\overbarcal{\CCprojector}}
\newcommand{\CChatprojectorBar}{\overhat{\CCprojectorBar}}
\newcommand{\CCembedorBar}{\overbarcal{\CCembedor}}

\newcommand{\CCunprojector}{\psi}
\newcommand{\CChatunprojector}{\overhat{\CCunprojector}}

\newcommand{\CCunprojectorBar}{\overbarcal{\CCunprojector}}
\newcommand{\CChatunprojectorBar}{\overhat{\CCunprojectorBar}}

\newcommand{\isom}{\;\cong\;}

\newcommand{\LeftRegRep}{\rho}
\newcommand{\RightRegRep}{\overbarcal{\rho}}
\newcommand{\LeftTwistRep}{\tau}
\newcommand{\RightTwistRep}{\overbarcal{\tau}}

\theoremstyle{plain}
\newtheorem*{theorem*}{Theorem}
\newtheorem{theorem}{Theorem}
\numberwithin{theorem}{section} 
\newtheorem{prop}[theorem]{Proposition}
\numberwithin{prop}{section} 
\newtheorem{cor}[theorem]{Corollary}
\numberwithin{cor}{section} 
\newtheorem{lem}[theorem]{Lemma}
\numberwithin{lem}{section} 

\numberwithin{conj}{section} 

\numberwithin{quest}{section} 

\numberwithin{InductAssump}{section} 

\numberwithin{comment}{section}

\theoremstyle{definition}
\newtheorem{remark}[theorem]{Remark}
\numberwithin{remark}{section} 

\numberwithin{recipe}{section} 
\newtheorem{defn}[theorem]{Definition}
\numberwithin{defn}{section} 
 
\numberwithin{example}{section} 


\newcounter{parentnumber}

\def\clap#1{\hbox to 0pt{\hss#1\hss}}
\def\mathclap{\mathpalette\mathclapinternal}

\def\mathclapinternal#1#2{%
\clap{$\mathsurround=0pt#1{#2}$}}

\makeatletter
\DeclareRobustCommand{\cev}[1]{%
  \mathpalette\do@cev{#1}%
}
\newcommand{\do@cev}[2]{%
  \fix@cev{#1}{+}%
  \reflectbox{$\m@th#1\vec{\reflectbox{$\fix@cev{#1}{-}\m@th#1#2\fix@cev{#1}{+}$}}$}%
  \fix@cev{#1}{-}%
}
\newcommand{\fix@cev}[2]{%
  \ifx#1\displaystyle
    \mkern#23mu
  \else
    \ifx#1\textstyle
      \mkern#23mu
    \else
      \ifx#1\scriptstyle
        \mkern#22mu
      \else
        \mkern#22mu
      \fi
    \fi
  \fi
}

\numberwithin{equation}{section}

\numberwithin{figure}{section}

\renewcommand{\thesection}{\arabic{section}} 

\makeatother

\allowdisplaybreaks

\begin{document}
\title{Higher-spin quantum and classical Schur-Weyl duality for $\mathfrak{sl}_2$ \vspace*{.5cm}}


\author{\bf Steven M. Flores}
\affiliation{\blue{\tt \small steven.miguel.flores@gmail.com} \\ 
Department of Mathematics and Systems Analysis, \\ 
P.O. Box 11100, FI-00076, Aalto University, Finland}

\author{\bf Eveliina Peltola}
\affiliation{\blue{\tt \small eveliina.peltola@hcm.uni-bonn.de} \\ 
Institute for Applied Mathematics, University of Bonn, \\
Endenicher Allee 60, D-53115 Bonn, Germany
\vspace*{.5cm}}


\begin{abstract}
\begingroup
\setlength{\parindent}{1.5em}
\setlength{\parskip}{.5em}

It is well-known that the commutant algebra of the $U_q(\mathfrak{sl}_2)$-action 
on the $n$-fold tensor product of its fundamental module 
is isomorphic to the Temperley-Lieb algebra $\TL_n(\nu)$ with fugacity parameter $\nu = -q - q^{-1}$
(at least in the generic case, i.e., when $q$ is not a root of unity, or $n$ is small enough). 
Furthermore, the simple $U_q(\mathfrak{sl}_2)$-modules appearing in the direct-sum decomposition 
of the $n$-fold tensor product module are in one-to-one correspondence with those of the Temperley-Lieb algebra.
This double-commutant property is referred to as \emph{quantum Schur-Weyl duality}.

In this article, we investigate such a duality in great detail.
We prove that the commutant of the $U_q(\mathfrak{sl}_2)$-action on any generic type-one tensor product module
is isomorphic to a diagram algebra that we call the valenced Temperley-Lieb algebra $\TL_\multii(\nu)$.
This corresponds to representations with higher spin, which results in the need of valences (or colors) in the Temperley-Lieb diagrams.
We establish detailed direct-sum decompositions exhibiting this duality and find explicit bases 
amenable to concrete calculations, important in applications.
We also include a double-commutant type property for homomorphisms between different 
$U_q(\mathfrak{sl}_2)$-modules, realized by valenced diagrams. The diagram calculus is reminiscent to
Kauffman's recoupling theory and the graphical methods developed among others by Penrose and Frenkel \& Khovanov.
The results also contain the standard quantum Schur-Weyl duality as a special case, and when specialized to $q \rightarrow 1$,
imply the classical Frobenius-Schur-Weyl duality  for the Lie algebra $\mathfrak{sl}_2(\bC)$ and a higher-spin version thereof.

We especially aim to provide a self-contained and elementary presentation 
useful to applications also in areas remote from algebra and representation theory. 
Only very basic facts about representations are needed to understand this article.
For this reason, we also include basic known results as well as results that can be regarded as folklore
but are lacking systematic or easy references.

\endgroup
\end{abstract}

\maketitle

\newpage

\vspace*{-1.5cm}
\renewcommand{\tocname}{}
{\hypersetup{linkcolor=black}
\tableofcontents
}

\begingroup
\setlength{\parindent}{1.5em}
\setlength{\parskip}{.5em}
\section{Introduction}
Quantum group symmetries, hidden or apparent, are ubiquitous in mathematics and mathematical physics.
Traditionally, they have played a central role especially in certain areas of low-dimensional topology, 
but also in numerous other fields of modern mathematics.
The emergence of quantum groups can be argued to originate 
from observations of hidden symmetries in certain physical quantum systems, deforming the classical Lie algebra symmetries,
as pioneered by V.~Drinfeld, L.~Faddeev, and M.~Jimbo in the 1980s.
Thereafter, such symmetries have been observed, e.g., in topological quantum field theory, conformal field theory,
and integrable models in statistical physics, e.g.,~\cite{wit, mr, ps, gras, bax, pm, var}. 
Quantum groups also immediately gained the interest of mathematicians, with remarkable success.
From the Yang-Baxter equation underlying these symmetries, one gets to the theory of tensor categories, 
higher representation theory, and categorification, see~\cite{fks, cs}.
Using quantum groups and Hecke algebras, one constructs invariants of knots, tangles, and links~\cite{vj, lk, rt}.
Different formulations of quantum groups led to advances in noncommutative geometry~\cite{ack, man}, 
and others in representation theory~\cite{Lus89,AKP91,fkk}, for instance.
In the relatively new field of random geometry, one finds quantum group symmetries related to probabilistic models.
For instance, in \cite{jjk, kp, kkp, kp2, ep2, fp2, fp1} explicit 
highest-weight vectors are used to construct specific correlation functions in conformal field theory,
with applications to conformally invariant random geometry.
Various quantum group symmetries and dualities have also become useful tools in integrable probability~\cite{BP14, BP15}, e.g.,
to analyze non-equilibrium characteristics in asymmetric particle processes and stochastic vertex models
(see~\cite{CGRS16, KMMO16} and references therein).

Due to the enormous number of applications, various tools and theory on quantum groups have been developed. 
In the present article, our aim is to focus on concrete and elementary understanding of the simplest,
but also most commonly encountered, quantum group associated to the Lie algebra $\mathfrak{sl}_2(\bC)$.
In order to utilize the full power of these symmetries in applications, detailed knowledge is crucial,
and in light of 
their emergence in many seemingly remote areas, 
we believe that elementary tools for obtaining concrete information 
are becoming more and more important.
In the present work, we specifically consider the quantum group $U_q(\mathfrak{sl}_2)$ and its tensor product representations.
(Despite the commonly used term ``group," the quantum group that we consider is a non-commutative
and non-cocommutative Hopf algebra, whose commutation relations are obtained from those of $\mathfrak{sl}_2$ 
by a deformation by a non-zero complex parameter $q \neq \pm 1$.) 
The presentation here is mostly self-contained, with little need of prerequisites, 
and in particular, all used techniques are completely elementary. 

It has been known for quite a while, at least since the work~\cite{rtw, pen} about 90 years ago,
that planar diagrams, now called tangles in the Temperley-Lieb category~\cite{tl, vj, kl, vt, cfs}
completely describe homomorphisms between tensor powers $\smash{\Wd\sub{1}^{\otimes n}}$ 
of the fundamental module $\Wd\sub{1}$ of the Lie algebra $\mathfrak{sl}_2$.
On the other hand, the classical ``Schur-Weyl duality" dating back 100 years ago~\cite{schur, hew},
shows that the algebra of all endomorphisms of $\smash{\Wd\sub{1}^{\otimes n}}$ 
that commute with the action of the enveloping algebra $U(\mathfrak{sl}_2)$ is generated by transpositions
of the tensor factors (thus being a quotient of the group algebra of the symmetric group acting naturally on $\smash{\Wd\sub{1}^{\otimes n}}$).
Interestingly, the latter agrees with the former: transpositions have a natural diagram interpretation,
whose action coincides with the planar diagram algebra known as Temperley-Lieb algebra $\TL_n(\nu)$ with (loop) fugacity parameter $\nu = -2$.

Analogous properties are inherited by the quantum group $\Uqsltwo := U_q(\mathfrak{sl}_2)$~\cite{kr}.
Indeed, when $q$ is not a root of unity, $\Uqsltwo$ is semisimple and its representation theory is analogous to that of 
the classical Lie algebra $\mathfrak{sl}_2$, and the planar diagram approach generalizes nicely~\cite{cfs, fk}. 
In this case, the ``quantum Schur-Weyl duality" holds for the Temperley-Lieb algebra $\TL_n(\nu)$ with $\nu = -[2] = -q-q^{-1}$.
M.~Jimbo~\cite{mj2} (independently, R.~Dipper and G.~James~\cite{dj})
observed this from the quasi-triangular structure of $\Uqsltwo$: 
braiding of tensor components in $\smash{\Wd\sub{1}^{\otimes n}}$, where $\Wd\sub{1}$ is now the fundamental $\Uqsltwo$-module, 
gives rise to an action of the $q$-deformation of the symmetric group algebra, the Hecke algebra.
However, this action is not faithful. Factoring out by its kernel gives the Temperley-Lieb algebra $\TL_n(\nu)$~\cite{pm, ppm}.
(A version of such a duality also holds in the non-semisimple case~\cite{ppm, gv, psa}, 
but $\Uqsltwo$ has to be extended by so-called divided powers and also the Temperley-Lieb algebra is then non-semisimple.)

The main purpose of the present article is to investigate in detail the commuting actions of the quantum group $\Uqsltwo$
and the valenced Temperley-Lieb algebra $\TL_\multii(\nu)$ on arbitrary tensor product (type-one) modules of $\Uqsltwo$.
The latter is a planar diagram algebra comprising diagrams of valenced (or colored) tangles, containing 
the usual Temperley-Lieb type tangles as a special case.
Importantly, this algebra can be used to determine all homomorphisms between arbitrary tensor product (type-one) modules of $\Uqsltwo$.
The well-known standard version of the Temperley-Lieb algebra was also crucial in the development of the theory of knot invariants~\cite{vj, kl},
and has many historical connections to mathematical physics, e.g., 
in the field of exactly solvable models in statistical physics~\cite{tl, pm, bax},

We develop in parallel an explicit algebraic, or combinatorial, description of the structure of these tensor product modules,
and a diagram calculus that facilitates many computations and makes the appearance of the valenced Temperley-Lieb algebra
apparent. 
We have special emphasis on understanding highest-weight vectors, also important in applications, e.g., to conformally invariant random geometry.
In addition to the ``usual" case of $\Uqsltwo$, we also include analogous results for its variants
obtained from changing the way how the algebra acts on tensor products, or sending the parameter $q$ to its inverse.
Such variants and their mutual connections are also needed in applications.

We take the approach of studying directly the quantum group and recover the classical results as a consequence.
Traditionally, the classical case is considered more elementary, but the proof of the classical Schur-Weyl duality 
also relies on quite specific combinatorics. 
Treating directly the quantum group puts all parameter values $q$ or $\nu$ to equal footing, 
and also gives a route to go beyond the semisimple case. 


\subsection{Background and motivation --- the quantum group}

In this work, we frequently use the \emph{$q$-integers}, \emph{$q$-factorials}, and \emph{$q$-binomial coefficients}, 
defined for any integers $k \in \bZ$ and $0 \leq \ell \leq m$ and for any $q \in \bC \setminus \{0\} =: \bC^\times$ respectively as
\begin{align}  \label{Qinteger} 
[k] := \frac{q^{k} - q^{-k}}{q - q^{-1}} , \qquad  \qquad
[m]! := \prod_{i=1}^m [i] , \qquad \qquad
\qbin{m}{\ell} := \frac{[m]!}{[\ell]![m-\ell]!} .
\end{align}
When $q \in \{ \pm 1\}$, these become the usual integers, factorials, and binomial coefficients.
We note that $[m]$ vanishes if and only if $m$ is zero or $q \notin \{ \pm 1\}$ and $m$ is an integer multiple of the number $\pmin(q)$ defined as
\begin{align} 
\pmin(q) := 
\begin{cases} 
\infty, & \text{$q$ is not a root of unity}, \\
p, & \text{$q=e^{\pi \ii p'/p}$ for coprime $p,p' \in \bZpos$} .
\end{cases}
\end{align}
Thus, $\pmin(q)$ is the smallest power $p = \pmin(q)$ of $q$ such that $q^p$ equals $+1$ or $-1$.

One of the main characters in this work is the quantum group $\Uqsltwo := U_q(\mathfrak{sl}_2)$.
As a $\bC$-algebra, it may be thought of as a deformed version of the universal enveloping algebra $U(\mathfrak{sl}_2)$
of the Lie algebra $\mathfrak{sl}_2$ of traceless $(2 \times 2)$-matrices, with deformation parameter $q \in \bC^\times \setminus \{\pm 1\}$. 
In a suitable sense, the latter is recovered as $q \rightarrow 1$.
Explicitly, $\Uqsltwo$ is the infinite-dimensional 
associative algebra with unit $1$ and generators $E$, $F$, $K$, $K^{-1}$ exclusively satisfying the relations
\begin{align} 
KK^{-1} = K^{-1}K = 1 , \qquad 
KE = q^2 EK , \qquad 
KF = q^{-2} FK , \qquad 
[E,F] = \frac{K - K^{-1}}{q - q^{-1}} .
\end{align}

For each nonnegative integer $s < \pmin(q)$, we let $\Wd\sub{s}$ denote the type-one simple 
$(s + 1)$-dimensional $\Uqsltwo$-module defined below, and we let 
$\LeftRegRep\sub{s} \colon \Uqsltwo \longrightarrow \End \Wd\sub{s}$ denote its corresponding irreducible representation.  
We often use the shorthand notation $x.v := \LeftRegRep\sub{s}(x)(v)$
for the action of elements $x \in \Uqsltwo$ on vectors $v \in \Wd\sub{s}$.
This action is completely determined by the action of the generators $\{E,F,K\}$ of
$\Uqsltwo$ on the standard basis $\{ \smash{\Basis_\ell\super{s}} \, | \, 0 \leq \ell \leq s \}$ 
for $\Wd\sub{s}$, given by
\begin{align} \label{HopfRepIntro} 
F.\Basis_\ell\super{s} := 
\begin{cases}
\Basis_{\ell+1}\super{s}, & 0 \leq \ell \leq s - 1 , \\ 
0, & \ell = s ,
\end{cases}
\qquad 
E.\Basis_\ell\super{s} := 
\begin{cases}
[\ell][s - \ell + 1] \Basis_{\ell-1}\super{s}, & 1 \leq \ell \leq s , \\ 
0, & \ell = 0 ,
\end{cases}
\qquad 
K.\Basis_\ell\super{s} := q^{s - 2\ell} \Basis_\ell\super{s} .
\end{align}
We call these modules ``simple type-one modules.'' 
They are analogues of the highest-weight modules of $\mathfrak{sl}_2$.


The algebra $\Uqsltwo$ has a bialgebra (and even a Hopf algebra) structure, which enables us to 
form $\Uqsltwo$-modules from the ground field $\bC$ and from tensor products of its modules. 
We recommend, e.g., the textbooks~\cite{cp,ck,krt} for background, although the present work is mostly self-contained.
Of particular importance for this article are 
\begin{align} \label{PreSpinChain} 
\Module{\VecSp_\multii}{\Uqsltwo} := \Wd\sub{\sIndex_1} \otimes \Wd\sub{\sIndex_2} \otimes \dotsm \otimes \Wd\sub{\sIndex_{\np_\multii}} ,
\qquad \qquad \text{with} \quad \multii := (\sIndex_1, \sIndex_2, \ldots, \sIndex_{\np_\multii}) \in \bZpos^{\np_\multii} ,
\qquad \max \multii < \pmin(q) ,
\end{align}
called ``type-one'' tensor product modules 
with ``spins" $\sIndex_1, \sIndex_2, \ldots, \sIndex_{\np_\multii}$, carrying the $\Uqsltwo$-action 
\begin{align} 
\LeftRegRep_\multii \colon \Uqsltwo \longrightarrow \End{\VecSp_\multii} ,
\qquad \qquad 
\LeftRegRep_\multii := (\LeftRegRep\sub{\sIndex_1} \otimes \LeftRegRep\sub{\sIndex_2} \otimes \dotsm \otimes \LeftRegRep\sub{\sIndex_{\np_\multii}}) \circ \Delta\super{\np_\multii} ,
\end{align}
where $\Delta \colon \Uqsltwo \longrightarrow \Uqsltwo \otimes \Uqsltwo$ is a coproduct given in~\eqref{CoProd}
in section~\ref{UqSect} and $\smash{\Delta\super{\np_\multii}}$ is its $(\np_\multii-1)$:th iterate given in~\eqref{IteratedCoProd}.
When $\np_\multii=0$, we use the convention that $\multii = (0)$. In this case, $\Wd\sub{0}$ is the trivial one-dimensional module, 
which we always identify with the ground field $\bC$, omitting it from all tensor products.

If $q$ is not a root of unity, $\Uqsltwo$ is a semisimple algebra. However, when $q$ is a root of unity,
the representation theory of $\Uqsltwo$ is non-semisimple and rather complicated.
For small enough spins, 
the module $\Module{\VecSp_\multii}{\Uqsltwo}$ is semisimple and it has 
the following direct-sum decomposition into simple $\Uqsltwo$-submodules 
with multiplicities $\smash{\Dim_\multii\super{s}} \in \bZpos$ given in~\eqref{CountLP}: 
\begin{align} \label{MoreGenDecompIntro} 
\Module{\VecSp_\multii}{\Uqsltwo} \isom 
\bigoplus_{s \, \in \, \DefectSet_\multii} \Dim_\multii\super{s} \Wd\sub{s} \qquad \qquad \text{when} \qquad
\Summed_\multii := \sIndex_1 + \sIndex_2 + \cdots + \sIndex_{\np_\multii}  < \pmin(q) ,
\end{align}
where $\Wd\sub{s}$ are non-isomorphic simple $\Uqsltwo$-modules, with index set $\DefectSet_\multii \ni s$.
Analogously to the representation theory of classical Lie algebras, for this semisimple case, 
each submodule is generated by a highest-weight vector in 
\begin{align} \label{HWWSpIntro}
\HWsp_\multii := \big\{ v \in \VecSp_\multii \, \big| \, E.v = 0 \big\} 
= \bigoplus_{s \, \in \, \DefectSet_\multii} \HWsp_\multii\super{s} ,
\end{align}
whose weights $q^{s}$ are the $K$-eigenvalues: we have $K.v = q^{s} v$ for highest-weight vectors $v \in \smash{\HWsp_\multii\super{s}}$.
(Note that this structure fails in general when $\Summed_\multii \geq \pmin(q)$.)
Direct-sum decomposition~\eqref{MoreGenDecompIntro} is well-known, although it is usually stated only for the case when $q$ is not a root of unity. 
In the present article, we investigate decompositions of this type, e.g., finding explicit bases realizing them. 
The combinatorial numbers $\smash{\Dim_\multii\super{s}}$, discussed in detail in section~\ref{CobloVecSec}, 
enumerate, e.g., so-called ``valenced link patterns" 
$\alpha$, defined in section~\ref{TLReviewSec}, or alternatively, ``walks over the multiindex $\multii$,"  defined in section~\ref{CobloVecSec}. 
Both are combinatorial objects with useful diagram representations facilitating the analysis of 
the structure of~\eqref{MoreGenDecompIntro} also in more general situations.

\subsection{Concrete understanding of highest-weight vectors}

Let us examine the 
module $\Module{\VecSp_\multii}{\Uqsltwo}$.
To understand its submodule structure, the first task is to find linearly independent highest-weight vectors in $\VecSp_\multii$.
A particularly useful set of such vectors is indexed by the numbers $\smash{\Dim_\multii\super{s}}$, even if $\Summed_\multii \geq \pmin(q)$.
We call them (valenced) link-pattern basis vectors. 
These vectors not only give detailed information about the structure of $\Module{\VecSp_\multii}{\Uqsltwo}$,
but are also of independent interest due to their useful properties, which translate to features of other objects in applications 
of quantum group symmetries (e.g.,~\cite{bf, IWWZ, kp, kp2, fp2, fp1}). 
Similar bases seem to be known, at least implicitly, for the community in diagrammatic representation theory~\cite{fk}.


\begin{prop} \label{DirectSumPropCombination}
\textnormal{(Link-pattern basis vectors):} 
In $\Module{\VecSp_\multii}{\Uqsltwo}$, 
there exists a collection of linearly independent 
highest-weight vectors $\Sing_\alpha$ indexed by valenced link patterns $\alpha$
\textnormal{(}definition~\ref{SingletBasisDefinition} in section~\ref{subsec: link state hwv for n}\textnormal{)}.
Diagrammatically, we have  
\begin{align} \label{SingAlphaDiagramValIntro}
\alpha \quad & = \quad \vcenter{\hbox{\includegraphics[scale=0.275]{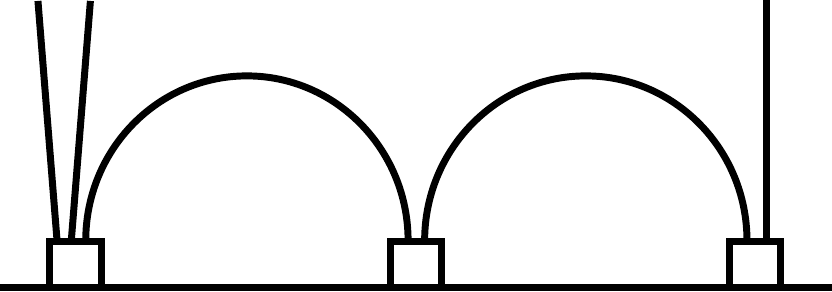}}}
\qquad \qquad \Longrightarrow \qquad \qquad
\Sing_\alpha \quad = \quad  \vcenter{\hbox{\includegraphics[scale=0.275]{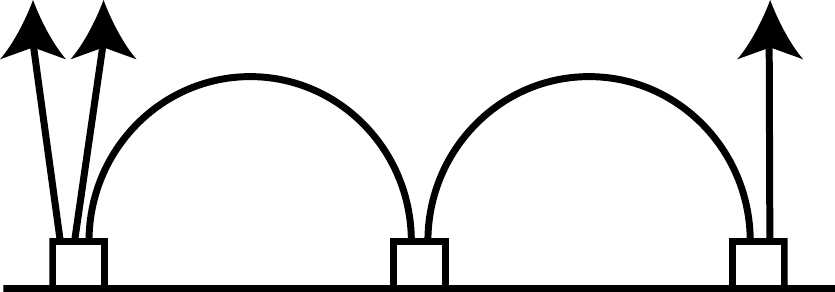} ,}} 
\end{align}
and the descendants of $\Sing_\alpha$ also have diagram representations 
\begin{align} \label{DescDiagram2Intro} 
F^\ell.\Sing_\alpha  
\quad = \quad \frac{[s]!}{[s-\ell]!} \,\, \times \,\, 
\vcenter{\hbox{\includegraphics[scale=0.275]{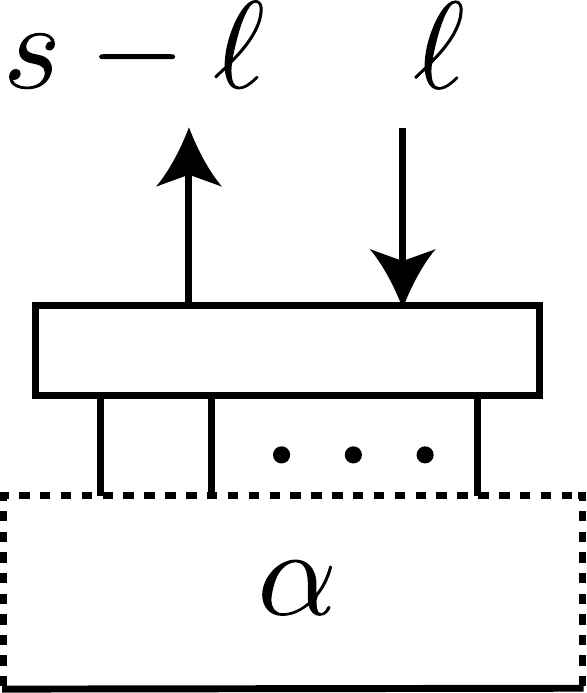} .}} 
\end{align}
These vectors give rise to an embedding of left $\Uqsltwo$-modules \textnormal{(}with multiplicities $\smash{\Dim_\multii\super{s}}$ given in~\eqref{CountLP}\textnormal{)},
\begin{align} \label{DirectSumInclusion2Intro}
\bigoplus_{\substack{s \, \in \, \DefectSet_\multii \\ s \, < \, \pmin(q) }} \Dim_\multii\super{s} \Wd\sub{s} 
\quad \lhook\joinrel\rightarrow \quad \Module{\VecSp_\multii}{\Uqsltwo} .
\end{align}
If $\Summed_\multii < \pmin(q)$, then this embedding is an isomorphism.
\end{prop}

\begin{proof}
The vectors $\Sing_\alpha$ are explicitly constructed in
definition~\ref{SingletBasisDefinition} and lemmas~\ref{sGradingLem2} and~\ref{SingletBasisIsLinIndepLem2}.
Lemma~\ref{DescLem2} gives the diagram representations for the $F$-descendants. 
Proposition~\ref{MoreGenDecompAndEmbProp2} gives the embedding or isomorphism~\eqref{DirectSumInclusion2Intro}. 
\end{proof}

There is another interesting set of highest-weight vectors that yields the direct-sum decomposition~\eqref{MoreGenDecompIntro}
when $q$ is not a root of unity, and an embedding slightly weaker than~\eqref{DirectSumInclusion2Intro} 
when $q$ is a root of unity. Despite this weakness, this basis is very useful for instance because it is orthogonal with respect to 
a standard invariant bilinear pairing $\SPBiForm{\cdot}{\cdot}$.
We call this basis the ``conformal-block basis" due to its connections to conformal field theory. 
(Indeed, these conformal-block vectors should correspond to intertwiners in a Virasoro VOA~\cite[section~\red{3}]{kkp}.)
These highest-weight vectors $\smash{\HWvec^{\varrho}_{\multii}}$ are also enumerated by the combinatorial numbers $\smash{\Dim_\multii\super{s}}$,
which count certain walks (see~\eqref{WalkHeights} in section~\ref{CobloVecSec}). 
Below, the embedding has multiplicities $\smash{\hat{\Dim}_\multii\super{s}}$ 
that can be smaller than $\smash{\Dim_\multii\super{s}}$ due to a technical condition.

The bilinear pairing is defined as $\SPBiForm{\cdot}{\cdot} \colon \VecSpBar_\multii \times \VecSp_\multii \longrightarrow \bC$,
\begin{align} 
\SPBiFormBig{\BasisBar_{\ell_1}\super{\sIndex_1} \otimes \BasisBar_{\ell_2}\super{\sIndex_2} \otimes \dotsm \otimes \BasisBar_{\ell_{\np_\multii}}\super{\sIndex_{\np_\multii}}}{\Basis_{m_1}\super{\sIndex_1} \otimes \Basis_{m_2}\super{\sIndex_2} \otimes \dotsm \otimes \Basis_{m_{\np_\multii}}\super{\sIndex_{\np_\multii}}}
= \prod_{k \, = \, 1}^{\np_\multii} \delta_{\ell_k, m_k}[\ell_k]!^2\,\qbin{\sIndex_k}{\ell_k} ,
\end{align} 
where 
\begin{align}
\VecSp_\multii = \Span \big\{ \Basis_{m_1}\super{\sIndex_1} \otimes \Basis_{m_2}\super{\sIndex_2} \otimes \dotsm \otimes \Basis_{m_{\np_\multii}}\super{\sIndex_{\np_\multii}} \, \big| \, 0 \leq m_1 \leq \sIndex_1, \, 0 \leq m_2 \leq \sIndex_2 , \; \ldots , \; 0 \leq m_{\np_\multii} \leq \sIndex_{\np_\multii} \big\}
\end{align}
is the vector space underlying the left $\Uqsltwo$-module $\Module{\VecSp_\multii}{\Uqsltwo}$ and 
\begin{align}
\VecSpBar_\multii = \Span \big\{  \BasisBar_{\ell_1}\super{\sIndex_1} \otimes \BasisBar_{\ell_2}\super{\sIndex_2} \otimes \dotsm \otimes \BasisBar_{\ell_{\np_\multii}}\super{\sIndex_{\np_\multii}} \, \big| \, 0 \leq \ell_1 \leq \sIndex_1, \, 0 \leq \ell_2 \leq \sIndex_2 , \; \ldots , \; 0 \leq \ell_{\np_\multii} \leq \sIndex_{\np_\multii} \big\}
\end{align}
is a vector space of the same dimension but with a right $\Uqsltwo$-action $\RModule{\VecSpBar_\multii}{\Uqsltwo}$
discussed in section~\ref{RepTheorySect}.

\begin{prop} \label{DirectSumPropCoblo}
\textnormal{(Conformal-block basis vectors):} 
In $\Module{\VecSp_\multii}{\Uqsltwo}$ and $\RModule{\VecSpBar_\multii}{\Uqsltwo}$, 
there exist collections of linearly independent 
highest-weight vectors $\smash{\HWvec^{\varrho}_{\multii}}$ and $\smash{\HWvecBar^{\varrho}_{\multii}}$
indexed by walks $\varrho$ over $\multii$, which are also orthogonal:
\begin{align} \label{WalkBiForm2Intro} 
\SPBiForm{\HWvecBar^{\varrho}_{\multii}}{\HWvec^{\varrho'}_{\multii}} 
= \delta_{\varrho, \varrho'} \prod_{j \, = \, 1}^{\np_\multii - 1}  
\frac{\ThetaNet( r_j, r_{j+1}, \sIndex_{j+1} )}{(q - q^{-1})^{r_j + \sIndex_{j+1} - r_{j+1}} \big[ \frac{r_j + \sIndex_{j+1} - r_{j+1}}{2} \big]!^2 \, [r_{j + 1}+1]} ,
\end{align}
where $\ThetaNet$ is an explicit constant given by the evaluation of the Theta network~\eqref{ThetaFormula}.
Diagrammatically, we have 
\begin{align} \label{EmbTheta2Intro} 
\HWvec^{\varrho}_{\multii} \quad = \quad
\Bigg( \prod_{j \, = \, 1}^{\np_\multii - 1} \frac{(\ii q^{1/2})^{\frac{r_j + \sIndex_{j+1} - r_{j+1}}{2}}}{(q - q^{-1})^{\frac{r_j + \sIndex_{j+1} - r_{j+1}}{2}}[\frac{r_j + \sIndex_{j+1} - r_{j+1}}{2}]!} \Bigg)
\,\, \times
\hspace*{-5mm}
\vcenter{\hbox{\includegraphics[scale=0.275]{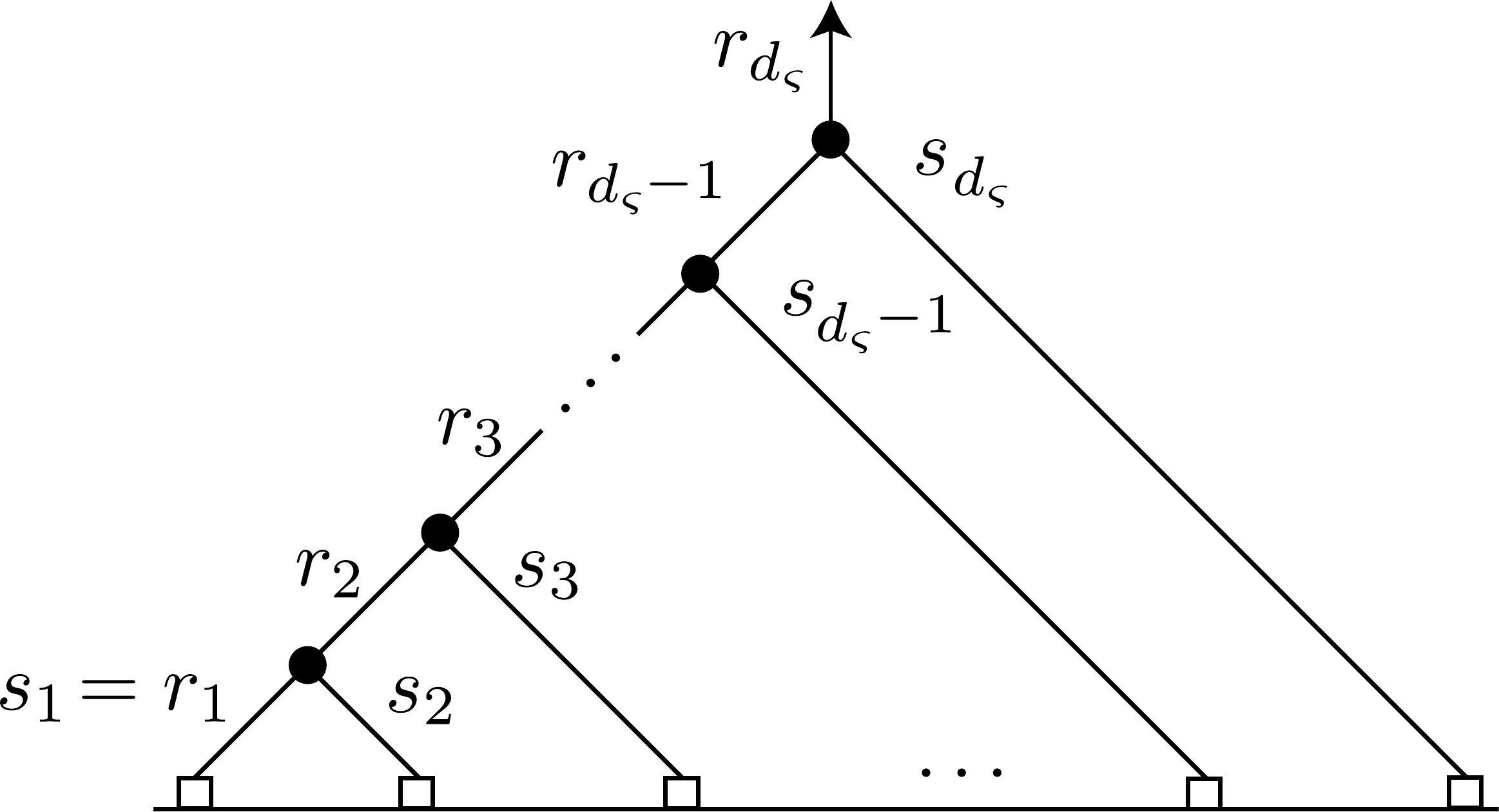} .}} 
\end{align}
These vectors give rise to an embedding of left $\Uqsltwo$-modules 
\textnormal{(}with multiplicities $\smash{\hat{\Dim}_\multii\super{s}}$ given in~\eqref{DimTilde}\textnormal{)},
\begin{align} \label{DirectSumInclusionIntro}
\bigoplus_{\substack{s \, \in \, \DefectSet_\multii \\ s \, < \, \pmin(q) }} \hat{\Dim}_\multii\super{s} \Wd\sub{s} 
\quad \lhook\joinrel\rightarrow \quad \Module{\VecSp_\multii}{\Uqsltwo} .
\end{align}
If $\Summed_\multii < \pmin(q)$, then this embedding is an isomorphism.
\end{prop}

\begin{proof}
The vectors $\smash{\HWvec^{\varrho}_{\multii}}$ are given in definition~\ref{CobloBasisDefinition}.
Lemma~\ref{EmbTheta2Lem} gives the diagram representation~\eqref{EmbTheta2Intro},  
and lemma~\ref{CoBloOrthBasisLem} shows that the vectors are orthogonal as in~\eqref{WalkBiForm2Intro}.
Proposition~\ref{MoreGenDecompAndEmbProp} gives the embedding or isomorphism~\eqref{DirectSumInclusionIntro}. 
\end{proof}

We cautiously note that the vectors $\smash{\HWvec^{\varrho}_{\multii}}$ and $\smash{\HWvecBar^{\varrho}_{\multii}}$ 
are not complex conjugates of each other.
They represent the holomorphic and anti-holomorphic sectors in conformal field theory.

\subsection{Quantum Schur-Weyl duality}

More information about the type-one $\Uqsltwo$-modules 
can be gained from understanding 
the commutant algebra $\EndMod{\Uqsltwo}{\VecSp_\multii}$,
which consists of all $\Uqsltwo$-homomorphisms from $\VecSp_\multii$ to itself.
For instance, all projections from $\VecSp_\multii$ onto its $\Uqsltwo$-submodules belong to the commutant, by definition.
In fact, the collection of all submodule projectors that act strictly on consecutive pairs of tensorands in $\VecSp_\multii$ 
generate the whole commutant algebra $\EndMod{\Uqsltwo}{\VecSp_\multii}$ when $\Summed_\multii < \pmin(q)$.
For a tensor product of two type-one modules, all multiplicities in the direct-sum decomposition~\eqref{MoreGenDecompIntro}
equal one:
\begin{align}
\Module{\VecSp\sub{r,t}}{\Uqsltwo}  
& \isom \Wd\sub{|r-t|} \oplus \Wd\sub{|r-t|+2} \oplus \Wd\sub{|r-t|+4} \oplus \dotsm \oplus \Wd\sub{r+t-2} \oplus \Wd\sub{r+t} 
\qquad \qquad \text{when} \qquad r + t < \pmin(q) .
\end{align}
We denote each projection onto the $(s+1)$-dimensional simple $\Uqsltwo$-submodule 
$\smash{\CCprojector\superscr{(r,t);(s)}\sub{r,t} } (\VecSp\sub{r,t}) \cong \Wd\sub{s}$ by 
\begin{align}
\CCprojector\superscr{(r,t);(s)}\sub{r,t} \colon \VecSp\sub{r,t} \longrightarrow \VecSp\sub{r,t} , \qquad \qquad
\CCprojector\superscr{(r,t);(s)}\sub{r,t} \big( F^\ell.\HWvec\sub{r,t}\super{p} \big) :=\delta_{p,s} F^\ell.\HWvec\sub{r,t}\super{s} 
\qquad \text{for} \quad \text{$\ell \in \{0, 1, \ldots, s\}$.}
\end{align}

\begin{restatable}{prop}{GeneratorThmComm} \label{GeneratorThmComm}
Suppose $\Summed_\multii < \pmin(q)$.  Then, the commutant algebra $\EndMod{\Uqsltwo} \VecSp_\multii$
is generated by the collection of all submodule projectors that act strictly on consecutive pairs of 
tensorands of vectors in $\VecSp_\multii$\textnormal{:}
\begin{align} \label{GeneratorsComm} 
\EndMod{\Uqsltwo} \VecSp_\multii 
= \big\langle \CCprojector\sub{\sIndex_i,\sIndex_{i+1}}\superscr{(\sIndex_i,\sIndex_{i+1});(s)} \, \; \big| \, \;
s \in \DefectSet\sub{\sIndex_i,\sIndex_{i+1}}, \, i \in \{1, 2, \ldots, \np_\multii - 1\} \big\rangle ,
\end{align} 
where $\DefectSet\sub{\sIndex_i,\sIndex_{i+1}} = \{ \, |\sIndex_{i+1} - \sIndex_{i}|, \, |\sIndex_{i+1} - \sIndex_{i}| + 2, \, \ldots, \, \sIndex_{i+1} + \sIndex_{i} \}$.
\end{restatable}

\begin{proof}
We prove this (known) result in section~\ref{GeneratorThmCommSubSec}, 
identifying the projectors with valenced diagrams, cf.~\eqref{SentIntro}.
\end{proof}

The above fact
is no surprise.
Indeed, the quantum group is quasi-triangular with R-matrix which gives 
an action of the braid group on tensor products of $\Uqsltwo$-modules, commuting with the $\Uqsltwo$-action.
This braiding can be related to the projectors $\smash{\CCprojector\sub{\sIndex_i,\sIndex_{i+1}}\superscr{(\sIndex_i,\sIndex_{i+1});(s)}}$
(although the explicit connection with general higher-spin representations in $\VecSp_\multii$ is somewhat complicated).
Such braid representations cannot be faithful (i.e., injective), 
as the group algebra of the braid group is infinite-dimensional. 
Alternatively, the commutant algebra is also isomorphic to a quotient of the Hecke algebra, that is,
the $q$-deformation of the symmetric group algebra. 
This quotient is an explicit planar diagram algebra, the valenced Temperley-Lieb algebra $\TL_\multii(\nu)$, discussed shortly. 

\begin{figure}
\includegraphics[scale=0.275]{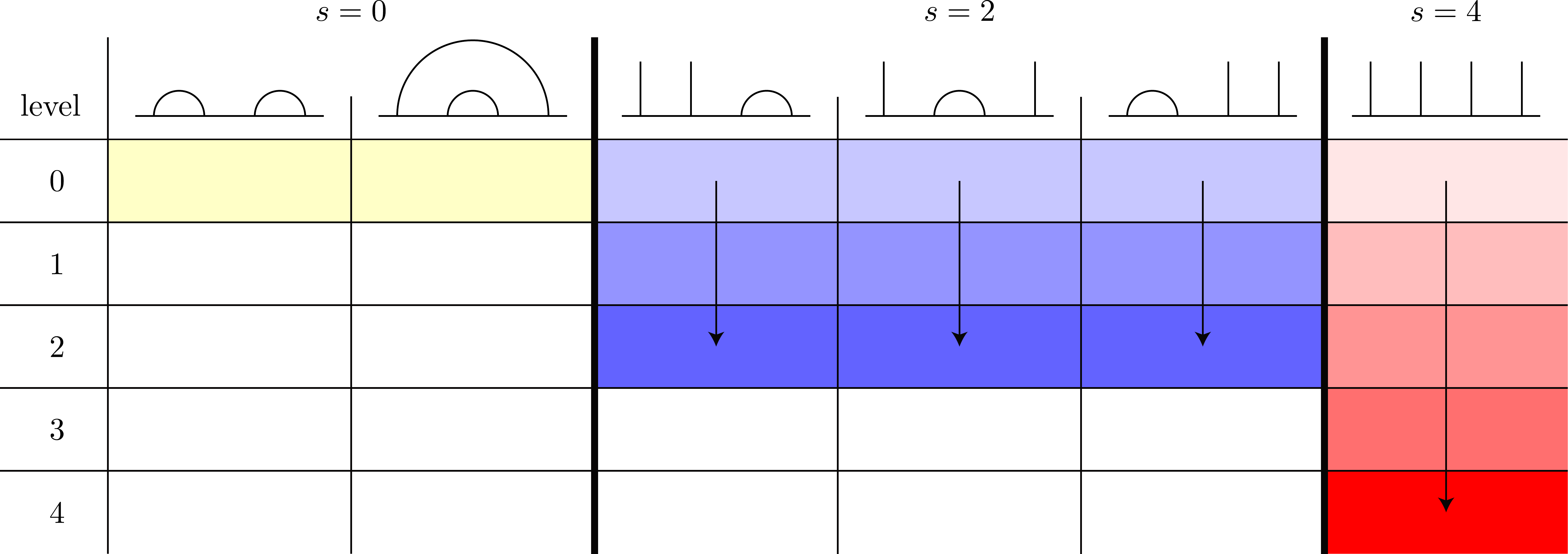}
\caption{
Illustration of the quantum Schur-Weyl duality~\eqref{QSWEx}. 
In the rows, adjacent boxes with the same 
coloring form standard modules $\smash{\LS_n^{(s)}}$, closed under the action of $\TL_n(\nu)$. 
The columns form simple type-one modules $\Wd_{(s)}$, closed under the action of $\Uqsltwo$. 
Arrows indicate the action of $F \in \Uqsltwo$, which maps from one copy of $\smash{\LS_n^{(s)}}$ to another.
}
\label{Qsw1}
\end{figure}

For illustration, let us briefly study the case of 
the $n$:th tensor power of the fundamental $\Uqsltwo$-module $\Wd\sub{1}$,
\begin{align} \label{TensorPowerN}
\Module{\VecSp_n}{\Uqsltwo} = \; & \Wd\sub{1}^{\otimes n} , \qquad \text{assuming that} \quad n < \pmin(q) .
\end{align}
In this case, the submodule projectors
\begin{align} \label{TLAsProjectors}
\pi_j := \id^{\otimes(j-1)} \otimes \CCprojector\sub{1,1}\superscr{(1,1); (0)} \otimes \id^{\otimes(n-j-1)} \; \in \; \EndMod{\Uqsltwo}{\VecSp_n} ,
\end{align}
projecting the $j$:th and $(j+1)$:st tensor components onto the trivial module $\Wd\sub{0}$,
satisfy the relations
\begin{align} 
\nonumber
\pi_j \smash{\pi_{j \pm 1}} \pi_j &= \frac{1}{[2]^2}
\,  \pi_j, \quad \quad \; \text{if $1 \leq i\pm1 \leq n-1$} \\ 
 \nonumber
\pi_j^2 &= \pi_j, && \\
\pi_i \pi_j &= \pi_j \pi_i, \qquad \quad \text{if $|i-j| > 1$}.
\end{align}
After normalizing them differently, we see that these relations are exactly the famous Temperley-Lieb relations.
Indeed, the operators $\nu \pi_j$, where we parameterize the \emph{loop fugacity} $\nu$ in terms of $q \in \bC^\times$ as
\begin{align} \label{fugacity} 
\nu = -q - q^{-1} = - [2] \in \bC 
\end{align} 
(the case of $q=1$ corresponding to representations of the Lie algebra $\mathfrak{sl}_2$ discussed in appendix~\ref{ClassicalApp})
satisfy the same relations
as the generating set $\{ \mathbf{1}_{\TL_n}, \Gen_1, \Gen_2, \ldots, \Gen_{n-1} \}$ of
the \emph{Temperley-Lieb algebra} $\TL_n(\nu)$~\cite{tl, vj}:
\begin{alignat}{2} 
\nonumber
\Gen_i \Gen_{i \pm 1} \Gen_i &= \Gen_i, \qquad  &&\text{if $1 \leq i\pm1 \leq n-1$}, \\ 
\nonumber
\Gen_i^2 &= \nu \Gen_i, \qquad && \\
\label{WordRelations}
\Gen_i \Gen_j &= \Gen_j \Gen_i, \qquad  &&\text{if $|i-j| > 1$}
\end{alignat} 
for all $i,j \in \{ 1, 2, \ldots, n - 1 \}$.
In section~\ref{TLReviewSec}, we give a diagram presentation~(\ref{LRtoGen},~\ref{LRtoUnit}) for these 
generators.

Conversely, one can show that the image of the algebra $\Uqsltwo$ under its representation on $\Module{\VecSp_n}{\Uqsltwo}$
coincides with the commutant algebra $\EndMod{\TL} \VecSp_n$ of the action of the Temperley-Lieb algebra $\TL_n(\nu)$ on $\CModule{\VecSp_n}{\TL}$.
This relationship is a well-known fact, observed by M.~Jimbo~\cite{mj2} and independently by R.~Dipper and G.~James~\cite{dj}.  
In fact, this property is an easy consequence of the semisimple structure of the module $\Module{\VecSp_n}{\Uqsltwo}$,
essentially following only from Schur's lemma (see proposition~\ref{DoubleCommCor} in appendix~\ref{DCApp}). 
A slightly more elaborate but still elementary application of  Schur's lemma (theorem~\ref{DoubleMainTheorem} in appendix~\ref{DCApp})
gives a duality decomposition for the bimodule structure of $\VecSp_n$ under both actions of $\Uqsltwo$ and $\TL_n(\nu)$. 
This is known as the ``quantum Schur-Weyl duality'' decomposition 
(see Figure~\ref{Qsw1})
\begin{align}  \label{QSWEx}
\tag{q-SW$_n$}
\Wd\sub{1}^{\otimes n}
\isom 
\bigoplus_{s \, \in \, \DefectSet_n} \LS_n\super{s} \otimes \Wd\sub{s}  ,
\qquad \qquad \textnormal{when} \qquad n < \pmin(q) , 
\end{align} 
where 
$\smash{\{ \LS_n\super{s} \, | \, s \in \DefectSet_n \}}$ is the complete set of simple $\TL_n(\nu)$-modules. 
Theorem~\ref{HighQSchurWeylThm2}, discussed shortly, implies this and a more refined statement
including a complete description of the general type-one tensor product module $\Module{\VecSp_\multii}{\Uqsltwo}$.

Let us also remark that if $n \geq \pmin(q)$, the action of the Temperley-Lieb algebra on $\VecSp_n$ still admits 
a quantum Schur-Weyl duality, but its commutant algebra  
is strictly larger than just the image of the quantum group $\Uqsltwo$, 
the latter being extended by additional generators called ``divided powers"~\cite{Lusz, ppm, mma}.
The direct-sum decomposition of $\CModule{\VecSp_n}{\TL}$ into indecomposable modules is also known~\cite{rs2, gv, psa},
see also~\cite{BFGT09, GST14}.


More generally (assuming $\Summed_\multii < \pmin(q)$), 
the commutant algebra $\EndMod{\Uqsltwo}{\VecSp_\multii}$ 
on the type-one tensor product module~\eqref{PreSpinChain}
is also a diagram algebra, the \emph{valenced Temperley-Lieb algebra} $\TL_\multii(\nu)$, 
generated by the valenced tangles~\cite{fp3a}
\begin{align}  
\label{TL_valenced_Unit}
\hphantom{\vcenter{\hbox{\includegraphics[scale=0.275]{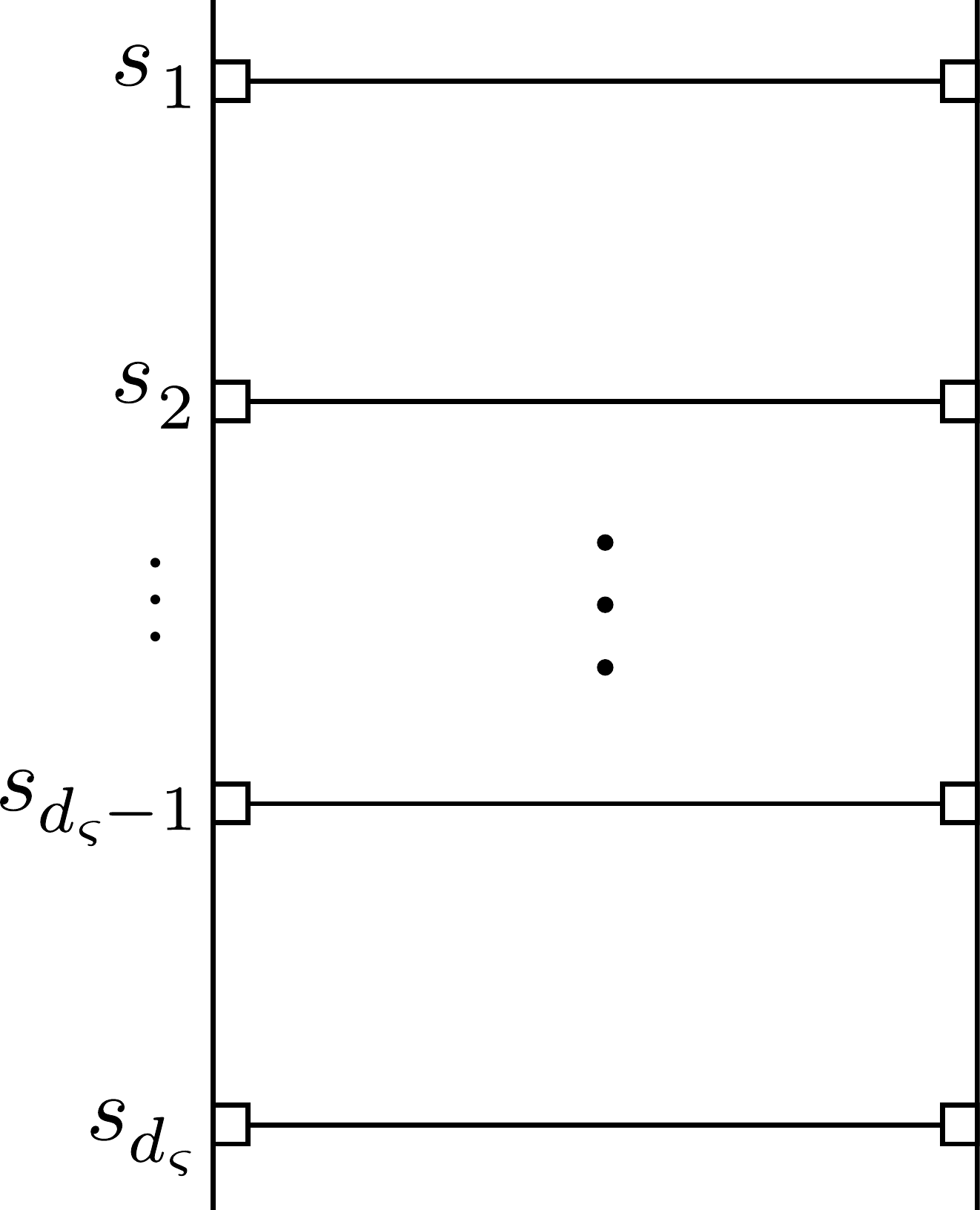}}}}
\mathbf{1}_{\TL_\multii} \quad & =
&& \vcenter{\hbox{\includegraphics[scale=0.275]{Figures/e-CompositeProjector_valenced.pdf}}} \quad \in \TL_\multii(\nu) \\[1em]
\label{TL_valenced_Gen}
\hphantom{\vcenter{\hbox{\includegraphics[scale=0.275]{Figures/e-CompositeProjector_valenced.pdf}}}}
\ValGen_i \quad & =
&& \vcenter{\hbox{\includegraphics[scale=0.275]{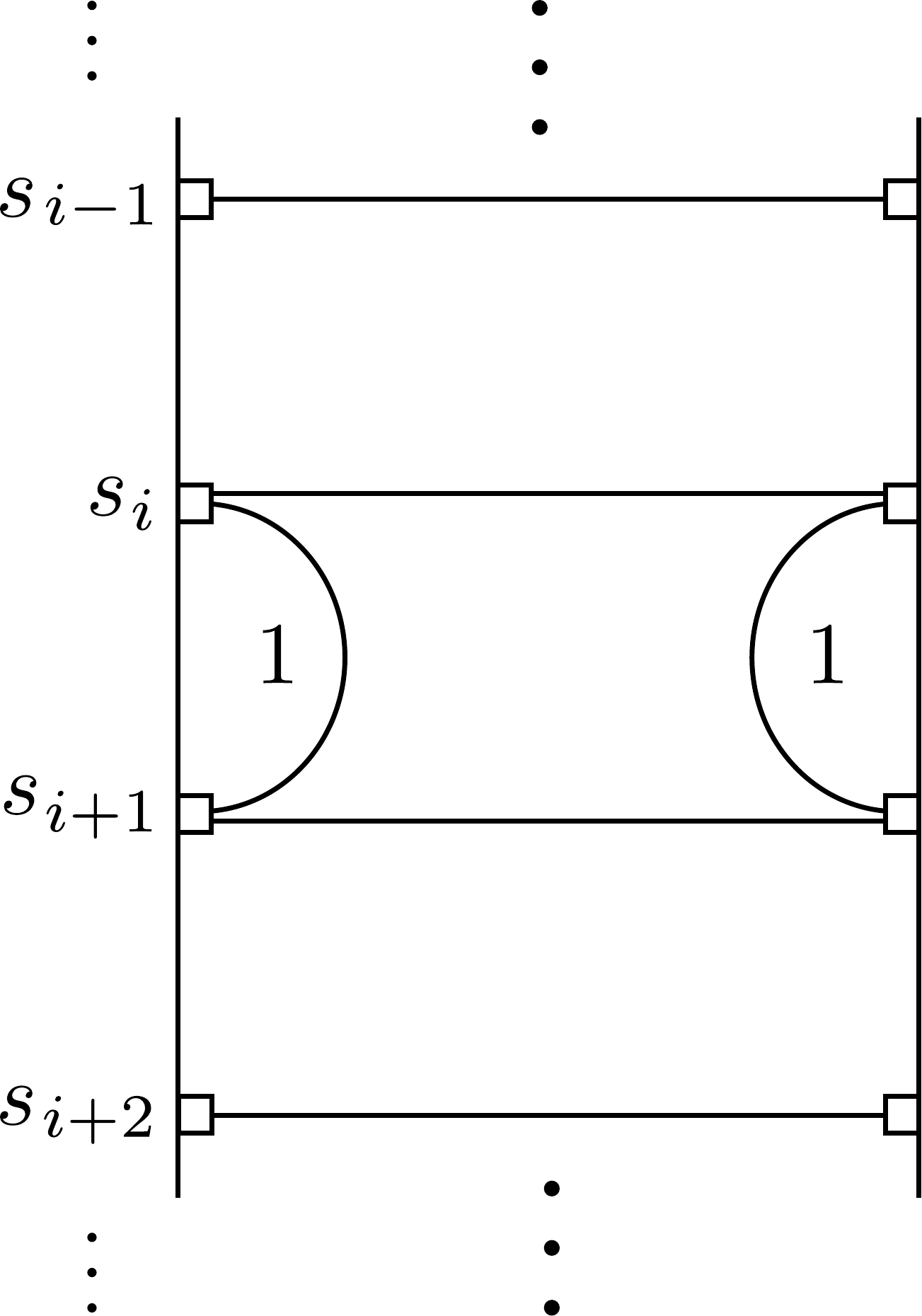}}} \quad \in \TL_\multii(\nu) 
\quad \text{for all $i \in \{ 1, 2, \ldots, \np_\multii - 1 \}$} ,
\end{align} 
where~\eqref{TL_valenced_Unit} is the unit, and the ``valenced nodes'' stand for multiple strands emerging from the same point:
\begin{align} 
\vcenter{\hbox{\includegraphics[scale=0.275]{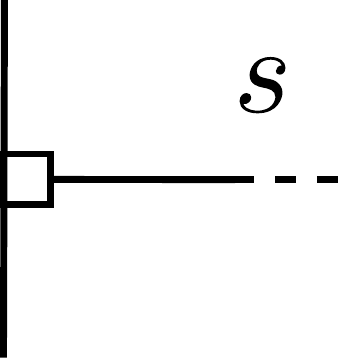}}} 
\qquad \qquad \Longleftrightarrow \qquad \qquad 
\begin{rcases}
\vcenter{\hbox{\includegraphics[scale=0.275]{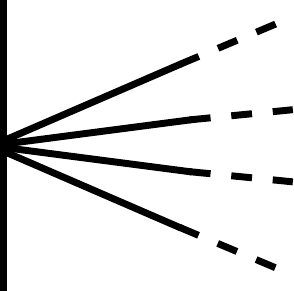}}} \quad
\end{rcases} \; s .
\end{align} 
$\TL_\multii(\nu)$ is an associative unital algebra consisting of valenced diagrams with multiplication defined in~(\ref{ExmpleConcat},~\ref{TangleHom}) in section~\ref{TLReviewSec}, using the Jones-Wenzl projector~\cite{vj, hw}.
The special case of $\multii = (1,1,\ldots,1)$ is the ordinary Temperley-Lieb algebra $\TL_n(\nu)$.
The restriction to $\Summed_\multii < \pmin(q)$ should not be essential for 
$\{\mathbf{1}_{\TL_\multii}, \ValGen_1, \ValGen_2, \ldots, \ValGen_{\np_\multii-1}\}$ to be the whole generating set of $\TL_\multii(\nu)$,
and the algebra $\TL_\multii(\nu)$ itself is defined whenever $\max \multii < \pmin(q)$~\cite{fp3a, fp3b}.

The valenced Temperley-Lieb algebra has an action on the vector space $\VecSp_\multii$, 
which coincides with the action of projectors as in proposition~\ref{GeneratorThmComm}.
See also Figure~\ref{Qsw2}. 
Indeed, an explicit isomorphism from $\TL_\multii(\nu)$ to~\eqref{GeneratorsComm} is given~by
\begin{align} \label{SentIntro} 
\frac{(-1)^s [s + 1]}{\ThetaNet(\sIndex_i, \sIndex_{i+1}, s)} \,\, \times \,\, \vcenter{\hbox{\includegraphics[scale=0.275]{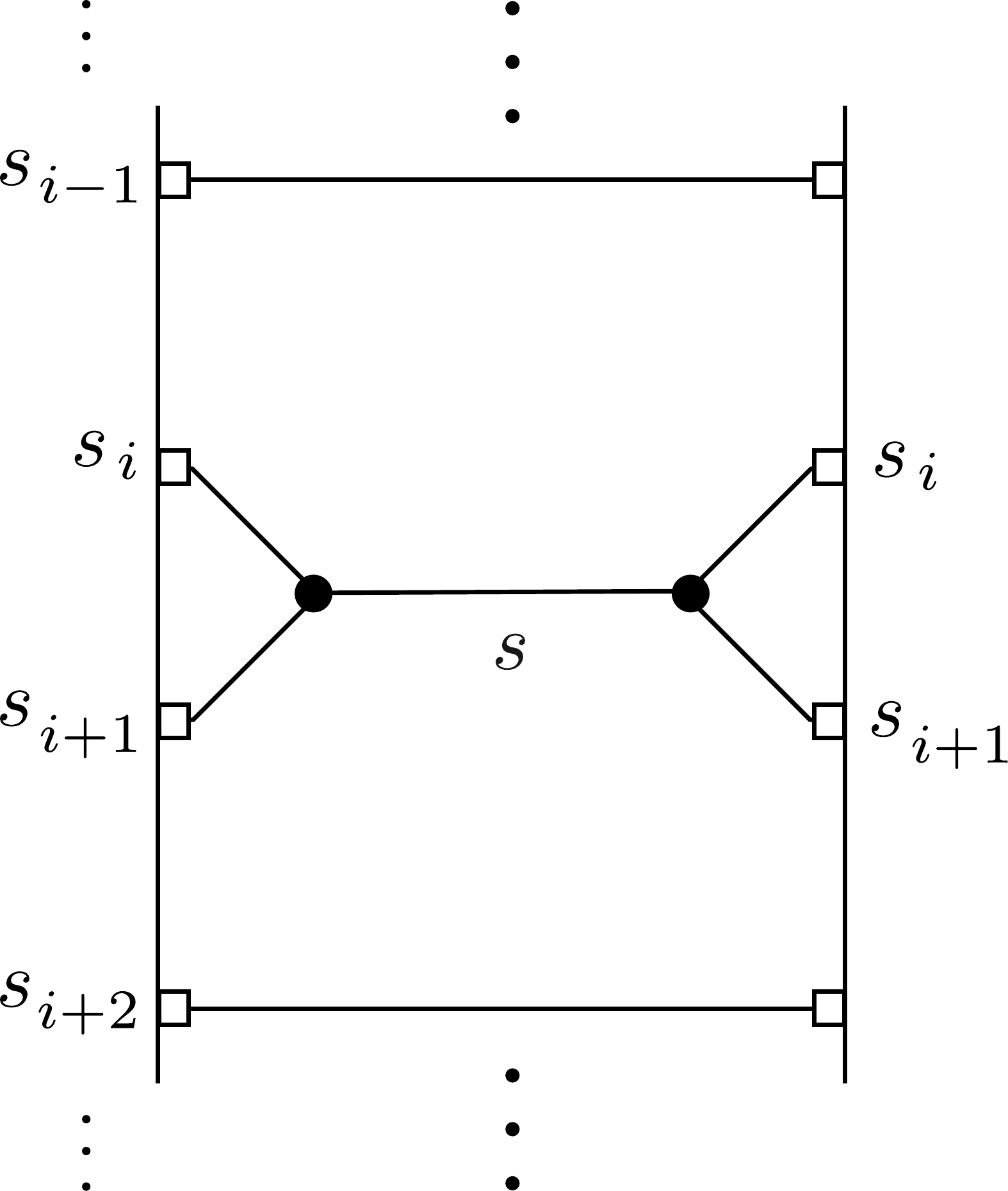}}} 
\quad \qquad \qquad \longmapsto \qquad \qquad 
\CCprojector\sub{\sIndex_i,\sIndex_{i+1}}\superscr{(\sIndex_i,\sIndex_{i+1});(s)} ,
\end{align}
where $\ThetaNet$ is again given by the evaluation of the Theta network~\eqref{ThetaFormula},
and where the left side of~\eqref{SentIntro} is an alternative generating set for $\TL_\multii(\nu)$
by~\cite[proposition~\red{2.10}]{fp3a}, 
using Kauffman's three-vertex notation~\cite{kl, mv, cfs}
\begin{align}\label{3vertex1} 
\text{for $s \in \DefectSet\sub{r,t}$} , \qquad 
\vcenter{\hbox{\includegraphics[scale=0.275]{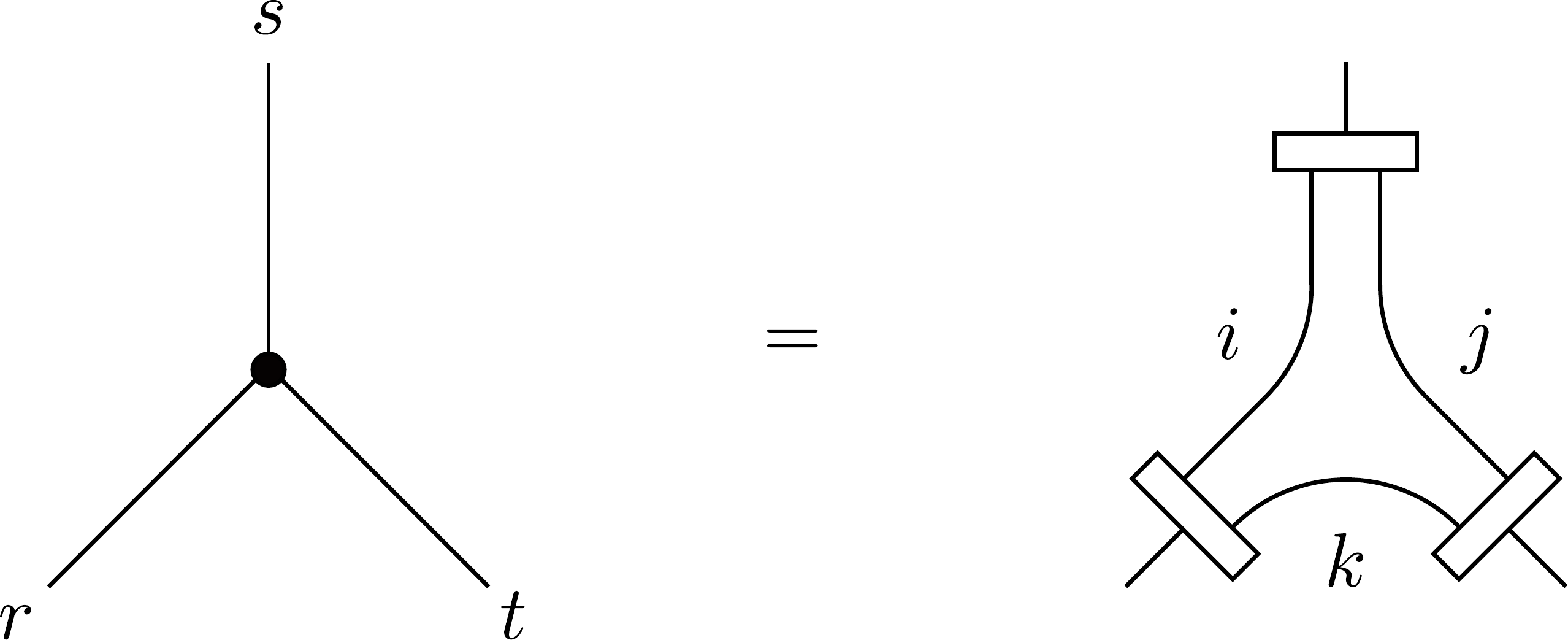} ,}} \qquad  \qquad \qquad
\begin{aligned} 
i & = \frac{r + s - t}{2} , \\[.7em] 
j & = \frac{s + t - r}{2} , \\[.7em] 
k & = \frac{t + r - s}{2} ,
\end{aligned}
\end{align}
where the box represents the Jones-Wenzl projector, as detailed in section~\ref{TLReviewSec}.
Identification~\eqref{SentIntro} is a consequence of the graphical representation of the $\Uqsltwo$-submodule projectors
from proposition~\ref{TLProjectionLem3}  in section~\ref{GraphicalProjSect} (cf. lemma~\ref{GeneratorLemTL}).
Meanwhile, the action of the algebra $\TL_\multii(\nu)$ on $\CModule{\VecSp_\multii}{\TL}$ is introduced and investigated in detail in section~\ref{DiagAlgSect}.

Now, as the vector space $\VecSp_\multii$ is endowed with actions of both algebras $\Uqsltwo$ and $\TL_\multii(\nu)$,
with $\nu = -q-q^{-1}$, we denote the obtained module as
\begin{align}
\Wd_\multii := \BIModule{\VecSp_\multii}{\Uqsltwo}{\TL} , \qquad \qquad
\text{with representations} \qquad 
\Trep_\multii \colon \TL_\multii(\nu) \longrightarrow \End \VecSp_\multii ,
\qquad 
\LeftRegRep_\multii \colon \Uqsltwo \longrightarrow \End \VecSp_\multii .
\end{align}

Recalling from proposition~\ref{DirectSumPropCombination} that the valenced link patterns $\alpha$  
give rise to an explicit direct-sum decomposition~(\ref{MoreGenDecompIntro},~\ref{DirectSumInclusion2Intro}) 
with basis vectors having suggestive diagram representations, it is no surprise that 
the valenced Temperley-Lieb diagram algebra $\TL_\multii(\nu)$ acts naturally on the spaces $\smash{\HWsp_\multii\super{s}}$
of $\Uqsltwo$-highest-weight vectors~\eqref{HWWSpIntro}.
These spaces are graded by the $K$-eigenvalues $q^s$, 
or alternatively, by a geometric quantity for the link patterns, namely the number $s$ of  ``defects."
We denote by $\smash{\LP_\multii\super{s}}$ the set of valenced link patterns with $s$ defects and set 
$\smash{\LS_\multii\super{s}} := \Span \smash{\LP_\multii\super{s}}$. 
These spaces also carry a diagram action of the algebra $\TL_\multii(\nu)$, explicated in section~\ref{TLReviewSec}.
Each $\smash{\LS_\multii\super{s}}$ is called a $\TL_\multii(\nu)$-\emph{standard module}.
(In fact, $\TL_\multii(\nu)$ is a cellular algebra~\citep{gl, gl2, fp3b}, and these are its cell modules.)

\begin{restatable}{theorem}{SecondHighQSchurWeylThm} \label{HighQSchurWeylThm2}
\textnormal{(Higher-spin quantum Schur-Weyl duality):} 
Suppose $\Summed_\multii < \pmin(q)$.   Then, the following hold:
\begin{enumerate} 
\itemcolor{red}
\item \label{HQsw21} 
The images of the maps 
$\Trep_\multii \colon \TL_\multii(\nu) \longrightarrow \End \VecSp_\multii$ 
and $\LeftRegRep_\multii \colon \Uqsltwo \longrightarrow \End \VecSp_\multii$ 
are semisimple algebras, 
which equal
\begin{align} \label{TLandUQdoubleCommutants}
\TL_\multii(\nu)
\isom \Trep_\multii ( \TL_\multii(\nu) )
= \EndMod{\Uqsltwo} \VecSp_\multii 
\qquad\qquad \textnormal{and} \qquad \qquad
\LeftRegRep_\multii(\Uqsltwo) = 
\EndMod{\TL} \VecSp_\multii  .
\end{align}

\item \label{HQsw22}
The collections $\smash{\{ \Wd\sub{s} \, | \, s \in \DefectSet_\multii \}}$ and 
$\smash{\{ \LS_\multii\super{s} \, | \, s \in \DefectSet_\multii \}}$
are respectively the complete sets of simple non-isomorphic $\LeftRegRep_\multii(\Uqsltwo)$-modules
and $\TL_\multii(\nu)$-modules, and we have the direct-sum decomposition
\begin{align} \label{GenDecompWJ} \tag{q-SW$_\multii$}
\Wd_\multii \isom 
\bigoplus_{s \, \in \, \DefectSet_\multii} \Wd\sub{s} \otimes \LS_\multii\super{s} .
\end{align} 

\item \label{HQsw33}
The linear extension of the following map gives an explicit isomorphism for~\eqref{GenDecompWJ}\textnormal{:}
with $\Sing_\alpha$ the explicit $\Uqsltwo$-highest-weight vectors constructed in definition~\ref{SingletBasisDefinition},
$\alpha \in \smash{\LP_\multii\super{s}}$, $s \in \DefectSet_\multii$, and $\ell \in \{0, 1, \ldots, s\}$, 
\begin{align} 
F^\ell.\Sing_\alpha \longmapsto \Basis_\ell\super{s} \otimes \alpha .
\end{align}

\end{enumerate}
\end{restatable}

\begin{proof}
We prove this result in section~\ref{GeneratorThmCommSubSec}. It is a consequence of the double-commutant theorem
(theorem~\ref{DoubleMainTheorem}, appendix~\ref{DCApp}) combined with explicit knowledge of the structure 
of the bi-module $\Wd_\multii$ from proposition~\ref{DirectSumPropCombination} for the $\Uqsltwo$-structure 
and proposition~\ref{TLPropCombination} (see below) for the $\TL_\multii(\nu)$-structure.
In addition, we need a dimension count and an injectivity result, 
which shows that $\TL_\multii(\nu) \cong \Trep_\multii ( \TL_\multii(\nu) )$ indeed constitutes the whole commutant algebra. 
\end{proof}

\begin{prop} \label{FaithfulPropIntro}
\textnormal{[Special case of proposition~\ref{PreFaithfulPropGen}]:}
The representation $\Trep_\multii \colon \TL_\multii(\nu) \longrightarrow \End \VecSp_\multii$ is faithful. 
\end{prop}

Interestingly, the above proposition applies not only in the (semisimple) case of $\Summed_\multii < \pmin(q)$
but in the complete range $\max \multii < \pmin(q)$ where the algebra $\TL_\multii(\nu)$ is defined.
As a special case, we recover the fact that the action of the ordinary Temperley-Lieb algebra $\TL_n(\nu)$
on the $n$:th tensor power~\eqref{TensorPowerN} of the fundamental $\Uqsltwo$-module $\Wd\sub{1}$ is faithful for all values of $q$,
a result proved independently by P.~Martin~\cite[theorem~\red{1}]{ppm} and F.~Goodman and H.~Wenzl~\cite[theorem~\red{2.4}]{gwe}.
The faithfulness is crucial in order to identify $\TL_\multii(\nu)$ as the commutant algebra in item~\ref{HQsw21}
of theorem~\ref{HighQSchurWeylThm2} --- otherwise, the commutant algebra could be a complicated quotient of $\TL_\multii(\nu)$.

\bigskip

Next, we briefly consider the more general space $\HomMod{\Uqsltwo} (\VecSp_\multii, \VecSp_\multiii)$ of homomorphisms of 
$\Uqsltwo$-modules between pairs of type-one modules $\Module{\VecSp_\multii}{\Uqsltwo}$ and $\Module{\VecSp_\multiii}{\Uqsltwo}$
as in~\eqref{PreSpinChain}, with 
\begin{align} 
\label{MultiindexNotation}
\multii = (\sIndex_1, \sIndex_2,\ldots, \sIndex_{\np_\multii}) \in \bZpos^{\np_\multii} 
\qquad \qquad \text{and} \qquad \qquad 
\multiii = (p_1, p_2, \ldots, p_{\np_\multiii}) \in \bZpos^{\np_\multiii} , \\ 
\label{ndefn} 
\Summed_\multii := \sIndex_1 + \sIndex_2 + \cdots + \sIndex_{\np_\multii}
\qquad \qquad \text{and} \qquad \qquad 
\Summed_\multiii := p_1 + p_2 + \cdots + p_{\np_\multiii} ,
\end{align}
such that $\max (\multii, \multiii) < \pmin(q)$ 
and $\Summed_\multii + \Summed_\multiii = 0 \Mod 2$.
It turns out that elements in the spaces $\HomMod{\Uqsltwo} (\VecSp_\multii, \VecSp_\multiii)$ 
can be realized as $(\multii, \multiii)$-tangles,
slightly more general than the elements of $\TL_\multii(\nu)$, which are just $(\multii, \multii)$-tangles --- see section~\ref{TLReviewSec}.
We denote the space of $(\multii, \multiii)$-tangles  by $\TL_\multii^\multiii$. A generic element in this space looks like 
\begin{align} 
\vcenter{\hbox{\includegraphics[scale=0.275]{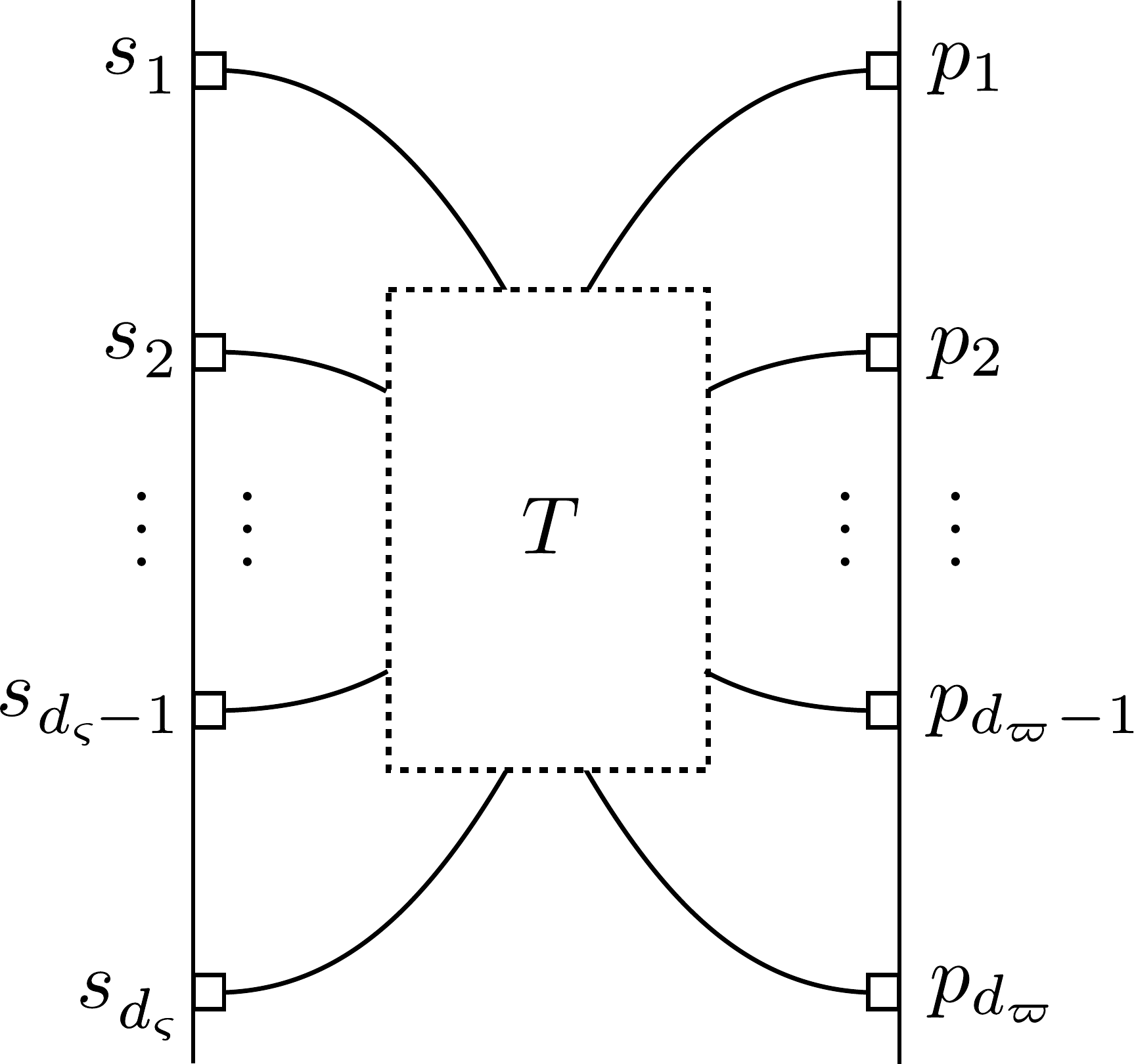} }}  \quad \in \quad \TL_\multii^\multiii ,
\end{align}
where $T$ in the dashed box is a planar network, explicated in sections~\ref{TLReviewSec} and~\ref{LSandBiformandSCGrapgSec}. 
In section~\ref{GenDiacActTypeOneSec}, we define an action 
\begin{align}
\Trep_\multii^\multiii \colon \TL_\multii^\multiii \longrightarrow \Hom (\VecSp_\multii, \VecSp_\multiii) 
\end{align}
of the $(\multii, \multiii)$-tangles, which actually commutes with the $\Uqsltwo$-action. 
Moreover, when $\max(\Summed_\multii, \Summed_\multiii) < \pmin(q)$,
these elements constitute the whole space $\HomMod{\Uqsltwo} (\VecSp_\multii, \VecSp_\multiii)$,
analogously to theorem~\ref{HighQSchurWeylThm2}.
For instance, the valenced tangles
\begin{align}
\BarAction \alpha \;\; \ProjBox \;\; \beta\superscr{\cheque} \BarAction 
\qquad \overset{\Trep_\multii^\multiii}{\longmapsto} \qquad  \CCprojector_\alpha^{\beta} 
\; \in \; \HomMod{\Uqsltwo} ( \VecSp_\multii , \VecSp_\multiii ) ,
\end{align}
where ``$\ProjBox$" is a Jones-Wenzl projector box placed across all through-paths of the tangle, 
give natural elements in this space, again indexed by valenced link patterns $\alpha, \beta$
as explained in lemma~\ref{ThisLemma} in section~\ref{GraphicalProjSect}.

\begin{restatable}{theorem}{GeneralCommutantThm} \label{GeneralCommutantThm}
\textnormal{(Double-commutant property):} 
Suppose $\max(\Summed_\multii, \Summed_\multiii) < \pmin(q)$.
Then, the following hold:
\begin{enumerate}
\itemcolor{red} 
\item \label{GeneralCommutantThmItem1}
Let $L ,R \in \Hom (\VecSp_\multiii , \VecSp_\multii)$. The diagram
\begin{equation} \label{1stThmCommute}
\begin{tikzcd}[ampersand replacement=\&, column sep=2cm, row sep=1.5cm]
\VecSp_\multiii \arrow{r}{\LeftRegRep_\multiii(x)} \arrow{d}[swap]{L}
\& \arrow{d}{R} \VecSp_\multiii \\ 
\VecSp_\multii \arrow{r}{\LeftRegRep_\multii(x)}
\& \VecSp_\multii
\end{tikzcd}
\end{equation}
commutes for all elements $x \in \Uqsltwo$ if and only if we have 
$L = R = \Trep_\multii^\multiii(T)$ for some valenced tangle $T \in \TL_\multii^\multiii$.

\item \label{GeneralCommutantThmItem2}
Let $L \in \End \VecSp_\multiii$ and $R \in \End \VecSp_\multii$. 
Let $\smash{\CCprojector_\multiii\super{s} \in \End \VecSp_\multiii}$ and 
$\smash{\CCprojector_\multii\super{s} \in \End \VecSp_\multii}$ 
be the respective projections onto the $s$:th summand of the direct-sum decompositions~\eqref{MoreGenDecompIntro} 
of $\VecSp_\multiii$ and $\VecSp_\multii$.
The diagram
\begin{equation} \label{2ndThmCommute}
\begin{tikzcd}[ampersand replacement=\&, column sep=2cm, row sep=1.5cm]
\VecSp_\multiii \arrow{r}{\Trep^\multiii_\multii(T)} \arrow{d}[swap]{L}
\& \arrow{d}{R} \VecSp_\multii \\ 
\VecSp_\multiii \arrow{r}{\Trep^\multiii_\multii(T)}
\& \VecSp_\multii
\end{tikzcd}
\end{equation}
commutes for all valenced tangles $T \in \TL_\multii^\multiii$ if and only if we have
\begin{align} 
L \,\, = \,\, \LeftRegRep_\multiii(x) \,\, + \sum_{ s \, \in \, \DefectSet_\multiii \, \setminus \, \DefectSet_\multii  } \CCprojector_\multiii\super{ s } \circ L'
\qquad\qquad \textnormal{and} \qquad\qquad
R \,\, = \,\, \LeftRegRep_\multii(x) \,\, + \sum_{ s \, \in \, \DefectSet_\multii \, \setminus \, \DefectSet_\multiii } R\,'  \,\circ \CCprojector_\multii\super{ s } 
\end{align}
for some element  $x \in \Uqsltwo$ and endomorphisms $L' \in \End \VecSp_\multiii$ and $R' \in \End \VecSp_\multii$.
\end{enumerate}
\end{restatable}

\begin{proof}
We prove this result in section~\ref{DCProofSec}. 
It is a consequence of a version of the double-commutant theorem (proposition~\ref{DoubleCommGenProp}, appendix~\ref{DCApp}),
and again, a dimension count and injectivity for $\smash{\Trep_\multii^\multiii}$,
namely proposition~\ref{PreFaithfulPropGen} in section~\ref{KerImSubSec}.
\end{proof}

Diagram representations for homomorphisms between tensor powers of type~\eqref{TensorPowerN}
of the fundamental module of the classical Lie algebra $\mathfrak{sl}_2$
have been known for a long time~\cite{rtw, pen, cfs}. 
In the semisimple case, this calculus generalizes to the case of $\Uqsltwo$ and its fundamental modules.
These diagrams are called tangles in the ``Temperley-Lieb category"~\cite{tl, vj, vt, kl, cfs}.
However, the general, valenced tangles as maps between arbitrary type-one tensor product $\Uqsltwo$-modules $\Module{\VecSp_\multii}{\Uqsltwo}$
seem to be less explicitly investigated, although implicit in the literature. 
The main appeal of theorem~\ref{GeneralCommutantThm} is that it follows from only a few elementary ingredients:
Schur's lemma, leading to general double-commutant tools
(completely elementary but still included in appendix~\ref{DCApp} for the sake of exposition),
and injectivity of the map $\smash{\Trep_\multii^\multiii}$, 
to recognize the commutant algebra.

\subsection{Interlude --- the classical case}

\begin{figure}
\includegraphics[scale=0.275]{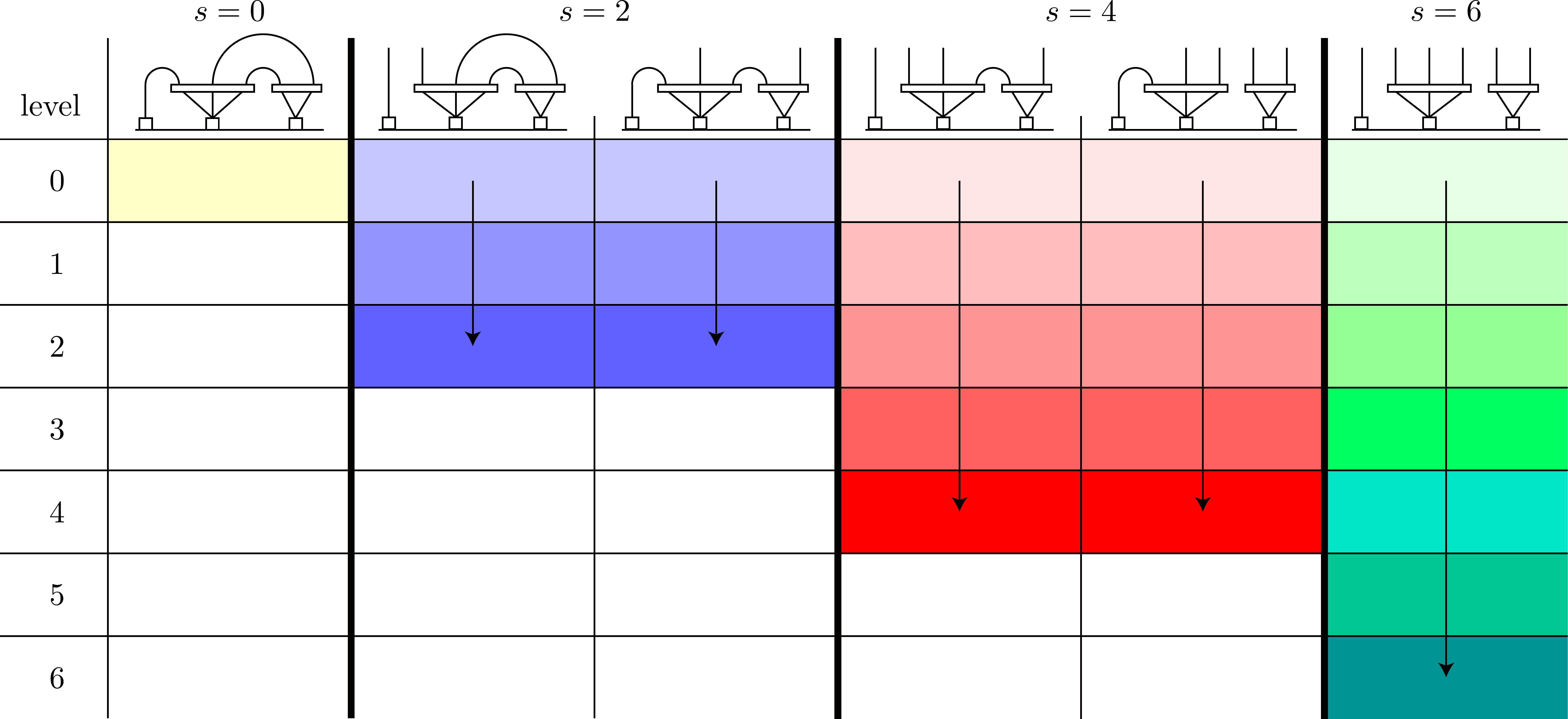}
\caption{
Illustration of the higher-spin quantum Schur-Weyl duality~\eqref{GenDecompWJ}.
In the rows, adjacent boxes with the same 
coloring form  standard modules  $\smash{\LS_\multii^{(s)}}$, closed under the action of $\TL_\multii(\nu)$.
The columns form simple type-one modules $\Wd_{(s)}$, closed under the action of $\Uqsltwo$. Arrows indicate the action of $F \in \Uqsltwo$, 
which maps from one copy of $\smash{\LS_\multii^{(s)}}$ to another.
The wide boxes represent Jones-Wenzl projectors and the square nodes represent valenced nodes.
}
\label{Qsw2}
\end{figure}

In the case $q=1$, the above results do not make sense as such. 
Instead, considering the universal enveloping algebra $\Usltwo := U(\mathfrak{sl}_2)$ of the classical Lie algebra $\mathfrak{sl}_2(\bC)$,
with generators $E$, $F$, and $H$ and relations
\begin{align} 
[H, E] = 2 E , \qquad 
[H, F] = -2 F , \qquad 
[E ,F ] = H ,
\end{align} 
we recover a generalized version of the classical Schur-Weyl duality~\cite{schur, hew}.

\begin{restatable}{theorem}{SecondHighCSchurWeylThm} \label{HighCSchurWeylThm2} 
\textnormal{(Higher-spin Schur-Weyl duality):} The following hold:
\begin{enumerate} 
\itemcolor{red}
\item \label{HCsw21} 
The images of the maps 
$\Trep_\multii \colon \TL_\multii(-2) \longrightarrow \End \VecSp_\multii$ 
and $\LeftRegRep_\multii \colon \Usltwo \longrightarrow \End \VecSp_\multii$
are semisimple algebras, 
which equal
\begin{align} 
\TL_\multii(-2)
\isom \Trep_\multii ( \TL_\multii(-2) )
= \EndMod{\Usltwo} \VecSp_\multii 
\qquad\qquad \textnormal{and} \qquad \qquad
\LeftRegRep_\multii(\Usltwo) = 
\EndMod{\TL} \VecSp_\multii .
\end{align}

\item \label{HCsw22}

The collections $\smash{\{ \Wd\sub{s} \, | \, s \in \DefectSet_\multii \}}$ and 
$\smash{\{ \LS_\multii\super{s} \, | \, s \in \DefectSet_\multii \}}$
are respectively the complete sets of simple non-isomorphic $\LeftRegRep_\multii(\Usltwo)$-modules
and $\TL_\multii(-2)$-modules, and we have the direct-sum decomposition
\begin{align} \label{HCGenDecomp} \tag{SW$_\multii$}
\Wd_\multii \isom 
\bigoplus_{s \, \in \, \DefectSet_\multii} \Wd\sub{s} \otimes \LS_\multii\super{s} .
\end{align}

\item \label{HCsw33}
The linear extension of the following map gives an explicit isomorphism for~\eqref{HCGenDecomp}\textnormal{:}
with $\Sing_\alpha$ the explicit $\Usltwo$-highest-weight vectors constructed in definition~\ref{SingletBasisDefinition} with $q=1$,
$\alpha \in \smash{\LP_\multii\super{s}}$, $s \in \DefectSet_\multii$, and $\ell \in \{0, 1, \ldots, s\}$, 
\begin{align} 
F^\ell.\Sing_\alpha \longmapsto \Basis_\ell\super{s} \otimes \alpha .
\end{align}
\end{enumerate}
\end{restatable}

\begin{proof}
We summarize the proof for this result in appendix~\ref{ClassicalApp}.
\end{proof}

The other results discussed above (and below) also have obvious counterparts in the case of $q=1$ and $\nu = -2$.

\subsection{Structure under the valenced Temperley-Lieb algebra}

Lastly, we summarize results concerning the ``dual" structure of the module $\CModule{\VecSp_\multii}{\TL}$. 
Theorem~\ref{HighQSchurWeylThm2} already gives an explicit direct-sum decomposition for it when $\Summed_\multii < \pmin(q)$.
However, the module $\CModule{\VecSp_\multii}{\TL}$ is not always semisimple, and the failure of semisimplicity 
has gained increasing interest~\cite{ppm, mma, rs2, BFGT09, gv, psa, GST14, rsa, ALZ15}.
Non-semisimple cases should correspond to physically relevant applications to 
conformal field theory (minimal models, logarithmic phenomena), but are also of interest from representation theoretical point of view.

We recall from the above discussion that valenced link patterns $\alpha \in \smash{\LP_\multii\super{s}}$ with $s$ defects
form a finite set of cardinality $\smash{\Dim_\multii\super{s}}$. 
For example, the following valenced link pattern has $s=3$ defects (see section~\ref{TLReviewSec} for the definitions):
\begin{align} 
\vcenter{\hbox{\includegraphics[scale=0.275]{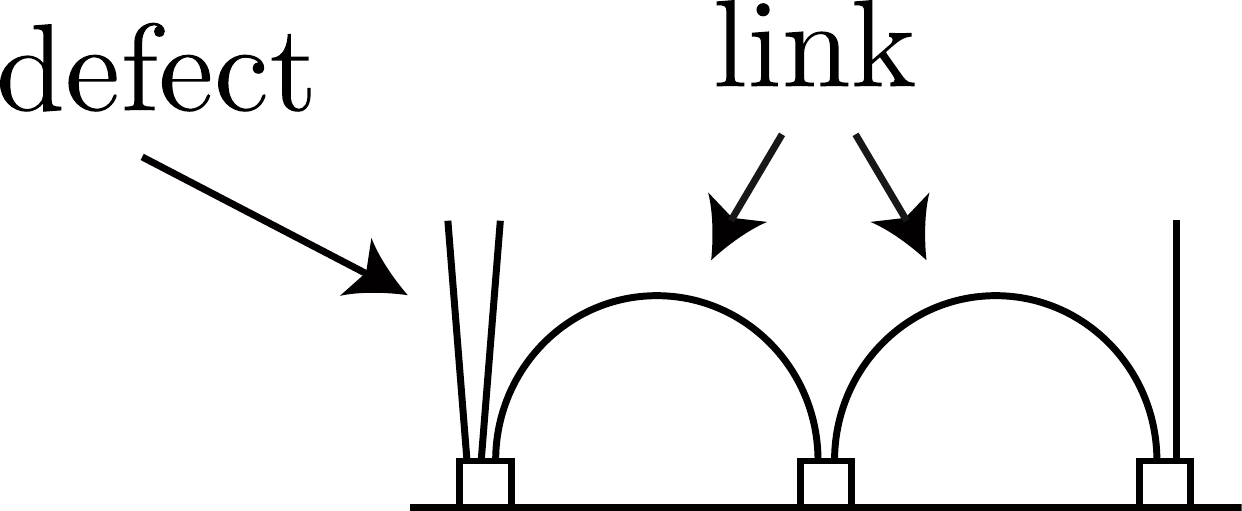} ,}} 
\end{align}
Diagram concatenation using the Jones-Wenzl projectors~\cite{vj, hw, kl}, as explained in section~\ref{TLReviewSec}, 
gives rise to a diagram action of the valenced Temperley-Lieb algebra $\TL_\multii(\nu)$ 
on formal linear combinations of valenced link patterns, standard modules (denoted $\smash{\LS_\multii\super{s}}$ for each $s$).
Proposition~\ref{DirectSumPropCombination} gives a map $\alpha \mapsto \Sing_\alpha$ 
sending each valenced link pattern to a highest weight vector in $\Module{\VecSp_\multii}{\Uqsltwo}$, 
and all these vectors $\Sing_\alpha$ are linearly independent. The diagram action of $\TL_\multii(\nu)$ on the link patterns $\alpha$
thus gives a natural action on the vectors $\Sing_\alpha$. In fact, this action coincides with the $\TL_\multii(\nu)$-action
appearing in the quantum Schur-Weyl duality (which is via projections as in proposition~\ref{GeneratorThmComm}).

When $\Summed_\multii < \pmin(q)$,  items~\ref{TLPropCombinationitem1} and~\ref{TLPropCombinationitem3} of the next proposition imply that 
$\CModule{\HWsp_\multii\super{s}}{\TL} \cong \smash{\LS_\multii\super{s}}$, i.e., the highest-weight vector spaces are isomorphic to the valenced 
Temperley-Lieb standard modules. 
Otherwise, there are additional highest-weight vectors not reachable via the valenced link-pattern basis  $\Sing_\alpha$,
due to degeneracies in the structure of the representation theory of $\Uqsltwo$~\cite{Lusz, ck, cp}. In this case, after quotienting out
highest-weight vectors which are orthogonal to the vectors $\Sing_\alpha$, we still establish that the quotient $\TL_\multii(\nu)$-modules 
thus obtained are isomorphic to simple $\TL_\multii(\nu)$-modules $\smash{\Quo_\multii\super{s}}$ obtained by 
taking the quotients of the standard modules $\smash{\LS_\multii\super{s}}$ by the radical $\smash{\rad \LS_\multii\super{s}}$ of 
a natural invariant $\TL_\multii(\nu)$-bilinear form on them (section~\ref{LSandBiformandSCGrapgSec}):
\begin{align}  
\Quo_\multii\super{s} := \LS_\multii\super{s} / \rad \LS_\multii\super{s}  
\qquad \qquad
\text{and}
\qquad \qquad
\Quo_\multii := 
\bigoplus_{s \, \in \, \DefectSet_\multii} \Quo_\multii\super{s} .
\end{align}

\begin{prop} \label{TLPropCombination}
\textnormal{(Link-state -- highest weight vector correspondence):} 
\begin{enumerate} 
\itemcolor{red}

\item \label{TLPropCombinationitem1}
If $\Summed_\multii < \pmin(q)$, then we have the following isomorphism of left $\TL_\multii(\nu)$-modules: 
\begin{align} 
\CModule{\VecSp_\multii}{\TL}
\isom
\bigoplus_{s \, \in \, \DefectSet_\multii} (s + 1) \, \CModule{\HWsp_\multii\super{s}}{\TL} .
\end{align}

\item \label{TLPropCombinationitem2} 
The linear map $\alpha \mapsto \Sing_\alpha$ 
induces the following embedding of $\TL_\multii(\nu)$-modules:
\begin{align} 
\LS_\multii := \bigoplus_{s \, \in \, \DefectSet_\multii} \LS_\multii\super{s}
\quad \lhook\joinrel\rightarrow \quad 
\bigoplus_{s \, \in \, \DefectSet_\multii} \CModule{\HWsp_\multii\super{s}}{\TL}
\subset \CModule{\HWsp_\multii}{\TL} .
\end{align}

\item \label{TLPropCombinationitem3} 
The linear map $\alpha \mapsto \Sing_\alpha$ 
induces the following embedding of $\TL_\multii(\nu)$-modules:
\begin{align} \label{ContainmentEmbIntro}
\bigoplus_{ \substack{s \, \in \, \DefectSet_\multii \\ s \, < \, \pmin(q)} } (s + 1) \, \LS_\multii\super{s}
\quad \lhook\joinrel\rightarrow \quad \CModule{\VecSp_\multii}{\TL} .
\end{align}
If $\Summed_\multii < \pmin(q)$, then this embedding~\eqref{ContainmentEmbIntro} is an isomorphism.

\item \label{TLPropCombinationitem4} 
The linear map $\alpha \mapsto \Sing_\alpha$ 
induces the following isomorphism of $\TL_\multii(\nu)$-modules:
\begin{align}  \label{QuotientCorItem3DirSumIntro}
\Quo_\multii \isom 
\frac{\CModule{\HWsp_\multii}{\TL}}{\CModule{\{ \Sing_\alpha \, | \, \alpha \in \LS_\multii \}\superscr{\perp}}{\TL}}
= \bigoplus_{s \, \in \, \DefectSet_\multii}
\frac{\CModule{\HWsp_\multii\super{s}}{\TL}}{\CModule{\{ \Sing_\alpha \, | \, \alpha \in \LS\super{s}_\multii \}\superscr{\perp}}{\TL}} 
\isom  \bigoplus_{s \, \in \, \DefectSet_\multii} \Quo_\multii\super{s} ,
\end{align}
where $\smash{\{ \cdots \}\superscr{\perp}}$ denotes the orthocomplement as detailed in section~\ref{subsec: link state hwv correspondence}
and appendix~\ref{LinAlgAppAux}.
\end{enumerate}
\end{prop}

\begin{proof}
Item~\ref{TLPropCombinationitem1} is a part of proposition~\ref{HWspaceDecTL} in section~\ref{GenDiacActTypeOneSec}. 
Items~\ref{TLPropCombinationitem2} and~\ref{TLPropCombinationitem3} are respectively 
parts of proposition~\ref{HWspLem2} and proposition~\ref{HWspacePropEmbAndIso} in section~\ref{subsec: link state hwv correspondence}. 
Lastly, item~\ref{TLPropCombinationitem4} is a part of proposition~\ref{QuotientProp} in section~\ref{QuotientRadicalSect}.
This last result~\eqref{QuotientCorItem3DirSumIntro} is the only part of the present work which is not self-contained,
relying on results in~\cite{fp3a}. 
\end{proof}

\bigskip

\begin{center}
\bf Organization of this article
\end{center}

Section~\ref{RepTheorySect} is mainly devoted to preliminaries concerning the quantum group $\Uqsltwo$, its representations,
embedding and projection operators, and a canonical invariant bilinear pairing.
We also introduce the ``conformal-block vectors'' $\smash{\HWvec^{\varrho}_{\multii}}$ 
(recall proposition~\ref{DirectSumPropCoblo}), 
which are highest-weight vectors especially relevant for applications to conformal field theory~\cite{kkp, fp1}.
Appendix~\ref{PreliApp} supplements these preliminaries and gathers formulas needed in applications.

In section~\ref{DiagAlgSect} we introduce and investigate the valenced Temperley-Lieb algebra $\TL_\multii(\nu)$
and its action on the tensor product  $\Uqsltwo$-modules.
We also discuss a diagram representation for vectors in these tensor product modules,
analogous to the one developed by I.~Frenkel and M.~Khovanov~\cite{fk}, 
inspired by Kauffman's Temperley-Lieb recoupling theory~\cite{kl, cfs}.
We show that a natural invariant bilinear pairing defined on
valenced link states, which form building blocks for representations of $\TL_\multii(\nu)$, 
corresponds with the bilinear pairing for the $\Uqsltwo$-tensor product module. 

Section~\ref{GraphUQSect} concerns the 
``link state -- highest-weight vector correspondence'' (proposition~\ref{DirectSumPropCombination} 
and items~\ref{TLPropCombinationitem2} and~\ref{TLPropCombinationitem3} of proposition~\ref{TLPropCombination}). 
Specifically, we explicate the connection between valenced link states $\alpha$ 
and the valenced link-pattern basis vectors $\Sing_\alpha$, and show that these vectors are linearly independent 
highest-weight vectors in the $\Uqsltwo$-tensor product module.
We also present a graphical calculus for these vectors and their $F$-descendants.

In section~\ref{GraphTLSect}, we find diagram presentations for the conformal-block vectors 
and show that they are orthogonal (as stated in proposition~\ref{DirectSumPropCoblo}).  
Furthermore, we find diagram presentations for projection operators between various $\Uqsltwo$-modules 
(given in lemma~\ref{ThisLemma} and proposition~\ref{TLProjectionLem3}).
We also show that certain $\Uqsltwo$-submodule projectors generate 
the image of the representation $\Trep_\multii$ of the valenced Temperley-Lieb algebra (lemma~\ref{GeneratorLemTL}), 
which later implies proposition~\ref{GeneratorThmComm} in section~\ref{HigherQSchurWeylSect}. 
Finally, we identify quotients of highest-weight vector spaces with simple $\TL_\multii(\nu)$-modules,
interestingly, even in the case when the $\Uqsltwo$-tensor product module is not semisimple 
(item~\ref{TLPropCombinationitem4} of proposition~\ref{TLPropCombination}).

The final section~\ref{HigherQSchurWeylSect} is devoted to the quantum Schur-Weyl duality itself. 
We prove both theorems~\ref{HighQSchurWeylThm2} and~\ref{GeneralCommutantThm}, 
combining basic general results from appendix~\ref{DCApp} with specific knowledge of the $\Uqsltwo$- and $\TL_\multii(\nu)$-actions. 
Importantly, we prove that the $\TL_\multii(\nu)$-action is faithful whenever $\max \multii < \pmin(q)$ (proposition~\ref{FaithfulPropIntro}).

In appendix~\ref{ExceptionalQSect}, we briefly discuss the special value $q = \pm \ii$ for which the loop fugacity $\nu$ vanishes~\eqref{fugacity}.
Appendix~\ref{ClassicalApp} concerns the case of $q=1$, which corresponds to the case of the classical Lie algebra $\mathfrak{sl}_2$.
In appendix~\ref{LinAlgAppAux}, we include a few basic facts concerning bilinear pairings and orthogonality.
The last appendix~\ref{DCApp} is a self-contained summary of how the double-commutant properties 
follow from very basic results, combining linear algebra with Schur's lemma.

\bigskip

\begin{center}
\bf Acknowledgements
\end{center}

We are very grateful to K.~Kyt\"ol\"a for numerous discussions and encouragement. 
We thank A.~Gainutdinov, C.~Hongler, J.~Jacobsen, V.~Jones, R.~Kashaev, V.~Pasquier, D.~Radnell, 
N.~Reshetikin, D.~Ridout, Y. Saint-Aubin, and H.~Saleur for helpful discussions.
E.P. is funded by the Deutsche Forschungsgemeinschaft (DFG, German Research Foundation) 
under Germany's Excellence Strategy -- GZ 2047/1, Projekt-ID 390685813.
She also acknowledges earlier support from the ERC AG COMPASP, the NCCR SwissMAP, and the Swiss NSF. 
During this work, S.F. was supported by the Academy of Finland grant number 288318, 
``Algebraic structures and random geometry of stochastic lattice models.'' 

\section{Quantum group action on its type-one modules} 
\label{RepTheorySect}

In this section, we discuss the quantum group $\Uqsltwo = U_q(\mathfrak{sl}_2)$ and its representation theory.
We give needed definitions and notation and introduce basic tools in section~\ref{UqSect},
and gather auxiliary results and details in appendix~\ref{PreliApp}.
In this article, we only assume very basic familiarity with the concept of a representation. 
In section~\ref{CobloVecSec} we introduce ``conformal-block vectors,''  special highest-weight vectors
especially relevant for applications to conformal field theory~\cite{kkp, fp1}.
Using the conformal-block vectors, we give explicit direct-sum decompositions for type-one $\Uqsltwo$-modules
(the well-known proposition~\ref{MoreGenDecompAndEmbProp}). 
In section~\ref{EmbAndProjSec}, we introduce embedding and projection operators used throughout this article.
The final section~\ref{BilinSect} concerns a canonical invariant bilinear pairing for type-one $\Uqsltwo$-modules. 

\subsection{The quantum group $\Uqsltwo$ and its type-one modules}
\label{UqSect}

The purpose of this section is to define (variants of) 
the quantum group $\Uqsltwo$ and summarize salient facts about its representation theory.
Some basics concerning representations of associative algebras are discussed in appendix~\ref{DCApp}.
Throughout this section, we fix $q \in \bC \setminus\{0, \pm1\}$ and use the notation
\begin{align}  \label{MinPower} 
\pmin(q) := 
\begin{cases} 
\infty, & \text{$q$ is not a root of unity}, \\
p, & \text{$q=e^{\pi \ii p'/p}$ for coprime $p,p' \in \bZpos$} .
\end{cases}
\end{align}
$\Uqsltwo := U_q(\mathfrak{sl}_2)$ is the infinite-dimensional 
associative $\bC$-algebra with unit $1$, generators $E$, $F$, $K$, $K^{-1}$, and relations
\begin{align} \label{AlgRelations} 
KK^{-1} = K^{-1}K = 1 , \qquad 
KE = q^2 EK , \qquad 
KF = q^{-2} FK , \qquad 
[E,F] = \frac{K - K^{-1}}{q - q^{-1}} .
\end{align}
As a vector space, $\Uqsltwo$ has a Poincar\'e-Birkhoff-Witt type basis
\begin{align}
\label{PBWBasis}
\{ E^k K^m F^\ell \,|\, k, \ell \in \bZnn, m \in \bZ \} .
\end{align}

A (left) \emph{representation} of $\Uqsltwo$ is a homomorphism 
$\rho \colon \Uqsltwo \longrightarrow \End \mathsf{V}$ of algebras 
from $\Uqsltwo$ to the algebra $\End \mathsf{V}$ of endomorphisms of some finite-dimensional vector space $\mathsf{V}$. 
We call $\Module{\mathsf{V}}{\Uqsltwo}$ 
an $\Uqsltwo$-\emph{module}, emphasizing the action in the notation.
We call $\Module{\mathsf{V}}{\Uqsltwo}$ \emph{simple}, and $\rho$ \emph{irreducible}, 
if it is not zero and it 
contains no non-zero proper submodules.
We call $\Module{\mathsf{V}}{\Uqsltwo}$ \emph{semisimple}, and $\rho$ \emph{completely reducible}, 
if $\Module{\mathsf{V}}{\Uqsltwo}$ is isomorphic to a direct sum of simple $\Uqsltwo$-modules.
If we have $K.v = \lambda v$ for some $v \in \mathsf{V} \setminus \{0\}$ 
and $\lambda \in \bC$, then we call $\lambda$ a \emph{weight} of $\Module{\mathsf{V}}{\Uqsltwo}$. Also, if a vector
$v \in \mathsf{V} \setminus \{0\}$ satisfies $E.v = 0$ and $K.v = \lambda v$ for some $\lambda \in \bC$, 
then we call $v$ a \emph{highest-weight vector}, and we call 
the $\Uqsltwo$-module that $v$ generates a \emph{highest-weight module}. 
The following facts about the representation theory of $\Uqsltwo$ are used throughout.

\begin{enumerate}[leftmargin=*, label = $\Usltwo$\arabic*., ref = $\Usltwo$\arabic*] 
\itemcolor{red}
\item \label{HWVFact0} 
\textnormal{\cite[lemma~\red{VI.3.4}]{ck}:}
Let $v_0$ be a highest-weight vector of weight $\lambda \in \bC$ and $v_\ell := F^\ell. v_0$. Then, for all $\ell \in \bZnn$,
\begin{align}
K . v_\ell = \lambda q^{-2\ell} v_\ell , \qquad
E . v_\ell = \frac{q^{-(\ell-1)} \lambda - q^{\ell - 1} \lambda^{-1} }{q - q^{-1}} [\ell] v_{\ell-1} , \qquad
F . v_\ell = v_{\ell+1} .
\end{align}

\item \label{HWVFact} 
\textnormal{\cite[theorem~\red{VI.3.5}]{ck}:}
Let $\mathsf{N}$ be a $\Uqsltwo$-module generated by a highest-weight vector $v_0$ of weight $\lambda \in \bC$, and let $v_\ell := F^\ell. v_0$. 
If $0 < \dim \mathsf{N} = s + 1 < \pmin(q)$, then
\begin{enumerate}
\itemcolor{red}
\item[(a):] 
we have $\lambda \in \{\pm q^s\}$,
\item[(b):] 
we have $v_\ell = 0$ for $\ell > s$, and $\{v_0, v_1, \ldots, v_s\}$ is a basis for $\mathsf{N}$,
\item[(c):]
we have $K.v_\ell =  \lambda q^{-2\ell} v_\ell$ for all $\ell \in \{0, 1, \ldots, s\}$,
\item[(d):]
any other highest-weight vector of $\mathsf{N}$ is a scalar multiple of $v_0$, and 
\item[(e):] 
$\mathsf{N}$ is simple.
\end{enumerate}

\item \label{SimpleModFact1}
\textnormal{\cite[proposition~\red{VI.5.1}]{ck}:}
If $s + 1 < \pmin(q)$, then any simple $\Uqsltwo$-module of dimension $s + 1$ is isomorphic to 
a highest-weight module $\mathsf{N}\sub{s}(\chi)$ with $\chi \in \{\pm1\}$, having basis $\{v_0, v_1, \ldots, v_s\}$ and $\Uqsltwo$-action
\begin{align} \label{HopfRep0}
F.v_\ell = 
\begin{cases}
v_{\ell+1} , & 0 \leq \ell \leq s-1 , \\ 
0, & \ell = s ,
\end{cases}
\qquad E.v_\ell = 
\begin{cases}
\chi [\ell][s - \ell + 1] v_{\ell-1} , & 1 \leq \ell \leq s , \\ 
0, & \ell = 0 ,
\end{cases}
\qquad K.v_\ell = \chi q^{s-2\ell} v_\ell .
\end{align}

\item \label{SimpleModFact2}
\textnormal{\cite[theorem~\red{VI.5.5}]{ck}:}
If $s + 1 = \pmin(q)$, then the highest-weight modules $\mathsf{N}\sub{s}(\chi)$ are simple,
but these are not all of the simple $\Uqsltwo$-modules of dimension $s + 1$.

\item \label{SimpleModFact3}
\textnormal{\cite[proposition~\red{VI.5.2}]{ck}:}
If $s + 1 > \pmin(q)$, then there is no simple $\Uqsltwo$-module of dimension $s + 1$. 
\end{enumerate}

Throughout, for each $s \in \bZnn$, we fix a vector space (spin chain)
$\VecSp\sub{s} := \Span \smash{\{ \Basis_0\super{s}, \Basis_1\super{s}, \ldots, \Basis_s\super{s} \}}$
of dimension $s + 1$ and with given basis. 
We define a left $\Uqsltwo$-module structure on $\VecSp\sub{s}$ via the rules
\begin{align} \label{HopfRep}
F.\Basis_\ell\super{s} := 
\begin{cases}
\Basis_{\ell+1}\super{s}, & 0 \leq \ell \leq s - 1 , \\ 
0, & \ell = s ,
\end{cases}
\qquad 
E.\Basis_\ell\super{s} := 
\begin{cases}
[\ell][s - \ell + 1] \Basis_{\ell-1}\super{s}, & 1 \leq \ell \leq s , \\ 
0, & \ell = 0 ,
\end{cases}
\qquad 
K.\Basis_\ell\super{s} := q^{s - 2\ell} \Basis_\ell\super{s} ,
\end{align}
and we call the resulting $\Uqsltwo$-module a (left) \emph{type-one} $\Uqsltwo$-module and denote it by 
\begin{align}
\Wd\sub{s} := \Module{\VecSp\sub{s}}{\Uqsltwo} .
\end{align}
Facts~\ref{HWVFact}--\ref{SimpleModFact3} imply 
that, when $s + 1 \leq \pmin(q)$, the $\Uqsltwo$-module
$\Wd\sub{s}$ is simple and isomorphic to $\mathsf{N}\sub{s}(+1)$.
In this article, we do not consider modules of type $\mathsf{N}\sub{s}(-1)$.
They can, in fact, be constructed from the simple type-one modules and
the one-dimensional module $\mathsf{N}\sub{0}(-1)$ 
via $\mathsf{N}\sub{s}(-1) \cong \mathsf{N}\sub{s}(+1) \otimes \mathsf{N}\sub{0}(-1)$.

The two-dimensional simple module $\Wd\sub{1}$, called the \emph{fundamental module},
is the building block of all type-one $\Uqsltwo$-modules.
We denote its basis vectors by
$\FundBasis_0 := \smash{\Basis\super{1}_0}$ and $\FundBasis_1 := \smash{\Basis\super{1}_1}$, 
so action~\eqref{HopfRep} with $s = 1$ becomes
\begin{align} \label{UAction} 
F.\FundBasis_0 = \FundBasis_1 , \qquad 
F.\FundBasis_1 = 0, \qquad 
E.\FundBasis_0 = 0, \qquad 
E.\FundBasis_1 = \FundBasis_0, \qquad 
K.\FundBasis_0 = q \FundBasis_0 , \qquad 
K.\FundBasis_1 = q^{-1} \FundBasis_1 .
\end{align}

In parallel, 
for each $s \in \bZnn$, we fix another vector space 
$\VecSpBar\sub{s} := \Span \smash{\{ \BasisBar_0\super{s}, \BasisBar_1\super{s}, \ldots, \BasisBar_s\super{s} \}}$,
of dimension $s + 1$ and with given basis, 
we define a right $\Uqsltwo$-module structure on $\VecSpBar\sub{s}$ via the rules
\begin{align} \label{HopfRepRight}
\BasisBar_\ell\super{s} . E := 
\begin{cases}
\BasisBar_{\ell+1}\super{s}, & 0 \leq \ell \leq s - 1 , \\ 
0, & \ell = s ,
\end{cases}
\qquad 
\BasisBar_\ell\super{s} . F := 
\begin{cases}
[\ell][s - \ell + 1] \BasisBar_{\ell-1}\super{s}, & 1 \leq \ell \leq s , \\ 
0, & \ell = 0 ,
\end{cases}
\qquad 
\BasisBar_\ell\super{s} . K :=  q^{s - 2\ell} \BasisBar_\ell\super{s} ,
\end{align}
and we denote the resulting right type-one $\Uqsltwo$-module by
\begin{align}
\WdBar\sub{s} := \RModule{\VecSpBar\sub{s}}{\Uqsltwo} .
\end{align}

Importantly, the algebra $\Uqsltwo$ also has a bialgebra 
structure~\cite{ck,cp,krt}, given by the coproduct $\Delta \colon \Uqsltwo \longrightarrow \Uqsltwo \otimes \Uqsltwo$
and counit $\epsilon \colon \Uqsltwo \longrightarrow \bC$, 
the algebra homomorphisms defined by homomorphic extensions of 
\begin{align} 
\label{CoProd} 
\Delta(E) = & \; E \otimes K + 1 \otimes E , \qquad 
\Delta(F) = F \otimes 1 + K^{-1} \otimes F , \qquad 
\Delta(K) = K \otimes K, \qquad
\Delta(K^{-1}) = K^{-1} \otimes K^{-1} , \\
\label{CoUnit}
\epsilon(E) = & \; \epsilon(F) = 0 , \qquad \qquad \quad \,
\epsilon(K) = \epsilon( K^{-1} ) = 1 .
\end{align}
We use the counit to define a $\Uqsltwo$-action on the ground field $\bC$ 
by $x.\lambda = \epsilon(x)\lambda$ for all $x \in \Uqsltwo$ and $\lambda \in \bC$. 
The $\Uqsltwo$-module $\bC$ is isomorphic to the one-dimensional simple module $\Wd\sub{0}$ and we call it the \emph{trivial module}. 
Using the coproduct~\eqref{CoProd}, we define left and right $\Uqsltwo$-module structures on tensor products of the form 
\begin{align} \label{VecSpTensProd}
\begin{array}{l}
\VecSp_\multii := \VecSp\sub{\sIndex_1} \otimes \VecSp\sub{\sIndex_2} \otimes \dotsm \otimes \VecSp\sub{\sIndex_{\np_\multii}} , \\[3pt]
\VecSpBar_\multii := \VecSpBar\sub{\sIndex_1} \otimes \VecSpBar\sub{\sIndex_2} \otimes \dotsm \otimes \VecSpBar\sub{\sIndex_{\np_\multii}} , 
\end{array}
\qquad \qquad \text{with} \quad \multii := (\sIndex_1, \sIndex_2, \ldots, \sIndex_{\np_\multii}) \in \bZpos^{\np_\multii},
\end{align}
that we call \emph{type-one modules} and denote by
\begin{align} \label{TensProdModules}
\Module{\VecSp_\multii}{\Uqsltwo} & := \Wd\sub{\sIndex_1} \otimes \Wd\sub{\sIndex_2} \otimes \dotsm \otimes \Wd\sub{\sIndex_{\np_\multii}} 
\qquad\qquad \text{and}\qquad\qquad 
\Module{\VecSp_\multii}{\Uqsltwo} := \WdBar\sub{\sIndex_1} \otimes \WdBar\sub{\sIndex_2} \otimes \dotsm \otimes \WdBar\sub{\sIndex_{\np_\multii}} .
\end{align} 
In the special case of $\multii = \OneVec{n}$ for some $n \in \bZpos$, as in~(\ref{TensorPowerN},~\ref{OneVecDefn}), 
we write $\VecSp_n = \smash{\VecSp_{\OneVec{n}}}$, and
\begin{align} \label{TensorPowerNRecall}
\Module{\VecSp_n}{\Uqsltwo} = \Wd\sub{1}^{\otimes n} , \qquad\qquad
\RModule{\VecSpBar_n}{\Uqsltwo} = \WdBar\sub{1}^{\otimes n}
\end{align}
for the $n$:th tensor powers of the left and right fundamental $\Uqsltwo$-modules.
The $\Uqsltwo$-actions on these modules
are defined by iterating the coproduct~\eqref{CoProd},
\begin{align} 
\label{IteratedCoProd}
\Delta\super{\np} & := (\Delta \otimes \id^{\otimes(\np-2)}) \circ (\Delta \otimes \id^{\otimes(\np-3)}) \circ \dotsm \circ (\Delta \otimes \id) \circ \Delta  ,
\end{align}
so that, for all elements $x \in \Uqsltwo$ and vectors $v \in \VecSp_\multii$ and $\overbarStraight{v} \in \VecSpBar_\multii$, we set
\begin{align} \label{UqTensorAction}
x.v := \Delta\super{\np_\multii}(x) .v .
\qquad \qquad \text{and} \qquad \qquad
\overbarStraight{v}.x := \overbarStraight{v} . \Delta\super{\np_\multii}(x) .
\end{align}
The coassociativity property of the coproduct $\Delta$
guarantees that above, the order in which the tensor products are formed does not matter for the $\Uqsltwo$-module structure.

We implicitly identify the trivial one-dimensional module $\Wd\sub{0}$ with the ground field $\bC$,
and in particular, we omit the factor $\Wd\sub{0}$ from all tensor products via the canonical isomorphism 
$\bC \otimes \Wd\sub{s} \cong \Wd\sub{s} \cong \Wd\sub{s} \otimes \bC$.
With this convention, if $\Module{\VecSp_\multiii}{\Uqsltwo}$ and $\Module{\VecSp_\multii}{\Uqsltwo}$ are two type-one modules, 
their tensor product is a type-one module too, because
\begin{align} \label{ConcatenateMultiindices}
\Module{\VecSp_\multiii}{\Uqsltwo} \otimes \Module{\VecSp_\multii}{\Uqsltwo} 
\isom \Module{\VecSp_{\multiii \oplus \multii}}{\Uqsltwo} , 
\end{align}
where $\multiii \oplus \multii$ is the multiindex obtained by concatenating $\multiii$ to the left of $\multii$ and removing possible zero entries.
Furthermore, if $f \colon \Module{\VecSp_\multiii}{\Uqsltwo} \longrightarrow \Module{\VecSp_\varepsilon}{\Uqsltwo}$ and 
$g \colon \Module{\VecSp_\multii}{\Uqsltwo} \longrightarrow \Module{\VecSp_\varphi}{\Uqsltwo}$ are
homomorphisms of $\Uqsltwo$-modules, then the map 
\begin{align} \label{TensorMap}
f \otimes g \colon \Module{\VecSp_\multiii}{\Uqsltwo} \otimes \Module{\VecSp_\multii}{\Uqsltwo}  
\longrightarrow \Module{\VecSp_\varepsilon}{\Uqsltwo} \otimes \Module{\VecSp_\varphi}{\Uqsltwo} , \qquad
(f \otimes g)(v \otimes w) := f(v) \otimes g(w),
\end{align}
extends linearly to a homomorphism of $\Uqsltwo$-modules.
(In fact, finite-dimensional $\Uqsltwo$-modules form a monoidal category.)

\bigskip

In applications to conformal field theory~\cite{fp1}, another version of the algebra $\Uqsltwo$ is also needed.
This algebra $\UqsltwoBar := \overbarStraight{U}_q(\mathfrak{sl}_2)$ differs from $\Uqsltwo$ only via its coalgebra structure.
Explicitly, $\UqsltwoBar$ is the associative $\bC$-algebra with unit $\overbarStraight{1}$ and generators $\EBar$, $\FBar$, $\KBar$, $\KBar^{-1}$, 
satisfying the same relations~\eqref{AlgRelations},
\begin{align} \label{AlgRelationsBar} 
\KBar \KBar^{-1} = \KBar^{-1} \KBar = 1 , \qquad 
\KBar \EBar = q^2 \EBar \KBar , \qquad 
\KBar \FBar = q^{-2} \FBar \KBar , \qquad 
[\EBar,\FBar] = \frac{\KBar - \KBar^{-1}}{q - q^{-1}} ,
\end{align}
with Poincar\'e-Birkhoff-Witt type basis $\{ \EBar^k \KBar^m \FBar^\ell \,|\, k, \ell \in \bZnn, m \in \bZ \}$, but with
bialgebra 
structure given by the following coproduct and counit:
\begin{align} 
\label{CoProdBar} 
\DeltaBar(\EBar) = & \; \KBar^{-1} \otimes \EBar + \EBar \otimes \overbarStraight{1}, \qquad 
\DeltaBar(\FBar) = \FBar \otimes \KBar + \overbarStraight{1} \otimes \FBar , \qquad
\DeltaBar(\KBar) = \KBar \otimes \KBar, \qquad
\DeltaBar(\KBar^{-1}) = \KBar^{-1} \otimes \KBar^{-1}  , \\
\label{CoUnitBar}
\overbarcal{\epsilon}(\EBar) = & \; \overbarcal{\epsilon}(\FBar) = 0 , \qquad \qquad \qquad \,
\overbarcal{\epsilon}(\KBar) = \overbarcal{\epsilon}( \KBar^{-1} ) = 1 .
\end{align}
We define a left $\UqsltwoBar$-module structure on the vector space $\VecSp\sub{s}$ via the rules
\begin{align} \label{HopfRepBar}
\FBar.\Basis_\ell\super{s} := q^{-1 + (s - 2\ell)} F.\Basis_\ell\super{s} ,
\qquad 
\EBar.\Basis_\ell\super{s} := q^{-1 - (s - 2\ell)} E.\Basis_\ell\super{s} ,
\qquad 
\KBar.\Basis_\ell\super{s} := K.\Basis_\ell\super{s} ,
\end{align}
and we denote by $\Module{\VecSp\sub{s}}{\UqsltwoBar}$ the resulting left $\UqsltwoBar$-module.
Similarly, we define a right $\UqsltwoBar$-module structure on $\VecSpBar\sub{s}$ via 
\begin{align} \label{HopfRepRightBar}
\BasisBar_\ell\super{s} . \EBar := q^{1 - (s - 2\ell)} \BasisBar_\ell\super{s} . E ,
\qquad 
\BasisBar_\ell\super{s} . \FBar := q^{1 + (s - 2\ell)} \BasisBar_\ell\super{s} . F ,
\qquad 
\BasisBar_\ell\super{s} . \KBar := \BasisBar_\ell\super{s} . K ,
\end{align}
and we denote the resulting right $\UqsltwoBar$-module by $\RModule{\VecSpBar\sub{s}}{\UqsltwoBar}$.
Finally, iterating the coproduct $\DeltaBar$, 
we define the tensor product modules $\Module{\VecSp_\multii}{\UqsltwoBar}$ and $\RModule{\VecSp_\multii}{\UqsltwoBar}$ with left and right $\UqsltwoBar$-actions
denoted by $\overbarStraight{x}.v$ and $\overbarStraight{v}.\overbarStraight{x}$, respectively.

The two bialgebras $\Uqsltwo$ and $\UqsltwoBar$ are related in a simple way: 
a calculation shows that the linear map 
\begin{align} 
& E^k K^m F^\ell 
\quad \longmapsto \quad 
(E^k K^m F^\ell)^* := \EBar^\ell \KBar^m \FBar^k , \\
& \EBar^k \KBar^m \FBar^\ell 
\quad \longmapsto \quad 
(\EBar^k \KBar^m \FBar^\ell)^* := E^\ell K^m F^k ,
\end{align}
gives an anti-isomorphim of associative, unital algebras, 
as well as an isomorphism of coassociative, counital coalgebras. 

The tensor product $\VecSp_\multii$ (similarly, $\VecSpBar_\multii$) defined in~\eqref{VecSpTensProd} admits the following $s$-grading:
\begin{align} 
\label{Kgraded}
\VecSp_\multii 
= \; & \bigoplus_{s \, \in \, \DefectSet_{\Summed_\multii}\superscr{\pm}} \Ksp_\multii\super{s} , 
\qquad \text{where} \qquad
\DefectSet_{\Summed_\multii}\superscr{\pm}
:= \{-\Summed_\multii, \, -\Summed_\multii + 2, \, \ldots , \, \Summed_\multii - 2 , \, \Summed_\multii \} , \\
\label{sGrading} 
\Ksp_\multii\super{s} := \; & \Span\big\{ \Basis_{\ell_1}\super{\sIndex_1} \otimes \Basis_{\ell_2}\super{\sIndex_2} \otimes \dotsm \otimes 
\Basis_{\ell_{\np_\multii}}\super{\sIndex_{\np_\multii}} \in \VecSp_\multii \, \,\big| \,\, \Summed_\multii - 2(\ell_1 + \ell_2 + \dotsm + \ell_{\np_\multii}) = s \big\} .
\end{align}
By~(\ref{HopfRep},~\ref{CoProd}), all $K$-eigenvalues in $\Module{\VecSp_\multii}{\Uqsltwo}$ 
have the form $q^s$ with $s \in \smash{\DefectSet_{\Summed_\multii}\superscr{\pm}}$.
If $q$ is a root of unity, some of these eigenvalues 
may coincide. They are all distinct if 
\begin{align} \label{Kweightsdistinct}
\Summed_\multii < \pmin(q) ,
\qquad \qquad \text{where} \qquad \qquad
\pmin(q) \overset{\eqref{MinPower}}{:=} 
\begin{cases} 
\infty, & \text{$q$ is not a root of unity}, \\
p, & \text{$q=e^{\pi \ii p'/p}$ for coprime $p,p' \in \bZpos$,} 
\end{cases}
\end{align}
and in this case, we have
\begin{align}
\label{npqcond} 
\Summed_\multii < \pmin(q)
\qquad \qquad \Longrightarrow \qquad \qquad
\begin{array}{l}
\Ksp_\multii\super{s} 
= \big\{ v \in \VecSp_\multii \; \big| \, K.v = q^s.v \big\},  \\[3pt]
\KspBar_\multii\super{s} 
= \big\{ \overbarStraight{v} \in \VecSp_\multii \, \big| \, \overbarStraight{v} . K = q^s. \overbarStraight{v} \big\} .
\end{array}
\end{align}
The $s$-grade of is moved up and down via the generators $E$ and $F$,
as specified in lemma~\ref{MergeLem} in appendix~\ref{PreliApp}; e.g.
\begin{align} 
E.v, \EBar.v \in \Ksp_\multii\super{s + 2}, \qquad 
F.v, \FBar.v \in \Ksp_\multii\super{s - 2}, \qquad 
K^{\pm1}.v, \KBar^{\pm1}.v \in \Ksp_\multii\super{s} 
\qquad \text{for all} \quad v \in \Ksp_\multii\super{s} .
\end{align}
Also, we can trade the generators $E$ and $F$ of $\Uqsltwo$ with the generators $\EBar$ and $\FBar$ of $\UqsltwoBar$ via  
\begin{align} 
(\EBar^k \KBar^m \FBar^\ell) .v 
= q^{-k(s - 2\ell + k) + m(s - 2\ell) + \ell(s - \ell)} \, (E^k K^m F^\ell) .v 
\qquad \text{for all} \quad v \in \Ksp_\multii\super{s} .
\end{align}

\bigskip

Next, we consider highest-weight vectors in the left and right $\Uqsltwo$-type-one modules:
\begin{align} \label{HWsp} 
\HWsp_\multii := \{ v \in \VecSp_\multii \, | \, E.v = 0 \} 
\qquad \qquad \text{and} \qquad \qquad
\HWspBar_\multii := \{ \overbarStraight{v} \in \VecSpBar_\multii \, | \, \overbarStraight{v}.F = 0 \} .
\end{align}  
These spaces are graded as in~(\ref{Kgraded}--\ref{sGrading}):
\begin{align}
\label{HWspace2}
\HWsp_\multii\super{s} := \HWsp_\multii \cap \Ksp_\multii\super{s}
\qquad \qquad \textnormal{and} \qquad \qquad
\HWspBar_\multii\super{s} := \HWspBar_\multii \cap \KspBar_\multii\super{s} ,
\end{align}
with $K$-eigenvalues of the form $q^s$ for integers $s \in \smash{\DefectSet_{\Summed_\multii}\superscr{\pm}}$.
Analogously to~\eqref{npqcond}, we note that 
\begin{align}
\label{npqcond2} 
\Summed_\multii < \pmin(q)
\qquad \qquad \Longrightarrow \qquad \qquad
\begin{array}{l}
\HWsp _\multii\super{s} 
= \big\{ v \in \VecSp_\multii \; \big| \, E.v = 0, \, K.v = q^s.v \big\}, \\[3pt]
\HWspBar _\multii\super{s} 
= \big\{ \overbarStraight{v} \in \VecSpBar_\multii \, \big| \, \overbarStraight{v}.F = 0, \, \overbarStraight{v} . K = q^s. \overbarStraight{v} \big\}.
\end{array}
\end{align}

\begin{lem} \label{HWDirectSumLem}
Suppose $q \in \bC^\times \setminus \{\pm1\}$. 
The highest-weight vector spaces have the following direct-sum decompositions:
\begin{align} \label{HWDirectSumPM}
\HWsp_\multii = \bigoplus_{s \, \in \, \DefectSet_{\Summed_\multii}\superscr{\pm}} \HWsp_\multii\super{s} 
\qquad \qquad \textnormal{and} \qquad \qquad
\HWspBar_\multii = \bigoplus_{s \, \in \, \DefectSet_{\Summed_\multii}\superscr{\pm}} \HWspBar_\multii\super{s}. 
\end{align}
\end{lem}

\begin{proof}
Definitions~\eqref{HWspace2} immediately give the containments ``$\subset$'' in~\eqref{HWDirectSumPM},
and with $s$-gradings~\eqref{Kgraded}, we see that the sums in~\eqref{HWDirectSumPM} are direct. 
To prove the containments ``$\supset$,'' 
we decompose a vector $v \in \HWsp_\multii$ as 
\begin{align} \label{KDecomp}
v = \bigoplus_{s \, \in \, \DefectSet_{\Summed_\multii}\superscr{\pm}} v\super{s}, \qquad \text{where} \qquad v\super{s} \in \Ksp_\multii\super{s} ,
\end{align}
and use lemma~\ref{MergeLem} from appendix~\ref{PreliApp} to obtain
\begin{align} 
v \in \HWsp_\multii
\qquad \overset{\eqref{HWsp}}{\Longrightarrow} \qquad 
E.v = 0 
\qquad & \underset{\eqref{vShift}}{\overset{\eqref{KDecomp}}{\Longrightarrow}}
\qquad E.v\super{s} = 0 \text{ for all } s \in \DefectSet_{\Summed_\multii}\superscr{\pm} \\
& \underset{\hphantom{\eqref{KDecomp}}}{\overset{\eqref{HWspace2}}{\Longrightarrow}} 
\qquad v\super{s} \in \HWsp_\multii\super{s} \text{ for all } s \in \DefectSet_{\Summed_\multii}\superscr{\pm} 
\qquad \overset{\eqref{KDecomp}}{\Longrightarrow} \qquad 
v \in \bigoplus_{s \, \in \, \DefectSet_{\Summed_\multii}\superscr{\pm}} \HWsp_\multii\super{s} ,
\end{align}
which gives the containment ``$\supset$'' for the left equation of~\eqref{HWDirectSumPM}.
The right side of~\eqref{HWDirectSumPM} follows similarly.
\end{proof}

We can describe the highest-weight vectors recursively. For this purpose, and
throughout this article, for a multiindex $\multii \in \smash{\bZpos^\# := \bZpos \cup \bZpos^2 \cup \bZpos^3 \cup \dotsm}$, we denote
\begin{align} \label{hats}
\multii = (\sIndex_1, \sIndex_2, \ldots, \sIndex_{\np_\multii}) 
\qquad \qquad \Longrightarrow \qquad \qquad 
\lds := (\sIndex_1, \sIndex_2, \ldots, \sIndex_{\np_\multii-1}) \quad \textnormal{and} \quad t := \sIndex_{\np_\multii} .
\end{align}

\begin{lem}
\label{HWform1GenLem} 
Suppose $q \in \bC^\times \setminus \{\pm1\}$. 
The following hold:
\begin{enumerate}
\itemcolor{red}
\item \label{HWform1GenItem} 
Write the vector $v \in \smash{\VecSp_\multii\super{s}}$ in the form
\begin{align} \label{HWform1Gen} 
v = \sum_{\ell \, = \, 0}^{t} v_\ell \otimes \Basis_{t - \ell}\super{t} , 
\qquad \textnormal{where} \qquad 
v_\ell \in \Ksp_{\lds}\super{s + t - 2 \ell}.
\end{align}
Then $v \in \smash{\HWsp_\multii\super{s}}$ if and only if
\begin{align}
\label{EvAct}
E.v_\ell = 
\begin{cases}
-q^{t-2 \ell} [t - \ell + 1] [\ell]  \, v_{\ell - 1} , & \ell > 0 , \\
0 , & \ell = 0.
\end{cases}
\end{align}

\item \label{HWform1GenBarItem3}
Similarly, write the vector $\overbarStraight{v} \in \smash{\VecSpBar_\multii\super{s}}$ in the form
\begin{align} \label{HWform1GenBar3} 
\overbarStraight{v} = \sum_{\ell \, = \, 0}^{t} \overbarStraight{v}_\ell \otimes \BasisBar_{t - \ell}\super{t}   ,
\qquad \textnormal{where} \qquad 
\overbarStraight{v}_\ell \in \KspBar_{\lds}\super{s + t - 2 \ell}.
\end{align}
Then $\overbarStraight{v} \in \smash{\HWspBar_\multii\super{s}}$ if and only if
\begin{align}
\overbarStraight{v}_\ell .F = 
\begin{cases}
- q^{-s - t + 2 \ell - 2} [t - \ell + 1][\ell] \, \overbarStraight{v}_{\ell - 1} , & \ell > 0 , \\
0 , & \ell = 0.
\end{cases}
\end{align}
\end{enumerate}
\end{lem}

\begin{proof} 
Writing a vector $v \in \smash{\Ksp_\multii\super{s}}$ in the generic form~\eqref{HWform1Gen}, we have
\begin{align}
E.v \underset{\eqref{HWform1Gen}}{\overset{(\ref{HopfRep},~\ref{CoProd})}{=}} 
q^{-t} (E.v_0) \otimes \Basis_{t}\super{t} 
+ \sum_{\ell \, = \, 1}^{t} \Big( q^{2 \ell-t} (E.v_\ell) \otimes \Basis_{t - \ell}\super{t} 
+ [t - \ell + 1] [\ell]  \, v_{\ell-1} \otimes \Basis_{t - \ell}\super{t} \Big) ,
\end{align}
which clearly vanishes if and only if~\eqref{EvAct} holds. 
This proves item~\ref{HWform1GenItem}, and item~\ref{HWform1GenBarItem3} can be proven similarly.
\end{proof}

\begin{cor} \label{UpperBoundDimension2Coro}
Suppose $\max \multii < \pmin(q)$. Then, for all $s \in \bZ$, we have
\begin{align} \label{UpperBoundDimension2}
\dim\HWsp_\multii\super{s} = \dim\HWspBar_\multii\super{s} \leq
B_\multii\super{s} ,
\end{align}
where the integers $B_\multii\super{s}$ are unique solution to the recursion problem
\begin{align} \label{Brecursion}
B_\multii\super{s} =
\sum_{\substack{r \, = \, s - t \\ r + s + t \, = \, 0 \Mod 2}}^{s + t} B_{\lds}\super{r}, 
\qquad \qquad B\sub{t}\super{s}  = 
\delta_{t,s} .
\end{align}
\end{cor}

\begin{proof}
We first prove that $\dim \smash{\HWsp_\multii\super{s}} \leq \smash{B_\multii\super{s}}$ 
by induction on $\np_\multii \in \bZpos$.
In the initial case $\np_\multii = 1$, we have $\multii = (t)$ for some $t \in \bZpos$, 
and~\eqref{Brecursion} follows from~\eqref{HopfRep} and~(\ref{Qinteger},~\ref{MinPower}), assuming that $t < \pmin(q)$. 
Next, we let $d \geq 2$ and assume that $\dim \smash{\HWsp_\lds\super{s}} \leq \smash{B_\lds\super{s}}$ holds 
for any multiindex $\smash{\lds \in \bZpos^{d - 1}}$ with $\max \smash{\lds} < \pmin(q)$. 
We let $\smash{\HWsp_\multii\superscr{(s); \, p}}$ denote the vector space
\begin{align} \label{GenHWsp}
\HWsp_\multii\superscr{(s); \, m} := \{ v \in \Ksp_\multii\super{s} \,|\, E^m.v = 0\}.
\end{align}
Then, taking $\multii \in \smash{\bZpos^d}$ with $\max \multii < \pmin(q)$ and using notation~\eqref{hats}, we consider the linear map
\begin{align}
\label{VtLinMap}
v_t \quad \mapsto \quad \sum_{\ell \, = \, 0}^t \lambda_\ell\super{t} (E^{t - \ell}.v_t) \otimes \Basis_{t - \ell}\super{t}, 
\qquad \textnormal{where} \qquad 
\lambda_\ell\super{t} = (-1)^{t - \ell} q^{(t - \ell)(\ell + 1)} \frac{[\ell]!}{[t - \ell]! [t]!}  \neq 0 ,
\end{align}
sending $\smash{\HWsp_{\lds}\superscr{(s - t); \, t + 1}}$ into $\smash{\VecSp_\multii\super{s}}$
(where the constants $\lambda_\ell\super{t}$ are motivated by item~\ref{HWform1GenItem} of lemma~\ref{HWform1GenLem}).
We observe that
\begin{align}
\sum_{\ell \, = \, 0}^t \lambda_\ell\super{t} (E^{t - \ell}.v_t) \otimes \Basis_{t - \ell}\super{t} = 0 
\qquad \Longrightarrow \qquad
v_t \otimes \Basis_0\super{t} = 0
\qquad \Longrightarrow \qquad 
v_t = 0,
\end{align}
so this map is injective. 
Also, repeated application of~\eqref{EvAct} from item~\ref{HWform1GenItem} of lemma~\ref{HWform1GenLem} 
implies that the image of this map is $\smash{\HWsp_\multii\super{s}}$. 
In conclusion, the map~\eqref{VtLinMap} is a linear isomorphism from $\smash{\HWsp_{\lds}\superscr{(s - t); \, t + 1}}$ to $\smash{\HWsp_\multii\super{s}}$,
so we have
\begin{align}
\label{EqDims}
\dim \HWsp_\multii\super{s} = \dim \smash{\HWsp_{\lds}\superscr{(s - t); \, t + 1}}.
\end{align}
By definition~\eqref{GenHWsp} of $\smash{\HWsp_\multii\superscr{(s); \, p}}$ and~\eqref{vShift} from lemma~\ref{MergeLem}, 
the action of 
$E$ induces a sequence of $t$ linear maps, 
\begin{align} \label{Chain}
\HWsp_{\lds}\superscr{(s - t); \, t + 1} 
\quad \overset{E.(\cdot) \,}{\longrightarrow} \quad \HWsp_{\lds}\superscr{(s - t + 2); \, t}
\quad \overset{E.(\cdot) \,}{\longrightarrow} \quad \HWsp_{\lds}\superscr{(s - t + 4); \, t - 1}
\quad \overset{E.(\cdot) \,}{\longrightarrow} \quad \dotsm 
\quad \overset{E.(\cdot) \,}{\longrightarrow} \quad \HWsp_{\lds}\superscr{(s + t); \, 1} = \HWsp_{\lds}\super{s + t}.
\end{align}
By definition~\eqref{HWsp}, the vector space $\smash{\HWsp_{\lds}\super{s - t}}$ contains 
the kernel of the first (leftmost) map in this chain~\eqref{Chain}, so  
\begin{align}
\label{IterateMe}
\dim \HWsp_{\lds}\superscr{(s - t); \, t + 1} \leq \dim \smash{\HWsp_{\lds}\super{s - t}} + \dim \HWsp_{\lds}\superscr{(s - t + 2); \, t}.
\end{align}
Iterating this argument $(t - 1)$ more times gives
\begin{align}
\nonumber
\dim \HWsp_\multii\super{s} \overset{\eqref{EqDims}}{=} \dim \smash{\HWsp_{\lds}\superscr{(s - t); \, t + 1}} 
\overset{\eqref{IterateMe}}{\leq} \; & 
\sum_{\substack{r \, = \, s - t \\ r + s + t \, = \, 0 \Mod 2}}^{s + t} \dim \HWsp_{\lds}\super{r} 
\overset{\eqref{UpperBoundDimension2}}{\leq} 
\sum_{\substack{r \, = \, s - t \\ r + s + t \, = \, 0 \Mod 2}}^{s + t} B_{\lds}\super{r} \overset{\eqref{Brecursion}}{=} B_\multii\super{s},
\end{align}
where the rightmost inequality follows from the induction hypothesis. 
This completes the induction step. 

To finish, using the linear isomorphism $( \, \cdot \,)^* \colon \VecSp_\multii \longrightarrow \VecSpBar_\multii$ 
from lemma~\ref{StarLem} 
combined again with lemma~\ref{MergeLem}, we have
\begin{align}
v \in \HWsp_\multii\super{s}
 \qquad \underset{\eqref{HWspace2}}{\overset{\eqref{HWsp}}{\Longrightarrow}} \qquad E.v = 0 
\qquad \overset{\eqref{BarToNoneLeft}}{\Longrightarrow} \qquad \EBar.v = 0 
\qquad \overset{\eqref{StarAnti1}}{\Longrightarrow} \qquad v^*. F = 0 
 \qquad \underset{\eqref{HWspace2}}{\overset{\eqref{HWsp}}{\Longrightarrow}} \qquad v^* \in \HWspBar_\multii\super{s}, \\
\overbarStraight{v} \in \HWspBar_\multii\super{s}
\qquad \underset{\eqref{HWspace2}}{\overset{\eqref{HWsp}}{\Longrightarrow}} \qquad \overbarStraight{v}.F = 0
\qquad \overset{\eqref{BarToNoneRight}}{\Longrightarrow} \qquad \overbarStraight{v}.\FBar = 0
\qquad \overset{\eqref{StarAnti2}}{\Longrightarrow} \qquad E.\overbarStraight{v}^* = 0
 \qquad \underset{\eqref{HWspace2}}{\overset{\eqref{HWsp}}{\Longrightarrow}} \qquad \overbarStraight{v}^* \in \HWsp_\multii\super{s} ,
\end{align}
so the linear map $( \, \cdot \,)^*$ restricted to the subspace $\smash{\HWsp_\multii\super{s}}$ 
is a linear injection into $\smash{\HWspBar_\multii\super{s}}$, and its inverse $( \, \cdot \,)^*$ restricted to 
the subspace $\smash{\HWspBar_\multii\super{s}}$ is a linear injection into $\smash{\HWsp_\multii\super{s}}$. 
This shows that $\dim \smash{\HWsp_\multii\super{s}} = \dim \smash{\HWspBar_\multii\super{s}}$ and finishes the proof. 
\end{proof}

\subsection{Conformal-block vectors and direct-sum decomposition}
\label{CobloVecSec}

It is well known that when $\Summed_\multii < \pmin(q)$, the type-one $\Uqsltwo,\UqsltwoBar$-modules 
$\Module{\VecSp_\multii}{\Uqsltwo,\UqsltwoBar}$ and $\RModule{\VecSpBar_\multii}{\Uqsltwo,\UqsltwoBar}$
are semisimple, 
as explicated in proposition~\ref{MoreGenDecompAndEmbProp} below. 
The main aim of this section is to construct particular highest-weight vectors, that we call ``conformal-block vectors,'' 
which give rise to explicit decompositions of these modules into direct sums of simple submodules.
The main importance of the conformal-block vectors are their applications to conformal field theory~\cite{fp1}, 
which also explains our choice of terminology.
When $\Summed_\multii \geq \pmin(q)$, the modules 
$\Module{\VecSp_\multii}{\Uqsltwo,\UqsltwoBar}$ and $\RModule{\VecSpBar_\multii}{\Uqsltwo,\UqsltwoBar}$
are no longer semisimple and also the direct construction of the conformal-block vectors fails, 
(reflecting degeneracies also in conformal field theory).
In section~\ref{GraphUQSect}, we find other highest-weight vectors, that we call the ``(valenced) link-pattern basis vectors,''
useful beyond the range $\Summed_\multii < \pmin(q)$.

The conformal-block vectors are indexed by certain walks. 
To define and study them, we first introduce some combinatorial notation and make simple observations. 
For $r,s,t \in \bZnn$, we set 
\begin{align}\label{SpecialDefSet} 
\DefectSet\sub{s} = \{s\} \qquad \textnormal{and} \qquad \DefectSet\sub{r,t} = \{ |r-t| , |r-t| + 2, \ldots, r+t\} . 
\end{align}
We immediately note the symmetry relations
\begin{align} \label{SameDefSet} 
s \in \DefectSet\sub{r,t} \qquad \Longleftrightarrow \qquad r \in \DefectSet\sub{t,s} \qquad \Longleftrightarrow \qquad t \in \DefectSet\sub{s,r} . 
\end{align}

For a multiindex $\multii = (\sIndex_1, \sIndex_2,\ldots, \sIndex_{\np_\multii}) \in \bZnn^{\np_\multii}$, we define a \emph{walk over $\multii$} to be 
a new multiindex $\varrho = (r_1, r_2, \ldots, r_{\np_\multii})$
whose entries, each called a \emph{height}, satisfy the following two conditions relative to $\multii$:
\begin{align}\label{WalkHeights} 
r_0 = 0 , \qquad \text{$r_{i+1} \in \DefectSet\sub{r_i,\sIndex_{i+1}} 
\overset{\eqref{SpecialDefSet}}{=} \{ |r_i - \sIndex_{i+1}|, |r_i - \sIndex_{i+1}| + 2, \ldots, r_i + \sIndex_{i+1} \}$ 
\quad for all $i \in \{1,2, \ldots, \np_{\multii}-1\}$} . 
\end{align}
As a notation convention, we do not explicitly show the zeroth height $r_0 = 0$ of the walk 
$\varrho = (r_1, r_2, \ldots, r_{\np_\multii})$ so the length of $\varrho$ matches that of $\multii$.  
We call the last height $\defect{\varrho} := r_{\np_\multii}$ the \emph{defect} of the walk $\varrho$. 
Then, we set 
\begin{align} 
\label{CountLP} 
\Dim_\multii\super{s} := \; & \# \{ \textnormal{walks over $\multii$ with defect $s$} \}, \\
\label{DefSetDefinition}
\DefectSet_\multii := \; & \{ s \in \bZnn \, | \, 
\textnormal{there exists a walk over $\multii$ with defect $s$} \} , \\
\label{DimensionsAllSummed}
\Dim_\multii := \; & \sum_{s \, \in \, \DefectSet_\multii} \Dim_\multii\super{s} = \# \{ \textnormal{walks over $\multii$} \} .
\end{align}
In the special case of $\multii = \OneVec{n}$, we simply have
\begin{align}\label{DefectSet} 
\DefectSet_n = \{ n \; \text{mod} \; 2, \; (n \; \text{mod} \; 2) + 2, \; \ldots, \; n \} ,
\end{align} 
and more generally, there are integers $\smin(\multii), \smax(\multii) \in \bZnn$ such that (see also \cite[lemmas~\red{2.1}--\red{2.3}]{fp3a})
\begin{align}\label{DefSet2} 
\DefectSet_\multii = \{ \smin(\multii), \smin(\multii) + 2, \; \ldots, \; \smax(\multii) \} \, \subset \, \DefectSet_{\Summed_\multii}\superscr{\pm} ,
\qquad \text{where} \qquad \smax(\multii) = \Summed_\multii .
\end{align}

In fact, the number of walks over $\multii$ can be determined recursively, 
using notation~\eqref{hats}:
\begin{lem} \label{DnumberLem}
The following hold:
\begin{enumerate}
\itemcolor{red}
\item \label{DnumberItem1}
The integers $\Dim_\multii\super{s}$ defined in~\eqref{CountLP}, for $s \in \bZ$, are the unique solution to the following recursion problem:
\begin{align} 
\label{Recursion2Def} 
\Dim_\multii\super{s} \hspace{.1cm}
 = \hspace{.5cm} \sum_{\mathclap{r \, \in \, \DefectSet\sub{s,t} }}  \quad \Dim_{\lds}\super{r} 
\quad \qquad \textnormal{and} \quad \qquad 
\Dim\sub{s}\super{t} = \delta_{s,t} .
\end{align}

\item \label{DnumberItem2}
In particular, we have
\begin{align} \label{DnumiszeroOutOfRange} 
s \notin \DefectSet_\multii \qquad \Longleftrightarrow \qquad \Dim_\multii\super{s} = 0 ,
\end{align}
so recursion~\eqref{Recursion2Def} simplifies to
\begin{align} \label{Recursion2}
\Dim_\multii\super{s} \hspace{.1cm}
 = \hspace{.5cm} \sum_{\mathclap{r \, \in \, \DefectSet_{\lds} \, \cap \, \DefectSet\sub{s,t} }} \quad \Dim_{\lds}\super{r} 
\quad \qquad \textnormal{and} \quad \qquad 
\Dim\sub{s}\super{t} = \delta_{s,t} ,
\end{align}

\item \label{DnumberItem3} 
The following identity holds: 
\begin{align} \label{DimID}
\sum_{s \, \in \, \DefectSet_\multii} (s + 1) \Dim_\multii\super{s} 
= \prod_{i \, = \, 1}^{\np_\multii} (\sIndex_i + 1) .
\end{align}

\end{enumerate}
\end{lem}
\begin{proof}
Item~\ref{DnumberItem1} is proved in~\cite[lemma~\red{4.1}, item~\red{4}]{fp3a}. 
Item~\ref{DnumberItem2} follows immediately from this and~\eqref{CountLP}.
We prove item~\ref{DnumberItem3} by induction on $\np_\multii  \in \bZpos$. It is trivial for $\np_\multii = 1$.
In the induction step, we need the case $\np_\multii = 2$ as well. 
In that case, with $\multii = (r,t)$ and for all $s \in \DefectSet\sub{r,t}$, we immediately obtain from~\eqref{Recursion2} that
\begin{align} \label{DimIDSpec}
\Dim\sub{r,t}\super{s} \overset{\eqref{Recursion2}}{=} 1 
\qquad\qquad \Longrightarrow \qquad\qquad
\sum_{s \, \in \, \DefectSet\sub{r,t}} (s + 1) \Dim\sub{r,t}\super{s} =
\sum_{s \, \in \, \DefectSet\sub{r,t}} (s + 1) \overset{\eqref{SpecialDefSet}}{=} (r + 1) (t + 1),
\end{align}
which proves~\eqref{DimID} for the case $\np_\multii = 2$. 
To finish, we assume that~\eqref{DimID}
holds for $\np_\multii = d - 1$ for some integer $d \geq 3$ and show that it holds if $\np_\multii = d$ too.  
Indeed, using notation~\eqref{hats}, we have
\begin{align} 
\nonumber
\prod_{i \, = \, 1}^{\np} (\sIndex_i + 1) 
&\; \underset{\eqref{DimID}}{\overset{\eqref{hats}}{=}} \; (t + 1) \sum_{r \, \in \, \DefectSet_{\lds}} (r + 1) \Dim_{\lds}\super{r} \\
\nonumber
&\; \overset{\eqref{DimIDSpec}}{=} \;  \sum_{r \, \in \, \DefectSet_{\lds}} \sum_{s \, \in \, \DefectSet\sub{r, t}} 
(s + 1) \Dim\sub{r,t}\super{s} \Dim_{\lds}\super{r} \\
& \; \underset{\eqref{DimIDSpec}}{\overset{\eqref{DefSetDefinition}}{=}} \; \sum_{s \, \in \, \DefectSet_\multii} 
\hspace{.65cm} \sum_{\mathclap{r \, \in \, \DefectSet_{\lds} \, \cap \, \DefectSet\sub{s,t} }}  
\quad 
(s + 1) \Dim_{\lds}\super{r} 
\; \overset{\eqref{Recursion2}}{=} \;  \sum_{s \, \in \, \DefectSet_\multii} (s + 1) \Dim_\multii\super{s} ,
\end{align}
which completes the induction step. 
\end{proof}

%

%

Our next aim is to define the conformal-block vectors and prove that they are linearly independent highest-weight vectors.
To begin, if $\max \multii < \pmin(q)$ and $\multii = \varrho \oplus \vartheta$, then 
we set
\begin{align} 
\label{GenHWV} 
\HWvecMap\sub{r,t}\super{s}(v \otimes w)
& := \sum_{i \, = \, 0}^{\frac{r + t - s}{2}} \sum_{j \, = \, 0}^{\frac{r + t - s}{2}} \delta_{i + j,\frac{r + t - s}{2}} 
\frac{(-1)^{j} q^{j(t + 1 - j)}}{(q - q^{-1})^{(r + t - s) / 2}} 
\frac{[r - i]![t - j]!}{[i]![j]![r]![t]!} \;
F^i.v \otimes F^j.w \\
\label{GenHWVbar} 
\HWvecMapBar\sub{r,t}\super{s}(\overbarStraight{v} \otimes \overbarStraight{w})
& := \sum_{i \, = \, 0}^{\frac{r + t - s}{2}} \sum_{j \, = \, 0}^{\frac{r + t - s}{2}} \delta_{i + j,\frac{r + t - s}{2}} 
\frac{(-1)^{-i} q^{-i(r + 1 - i)}}{(q - q^{-1})^{(r + t - s) / 2}} 
\frac{[r - i]![t - j]!}{[i]![j]![r]![t]!}  \;
\overbarStraight{v} . E^i \otimes \overbarStraight{w} . E^j 
\end{align}
for all vectors $v \in \smash{\HWsp_\varrho\super{r}}$, $w \in \smash{\HWsp_\vartheta\super{t}}$, 
$\overbarStraight{v} \in \smash{\HWspBar_\varrho\super{r}}$, and $\overbarStraight{w} \in \smash{\HWspBar_\vartheta\super{t}}$, 
and for all $s \in \DefectSet\sub{r,t}$ such that $\max(r,t) < \pmin(q)$. 
Also, for $v = \smash{\Basis_0\super{r}}$, $w = \smash{\Basis_0\super{t}}$, 
$\overbarStraight{v} = \smash{\BasisBar_0\super{r}}$, and $\overbarStraight{w} = \smash{\BasisBar_0\super{t}}$, we simply denote 
\begin{align} 
\label{tau} 
\HWvec\sub{r,t}\super{s} 
:= \sum_{i \, = \, 0}^{\frac{r + t - s}{2}} \sum_{j \, = \, 0}^{\frac{r + t - s}{2}} \delta_{i + j,\frac{r + t - s}{2}} 
\frac{(-1)^{j} q^{j(t + 1 - j)}}{(q - q^{-1})^{(r + t - s) / 2}} 
\frac{[r - i]![t - j]!}{[i]![j]![r]![t]!}  \;
\Basis_i\super{r} \otimes \Basis_j\super{t} , \\
\label{taubar} 
\HWvecBar\sub{r,t}\super{s} 
:= \sum_{i \, = \, 0}^{\frac{r + t - s}{2}} \sum_{j \, = \, 0}^{\frac{r + t - s}{2}} \delta_{i + j,\frac{r + t - s}{2}} 
\frac{(-1)^{-i} q^{-i(r + 1 - i)}}{(q - q^{-1})^{(r + t - s) / 2}} 
\frac{[r - i]![t - j]!}{[i]![j]![r]![t]!}  \;
\BasisBar_i\super{r} \otimes \BasisBar_j\super{t} .
\end{align}

In addition to notation~\eqref{hats}, we frequently use the following notation for a walk $\varrho$ over a multiindex $\multii \in \bZpos^{\np_\multii}$:
\begin{align} \label{hatsRho}
\varrho = (r_1, r_2, \ldots, r_{\np_\multii}) 
\qquad \qquad \Longrightarrow \qquad \qquad
\hat{\varrho} = (r_1, r_2, \ldots, r_{\np_\multii-1}) ,
\qquad r = r_{\np_\multii-1}, \quad s = r_{\np_\multii}.
\end{align}


\begin{defn} \label{CobloBasisDefinition}
When $\max (\multii, \hat{\varrho}) < \pmin(q)$, we recursively define the \emph{conformal-block vectors}
\begin{align} 
\label{ConfBlockDefn} 
\HWvec\super{t}\sub{t} := \Basis_0\super{t}
\qquad \qquad \text{and} \qquad \qquad
\HWvec^{\varrho}_{\multii} := \HWvecMap\sub{r,t}\super{s}\big(\HWvec^{\hat{\varrho}}_{\lds} \otimes \Basis_0\super{t} \big) , \\
\label{ConfBlockDefnBar} 
\HWvecBar\super{t}\sub{t} := \BasisBar_0\super{t}
\qquad \qquad \text{and} \qquad \qquad
\HWvecBar^{\varrho}_{\multii} := \HWvecMapBar\sub{r,t}\super{s}\big(\HWvecBar^{\hat{\varrho}}_{\lds} \otimes \BasisBar_0\super{t} \big) .
\end{align}
\end{defn}

In the special case of $\multii = \OneVec{n}$, by definition~\eqref{WalkHeights}, all walks $\varrho = (r_1, r_2, \ldots, r_n)$ 
over $\OneVec{n}$ have only steps of size one, 
and the conformal-block vectors can be written very explicitly using definition~(\ref{hatsRho}--\ref{ConfBlockDefn}) and formula~\eqref{GenHWV}:
\begin{align} \label{ConfBlockDefnN}
\HWvec^{\varrho}_{n} 
\overset{\eqref{ConfBlockDefn}}{=} \HWvec\super{r_n}\big( \HWvec^{\hat{\varrho}}_{n-1} \otimes \FundBasis_0 \big) 
& \overset{\eqref{GenHWV}}{=} 
\begin{cases} 
\HWvec^{\hat{\varrho}}_{n-1} \otimes \FundBasis_0 , & r_n = r_{n-1} + 1 , \\ 
\frac{1}{q - q^{-1}} \left( -q \, \HWvec^{\hat{\varrho}}_{n-1} \otimes \FundBasis_1 
+ \frac{1}{[r_{n-1}]} \; F.\HWvec^{\hat{\varrho}}_{n-1} \otimes \FundBasis_0 \right) , & r_n = r_{n-1} - 1 ,
\end{cases} 
\end{align}
where $\hat{\varrho} = (r_1, r_2, \ldots, r_{n-1})$.
We note that these are essentially the same vectors that were constructed in~\cite{kkp}.

\begin{lem} \label{HWPropLem0} 
Suppose $\max(\varrho, \vartheta, r, t) < \pmin(q)$.
Then, for all $s \in \DefectSet\sub{r,t}$, we have 
\begin{align} \label{HWPropLem0ID} 
& \begin{cases}
v \in \HWsp_\varrho\super{r} \\ 
w \in \HWsp_\vartheta\super{t} 
\end{cases} 
\qquad \Longrightarrow \qquad 
\HWvecMap\super{s}\sub{r,t}(v \otimes w) \in \HWsp_{\varrho \, \oplus \, \vartheta}\super{s} .
\end{align}
Similarly, this lemma holds after the symbolic replacements
$v \mapsto \overbarStraight{v}$, $\HWsp \mapsto \HWspBar$, $w \mapsto \overbarStraight{w}$, and $\HWvecMap \mapsto \HWvecMapBar$.
\end{lem}

\begin{proof} 
We only prove~\eqref{HWPropLem0ID}; the other case can be proven similarly.
Without loss of generality, we may assume that 
\begin{align} 
\varrho = ( r ), \quad \vartheta = ( t ) \qquad \Longrightarrow \qquad \smash{\HWvecMap\super{s}\sub{r,t}(v \otimes w) = \HWvec\sub{r,t}\super{s}}. 
\end{align}
Then using formula~\eqref{tau}, it is a straightforward calculation to show
that $\smash{E.\HWvec\sub{r,t}\super{s} = 0}$ and 
$\smash{K.\HWvec\sub{r,t}\super{s} = q^s \HWvec\sub{r,t}\super{s}}$.
\end{proof}

We next prove that the conformal-block vectors are linearly independent.

\begin{lem} \label{ConfBlockLinIndepLem} 
Suppose $\max \multii < \pmin(q)$. The following hold:
\begin{enumerate}
\itemcolor{red}
\item \label{ConfBlockAreLinIndepItem1} 
The collection $\smash{\{ \HWvec^{\varrho}_{\multii} \, | \, \max \hat{\varrho} <\pmin(q) \}}$ 
is a linearly independent subset of $\HWsp_\multii$.

\item \label{ConfBlockAreLinIndepItem2}
For each $s \in \DefectSet_\multii$, the collection 
$\smash{\{ \HWvec^{\varrho}_{\multii} \, | \, \max \hat{\varrho} <\pmin(q) \text{ and } \defect{\varrho} =s \}}$ 
is a linearly independent subset of $\smash{\HWsp_\multii\super{s}}$.
\end{enumerate}
Similarly, this lemma holds after the symbolic replacements
$\HWsp \mapsto \HWspBar$ and $\HWvec \mapsto \HWvecBar$.
\end{lem}

\begin{proof}
In light of lemma~\ref{HWDirectSumLem}, item~\ref{ConfBlockAreLinIndepItem1} follows from item~\ref{ConfBlockAreLinIndepItem2},
so it suffices to prove the latter. To this end, 
by lemma~\ref{HWPropLem0}, we only have to show that the collection
$\smash{\{ \HWvec^{\varrho}_{\multii} \, | \, \max \hat{\varrho} <\pmin(q) \text{ and } \defect{\varrho} = s \}}$ is linearly independent.
We prove this by induction on $\np_\multii  \in \bZpos$.  The case $\np_\multii = 1$ is trivial.
For notational simplicity, supplementing notation~\eqref{hatsRho}, we write 
\begin{align}
\hat{\varrho}' := \hat{\hat{\varrho}} = (r_1, r_2, \ldots, r_{\np_\multii-2}) .
\end{align}
Now, we assume that $\np_\multii = d-1$ for some integer $d \geq 2$ and that the collection 
\begin{align} \label{LinIndeIndHypo}
\big\{ \HWvec^{\hat{\varrho}}_{\lds} \, \big| \, \max \hat{\varrho}' <\pmin(q) \text{ and } \defect{{\hat{\varrho}}} =s \big\}
\end{align}
is linearly independent. Then, we suppose that
\begin{align} \label{VanishSum} 
\sum_{\substack{ \varrho \colon r_d \, = \, s \\ \max \hat{\varrho} \, < \, \pmin(q) } } c_\varrho \, \HWvec^{\varrho}_{\multii} = 0
\end{align}
is a vanishing linear combination of vectors in the collection
$\smash{\{ \HWvec^{\varrho}_{\multii} \, | \, \max \hat{\varrho} <\pmin(q) \text{ and } \defect{\varrho} =s \}}$, 
where $c_\varrho \in \bC$ are some constants. 
We write each vector $\smash{\HWvec^{\varrho}_{\multii}}$ 
in terms of its recursive definition~\eqref{ConfBlockDefn}, to obtain
\begin{align} \label{VanishSum2} 
0 = \sum_{\substack{ \varrho \colon r_d \, = \, s \\ \max \hat{\varrho} \, < \, \pmin(q) } } 
c_\varrho \, \HWvecMap\sub{r,t}\super{s} \big(\HWvec^{\hat{\varrho}}_{\lds} \otimes \Basis\super{t}_0\big) = \sum_{i \, = \, 0}^{\frac{r + t - s}{2}} w_i\super{s}, 
\end{align}
where
\begin{align}
w_i\super{s} := \Bigg( \sum_{\substack{ \varrho \colon r_d  \, = \,  s \\ \max \hat{\varrho} \, < \, \pmin(q) } } 
c_\varrho \, \frac{(-1)^{j}q^{j(t + 1 - j)}}{(q-q^{-1})^{(r + t - s) / 2 }} \frac{[r - i]! [t - j]!}{[i]! [j]! [r]! [t]!} \; F^i. \HWvec^{\hat{\varrho}'}_{\lds} \Bigg) 
\otimes \Basis_j\super{t}, \qquad \text{and $j = \dfrac{r + t - s}{2} - i$}.
\end{align} 
Now, by definition~\ref{CobloBasisDefinition} and lemma~\ref{HWPropLem0}, the $(K \otimes 1)$-eigenvalue of $\smash{w_i\super{s}}$ is $q^{r - 2i}$, 
and because $0 \leq i \leq r < \pmin(q)$, these eigenvalues are distinct.  In light of this fact and~\eqref{VanishSum2}, 
we have $\smash{w_i\super{s}} = 0$ for all $i \in \{0,1,\ldots, \frac{1}{2}(r-s+t)\}$, so
\begin{align} \label{w0Term} 
w_0\super{s} = \Bigg( \sum_{\substack{ \varrho \colon r_d  \, = \,  s \\ \max \hat{\varrho} \, < \, \pmin(q) } } 
c_\varrho \, \frac{(-1)^{j}q^{j(t + 1 - j)}}{(q-q^{-1})^{(r + t - s) / 2 }}\frac{[r]! [t - j]!}{[j]! [r]! [t]!} \; \HWvec^{\hat{\varrho}'}_{\lds} \Bigg) \otimes \Basis_j\super{t} = 0, 
\qquad \text{with $j = \frac{r + t - s}{2}$},
\end{align}
so the sum in~\eqref{w0Term} vanishes. 
Then, because collection~\eqref{LinIndeIndHypo} is linearly independent by the induction hypothesis, 
we have $c_\varrho=0$ for each term in~\eqref{w0Term}. 
This shows that $\smash{\{ \HWvec^{\varrho}_{\multii} \, | \, \max \hat{\varrho} <\pmin(q) \text{ and } \defect{\varrho} =s \}}$ is linearly independent.

The statements with $\HWsp \mapsto \HWspBar$ and $\HWvec \mapsto \HWvecBar$ can be proven similarly.
\end{proof}

The conformal-block vectors 
generate submodules enumerated by the integers
\begin{align} \label{DimTilde}
\hat{\Dim}_\multii\super{s} := \# 
\big\{ \text{walks $\varrho$ over $\multii$}  \, \big| \, \max \hat{\varrho} < \pmin(q) \text{ and } \defect{\varrho} = s \big\} 
\, \overset{\eqref{CountLP}}{\leq} \Dim_\multii\super{s} .
\end{align}
We show next that these submodules form direct-sum submodules inside $\Module{\VecSp_\multii}{\Uqsltwo,\UqsltwoBar}$
and $\RModule{\VecSpBar_\multii}{\Uqsltwo,\UqsltwoBar}$.
(See also proposition~\ref{MoreGenDecompAndEmbProp2} in 
section~\ref{subsec: dual direct-sum decomposition} for a refinement of this result.)

\begin{prop} \label{MoreGenDecompAndEmbProp}
Suppose $\max \multii < \pmin(q)$.
There exists an embedding of left $\Uqsltwo$-modules
\begin{align} \label{DirectSumInclusion}
\bigoplus_{\substack{s \, \in \, \DefectSet_\multii \\ s \, < \, \pmin(q) }} \hat{\Dim}_\multii\super{s} \Wd\sub{s} 
\quad \lhook\joinrel\rightarrow \quad \Module{\VecSp_\multii}{\Uqsltwo}
\end{align}
such that the following hold:
\begin{enumerate}
\itemcolor{red}
\item  \label{DirectSumInclusionItem1}
For each walk $\varrho$ over $\multii$ with $\max \hat{\varrho}  < \pmin(q)$ and $\defect{\varrho} = s$, the collection
\begin{align} 
\label{BasisForm5} 
\big\{ F^\ell.\HWvec^{\varrho}_{\multii} \, \big| \, 0 \leq \ell \leq s \big\}
\end{align}
is a basis for the image of a unique direct summand $\Wd\sub{s}$ in~\eqref{DirectSumInclusion}.

\item  \label{DirectSumInclusionItem2} 
The image of each summand $\Wd\sub{s}$ has a unique basis of the form~\eqref{BasisForm5} with $\max \hat{\varrho}  < \pmin(q)$ and $\defect{\varrho} = s$.

\item  \label{DirectSumInclusionItem3} 
If $\Summed_\multii < \pmin(q)$, then~\eqref{DirectSumInclusion} is an isomorphism of left $\Uqsltwo$-modules,
\begin{align} \label{MoreGenDecomp} 
\Module{\VecSp_\multii}{\Uqsltwo} \isom 
\bigoplus_{s \, \in \, \DefectSet_\multii} \Dim_\multii\super{s} \Wd\sub{s} .
\end{align}
\end{enumerate}
Similarly, this proposition holds for right $\Uqsltwo$-modules, after the symbolic replacements
\begin{align} \label{DirectSumInclusionLemReplaceFixed}
\Wd \mapsto \WdBar, \qquad 
\Module{\VecSp_\multii}{\Uqsltwo} \mapsto \RModule{\VecSpBar_\multii}{\Uqsltwo} ,
\qquad \textnormal{and} \qquad 
F^\ell.\HWvec^{\varrho}_{\multii} \mapsto \HWvecBar^{\varrho}_{\multii}.E^\ell .
\end{align}
Finally, both the left-action and right-action versions of this proposition hold after replacing $\Uqsltwo \mapsto \UqsltwoBar$ in either.
\end{prop}

\begin{proof}
We first show that, for each walk $\varrho$ over $\multii$ with $\max \hat{\varrho} < \pmin(q)$, 
the 
vector $\smash{\HWvec^{\varrho}_{\multii}}$ with $\defect{\varrho} = s$ 
generates a submodule of $\Module{\VecSp_\multii}{\Uqsltwo}$ isomorphic to the simple type-one module $\smash{\Wd\sub{s}}$ 
with basis~\eqref{BasisForm5}.
By lemma~\ref{HWPropLem0}, each $\smash{\HWvec^{\varrho}_{\multii}}$ with $\defect{\varrho} = s$ 
is a highest-weight vector with weight $q^{s}$.
By fact~\ref{HWVFact0}, the vectors $v_\ell := F^\ell . \smash{\HWvec^{\varrho}_{\multii}}$,
for $\ell \in \bZnn$, satisfy
\begin{align} \label{HWaction}
K . v_\ell = q^{s-2\ell} v_\ell , \qquad
E . v_\ell = [s-\ell+1] [\ell] v_{\ell-1} , \qquad
F . v_\ell = v_{\ell+1} .
\end{align}
Thus, in order to show that the submodule generated by $\smash{\HWvec^{\varrho}_{\multii}}$ is isomorphic to $\smash{\Wd\sub{s}}$,
with left $\Uqsltwo$-action~\eqref{HopfRep}, 
it suffices to prove that the set $\{v_0, v_1, \ldots , v_s \}$, i.e.,~\eqref{BasisForm5}, is linearly independent 
and $v_{s+1} := F^{s+1} . \smash{\HWvec^{\varrho}_{\multii}} = 0$.

If $q \in \bC^\times$ is a not a root of unity, then $F^{s+1} . \smash{\HWvec^{\varrho}_{\multii}} = 0$ by fact~\ref{HWVFact}.
Also, the explicit formulas in definition~\ref{CobloBasisDefinition}
of $\smash{\HWvec^{\varrho}_{\multii}}$ combined with 
the explicit formula~\eqref{CoproductFormulas} 
for $\Delta(F^\ell)$ from lemma~\ref{CoproductFormulas}
imply that, for each $\ell \in \{0,1,\ldots,s+1\}$, the function
$q' \mapsto F^\ell . \smash{\HWvec^{\varrho}_{\multii}}(q')$ is continuous for $q' \in \{ q \in \bC^\times \, | \, \pmin(q) > s \}$.
Therefore, the limit of $F^{s+1} . \smash{\HWvec^{\varrho}_{\multii}}(q')$ as $q' \to q$ 
along a sequence not containing roots of unity exists and equals zero. 
Hence, we have $v_{s+1}  = 0$ in any case. 

For $0 \leq \ell \leq s < \pmin(q)$, the $K$-eigenvalues $q^{s-2\ell}$ of $\{v_0, v_1, \ldots , v_s \}$ are all distinct.
Therefore, the set~\eqref{BasisForm5} 
is linearly independent if all of its vectors are non-zero.
To show that this is the case, we let $m \in \bZpos$ be the smallest integer such that $v_m = 0$,
and we show that $m=s+1$.  Indeed, the property
\begin{align} \label{EKills}
0 = E . v_m = [s-m+1] [m] v_{m-1} 
\end{align}
implies that either $\pmin(q) \mid m$ or $\pmin(q) \mid (s-m+1)$, and
since $m$ is the smallest such integer, we have $m \leq s + 1$, so
\begin{align}
\begin{cases}
s \leq \max \varrho < \pmin(q) , \\
1 \leq m \leq s+1 , \\
\text{either $\pmin(q) \mid m$ or $\pmin(q) \mid (s-m+1)$}
\end{cases}
\qquad \qquad \Longrightarrow \qquad \qquad 
m = s + 1.
\end{align}
We conclude that the set~\eqref{BasisForm5} is a basis for the submodule generated by $\smash{\HWvec^{\varrho}_{\multii}}$,
and by~\eqref{HWaction}, this submodule has left $\Uqsltwo$-action~\eqref{HopfRep}, so it is indeed isomorphic to 
the simple type-one module $\smash{\Wd\sub{s}}$.

In summary, for each walk $\varrho$ over $\multii$ with $\max \hat{\varrho} < \pmin(q)$, 
the conformal-block vector $\smash{\HWvec^{\varrho}_{\multii}}$ generates a submodule of $\Module{\VecSp_\multii}{\Uqsltwo}$
with basis $\smash{\{ F^\ell.\HWvec^{\varrho}_{\multii} \, | \, 0 \leq \ell \leq s = \defect{\varrho} \}}$
 isomorphic to $\Wd\sub{s}$.
Therefore, we have an embedding of $\Uqsltwo$-modules
\begin{align} \label{SumSpaces} 
\sum_{\substack{s \, \in \, \DefectSet_\multii \\ s \, < \, \pmin(q) }} \hat{\Dim}_\multii\super{s} \Wd\sub{s} 
\quad \lhook\joinrel\rightarrow \quad \Module{\VecSp_\multii}{\Uqsltwo} .
\end{align}
Now, to prove~\eqref{DirectSumInclusion} and items~\ref{DirectSumInclusionItem1} and~\ref{DirectSumInclusionItem2}, 
it remains to show that this sum is direct.
We prove this by induction on the number of summands.
The initial case with one summand is trivial. For the induction step, we recall that 
the intersection $\mathsf{N}_1 \cap \, \mathsf{N}_2$ of two submodules $\mathsf{N}_1$ and $\mathsf{N}_2$ is itself a submodule.  
Hence, if $\mathsf{N}_1$ is simple, then we have
\begin{align}
\label{IntersectEqual}
\mathsf{N}_1 \cap \mathsf{N}_2 \neq \{0\}
\qquad \Longrightarrow \qquad \mathsf{N}_1 = \mathsf{N}_1 \cap \mathsf{N}_2
\qquad \Longrightarrow \qquad \mathsf{N}_1 \subset \mathsf{N}_2.
\end{align}
Now we assume that the sum of the first $j - 1$ summands $\mathsf{N}_1, \mathsf{N}_2, \ldots, \mathsf{N}_{j - 1}$ in~\eqref{SumSpaces} is direct for some $j \geq 2$. If the sum of the next summand $\mathsf{N}_j$ with the prior sum $\mathsf{N}_1 \oplus \mathsf{N}_2 \oplus \cdots \oplus \mathsf{N}_{j - 1}$ is not direct, then with $\mathsf{N}_j$ simple, we have
\begin{align}
\mathsf{N}_j \cap (\mathsf{N}_1 \oplus \mathsf{N}_2 \oplus \cdots \oplus \mathsf{N}_{j - 1}) \neq \{0\}
\qquad \overset{\eqref{IntersectEqual}}{\Longrightarrow} \qquad 
\mathsf{N}_j \subset \mathsf{N}_1 \oplus \mathsf{N}_2 \oplus \cdots \oplus \mathsf{N}_{j - 1}.
\end{align}
Hence, the highest-weight vector space of $\mathsf{N}_j$, 
spanned by a highest-weight vector $\smash{\HWvec_\multii^\varrho}$ for some walk $\varrho$  over $\multii$, 
must lie inside the highest-weight vector space of the direct sum $\mathsf{N}_1 \oplus \mathsf{N}_2 \oplus \cdots \oplus \mathsf{N}_{j - 1}$. 
However, the latter space is spanned by a set of highest-weight vectors $\smash{\{\HWvec_\multii^{\varrho'}\}}$ for other walks 
$\varrho' \neq \varrho$ over $\multii$. Thus, the former vector $\smash{\HWvec_\multii^\varrho}$ is a linear combination of these latter vectors, 
which contradicts lemma~\ref{ConfBlockLinIndepLem}. 
Therefore, the sum of $\mathsf{N}_j$ with $\mathsf{N}_1 \oplus \mathsf{N}_2 \oplus \cdots \oplus \mathsf{N}_{j - 1}$ is 
direct.

To prove item~\ref{DirectSumInclusionItem3}, we note that if $\Summed_\multii  < \pmin(q)$, then
we have $s \leq \smax(\multii) = \Summed_\multii  < \pmin(q)$ for any $s \in \DefectSet_\multii$,
and also, any walk $\varrho$ over $\multii$ satisfies $\max\varrho \leq \Summed_\multii < \pmin(q)$.
Therefore, we have $\smash{\Dim_\multii\super{s} = \hat{\Dim}_\multii\super{s}}$ and 
embedding~\eqref{DirectSumInclusion} is an isomorphism of left $\Uqsltwo$-modules, 
because the dimensions of both sides of~\eqref{DirectSumInclusion} are now equal by 
item~\ref{DnumberItem3} of lemma~\ref{DnumberLem}. This proves~\eqref{MoreGenDecomp}.

The statements with replacements~\eqref{DirectSumInclusionLemReplaceFixed} or $\Uqsltwo \mapsto \UqsltwoBar$ can be proven similarly.
\end{proof}

\begin{cor}  \label{CobloBasisCor} 
Suppose $\Summed_\multii < \pmin(q)$. Then, the following hold:
\begin{enumerate}
\itemcolor{red}
\item \label{CobloBasisItem1} 
The collection 
$\{ \HWvec^{\varrho}_{\multii}\}$ is a basis for $\HWsp_\multii$.

\item \label{CobloBasisItem2} 
For each $s \in \DefectSet_\multii$, the collection $\{ \HWvec^{\varrho}_{\multii} \, | \, \defect{\varrho} = s \}$ is a basis for $\smash{\HWsp_\multii\super{s}}$.
\end{enumerate}
Similarly, items~\ref{CobloBasisItem1}--\ref{CobloBasisItem2} hold 
after the symbolic replacements $\HWsp \mapsto \HWspBar$ and $\HWvec \mapsto \HWvecBar$.
\begin{enumerate}
\setcounter{enumi}{2}
\itemcolor{red}
\item \label{CobloBasisItem3}  
We have
\begin{align} 
\label{HWspDimension}
\hspace*{3mm}
\dim \HWsp_\multii\super{s} = \dim \HWspBar_\multii\super{s} = \Dim_\multii\super{s}
\qquad \qquad \textnormal{and} \qquad \qquad 
\dim \HWsp_\multii = \dim \HWspBar_\multii = \Dim_\multii,
\end{align}
and the direct-sum decompositions
\begin{align}
\label{KHWspDirSumGeneric}  
\HWsp_\multii = \bigoplus_{s \, \in \, \DefectSet_\multii} \HWsp_\multii\super{s} 
\qquad \qquad \textnormal{and} \qquad \qquad 
\HWspBar_\multii = \bigoplus_{s \, \in \, \DefectSet_\multii} \HWspBar_\multii\super{s} .
\end{align}
In particular, we have
\begin{align} \label{EmptyHWsp}
s \notin \DefectSet_\multii \quad \textnormal{and} \quad \Summed_\multii < \pmin(q)
\qquad \qquad \Longrightarrow \qquad \qquad
\dim \HWsp_\multii\super{s} = \dim \HWspBar_\multii\super{s} = 0 .
\end{align}
\end{enumerate}
\end{cor}

\begin{proof}
Items~\ref{CobloBasisItem1}--\ref{CobloBasisItem2} immediately 
follow from lemma~\ref{ConfBlockLinIndepLem} and 
proposition~\ref{MoreGenDecompAndEmbProp}. 
For item~\ref{CobloBasisItem3}, equalities~\eqref{HWspDimension} follow from items~\ref{CobloBasisItem1}--\ref{CobloBasisItem2}
combined with~(\ref{DefSetDefinition},~\ref{DimensionsAllSummed}),
and~\eqref{KHWspDirSumGeneric} and~\eqref{EmptyHWsp} then follow from lemma~\ref{HWDirectSumLem} with a dimension count.
\end{proof}

\subsection{Submodule projectors and embeddings}
\label{EmbAndProjSec}

Proposition~\ref{MoreGenDecompAndEmbProp} and recursion~\eqref{Recursion2} applied to $\multii = \OneVec{n}$ 
as in~\eqref{TensorPowerNRecall} show that if $n < \pmin(q)$, then the type-one module $\Module{\VecSp_n}{\Uqsltwo}$ 
has a unique simple $(n+1)$-dimensional submodule isomorphic to $\Wd\sub{n}$, 
generated by the highest-weight vector 
\begin{align} \label{MThwv} 
\MTbas_0\super{n} := \underbrace{ \FundBasis_0 \otimes \FundBasis_0 \otimes \dotsm \otimes \FundBasis_0}_{\text{$n$ tensorands}} 
\overset{\eqref{ConfBlockDefnN}}{=} \HWvec\superscr{(1, 2, \ldots, n)}_n \in \HWsp_n\super{n} 
\end{align}
associated to the highest walk $\varrho = (1, 2, \ldots, n)$ over $\OneVec{n}$.
This submodule has basis 
\begin{align} \label{MTFActhwv} 
\MTbas_\ell\super{n} 
& := F^\ell. \MTbas_0\super{n} \overset{\eqref{BarToNoneLeft}}{=} q^{-\ell(n - \ell)} \FBar^\ell.\MTbas_0\super{n}
\qquad  \text{for all $\ell \in \{0,1,\ldots,n\}$.} 
\end{align}
The corresponding basis for the submodule isomorphic to $\WdBar\sub{n}$ in $\RModule{\VecSpBar_n}{\Uqsltwo}$ is
\begin{align} 
\label{MThwvBar}
\MTbasBar_0\super{n} 
& := \FundBasisBar_0 \otimes \FundBasisBar_0 \otimes \dotsm \otimes \FundBasisBar_0, \\
\label{MTFActhwvBar}
\MTbasBar_\ell\super{n} 
& := \MTbasBar_0\super{n}.E^\ell \overset{\eqref{BarToNoneRight}}{=} q^{+\ell(n - \ell)} \MTbasBar_0\super{n}. \EBar^\ell
\qquad \text{for all $\ell \in \{0,1,\ldots,n\}$.}
\end{align}
These basis vectors have explicit formulas:

\begin{lem}
\label{MTbasExplicitLem}
Suppose $q \in \bC^\times \setminus \{\pm1\}$.  For all $\ell \in \{0,1,\ldots,n\}$, we have
\begin{align} 
\label{MTbasExplicit}
\MTbas_{\ell}\super{n} 
= q^{\binom{\ell}{2}}  [\ell]!
\sum_{1 \, \leq \, r_1 \, < \, \cdots \, < \, r_\ell \, \leq \, n} q^{\sum_{i \, = \, 1}^\ell (1-r_i)}  
\left(\FundBasis_{k_1(\varrho)}  \otimes \FundBasis_{k_2(\varrho)} \otimes \cdots \otimes \FundBasis_ {k_{n}(\varrho)}\right) 
\hphantom{q^-}
\end{align} 
and similarly,
\begin{align} 
\label{MTbasBarExplicit} 
\MTbasBar_{\ell}\super{n} 
= q^{-\binom{\ell}{2}}  [\ell]!
\sum_{1 \, \leq \, r_1 \, < \, \cdots \, < \, r_\ell \, \leq \, n} q^{\sum_{i \, = \, 1}^\ell (n-r_i)}  
\left(\FundBasisBar_{k_1(\varrho)}  \otimes \FundBasisBar_{k_2(\varrho)} \otimes \cdots \otimes \FundBasisBar_ {k_{n}(\varrho)}\right) ,
\end{align}
where $\varrho := (r_1,r_2,\ldots,r_\ell)$, and for each $i \in \{1,2,\ldots,n \}$, 
\begin{align}
k_i(\varrho) = \; & 
\begin{cases}
1 , & \textnormal{if } i \in \{r_1,r_2,\ldots,r_\ell \} , \\
0 , & \textnormal{otherwise} .
\end{cases}
\end{align}
In particular, we have
\begin{align} 
\label{MTbashighestK}
\MTbas_n\super{n} = [n]! \, \FundBasis_1 \otimes \FundBasis_1 \otimes \dotsm \otimes \FundBasis_1 
 \qquad \textnormal{and} \qquad
\MTbasBar_n\super{n} 
= [n]! \, \FundBasisBar_1 \otimes \FundBasisBar_1 \otimes \dotsm \otimes \FundBasisBar_1 .
\end{align}
\end{lem}

\begin{proof}
Identity~\eqref{MTbasExplicit} was proved in~\cite[lemma~\red{B.1}]{ep2}
and identity~\eqref{MTbasBarExplicit} can be proven similarly. 
The observation
\begin{align}
- \sum_{i \, = \,1}^s (1-i) = \binom{n}{2} = \sum_{i \, = \,1}^n (n-i) 
\end{align}
proves identities~\eqref{MTbashighestK} for the special case $\ell = n$.
\end{proof}

We let $\smash{\Projection\sub{n} \colon \VecSp_n \longrightarrow \VecSp_n}$ 
and $\smash{\ProjectionBar\sub{n} \colon \VecSpBar_n \longrightarrow \VecSpBar_n}$ denote 
the respective projections from $\Module{\VecSp_n}{\Uqsltwo}$ and $\RModule{\VecSpBar_n}{\Uqsltwo}$ onto
their unique $(n+1)$-dimensional submodules respectively isomorphic to $\Wd\sub{n}$ and $\WdBar\sub{n}$.
Then, we have 
\begin{align}
\label{ProjectionDefn}
\Projection\sub{n}(v) =  \; &
\begin{cases} 
v, & \quad v \in \Span \big\{ \MTbas_\ell\super{n} \,\big|\, 0 \leq \ell \leq n \big\}, \\
0, & \quad \text{otherwise} 
\end{cases} 
\qquad \qquad \text{for $n < \pmin(q)$}, \\
\label{ProjectionDefnBar}
\ProjectionBar\sub{n}(\overbarStraight{v}) = \; &
\begin{cases} 
\overbarStraight{v}, & \quad v \in \Span \big\{ \MTbasBar_\ell\super{n} \,\big|\, 0 \leq \ell \leq n \big\}, \\
0, & \quad \text{otherwise} 
\end{cases}
\qquad \qquad \text{for $n < \pmin(q)$.}
\end{align}
Also, using bases~(\ref{MTFActhwv}--\ref{MTFActhwvBar}), we define the embeddings 
$\smash{\Embedding\sub{n} \colon \Wd\sub{n} \lhook\joinrel\rightarrow \VecSp_n}$ and 
$\smash{\EmbeddingBar\sub{n} \colon \WdBar\sub{n} \lhook\joinrel\rightarrow \VecSpBar_n}$
by linearly extending
\begin{align}
\label{EmbeddingDef}
\Embedding\sub{n}\big(\Basis_\ell\super{n}\big) := \MTbas_\ell\super{n}
\qquad \qquad \text{and} \qquad \qquad 
\EmbeddingBar\sub{n}\big(\BasisBar_\ell\super{n}\big) := \MTbasBar_\ell\super{n} 
\qquad \qquad \text{for $\ell \in \{0,1,\ldots,n\}$,}
\end{align}
and the projectors $\smash{\Projectionhat\sub{n} \colon \VecSp_n \longrightarrow \Wd\sub{n}}$
and $\smash{\ProjectionhatBar\sub{n} \colon \VecSpBar_n \longrightarrow \WdBar\sub{n}}$ as 
inverses of these embeddings on their images: 
\begin{align}
\label{ProjectionHatDefn}
\Projectionhat\sub{n}(v) := 
\begin{cases} 
\Basis_\ell\super{n} , & \quad v = \MTbas_\ell\super{n} \text{ for some } 0 \leq \ell \leq n , \\
0, & \quad v \, \notin \, \Span \big\{ \MTbas_\ell\super{n} \,\big|\, 0 \leq \ell \leq n \big\} 
\end{cases}
\qquad \qquad \text{for $n < \pmin(q)$}, \\
\label{ProjectionHatDefnBar}
\ProjectionhatBar\sub{n}(v) := 
\begin{cases} 
\BasisBar_\ell\super{n} , & \quad v = \MTbasBar_\ell\super{n} \text{ for some } 0 \leq \ell \leq n , \\
0, & \quad v \, \notin \, \Span \big\{ \MTbasBar_\ell\super{n} \,\big|\, 0 \leq \ell \leq n \big\} 
\end{cases}
\qquad \qquad \text{for $n < \pmin(q)$}.
\end{align}
More generally, assuming $\max \multii < \pmin(q)$, we define the composite embeddings and 
composite projectors 
\begin{align}
\label{Composites}
\hspace*{-3mm}
\Embedding_\multii & 
:= \Embedding\sub{\sIndex_1} \otimes \Embedding\sub{\sIndex_2} \otimes \dotsm \otimes \Embedding\sub{\sIndex_{\np_\multii}}, \qquad
\Projection_\multii 
:= \Projection\sub{\sIndex_1} \otimes \Projection\sub{\sIndex_2} \otimes \dotsm \otimes \Projection\sub{\sIndex_{\np_\multii}} , \qquad
\Projectionhat_\multii 
:= \Projectionhat\sub{\sIndex_1} \otimes \Projectionhat\sub{\sIndex_2} \otimes \dotsm \otimes \Projectionhat\sub{\sIndex_{\np_\multii}} , \\
\label{CompositesBar}
\hspace*{-3mm}
\EmbeddingBar_\multii & 
:= \EmbeddingBar\sub{\sIndex_1} \otimes \EmbeddingBar\sub{\sIndex_2} \otimes \dotsm \otimes \EmbeddingBar\sub{\sIndex_{\np_\multii}}, \qquad
\ProjectionBar_\multii 
:=\ProjectionBar\sub{\sIndex_1} \otimes \ProjectionBar\sub{\sIndex_2} \otimes \dotsm \otimes \ProjectionBar\sub{\sIndex_{\np_\multii}} , \qquad
\ProjectionhatBar_\multii
:= \ProjectionhatBar\sub{\sIndex_1} \otimes \ProjectionhatBar\sub{\sIndex_2} \otimes \dotsm \otimes \ProjectionhatBar\sub{\sIndex_{\np_\multii}} .
\end{align}
In lemma~\ref{EmbProjLem} in appendix~\ref{PreliApp}, we gather useful properties of these $\Uqsltwo,\UqsltwoBar$-homomorphisms. 
For instance, the following diagram commutes:
\begin{equation} 
\begin{tikzcd}[column sep=2cm, row sep=1.5cm]
& \arrow{ld}[swap]{\Projectionhat_\multii} \arrow{d}{\Projection_\multii}
\VecSp_{\Summed_\multii}
= \VecSp_{\sIndex_1} \otimes \VecSp_{\sIndex_2} \otimes \dotsm \otimes \VecSp_{\sIndex_{\np_\multii}} \\ 
\VecSp_\multii = \VecSp\sub{\sIndex_1} \otimes \VecSp\sub{\sIndex_2} \otimes \dotsm \otimes \VecSp\sub{\sIndex_{\np_\multii}} 
\arrow{r}{\Embedding_\multii}
& \im \Embedding_\multii = \im \Projection_\multii \subset \VecSp_{\Summed_\multii}
\end{tikzcd}
\end{equation}

In section~\ref{DiagAlgSect}, we give diagram interpretations for the mappings 
$\Embedding_\multii$, $\Projection_\multii$, and $\Projectionhat_\multii$ 
(lemma~\ref{WJprojLem} and corollaries~\ref{CompositeProjCor} and~\ref{CompositeProjCorHatEmb}). 
Also, in section~\ref{GraphicalProjSect}, 
we discuss projections onto other submodules in $\Module{\VecSp_n}{\Uqsltwo}$ and $\Module{\VecSp_\multii}{\Uqsltwo}$
(lemma~\ref{ThisLemma} and proposition~\ref{TLProjectionLem3}).

A special case of decomposition~\eqref{MoreGenDecomp} and recursion~\eqref{Recursion2} for 
the multiindex $\multii = (r,t)$ together show that 
\begin{align} \label{2TensDecomp} 
r + t < \pmin(q) 
\qquad \qquad \Longrightarrow \qquad \qquad \Module{\VecSp\sub{r,t}}{\Uqsltwo}  
& \isom 
= \bigoplus_{s \, \in \, \DefectSet_{(r,t)}} \Wd\sub{s}, 
\end{align}
where $\DefectSet\sub{r,t}$ is the set~\eqref{SpecialDefSet}. 
Because no two summands in decomposition~\eqref{2TensDecomp} are isomorphic,
for each index $s \in \DefectSet_{(r,t)}$, 
we may define the embeddings $\smash{ \CCembedor\super{s}\sub{r,t}  \colon \VecSp\sub{s} \longrightarrow \VecSp\sub{r,t}}$
and $\smash{ \CCembedorBar\super{s}\sub{r,t}  \colon \VecSpBar\sub{s} \longrightarrow \VecSpBar\sub{r,t}}$
by linear extension of the rules
\begin{align}  \label{EmbeddingDef2x2} 
\begin{array}{l} 
\CCembedor\super{s}\sub{r,t}  \big( \Basis_\ell\super{s} \big) := F^\ell.\HWvec\sub{r,t}\super{s} ,  \\[5pt]
\CCembedorBar\super{s}\sub{r,t}  \big( \BasisBar_\ell\super{s} \big) := \HWvecBar\sub{r,t}\super{s} . E^\ell, 
\end{array}
\qquad\qquad \text{$\ell \in \{0, 1, \ldots, s\}$,}
\end{align}
the projectors 
$\smash{ \CCprojector\superscr{(r,t);(s)}\sub{r,t} } \colon \VecSp\sub{r,t} \longrightarrow \VecSp\sub{r,t}$ 
and $\smash{ \CCprojectorBar\superscr{(r,t);(s)}\sub{r,t} } \colon \VecSpBar\sub{r,t} \longrightarrow \VecSpBar\sub{r,t}$ 
by linear extension of the rules
\begin{align} \label{ProjectionDefn2x2} 
\begin{array}{l} 
\CCprojector\superscr{(r,t);(s)}\sub{r,t}  \big( F^\ell.\HWvec\sub{r,t}\super{p} \big) :=\delta_{p,s} F^\ell.\HWvec\sub{r,t}\super{s} ,  \\[5pt]
\CCprojectorBar\superscr{(r,t);(s)}\sub{r,t}  \big( \HWvecBar\sub{r,t}\super{p} . E^\ell \big) := \delta_{p,s} \HWvecBar\sub{r,t}\super{s} . E^\ell ,
\end{array}
\qquad\qquad \text{$\ell \in \{0, 1, \ldots, s\}$,}
\end{align}
and the maps 
$\smash{ \CChatprojector\sub{s}\super{r,t}   \colon \VecSp\sub{r,t} \longrightarrow \VecSp\sub{s}}$ 
and $\smash{ \CChatprojectorBar\sub{s}\super{r,t}   \colon \VecSpBar\sub{r,t} \longrightarrow \VecSpBar\sub{s}}$ 
by linear extensions of the rules
\begin{align}  \label{ProjectioHatDefn2x2} 
\begin{array}{l} 
\CChatprojector\sub{s}\super{r,t}  \big( F^\ell.\HWvec\sub{r,t}\super{p} \big) := \delta_{p,s} \, \Basis_\ell\super{s} ,  \\[5pt]
\CChatprojectorBar\sub{s}\super{r,t}  \big(\HWvecBar\sub{r,t}\super{p} . E^\ell\big) :=  \delta_{p,s} \, \BasisBar_\ell\super{s} ,
\end{array}
\qquad\qquad \text{$\ell \in \{0, 1, \ldots, s\}$.}
\end{align}
In lemma~\ref{EmbProjLem2} in appendix~\ref{PreliApp}, we gather useful properties of these $\Uqsltwo,\UqsltwoBar$-homomorphisms. 
For instance, the following diagram commutes:
\begin{equation} 
\begin{tikzcd}[column sep=2cm, row sep=1.5cm]
& \arrow{ld}[swap]{ \CChatprojector\sub{s}\super{r,t} } \arrow{d}{ \CCprojector\superscr{(r,t);(s)}\sub{r,t} }
\VecSp\sub{r,t}
\\ 
\VecSp\sub{s} 
\arrow{r}{ \CCembedor\super{s}\sub{r,t} }
& \im  \CCembedor\super{s}\sub{r,t}  = \im  \CCprojector\superscr{(r,t);(s)}\sub{r,t}  \subset \VecSp\sub{r,t}
\end{tikzcd}
\end{equation}
In proposition~\ref{TLProjectionLem3} in section~\ref{GraphicalProjSect}, 
we give diagram interpretations for these mappings.

\subsection{Bilinear pairing of type-one modules}
\label{BilinSect}

Now, we assign a (canonical) bilinear pairing $\SPBiForm{\cdot}{\cdot}$ to the vector spaces $\VecSpBar_\multii$ and $\VecSp_\multii$.
The standard tensor product basis vectors are orthogonal with respect to this pairing, 
and the pairing is invariant under the $\Uqsltwo$ and $\UqsltwoBar$-actions.
Furthermore, the conformal-block basis vectors $\smash{\HWvecBar\sub{r,t}\super{s}}$ and $\smash{ \HWvec\sub{r,t}\super{s}}$
are also orthogonal with respect to $\SPBiForm{\cdot}{\cdot}$, as we prove in lemma~\ref{CoBloOrthBasisLem} in section~\ref{CoBloGraphical} . 
Later, in section~\ref{LSandBiformandSCGrapgSec}, we provide a diagram interpretation~\cite{kl,fk} for $\SPBiForm{\cdot}{\cdot}$, see lemma~\ref{BiFormLem}.

\begin{lem} \label{BiFormDefLem}
Suppose $\max \multii < \pmin(q)$, and let $\SPBiForm{\cdot}{\cdot} \colon \VecSpBar_\multii \times \VecSp_\multii \longrightarrow \bC$ 
be a bilinear pairing such that 
\begin{align} \label{biformnormalization}
\SPBiFormBig{\BasisBar_{0}\super{\sIndex_1} \otimes \BasisBar_{0}\super{\sIndex_2} \otimes \dotsm \otimes \BasisBar_{0}\super{\sIndex_{\np_\multii}}}{\Basis_{0}\super{\sIndex_1} \otimes \Basis_{0}\super{\sIndex_2} \otimes \dotsm \otimes \Basis_{0}\super{\sIndex_{\np_\multii}}}
= 1 .
\end{align}
Then, the following statements are equivalent:
\begin{enumerate}
\itemcolor{red}

\item \label{biformEquivitem1}
For all elements $x \in \UqsltwoPow{\np_\multii}$ and $\bar{x} \in \UqsltwoBarPow{\np_\multii}$ and vectors 
$\overbarStraight{v} \in \VecSp_\multii$ and $w \in \VecSp_\multii$, we have
\begin{align} \label{biformEquivitem1form}
\SPBiForm{\overbarStraight{v}}{x.w} = \SPBiForm{\overbarStraight{v}.x}{w}
\qquad \qquad \textnormal{and} \qquad \qquad
\SPBiForm{\overbarStraight{v}}{\bar{x}.w} = \SPBiForm{\overbarStraight{v}.\bar{x}}{w} .
\end{align}

\item \label{biformEquivitem2}
We have
\begin{align} \label{biformDefnBasisvec}
\SPBiFormBig{\BasisBar_{\ell_1}\super{\sIndex_1} \otimes \BasisBar_{\ell_2}\super{\sIndex_2} \otimes \dotsm \otimes \BasisBar_{\ell_{\np_\multii}}\super{\sIndex_{\np_\multii}}}{\Basis_{m_1}\super{\sIndex_1} \otimes \Basis_{m_2}\super{\sIndex_2} \otimes \dotsm \otimes \Basis_{m_{\np_\multii}}\super{\sIndex_{\np_\multii}}}
= \prod_{k \, = \, 1}^{\np_\multii} \delta_{\ell_k, m_k}[\ell_k]!^2\,\qbin{\sIndex_k}{\ell_k}. 
\end{align} 
\end{enumerate}
\end{lem}

\begin{proof}
Without loss of generality, we consider the case of $\np_\multii = 1$, when $\multii = (t)$ for some $t \in \bZnn$. 
\begin{enumerate}[leftmargin=3.5em]
\item[\ref{biformEquivitem1} $\Rightarrow$~\ref{biformEquivitem2}:] 
Assuming that item~\ref{biformEquivitem1} holds, we derive assertion~\eqref{biformDefnBasisvec} of item~\ref{biformEquivitem2}:
\begin{align} 
\nonumber
\SPBiFormBig{ \BasisBar_\ell\super{t} }{ \Basis_m\super{t} }
\underset{\eqref{HopfRepRight}}{\overset{\eqref{HopfRep}}{=}} \SPBiFormBig{ \BasisBar_0\super{t}. E^\ell }{ F^m.\Basis_0\super{t} }
\overset{\eqref{biformEquivitem1form}}{=} 
\; & 
\begin{cases}
\SPBiFormBig{ \BasisBar_0\super{t} }{ E^\ell F^m.\Basis_0\super{t} }, & \ell \geq m  \\[5pt]
\SPBiFormBig{ \BasisBar_0\super{t}. E^\ell F^m }{ \Basis_0\super{t} }, & \ell \leq m
\end{cases}
\\
\label{BaissVecBiForm}
\overset{\hphantom{\eqref{biformEquivitem1form}}}{ 
\underset{\eqref{HopfRepRight}}{\overset{\eqref{HopfRep}}{=}}} 
\; & \delta_{\ell, m}[\ell]!^2\,\qbin{t}{\ell}\SPBiFormBig{ \BasisBar_0\super{t} }{ \Basis_0\super{t} }
\overset{\eqref{biformnormalization}}{=}\delta_{\ell, m}[\ell]!^2\,\qbin{t}{\ell} .
\end{align}

\item[\ref{biformEquivitem2} $\Rightarrow$~\ref{biformEquivitem1}:] 
To prove the left equality of~\eqref{biformEquivitem1form}, 
without loss of generality, we assume that $x \in \{E, F, K^{\pm1}\}$, and
$\smash{\overbarStraight{v} = \BasisBar_{\ell}\super{t}}$ and $\smash{w = \Basis_{m}\super{t}}$.
Then, assuming that item~\ref{biformEquivitem2} holds, we calculate, e.g., for $x = E$, 
\begin{align} 
\SPBiFormBig{ \BasisBar_\ell\super{t} }{ E. \Basis_m\super{t} } 
\underset{\hphantom{\eqref{biformDefnBasisvec}}}{\overset{\eqref{HopfRep}}{=}} 
[m][t + 1 - m] \, \SPBiFormBig{ \BasisBar_\ell\super{t} }{ \Basis_{m - 1}\super{t} } 
\underset{\eqref{biformDefnBasisvec}}{\overset{\eqref{Qinteger}}{=}} 
\SPBiFormBig{ \BasisBar_{\ell + 1}\super{t} }{ \Basis_m\super{t} } 
\overset{\eqref{HopfRep}}{=} 
\SPBiFormBig{ \BasisBar_\ell\super{t}.E }{ \Basis_m\super{t} },
\end{align}
with the convention that $\smash{\Basis_{-1}\super{t}} = 0$ and $\smash{\BasisBar_{t + 1}\super{t}} = 0$.
The other cases  $x \in \{F, K^{\pm1}\}$
and the right equality of~\eqref{biformEquivitem1form} can be proven similarly.
This shows that item~\ref{biformEquivitem2} implies item~\ref{biformEquivitem1}. 
\end{enumerate}
\end{proof}

\begin{remark}
Formulas~(\ref{biformnormalization},~\ref{biformDefnBasisvec}) in lemma~\ref{BiFormDefLem} uniquely define a $\Uqsltwo$-invariant bilinear pairing 
$\SPBiForm{\cdot}{\cdot} \colon \VecSpBar_\multii \times \VecSp_\multii \longrightarrow \bC$.
See also~\cite[theorem~{VII.6.2.}]{ck}.
\end{remark}

In appendix~\ref{PreliApp}, we consider a bilinear form 
on $\VecSp_\multii$ obtained by taking $q \mapsto q^{-1}$ instead of $\VecSpBar_\multii$ in the first component.

\begin{lem} \label{biformPropertyLem}
Suppose $\max \multii < \pmin(q)$.  The following hold:
\begin{enumerate}
\itemcolor{red}

\item \label{biformitem1}
For all vectors 
$\overbarStraight{v}_j \in \VecSpBar\sub{\sIndex_j}$ and $w_j \in \VecSp\sub{\sIndex_j}$,
with $j \in \{1,2,\ldots,\np_\multii\}$, we have the factorization 
\begin{align} \label{biformfactorize}
\SPBiForm{\overbarStraight{v}_1 \otimes \overbarStraight{v}_2 \otimes \dotsm \otimes \overbarStraight{v}_{\np_\multii}}{
w_1 \otimes w_2 \otimes \dotsm \otimes w_{\np_\multii}}
= \prod_{k \, = \, 1}^{\np_\multii} \SPBiForm{\overbarStraight{v}_k}{w_k} .
\end{align} 

\item \label{biformitem2}
For all elements $x \in \Uqsltwo$ and $\bar{x} \in \UqsltwoBar$ and vectors 
$\overbarStraight{v} \in \VecSp_\multii$ and $w \in \VecSp_\multii$, we have
\begin{align} \label{biformEquivitem1form2}
\SPBiForm{\overbarStraight{v}}{x.w} = \SPBiForm{\overbarStraight{v}.x}{w}
\qquad \qquad \textnormal{and} \qquad \qquad
\SPBiForm{\overbarStraight{v}}{\bar{x}.w} = \SPBiForm{\overbarStraight{v}.\bar{x}}{w} .
\end{align}

\item \label{biformitem3}
The subspaces $\smash{\VecSpBar_\multii\super{s}}$ and $\smash{\VecSp_\multii\super{t}}$ are orthogonal\textnormal{:}
\begin{align} \label{OrthoSubspaces}
\SPBiForm{\overbarStraight{v}}{w} = 0 \qquad
\textnormal{for all} \quad 
\overbarStraight{v} \in \VecSpBar_\multii\super{s} \textnormal{ and } 
w \in \VecSp_\multii\super{t}
\textnormal{ with } s \neq t .
\end{align}
Also, for all highest-weight vectors 
$\overbarStraight{v} \in \smash{\HWspBar_\multii\super{s}}$ and $w \in \smash{\HWsp_\multii\super{t}}$, we have
\begin{align} \label{FactBiformId}
\SPBiForm{\overbarStraight{v} . E^\ell}{F^m . w} 
= \delta_{s,t} \delta_{\ell, m}[\ell]!^2\,\qbin{t}{\ell} \SPBiForm{\overbarStraight{v}}{w} 
= \SPBiForm{\overbarStraight{v} . \EBar^\ell}{\FBar^m . w}.  
\end{align}

\item \label{biformitem4}
For all vectors 
$\overbarStraight{v} \in \VecSp_\multii$ and $w \in \VecSp_\multii$, we have
\begin{align} \label{SPBiFormNewEmbed}
\SPBiForm{\overbarStraight{v}}{w} = 
\SPBiForm{\EmbeddingBar_\multii(\overbarStraight{v})}{\Embedding_\multii(w)} .
\end{align}
\end{enumerate}
\end{lem}

\begin{proof} 
Using lemma~\ref{BiFormDefLem}, we prove items~\ref{biformitem1}--\ref{biformitem4} as follows:
\begin{enumerate}[leftmargin=*]
\itemcolor{red}
\item Factorization~\eqref{biformfactorize} immediately follows from the formula~\eqref{biformDefnBasisvec} for the bilinear pairing.

\item Identities~\eqref{biformEquivitem1form2} readily follow from~\eqref{biformEquivitem1form} and the definitions of 
the $\Uqsltwo,\UqsltwoBar$-actions on $\VecSp_\multii$ and $\VecSpBar_\multii$.

\item Orthogonality~\eqref{OrthoSubspaces} immediately follows from definitions (cf.~\eqref{sGrading}) 
with~\eqref{biformDefnBasisvec} from item~\ref{biformEquivitem2} of lemma~\ref{BiFormDefLem}.
To verify~\eqref{FactBiformId}, we note that the orthogonality of the spaces $\smash{\HWspBar_\multii\super{s}}$ and $\smash{\HWsp_\multii\super{t}}$ 
together with a similar calculation as in~\eqref{BaissVecBiForm} readily imply
the first equality in~\eqref{FactBiformId}, and the second equality follows similarly.

\item 
The value of $\SPBiForm{\BasisBar_\ell\super{t}}{\Basis_m\super{t}}$ is given in~\eqref{BaissVecBiForm},
and the same calculation shows that this value equals 
\begin{align}
\SPBiForm{\EmbeddingBar_\multii(\BasisBar_\ell\super{t})}{\Embedding_\multii(\Basis_m\super{t})} 
\overset{\eqref{EmbeddingDef}}{=} \SPBiForm{\MTbas_\ell\super{t}}{\MTbasBar_m\super{t}} 
\underset{\eqref{BaissVecBiForm}}{\overset{\textnormal{(\ref{MThwv},~\ref{MThwvBar})}}{=}}
\delta_{\ell, m}[\ell]!^2 \qbin{t}{\ell} 
\overset{\eqref{BaissVecBiForm}}{=}
\SPBiFormBig{\BasisBar_\ell\super{t}}{\Basis_m\super{t}} .
\end{align}
Asserted identity~\eqref{SPBiFormNewEmbed} follows from this by linearity 
and the factorization property~\eqref{biformfactorize} in item~\ref{biformitem1} 
together with the analogous factorization property
for the embedding maps $\Embedding_\multii$ and $\EmbeddingBar_\multii$ in~(\ref{Composites}--\ref{CompositesBar}).
\end{enumerate}
This concludes the proof.
\end{proof}

In the next identity, we encounter the evaluation of the \emph{Theta network}~\cite{kl} from~\cite[lemma~\red{A.7}]{fp3a}:
\begin{align} \label{ThetaFormula}
\ThetaNet(r,s,t)
= \frac{(-1)^{\frac{r + s + t}{2}} \left[ \frac{r + s + t}{2} + 1 \right]! \left[ \frac{ r + s - t }{2} \right]! \left[ \frac{ s + t - r}{2} \right]! \left[ \frac{t + r - s}{2} \right]! }{[ r ]! [s ]! [ t ]!} . 
\end{align}

\begin{lem} \label{BiFormTauLem} 
Suppose $\max(r,t) < \pmin(q)$.  We have 
\begin{align} 
\label{FinalBiFormTau} 
\SPBiFormBig{ \HWvecBar\sub{r,t}\super{s} }{ \HWvec\sub{r,t}\super{s'} } 
& = 
\delta_{s,s'} \, \frac{\ThetaNet(r,s,t)}{(q - q^{-1})^{r + t - s}[\frac{r + t - s}{2}]!^2 \, [s + 1]} .  
\end{align}
\end{lem} 

\begin{proof}
Using definitions~(\ref{tau},~\ref{taubar}) and 
lemmas~\ref{BiFormDefLem} and~\ref{biformPropertyLem}, we have
\begin{align}
\nonumber
\SPBiFormBig{ \HWvecBar\sub{r,t}\super{s} }{ \HWvec\sub{r,t}\super{s'} } 
\overset{\textnormal{(\ref{tau},~\ref{taubar})}}{\underset{\textnormal{(\ref{biformDefnBasisvec},~\ref{FactBiformId})}}{=}} 
\; & \delta_{s,s'} \, \frac{ (-1)^{\frac{r + s + t}{2}} \; q^{- (2 + r + s - t)(r + t - s)/4}}{(q - q^{-1})^{r + t - s} [\frac{r + t - s}{2}]! [r]![t]!}  \\
\; & \times \sum_{j \, = \, 0}^{\frac{r + t - s}{2}} 
q^{j (2 + s)} \; \qbin{(r + t - s)/2}{j} \; \Big[r - \frac{r + t - s}{2} + j \Big]! \, [t - j]! .
\end{align}
By~\cite[Lemma~\red{A.1}, item~(\red{d})]{ep2} (with $\nu_1 = t$, $\nu_2 = r$, and $n = \frac{r + t - s}{2}$ in that lemma), we can explicitly evaluate this sum.
Then, using formula~\eqref{ThetaFormula} for the Theta network, we arrive with~\eqref{FinalBiFormTau}.
\end{proof}

\section{Temperley-Lieb action on type-one $\Uqsltwo$-modules} 
\label{DiagAlgSect}
In this section, we define an action of the valenced Temperley-Lieb algebra $\TL_\multii(\nu)$ on the 
tensor products $\VecSp_\multii$ and $\smash{\VecSpBar_\multii}$. 
For this purpose, in section~\ref{TLReviewSec} we recall definitions and notation from~\cite{fp3a} to be used throughout. 
Then, in sections~\ref{DiacActTypeOneSec}--\ref{GenDiacActTypeOneSec} 
we define the Temperley-Lieb-actions on 
$\CModule{\VecSp_\multii}{\TL}$ and $\CRModule{\VecSpBar_\multii}{\TL}$, 
commuting with the $\Uqsltwo$-action
on $\Module{\VecSp_\multii}{\Uqsltwo}$ and $\RModule{\VecSpBar_\multii}{\Uqsltwo}$.
In section~\ref{LSandBiformandSCGrapgSec}, we give a diagram representation for vectors in $\VecSp_n$ and $\smash{\VecSpBar_n}$
analogous to the one developed by I.~Frenkel and M.~Khovanov~\cite{fk}
(however, our conventions and purposes are somewhat different).
Although such graphical ideas are relatively well known~\cite{kl, cfs}, we present them in detail for the sake of exposition
(and because the conventions vary greatly in the literature).
We also define a natural invariant bilinear pairing on valenced link state diagrams, and relate it to the bilinear pairing discussed in section~\ref{BilinSect}.


\subsection{Valenced tangles and link states} \label{TLReviewSec}

Here, we recall definitions and notation from~\cite{fp3a}, to be used throughout this article. 

{\bf Temperley-Lieb category.} 
To begin, we 
consider planar non-crossing tangles in the \emph{Temperley-Lieb category} $\smash{\TL^1(\nu)}$~\cite{kl, vt, ck, gl2}.
Its object class comprises the special multiindices with all entries equal to one:
\begin{align} \label{OneVecDefn}
\text{Ob} \, \TL^1(\nu) &= \big\{ \OneVec{n} \,\big| \, n \in \bZnn \big\}, \qquad \text{where } \qquad
\OneVec{0} := (0), \qquad \OneVec{n} := (\underbrace{1,1,\ldots,1}_{\text{$n$ times}}) \quad \text{for $n \in \bZpos$} .
\end{align}
The morphisms in $\smash{\TL^1(\nu)}$ are \emph{$(n,m)$-tangles}, that is, formal linear combinations $T \in \smash{\TL_n^m}$
of \emph{$(n,m)$-link diagrams} 
\begin{align} \label{LinkDiagEx}
\vcenter{\hbox{\includegraphics[scale=0.275]{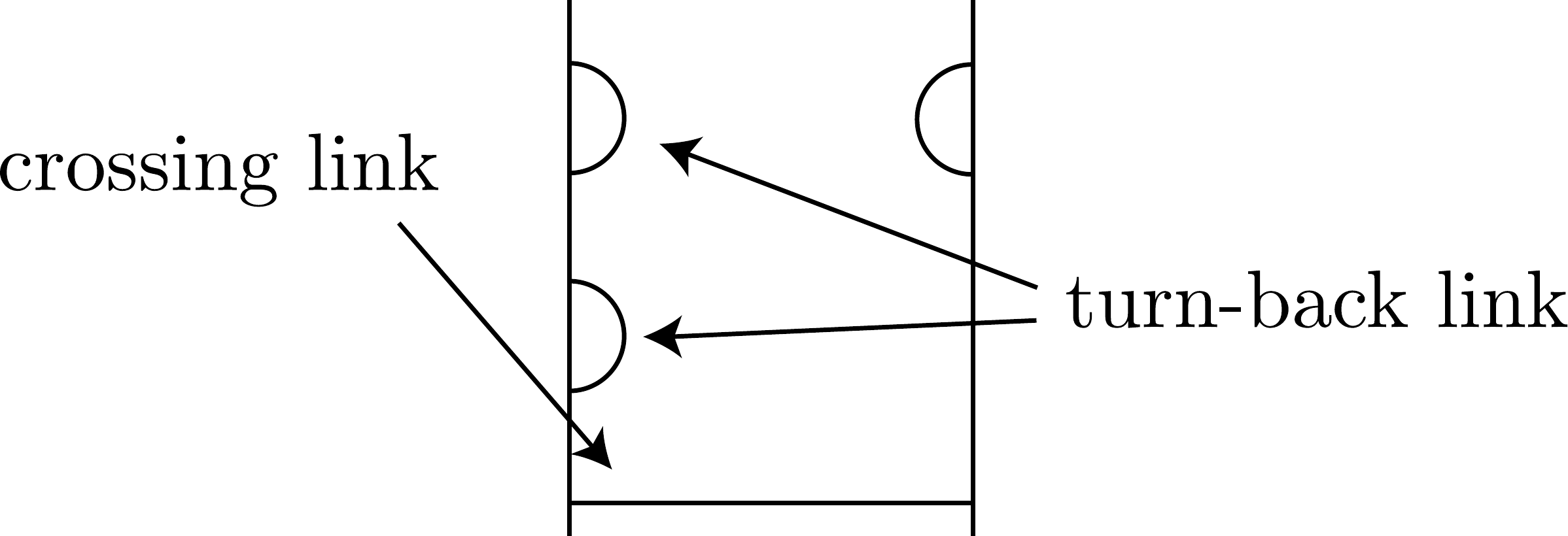} ,}}
\end{align}
which consist of two vertical lines with $n$ (resp.~$m$) \emph{nodes} anchored to the left (resp.~right) line,
and a collection of \emph{links} between the lines, connecting the nodes pairwise, and specified up to isotopy.
The source and target associated with a tangle $T \in \smash{\TL_n^m}$ are the objects $\OneVec{m}$ and $\OneVec{n}$ respectively.
In summary, we have
\begin{align} 
\Hom  \TL^1(\nu) &= \big\{ \TL_n^m \,\big|\, \text{$n, m \in \bZnn$ with $n + m = 0 \Mod 2$} \big\}.
\end{align}
The composition of two morphisms $T, U \in \Hom \smash{\TL^1(\nu)}$ is given by diagram concatenation,
where each loop formed by the concatenation is replaced by a factor of the  \emph{loop fugacity} $\nu \in \bC$: 
\begin{align}  \label{ExmpleConcat}
& \vcenter{\hbox{\includegraphics[scale=0.275]{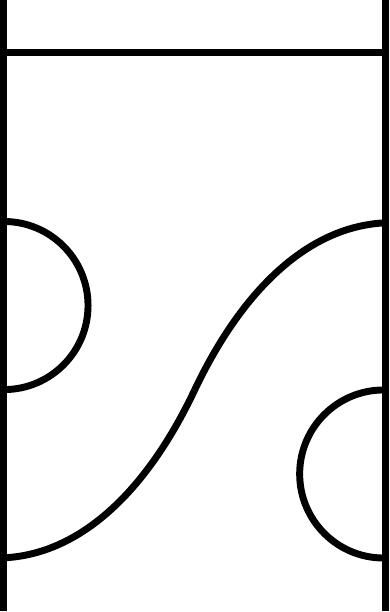}}} \quad 
\vcenter{\hbox{\includegraphics[scale=0.275]{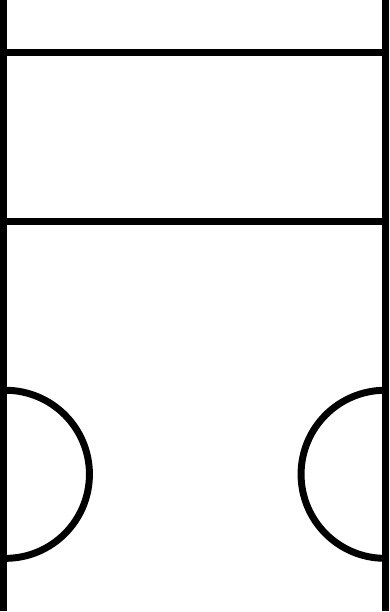}}} \quad := \quad 
\vcenter{\hbox{\includegraphics[scale=0.275]{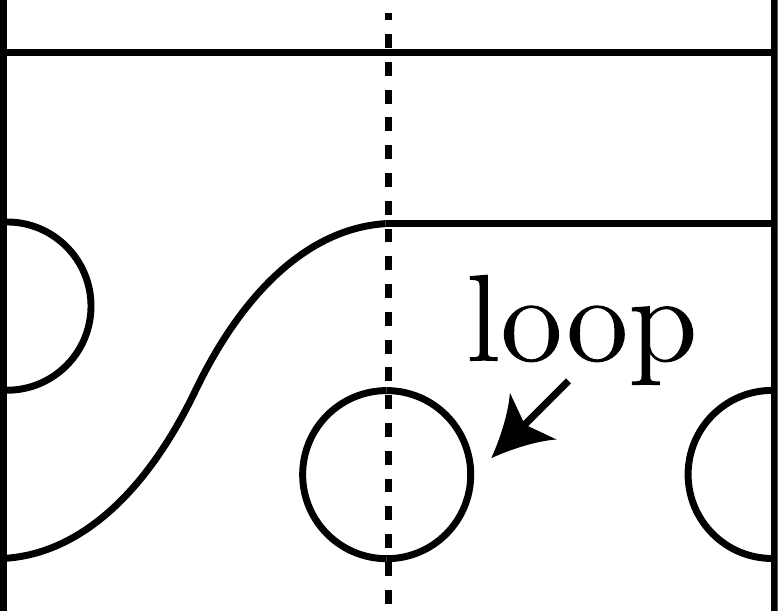}}} \quad = \quad \nu \,\, \times \,\, 
\vcenter{\hbox{\includegraphics[scale=0.275]{Figures/e-TLalgebra2.pdf}  .}}
\end{align} 
When $m=n$, we omit the superscript for $\smash{\TL_n = \TL_n^n}$.
This space is an associative unital algebra~\cite[theorem~\red{2.4}]{rsa}, 
the ``Temperley-Lieb algebra" $\TL_n(\nu)$,
generated by tangles~(\ref{LRtoGen},~\ref{LRtoUnit}) given below, and satisfying relations~\eqref{WordRelations}.


For later use, we determine a minimal collection of generators for the morphism class $\Hom \smash{\TL^1(\nu)}$.  
These constitute the \emph{left and right generators} (also known as the ``evaluation'' and ``coevaluation'' maps), 
defined as 
\begin{align} \label{LgenForm} 
\Lgen_i \quad := \quad \vcenter{\hbox{\includegraphics[scale=0.275]{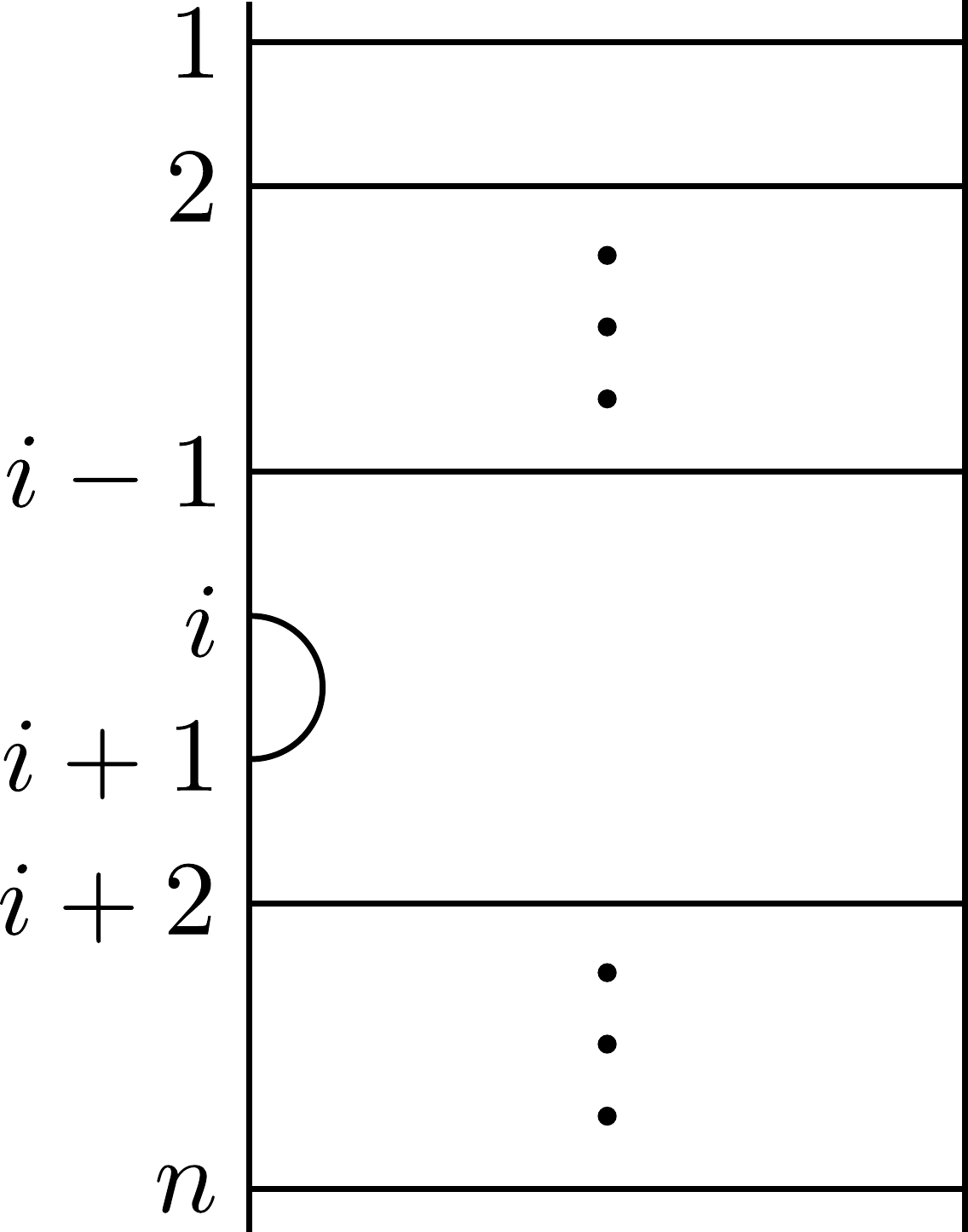}}} \;\; \in \TL_n^{n-2} ,
\qquad \qquad  
\Rgen_j\quad := \quad \vcenter{\hbox{\includegraphics[scale=0.275]{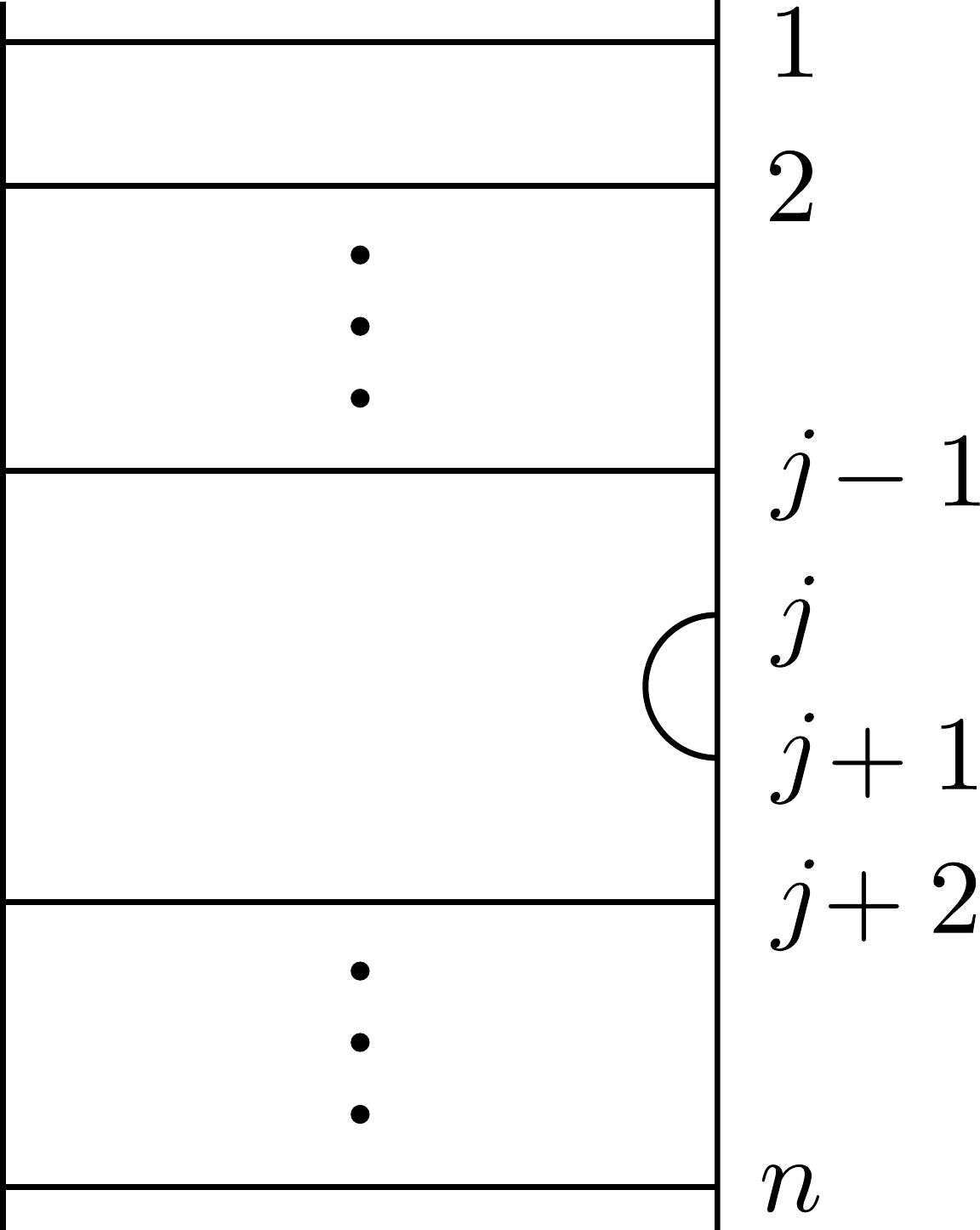}}} \;\; \in \TL_{n-2}^{n} 
\end{align} 
for all integers $n \geq 2$ and $i,j \in \{1,2,\ldots,n - 1\}$.
Now, if $T$ is an arbitrary $(n,m)$-link diagram with $s$ crossing links, we can construct $T$
by an insertion of all $\smash{\ell_L := \frac{n-s}{2}}$ left links of $T$ into the unit diagram $\mathbf{1}_{\TL_s}$
by repeated application of the left generators $\Lgen_i$, 
followed by an insertion of all $\smash{\ell_R := \frac{m-s}{2}}$ right links of $T$
by repeated application of the right generators $\Rgen_j$, that is,
\begin{align} \label{Tword} 
T = \Lgen_{i_{\ell\subscr{L}}} \Lgen_{i_{\ell\subscr{L} - 1}} \dotsm \Lgen_{i_2} \Lgen_{i_1} \mathbf{1}_{\TL_s} \Rgen_{j_1} \Rgen_{j_2} 
\dotsm \Rgen_{j_{\ell\subscr{R} - 1}} \Rgen_{j_{\ell\subscr{R}}}. 
\end{align} 
For example, 
\begin{align}
\vcenter{\hbox{\includegraphics[scale=0.275]{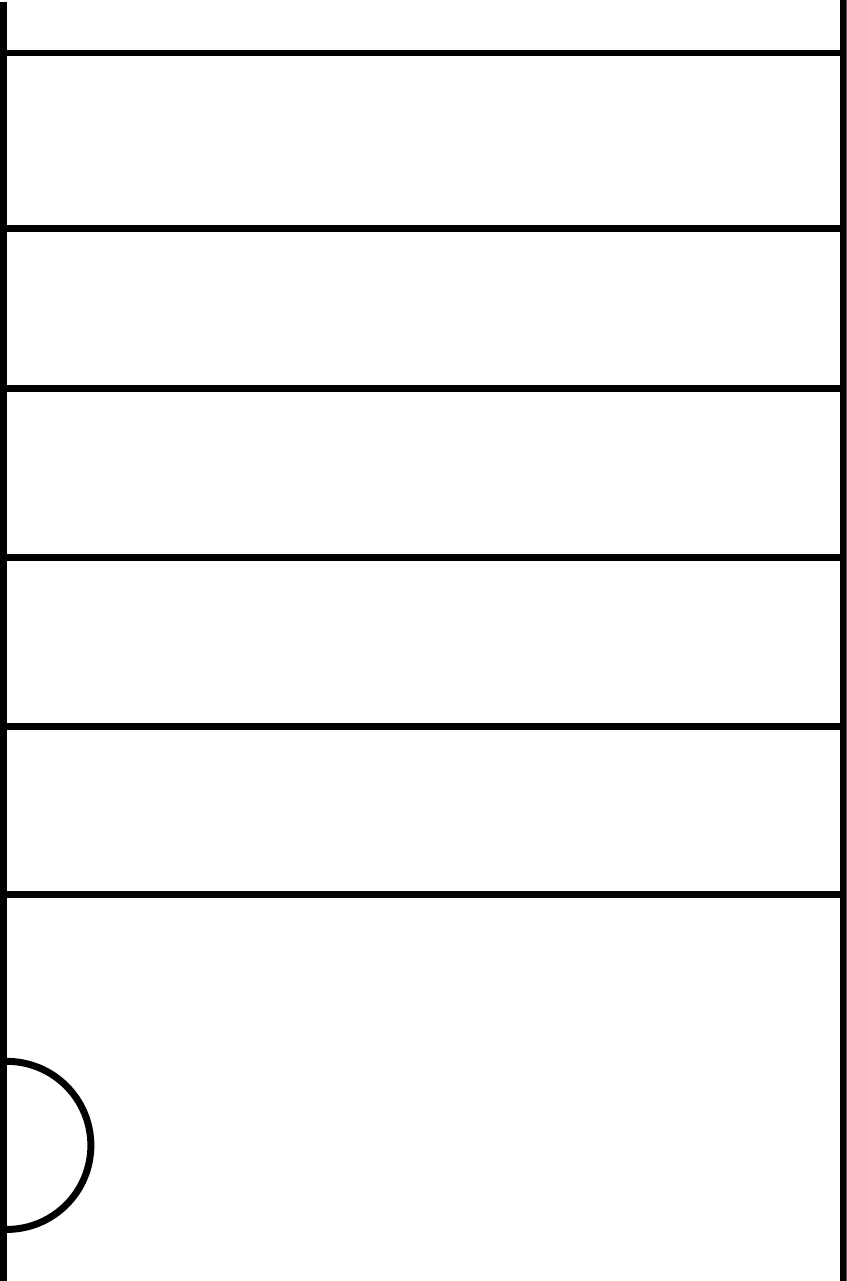}}} \;
\vcenter{\hbox{\includegraphics[scale=0.275]{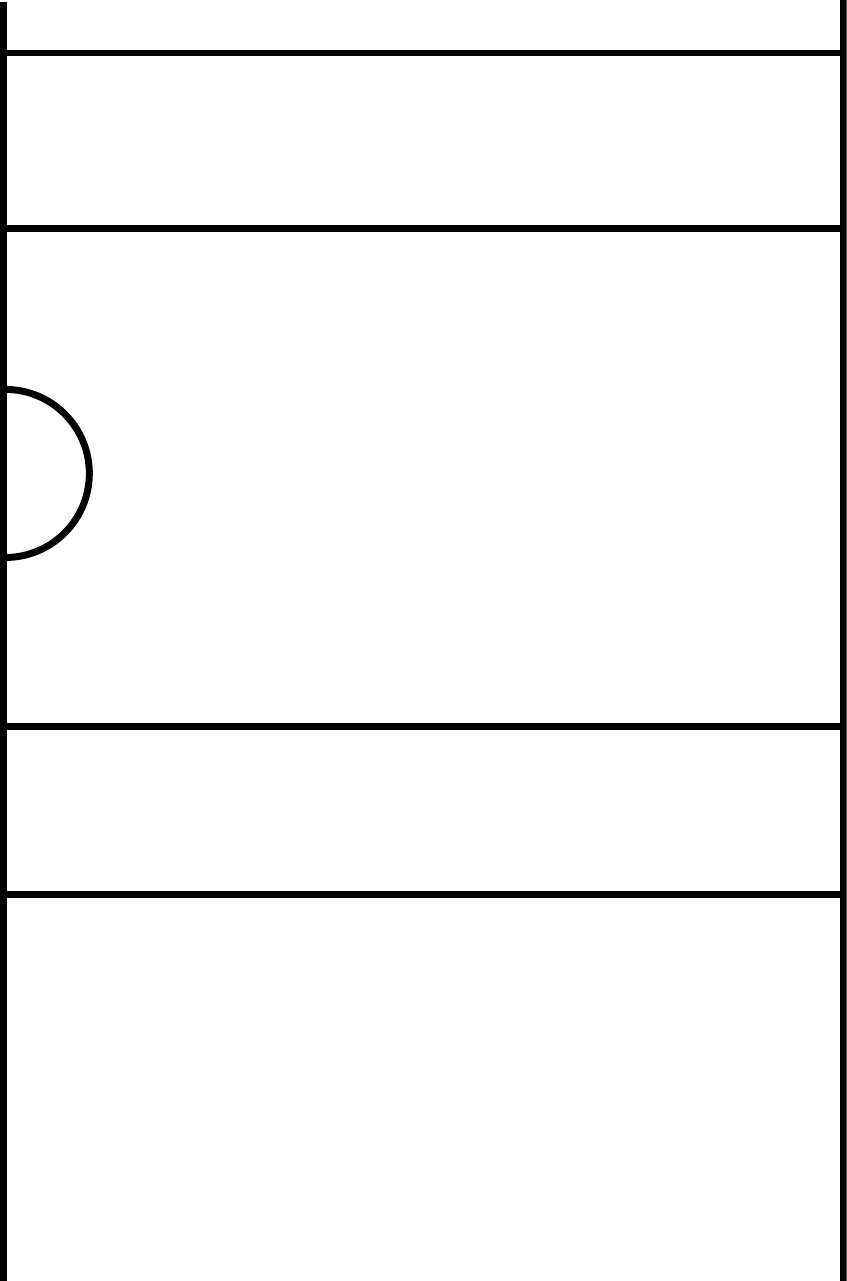}}} \;
\vcenter{\hbox{\includegraphics[scale=0.275]{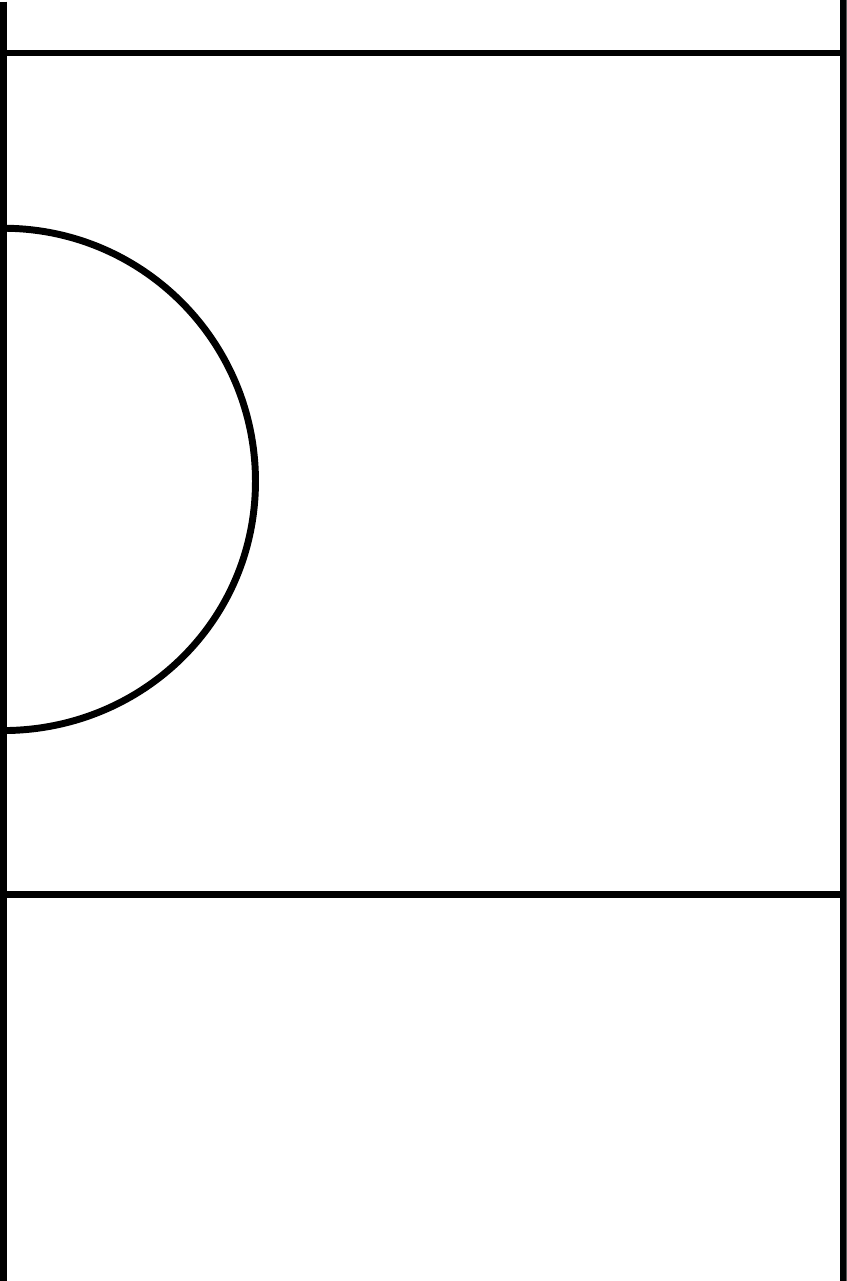}}} \;
\vcenter{\hbox{\includegraphics[scale=0.275]{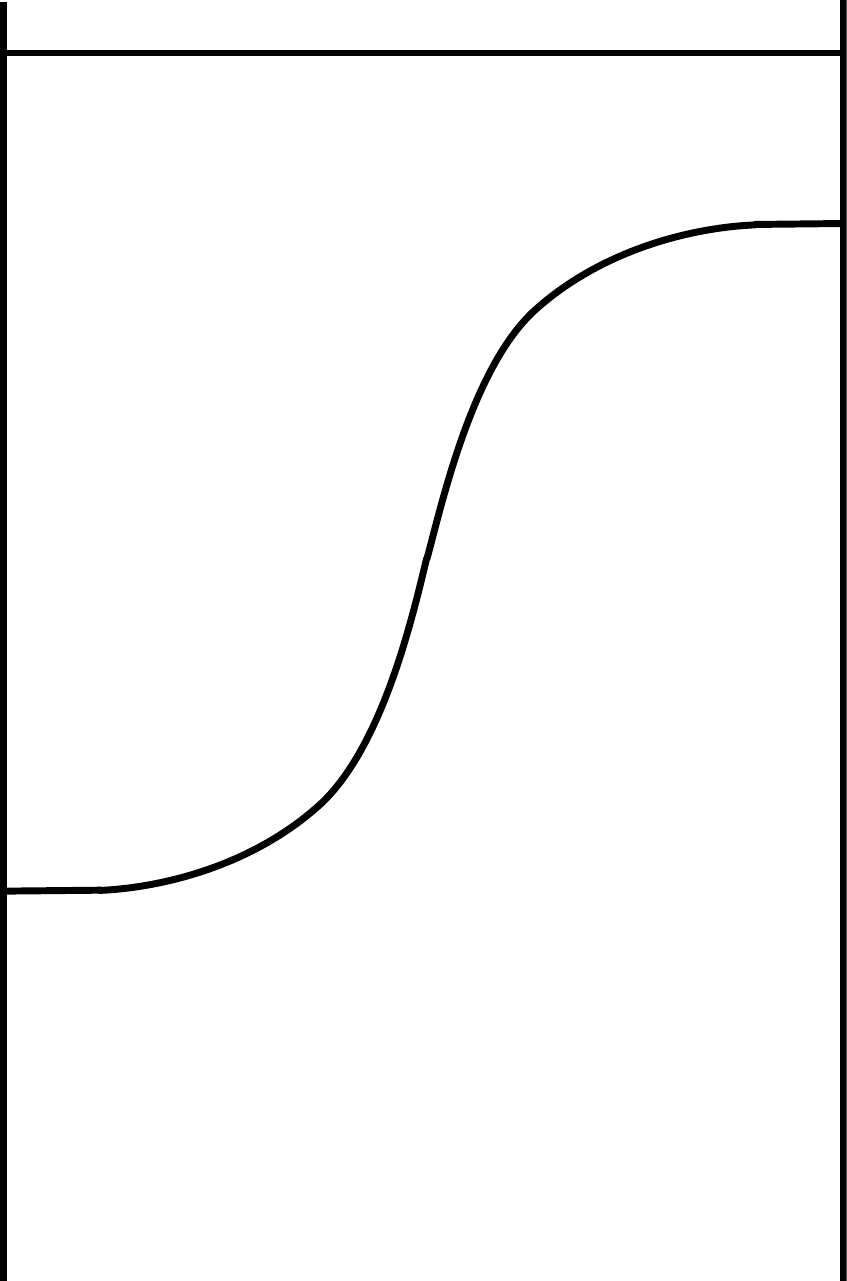}}} \;
\vcenter{\hbox{\includegraphics[scale=0.275]{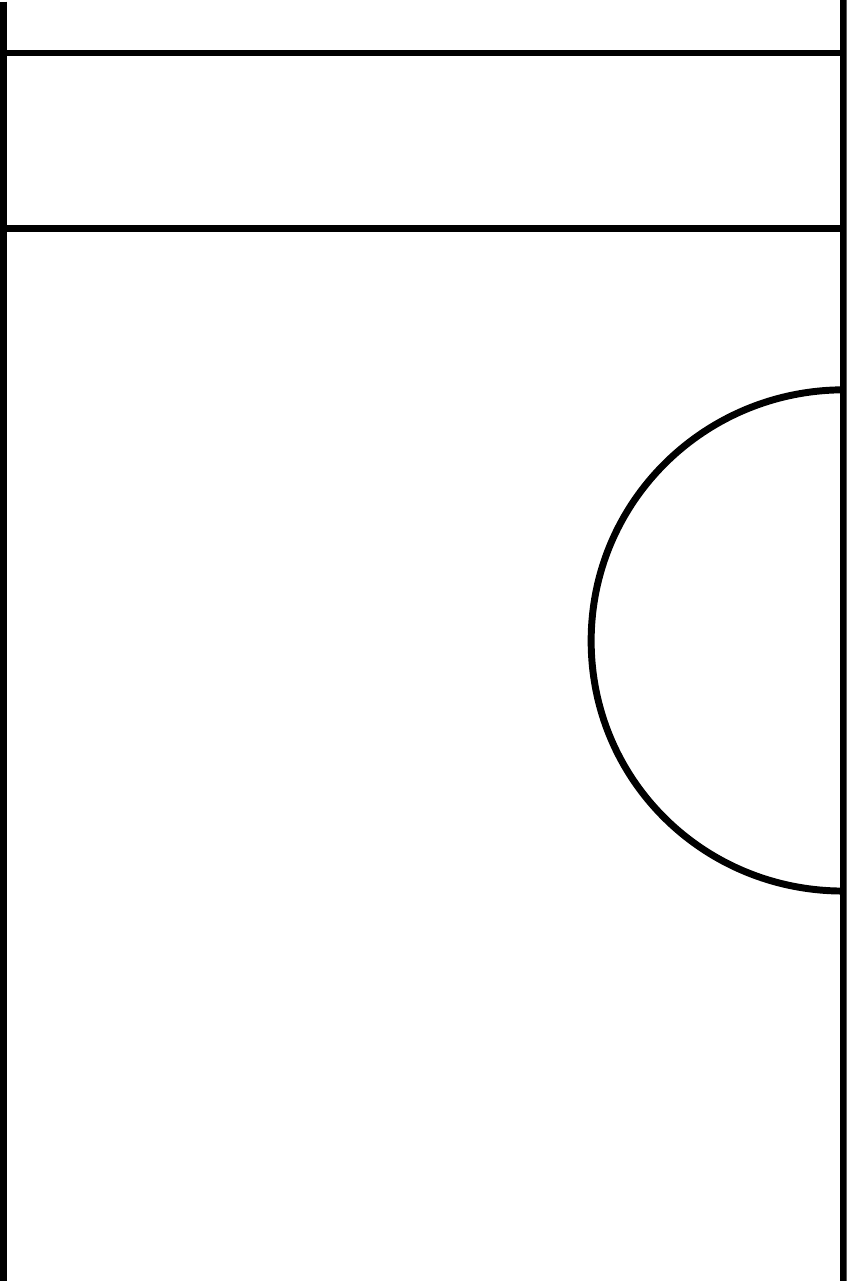}}} \;
\vcenter{\hbox{\includegraphics[scale=0.275]{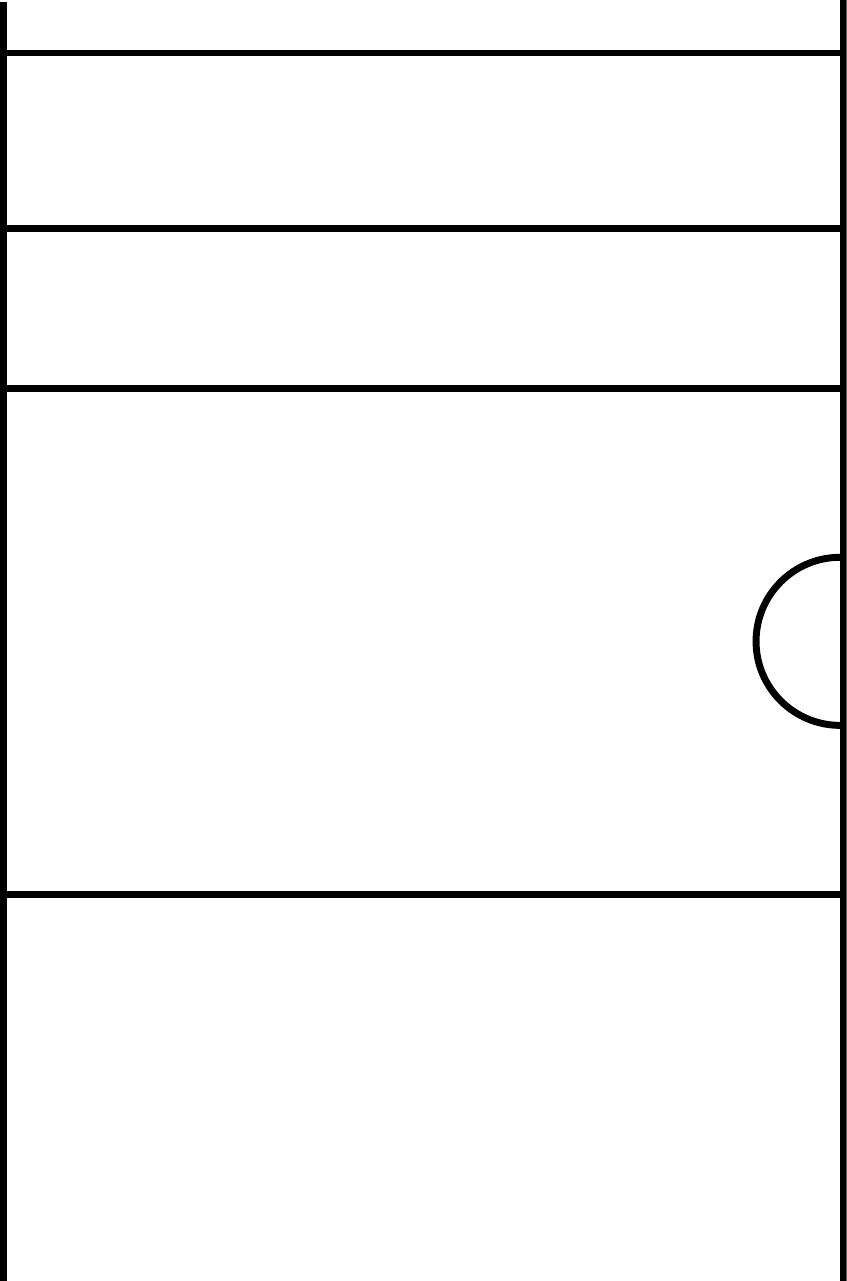}}} 
\end{align} 
gives the tangle 
\begin{align}
\vcenter{\hbox{\includegraphics[scale=0.275]{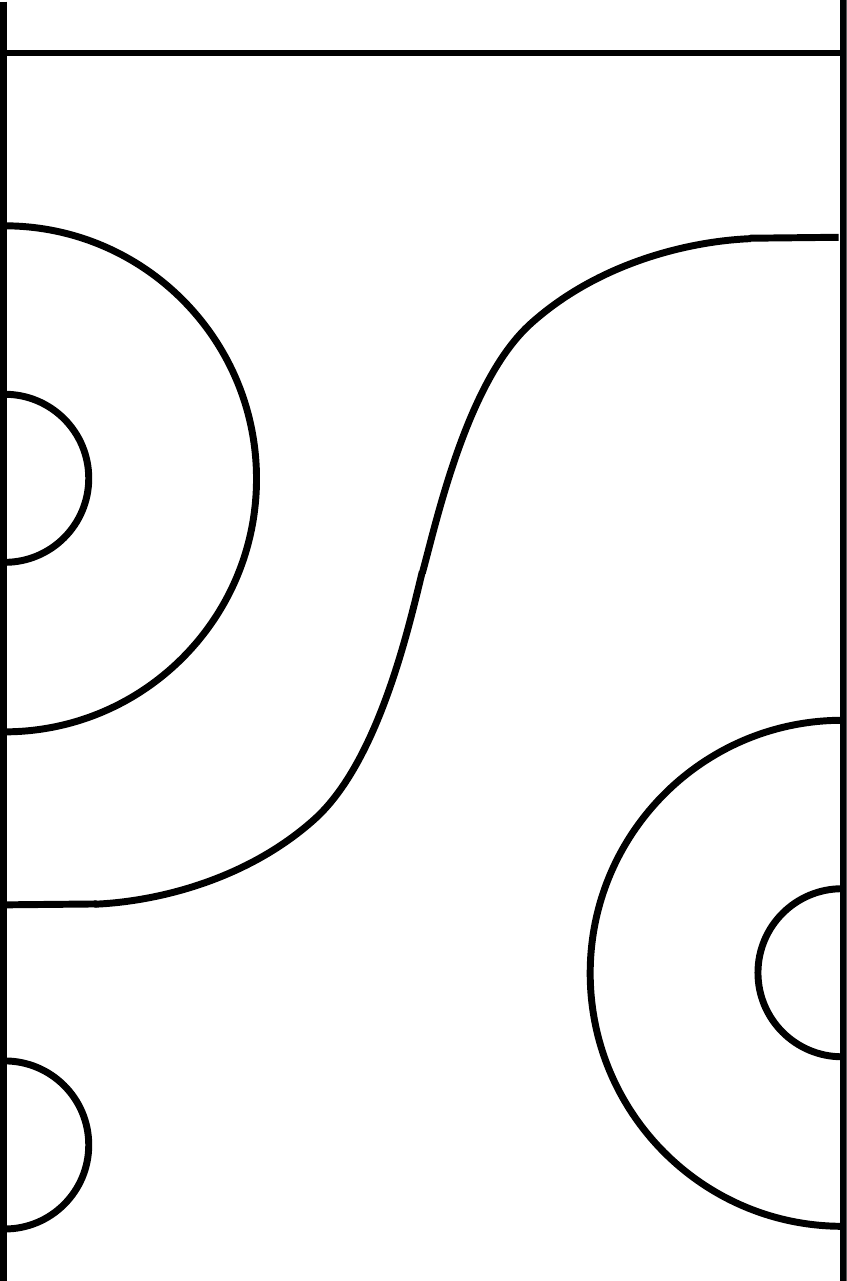}}}  \;\; \in \TL_{8}^{6} .
\end{align} 
As shown, we include the unit in the middle of the product to emphasize that $\Lgen_{i_1}$ is an $(s+2,s)$-link diagram and 
$\Rgen_{j_1}$ is an $(s,s+2)$-link diagram, in spite of the following obvious relations which imply that the unit can be dropped:
\begin{align}\label{omitId} 
\Lgen_i \mathbf{1}_{\TL_s} = \Lgen_i \qquad \text{and} \qquad \mathbf{1}_{\TL_s} \Rgen_j = \Rgen_j. 
\end{align} 
In~\eqref{Tword}, we order the left generators $\Lgen_i$ such that if the upper endpoint of one left link of $T$ is above the the upper endpoint of another left link, then 
the former is inserted before the latter, and similarly for the $\Rgen_i$.  This implies that
\begin{align} \label{ordering} 
i_1 < i_2 < \ldots < i_{\ell\subscr{L}}  \qquad \text{and} \qquad j_1 < j_2 < \ldots < j_{\ell\subscr{R}}. 
\end{align} 
We say that any product of left and right generators of the form in (\ref{Tword},~\ref{ordering}) is in \emph{standard form}.

\begin{lem} \label{StdLem} 
Each $(n,m)$-link diagram equals a unique product of left and right generators 
in standard form \textnormal{(}\ref{Tword},~\ref{ordering}\textnormal{)},
and every such product equals a unique $(n,m)$-link diagram.
\end{lem}

\begin{proof} 
This is immediate from~(\ref{Tword},~\ref{omitId}) and the chosen ordering~\eqref{ordering}.
\end{proof}

\begin{lem}  \label{RelLem} 
The following is a complete list of independent relations satisfied by the left and right generators:
\begin{align} \label{MaxRelations} 
\Rgen_j \Lgen_i = 
\begin{cases} 
\mathbf{1}_{\TL_s}, & i = j \pm 1, \\ \nu \mathbf{1}_{\TL_s}, & i = j, \\ 
\Lgen_i \Rgen_{j-2}, &  i \leq j-2, \\ \Lgen_{i-2} \Rgen_j, &  j \leq i-2, 
\end{cases} 
\qquad \qquad 
\begin{aligned} 
& \Lgen_j \Lgen_i = \Lgen_{i+2} \Lgen_j, && j \leq i, \\ 
& \Rgen_j \Rgen_{i-2} = \Rgen_i \Rgen_j, && j \leq i, 
\end{aligned} 
\end{align} 
where $s$ is the number of crossing links in $\Lgen_i$ and $\Rgen_j$.
\end{lem}

\begin{proof} 
Each relation~\eqref{MaxRelations} is easy to verify with a diagram. Also,
relations~\eqref{MaxRelations} allow to write any word formed from the right and left generators in standard form.
To see that~\eqref{MaxRelations} are all of the independent relations, we let 
\begin{align} \label{Xtra} 
\sum_{k_1, k_2, \ldots, k_l} c_{k_1, k_2, \ldots, k_l} T_{k_1} T_{k_2} \dotsm T_{k_l} = 0 , 
\qquad \text{with $c_{k_1, k_2, \ldots, k_l} \in \bC$ and $T_{k_p} \in \{ \Lgen_i, \Rgen_j \,|\, i,j \in \bZpos\}$, $k_p \in \bZpos$,} 
\end{align} 
be a relation where all terms $T_{k_1} T_{k_2} \dotsm T_{k_l}$ are in standard form.  
Because the link diagrams are linearly independent, lemma~\ref{StdLem} implies that
all of the coefficients 
$c_{k_1, k_2, \ldots, k_l}$ must vanish, so relation~\eqref{Xtra} is trivial. 
\end{proof}

The \emph{Temperley-Lieb algebra} $\TL_n(\nu)$ is the associative unital algebra  
with generating set $\{ \mathbf{1}_{\TL_n}, \Gen_1, \Gen_2, \ldots, \Gen_{n-1} \}$ subject to relations~\eqref{WordRelations}~\cite{tl, vj}.
These generators have the following diagrammatic form:
\begin{align} \label{LRtoGen} 
\Lgen_i \Rgen_i 
\quad = \quad 
\vcenter{\hbox{\includegraphics[scale=0.275]{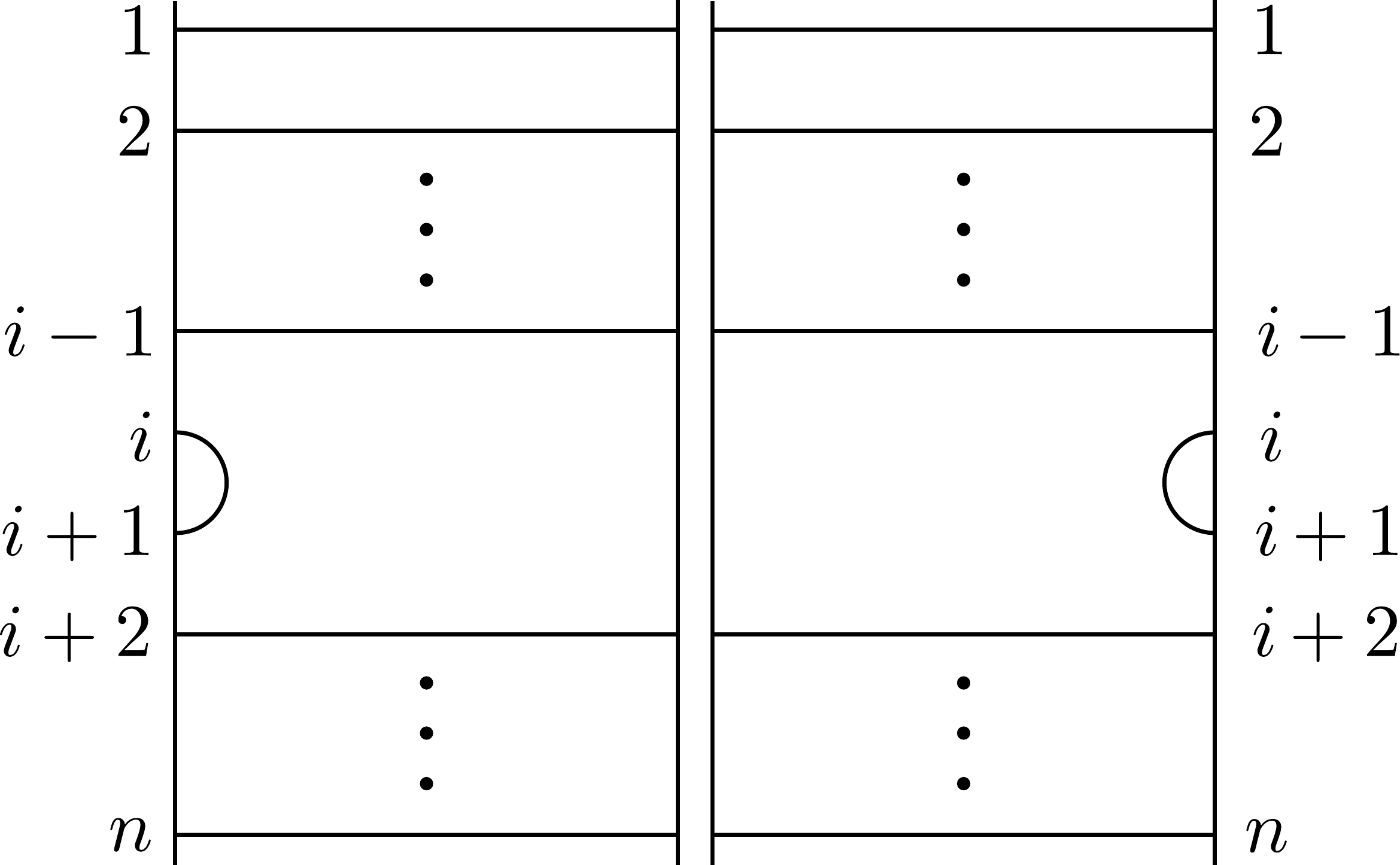}}} 
\quad = \quad 
\vcenter{\hbox{\includegraphics[scale=0.275]{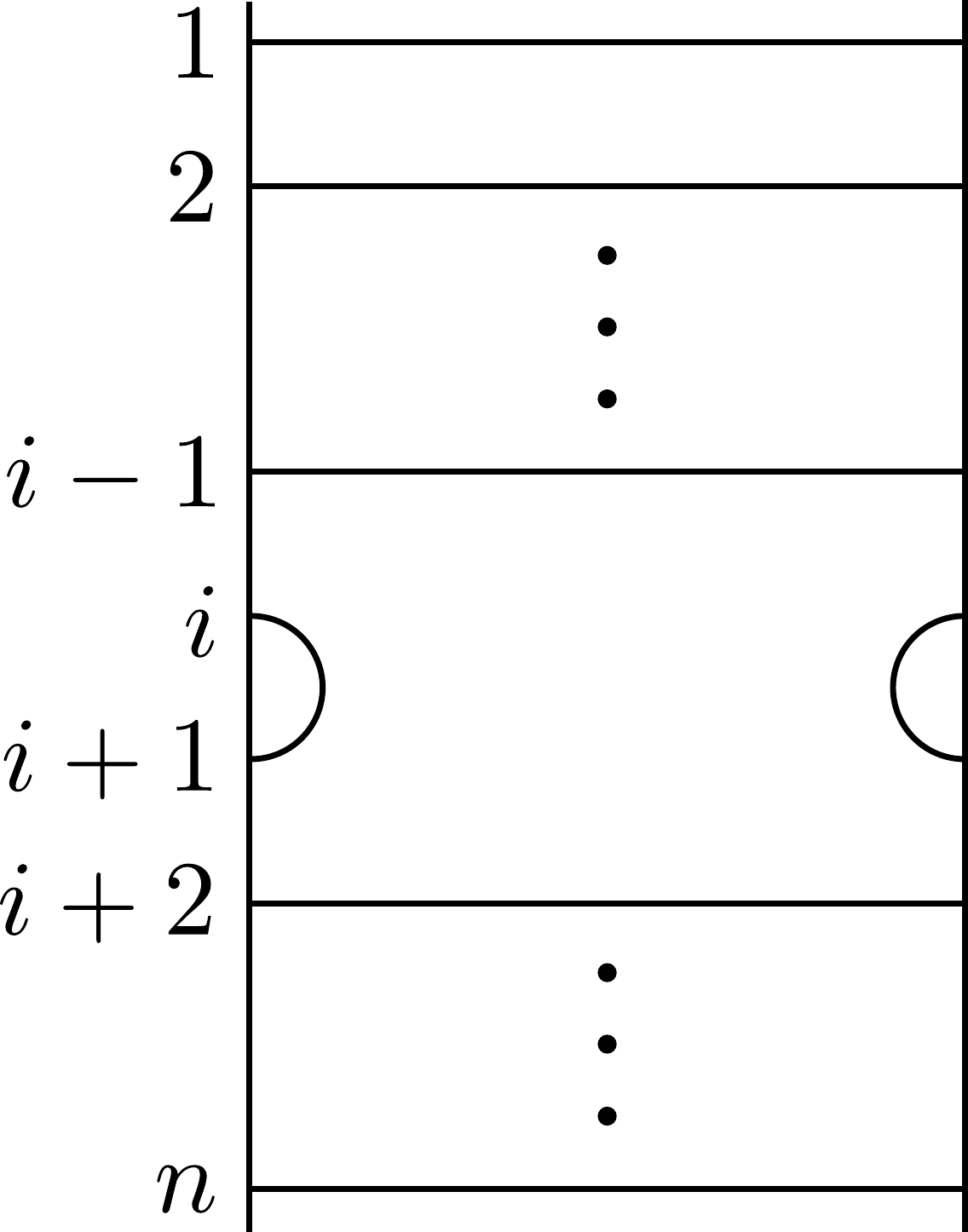}}} \quad =: \quad \Gen_i  
\end{align}
for all $i \in \{1,2, \ldots, n-1\}$, and the unit is
\begin{align} \label{LRtoUnit} 
\hspace*{-3mm}
\Rgen_{i-1} \Lgen_i  
\quad = \quad 
\vcenter{\hbox{\includegraphics[scale=0.275]{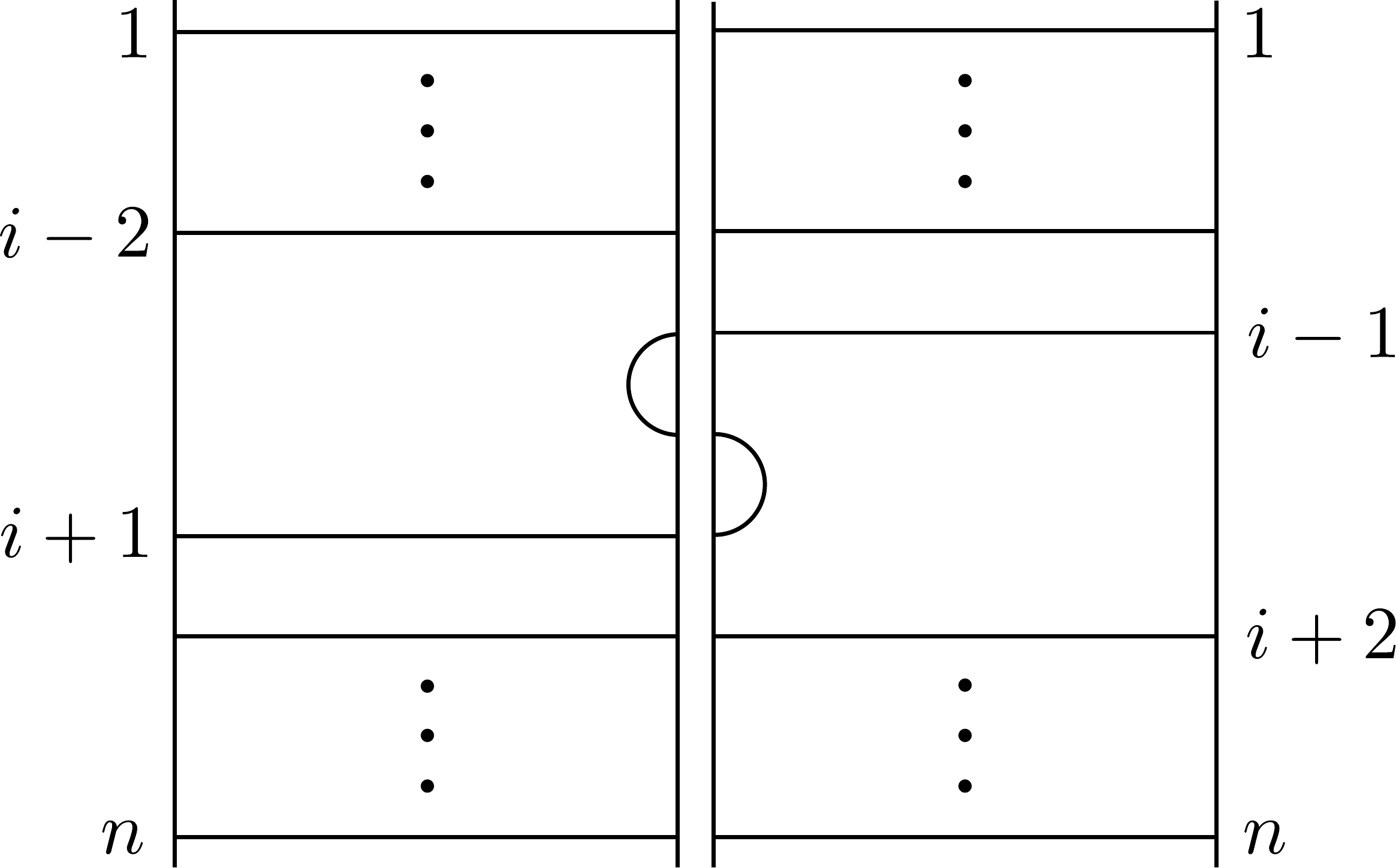}}} 
\quad = \quad 
\vcenter{\hbox{\includegraphics[scale=0.275]{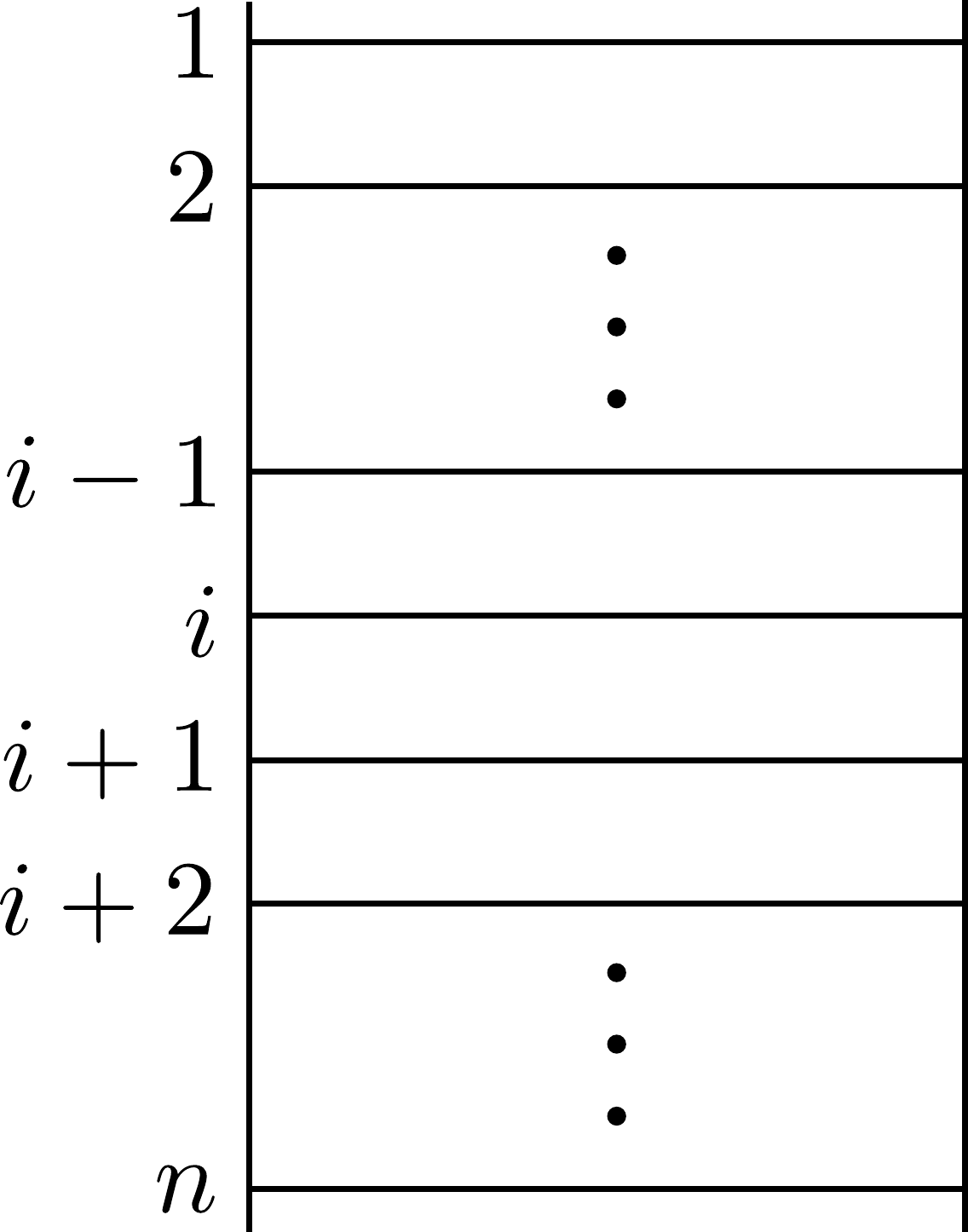}}} \quad =: \quad \mathbf{1}_{\TL_n} 
\end{align}
Relations~\eqref{WordRelations} of the Temperley-Lieb algebra $\TL_n(\nu)$ follow from~(\ref{MaxRelations},~\ref{LRtoGen})
with the diagram concatenation rules.

\bigskip

{\bf Link states.}
Standard modules 
consisting of link states
are building blocks for representations of the Temperley-Lieb algebra $\TL_n(\nu)$. 
Given an $n$-link diagram with $s$ crossing links, we create an \emph{$(n,s)$-link pattern}
by dividing the link diagram vertically in half, discarding the right half, and rotating the left half by $\pi/2$ radians:
\begin{align} 
\nonumber
\vcenter{\hbox{\includegraphics[scale=0.275]{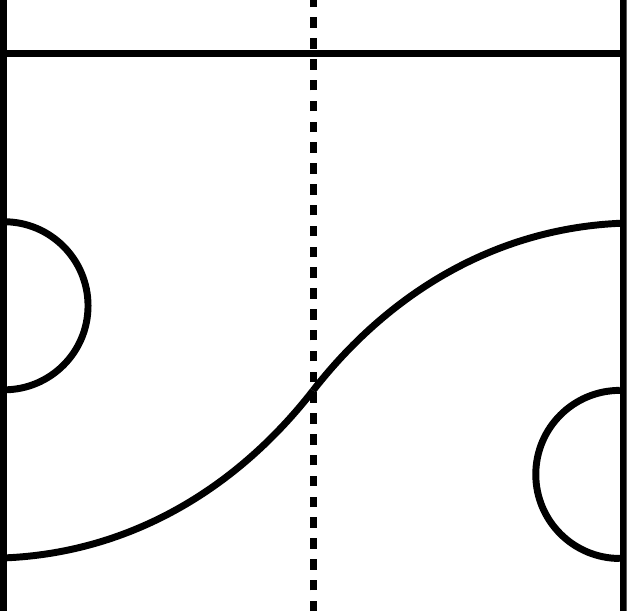}}} \qquad \qquad & \longmapsto \qquad \qquad
\vcenter{\hbox{\includegraphics[scale=0.275]{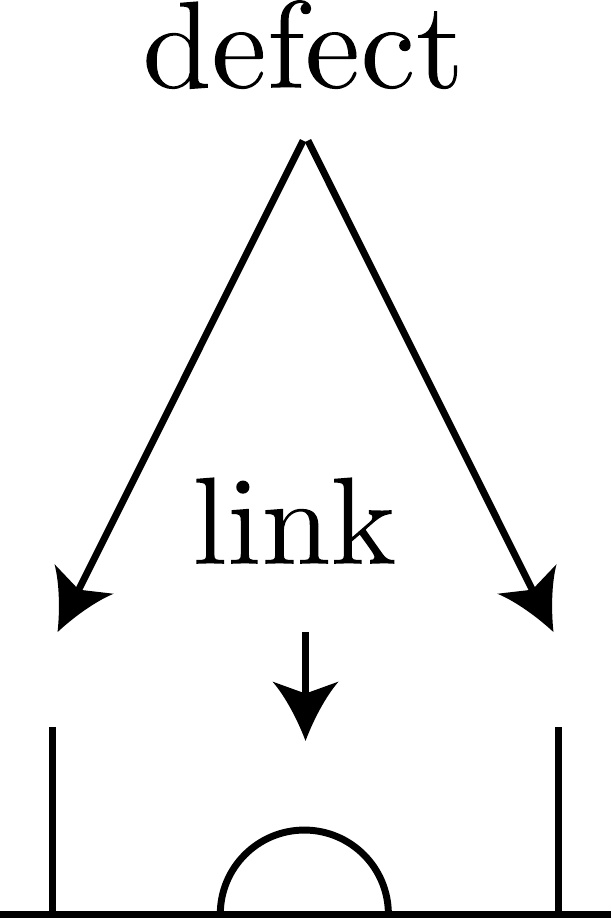}}} \\[1em]
\label{CutInHalf}
\text{link diagram} \qquad \qquad & \longmapsto \qquad \qquad \text{link pattern}.
\end{align}
We call the broken links in the $(n,s)$-link pattern \emph{defects}.  
We denote the set of $(n,s)$-link patterns by $\smash{\LP_n\super{s}}$, and let
\begin{align}
\LP_n := \bigcup_{s \, \in \, \DefectSet_n} \LP_n\super{s} , 
\qquad \text{where} \quad \DefectSet_n = \{ n \; \text{mod} \; 2, \; (n \; \text{mod} \; 2) + 2, \; \ldots, \; n \} ,
\end{align}
denote the set of \emph{$n$-link patterns} with any number of crossing links. 
We call a formal linear combination of $(n,s)$-link patterns with complex coefficients an \emph{$(n,s)$-link state},
and let $\smash{\LS_n\super{s}}$ denote the complex vector space of $(n,s)$-link states.

We endow the space $\smash{\LS_n\super{s}}$
with a $\TL_n(\nu)$-action via the following diagram concatenation recipe. 
Given an $n$-link diagram $T$ and an $(n,s)$-link pattern $\alpha \in \smash{\LP_n\super{s}}$, 
the latter rotated $-\pi/2$ radians, we concatenate $T$ to the left of $\alpha$,
remove the $k \geq 0$ loops formed by this concatenation, and multiply the result by $\nu^{k}$:
\begin{align} \label{loopex}
& \vcenter{\hbox{\includegraphics[scale=0.275]{Figures/e-TLalgebra1.pdf}}} \quad 
\vcenter{\hbox{\includegraphics[scale=0.275]{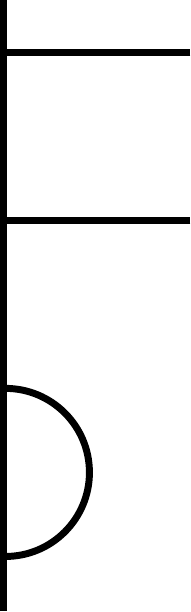}}} \quad
= \quad \vcenter{\hbox{\includegraphics[scale=0.275]{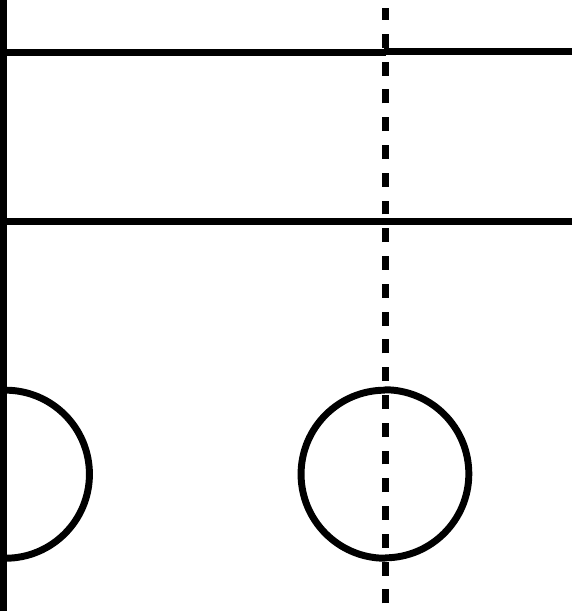}}}  \quad
 = \quad \nu \,\, \times \,\, 
\vcenter{\hbox{\includegraphics[scale=0.275]{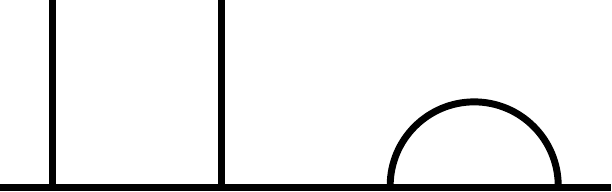} .}} 
\end{align}
Importantly, we set diagrams containing \emph{turn-back paths} to zero,
so $\TL_n(\nu)$ preserves the number $s$ of defects:
\begin{align} \label{turnbackex} 
& \vcenter{\hbox{\includegraphics[scale=0.275]{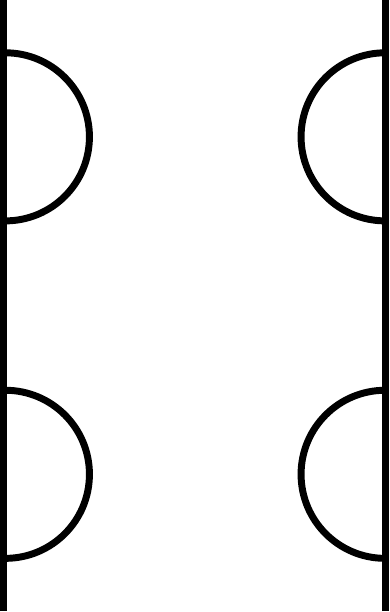}}} \quad 
\vcenter{\hbox{\includegraphics[scale=0.275]{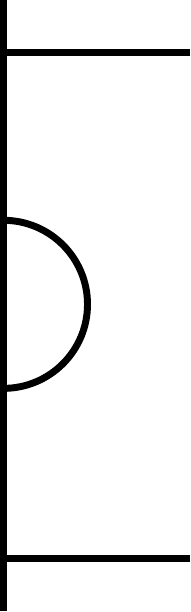}}} \quad
= \quad \vcenter{\hbox{\includegraphics[scale=0.275]{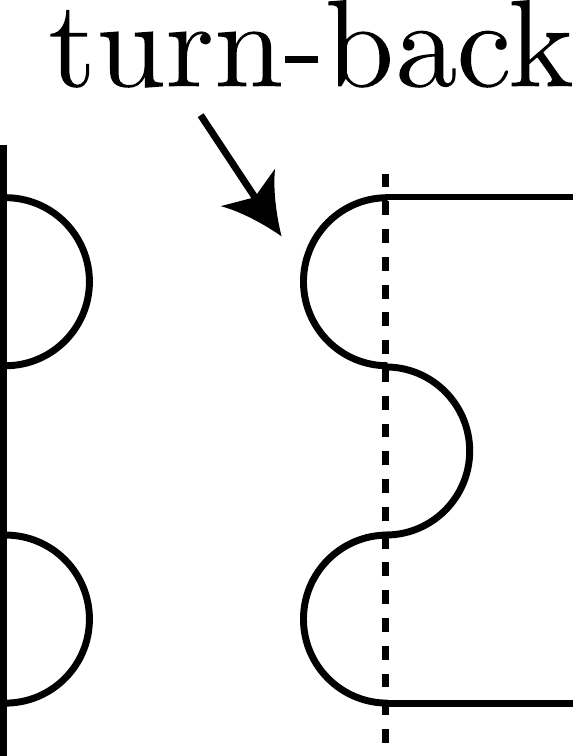}}} \quad
= \quad 0 \,\, \times \,\, 
\vcenter{\hbox{\includegraphics[scale=0.275]{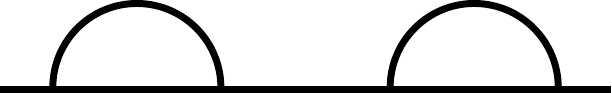}}} 
\quad = \quad 0 .
\end{align}
Bilinear extension of this recipe defines a $\TL_n(\nu)$-module structure on 
$\smash{\LS_n\super{s}}$, and we thus call $\smash{\LS_n\super{s}}$ a $\TL_n(\nu)$-\emph{standard module}. 
We also define the $\TL_n(\nu)$-\emph{link state module} to be the direct sum of all of the standard modules,
\begin{align} 
\LS_n :=  \bigoplus_{s \, \in \, \DefectSet_n} \LS_n\super{s} .
\end{align}
The algebra $\TL_n(\nu)$ is semisimple if and only if the parameter $q \in \bC^\times$ 
that determines the fugacity $\nu$ via~\eqref{fugacity} satisfies
\begin{align}  
\text{either} \qquad q \in \{\pm1\}, \qquad \text{or} 
\qquad n < \pmin(q), \qquad \text{or} \qquad q \in \{\pm \ii\} \text{ if $n$ is odd} ,
\end{align}
see, e.g.,~\cite[theorem~\red{8.1}]{rsa} and~\cite[corollary~\red{6.10}]{fp3a}.
In this case, the collection $\smash{\{ \LS_n\super{s} \, | \, s \in \DefectSet_n \}}$ 
is the complete set of non-isomorphic simple $\TL_n(\nu)$-modules. 
If $\TL_n(\nu)$ is not semisimple, some of its standard modules $\smash{\LS_n\super{s}}$ are not simple 
(but still indecomposable), but instead, certain quotients $\smash{\Quo_\multii\super{s}}$ of the standard modules are simple,
see~\eqref{QuoSp}.

\bigskip

{\bf Valenced tangles and link patterns.}
Next, for two multiindices 
as in~(\ref{MultiindexNotation},~\ref{ndefn}), we
consider the set of $(\multii, \multiii)$-\emph{valenced link diagrams},
\begin{align} \label{ValDiagEx}
\vcenter{\hbox{\includegraphics[scale=0.275]{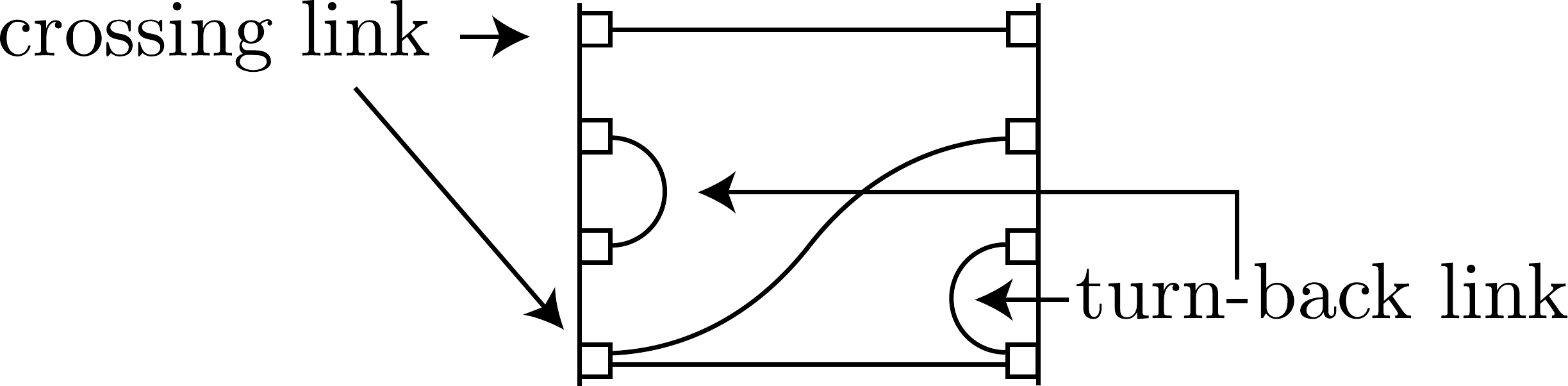} ,}}
\end{align}
which consist of two vertical lines with $\np_\multii$ (resp.~$\np_\multiii$) \emph{valenced nodes} 
\begin{align}\label{boxval} 
\vcenter{\hbox{\includegraphics[scale=0.275]{Figures/e-valenced_box.pdf}}} 
\qquad \qquad \Longleftrightarrow \qquad \qquad 
\begin{rcases}
\vcenter{\hbox{\includegraphics[scale=0.275]{Figures/e-valenced_node1.pdf}}} \quad
\end{rcases} \; s 
\end{align} 
anchored to the left (resp.~right) line
and a collection of \emph{links} between the lines, connecting the nodes
in such a way that, first, for each $i \in \{ 1,2,\ldots,\np_\multii \}$ (resp.~$j \in \{ 1,2,\ldots,\np_\multiii \}$),
exactly $\sIndex_i$ (resp.~$p_j$) endpoints reside at the $i$:th left (resp.~$j$:th right) node counted from the top,
and second, the two endpoints of each link are distinct
(that is, we exclude ``loop-links'' that start and end at the same node; see~\cite[section~\red{2}]{fp3a} for details).
Each link is specified up to isotopy. 
We let $\TL_\multii^\multiii$ denote the complex vector space of \emph{$(\multii, \multiii)$-valenced tangles}, that is,
formal linear combinations of $(\multii, \multiii)$-valenced link diagrams.
In the special case when $\multii = \OneVec{n}$ and $\multiii = \OneVec{m}$, 
we omit the arrow in superscripts and subscripts, writing $\smash{\TL_n^m}$. 
Also, when the superscript and subscript are equal, we write
$\TL_\multii = \smash{\TL_\multii^\multii}$.

For a multiindex $\multii$ as in~(\ref{MultiindexNotation},~\ref{ndefn}) and an integer $s$, we denote by
$\smash{\LP_\multii\super{s}}$ the set of \emph{$(\multii,s)$-valenced link patterns},
\begin{align} \label{ValPattEx}
\vcenter{\hbox{\includegraphics[scale=0.275]{Figures/e-ValencedLinkState_with_defect_and_link.pdf} ,}} 
\end{align}
consisting of $\frac{1}{2}(\Summed_\multii-s)$ \emph{links} and $s$ \emph{defects} (specified up to isotopy)
attached to $\np_\multii$ valenced nodes on and above a horizontal line,
in such a way that, first, for each $i \in \{ 1,2,\ldots,\np_\multii \}$, the $i$:th node counted from the left has exactly $\sIndex_i$ strands attached to it,
second, the two endpoints of each link are distinct, and third, each defect has one endpoint in the unbounded component of 
the complement of the links and the horizontal line and another endpoint among the $\np_\multii$ nodes.
We let $\smash{\LS_\multii\super{s} := \Span \LP_\multii\super{s}}$ denote the complex vector space of \emph{$(\multii,s)$-valenced link states}.
We remark that by~\cite[lemma~\red{4.1}]{fp3a}, $\smash{\LP_\multii\super{s}}$ is nonempty 
(and thus $\smash{\LS_\multii\super{s}}$ is nontrivial) if and only if $s \in \DefectSet_\multii$. 
We also write 
\begin{align} \label{LSDirSum} 
\LP_\multii := \bigcup_{s \, \in \, \DefectSet_\multii} \LP_\multii\super{s} 
\qquad \qquad \text{and} \qquad \qquad 
\LS_\multii :=  \bigoplus_{s \, \in \, \DefectSet_\multii} \LS_\multii\super{s} 
\end{align}
respectively for the collections of \emph{$\multii$-valenced link patterns} and \emph{$\multii$-valenced link states}.
When $\multii = \OneVec{n}$, we omit the arrow in superscripts and subscripts: 
$\smash{\LS_n\super{s}}$ and $\smash{\LP_n\super{s}}$ stand for the spaces of $(n,s)$-link states and $(n,s)$-link patterns.

We also define an analogous space $\smash{\LSBar_\multii\super{s}}$ of valenced link states,
spanned by $(\multii,s)$-valenced link patterns $\alphaBar \in \smash{\LPBar_\multii\super{s}}$,  
\begin{align} 
\vcenter{\hbox{\includegraphics[scale=0.275]{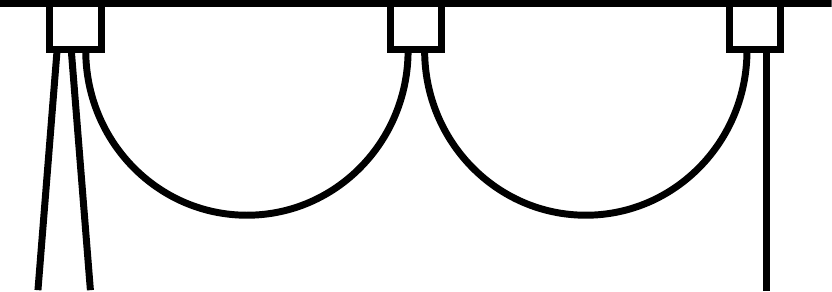} .}} 
\end{align}
We note that space $\smash{\LSBar_\multii\super{s}}$ is naturally isomorphic with $\smash{\LS_\multii\super{s}}$,
via the linear extension of the following reflection map: 
we let $\smash{\alpha^* \in \LPBar_\multii\super{s}}$ denote the reflection of $\alpha \in \smash{\LP_\multii\super{s}}$ about a horizontal axis:
\begin{align}\label{StarMapAlpha} 
\alpha \quad = \quad \vcenter{\hbox{\includegraphics[scale=0.275]{Figures/e-LinkPattern3_valenced.pdf}}}
\quad  \in \; \smash{\LP\sub{3,2,2}\super{3}}
\qquad \qquad \Longrightarrow \qquad \qquad
\alpha^* \quad := \quad \vcenter{\hbox{\scalebox{1}[-1]{\includegraphics[scale=0.275]{Figures/e-LinkPattern3_valenced.pdf}}}} 
\quad  \in \; \smash{\LPBar\sub{3,2,2}\super{3}} ,
\end{align}
and we similarly define 
$\alphaBar^* \in \smash{\LP_\multii\super{s}}$ for  $\alphaBar \in \smash{\LPBar_\multii\super{s}}$.
We also define the union $\LPBar_\multii$ and direct sum $\LPBar_\multii$ similarly to~\eqref{LSDirSum}.

\bigskip

{\bf Projectors.}
Recalling definition~\eqref{MinPower} of $\pmin(q)$, 
we define the \emph{Jones-Wenzl projector of size} $s \in \{0, 1, \ldots, \pmin(q) - 1 \}$ to be the nonzero tangle $\WJProj\sub{s} \in \TL_s$ 
uniquely determined by the following two properties~\cite{vj, hw, kl}:
\begin{enumerate}[leftmargin=*, label = P\arabic*., ref = P\arabic*]
\itemcolor{red}
\item \label{wj1} $\smash{\WJProj\sub{s}^2} = \smash{\WJProj\sub{s}}$, and
\item \label{wj2}
$\smash{\Gen_i} \smash{\WJProj\sub{s}} = \smash{\WJProj\sub{s}} \smash{\Gen_i} = 0$ for all $i \in \{1, 2, \ldots, s - 1\}$.
\end{enumerate}
For example, we have
\begin{align}\label{WJ2} 
\WJProj\sub{0} = \text{the empty tangle}, \qquad \qquad
\WJProj\sub{1} = \mathbf{1}_{\TL_1}, \qquad \qquad \text{and} \qquad \qquad
\WJProj\sub{2} = \mathbf{1}_{\TL_2} - \nu^{-1} \Gen_1 .
\end{align} 
We represent the tangle $\smash{\WJproj\sub{s}}$ as the empty \emph{projector box} 
\begin{align} \label{ProjBoxDiag}
s
\begin{cases}
\quad \vcenter{\hbox{\includegraphics[scale=0.275]{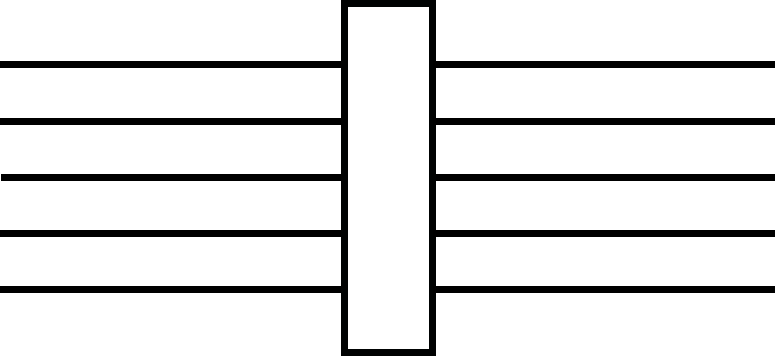}}}
\end{cases}
\; = \quad 
\vcenter{\hbox{\includegraphics[scale=0.275]{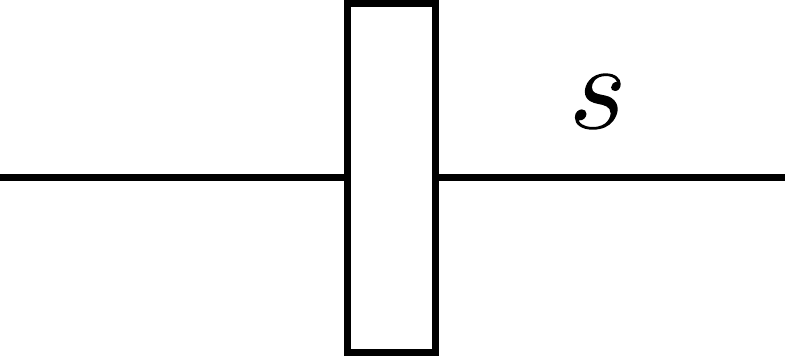} .}}
\end{align}
We can insert this projector box within any larger link diagram,
and abusing notation, we let the symbol $\WJProj\sub{s} \in \TL_n$ also denote this tangle 
in $\TL_n$ with $i = 1$.  Then the various projectors 
$\WJProj\sub{1} = \mathbf{1}_{\TL_n},$ $ \WJProj\sub{2}, \ldots,\WJProj\sub{n} \in \TL_n$
satisfy the following recursive relation~\cite{hw, vj, kl}: 
\begin{align}\label{wjrecursion} 
\smash{\WJproj\sub{1}} = \mathbf{1}_{\TL_n} , \qquad 
\smash{\WJproj\sub{s+1}} = \smash{\WJproj\sub{s}} + \left( \frac{[s]}{[s+1]} \right) 
\smash{\WJproj\sub{s}} \smash{\Gen_s} \smash{\WJproj\sub{s}}  \quad 
\text{for all $s \in \{1, 2, \ldots, n - 1\}$.}
\end{align}
In terms of diagrams, recursion relation~\eqref{wjrecursion} reads
\begin{align}\label{RecursionDiagram}
\vcenter{\hbox{\includegraphics[scale=0.275]{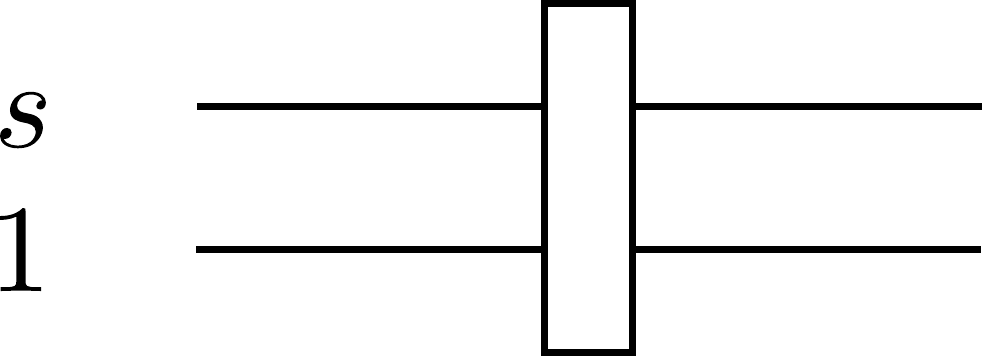}}} \quad = \quad 
\vcenter{\hbox{\includegraphics[scale=0.275]{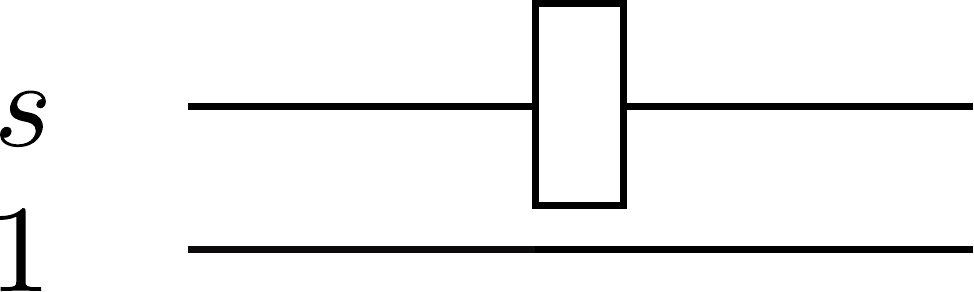}}} \quad + \quad 
\frac{[s]}{[s+1]} \,\, \times \,\, \vcenter{\hbox{\includegraphics[scale=0.275]{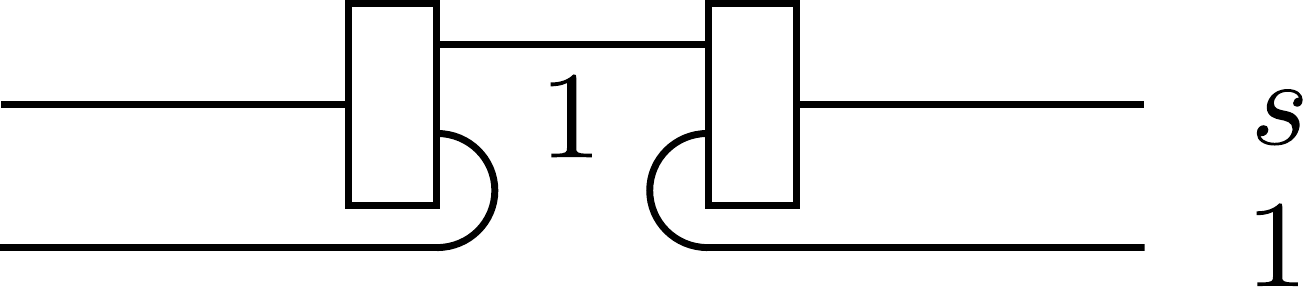} .}} 
\end{align}
With the graphical representation, we can write properties~\ref{wj1} and~\ref{wj2} in the following form
(where property~\eqref{ProjectorID1} below is slightly stronger than property~\ref{wj1} above)
\begin{align} 
\label{ProjectorID1}
\tag{\ref{wj1}$\red{'}$} 
\vcenter{\hbox{\includegraphics[scale=0.275]{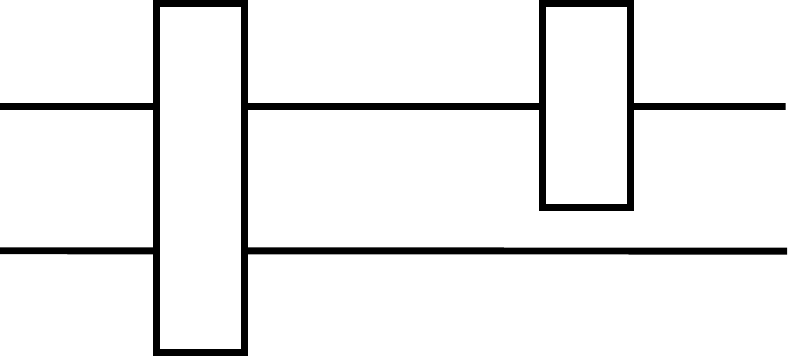}}} 
\quad = \quad \vcenter{\hbox{\includegraphics[scale=0.275]{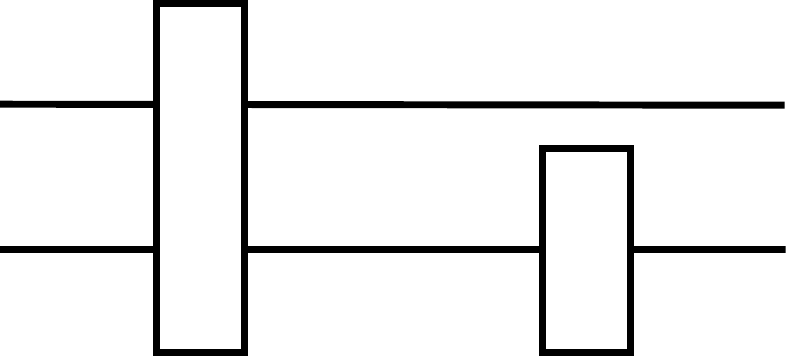}}} 
\quad & = \quad \vcenter{\hbox{\includegraphics[scale=0.275]{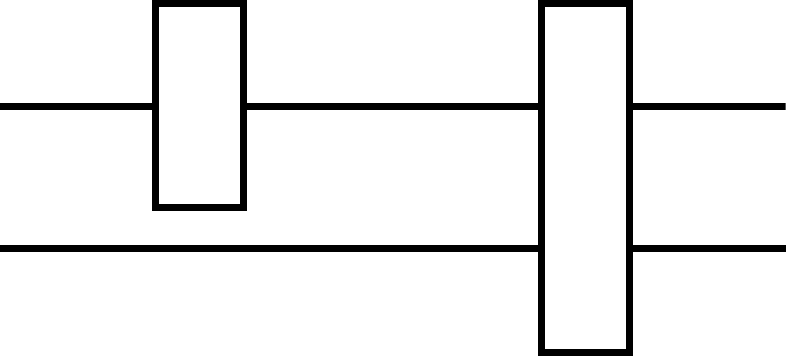}}} 
\quad = \quad \vcenter{\hbox{\includegraphics[scale=0.275]{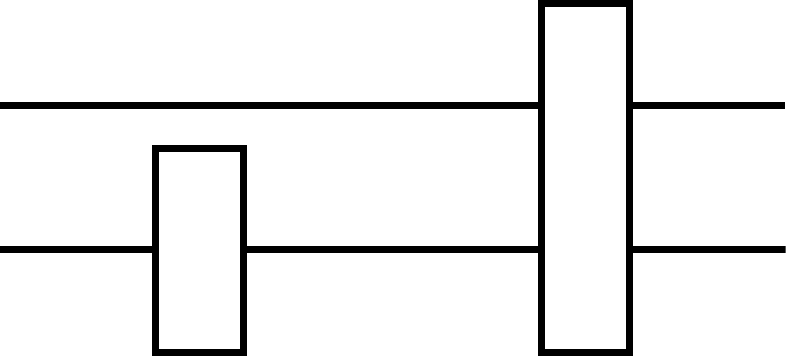}}} 
\quad = \quad \vcenter{\hbox{\includegraphics[scale=0.275]{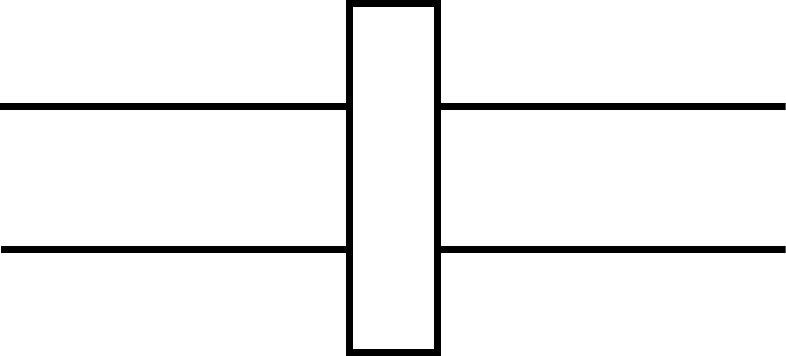} ,}} \\[1em]
\label{ProjectorID2} 
\tag{\ref{wj2}} 
\vcenter{\hbox{\includegraphics[scale=0.275]{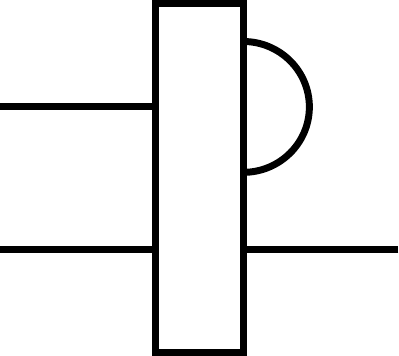}}} 
\quad = \quad \vcenter{\hbox{\includegraphics[scale=0.275]{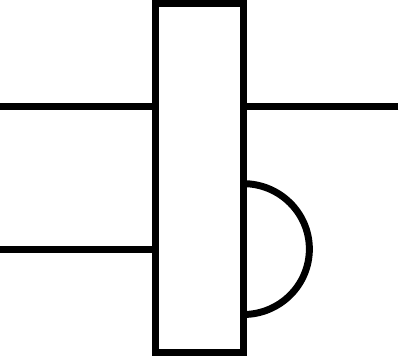}}} 
\quad & = \quad \vcenter{\hbox{\includegraphics[scale=0.275]{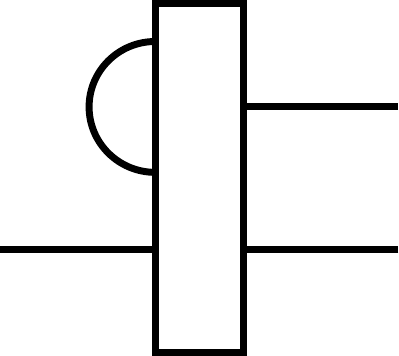}}} 
\quad = \quad \vcenter{\hbox{\includegraphics[scale=0.275]{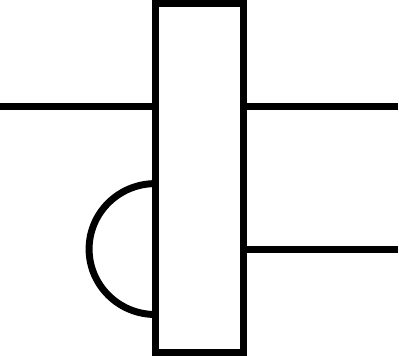}}}
\quad = \quad 0.
\end{align}
and~\cite[appendix~\red{A}]{fp3b}.
Also, if $\max \multii < \pmin(q)$, then we define the \emph{Jones-Wenzl composite projector}
\begin{align}\label{WJCompProj} 
\WJProj_\multii \quad := \quad \vcenter{\hbox{\includegraphics[scale=0.275]{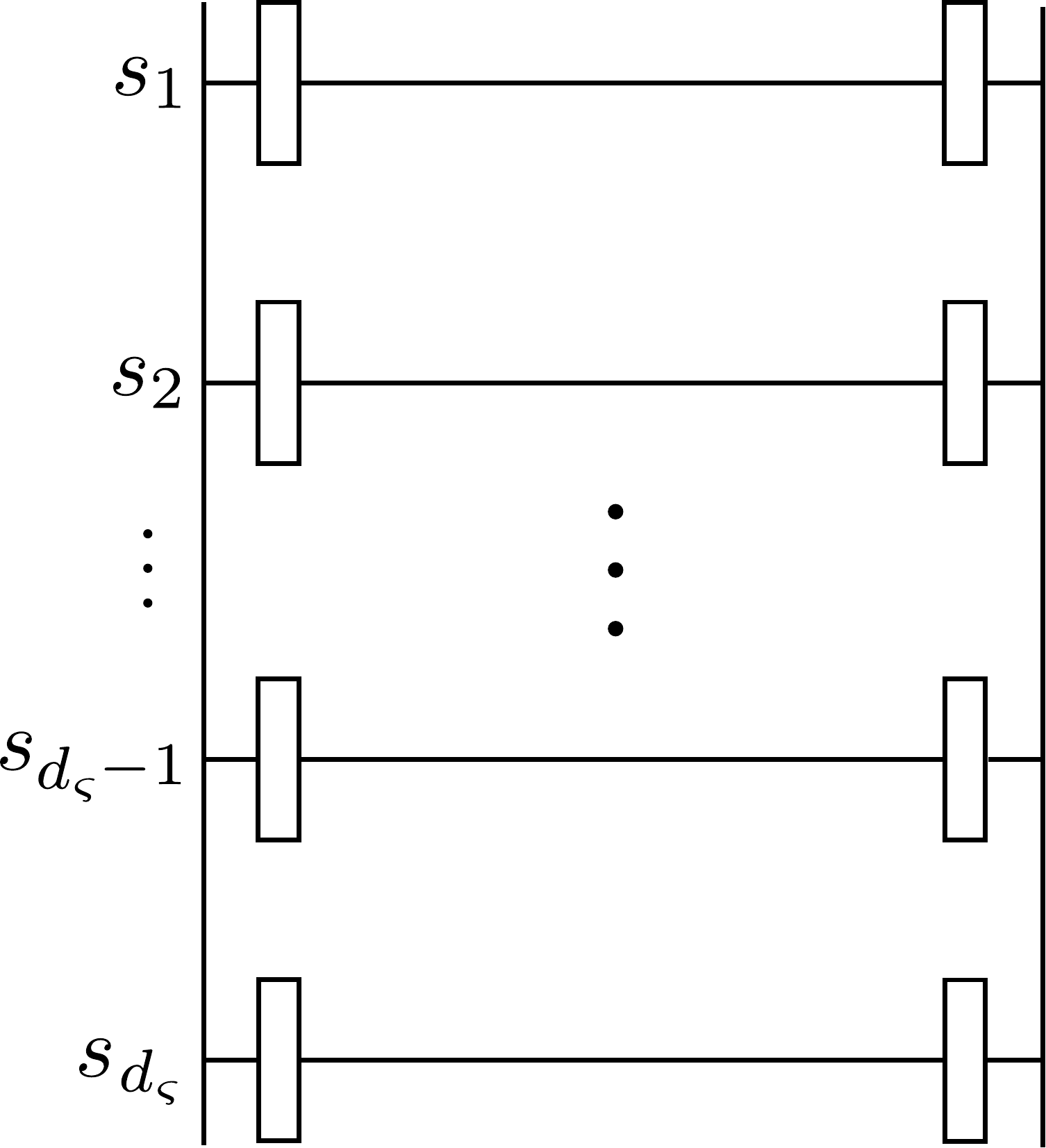} .}}
\hphantom{\WJProj_\multii \quad := \quad}
\end{align} 
and the \emph{Jones-Wenzl composite embedder} and its reflection
\begin{align}\label{WJCompEmbAndProjHat} 
\WJEmb_\multii \quad := \quad \vcenter{\hbox{\includegraphics[scale=0.275]{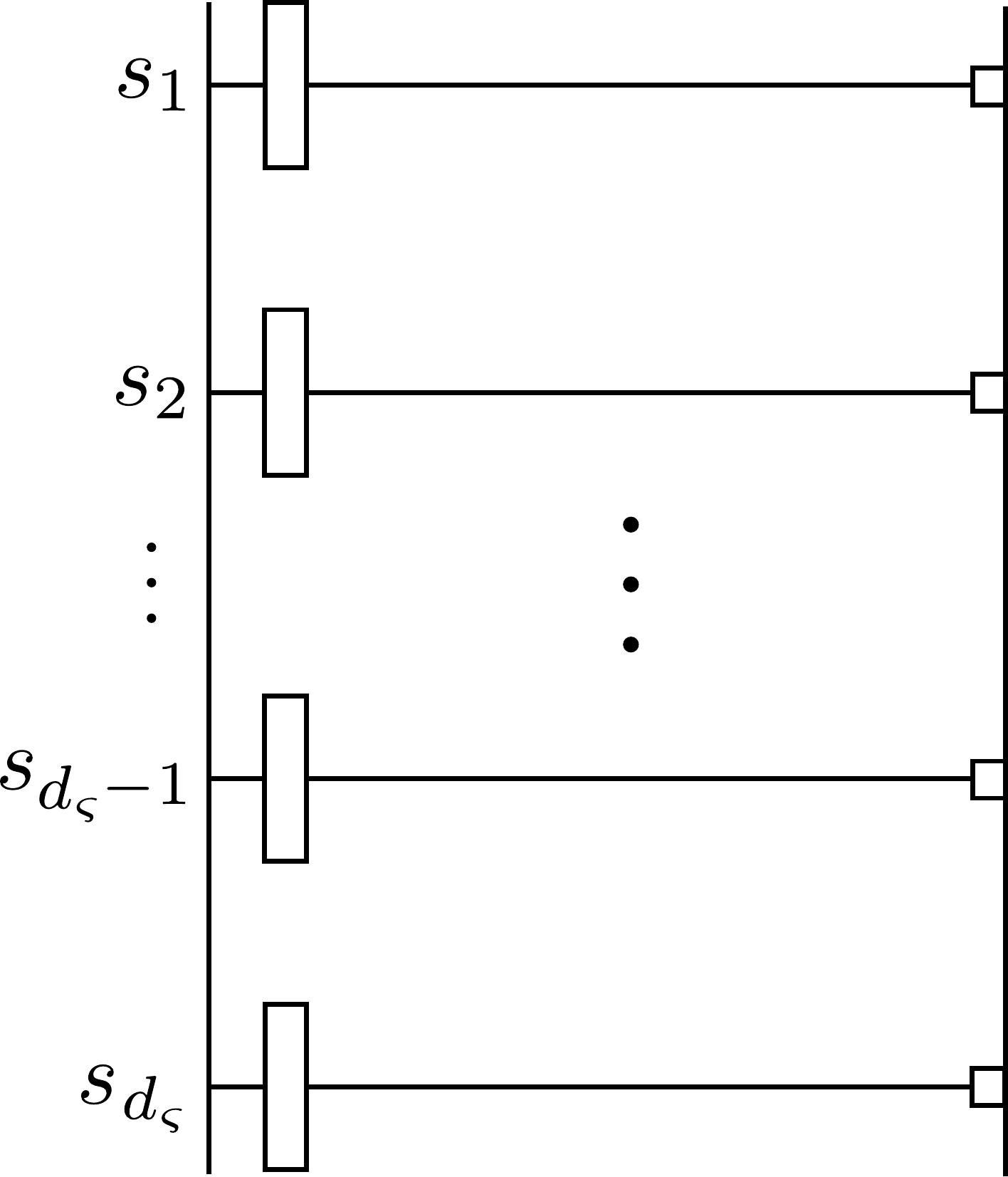} ,}}
\qquad \qquad \text{and} \qquad \qquad 
\WJProjHat_\multii \quad := \quad \vcenter{\hbox{\includegraphics[scale=0.275]{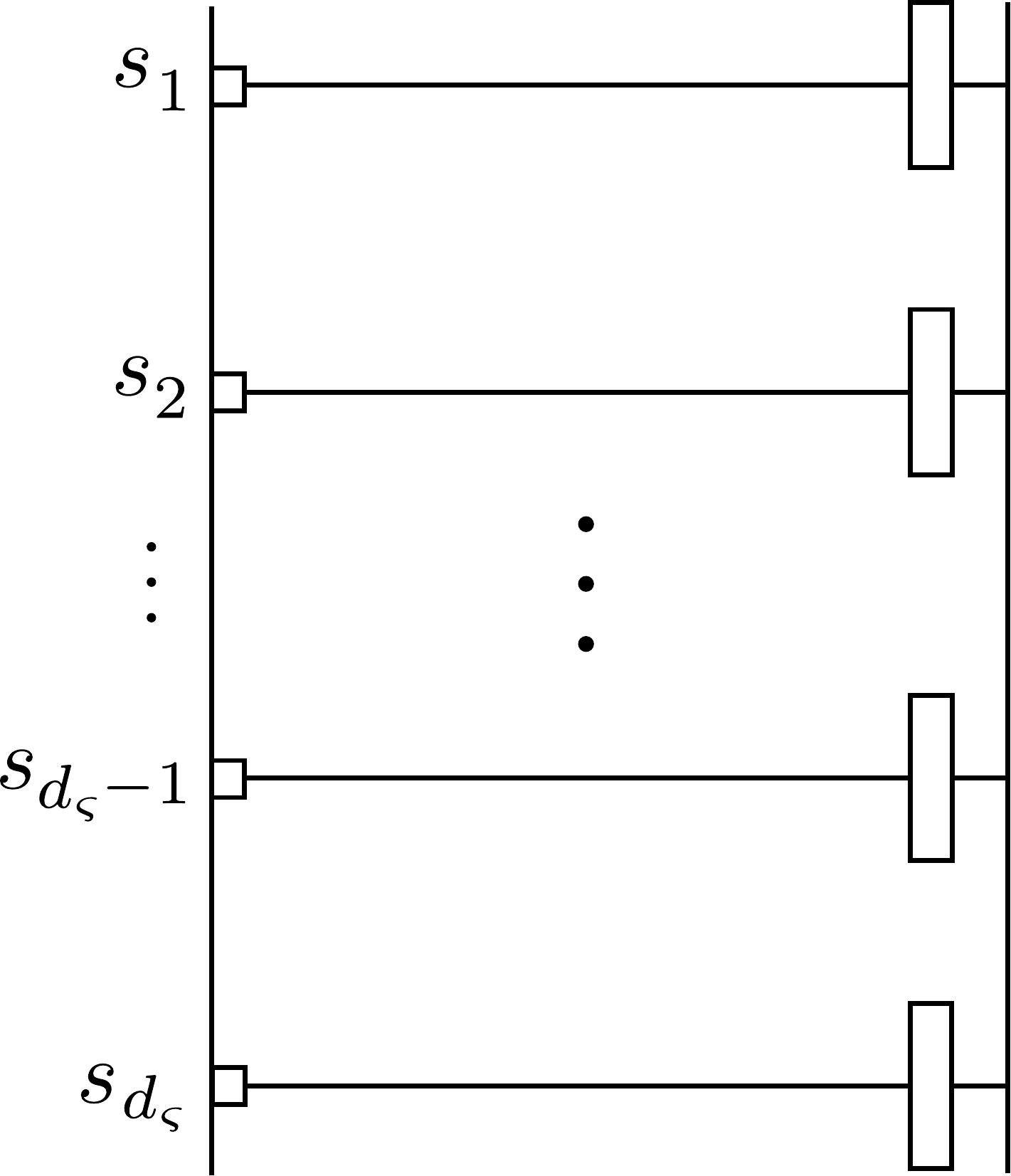} .}}
\end{align}
These valenced tangles have the properties 
\begin{align} \label{IdCompAndWJPhatPEmb} 
\WJProjHat_\multii \WJEmb_\multii = \mathbf{1}_{\TL_\multii} , \qquad
\WJEmb_\multii \WJProjHat_\multii = \WJProj_\multii , \qquad
\WJProjHat_\multii \WJProj_\multii \overset{\eqref{wj1}}{=} \WJProjHat_\multii ,
\qquad \text{and} \qquad 
\WJProj_\multii \WJEmb_\multii \overset{\eqref{wj1}}{=} \WJEmb_\multii . 
\end{align}

\bigskip
{\bf Embeddings.}
We identify any $(\multii, \multiii)$-valenced tangle $T \in \smash{\TL_\multii^\multiii}$ with 
an $(\Summed_\multii, \Summed_\multiii)$-tangle 
$\WJEmb_\multii T \smash{\WJProjHat}_\multiii \in \TL_{\Summed_\multii}^{\Summed_\multiii}$, 
and any $\multii$-valenced link state $\alpha \in \smash{\LS_\multii}$ with a link state
$\WJEmb_\multii(\alpha) \in \smash{\LS_{\Summed_\multii}}$, as follows:
for each $i \in \{ 1, 2, \ldots, \np_\multii \}$ (resp.~$j \in \{ 1, 2, \ldots, \np_\multiii \}$),
\begin{enumerate}[leftmargin=*, label = I\arabic*., ref = I\arabic*]
\itemcolor{red}
\item \label{step1} we separate the $i$:th (resp.~$j$:th) node of $T$ or $\alpha$ into a collection of $\sIndex_i$ (resp.~$p_j$) adjacent nodes, and 
\item \label{step2} we place a Jones-Wenzl projector box~\eqref{ProjBoxDiag} across the strands that attach to the node:
\begin{align}
\vcenter{\hbox{\includegraphics[scale=0.275]{Figures/e-valenced_box.pdf}}} 
\qquad \qquad \longmapsto \qquad \qquad 
\vcenter{\hbox{\includegraphics[scale=0.275]{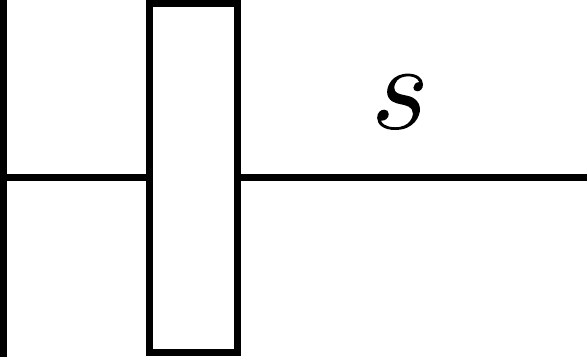}}}
\quad = \quad 
\begin{rcases}
\vcenter{\hbox{\includegraphics[scale=0.275]{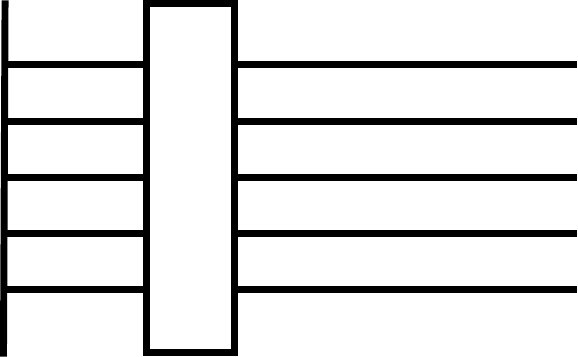}}} \quad
\end{rcases}
\; s .
\end{align}
\end{enumerate}
This identification defines two injective linear maps~\cite[lemma~\red{B.1}]{fp3a},
\begin{align}  \label{TangleLinkEmbDef}
\WJEmb_\multii (\,\cdot\,) \smash{\WJProjHat}_\multiii  \colon \TL_\multii^\multiii \longrightarrow 
\TL_{\Summed_\multii}^{\Summed_\multiii} 
\qquad \qquad \text{and} \qquad \qquad
\WJEmb_\multii \colon \LS_\multii \longrightarrow \LS_{\Summed_\multii} ,
\end{align}
which 
can respectively be realized diagrammatically as
\begin{align} 
 \label{TangleEmbDef}
\vcenter{\hbox{\includegraphics[scale=0.275]{Figures/e-GenericTangle_valenced_general.pdf}}} 
 \qquad \qquad & \overset{\WJEmb_\multii (\,\cdot\,) \smash{\WJProjHat}_\multiii}{\longmapsto}  \qquad \qquad 
\vcenter{\hbox{\includegraphics[scale=0.275]{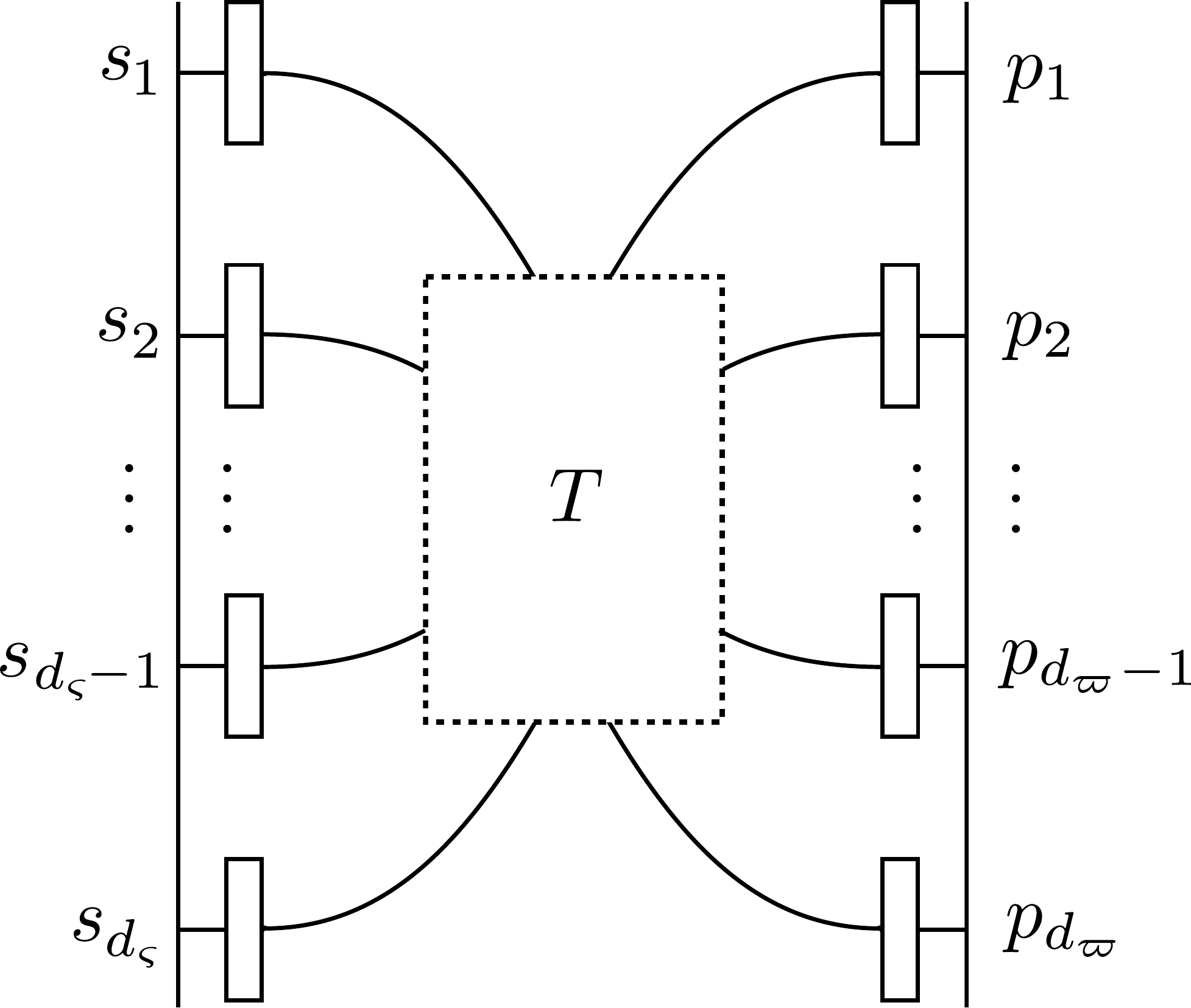} ,}} \\[1em]
\label{LinkEmbDef}
\vcenter{\hbox{\includegraphics[scale=0.275]{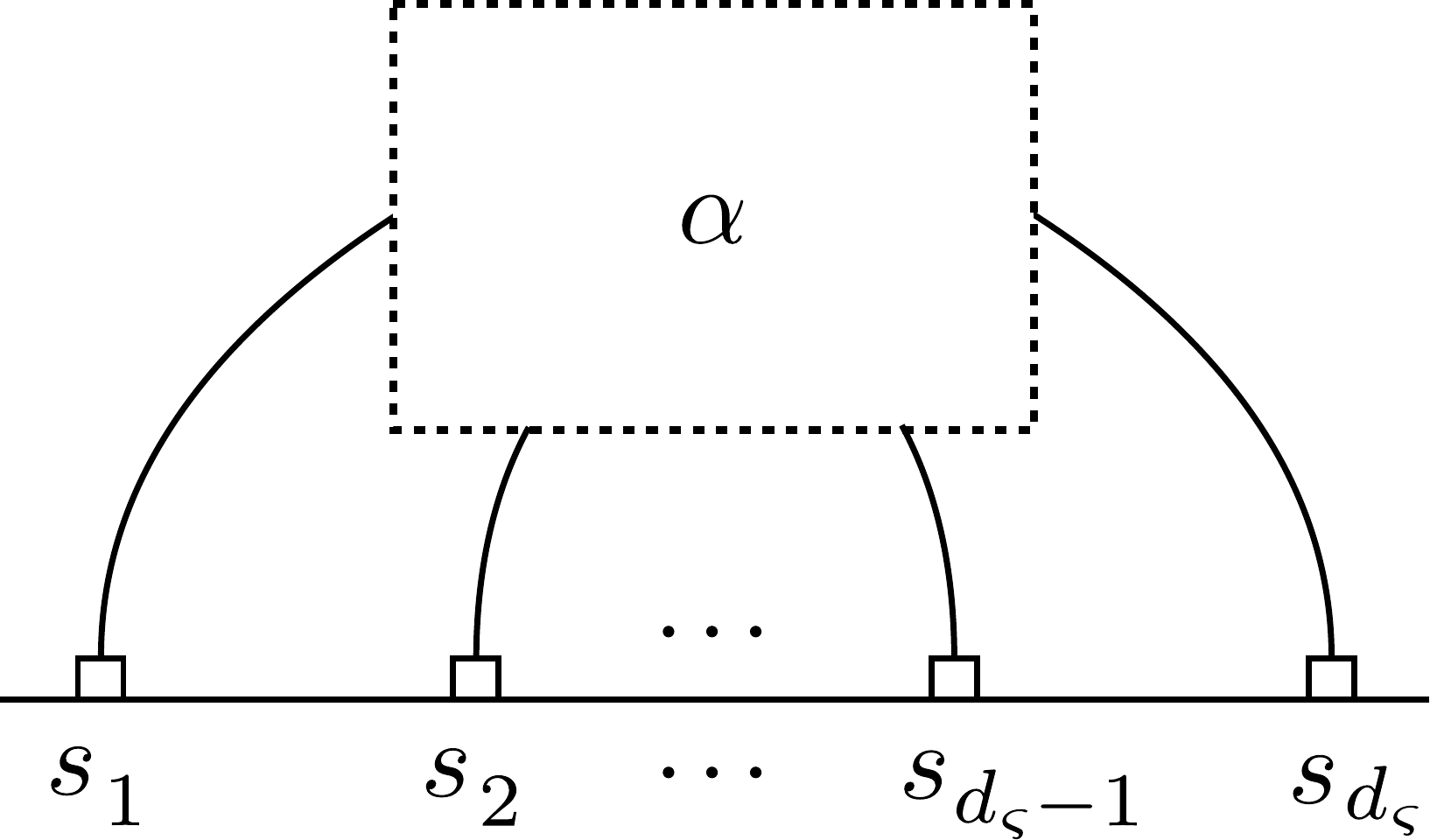}}} 
\qquad \qquad &  \underset{\hphantom{\WJEmb_\multii (\,\cdot\,) \smash{\WJProjHat}_\multiii}}{\overset{\WJEmb_\multii(\,\cdot\,)}{\longmapsto}} \qquad \qquad 
\vcenter{\hbox{\includegraphics[scale=0.275]{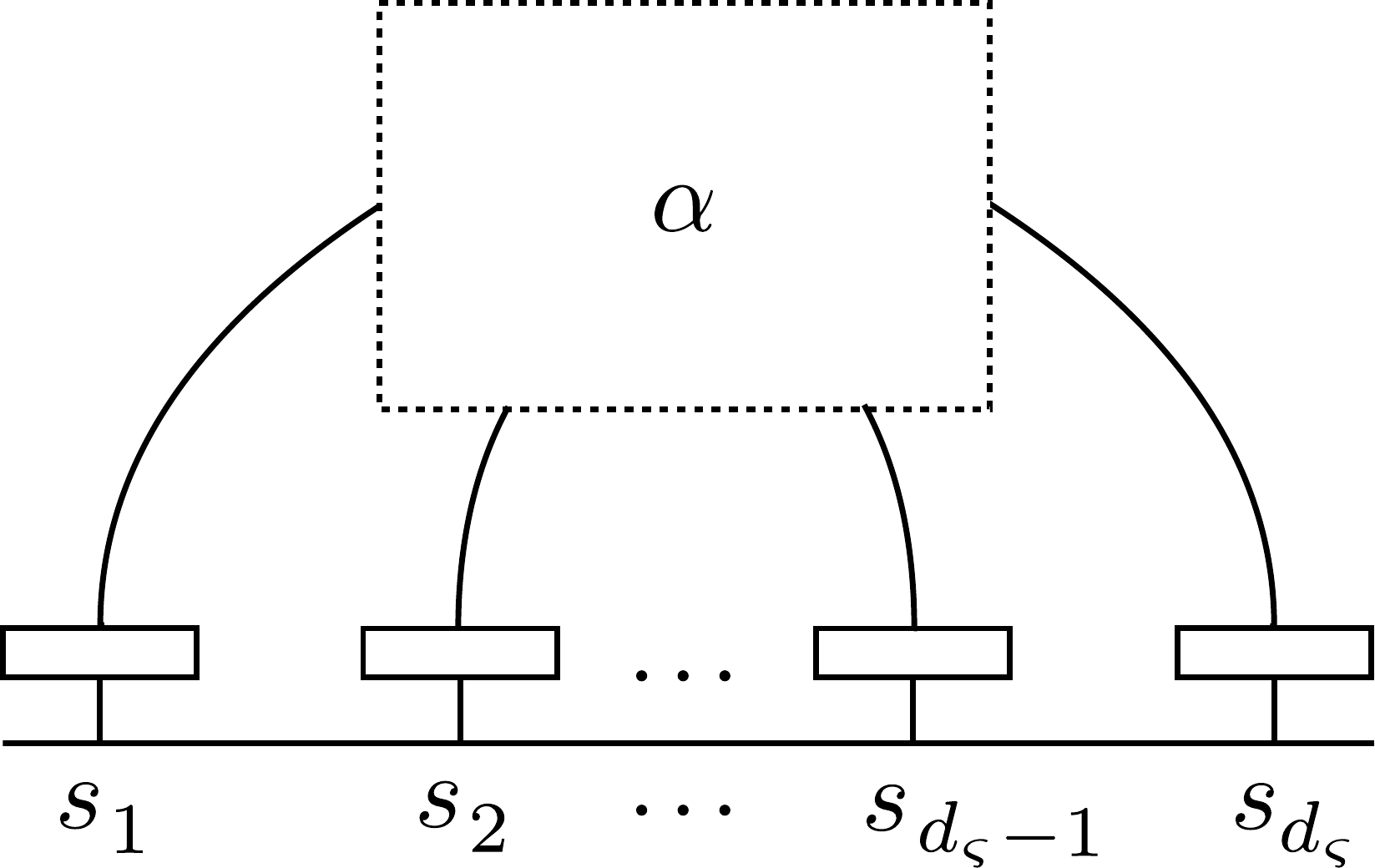} ,}}
\end{align}
where $T \in \smash{\TL_{\Summed_\multii}^{\Summed_\multiii}}$ 
(resp.~$\alpha \in \smash{\LS_{\Summed_\multii}}$) is a special tangle (resp.~link state)
lacking links that would result in loop-links in the valenced tangle (resp.~link state); 
cf. also property~\eqref{ProjectorID2} of the Jones-Wenzl projector.
These maps preserve the $s$-grading~\ref{LSDirSum} 
We summarize them in the following commuting diagrams~\cite[appendix~\red{B}]{fp3a}:
\begin{equation} \label{CommDiagram} 
\begin{tikzcd}[column sep=2cm, row sep=1.5cm]
& \arrow{ld}[swap]{\WJProjHat_\multii (\, \cdot \,)} \arrow{d}{\WJProj_\multii (\, \cdot \,)}
\LS_{\Summed_\multii} \\ 
\LS_\multii 
\arrow{r}{\WJEmb_\multii (\, \cdot \,)}
& \WJProj_\multii  \LS_{\Summed_\multii}
\end{tikzcd} 
\qquad \qquad \qquad
\begin{tikzcd}[column sep=2cm, row sep=1.5cm]
& \arrow{ld}[swap]{\WJProjHat_\multii (\, \cdot \,) \WJEmb_\multiii} \arrow{d}{\WJProj_\multii (\, \cdot \,) \WJProj_\multiii}
\TL_{\Summed_\multii}^{\Summed_\multiii} \\ 
\TL_\multii^{\multiii}
\arrow{r}{\WJEmb_\multii (\, \cdot \,) \WJProjHat_\multiii}
& \WJEmb_\multii \TL_\multii^{\multiii} \WJProjHat_\multiii 
\end{tikzcd} 
\end{equation}

\bigskip

{\bf Valenced Temperley-Lieb category.}
Valenced tangles can be viewed as morphisms in a category $\TL(\nu)$, 
that contains the Temperley-Lieb category $\TL^1(\nu)$.
The object class and morphism class of this category are respectively 
\begin{align} 
\text{Ob} \, \TL(\nu) &= \big\{\multii \in \{ \vec{0} \} \cup \bZpos^\# \, \big| \, \max \multii < \pmin(q) \big\} \\
\Hom  \TL(\nu) &= \big\{ \TL_\multii^\multiii \, \big| \, \text{$\multii, \multiii \in \text{Ob} \, \TL(\nu)$
 with $\Summed_\multii + \Summed_\multiii = 0 \Mod 2$} \big\} .
\end{align}
The source and target associated with the tangle $T \in \smash{\TL_\multii^\multiii}$ are the objects $\multiii$ and $\multii$ respectively.
The composition of morphisms of $\TL(\nu)$ is defined as the following diagram concatenation; 
see~\cite[section~\red{2D}]{fp3a}  
for details.
Given two valenced tangles $T \in \smash{\TL_\multii^\varepsilon}$ and $U \in \smash{\TL_\varepsilon^\multiii}$, 
we utilize the maps $\WJEmb_\multii (\,\cdot\,) \smash{\WJProjHat}_\varepsilon$ 
and $\WJEmb_\varepsilon (\,\cdot\,) \smash{\WJProjHat}_\multiii$ with properties~\eqref{IdCompAndWJPhatPEmb} to define 
\begin{align} \label{TangleHom} 
TU := \WJProjHat_\multii( \WJEmb_\multii T \smash{\WJProjHat}_\varepsilon )( \WJEmb_\varepsilon U \smash{\WJProjHat}_\multiii ) \WJEmb_\multiii 
\quad \in \TL_\multii^\multiii ,
\end{align}
performing on the right side the ordinary diagram concatenation without valenced nodes, as in~\eqref{ExmpleConcat}.
The identity morphism associated with the object $\multii$ is the unit~\eqref{TL_valenced_Unit} 
of the \emph{valenced Temperley-Lieb algebra} $\TL_\multii(\nu) = \smash{\TL_\multii^\multii}$, 
an associative unital algebra with multiplication given by~\eqref{TangleHom} 
with $\multii = \multiii = \varepsilon$. By~\cite[proposition~\red{2.10}]{fp3a}
(see also~\cite[theorem~\red{1.1}]{fp3b}), 
when $\Summed_\multii < \pmin(q)$, 
the algebra $\TL_\multii(\nu)$ is generated by its unit 
and the generators 
$\{\ValGen_1, \ValGen_2, \ldots, \ValGen_{\np_\multii-1}\}$ given by diagrams~(\ref{TL_valenced_Unit},~\ref{TL_valenced_Gen}).
The special case of $\multii = \OneVec{n}$ is the 
Temperley-Lieb algebra $\TL_n(\nu)$.
We remark that the restriction to $\Summed_\multii < \pmin(q)$ in~\cite{fp3a, fp3b} should not be essential for 
$\{\mathbf{1}_{\TL_\multii}, \ValGen_1, \ValGen_2, \ldots, \ValGen_{\np_\multii-1}\}$ to be the whole generating set of $\TL_\multii(\nu)$,
but no proof is known to us to date. 
Also, not all of the relations of this algebra are known to us in general, see~\cite{fp3b}.

We also let $\smash{T^* \in \TL_\multiii^\multii}$ denote the reflection of $T \in \smash{\TL_\multii^\multiii}$ about a vertical axis:
\begin{align} \label{DaggerRefl} 
T \quad = \quad \vcenter{\hbox{\includegraphics[scale=0.275]{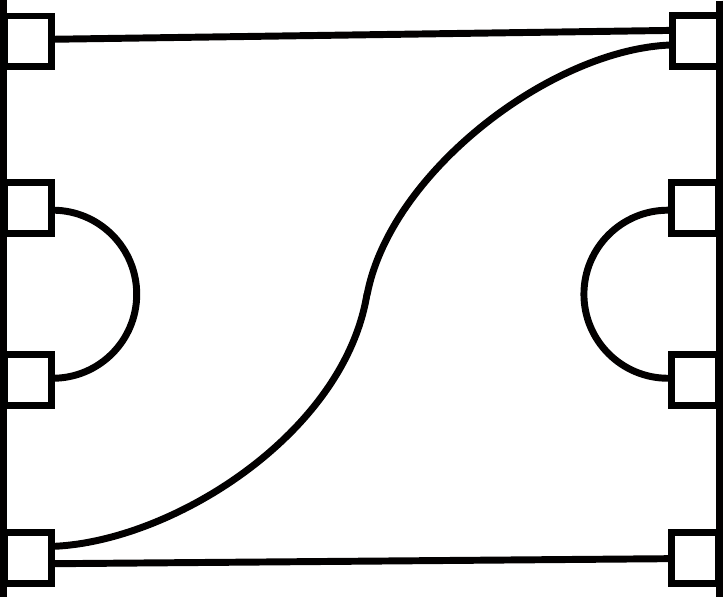}}} 
\qquad \qquad \Longrightarrow \qquad \qquad
T^* \quad = \quad \vcenter{\hbox{\includegraphics[scale=0.275]{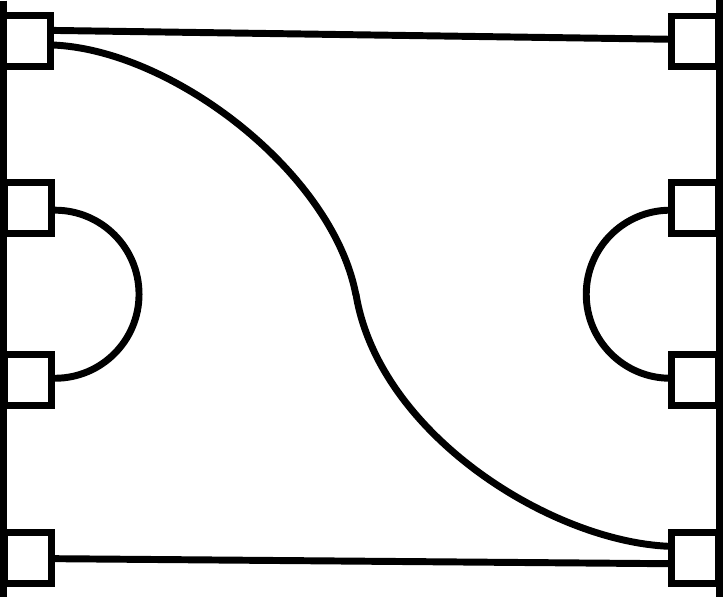} .}}
\end{align}
From~(\ref{StarMapAlpha},~\ref{WJCompEmbAndProjHat},~\ref{DaggerRefl}),
we immediately observe that
\begin{align} 
\label{StarMapAlphaEmbProj}
(\WJEmb_\multii \alpha)^* = \alpha^* \WJProjHat_\multii 
\qquad\qquad \text{and}\qquad\qquad 
\alphaBar \WJProjHat_\multii = (\WJEmb_\multii \alphaBar^*)^*
\end{align}

Given a valenced tangle $T \in \smash{\TL_\multii^\multiii}$ and a valenced link state $\alpha \in \smash{\LS_\multiii\super{s}}$,
we 
Tutilize the maps $\WJEmb_\multii (\,\cdot\,) \smash{\WJProjHat}_\multii$, and $\WJEmb_\multii (\,\cdot\,)$, and $\smash{\WJProjHat}_\multii(\,\cdot\,)$
form the concatenation 
\begin{align} \label{TLactionGeneral}
T \alpha := \smash{\WJProjHat}_\multii (\WJEmb_\multii T \smash{\WJProjHat}_\multiii) ( \WJEmb_\multiii \alpha ) 
\quad \in \LS_\multii\super{s} ,
\end{align}
performing on the right side the ordinary diagram concatenation, as in~(\ref{loopex},~\ref{turnbackex}).
With $\multiii = \multii$, rule~\eqref{TLactionGeneral} defines a left $\TL_\multii(\nu)$-action on the space $\smash{\LS_\multii\super{s}}$,
which we thus call a Temperley-Lieb  \emph{standard module}.
We also call the direct sum $\LS_\multii$ defined in~\eqref{LSDirSum} the \emph{link state module}.
We refer to~\cite{fp3a} for details about the $\TL_\multii(\nu)$-standard modules.

We also define analogous right $\TL_\multii(\nu)$-module structures on $\smash{\LSBar_\multii\super{s}}$ and $\LSBar_\multii$,
with right $\TL_\multii(\nu)$-action defined by diagram concatenation from the right:
\begin{align} 
\alphaBar T := \alphaBar \WJProjHat_\multii( \WJEmb_\multii T \smash{\WJProjHat}_\multii ) \WJEmb_\multii \quad \in \LSBar_\multii\super{s} .
\end{align}

We can enumerate the valenced link patterns in terms of the numbers $\smash{\Dim_\multii\super{s}} $ appearing in~\eqref{CountLP}.

\begin{lem} \label{LSDimLem2}
The following hold: 
\begin{enumerate}
\itemcolor{red}
\item \textnormal{\cite[Lemma~\red{2.8}]{fp3a}:}
We have 
\begin{align} \label{LSDim2} 
\dim \LS_\multii\super{s} & = \# \LP_\multii\super{s} = \Dim_\multii\super{s}
\qquad \qquad \textnormal{and} \qquad \qquad \dim \LS_\multii = \# \LP_\multii = \Dim_\multii.
\end{align}
Similarly, this identity 
holds after the symbolic replacements $\LS \mapsto \LSBar$ and $\LP \mapsto \LPBar$.

\item \textnormal{\cite[Corollary~\red{2.7}]{fp3a}:}
We have 
\begin{align}  \label{TLdim}
\dim \TL_\multii^\multiii 
= \sum_{s \, \in \, \DefectSet_\multii \cap \, \DefectSet_\multiii}  \Dim_\multiii\super{s} \Dim_\multii\super{s} .
\end{align} 
\end{enumerate}
\end{lem}

\subsection{Temperley-Lieb representations on fundamental $\Uqsltwo$-modules}
\label{DiacActTypeOneSec}

The purpose of this section is to define and study an action of the Temperley-Lieb algebra $\TL_n(\nu)$
on the tensor products $\VecSp_n$ and $\VecSpBar_n$.
Specifically, we explicitly define a left action $\CModule{\VecSp_n}{\TL}$ and a right action $\CRModule{\VecSpBar_n}{\TL}$ 
and prove that 
the action is via $\Uqsltwo,\UqsltwoBar$-homomorphisms.
(In the present section, we mostly consider the case $q \notin \{\pm 1, \pm \ii\}$;
the special case $q \in \{\pm \ii\}$ when $\nu=0$ is discussed in appendix~\ref{ExceptionalQSect} and 
the case $q=1$ in appendix~\ref{ClassicalApp}).
Indeed, the Temperley-Lieb generators correspond with multiples of 
$\Uqsltwo,\UqsltwoBar$-submodule projectors of type~\eqref{TLAsProjectors}. 

For the Temperley-Lieb action, it is useful to define the \emph{singlet vectors} $\sing \in \VecSp_2$ and $\singBar \in \smash{\VecSpBar_2}$ as
\begin{align} 
\label{singletVector} 
& \sing \overset{\hphantom{\eqref{tau}}}{:=}
\ii q^{1/2} \FundBasis_0 \otimes \FundBasis_1 - \ii q^{-1/2} \FundBasis_1 \otimes \FundBasis_0  
\qquad \qquad \text{and} \quad
&& \singBar \overset{\hphantom{\eqref{taubar}}}{:=}
\ii q^{1/2} \FundBasisBar_0 \otimes \FundBasisBar_1 - \ii q^{-1/2} \FundBasisBar_1 \otimes \FundBasisBar_0  \\
\label{singletVectorTau} 
& \hphantom{\sing} \overset{\eqref{tau}}{=} \left(\frac{q-q^{-1}}{\ii q^{1/2}}\right) \HWvec\sub{1,1}\super{0}  
\quad \text{if $q \neq \pm 1$}
&& \hphantom{\singBar} \overset{\eqref{taubar}}{=} 
\ii q^{1/2}(q-q^{-1}) \, \HWvecBar\sub{1,1}\super{0}  
\quad \text{if $q \neq \pm 1$} .
\end{align}
As shown, when 
$q \neq \pm 1$, 
the vectors $\sing$ and $\singBar$ are renormalized versions of the conformal-block vectors 
$\smash{\HWvec\sub{1,1}\super{0}}$ and $\smash{\HWvecBar\sub{1,1}\super{0}}$. 
The choice of normalization will become apparent with diagram representation of the $\TL_n(\nu)$-action (section~\ref{LSandBiformandSCGrapgSec}).

Next, for all integers $n \geq 2$ and $i,j \in \{1,2,\ldots,n - 1\}$ 
we define the (left) actions
\begin{align} \label{ExtendThis0} 
\Trep_{n}^{n-2}( \Lgen_i ) \colon \VecSp_{n-2} \longrightarrow \VecSp_n
\qquad \qquad \text{and} \qquad \qquad
\Trep_{n-2}^n( \Rgen_j ) \colon \VecSp_n \longrightarrow \VecSp_{n-2}
\end{align}
by extending linearly the rules
\begin{align}
\label{ExtendThis1} 
\Trep_{n}^{n-2}( \Lgen_i ) := \; & \id^{\otimes(i - 1)} \otimes \Trep_{2}^{0}( \Lgen_1 ) \otimes \id^{\otimes(n - i - 1)} ,
&& \Trep_{2}^{0}( \Lgen_1 ) (\Basis_0\super{0}) := \sing , \\
\label{ExtendThis2} 
\Trep_{n-2}^n( \Rgen_j ) := \; & \id^{\otimes(j - 1)} \otimes \Trep_{0}^2( \Rgen_1 ) \otimes \id^{\otimes(n - j - 1)} ,
&& \Trep_{0}^2( \Rgen_1 )(\FundBasis_{\ell_1} \otimes \FundBasis_{\ell_2}) := 
\begin{cases} 
\hphantom{-} 0, & \ell_1 = 0, \; \ell_2 = 0 , \\
\hphantom{-} \ii q^{1/2} \hphantom{{}^-} , & \ell_1 = 0, \; \ell_2 = 1 , \\
- \ii q^{-1/2} , & \ell_1 = 1, \; \ell_2 = 0 , \\
\hphantom{-} 0, & \ell_1 = 1, \;  \ell_2 = 1 ,
\end{cases} 
\end{align}
where the one-dimensional vector space $\VecSp_0 = \Span\{ \smash{\Basis_0\super{0}} \}$
is identified with the ground field $\bC$ and dropped from all tensor products.
In lemmas~\ref{TLprojLemNew} and~\ref{TLprojLemNewExceptional}, 
we relate (\ref{ExtendThis1},~\ref{ExtendThis2}) to left $\Uqsltwo,\UqsltwoBar$-submodule embeddings and projectors.

We analogously define the (right) actions 
(which define homomorphisms of right $\Uqsltwo,\UqsltwoBar$-modules)
\begin{align}
\TrepBar_{n-2}^n ( \Lgen_j ) \colon  \VecSpBar_n \longrightarrow \VecSpBar_{n-2}
\qquad \qquad \text{and} \qquad \qquad
\TrepBar_{n}^{n-2}( \Rgen_i ) \colon \VecSpBar_{n-2} \longrightarrow \VecSpBar_n ,
\end{align}
by extending linearly the rules
\begin{align} 
\label{ExtendThisBar1} 
\TrepBar_{n-2}^{n}( \Rgen_i )  := \; & \id^{\otimes(i - 1)} \otimes \TrepBar_0^{2} ( \Rgen_1 ) \otimes \id^{\otimes(n - i - 1)} ,
&& \TrepBar_0^{2} ( \Rgen_1 ) (\BasisBar_0\super{0})  := \singBar , \\
\label{ExtendThisBar2} 
\TrepBar_n^{n-2}( \Lgen_j ) := \; & \id^{\otimes(j - 1)} \otimes \TrepBar_2^{0}( \Lgen_1 ) \otimes \id^{\otimes(n - j - 1)} ,
&& \TrepBar_2^{0}( \Lgen_1 ) (\FundBasisBar_{\ell_1} \otimes \FundBasisBar_{\ell_2}) := 
\begin{cases} 
\hphantom{-} 0, & \ell_1 = 0, \; \ell_2 = 0 , \\
\hphantom{-} \ii q^{1/2} \hphantom{{}^-} , & \ell_1 = 0, \; \ell_2 = 1 , \\
- \ii q^{-1/2} , & \ell_1 = 1, \; \ell_2 = 0 , \\
\hphantom{-} 0, & \ell_1 = 1, \;  \ell_2 = 1 .
\end{cases} 
\end{align}

To begin, we verify that rules~(\ref{ExtendThis0}--\ref{ExtendThisBar2}) for the generator diagrams $\Lgen_i$, $\Rgen_j$ 
extend naturally to two families of morphisms $\smash{\Trep_n^m}$ and $\smash{\TrepBar_n^m}$, 
which in the special case of $m = n$ give rise to representations of the Temperley-Lieb algebra $\TL_n(\nu)$: 
a left $\TL_n(\nu)$-action on $\CModule{\VecSp_n}{\TL}$ and a right $\TL_n(\nu)$-action on $\CRModule{\VecSpBar_n}{\TL}$ (see corollary~\ref{RepCor}).
In general, these maps are not representations per se, as the set of tangles generated by $\Lgen_i$, $\Rgen_j$, and $\mathbf{1}_{\TL_n}$ is not an algebra.

\begin{lem} \label{HomoLem} 
\textnormal{(Temperley-Lieb actions):}
There exists a unique family of maps 
\begin{align} \label{BigFamily} 
\big\{ \Trep_n^m \colon \TL_n^m \longrightarrow \Hom ( \VecSp_m, \VecSp_n ) \,\big|\, n,m \in \bZnn, \, n + m \equiv 0 \Mod{2} \big\} 
\end{align}
with the following properties:
\begin{enumerate} 
\itemcolor{red}

\item \label{HomoLem21} 
All maps in family~\eqref{BigFamily} that are of the form $\Trep_{n}^{n-2}$ and $\Trep_{n-2}^n$ 
are determined by rules~\textnormal{(\ref{ExtendThis1},~\ref{ExtendThis2})}. 

\item \label{HomoLem22} 
For all maps $\Trep_n^k$ and $\Trep_k^m$ in family~\eqref{BigFamily} and for all tangles 
$T \in \smash{\TL_n^k}$ and $U \in \smash{\TL_k^m}$, we have 
\begin{align} \label{HomoLikeProp} 
\Trep_n^m(TU) = \Trep_n^k (T) \circ \Trep_k^m (U) . 
\end{align}
\end{enumerate}
Similarly, there exists a unique family of maps
\begin{align} \label{BigFamilyBar} 
\big\{ \TrepBar_n^m \colon \TL_n^m \longrightarrow \Hom ( \VecSpBar_n, \VecSpBar_m ) \,\big|\, n,m \in \bZnn, \, n + m \equiv 0 \Mod{2} \big\}
\end{align}
with the following properties:
\begin{enumerate}
\itemcolor{red}
\setcounter{enumi}{2}

\item \label{HomoLem21Bar} 
All maps in family~\eqref{BigFamilyBar} that are of the form $\TrepBar_{n-2}^{n}$ and $\TrepBar_n^{n-2}$ 
are determined by rules~\textnormal{(\ref{ExtendThisBar1},~\ref{ExtendThisBar2})}. 

\item \label{HomoLem22Bar} 
For all maps $\TrepBar_n^k$ and $\TrepBar_k^m$ in family~\eqref{BigFamilyBar} and for all tangles 
$T \in \smash{\TL_n^k}$ and $U \in \smash{\TL_k^m}$, we have 
\begin{align} \label{HomoLikePropBar} 
\TrepBar_n^m(TU) = \TrepBar_k^m (U) \circ \TrepBar_n^k (T) . 
\end{align}
\end{enumerate}
\end{lem}

\begin{proof} 
By lemma~\ref{StdLem}, any tangle $T \in \smash{\TL_n^m}$ equals a polynomial in the generators $\Lgen_i$ and $\Rgen_j$,
so rules (\ref{ExtendThis1},~\ref{ExtendThis2}) and (\ref{ExtendThisBar1},~\ref{ExtendThisBar2})
completely determine families~\eqref{BigFamily} and~\eqref{BigFamilyBar} 
via homomorphism properties~\eqref{HomoLikeProp} and~\eqref{HomoLikePropBar}.
\end{proof}

We often use the following shorthand notation, 
for all tangles $T \in \TL_n^m$ and vectors $v \in \VecSp_m$ and
$\overbarStraight{v} \in \VecSpBar_n$:
\begin{align} \label{shorthandN}
\Trep_n^m(T) (v) = T  v
\qquad\qquad \text{and} \qquad\qquad
\TrepBar_n^m(T) (\overbarStraight{v}) = \overbarStraight{v} T .
\end{align}

\begin{cor} \label{NoShiftPropertyCorN}
\textnormal{($s$-grading preservation):}
We have
\begin{align} \label{NoShiftPropertyN}
v \in \Ksp_m\super{s} , \quad 
\overbarStraight{v} \in \KspBar_n\super{s} ,
\quad \textnormal{and} \quad T \in \TL_n^m
\qquad \qquad \Longrightarrow \qquad \qquad 
T  v \in \Ksp_n\super{s} 
\quad \textnormal{and} \quad
\overbarStraight{v} T  \in \KspBar_m\super{s}  .
\end{align} 
\end{cor}

\begin{proof}
Assertion~\eqref{NoShiftPropertyN} immediately follows from rules~(\ref{ExtendThis0}--\ref{ExtendThisBar2}),
definitions~\eqref{singletVector}, and the $s$-grading~\eqref{sGrading}. 
\end{proof}

Next, we realize actions~(\ref{ExtendThis0}--\ref{ExtendThisBar2}) 
in terms of 
$\Uqsltwo,\UqsltwoBar$-homomorphisms. 
To this end, in lemma~\ref{TLprojLemNew} and corollary~\ref{RepCor} we treat the generic case $q \notin \{\pm \ii\}$, relating 
$\Lgen_i$, $\Gen_i$, and $\Rgen_j$ to special cases of the embeddings and projectors~(\ref{EmbeddingDef2x2}--\ref{ProjectioHatDefn2x2}).
In the exceptional case $q \in \{\pm \ii\}$, the embedding is still well-defined but the projectors cannot be defined,
for the direct-sum decomposition~\eqref{M02Decomp}, discussed below, fails.
We consider this case separately in appendix~\ref{ExceptionalQSect}.

When $\pmin(q) > 2$, 
proposition~\ref{MoreGenDecompAndEmbProp} 
with $\multii = (1,1)$ gives the following direct-sum decompositions of $\Uqsltwo,\UqsltwoBar$-modules:  
\begin{align} \label{M02Decomp}
\Module{\VecSp_2}{\Uqsltwo,\UqsltwoBar}
\overset{\eqref{TensProdModules}}{:=} \Wd\sub{1} \otimes \Wd\sub{1} 
\overset{\eqref{MoreGenDecomp}}{\isom} \Wd\sub{0} \oplus \Wd\sub{2} 
\qquad \qquad \text{and} \qquad \qquad
\RModule{\VecSpBar_2}{\Uqsltwo,\UqsltwoBar} \overset{\eqref{TensProdModules}}{:=} 
\WdBar\sub{1} \otimes \WdBar\sub{1} 
\overset{\eqref{MoreGenDecomp}}{\isom} \WdBar\sub{0} \oplus \WdBar\sub{2} ,
\end{align}
with bases
\begin{align} 
\label{M02Basis} 
& \HWvec\sub{1,1}\super{0} 
\quad \text{for} \quad 
\CCembedor\super{0}\sub{1,1} ( \Wd\sub{0} ) 
\qquad \qquad \text{and} \qquad \qquad
&& \begin{cases}
\HWvec\sub{1,1}\super{2} = \FundBasis_0 \otimes \FundBasis_0,  \\[.2em]
F. \HWvec\sub{1,1}\super{2} = q^{-1} \FundBasis_0 \otimes \FundBasis_1 + \FundBasis_1 \otimes \FundBasis_0,  \\[.2em]
F^2 . \HWvec\sub{1,1}\super{2}  = \nu \FundBasis_1 \otimes \FundBasis_1 
\end{cases}
\quad \text{for} \quad 
\CCembedor\super{2}\sub{1,1} ( \Wd\sub{2} ) , \\
\label{M02BasisBar} 
& \HWvecBar\sub{1,1}\super{0}
\quad \text{for} \quad  
\CCembedorBar\super{0}\sub{1,1} ( \WdBar\sub{0} ) 
\qquad \qquad \text{and} \qquad \qquad
&& \begin{cases}
\HWvecBar\sub{1,1}\super{2} = \FundBasisBar_0 \otimes \FundBasisBar_0,  \\[.2em]
\HWvecBar\sub{1,1}\super{2} .E = \FundBasisBar_0 \otimes \FundBasisBar_1 +  q \FundBasisBar_1 \otimes \FundBasisBar_0,  \\[.2em]
\HWvecBar\sub{1,1}\super{2} .E^2 = \nu \FundBasisBar_1 \otimes \FundBasisBar_1 
\end{cases}
\quad \quad \text{for} \quad 
\CCembedorBar\super{2}\sub{1,1} ( \WdBar\sub{2} ) .
\end{align}

Definitions~(\ref{EmbeddingDef2x2}--\ref{ProjectioHatDefn2x2}) of the homomorphisms use these bases~(\ref{M02Basis},~\ref{M02BasisBar}).

\begin{lem}  \label{TLprojLemNew} 
Suppose $q \in \bC^\times \setminus \{ \pm1, \pm \ii \}$. 
Then, for all integers $n \geq 2$ and $i , j \in \{1,2,\ldots,n-1\}$, we have
\begin{align} 
\label{GenProj2-1}
\Trep_{n}^{n-2}(\Lgen_i) & = \left(\frac{q-q^{-1}}{\ii q^{1/2}}\right) 
\big(\id^{\otimes(i-1)} \otimes \CCembedor\super{0}\sub{1,1} \otimes \id^{\otimes(n-i-1)}\big) , \\
\label{GenProj2-0} 
\Trep_{n-2}^n(\Rgen_j) & = \left(\frac{\ii \nu q^{1/2}}{q-q^{-1}}\right) 
\big(\id^{\otimes(j-1)} \otimes \CChatprojector\super{1,1}\sub{0} \otimes \id^{\otimes(n-j-1)}\big) ,
\end{align}
and similarly, 
\begin{align} 
\label{GenProj2-1Bar}
\TrepBar_{n-2}^{n}(\Rgen_j) & = \ii q^{1/2} (q-q^{-1})
\big(\id^{\otimes(j-1)} \otimes \CCembedorBar\super{0}\sub{1,1} \otimes 
\id^{\otimes(n-j-1)}\big) , \\
\label{GenProj2-0Bar} 
\TrepBar_n^{n-2}(\Lgen_i) & = \left(\frac{\nu}{\ii q^{1/2} (q-q^{-1})}\right) 
\big(\id^{\otimes(i-1)} \otimes \CChatprojectorBar\super{1,1}\sub{0} \otimes 
\id^{\otimes(n-i-1)}\big).
\end{align}
Analogous statements also hold for the case of $q \in \{\pm \ii \}$, given in lemma~\ref{TLprojLemNewExceptional} in 
appendix~\ref{ExceptionalQSect}.
\end{lem}

\begin{proof}  
By definitions~(\ref{EmbeddingDef2x2},~\ref{ProjectioHatDefn2x2}) of 
$\smash{\CCembedor\super{0}\sub{1,1}}$ and 
$\smash{\CChatprojector\super{1,1}\sub{0}}$, we have 
\begin{align} \label{tausingIDs}
\CCembedor\super{0}\sub{1,1}\big( \Basis_0\super{0} \big)  \overset{\eqref{EmbeddingDef2x2}}{=} \HWvec\sub{1,1}\super{0} 
\overset{\eqref{singletVector}}{=} \left( \frac{\ii q^{1/2}}{q-q^{-1}} \right)  \sing , \qquad
\CChatprojector\super{1,1}\sub{0} \big( \HWvec\sub{1,1}\super{0} \big)  \overset{\eqref{ProjectioHatDefn2x2}}{=} \Basis_0\super{0} ,
\qquad \text{and} \qquad
\CChatprojector\super{1,1}\sub{0}(\sing)  \overset{\eqref{singletVectorTau}}{=} \left( \frac{q-q^{-1}}{\ii q^{1/2}} \right) \Basis_0\super{0} .
\end{align}
The first identity in~\eqref{tausingIDs} together with~\eqref{ExtendThis1} imply the first assertion~\eqref{GenProj2-1}.
The identities
\begin{align} \label{Rsing} 
\Rgen_1  \big( F^\ell . \HWvec\sub{1,1}\super{2} \big) 
\; \underset{\eqref{ExtendThis2}}{\overset{\eqref{M02Basis}}{=}} \; 0
\; \overset{\eqref{ProjectioHatDefn2x2}}{=} \;  \CChatprojector\super{1,1}\sub{0} \big( F^\ell . \HWvec\sub{1,1}\super{2} \big) \quad \text{for all $\ell = 0,1,2$,}
\qquad \text{and} \qquad 
\Rgen_1  \sing \;  \underset{\eqref{ExtendThis2}}{\overset{\eqref{singletVector}}{=}} \; \nu .
\end{align}
combined with the last identity of~\eqref{tausingIDs} imply the second assertion~\eqref{GenProj2-0}.
Assertions~(\ref{GenProj2-1Bar},~\ref{GenProj2-0Bar}) are similar.
\end{proof}

\begin{lem} \label{UqHomoLemN} 
\textnormal{(Quantum group homomorphism properties):}
Suppose $q \in \bC^\times \setminus \{\pm1\}$.
\begin{enumerate}
\itemcolor{red}

\item 
\label{UqHomoIt1N}
For each map $\Trep_n^m$ in family~\eqref{BigFamily}, we have 
$\im \smash{\Trep_n^m} \subset \HomMod{\Uqsltwo} ( \VecSp_m, \VecSp_n)$. 
In other words, we have
\begin{align} \label{HomProp2N} 
T  (x.v) = x.(T  v)
\qquad \textnormal{for all tangles $T \in \smash{\TL_n^m}$, elements $x \in \Uqsltwo$, and vectors $v \in \VecSp_m$.}
\end{align}

\item
\label{UqHomoIt2N}
Similarly, for each map $\smash{\TrepBar_n^m}$ in family~\eqref{BigFamilyBar}, we have 
$\im \smash{\TrepBar_n^m} \subset \HomMod{\Uqsltwo} ( \VecSpBar_n , \VecSpBar_m )$.
In other words, we have
\begin{align} \label{HomProp2BarN} 
(\overbarStraight{v}.x)  T = (\overbarStraight{v}  T).x 
\qquad \textnormal{for all tangles $T \in \smash{\TL_n^m}$, elements $x \in \smash{\Uqsltwo}$, and vectors $\overbarStraight{v} \in \VecSpBar_n$.}
\end{align}
\end{enumerate}
Similarly, this lemma holds after the symbolic replacements
$x \mapsto \overbarStraight{x}$ and $\Uqsltwo \mapsto \UqsltwoBar$.
\end{lem}

\begin{proof}
By lemma~\ref{StdLem} and homomorphism properties~(\ref{HomoLikeProp},~\ref{HomoLikePropBar}) of lemma~\ref{HomoLem},
it suffices to consider $T \in \{ \Lgen_i, \Rgen_j\}$. 
If $q \neq \pm \ii$ (resp.~$q = \pm \ii$), 
this follows from lemma~\ref{TLprojLemNew} and item~\ref{2ndIt1} of lemma~\ref{EmbProjLem2}
(resp.~lemma~\ref{TLprojLemNewExceptional} and lemma~\ref{EmbProjLem2Exc}).
\end{proof}

\begin{cor} \label{NoShiftPropertyCorHWVN}
Suppose $q \in \bC^\times \setminus \{\pm1\}$.
We have
\begin{align} \label{NoShiftPropertyHWVN}
v \in \HWsp_m\super{s} , \quad 
\overbarStraight{v} \in \HWspBar_n\super{s} ,
\quad \textnormal{and} \quad T \in \TL_n^m
\qquad \qquad \Longrightarrow \qquad \qquad 
T  v \in \HWsp_n\super{s} 
\quad \textnormal{and} \quad
\overbarStraight{v} T  \in \HWspBar_m\super{s}  .
\end{align} 
\end{cor}

\begin{proof}
Asserted property~\eqref{NoShiftPropertyHWVN} follows from
corollary~\ref{NoShiftPropertyCorN}, lemma~\ref{UqHomoLemN}, and definition~\eqref{HWspace2}.
\end{proof}

As a corollary, we also recover the fact that rules~(\ref{ExtendThis0}--\ref{ExtendThisBar2}) induce 
representations $\Trep_n := \smash{\Trep_n^n}$ and $\TrepBar_n := \smash{\TrepBar_n^n}$
of the Temperley-Lieb algebra $\TL_n(\nu)$ on $\CModule{\VecSp_n}{\TL}$ and $\CRModule{\VecSpBar_n}{\TL}$, 
acting as projectors of type~\eqref{TLAsProjectors}.

\begin{cor} \label{RepCor} 
Suppose $q \in \bC^\times \setminus \{ \pm1, \pm \ii \}$. Then, 
$\Trep_n \colon \TL_n(\nu) \longrightarrow \End \VecSp_n$ and 
$\TrepBar_n \colon \TL_n(\nu) \longrightarrow \EndOp \VecSpBar_n$ 
are respectively left and right representations,
and for all $j \in \{1,2,\ldots,n-1\}$, we have
\begin{align}
\label{GenProj} 
\Trep_n(\Gen_j) = \nu 
\big(\id^{\otimes(j-1)} \otimes \CCprojector\sub{1,1}\superscr{(1,1); (0)} \otimes \id^{\otimes(n-j-1)}\big),
\end{align}
and similarly,
\begin{align}
\label{GenProjBar}
\TrepBar_ n(\Gen_j) = \nu
\big(\id^{\otimes(j-1)} \otimes \CCprojectorBar\sub{1,1}\superscr{(1,1); (0)} \otimes \id^{\otimes(n-j-1)}\big) .
\end{align}
Also, analogous statements hold for the case of $q \in \{\pm \ii \}$, given in corollary~\ref{RepCorExceptional} in 
appendix~\ref{ExceptionalQSect}.
\end{cor}

\begin{proof}
Item~\ref{2ndIt3} of lemma~\ref{EmbProjLem2} gives the relations
$\smash{\CCprojector\sub{1,1}\superscr{(1,1); (0)} = \CCembedor\super{0}\sub{1,1} \circ \CChatprojector\super{1,1}\sub{0}}$
and 
$\smash{\CCprojectorBar\sub{1,1}\superscr{(1,1); (0)} = \CCembedorBar\super{0}\sub{1,1} \circ \CChatprojectorBar\super{1,1}\sub{0}}$,
and the Temperley-Lieb relation~\eqref{LRtoGen} gives $U_j = L_j R_j$.
Assertions~(\ref{GenProj},~\ref{GenProjBar}) now follow from lemmas~\ref{TLprojLemNew} and~\ref{UqHomoLemN}. 
Also, relations~\eqref{MaxRelations} and rules~(\ref{ExtendThis0}--\ref{ExtendThisBar2})
show that $\smash{\Trep_n(\mathbf{1}_{\TL_n}) = \id_{\VecSp_n}}$ and $\smash{\TrepBar_n(\mathbf{1}_{\TL_n}) = \id_{\VecSpBar_n}}$.
Therefore, $\Trep_n$ and $\TrepBar_n$ are respectively left and right representations of 
the associative unital algebra $\TL_n(\nu)$ generated by $\mathbf{1}_{\TL_n}$ and $\Gen_1, \ldots, \Gen_{n-1}$.
\end{proof}

In fact, the representations $\Trep_n$ and $\TrepBar_n$ are always faithful,
as has been proven independently by P.~Martin~\cite[theorem~\red{1}]{ppm} 
and F.~Goodman and H.~Wenzl~\cite[theorem~\red{2.4}]{gwe}.
We give another proof for this in section~\ref{KerImSubSec}, 
where we also prove an analogous result for the valenced Temperley-Lieb algebra, in proposition~\ref{PreFaithfulPropGen}.

\bigskip

The next lemma shows that the diagram actions on $\VecSp_n$ and $\VecSpBar_n$ work well with tensor products of tangles,
\begin{align} \label{concatenateTanglesN}
\OneVec{n} \otimes \OneVec{m} \quad := \quad \OneVec{n} \oplus \OneVec{m} \quad := \quad (\underbrace{1,1,\ldots,1}_{\text{$n+m$ times}})
\qquad \qquad \text{and} \qquad \qquad
T \otimes U \quad := \quad \vcenter{\hbox{\includegraphics[scale=0.275]{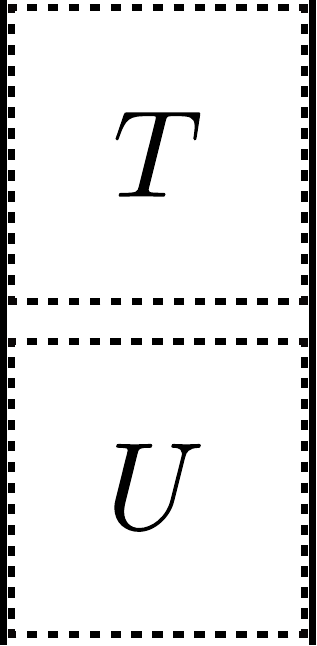} .}}
\end{align} 
(The Temperley-Lieb category $\TL^1(\nu)$ is a monoidal category, with identity object $\smash{\OneVec{0}}$, 
and the above tensor product.)

\begin{lem} \label{TensorLem} 
\textnormal{(Tensor product):} 
For all tangles $T \in \TL_n(\nu)$ and $U \in \TL_m(\nu)$, we have
\begin{align} \label{TensorID} 
\Trep_n(T) \otimes \Trep_m(U) = \Trep_{n+m} (T \otimes U) 
\qquad \qquad \textnormal{and} \qquad \qquad
\TrepBar_n(T) \otimes \TrepBar_m(U) = \TrepBar_{n+m} (T \otimes U) .
\end{align}
\end{lem}

\begin{proof}
By lemma~\ref{StdLem} and homomorphism properties~(\ref{HomoLikeProp},~\ref{HomoLikePropBar}) of lemma~\ref{HomoLem},
it suffices to consider $T \in \{ \Lgen_i, \Rgen_j\}$. 
Then, the assertion follows from the explicit construction of the diagram action in~(\ref{ExtendThis0}--\ref{ExtendThisBar2}).  
\end{proof}

The Jones-Wenz projector $\WJProj\sub{n}$~\eqref{ProjBoxDiag} 
corresponds to the submodule projectors $\Projection\sub{n}$ or $\ProjectionBar\sub{n}$, defined in~(\ref{ProjectionDefn},~\ref{ProjectionDefnBar}).

\begin{lem} \label{WJprojLem} 
Suppose 
$n < \pmin(q)$.  Then, we have
\begin{align} \label{Recover} 
\WJProj\sub{n} \quad \overset{\Trep_n}{\longmapsto} \quad \Projection\sub{n} 
\qquad\qquad \textnormal{and} \qquad\qquad
\WJProj\sub{n} \quad \overset{\TrepBar_n}{\longmapsto} \quad \ProjectionBar\sub{n} . 
\end{align}
\end{lem}

\begin{proof}
We prove the left side of~\eqref{Recover}; the right side is similar.
The assertion is trivial for $n=1$, as $\WJProj\sub{1}$ is just the unit. 
Hence assuming that $n > 1$, so $2 < \pmin(q)$, 
corollary~\ref{RepCor} and the fact that $\nu \neq 0$ (since $\pmin(q) > 2$) imply that
\begin{align}
\big(\id^{\otimes(i-1)} \otimes \CCprojector\sub{1,1}\superscr{(1,1); (0)} \otimes \id^{\otimes(n-i-1)}\big) (\WJProj\sub{n}v) 
\overset{\eqref{GenProj}}{=}
\nu^{-1}\Gen_i \WJProj\sub{n}v \overset{\eqref{wj2}}{=}
0
\end{align}
for any vector $v \in \VecSp_n$ and for all indices $i \in \{1,2,\ldots,n-1\}$.
As such~\cite[lemma~\red{2.4}]{kp} implies that $\smash{\WJProj\sub{n}v}$ is 
an element of the unique submodule of $\Module{\VecSp_n}{\Uqsltwo}$ isomorphic to $\Wd\sub{n}$,
which by item~\ref{CobloBasisItem2} of corollary~\ref{CobloBasisCor} is generated by
the highest-weight vector $\smash{\MTbas_0\super{n}}$ corresponding to 
the unique walk $\varrho$ over $\OneVec{n}$ with $\defect{\varrho} = n$.
On the other hand, 
we also have 
\begin{align} \label{HWPreservation}
v \in \HWsp_n\super{s}
\qquad \overset{\eqref{NoShiftPropertyHWVN}}{\Longrightarrow} \qquad
 \WJProj\sub{n}v \in \HWsp_n\super{s} ,
\end{align}
so item~\ref{DirectSumInclusionItem3} of proposition~\ref{MoreGenDecompAndEmbProp} 
implies that
$\WJProj\sub{n} v = 0$ for all vectors $v \notin \Span \{ \MTbas_\ell\super{n} \,|\, 0 \leq \ell \leq n \}$.
Because $\WJProj\sub{n}$ acts as an $\Uqsltwo$-homomorphism by lemma~\ref{UqHomoLemN},  
it remains to find the value of $\WJProj\sub{n} \smash{\MTbas_0\super{n}}$. 
To this end, the case of $n = 1$ being trivial, 
we assume that $\WJProj\sub{n - 1} \smash{\MTbas_0\super{n - 1}} = \smash{\MTbas_0\super{n - 1}}$ 
and apply recursion~\eqref{wjrecursion} 
along with lemmas~\ref{HomoLem} and~\ref{TensorLem} to obtain
\begin{align} 
\nonumber
\WJProj\sub{n} \MTbas_0\super{n}
:= \Trep_{n}(\WJProj\sub{n}) (\MTbas_0\super{n})
& \underset{\eqref{TensorID}}{\overset{\eqref{wjrecursion}}{=}}
\Big( \Trep_{n-1}(\WJproj\sub{n-1}) \otimes \id + \frac{[n-1]}{[n]} \Trep_{n-1}(\WJproj\sub{n-1} \Gen_{n-1} \WJproj\sub{n-1}) \Big)
\big( \MTbas_0\super{n} \big) \\
\nonumber
& \underset{\eqref{TensorID}}{\overset{\eqref{HomoLikeProp}}{=}}
\Big( \id + \frac{[n-1]}{[n]} \Trep_{n-1}(\WJproj\sub{n-1} \Gen_{n-1}) \Big)
\big( \Trep_{n-1}(\WJproj\sub{n-1}) \otimes \id \big) \big( \MTbas_0\super{n} \big) \\
\nonumber
& \overset{\eqref{MThwv}}{=}
\Big( \id + \frac{[n-1]}{[n]} \Trep_{n-1}(\WJproj\sub{n-1} \Gen_{n-1}) \Big)
\big( \WJproj\sub{n-1} \MTbas_0\super{n-1} \otimes \FundBasis_0 \big) \\
\label{AuxiliaryFormula}
&  \overset{\hphantom{\eqref{MThwv}}}{=}
\Big( \id + \frac{[n-1]}{[n]} \Trep_{n-1}(\WJproj\sub{n-1} \Gen_{n-1}) \Big)
\big( \MTbas_0\super{n-1} \otimes \FundBasis_0 \big).
\end{align}
After factoring $\Gen_{n-1} = \Lgen_{n-1} \Rgen_{n-1}$ according to~\eqref{LRtoGen} 
and using~\eqref{ExtendThis2} to simplify the last line, we arrive with
\begin{align}
\nonumber
& \WJproj\sub{n-1}\Gen_{n-1} \big( \MTbas_0\super{n-1} \otimes \FundBasis_0 \big) 
\overset{\eqref{HomoLikeProp}}{=}  
\WJproj\sub{n-1}\Lgen_{n-1} \Rgen_{n-1} ( \MTbas_0\super{n-1} \otimes \FundBasis_0 \big) 
\underset{\eqref{ExtendThis2}}{\overset{\eqref{MThwv}}{=}} 0 \\ 
 \label{AuxiliaryFormula2}
\Longrightarrow \qquad &
\WJProj\sub{n} \big( \MTbas_0\super{n} \big) 
\overset{\eqref{AuxiliaryFormula}}{=}
\MTbas_0\super{n-1} \otimes \FundBasis_0 
\overset{\eqref{MThwv}}{=} \MTbas_0\super{n} ,
\end{align}
so induction on $n \in \bZpos$ shows that $\WJProj\sub{n} \smash{\MTbas_0\super{n}} = \smash{\MTbas_0\super{n}}$.
To conclude, we recall that by definition, $\Projection\sub{n}$ is the projection 
\begin{align}
\Projection\sub{n}(v) = 
\begin{cases} 
v, & \quad v \in \Span \big\{ \MTbas_\ell\super{n} \,\big|\, 0 \leq \ell \leq n \big\} , \\
0, & \quad \text{otherwise}
\end{cases}
\end{align}
from $\Module{\VecSp_n}{\Uqsltwo}$ onto its unique submodule isomorphic to $\Wd\sub{n}$,
thus coinciding with $\Trep_{n}(\WJProj\sub{n})$. This finishes the proof.
\end{proof}

Similarly, the composite projector $\WJProj_\multii$~\eqref{WJCompProj} corresponds to the submodule 
projector $\Projection_\multii$ or $\ProjectionBar_\multii$, defined in~\eqref{Composites}.

\begin{cor} \label{CompositeProjCor} 
Suppose 
$\max \multii < \pmin(q)$. We have
\begin{align} \label{CompTwoProjs}  
\WJProj_\multii \quad \overset{\Trep_{\Summed_\multii}}{\longmapsto} \quad \Projection_\multii 
\qquad\qquad \textnormal{and} \qquad\qquad
\WJProj_\multii  \quad \overset{\TrepBar_{\Summed_\multii}}{\longmapsto} \quad \ProjectionBar_\multii . 
\end{align}
\end{cor}

\begin{proof} 
Lemma~\ref{WJprojLem} says that $\Trep_s(\WJProj\sub{s}) = \Projection\sub{s}$ for any $s \in \bZpos$.  Combining 
with lemma~\ref{TensorLem}, we find that 
\begin{align} 
\Projection_\multii
&\overset{\eqref{Composites}}{=} \Projection\sub{\sIndex_1} \otimes \Projection\sub{\sIndex_2} \otimes \dotsm \otimes \Projection\sub{\sIndex_{\np_\multii}} 
\overset{\eqref{Recover}}{=}  \Trep_{\sIndex_1}(\WJProj\sub{\sIndex_1}) \otimes \Trep_{\sIndex_2}(\WJProj\sub{\sIndex_2}) \otimes \dotsm 
\otimes \Trep_{\sIndex_{\np_\multii}}(\WJProj\sub{\sIndex_{\np_\multii}}) 
\overset{\eqref{TensorID}}{=} \Trep_{\Summed_\multii} (\WJProj_\multii),
\end{align}
which proves the left equation in~\eqref{CompTwoProjs}. The right equation can be proven similarly.
\end{proof}

\subsection{Valenced Temperley-Lieb representations on type-one $\Uqsltwo$-modules}
\label{GenDiacActTypeOneSec}

Next, we define actions of the valenced Temperley-Lieb algebra $\TL_\multii(\nu)$
on general tensor products of type $\VecSp_\multii$. 
The main result of this section is 
proposition~\ref{HWspaceDecTL}, where we establish an explicit direct-sum decomposition for these $\TL_\multii(\nu)$-modules
$\CModule{\VecSp_\multii}{\TL}$ and $\CRModule{\VecSpBar_\multii}{\TL}$
in terms of $\Uqsltwo,\UqsltwoBar$-highest-weight vectors. 
Slightly more generally, in lemma~\ref{HomoLem2} 
we define 
actions of $(\multii,\multiii)$-valenced tangles $T \in \smash{\TL_\multii^\multiii}$ 
via the map~\eqref{TangleEmbDef} and families~(\ref{BigFamily},~\ref{BigFamilyBar}).
The special case $\multiii = \multii$ gives rise to left and right representations of the algebra $\TL_\multii(\nu)$.
(Again, if $\multiii \neq \multii$, these maps are not representations per se, although composition of tangles is respected.)

\begin{lem} \label{HomoLem2}
\textnormal{(Valenced Temperley-Lieb actions):} 
There exists a unique family of maps 
\begin{align} \label{BigFamily2} 
\big\{ \Trep_\multii^\multiii \colon \TL_\multii^\multiii \longrightarrow 
\Hom (\VecSp_\multiii , \VecSp_\multii) \,\big|\, \max (\multii,\multiii) < \pmin(q), \, \Summed_\multii + \Summed_\multiii \equiv 0 \Mod{2} \big\}
\end{align}
with the following properties:
\begin{enumerate} 
\itemcolor{red}

\item \label{HomoLem21Gen} 
All maps in family~\eqref{BigFamily2} are given by  
the following rule, for all valenced tangles $T \in \smash{\TL_\multii^\multiii}$\textnormal{:}
\begin{align} \label{ImultHom} 
\Trep_\multii^\multiii(T) := 
\Projectionhat_\multii \circ \Trep_{\Summed_\multii}^{\Summed_\multiii}(\WJEmb_\multii T \smash{\WJProjHat}_\multiii) \circ \Embedding_\multiii .
\end{align}

\item \label{HomoLem22Gen} 
For all maps $\smash{\Trep_\multii^\varepsilon}$ and $\smash{\Trep_\varepsilon^\multiii}$ in family~\eqref{BigFamily2} 
and for all valenced tangles $T \in \smash{\TL_\multii^\varepsilon}$ and $U \in \smash{\TL_\varepsilon^\multiii}$, we~have
\begin{align} \label{HomoLikeProp2} 
\Trep_\multii^\multiii(TU) = \Trep_\multii^\varepsilon (T) \circ \Trep_\varepsilon^\multiii (U). 
\end{align}
\end{enumerate}
Similarly, there exists a unique family of maps
\begin{align} \label{BigFamily2Bar} 
\big\{ \TrepBar_\multii^\multiii \colon \TL_\multii^\multiii \longrightarrow \Hom (\VecSpBar_\multii , \VecSpBar_\multiii) \,\big|\, \max (\multii,\multiii) < \pmin(q), \, \Summed_\multii + \Summed_\multiii \equiv 0 \Mod{2} \big\}
\end{align}
with the following properties:
\begin{enumerate}
\itemcolor{red}
\setcounter{enumi}{2}

\item \label{HomoLem21GenBar} 
All maps in family~\eqref{BigFamily2Bar} are given by 
the following rule, for all valenced tangles $T \in \smash{\TL_\multii^\multiii}$\textnormal{:}
\begin{align} \label{ImultHomBar} 
\TrepBar_\multii^\multiii(T) := 
\ProjectionhatBar_\multiii \circ \TrepBar_{\Summed_\multii}^{\Summed_\multiii}(\WJEmb_\multii T \smash{\WJProjHat}_\multiii) \circ \EmbeddingBar_\multii .
\end{align}

\item \label{HomoLem22GenBar} 
For all maps $\smash{\TrepBar_\multii^\varepsilon}$ and $\smash{\TrepBar_\varepsilon^\multiii}$ in family~\eqref{BigFamily2Bar} and for all valenced tangles  
$T \in \smash{\TL_\multii^\varepsilon}$ and $U \in \smash{\TL_\varepsilon^\multiii}$, we have 
\begin{align} \label{HomoLikeProp2Bar} 
\TrepBar_\multii^\multiii(TU) = \TrepBar_\varepsilon^\multiii (U) \circ \TrepBar_\multii^\varepsilon (T). 
\end{align}
\end{enumerate}
\end{lem}

\begin{proof} 
Items~\ref{HomoLem21Gen} and~\ref{HomoLem21GenBar} just define all of the maps in families~(\ref{BigFamily2},~\ref{BigFamily2Bar}). 
To prove item~\ref{HomoLem22Gen} (similarly,~\ref{HomoLem22GenBar}), 
we use this definition
and the corresponding property~\eqref{HomoLikeProp} from lemma~\ref{HomoLem} to obtain
\begin{align} 
\nonumber
\Trep_\multii^\varepsilon(T) \circ \Trep_\varepsilon^\multiii(U) 
& \overset{\eqref{ImultHom}}{=} 
\Projectionhat_\multii \circ \Trep_{\Summed_\multii}^{\Summed_\varepsilon}( \WJEmb_\multii T \smash{\WJProjHat}_\varepsilon ) 
\circ \Embedding_\varepsilon \circ \Projectionhat_\varepsilon \circ \Trep_{\Summed_\varepsilon}^{\Summed_\multiii}( \WJEmb_\varepsilon U \smash{\WJProjHat}_\multiii ) 
\circ \Embedding_\multiii \\
\nonumber
& \overset{\eqref{UQIdComp}}{=}
\Projectionhat_\multii \circ \Trep_{\Summed_\multii}^{\Summed_\varepsilon}( \WJEmb_\multii T \smash{\WJProjHat}_\varepsilon ) 
\circ \Projection_\varepsilon \circ \Trep_{\Summed_\varepsilon}^{\Summed_\multiii}( \WJEmb_\varepsilon U \smash{\WJProjHat}_\multiii ) 
\circ \Embedding_\multiii \\
\nonumber
& \underset{\eqref{CompTwoProjs}}{\overset{\eqref{HomoLikeProp}}{=}} 
\Projectionhat_\multii \circ \Trep_{\Summed_\multii}^{\Summed_\multiii}( \WJEmb_\multii T \smash{\WJProjHat}_\varepsilon \WJProj_\varepsilon
\WJEmb_\varepsilon U \smash{\WJProjHat}_\multiii )  \circ \Embedding_\multiii \\
\nonumber
& \overset{\eqref{IdCompAndWJPhatPEmb}}{=} 
\Projectionhat_\multii \circ \Trep_{\Summed_\multii}^{\Summed_\multiii}( \WJEmb_\multii T \smash{\WJProjHat}_\varepsilon
\WJEmb_\varepsilon U \smash{\WJProjHat}_\multiii )  \circ \Embedding_\multiii \\
& \overset{\eqref{IdCompAndWJPhatPEmb}}{=} 
\Projectionhat_\multii \circ \Trep_{\Summed_\multii}^{\Summed_\multiii}( \WJEmb_\multii (TU) \smash{\WJProjHat}_\multiii ) 
\circ \Embedding_\multiii \overset{\eqref{ImultHom}}{=} \Trep_\multii^\multiii(TU) 
\end{align}
for any two maps $\Trep_\multii^\varepsilon$ and $\Trep_\varepsilon^\multiii$ in family~\eqref{BigFamily2} 
and for any two valenced tangles 
$T \in \smash{\TL_\multii^\varepsilon}$ and $U \in \smash{\TL_\varepsilon^\multiii}$.
\end{proof}

We often use the following shorthand notation, 
for all valenced tangles $T \in \TL_\multii^\multiii$ and vectors $v \in \VecSp_\multiii$ and~$\overbarStraight{v} \in \VecSpBar_\multii$:
\begin{align}
\Trep_\multii^\multiii(T) (v) = T  v
\qquad\qquad \text{and} \qquad\qquad
\TrepBar_\multii^\multii(T) (\overbarStraight{v}) = \overbarStraight{v} T .
\end{align}

\begin{cor} \label{NoShiftPropertyCor}
\textnormal{($s$-grading preservation):}
Suppose 
$\max (\multii, \multiii) < \pmin(q)$. 
We have
\begin{align} \label{NoShiftProperty}
v \in \Ksp_\multiii\super{s} , \quad 
\overbarStraight{v} \in \KspBar_\multii\super{s} ,
\quad \textnormal{and} \quad T \in \TL_\multii^\multiii
\qquad \qquad \Longrightarrow \qquad \qquad 
T  v \in \Ksp_\multii\super{s} 
\quad \textnormal{and} \quad
\overbarStraight{v} T  \in \KspBar_\multiii\super{s}  .
\end{align} 
\end{cor}

\begin{proof}
The case $\multii = \OneVec{n}$ and $\multiii = \OneVec{m}$ is the content of corollary~\ref{NoShiftPropertyCorN}, 
and using definitions~(\ref{ImultHom},~\ref{ImultHomBar}) and 
lemma~\ref{EmbProjLem} we readily extend~\eqref{NoShiftPropertyN} to~\eqref{NoShiftProperty} 
for general multiindices $\multii , \multiii \in \smash{\bZpos^\#}$ as in~(\ref{MultiindexNotation},~\ref{ndefn}). 
\end{proof}

Next, we record an analogue to corollary~\ref{CompositeProjCor} for the Jones-Wenzl composite 
embedder and projector~\eqref{WJCompEmbAndProjHat}.

\begin{cor} \label{CompositeProjCorHatEmb}
Suppose 
$\max \multii < \pmin(q)$. We have
\begin{align} \label{CompTwoProjsHatEmb}  
\WJProjHat_\multii \quad \overset{\Trep_{\multii}^{\Summed_\multii}}{\longmapsto} \quad \Projectionhat_\multii
\qquad\qquad \, \textnormal{and} \qquad \qquad
\WJEmb_\multii \quad \overset{\Trep_{\Summed_\multii}^{\multii}}{\longmapsto} \quad \Embedding_\multii ,
\end{align}
and similarly,
\begin{align} \label{CompTwoProjsHatEmbBar}
\; \WJProjHat_\multii \quad \overset{\TrepBar_{\multii}^{\Summed_\multii}}{\longmapsto} \quad \EmbeddingBar_\multii
\; \, \qquad\qquad \textnormal{and} \qquad\qquad
\WJEmb_\multii \quad \overset{\TrepBar_{\Summed_\multii}^{\multii}}{\longmapsto} \quad \ProjectionhatBar_\multii.
\end{align}
\end{cor}

\begin{proof} 
Using the trivial observation that $\smash{\WJProjHat_{\Summed_\multii} = \WJEmb_{\Summed_\multii} = \mathbf{1}_{\TL_{\Summed_\multii}}}$,
with definition~\eqref{ImultHom} and corollary~\ref{CompositeProjCor}, we have
\begin{align} 
\Trep_{\multii}^{\Summed_\multii} (\WJProjHat_\multii) 
\overset{\eqref{ImultHom}}{=} \; &
\Projectionhat_\multii \circ \Trep_{\Summed_\multii}(\WJEmb_\multii \WJProjHat_\multii) 
\overset{\eqref{IdCompAndWJPhatPEmb}}{=} 
\Projectionhat_\multii \circ 
\Trep_{\Summed_\multii}(\WJProj_\multii) 
\overset{\eqref{CompTwoProjs}}{=} \Projectionhat_\multii \circ \Projection_\multii 
\overset{\eqref{PhatP}}{=} \Projectionhat_\multii, \\[.5em]
\Trep_{\Summed_\multii}^{\multii} (\WJEmb_\multii) 
\overset{\eqref{ImultHom}}{=} \; &
\Trep_{\Summed_\multii}(\WJEmb_\multii \smash{\WJProjHat}_\multii) \circ \Embedding_\multii 
\overset{\eqref{IdCompAndWJPhatPEmb}}{=} 
\Trep_{\Summed_\multii}(\WJProj_\multii) \circ \Embedding_\multii 
\overset{\eqref{CompTwoProjs}}{=} \Projection_\multii \circ \Embedding_\multii  
\overset{\eqref{PhatP}}{=} \Embedding_\multii .
\end{align}
This proves~\eqref{CompTwoProjsHatEmb}, and~\eqref{CompTwoProjsHatEmbBar} can be proven similarly. 
\end{proof}

The diagram actions given by lemma~\ref{HomoLem2} are also $\Uqsltwo,\UqsltwoBar$-homomorphisms.

\begin{lem} \label{UqHomoLem2} 
\textnormal{(Quantum group homomorphism properties):}
Suppose $q \in \bC^\times \setminus \{\pm1\}$.
\begin{enumerate}
\itemcolor{red}

\item \label{UqHomo2It1}
For each map $\Trep_\multii^\multiii$ in family~\eqref{BigFamily2}, we have 
$\im \smash{\Trep_\multii^\multiii} \subset \HomMod{\Uqsltwo} ( \VecSp_\multiii , \VecSp_\multii)$. 
In other words, we have
\begin{align} \label{HomProp2} 
T  (x.v) = x.(T  v)
\qquad \textnormal{for all valenced tangles $T \in \smash{\TL_\multii^\multiii}$, elements $x \in \Uqsltwo$, and vectors $v \in \VecSp_\multiii$.}
\end{align}

\item \label{UqHomo2It2}
Similarly, for each map $\smash{\TrepBar_\multii^\multiii}$ in family~\eqref{BigFamily2Bar}, we have 
$\im \smash{\TrepBar_\multii^\multiii} \subset \HomMod{\Uqsltwo} (\VecSpBar_\multii , \VecSpBar_\multiii )$.
In other words, we have
\begin{align} \label{HomProp2Bar} 
T  (x.v) = x.(T  v)
\qquad \textnormal{for all valenced tangles $T \in \smash{\TL_\multii^\multiii}$, elements $x \in \smash{\Uqsltwo}$, and vectors $\overbarStraight{v} \in \VecSpBar_\multii$.}
\end{align}
\end{enumerate}
Similarly, this lemma holds after the symbolic replacements
$x \mapsto \overbarStraight{x}$ and $\Uqsltwo \mapsto \UqsltwoBar$.
\end{lem}

\begin{proof}
Lemmas~\ref{EmbProjLem} and~\ref{UqHomoLemN} show that all of the maps
$\Embedding_\multiii$, $\smash{\Projectionhat_\multii}$, $\EmbeddingBar_\multiii$, $\smash{\ProjectionhatBar_\multii}$,
$\smash{\Trep_{\Summed_\multii}^{\Summed_\multiii}(\WJEmb_\multii T \smash{\WJProjHat}_\multiii)}$  
and $\smash{\TrepBar_{\Summed_\multii}^{\Summed_\multiii}(\WJEmb_\multii T \smash{\WJProjHat}_\multiii)}$,
are $\Uqsltwo,\UqsltwoBar$-homomorphisms.  
The assertions follow from this by definitions~(\ref{ImultHom},~\ref{ImultHomBar}). 
\end{proof}

\begin{cor} \label{NoShiftPropertyCorHWV}
Suppose 
$\max (\multii, \multiii) < \pmin(q)$. 
We have
\begin{align} \label{NoShiftPropertyHWV}
v \in \HWsp_\multiii\super{s} , \quad 
\overbarStraight{v} \in \HWspBar_\multii\super{s} ,
\quad \textnormal{and} \quad T \in \TL_\multii^\multiii
\qquad \qquad \Longrightarrow \qquad \qquad 
T  v \in \HWsp_\multii\super{s} 
\quad \textnormal{and} \quad
\overbarStraight{v} T  \in \HWspBar_\multiii\super{s}  .
\end{align} 
\end{cor}

\begin{proof}
Asserted property~\eqref{NoShiftPropertyHWV} follows from
corollary~\ref{NoShiftPropertyCor}, lemma~\ref{UqHomoLem2}, and definition~\eqref{HWspace2}.
\end{proof}

We denote $\Trep_\multii := \smash{\Trep_\multii^\multiii}$ and $\TrepBar_\multii := \smash{\TrepBar_\multii^\multiii}$ when $\multiii = \multii$.
These maps 
\begin{align}
\Trep_\multii \colon \TL_\multii(\nu) \longrightarrow \End \VecSp_\multii
\qquad \qquad \text{and} \qquad \qquad 
\TrepBar_\multii \colon \TL_\multii(\nu) \longrightarrow \EndOp \VecSpBar_\multii
\end{align}
are respectively left and right representations of the valenced Temperley-Lieb algebra, 
because they send the valenced unit tangle $\mathbf{1}_{\TL_\multii}$~\eqref{TL_valenced_Unit} to the identity map,
by corollary~\ref{RepCor} and definitions~(\ref{ImultHom},~\ref{ImultHomBar}).
We next 
investigate the structure of the $\TL_\multii(\nu)$-modules $\CModule{\VecSp_\multii}{\TL}$ and $\CRModule{\VecSpBar_\multii}{\TL}$ in more detail.
In proposition~\ref{HWspaceDecTL}, we establish a direct-sum decomposition for these $\TL_\multii(\nu)$-modules in terms of 
$\Uqsltwo,\UqsltwoBar$-highest-weight vectors, when $\Summed_\multii < \pmin(q)$. 
Theorem~\ref{HighQSchurWeylThm2} in section~\ref{QSWProofSec} 
upgrades this direct-sum decomposition into a quantum Schur-Weyl duality. 
We will also prove in section~\ref{KerImSubSec} that these representations are always faithful (even if $\Summed_\multii \geq \pmin(q)$). 
This follows from proposition~\ref{PreFaithfulPropGen}, which says that in fact, all 
maps $\smash{\Trep_\multii^{\multiii}}$ and $\smash{\TrepBar_\multii^{\multiii}}$ in lemma~\ref{HomoLem2} are linear injections.

\begin{prop} \label{HWspaceDecTL}
Suppose 
$\max \multii < \pmin(q)$. 
\begin{enumerate}
\itemcolor{red}

\item \label{HWspaceDecTLItem1}
The vector space $\HWsp_\multii$ is closed under the left $\TL_\multii(\nu)$-action on it. 

\item \label{HWspaceDecTLItem2}
For each $s \in \DefectSet_\multii$, the vector space 
$\smash{\HWsp_\multii\super{s}}$ is closed under the left $\TL_\multii(\nu)$-action on it. 

\item \label{HWspaceDecTLItem3}
For each $0\leq\ell\leq s < \pmin(q)$, the left $\TL_\multii(\nu)$-modules $\CModule{\HWsp_\multii\super{s}}{\TL}$ and
$\CModule{F^\ell.\HWsp_\multii\super{s}}{\TL}$ are isomorphic.

\item \label{HWspaceDecTLItem5} 
If $\Summed_\multii < \pmin(q)$, then we have the following isomorphism of left $\TL_\multii(\nu)$-modules: 
\begin{align}
\label{TLDirDecomp}
\CModule{\VecSp_\multii}{\TL}
\isom
\bigoplus_{s \, \in \, \DefectSet_\multii} (s + 1) \, \CModule{\HWsp_\multii\super{s}}{\TL} .
\end{align}
\end{enumerate}
Similarly, this proposition holds for right $\TL_\multii(\nu)$-modules, after the symbolic replacements
\begin{align} 
\HWsp \mapsto \HWspBar , \qquad
\CModule{\HWsp_\multii\super{s}}{\TL} \mapsto \CRModule{\HWspBar_\multii\super{s}}{\TL} , \qquad
F^\ell.\HWsp_\multii\super{s} \mapsto \HWspBar_\multii\super{s}.E^\ell , 
\qquad \textnormal{and} \qquad 
\CModule{\VecSp_\multii}{\TL} \mapsto \CRModule{\VecSpBar_\multii}{\TL} .
\end{align}
\end{prop}

\begin{proof} 
Items~\ref{HWspaceDecTLItem1}--\ref{HWspaceDecTLItem2} immediately follow from corollary~\ref{NoShiftPropertyCorHWV}.
To prove item~\ref{HWspaceDecTLItem3}, we note that
for each $0\leq\ell\leq s < \pmin(q)$, the restriction of the action of $F^\ell$ 
to a map from $\smash{\HWsp_\multii\super{s}}$ onto $\smash{F^\ell.\HWsp_\multii\super{s}}$ 
is an isomorphism of vector spaces with inverse 
$[s]! [\ell]![s - \ell]!^{-1} E^\ell$ 
and, by lemma~\ref{UqHomoLem2}, a homomorphism of $\TL_\multii(\nu)$-modules.
To prove item~\ref{HWspaceDecTLItem5}, we note that
when $\Summed_\multii < \pmin(q)$, proposition~\ref{MoreGenDecompAndEmbProp} 
implies the following isomorphism of vector spaces, with $s \leq \Summed_\multii  < \pmin(q)$ by~\eqref{DefSet2}:
\begin{align} \label{MoreGenDecompVecSp}
\bigoplus_{s \, \in \, \DefectSet_\multii}
\bigoplus_{\ell \, = \, 0}^s F^\ell.\HWsp_\multii\super{s} \isom \VecSp_\multii .
\end{align}
Indeed, the $K$-eigenvalues of the different summands
$\smash{F^\ell.\HWsp_\multii\super{s}}$ over $\ell$ for fixed $s$ are distinct, so their sum is direct.
These summands are all isomorphic as $\TL_\multii(\nu)$-modules by item~\ref{HWspaceDecTLItem3},
and item~\ref{HWspaceDecTLItem2} shows that the $\TL_\multii(\nu)$-action on $\CModule{\VecSp_\multii}{\TL}$ respects 
the $s$-grading on the left side of~\eqref{MoreGenDecompVecSp}.
Asserted isomorphism~\eqref{TLDirDecomp} of $\TL_\multii(\nu)$-modules then follows.
Finally, the statements concerning the right $\TL_\multii(\nu)$-action on $\CRModule{\VecSpBar_\multii}{\TL}$ can be proven similarly.
\end{proof}

We show in proposition~\ref{HWspLem2} in section~\ref{subsec: link state hwv correspondence} that, when $\Summed_\multii < \pmin(q)$,
the $\TL_\multii(\nu)$-modules $\CModule{\HWsp_\multii\super{s}}{\TL}$ are in fact isomorphic to 
the standard modules $\smash{\LS_\multii\super{s}}$ (cf. section~\ref{TLReviewSec}).
In particular, we then know from~\cite{fp3a} that in this case, the algebra $\TL_\multii(\nu)$ is semisimple and 
$\smash{\LS_\multii\super{s}}$ 
are all of its simple modules, so proposition~\ref{HWspLem2} 
also implies that $\CModule{\HWsp_\multii\super{s}}{\TL}$ are simple. 
If $\Summed_\multii \geq \pmin(q)$, the modules $\CModule{\HWsp_\multii\super{s}}{\TL}$ might not be simple,
and $\TL_\multii(\nu)$ is generally not semisimple. We discuss this case in proposition~\ref{QuotientProp}
in section~\ref{subsec: link state hwv correspondence}.

\subsection{Link state bilinear pairing and graphical calculus for the spin chain}
\label{LSandBiformandSCGrapgSec}

Next, we define a bilinear pairing of valenced link patterns (see also~\cite[section~\red{3A}]{fp3a}). 
We begin with the special case of $\multii = \OneVec{n}$ for some $n \in \bZpos$.
To this end, 
given two link patterns $\smash{\alphaBar} \in \LPBar_n$ and $\beta \in \LP_n$,
we concatenate $\smash{\alphaBar}$ to $\beta$ from below, and delete the overlapping 
horizontal lines.  The resulting diagram is a network $\smash{\alphaBar} \BarAction \beta$.  For instance, 
\begin{align} \label{LSBiformExamples}
& \alphaBar \; = \; 
\raisebox{-10pt}{\includegraphics[scale=0.275]{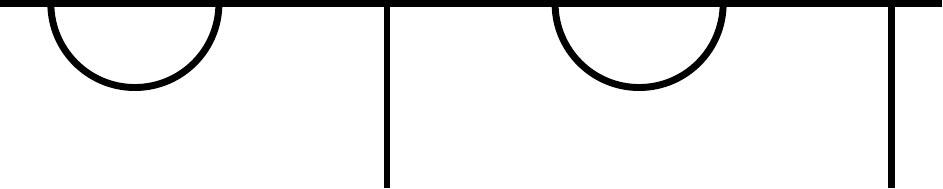}} \; , \qquad 
\beta \; = \; \raisebox{1pt}{\includegraphics[scale=0.275]{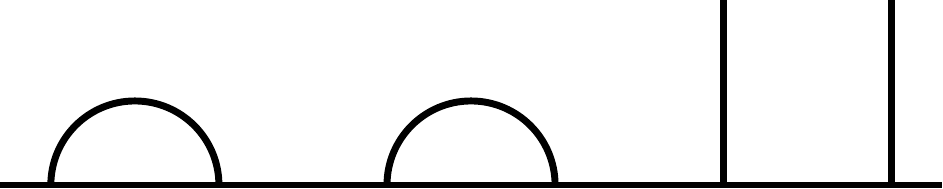}}
\qquad \qquad \Longrightarrow \qquad \qquad
\alphaBar \BarAction \beta \; = \;
\vcenter{\hbox{\includegraphics[scale=0.275]{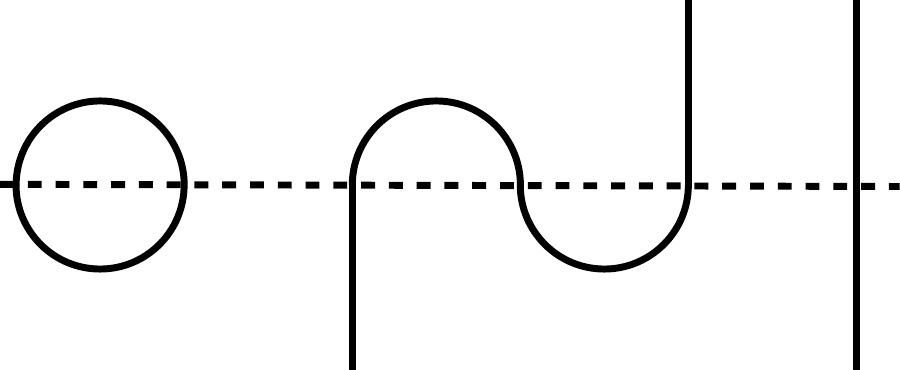}}}
\end{align}
Linear extension of the evaluation of the network $\smash{\alphaBar} \BarAction \beta$ determines a bilinear pairing
$\LSBiFormBar{\cdot}{\cdot} \colon \LSBar_n \times \LS_n \longrightarrow \bC$,
\begin{align}
\nonumber
\LSBiFormBar{\alphaBar}{\beta} 
:= & \; \prod \{ \text{the weights of all connected components in the network $\smash{\alphaBar} \BarAction \beta$} \} \\
= & \; 
\begin{cases} 
\nu^{\textnormal{\# loops in $\smash{\alphaBar} \BarAction \beta$}} , 
& \textnormal{if the network $\smash{\alphaBar} \BarAction \beta$ has no turn-back path} , \\ 
0 , & \textnormal{if the network $\smash{\alphaBar} \BarAction \beta$ has a turn-back path,}
\end{cases}
\end{align}
where we assign all loops, through-paths, and turn-back paths the following weights in $\bC$:
\begin{alignat}{7} 
\label{TLfugacity} 
&\hspace*{-3mm}  \text{loop weight (fugacity):} \quad \qquad 
& \vcenter{\hbox{\includegraphics[scale=0.275]{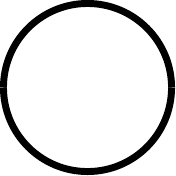}}}  \qquad & \text{and} \qquad  
& \raisebox{-18pt}{\includegraphics[scale=0.275]{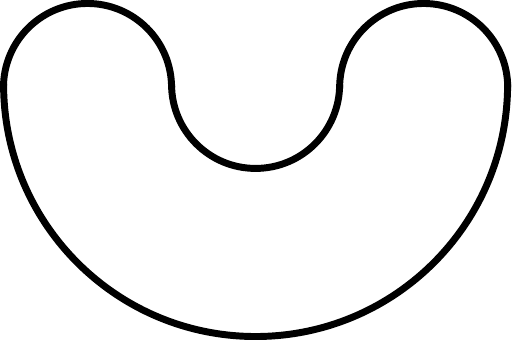}}  \qquad & \text{and} \qquad
& \raisebox{-11pt}{\includegraphics[scale=0.275]{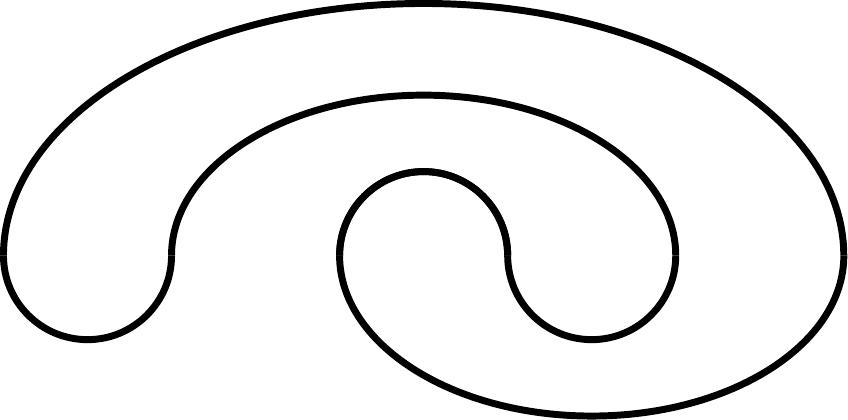}}  \qquad & \text{etc.}  
\quad = \quad \nu , \\[1em]
\label{TLThroughPathWeight}
& \hspace*{-3mm} \text{through-path weight:} \quad \qquad
& \vcenter{\hbox{\hspace*{-4mm} \includegraphics[scale=0.275]{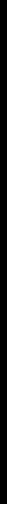}}}  \qquad & \text{and} \qquad
& \vcenter{\hbox{\includegraphics[scale=0.275]{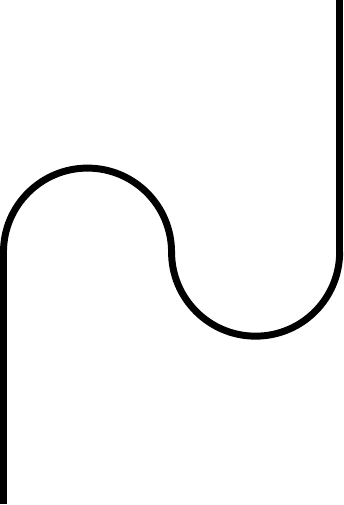} \hspace*{2mm}}}  \qquad & \text{and} \qquad
& \vcenter{\hbox{\includegraphics[scale=0.275]{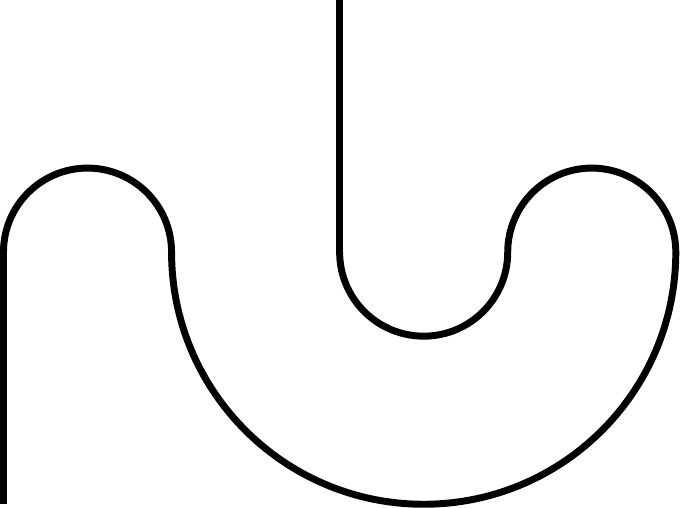}\hspace*{3mm}}}  \qquad & \text{etc.} 
\quad = \quad 1 , \\[1em]
\label{TLTurnBack0}
& \hspace*{-3mm} \text{turn-back path weight:} \quad \qquad
& \raisebox{-5pt}{\includegraphics[scale=0.275]{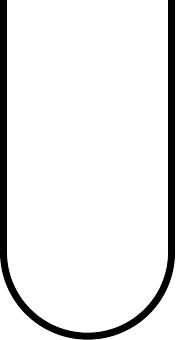}}  \qquad & \text{and} \qquad
& \raisebox{-5pt}{\includegraphics[scale=0.275]{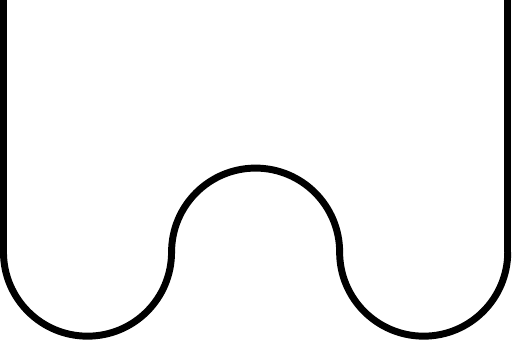}}  \qquad & \text{and} \qquad
& \raisebox{-19pt}{\includegraphics[scale=0.275]{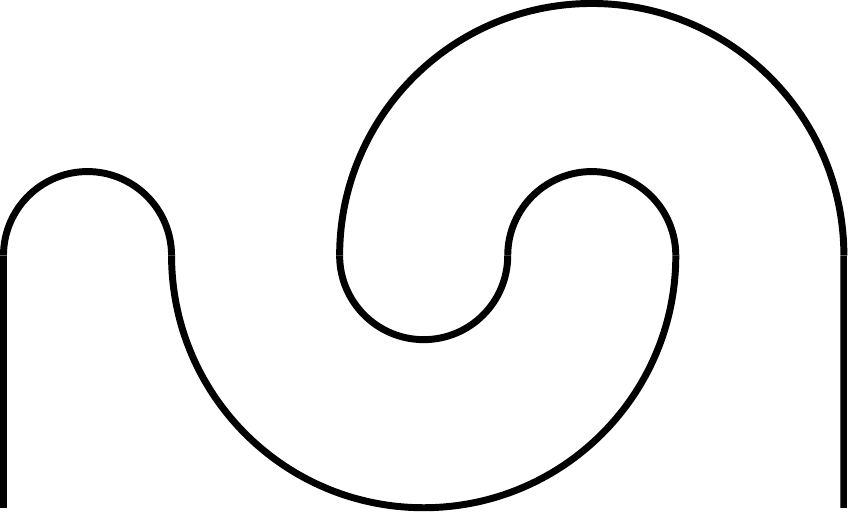}}  \qquad & \text{etc.} 
\quad = \quad 0 .
\end{alignat}
(Actually, we have already been using rules~\eqref{TLfugacity} and~\eqref{TLTurnBack0} in diagram concatenation in section~\ref{TLReviewSec}.)
Using this, we then define for general multiindices a bilinear pairing
$\LSBiFormBar{\cdot}{\cdot} \colon \LSBar_\multii \times \LS_\multii \longrightarrow \bC$ by
\begin{align}
\label{LSBiFormExt} 
\LSBiFormBar{\alphaBar}{\beta} &:= \LSBiFormBar{\alphaBar \WJProjHat_\multii}{\WJEmb_\multii \beta} .
\end{align}
For instance, 
\begin{align} \label{LSBiformExamplesValenced}
\bigg( \; 
\raisebox{-10pt}{\includegraphics[scale=0.275]{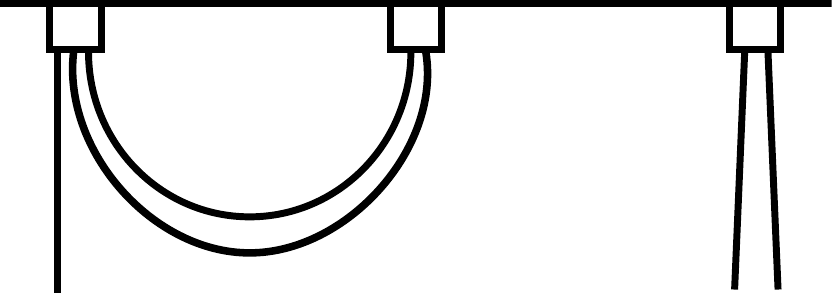}}  
\;\; \raisebox{-8pt}{,} \;\;
\raisebox{-5pt}{\includegraphics[scale=0.275]{Figures/e-LinkPattern3_valenced.pdf}}
\; \bigg)
\quad = \quad \bigg( \; \vcenter{\hbox{\includegraphics[scale=0.275]{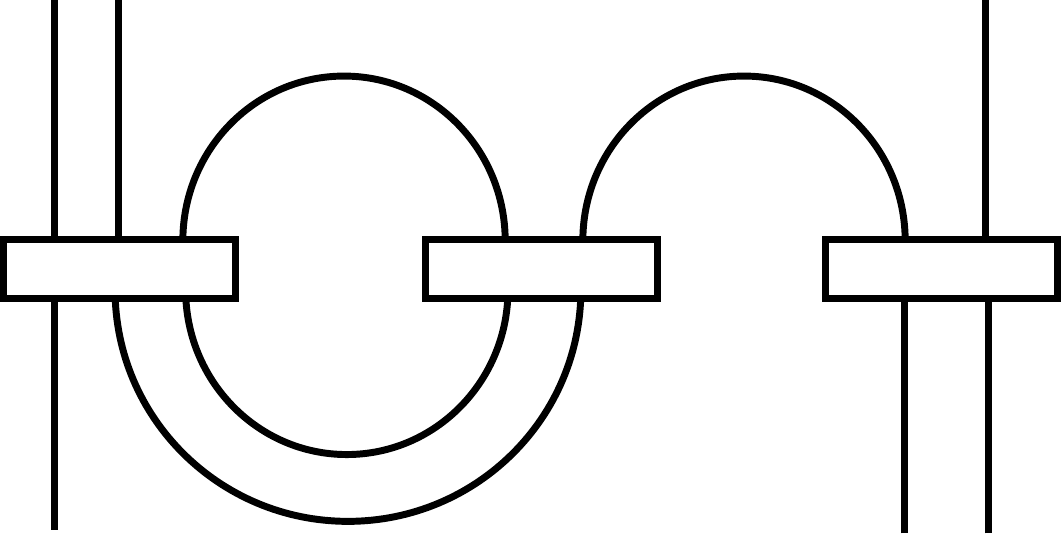}}} \; \bigg) .
\end{align}
The bilinear pairing $\LSBiFormBar{\cdot}{\cdot}$ is $\TL_\multii(\nu)$-invariant in the sense that~\cite[lemma~\red{3.1}]{fp3a}
\begin{align} \label{InvarProp} 
\LSBiFormBar{\alphaBar}{T \beta} = \LSBiFormBar{\alphaBar T}{\beta} 
\end{align}
for all valenced link patterns $\alphaBar \in \LSBar_\multii$ and $\beta \in \LS_\multii$ and 
for all valenced tangles $T \in \TL_\multii(\nu)$.
Using this property, it is straightforward to check that the spaces 
\begin{align} 
\label{LSRadical}
\rad \LS_\multii := \; & \big\{\alpha \in \LS_\multii \, \big| \, \text{$\LSBiFormBar{\betaBar}{\alpha} = 0$ for all $\betaBar \in \LSBar_\multii$} \big\}
= \bigoplus_{s \, \in \, \DefectSet_\multii} \rad \LS_\multii\super{s} , \\
\label{LSRadicalS}
\rad \LS_\multii\super{s} := \; & \big\{\alpha \in \LS_\multii\super{s} \, \big| \, \text{$\LSBiFormBar{\betaBar}{\alpha} = 0$
for all $\betaBar \in \LSBar_\multii\super{s}$} \big\} ,
\end{align}
are $\TL_\multii(\nu)$-submodules of $\LS_\multii$ and $\smash{\LS_\multii\super{s}}$, respectively. 
We denote the corresponding quotient modules by 
\begin{align} 
\label{QuoSp}
\Quo_\multii\super{s} := \LS_\multii\super{s} / \rad \LS_\multii\super{s}  
\qquad \qquad
\text{and}
\qquad \qquad
\Quo_\multii := \LS_\multii / \rad \LS_\multii 
= \bigoplus_{s \, \in \, \DefectSet_\multii} \Quo_\multii\super{s} .
\end{align}
We analogously define the right $\TL_\multii(\nu)$-modules 
$\rad \LSBar_\multii$, $\rad\smash{\LSBar_\multii\super{s}}$, $\smash{\QuoBar_\multii\super{s}}$, and $\QuoBar_\multii$.

By~\cite[proposition~\red{6.7}]{fp3a}, the collection 
$\smash{\{ \Quo_\multii\super{s} \,| \, s \in \DefectSet_{\multii}, \dim \Quo_\multii\super{s} > 0 \}}$ 
is the complete set of non-isomorphic simple left $\TL_\multii(\nu)$-modules.
Also, if $\Summed_\multii < \pmin(q)$, then by~\cite[theorem~\red{6.9}]{fp3a}, $\rad \LS_\multii = \{0\}$, 
the valenced Temperley-Lieb algebra $\TL_\multii(\nu)$ is semisimple, and the collection 
$\smash{\{ \LS_\multii\super{s} \, | \, s \in \DefectSet_\multii \}}$ 
is the complete set of all simple $\TL_\multii(\nu)$-modules.
Let us also remark that, by~\cite[corollary~\red{3.8}]{fp3a}, 
the link state representation of $\TL_\multii (\nu)$ on $\LS_\multii$ is faithful if and only if $\rad \LS_\multii = \{0\}$.  
These facts reflect the cellular structure of the valenced Temperley-Lieb algebra~\cite{gl, gl2, fp3b}. 

\bigskip

Next, we give a diagram representation for vectors in $\VecSp_n$ and $\VecSpBar_n$, 
which is analogous to the one introduced by I.~Frenkel and M.~Khovanov~\cite{fk}
(but we use different conventions, with applications to conformal field theory in mind).
For this purpose, we introduce networks and link states 
with orientation on through-paths, turn-back paths, and defects, but not on loops or links.
We call a collection of nonintersecting, non-self-intersecting planar loops and oriented paths within a rectangle an \emph{oriented network}.
We denote an unspecified orientation 
by $\pm$:
\begin{align}
\vcenter{\hbox{\includegraphics[scale=0.275]{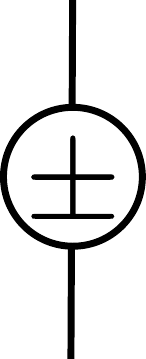}}} \quad = \quad 
\vcenter{\hbox{\includegraphics[scale=0.275]{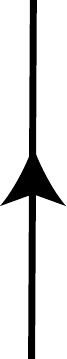}}} \quad \text{ or } \quad
\vcenter{\hbox{\includegraphics[scale=0.275]{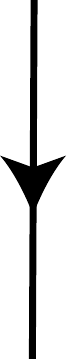} .}} 
\end{align}
We define the evaluation of the oriented network $T$ to be the following complex number: 
\begin{align} \label{evT} 
(\, T \,) &:= \prod \{ \text{the weights of all connected components in the oriented network $T$} \} ,
\end{align}
where, in addition to the loop weight $\nu = - q - q^{-1}$ in~\eqref{TLfugacity} 
(we emphasize that loops are never oriented),  
we assign all oriented through-paths and turn-back paths the following weights in $\bC$:
\begin{alignat}{7} 
\label{ThroughPathWeight}
\hspace*{-4mm} & \text{through-path weight:} \quad \qquad
& \vcenter{\hbox{\hspace*{-3mm} \includegraphics[scale=0.275]{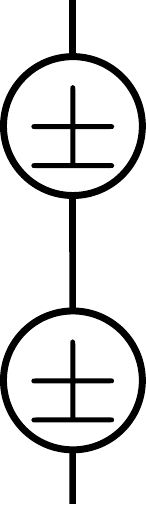}}}  \qquad & \text{and} \qquad
& \vcenter{\hbox{\includegraphics[scale=0.275]{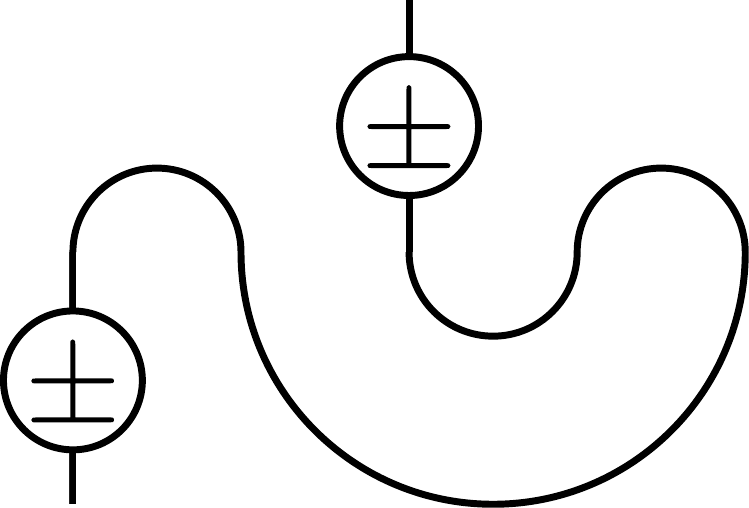}\hspace*{3mm}}}  \qquad & \text{etc.} 
\quad = \quad \hphantom{-} 1 , \\[1em]
\label{TurnBackWeightCW}
\hspace*{-4mm} & \text{clockwise turn-back path weight:} \quad \qquad
& \raisebox{-5pt}{\includegraphics[scale=0.275]{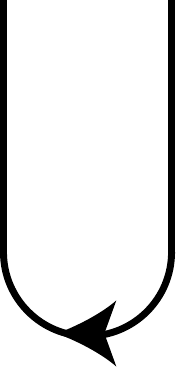}}  \qquad & \text{and} \qquad
& \raisebox{-19pt}{\includegraphics[scale=0.275]{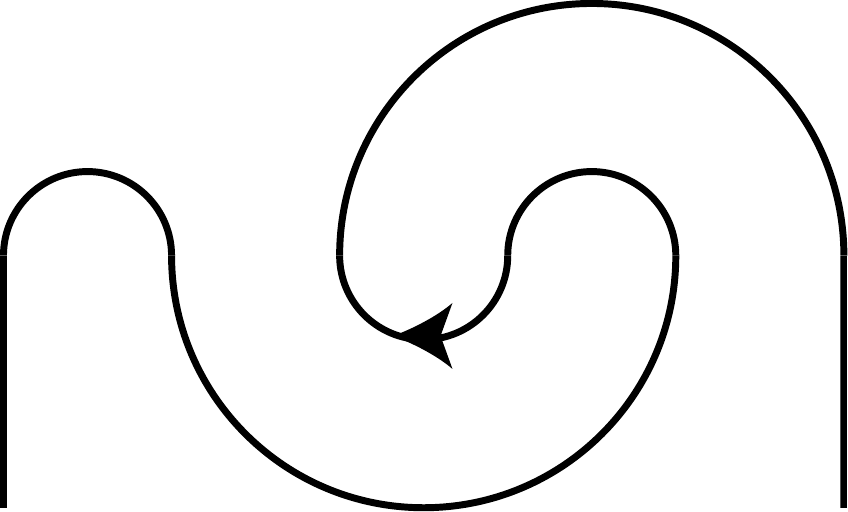}}  \qquad & \text{etc.} 
\quad = \quad \hphantom{-} \ii q^{1/2}  , \\[1em]
\label{TurnBackWeightCCW}
\hspace*{-4mm} & \text{counter-clockwise turn-back path weight:} \quad \qquad
& \raisebox{-5pt}{\includegraphics[scale=0.275]{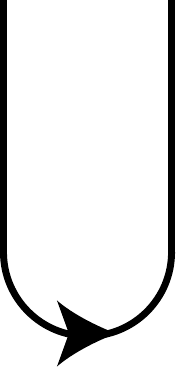}}  \qquad & \text{and} \qquad
& \raisebox{-19pt}{\includegraphics[scale=0.275]{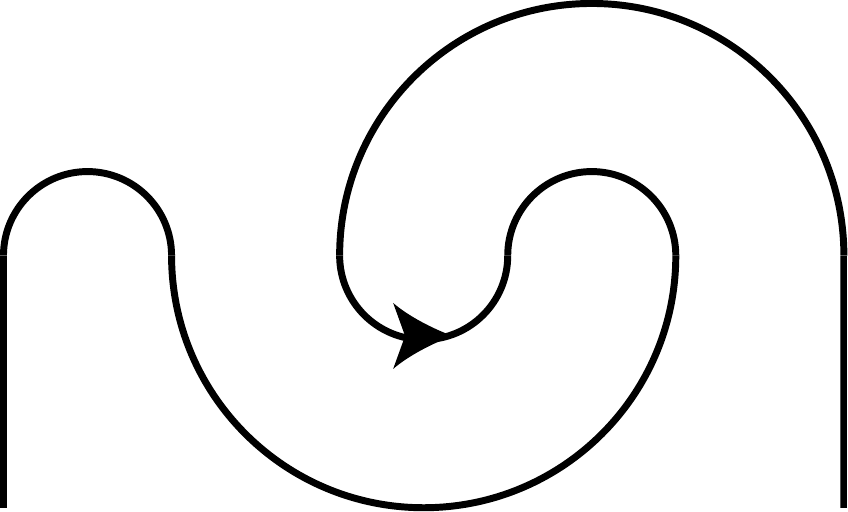}}  \qquad & \text{etc.} 
\quad = \quad - \ii q^{-1/2}  ,
\end{alignat}
and when encountering clashing orientations, we assign the weight zero:
\begin{alignat}{7} 
\label{ThroughPathWeightZero} 
\hspace*{-4mm} & \text{through-path weight:} \quad \qquad
& \vcenter{\hbox{\hspace*{-8mm} \includegraphics[scale=0.275]{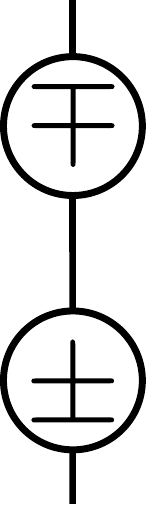}}}  \qquad & \text{and} \qquad
& \vcenter{\hbox{\includegraphics[scale=0.275]{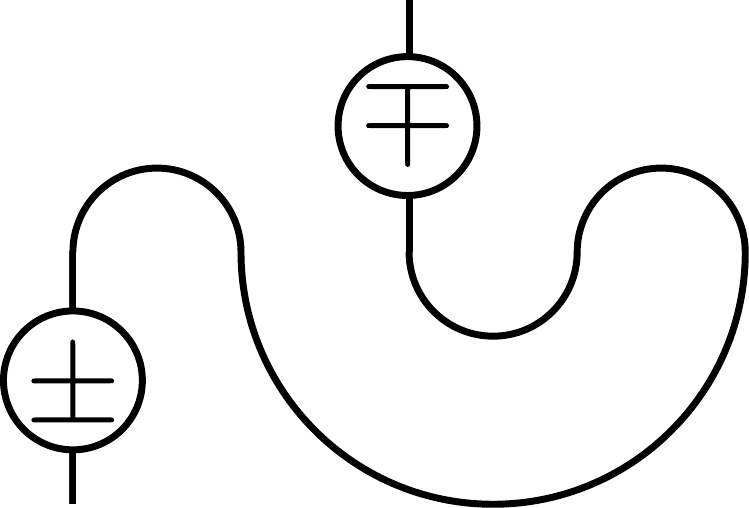}\hspace*{3mm}}}  \qquad & \text{etc.} 
\quad = \quad 0 , \\[1em]
\label{TurnBackWeightZero}
\hspace*{-4mm} & \text{turn-back path weight:} \quad \qquad
& \raisebox{-5pt}{\includegraphics[scale=0.275]{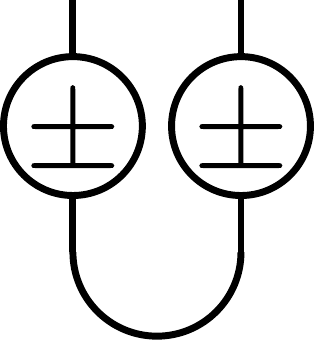}}  \qquad & \text{and} \qquad
& \raisebox{-19pt}{\includegraphics[scale=0.275]{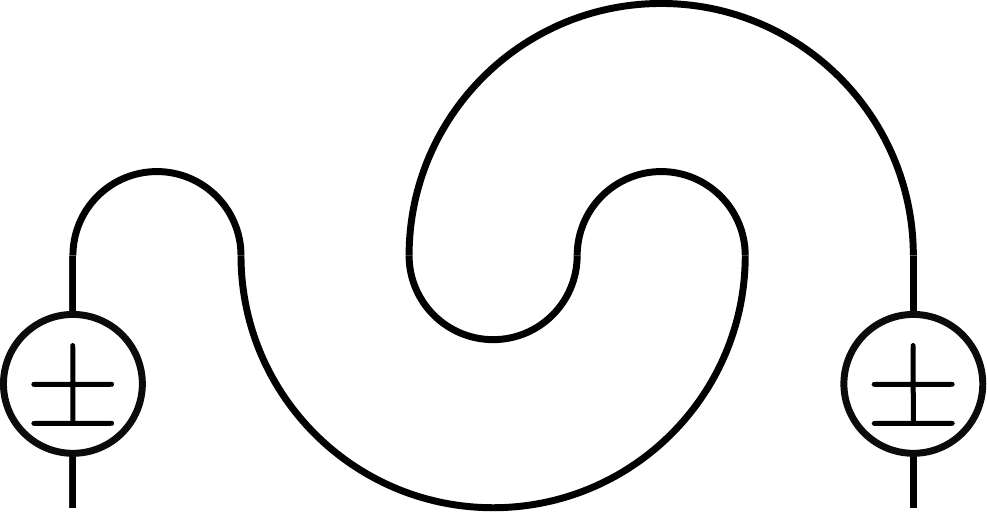}}  \qquad & \text{etc.} 
\quad = \quad 0  .
\end{alignat}
The choice of these weights becomes apparent soon.
We also consider evaluations of formal linear combinations of oriented networks, obtained by linearly extending~\eqref{evT}.
Abusing terminology, we call them oriented networks too.

Analogously, we consider link patterns and link states with oriented defects.
As in~\eqref{LSBiformExamples}, we can produce oriented networks from pairs of such link states:
\begin{align} \label{LSBiformExamplesOriented}
& \alphaBar \; = \; \raisebox{-10pt}{\includegraphics[scale=0.275]{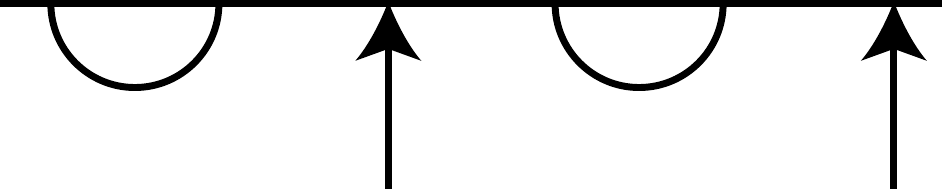}} \; , \qquad 
\beta \; = \; \raisebox{1pt}{\includegraphics[scale=0.275]{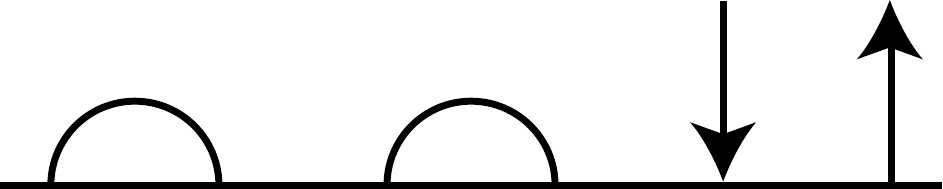}}
\qquad \qquad \Longrightarrow \qquad \qquad
\alphaBar \BarAction \beta \; = \;
\vcenter{\hbox{\includegraphics[scale=0.275]{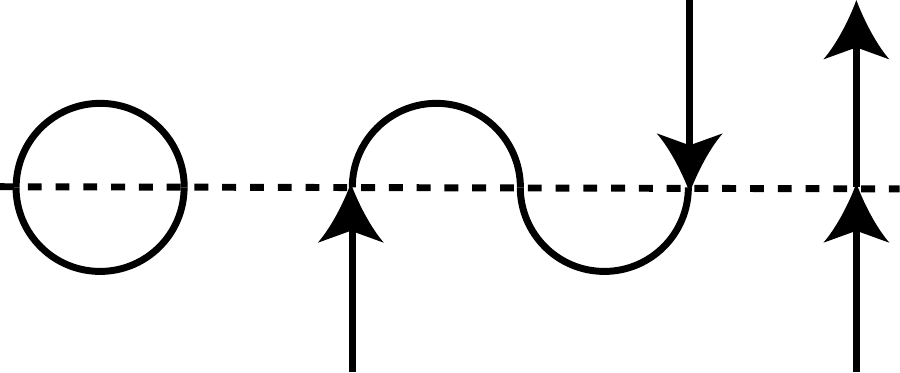} .}}
\end{align}
Analogously to~\eqref{concatenateTanglesN}, we define the bilinear tensor product $\alpha \otimes \beta$ 
of link states (with oriented or non-oriented defects) via bilinear extension of the operation of concatenating $\alpha$ to the left of $\beta$:
\begin{align}  \label{concatenateLinkStatesN}  
\alpha \otimes \beta 
\quad := \quad \vcenter{\hbox{\includegraphics[scale=0.275]{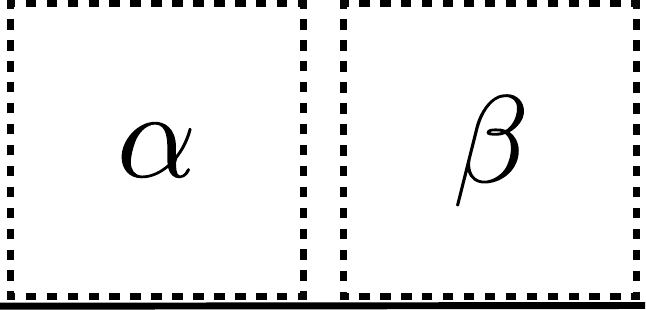} .}}
\end{align}

Now, for the basis vectors in $\VecSp\sub{1} = \Span \{ \FundBasis_0, \FundBasis_1 \}$
and $\VecSpBar\sub{1} = \Span \{ \FundBasisBar_0, \FundBasisBar_1 \}$, we use the notation~\cite{cfs, fk}
\begin{align} \label{OrientedDefects} 
\FundBasis_0 \quad = \quad \vcenter{\hbox{\includegraphics[scale=0.275]{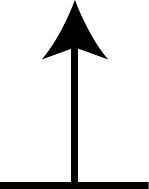} ,}}
\qquad \qquad 
\FundBasis_1 \quad = \quad \vcenter{\hbox{\includegraphics[scale=0.275]{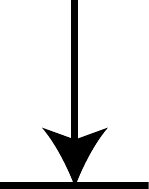} ,}}
\qquad \qquad 
\FundBasisBar_0 \quad = \quad \vcenter{\hbox{\includegraphics[scale=0.275]{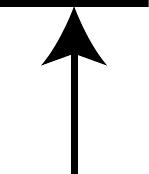} ,}}
\qquad \qquad \textnormal{and} \qquad\qquad 
\FundBasisBar_1 \quad = \quad \vcenter{\hbox{\includegraphics[scale=0.275]{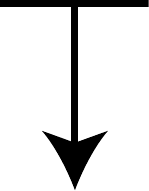} ,}} 
\end{align}
and the standard bases
$\{\FundBasis_{\ell_1} \otimes \FundBasis_{\ell_2} \otimes \dotsm \otimes \FundBasis_{\ell_n} \, | \, \ell_1, \ldots, \ell_n \in \{0,1\} \}$ 
and $\{\FundBasisBar_{\ell_1} \otimes \FundBasisBar_{\ell_2} \otimes \dotsm \otimes \FundBasisBar_{\ell_n} \, | \, \ell_1, \ldots, \ell_n \in \{0,1\} \}$ 
of $\VecSp_n$ and $\VecSpBar_n$ thus obtain a diagram notation via~(\ref{concatenateLinkStatesN},~\ref{OrientedDefects}): for instance, 
\begin{align} \label{SimpleTensor} 
\FundBasis_0 \otimes \FundBasis_1 \otimes \FundBasis_1
\quad = \quad \vcenter{\hbox{\includegraphics[scale=0.275]{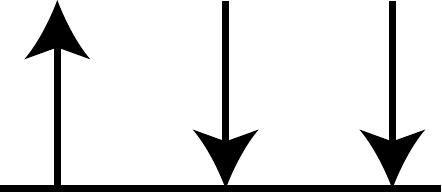}}}
\qquad\qquad \textnormal{and} \qquad\qquad
\FundBasisBar_0 \otimes \FundBasisBar_1 \otimes \FundBasisBar_1 
\quad = \quad \vcenter{\hbox{\includegraphics[scale=0.275]{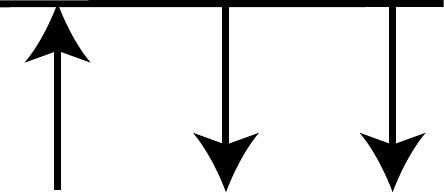} .}} 
\end{align}
We then extend this diagram representation linearly to all vectors in $\VecSp_n$ and $\VecSpBar_n$. 
For example, the singlet vectors $\sing$ and $\singBar$ defined in~\eqref{singletVector} 
read
\begin{align} 
\label{singletDiagramNotation}  
\sing \quad \underset{\eqref{OrientedDefects}}{\overset{\eqref{singletVector}}{=}} \quad  
\ii q^{1/2} \,\, \times \,\, \vcenter{\hbox{\includegraphics[scale=0.275]{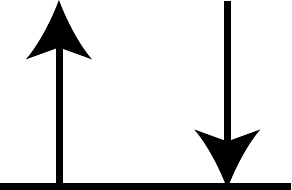}}} 
\quad - \quad \ii q^{-1/2} \,\, \times \,\, \vcenter{\hbox{\includegraphics[scale=0.275]{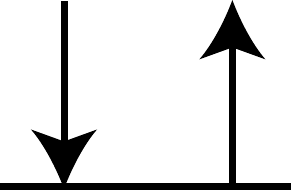} ,}}  \\[1em]
\label{singletDiagramNotationBar}
\singBar
\quad \underset{\eqref{OrientedDefects}}{\overset{\eqref{singletVector}}{=}} \quad  
\ii q^{1/2} \,\, \times \,\, \vcenter{\hbox{\includegraphics[scale=0.275]{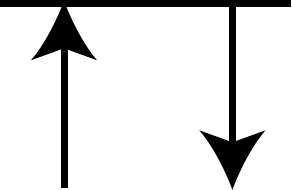}}} 
\quad - \quad \ii q^{-1/2} \,\, \times \,\, \vcenter{\hbox{\includegraphics[scale=0.275]{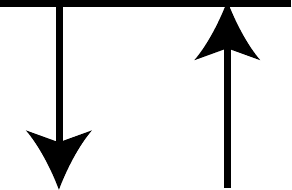} .}}
\end{align} 
By lemmas~\ref{BiFormDefLem} and~\ref{biformPropertyLem}, 
weights~(\ref{ThroughPathWeight},~\ref{ThroughPathWeightZero}) guarantee 
that evaluations of the oriented networks associated to 
graphical representations~(\ref{OrientedDefects},~\ref{SimpleTensor}) agree with the bilinear pairing~\eqref{biformnormalization}:
\begin{align}
\label{TrivialBLform1}
& 1 
\; \underset{\hphantom{\eqref{ThroughPathWeightZero}}}{\overset{\eqref{ThroughPathWeight}}{=}} \quad
\left( \; \vcenter{\hbox{\includegraphics[scale=0.275]{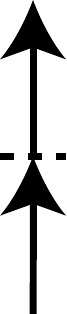}}} \; \right)
\quad \overset{\eqref{OrientedDefects}}{=} \;
\SPBiForm{\FundBasisBar_0}{\FundBasis_0}
\; \overset{\eqref{biformnormalization}}{=} \; 1 ,
&& 1 
\; \underset{\hphantom{\eqref{ThroughPathWeightZero}}}{\overset{\eqref{ThroughPathWeight}}{=}} \quad
\left( \; \vcenter{\hbox{\includegraphics[scale=0.275]{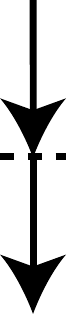}}} \; \right)
\quad \overset{\eqref{OrientedDefects}}{=} \;
\SPBiForm{\FundBasisBar_1}{\FundBasis_1}
\; \overset{\eqref{biformnormalization}}{=} \; 1 ,  \\[1em]
\label{TrivialBLform2}
& 0 
\; \underset{\hphantom{\eqref{ThroughPathWeight}}}{\overset{\eqref{ThroughPathWeightZero}}{=}} \quad
\left( \; \vcenter{\hbox{\includegraphics[scale=0.275]{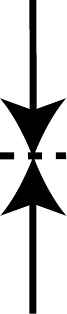}}} \; \right)
\quad \overset{\eqref{OrientedDefects}}{=} \;
\SPBiForm{\FundBasisBar_1}{\FundBasis_0}
\; \overset{\eqref{biformnormalization}}{=} \; 0 , 
&& 0 
\; \underset{\hphantom{\eqref{ThroughPathWeight}}}{\overset{\eqref{ThroughPathWeightZero}}{=}} \quad
\left( \; \vcenter{\hbox{\includegraphics[scale=0.275]{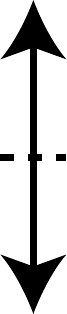}}} \; \right)
\quad \overset{\eqref{OrientedDefects}}{=} \;
\SPBiForm{\FundBasisBar_0}{\FundBasis_1}
\; \overset{\eqref{biformnormalization}}{=} \; 0 .
\end{align}

Thanks to lemma~\ref{HomoLem}, rules~(\ref{ExtendThis0}--\ref{ExtendThisBar2}) 
and weights~(\ref{TurnBackWeightCW},~\ref{TurnBackWeightCCW},~\ref{TurnBackWeightZero}) 
amount to natural graphical rules for the $\smash{\TL_n^m}$-action.
Indeed, the left action~\eqref{ExtendThis2} of the right generators is immediately governed by the rules for $\Rgen_1$:
\begin{align}
\label{RActOnE00}
\Rgen_1 (\FundBasis_0 \otimes \FundBasis_0) \quad \overset{\eqref{ExtendThis2}}{=} \quad & \quad 0 \qquad \;
\quad \overset{\eqref{TurnBackWeightZero}}{=} \quad
\vcenter{\hbox{\includegraphics[scale=0.275]{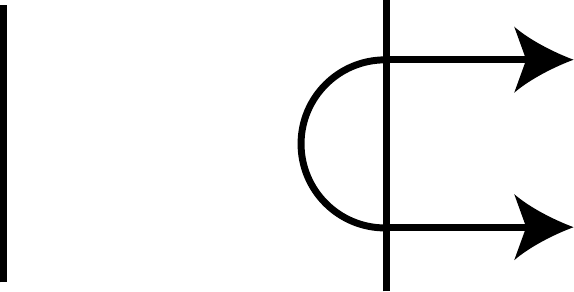}}} 
\quad \underset{\eqref{SimpleTensor}}{\overset{\eqref{LgenForm}}{=}} \quad
\Rgen_1 (\FundBasis_0 \otimes \FundBasis_0) , \\[1em]
\label{RActOnE01}
\Rgen_1 (\FundBasis_0 \otimes \FundBasis_1) \quad \overset{\eqref{ExtendThis2}}{=} \quad & \quad \; \ii q^{1/2} \; \;
\quad \overset{\eqref{TurnBackWeightCW}}{=} \quad 
\vcenter{\hbox{\includegraphics[scale=0.275]{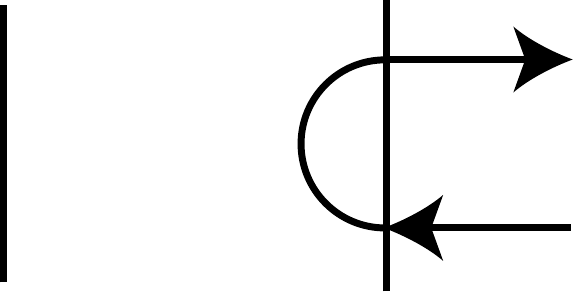}}} 
\quad \underset{\eqref{SimpleTensor}}{\overset{\eqref{LgenForm}}{=}} \quad
\Rgen_1 (\FundBasis_0 \otimes \FundBasis_1) ,  \\[1em]
\label{RActOnE10}
\Rgen_1 (\FundBasis_1 \otimes \FundBasis_0) \quad \overset{\eqref{ExtendThis2}}{=} \quad & - \ii q^{-1/2}
\quad \overset{\eqref{TurnBackWeightCCW}}{=} \quad 
\vcenter{\hbox{\includegraphics[scale=0.275]{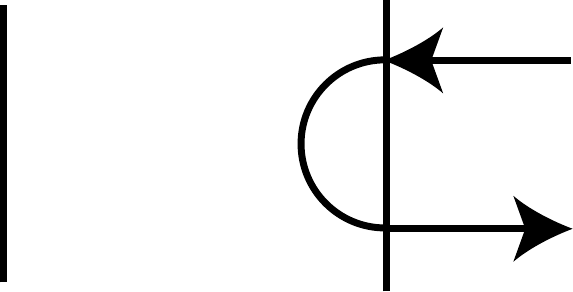}}} 
\quad \underset{\eqref{SimpleTensor}}{\overset{\eqref{LgenForm}}{=}} \quad
\Rgen_1 (\FundBasis_1 \otimes \FundBasis_0) , \\[1em]
\label{RActOnE11}
\Rgen_1 (\FundBasis_1 \otimes \FundBasis_1) \quad \overset{\eqref{ExtendThis2}}{=} \quad & \quad 0 \qquad \;
\quad \overset{\eqref{TurnBackWeightZero}}{=} \quad 
\vcenter{\hbox{\includegraphics[scale=0.275]{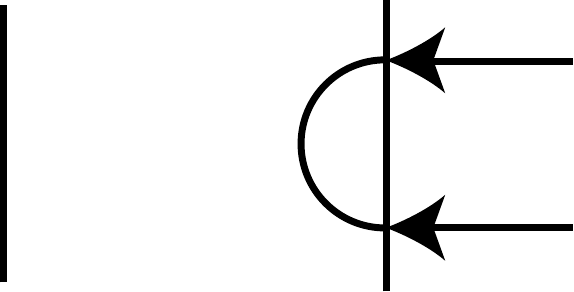}}} 
\quad \underset{\eqref{SimpleTensor}}{\overset{\eqref{LgenForm}}{=}} \quad
\Rgen_1 (\FundBasis_1 \otimes \FundBasis_1) ,
\end{align}
and analogous rules hold for the right action~\eqref{ExtendThisBar2} of the left generators.
The left action~\eqref{ExtendThis1} of the left generators 
(resp.~right action~\eqref{ExtendThisBar1} of the right generators)
gives rise to additional graphical rules: we have
\begin{align}
\label{singletLinkDiagram} 
\sing \quad \overset{\eqref{ExtendThis1}}{=} \quad
\Lgen_1 \big( \Basis_0\super{0} \big) 
\quad \overset{\eqref{LgenForm}}{=} \quad 
\vcenter{\hbox{\includegraphics[scale=0.275]{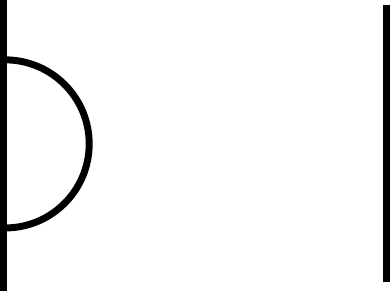}}} 
\quad = \quad
\vcenter{\hbox{\includegraphics[scale=0.275]{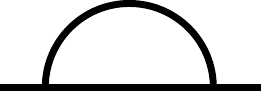},}} \\[1em]
\label{singletLinkDiagramBar} 
\singBar \quad \overset{\eqref{ExtendThisBar1}}{=} \quad
\Rgen_1 \big( \BasisBar_0\super{0} \big) 
\quad \overset{\eqref{LgenForm}}{=} \quad 
\vcenter{\hbox{\includegraphics[scale=0.275]{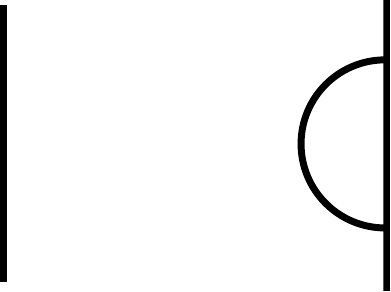}}} 
\quad = \quad
\vcenter{\hbox{\includegraphics[scale=0.275]{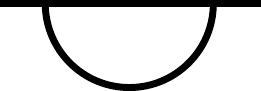}.}}
\end{align}
(We also recall from~\eqref{Rsing} that $\Rgen_1 \sing = \nu$, matching the fugacity weight~\eqref{TLfugacity}.) 
Therefore, in light of lemma~\ref{HomoLem}, 
if we identify the singlet vectors~(\ref{singletDiagramNotation},~\ref{singletDiagramNotationBar}) with simple links as above,
then the above $\smash{\TL_n^m}$-action agrees with the $\smash{\TL_n^m}$-action discussed in section~\ref{DiacActTypeOneSec}.
This also gives a decomposition rule for nested links obtained from repeated application of
the left generators from the left (for $\sing$), or the right generators from the right (for $\singBar$).
For example, 
\begin{align} 
\nonumber
\vcenter{\hbox{\includegraphics[scale=0.275]{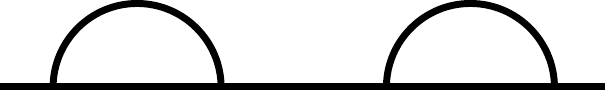}}} 
& \quad \underset{\hphantom{\eqref{singletDiagramNotation}}}{\overset{\eqref{singletLinkDiagram}}{=}} \quad 
\sing \otimes \sing  \\[1em] 
\label{VectorToLinkDiagramDefinitionEx1}
&  \quad \underset{\hphantom{\eqref{singletLinkDiagram}}}{\overset{\eqref{singletDiagramNotation}}{=}} \quad 
 - q \,\, \times \,\, \vcenter{\hbox{\includegraphics[scale=0.275]{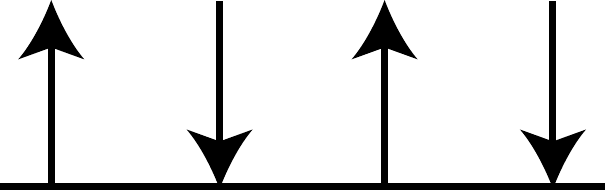}}} 
\quad + \quad \vcenter{\hbox{\includegraphics[scale=0.275]{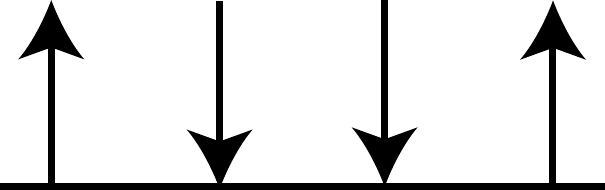}}} 
\quad + \quad \vcenter{\hbox{\includegraphics[scale=0.275]{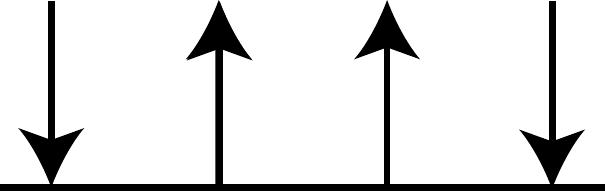}}} 
\quad - \quad q^{-1} \,\, \times \,\, 
\vcenter{\hbox{\includegraphics[scale=0.275]{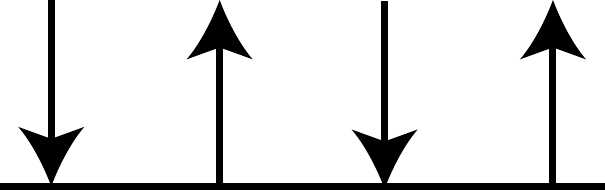} ,}} \\[1em] 
\label{VectorToLinkDiagramDefinitionEx2}
\vcenter{\hbox{\includegraphics[scale=0.275]{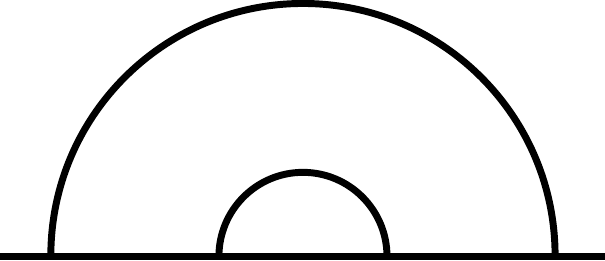}}} 
& \quad \underset{\eqref{singletLinkDiagram}}{\overset{\eqref{singletDiagramNotation}}{=}} \quad 
 - q \,\, \times \,\, \vcenter{\hbox{\includegraphics[scale=0.275]{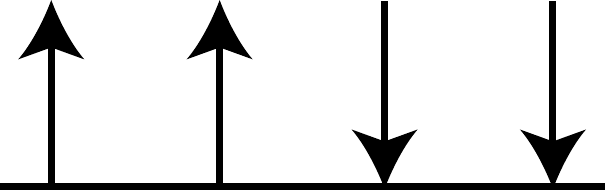}}} 
\quad + \quad \vcenter{\hbox{\includegraphics[scale=0.275]{Figures/e-vector_nested_singlets_split1.pdf}}} 
\quad + \quad \vcenter{\hbox{\includegraphics[scale=0.275]{Figures/e-vector_nested_singlets_split3.pdf}}} 
\quad - \quad q^{-1} \,\, \times \,\, 
\vcenter{\hbox{\includegraphics[scale=0.275]{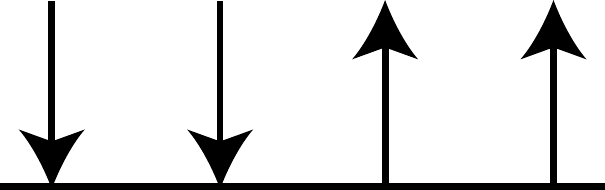} .}}  
\end{align}

Next, we verify that the above identifications~(\ref{singletLinkDiagram}--\ref{VectorToLinkDiagramDefinitionEx2})
still respect the bilinear pairing $\SPBiForm{\cdot}{\cdot}$.
For instance, 
\begin{align} \label{SingletBiFormDiagramWellDef}
(\, \singBar \BarAction \sing  \,) 
\; \underset{\eqref{singletLinkDiagramBar}}{\overset{\eqref{singletLinkDiagram}}{=}} \quad 
\left(\, \vcenter{\hbox{\includegraphics[scale=0.275]{Figures/e-Sing5_flipped.pdf}}} \BarAction \vcenter{\hbox{\includegraphics[scale=0.275]{Figures/e-Sing5.pdf}}}  \,\right)
\quad \overset{\eqref{LSBiformExamplesOriented}}{=} \quad
\left( \; \vcenter{\hbox{\includegraphics[scale=0.275]{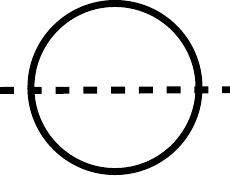}}} \; \right)
\quad \underset{\eqref{evT}}{\overset{\eqref{TLfugacity}}{=}} \quad 
\nu \; \overset{\eqref{fugacity}}{=} \; - q - q^{-1} 
\; \underset{\eqref{singletVector}}{\overset{\eqref{biformnormalization}}{=}} \;
\SPBiForm{\singBar}{\sing} .
\end{align}
We observe that the right side of~\eqref{SingletBiFormDiagramWellDef} is the bilinear pairing of the singlet vectors $\singBar$ and $\sing$,
while the left side of~\eqref{SingletBiFormDiagramWellDef} is the evaluation of an oriented network corresponding to these vectors. 
Alternatively, we have 
\begin{align} 
\nonumber
(\, \singBar \BarAction \sing  \,) 
\; \underset{\eqref{singletDiagramNotationBar}}{\overset{\eqref{singletDiagramNotation}}{=}} \; &
- q \, \left(\, \vcenter{\hbox{\includegraphics[scale=0.275]{Figures/e-Defects4_flipped.pdf}}} \BarAction \vcenter{\hbox{\includegraphics[scale=0.275]{Figures/e-Defects4.pdf}}} \,\right)
\; + \; 
\left(\, \vcenter{\hbox{\includegraphics[scale=0.275]{Figures/e-Defects4_flipped.pdf}}} \BarAction \vcenter{\hbox{\includegraphics[scale=0.275]{Figures/e-Defects5.pdf}}} \,\right)
\; + \;
\left(\, \vcenter{\hbox{\includegraphics[scale=0.275]{Figures/e-Defects5_flipped.pdf}}} \BarAction \vcenter{\hbox{\includegraphics[scale=0.275]{Figures/e-Defects4.pdf}}} \,\right)
\; - q^{-1} \, \left(\, \vcenter{\hbox{\includegraphics[scale=0.275]{Figures/e-Defects5_flipped.pdf}}} \BarAction \vcenter{\hbox{\includegraphics[scale=0.275]{Figures/e-Defects5.pdf}}} \,\right)  \\[1em]
\label{SingletBiFormDiagramWellDef2}
\;  \overset{\eqref{LSBiformExamplesOriented}}{=} \; &
- q \, \left( \; \vcenter{\hbox{\includegraphics[scale=0.275]{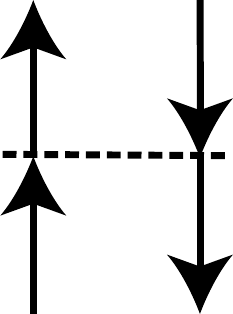}}} \; \right)
\; + \; \left( \; \vcenter{\hbox{\includegraphics[scale=0.275]{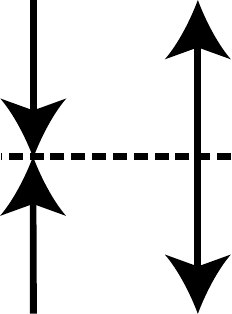}}} \; \right)
\; + \; \left( \; \vcenter{\hbox{\includegraphics[scale=0.275]{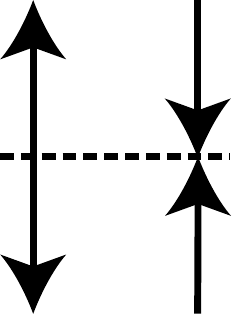}}} \; \right)
\; - q^{-1} \, \left( \; \vcenter{\hbox{\includegraphics[scale=0.275]{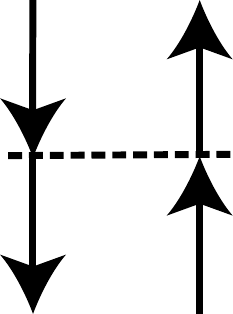}}} \; \right)
\; \underset{\textnormal{(\ref{ThroughPathWeight}, \ref{ThroughPathWeightZero})}}{\overset{\eqref{evT}}{=}} \; 
 - q - q^{-1}  \; \overset{\eqref{fugacity}}{=} \; \nu ,
\end{align}
which equals~\eqref{SingletBiFormDiagramWellDef}.
Hence, we see that both choices of representations of $\sing$ and $\singBar$ as link states with oriented defects yield the same network evaluation,
which coincides with the bilinear pairing $\SPBiForm{\singBar}{\sing}$.

\begin{lem}  \label{BiFormLem} 
For any two vectors $\overbarStraight{v} \in \VecSpBar_n$ and $w \in \VecSp_n$, 
the evaluation 
of an oriented network $\overbarStraight{v} \BarAction w$ 
equals
\begin{align} \label{WeightProd}
(\, \overbarStraight{v} \BarAction w \,)
= \SPBiForm{\overbarStraight{v}}{w} 
\end{align}
for any choice of representations of the vectors $\overbarStraight{v}$ and $w$ as link states with oriented defects.
\end{lem}

\begin{proof} 
We prove~\eqref{WeightProd} by induction on $n \geq 1$.
The initial case $n=1$ is governed by rules~(\ref{TrivialBLform1},~\ref{TrivialBLform2}). Hence, 
we assume that identity~\eqref{WeightProd} holds for all $n \in \{1,2,\ldots, m-1\}$ for some $m \geq 2$,
and we prove that it holds for $n=m$ too.
By linearity, it suffices to split the analysis of $(\, \overbarStraight{v} \BarAction w \,)$ into the following four scenarios: 
\begin{enumerate}[leftmargin=*]
\itemcolor{red}

\item \label{WellDefItem1}
The oriented network $\overbarStraight{v} \BarAction w$ has one of the following four forms, 
for some appropriate oriented sub-network $T$:
\begin{align} \label{GeneralSorm1}
\vcenter{\hbox{\includegraphics[scale=0.275]{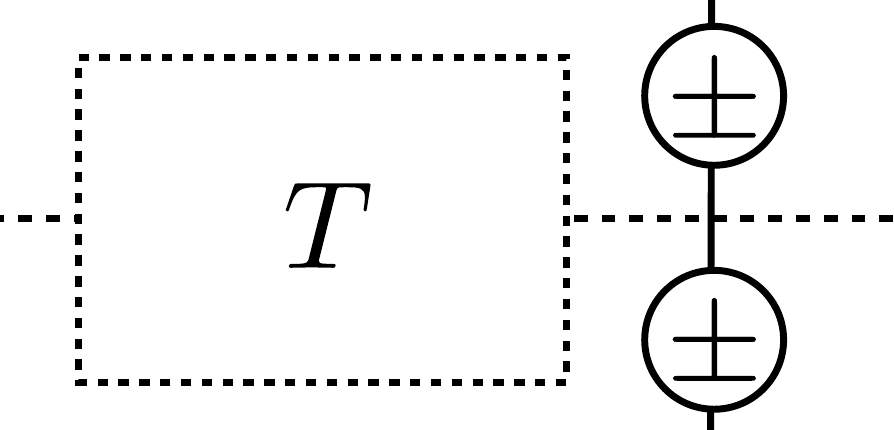},}} 
\qquad\qquad
\vcenter{\hbox{\includegraphics[scale=0.275]{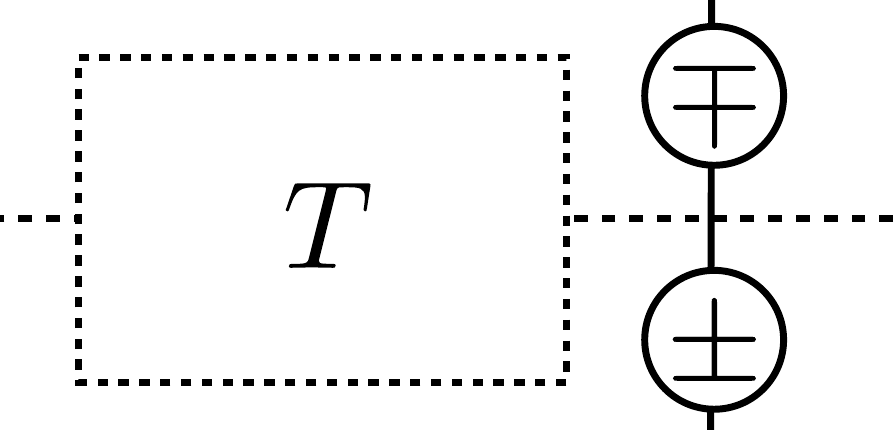}.}} 
\end{align}
As an example, we consider the leftmost case of~\eqref{GeneralSorm1} with $\pm \mapsto +$. Then, 
the vectors $\overbarStraight{v}$ and $w$ have the forms
\begin{align}
\label{RespForms1}
\overbarStraight{v} \overset{\eqref{OrientedDefects}}{=} 
\overbarStraight{v}' \otimes \FundBasisBar_0 \qquad \text{and} \qquad 
w \overset{\eqref{OrientedDefects}}{=} w' \otimes \FundBasis_0
\end{align}
for some other vectors $\overbarStraight{v}' \in \VecSpBar_{m-1}$ and $w' \in \VecSp_{m-1}$,
and we have $T = \overbarStraight{v}' \BarAction w'$.
The induction hypothesis gives
\begin{align}
\label{FactorForm1}
\SPBiForm{\overbarStraight{v}}{w} 
\overset{\eqref{RespForms1}}{=}
\SPBiForm{\overbarStraight{v}' \otimes \FundBasisBar_0}{w' \otimes \FundBasis_0} 
\underset{\eqref{biformfactorize}}{\overset{\eqref{biformnormalization}}{=}} \SPBiForm{\overbarStraight{v}'}{w'} \overset{\eqref{WeightProd}}{=}
(\, \overbarStraight{v}' \BarAction w' \, ) 
\overset{\eqref{GeneralSorm2}}{=} (\, T \, ).
\end{align}
On the other hand, we see from~\eqref{GeneralSorm1} that the evaluation of the oriented network $\overbarStraight{v} \BarAction w$
equals that of $T$ times the evaluation of the rightmost through-path. The latter has weight one 
by~\eqref{ThroughPathWeight}. Therefore,~\eqref{WeightProd} holds: 
\begin{align}
\SPBiForm{\overbarStraight{v}}{w} 
\overset{\eqref{FactorForm1}}{=} 
(\, T \, )
\underset{\textnormal{(\ref{ThroughPathWeight},~\ref{GeneralSorm1})}}{\overset{\eqref{evT}}{=}} 
( \, \overbarStraight{v} \BarAction w \, ) .
\end{align}
The other cases of~\eqref{GeneralSorm1} can be handled similarly. 

\item \label{WellDefItem2}
The oriented network $\overbarStraight{v} \BarAction w$ has one of the following 
eight forms, for some appropriate oriented sub-network $T$:
\begin{align} \label{GeneralSorm2}
\vcenter{\hbox{\includegraphics[scale=0.275]{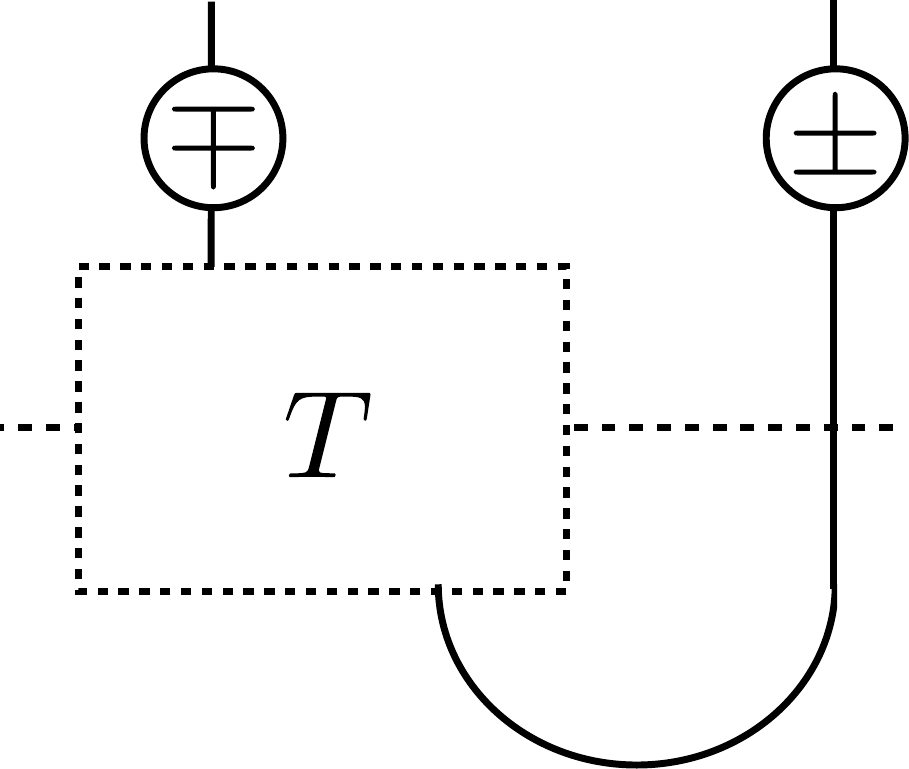} ,}} 
\qquad\qquad
\vcenter{\hbox{\includegraphics[scale=0.275]{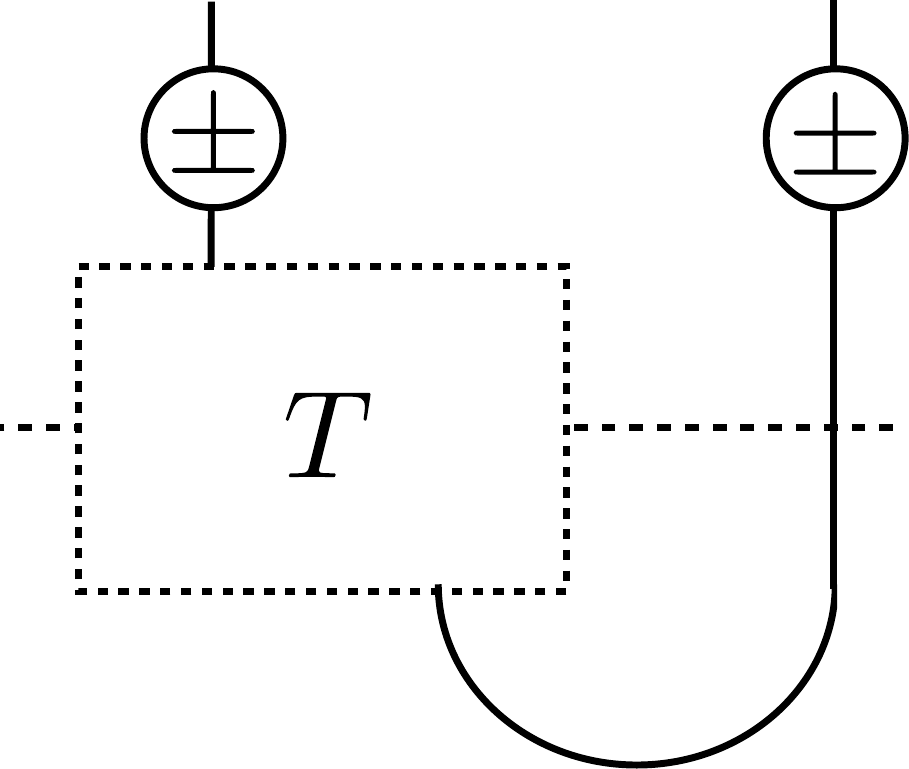} ,}}
\qquad \qquad
\vcenter{\hbox{\includegraphics[scale=0.275]{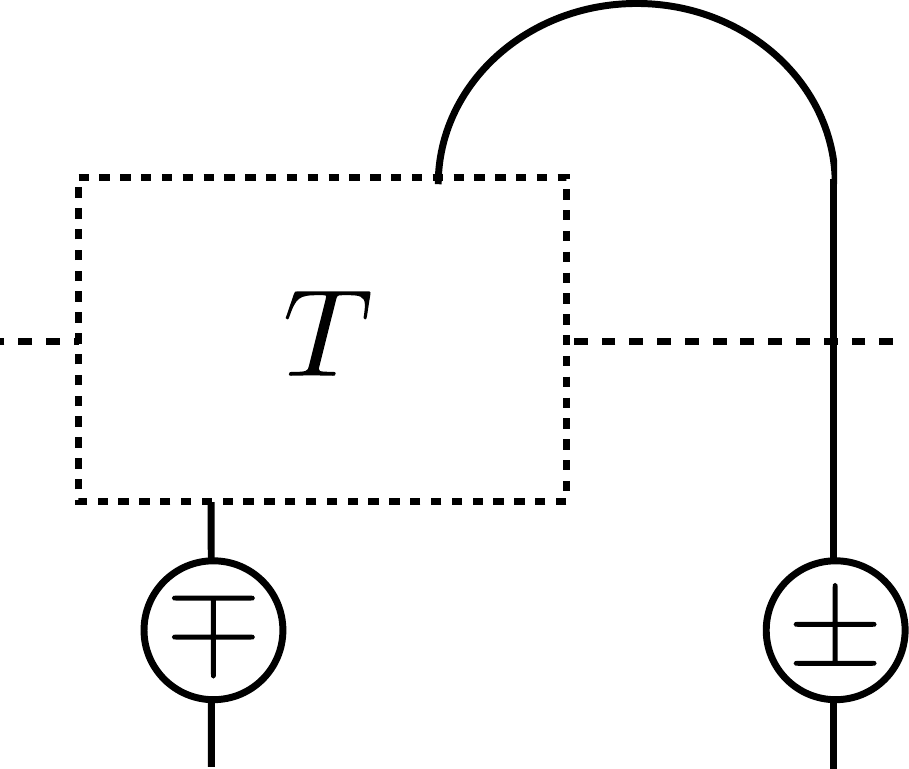} ,}} 
\qquad \qquad
\vcenter{\hbox{\includegraphics[scale=0.275]{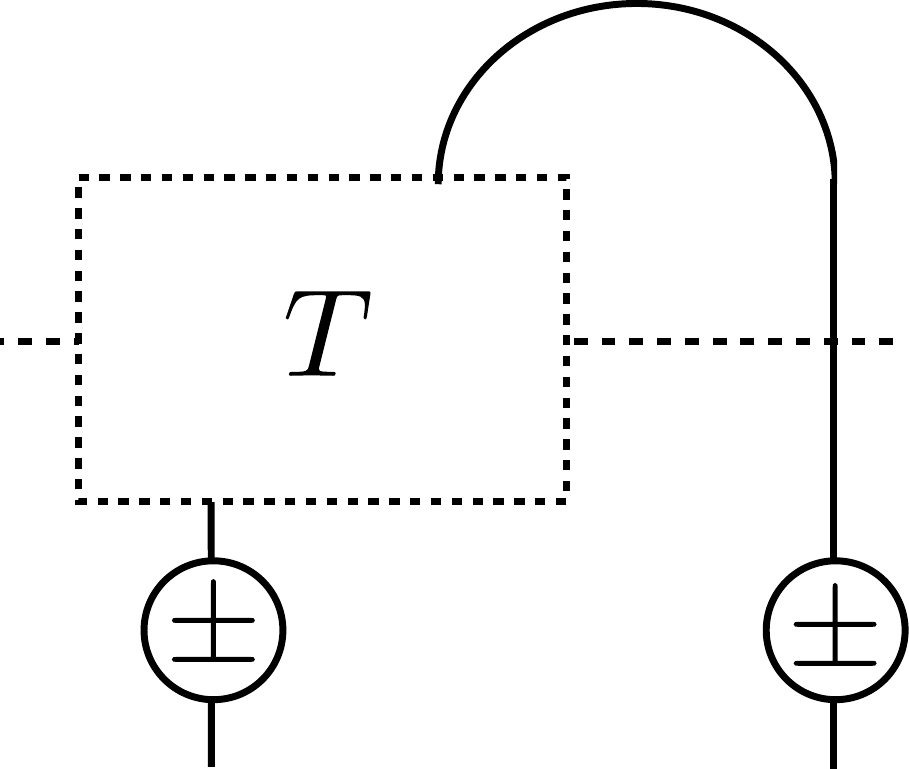} .}} 
\end{align}
Again, we restrict our attention to the leftmost case of~\eqref{GeneralSorm2} with 
$\mp \mapsto +$ and $\pm \mapsto -$; the other cases can be handled similarly.
Now, the network structure in~\eqref{GeneralSorm2} implies that the vectors $\overbarStraight{v}$ and $w$ have the forms
\begin{align}
\label{RespForms20}
\overbarStraight{v} \overset{\eqref{singletLinkDiagram}}{=} 
\ii q^{1/2} \, \overbarStraight{v}' \otimes \FundBasisBar_1
- \ii q^{-1/2} \, \overbarStraight{v}'' \otimes \FundBasisBar_0  
\qquad\qquad \text{and} \qquad \qquad
w  \overset{\eqref{OrientedDefects}}{=} w' \otimes \FundBasis_1 ,
\end{align}
for some other vectors
$\overbarStraight{v}' := \overbarStraight{v}_1 \otimes \FundBasisBar_0 \otimes \overbarStraight{v}_2 \in \VecSpBar_{m-1}$, 
$\overbarStraight{v}'' := \overbarStraight{v}_1 \otimes \FundBasisBar_1 \otimes \overbarStraight{v}_2 \in \VecSpBar_{m-1}$, 
and $w' \in \VecSp_{m-1}$. 
Thus, factorization property~\eqref{biformfactorize} from lemma~\ref{biformPropertyLem}
and the induction hypothesis give 
\begin{align}
\nonumber
\SPBiForm{\overbarStraight{v}}{w} 
&\underset{\hphantom{\eqref{biformfactorize}}}{\overset{\eqref{RespForms20}}{=}}  
\ii q^{1/2} 
\SPBiForm{\overbarStraight{v}' \otimes \FundBasisBar_1}{w' \otimes \FundBasis_0} 
- \ii q^{-1/2} 
\SPBiForm{\overbarStraight{v}'' \otimes \FundBasisBar_0}{w' \otimes \FundBasis_0} \\
\label{FactorForm2}
&\underset{\eqref{biformfactorize}}{\overset{\eqref{biformnormalization}}{=}} 
\ii q^{1/2}\SPBiForm{\overbarStraight{v}'}{w'} \overset{\eqref{WeightProd}}{=}
\ii q^{1/2}(\, \overbarStraight{v}' \BarAction w' \,) .
\end{align}
Now, the defect exiting $T$ through its bottom has upward orientation in order to connect to the turnback on the rightmost node with downward orientation.
Therefore, we see that $T = \overbarStraight{v}' \BarAction w'$, so we obtain
\begin{align} \label{FactorForm3}
\SPBiForm{\overbarStraight{v}}{w} 
\overset{\eqref{FactorForm2}}{=}
\ii q^{1/2}(\, T \, ) .
\end{align}
Finally, we observe that deforming the turn-back path of $\overbarStraight{v} \BarAction w$ shown in~\eqref{GeneralSorm2} into a through-path
gives the oriented network $T$, which completes the induction step for this case:
\begin{align}
\SPBiForm{\overbarStraight{v}}{w} 
\overset{\eqref{FactorForm3}}{=}
\ii q^{1/2} \left( \; \vcenter{\hbox{\includegraphics[scale=0.275]{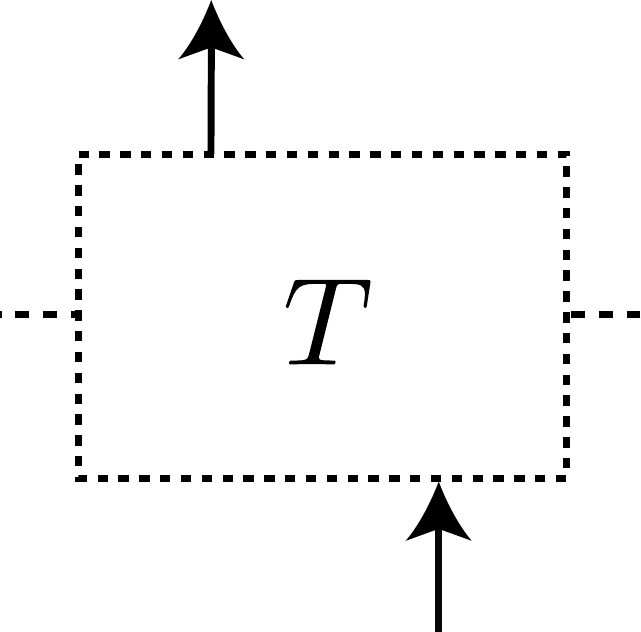}}}\; \right) 
\underset{\eqref{TurnBackWeightCW}}{\overset{\eqref{ThroughPathWeight}}{=}} 
\left( \; \vcenter{\hbox{\includegraphics[scale=0.275]{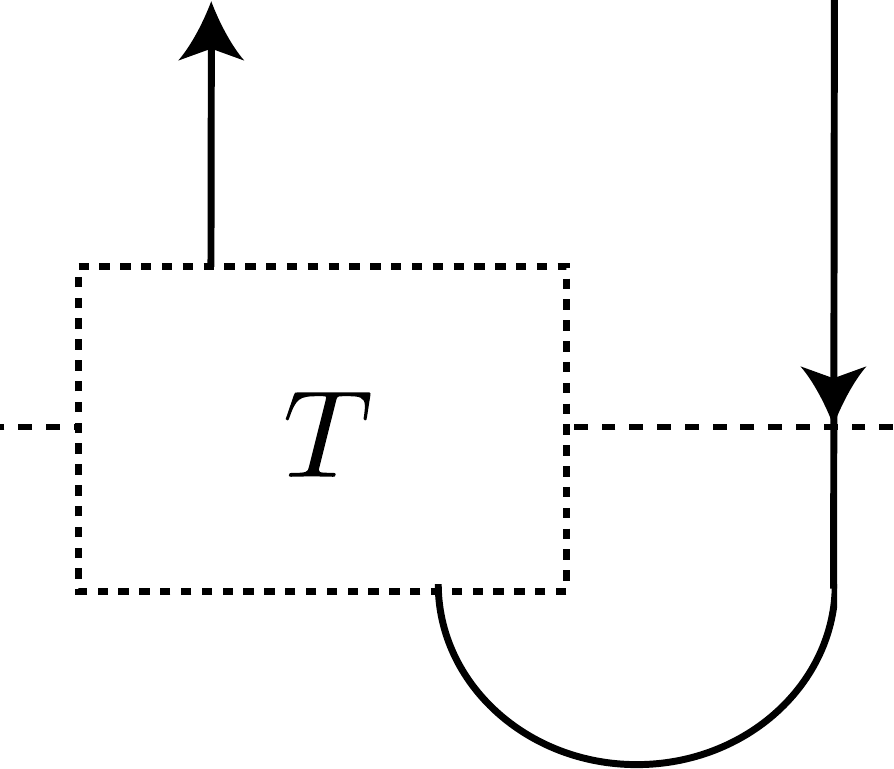}}}  \; \right)
\overset{\eqref{GeneralSorm2}}{=}
(\, \overbarStraight{v} \BarAction w \, ) .
\end{align}

\item \label{WellDefItem3}
The oriented network $\overbarStraight{v} \BarAction w$ has one of the following 
eight forms, for some appropriate oriented sub-network $T$:
\begin{align}
\vcenter{\hbox{\includegraphics[scale=0.275]{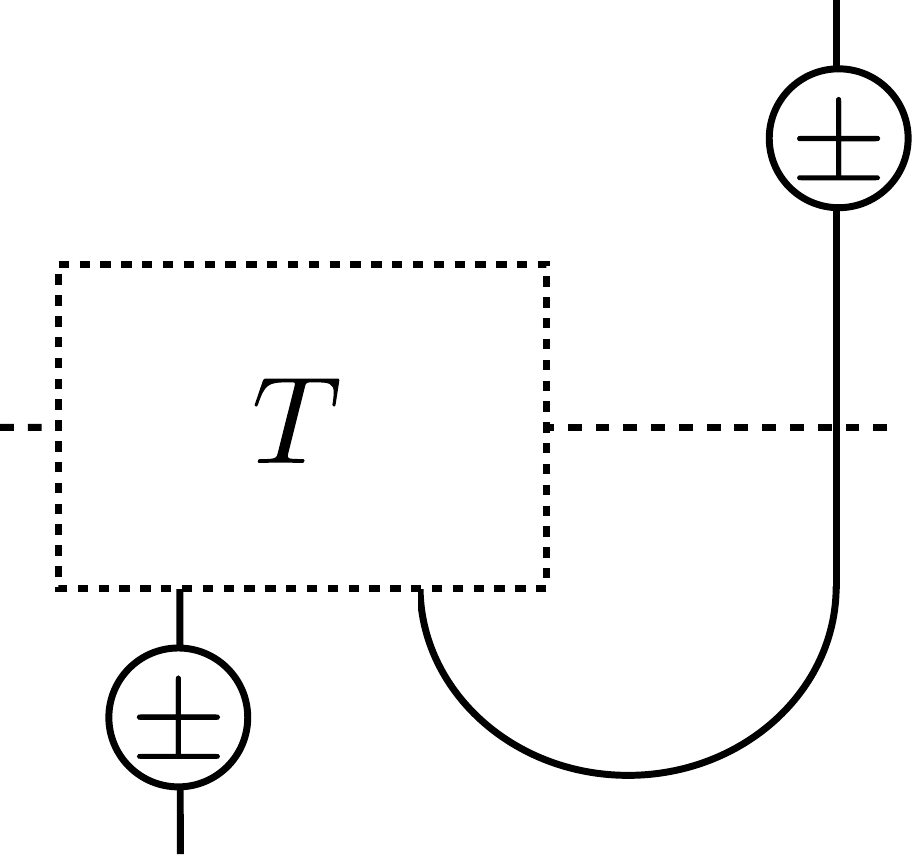} ,}} 
\qquad\qquad
\vcenter{\hbox{\includegraphics[scale=0.275]{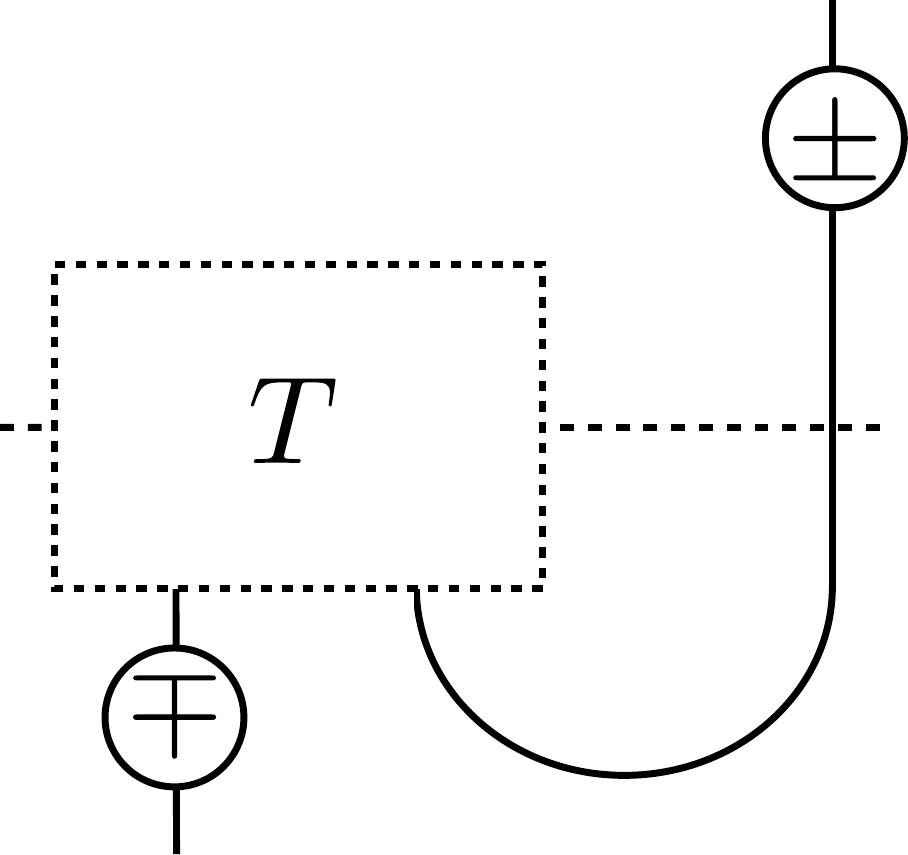} ,}}
\qquad\qquad
\vcenter{\hbox{\includegraphics[scale=0.275]{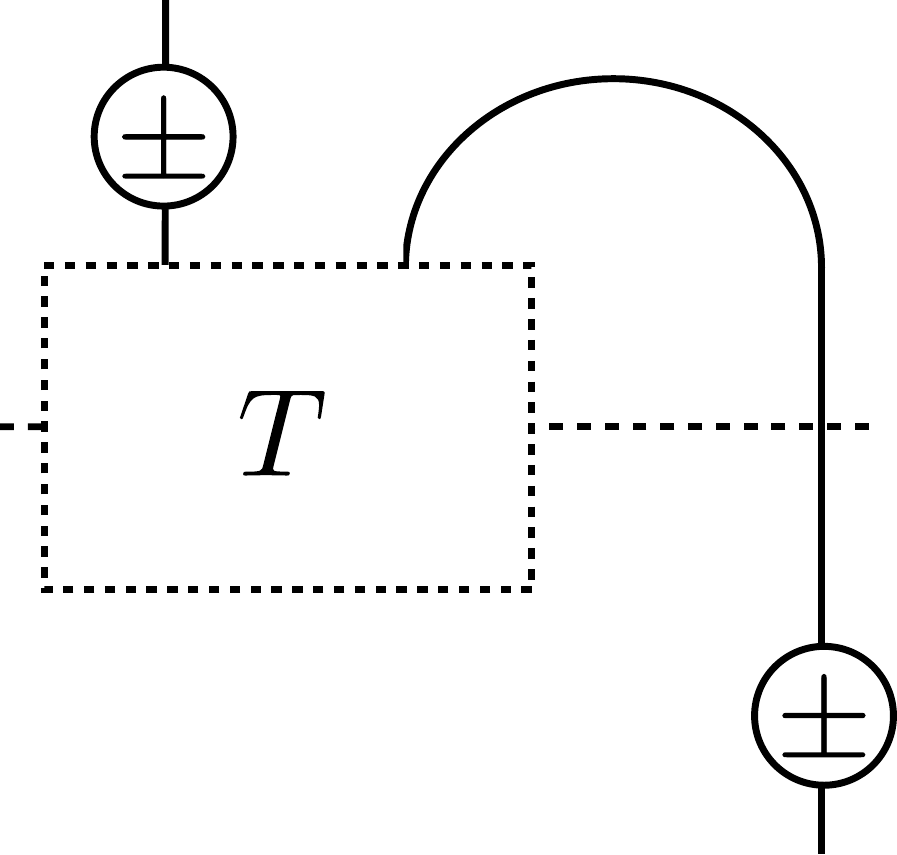} ,}} 
\qquad \qquad 
\vcenter{\hbox{\includegraphics[scale=0.275]{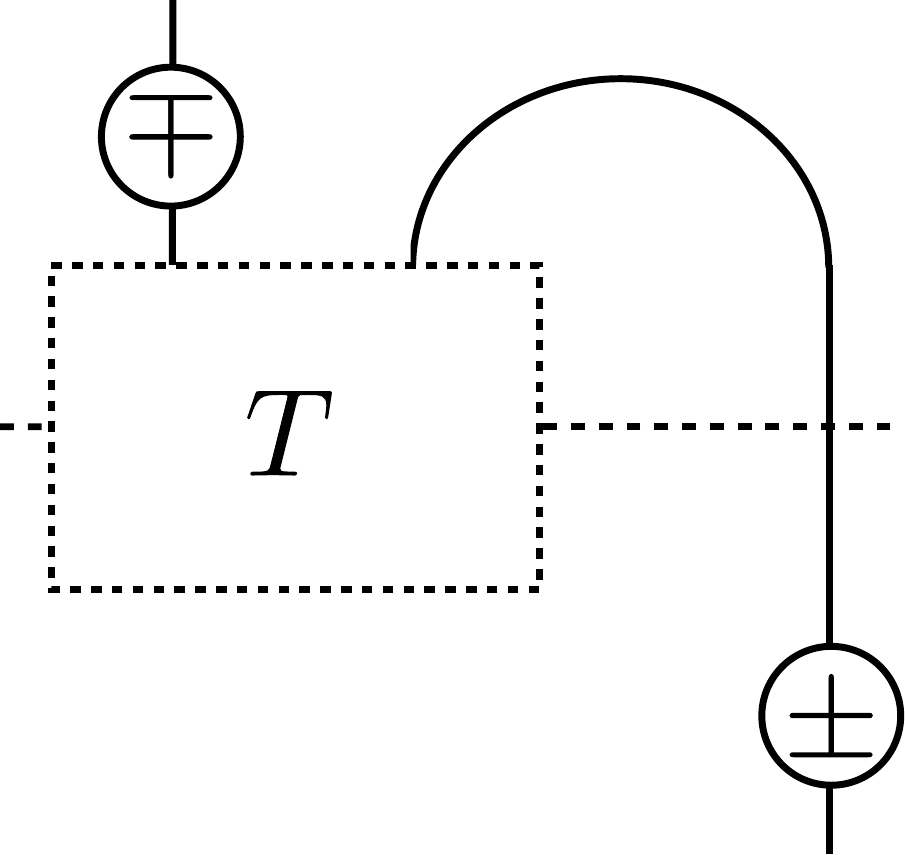} ,}} 
\end{align}
The proof of this case is very similar to that of item~\ref{WellDefItem3}. 
We leave it to the reader.

\item \label{WellDefItem4}
The oriented network $\overbarStraight{v} \BarAction w$ has the following form, for some appropriate 
sub-diagram $T$:
\begin{align} \label{GeneralSorm4}
\vcenter{\hbox{\includegraphics[scale=0.275]{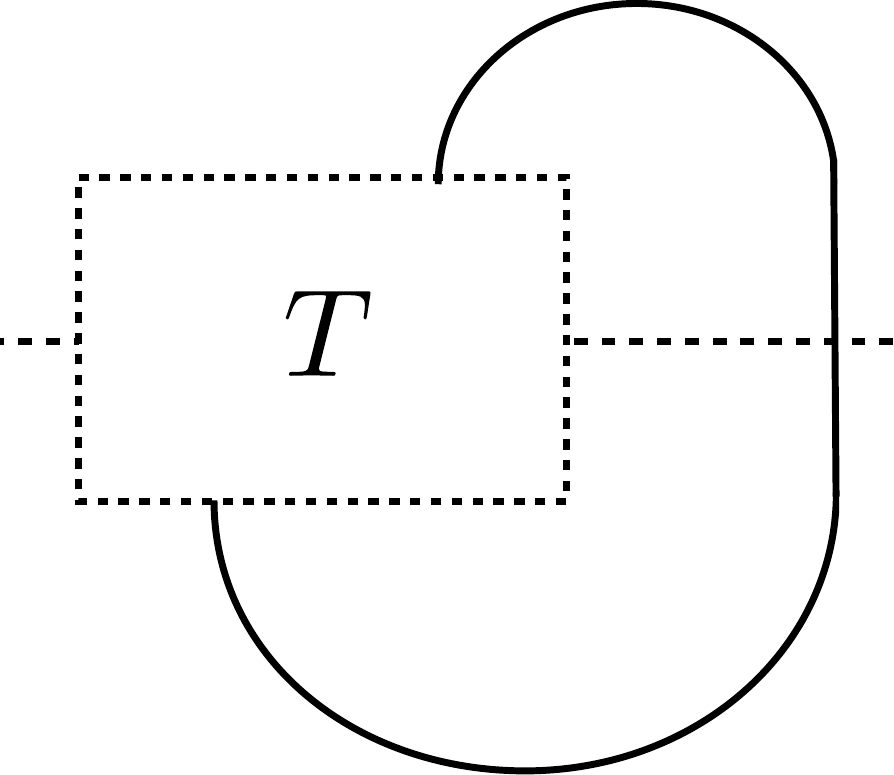} ,}} 
\end{align}
where $T$ contains a through-path joining the ends of the exterior path to form a loop. 
This network structure in~\eqref{GeneralSorm4} implies that the vectors $\overbarStraight{v}$ and $w$ have the forms
\begin{align} \label{RespForms30}
\overbarStraight{v} \overset{\eqref{singletLinkDiagram}}{=} 
\ii q^{1/2} \, \overbarStraight{v}' \otimes \FundBasisBar_1
- \ii q^{-1/2}\, \overbarStraight{v}'' \otimes \FundBasisBar_0 
\qquad\qquad \text{and} \qquad\qquad
w \overset{\eqref{singletLinkDiagram}}{=} 
\ii q^{1/2} \, w' \otimes \FundBasis_1 - \ii q^{-1/2} \, w'' \otimes \FundBasis_0 ,
\end{align}
for some other vectors 
$\overbarStraight{v}' := \overbarStraight{v}_1 \otimes \FundBasisBar_0 \otimes \overbarStraight{v}_2 \in \VecSpBar_{m-1}$, 
$\overbarStraight{v}'' := \overbarStraight{v}_1 \otimes \FundBasisBar_1 \otimes \overbarStraight{v}_2 \in \VecSpBar_{m-1}$, 
$w' := w_1 \otimes \FundBasis_0 \otimes w_2 \in \VecSp_{m-1}$, 
and 
$w'' := w_1 \otimes \FundBasis_1 \otimes w_2 \in \VecSp_{m-1}$. 
Thus, the induction hypothesis gives
\begin{align}
\nonumber
\SPBiForm{\overbarStraight{v}}{w} 
\overset{\eqref{RespForms30}}{=}  & \; 
- q \, \SPBiForm{\overbarStraight{v}' \otimes \FundBasisBar_1}{w' \otimes \FundBasis_1} 
+ \SPBiForm{\overbarStraight{v}'' \otimes \FundBasisBar_0}{w' \otimes \FundBasis_1} 
+ \SPBiForm{\overbarStraight{v}' \otimes \FundBasisBar_1}{w'' \otimes \FundBasis_0}
- q^{-1} \, \SPBiForm{\overbarStraight{v}'' \otimes \FundBasisBar_0}{w'' \otimes \FundBasis_0} \\
\label{FactorForm4}
\underset{\eqref{biformfactorize}}{\overset{\eqref{biformnormalization}}{=}}  & \; 
- q \SPBiForm{\overbarStraight{v}'}{w'} - q^{-1} \SPBiForm{\overbarStraight{v}''}{w''}
\overset{\eqref{WeightProd}}{=} 
- q (\, \overbarStraight{v}' \BarAction w' \,) - q^{-1} (\, \overbarStraight{v}'' \BarAction w''  \,) .
\end{align}
Now, the diagrams $\overbarStraight{v}' \BarAction w'$ and
$\overbarStraight{v}'' \BarAction w''$ only differ by the orientation of the through-path joining 
joining the ends of the exterior path in the diagram~\eqref{GeneralSorm4}:
\begin{align} \label{Vpp}
\overbarStraight{v}' \BarAction w'
\quad = \quad \vcenter{\hbox{\includegraphics[scale=0.275]{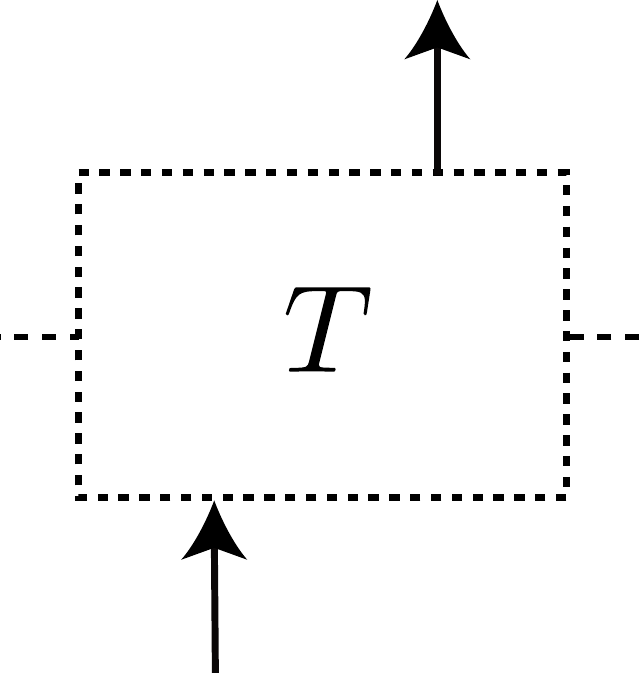}}} 
\qquad \qquad \text{and} \qquad \qquad
\overbarStraight{v}'' \BarAction w''
\quad = \quad \vcenter{\hbox{\includegraphics[scale=0.275]{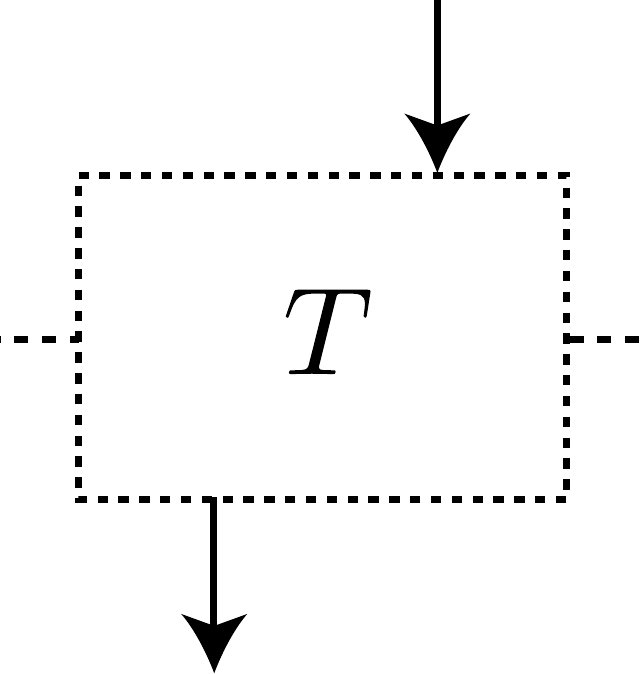} .}} 
\end{align}
Therefore, their evaluations $(\, \overbarStraight{v}' \BarAction w' \,)$ and 
$(\, \overbarStraight{v}'' \BarAction w''  \,)$ are equal by rules~(\ref{evT},~\ref{ThroughPathWeight}), and furthermore,
$\nu$ times either evaluation equals $(\, \overbarStraight{v} \BarAction w \, )$
according to rules~(\ref{TLfugacity},~\ref{TLThroughPathWeight}) and formula~\eqref{fugacity} for $\nu$.
Therefore, we have
\begin{align}
\SPBiForm{\overbarStraight{v}}{w} 
\overset{\eqref{FactorForm4}}{=}
- q (\, \overbarStraight{v}' \BarAction w' \,) - q^{-1} (\, \overbarStraight{v}'' \BarAction w''  \,) 
\overset{\eqref{fugacity}}{=} \nu (\, \overbarStraight{v}' \BarAction w' \,)
\underset{\eqref{Vpp}}{\overset{\eqref{GeneralSorm4}}{=}}
(\, \overbarStraight{v} \BarAction w \, )  .
\end{align}
This concludes the induction step.
\end{enumerate}
The proof is complete.
\end{proof}

It follows that the bilinear pairing $\SPBiForm{\cdot}{\cdot}$ is $\smash{\TL_\multii^\multiii}$-invariant in the following sense.

\begin{cor} \label{SwicthTCor}  
Suppose $\max \multii < \pmin(q)$.  For any valenced tangle $T \in \smash{\TL_\multii^\multiii}$ and vectors 
$\overbarStraight{v} \in \VecSpBar_\multii$ and $w \in \VecSp_\multiii$, we have 
\begin{align} \label{SwicthT} 
\SPBiForm{\overbarStraight{v}}{ T   w} = \SPBiForm{ \overbarStraight{v}  T}{w}. 
\end{align}
\end{cor}

\begin{proof} 
If $\multii = \OneVec{n}$ and $\multiii = \OneVec{m}$ for some $n,m \in \bZpos$,
then asserted identity~\eqref{SwicthT} follows 
from the definitions with lemma~\ref{BiFormLem}:
indeed, either expression~\eqref{SwicthT} equals the evaluation of an oriented network 
$\overbarStraight{v} \BarAction T w = \overbarStraight{v} T \BarAction w$;
e.g.,~for
\begin{align}
w \quad = & \quad \vcenter{\hbox{\includegraphics[scale=0.275]{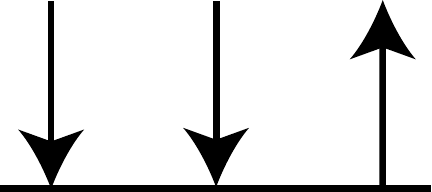} ,}} \qquad
\overbarStraight{v} \quad = \quad \vcenter{\hbox{\includegraphics[scale=0.275]{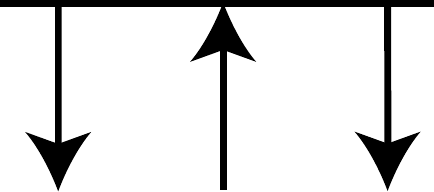} ,}} \qquad \text{and} \qquad
T \quad = \quad \vcenter{\hbox{\includegraphics[scale=0.275]{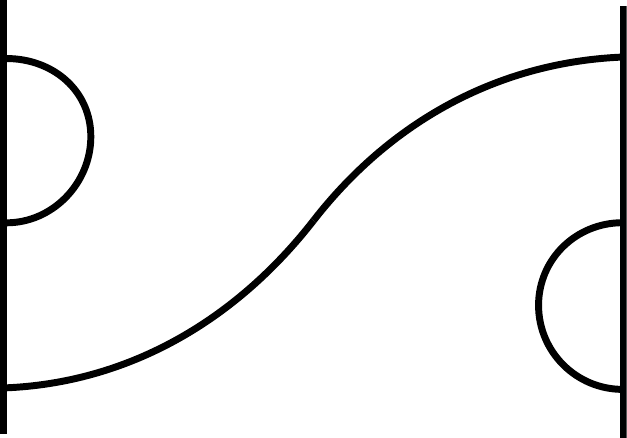} ,}} 
\end{align}
the following network (rotated by $-\pi/2$ radians)
represents either quantity~\eqref{SwicthT}:
\begin{align}
\vcenter{\hbox{\includegraphics[scale=0.275]{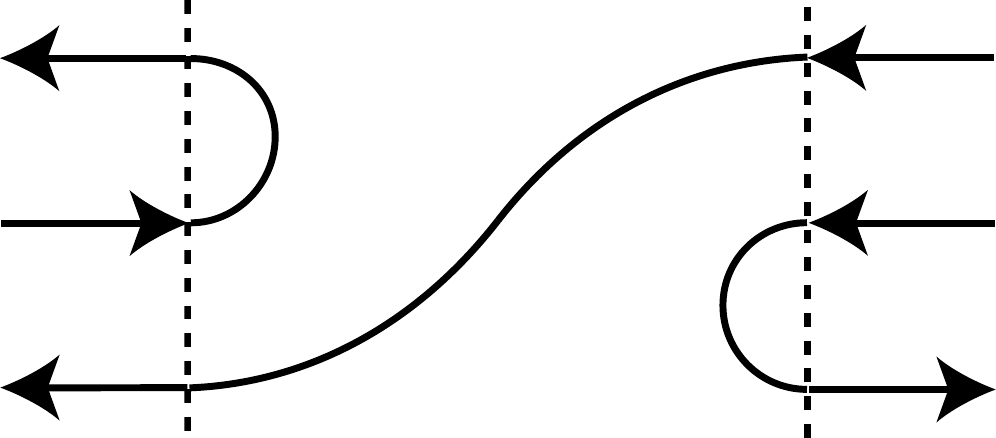} .}}
\end{align}
For the general case, combining 
the case of  $\multii = \OneVec{n}$ and $\multiii = \OneVec{m}$ with
item~\ref{biformitem4} of lemma~\ref{biformPropertyLem} and
corollary~\ref{CompositeProjCorHatEmb}, we have
\begin{align} 
\nonumber
\SPBiForm{\overbarStraight{v}}{T  w} 
\underset{\textnormal{(\ref{CompTwoProjsHatEmb}, \ref{CompTwoProjsHatEmbBar})}}{\overset{\eqref{SPBiFormNewEmbed}}{=}}
\SPBiForm{\overbarStraight{v}  \WJProjHat_\multii}{\WJEmb_\multii T   w}
&\overset{\eqref{IdCompAndWJPhatPEmb}}{=}
\SPBiForm{\overbarStraight{v}  \WJProjHat_\multii}{ (\WJEmb_\multii T \WJProjHat_\multiii) \WJEmb_\multiii   w} \\
&\overset{\eqref{SwicthT}}{=}
\SPBiForm{\overbarStraight{v}  \WJProjHat_\multii (\WJEmb_\multii T \WJProjHat_\multiii) }{  \WJEmb_\multiii   w} 
\overset{\eqref{IdCompAndWJPhatPEmb}}{=} \SPBiForm{\overbarStraight{v}  T \WJProjHat_\multiii }{  \WJEmb_\multiii   w}
\underset{\textnormal{(\ref{CompTwoProjsHatEmb}, \ref{CompTwoProjsHatEmbBar})}}{\overset{\eqref{SPBiFormNewEmbed}}{=}}
\SPBiForm{\overbarStraight{v}  T  }{ w},
\end{align}
which proves asserted identity~\eqref{SwicthT}.
\end{proof}

\begin{cor} \label{PmoveCor} 
Suppose 
$\max \multii < \pmin(q)$. 
For any vectors
$\overbarStraight{v} \in \VecSpBar_{\Summed_\multii}$ and $w \in \VecSp_{\Summed_\multii}$, we have
\begin{align} \label{Pmove} 
\SPBiForm{\ProjectionBar_\multii(\overbarStraight{v})}{w} = \SPBiForm{\overbarStraight{v}}{\Projection_\multii(w)} = \SPBiForm{\ProjectionBar_\multii(\overbarStraight{v})}{\Projection_\multii(w)}. 
\end{align}
\end{cor}

\begin{proof} The first equality in~\eqref{Pmove} follows by setting $T = \WJProj_\multii$ in~\eqref{SwicthT} and 
using corollary~\ref{CompositeProjCor}. The second equality in~\eqref{Pmove} follows by replacing $w$ with $\Projection_\multii(w)$ 
in the first equality and recalling the projection property $\Projection_\multii^2 = \Projection_\multii$.
\end{proof}

\section{The link state -- highest-weight vector correspondence} 
\label{GraphUQSect}
The purpose of this section is to 
explicate the connection between valenced link states $\alpha \in \LS_\multii$
and certain vectors in $\VecSp_\multii$. In particular, we show how valenced link patterns $\alpha \in \LP_\multii$
correspond to linearly independent highest-weight vectors (proposition~\ref{HWspLem2}). 
This fact gives rise to an embedding of $\TL_\multii(\nu)$-modules, 
\begin{align} 
\LS_\multii \overset{\eqref{LSDirSum}}{=}
\bigoplus_{s \, \in \, \DefectSet_\multii} \LS_\multii\super{s}
\quad \lhook\joinrel\rightarrow \quad 
\bigoplus_{s \, \in \, \DefectSet_\multii} \CModule{\HWsp_\multii\super{s}}{\TL}
\overset{\eqref{HWDirectSumPM}}{\subset}
\CModule{\HWsp_\multii}{\TL} .
\end{align}
In particular, if $\Summed_\multii < \pmin(q)$, then the $\TL_\multii(\nu)$-modules
$\CModule{\HWsp_\multii\super{s}}{\TL}$ discussed in section~\ref{GenDiacActTypeOneSec} are isomorphic to 
the standard 
modules $\smash{\LS_\multii\super{s}}$ discussed in section~\ref{TLReviewSec}.
Supplementing proposition~\ref{HWspaceDecTL},
we also have an isomorphism of $\TL_\multii(\nu)$-modules, 
\begin{align} 
\CModule{\VecSp_\multii}{\TL} \isom 
\bigoplus_{s \, \in \, \DefectSet_\multii} (s + 1) \, \LS_\multii\super{s} 
\qquad \qquad \textnormal{when} \qquad \Summed_\multii < \pmin(q) .
\end{align}
We thus understand this structure completely, thanks to Temperley-Lieb representation theory
(e.g.,~\cite{rsa, fp3a}).

Conversely, we recall from item~\ref{DirectSumInclusionItem3} of proposition~\ref{MoreGenDecompAndEmbProp} 
the direct-sum decomposition
\begin{align} 
\Module{\VecSp_\multii}{\Uqsltwo} \overset{\eqref{MoreGenDecomp}}{\isom}
\bigoplus_{s \, \in \, \DefectSet_\multii} \Dim_\multii\super{s} \Wd\sub{s} 
\qquad \qquad \textnormal{when} \qquad \Summed_\multii < \pmin(q) 
\end{align}
of the type-one module $\Module{\VecSp_\multii}{\Uqsltwo}$ 
into a direct sum of simple type-one submodules $\Wd\sub{s}$, 
whose multiplicities are explicitly given by solving the recursion problem~\eqref{Recursion2}, 
also appearing in lemma~\ref{LSDimLem2} and corollary~\ref{CobloBasisCor}, 
\begin{align} \label{HWdimension}
\Dim_\multii\super{s} 
\underset{\hphantom{\eqref{HWspDimension}}}{\overset{\eqref{LSDim2}}{=}} \; & \dim \LS_\multii\super{s} 
\, \qquad  \qquad \textnormal{for all} \qquad \multii \in \bZpos^\#  \\
\underset{\hphantom{\eqref{LSDim2}}}{\overset{\eqref{HWspDimension}}{=}}  \; &  \dim \HWsp_\multii\super{s} 
\qquad \qquad \textnormal{when} \qquad \Summed_\multii < \pmin(q) .
\end{align}
Each multiplicity space $\smash{\HWsp_\multii\super{s}}$~\eqref{npqcond2} 
comprises highest-weight vectors of weight $q^s$, and with $\Summed_\multii < \pmin(q)$,  these weights are all distinct
(cf.~(\ref{Kweightsdistinct},~\ref{npqcond})) 
and the conformal-block vectors $\smash{\HWvec^{\varrho}_{\multii}}$ 
spanning $\smash{\HWsp_\multii\super{s}}$ are all linearly independent. 
However, if $\Summed_\multii \geq \pmin(q)$, then proposition~\ref{MoreGenDecompAndEmbProp} 
only gives an embedding
with multiplicities $\smash{\hat{\Dim}_\multii\super{s}}$~\eqref{DimTilde} perhaps smaller than~$\smash{\Dim_\multii\super{s}}$, 
\begin{align} \label{DirectSumInclusionExampls}
\bigoplus_{\substack{s \, \in \, \DefectSet_\multii \\ s \, < \, \pmin(q) }} \hat{\Dim}_\multii\super{s} \Wd\sub{s} 
\quad \overset{\eqref{DirectSumInclusion}}{\lhook\joinrel\rightarrow} \quad \Module{\VecSp_\multii}{\Uqsltwo} .
\end{align}
In this case, some of the $K$-eigenvalues $q^s$ become equal, 
the type-one module $\Module{\VecSp_\multii}{\Uqsltwo}$ is usually not semisimple, 
and some of the conformal-block vectors $\smash{\HWvec^{\varrho}_{\multii}}$ become linearly dependent
(see appendix~\ref{ExceptionalQSect} for an example).
The embedding~\eqref{DirectSumInclusionExampls} can nevertheless be strengthened to have multiplicities $\smash{\Dim_\multii\super{s}}$,
see proposition~\ref{MoreGenDecompAndEmbProp2}.

The main aim of this section is to construct a 
set of highest-weight vectors in $\Module{\VecSp_\multii}{\Uqsltwo}$ well-defined 
and linearly independent whenever $\max \multii < \pmin(q)$. 
We call them ``(valenced) link-pattern basis vectors" $\Sing_\alpha$, indexed by valenced link patterns $\alpha$.
These vectors are particularly amenable to diagram calculations, namely,
the Temperley-Lieb action on them coincides with the 
diagram action of $\TL_\multii(\nu)$ on its standard modules.
This observation gives the ``link state -- highest-weight vector correspondence'' (proposition~\ref{HWspLem2}). 
Later, in proposition~\ref{QuotientProp} in sections~\ref{QuotientRadicalSect}--\ref{subsec: radical embedding} 
we use this correspondence to identify quotients of the 
highest-weight vector spaces 
with simple $\TL_\multii(\nu)$-modules.

\subsection{From link states to link-state vectors}
\label{subsec: link state hwv for n}

To begin, we consider the case of $\multii = \OneVec{n}$. 
We define the \emph{link-pattern basis vectors} $\Sing_\alpha$ for all $\alpha \in \LP_n$,
and prove that they are linearly independent highest-weight vectors in $\Module{\VecSp_n}{\Uqsltwo}$. 
These vectors and their $F$-descendants have a natural graphical representation, 
already appearing in section~\ref{LSandBiformandSCGrapgSec} and also discussed in section~\ref{DescGraphSec}. 
The link-pattern basis vectors are obtained by repeated insertion of (possibly nested) singlet vectors, corresponding to links, 
\begin{align} 
\label{singletDiagram} 
\vcenter{\hbox{\includegraphics[scale=0.275]{Figures/e-Sing5.pdf}}} 
\quad \overset{\eqref{singletLinkDiagram}}{=} \quad &
\sing 
\quad \overset{\eqref{singletVector}}{=} \quad 
\ii q^{1/2} \FundBasis_0 \otimes \FundBasis_1 - \ii q^{-1/2} \FundBasis_1 \otimes \FundBasis_0 
\quad \underset{\textnormal{$q \neq \pm 1$}}{\overset{\eqref{tau}}{=}} \quad \left(\frac{q-q^{-1}}{\ii q^{1/2}}\right) \HWvec\sub{1,1}\super{0} , \\
\label{singletDiagramBar} 
\vcenter{\hbox{\includegraphics[scale=0.275]{Figures/e-Sing5_flipped.pdf}}} 
\quad \overset{\eqref{singletLinkDiagramBar}}{=} \quad &
\singBar 
\quad \overset{\eqref{singletVector}}{=} \quad 
\ii q^{1/2} \FundBasisBar_0 \otimes \FundBasisBar_1 - \ii q^{-1/2} \FundBasisBar_1 \otimes \FundBasisBar_0  
\quad \underset{\textnormal{$q \neq \pm 1$}}{\overset{\eqref{taubar}}{=}} \quad \ii q^{1/2}(q-q^{-1}) \, \HWvecBar\sub{1,1}\super{0}  ,
\end{align} 
and standard basis vectors $\FundBasis_0, \FundBasisBar_0$, corresponding to upward-oriented defects,
\begin{align} 
\FundBasis_0 \quad \overset{\eqref{OrientedDefects}}{=} \quad \vcenter{\hbox{\includegraphics[scale=0.275]{Figures/e-Defects1.pdf} ,}}
\qquad \qquad 
\FundBasisBar_0 \quad \overset{\eqref{OrientedDefects}}{=} \quad \vcenter{\hbox{\includegraphics[scale=0.275]{Figures/e-Defects1_flipped.pdf} .}}
\end{align}
Precisely, the (nested) singlet vectors are inserted via (a repeated application of) the embedding
of the trivial $\Uqsltwo$-module $\Wd\sub{0}$ into the tensor product $\Uqsltwo$-module $\Wd\sub{1} \otimes \Wd\sub{1}$ 
implemented by the left generator tangle $\Lgen_1$:
\begin{align} \label{TildeEmbedor}
\Trep_{2}^{0}( \Lgen_1 ) \colon \Wd\sub{0} \longmapsto \Wd\sub{1} \otimes \Wd\sub{1} ,
\qquad \qquad
\Trep_{2}^{0}( \Lgen_1 ) \big( \Basis_0\super{0} \big) \overset{\eqref{ExtendThis1}}{=}  \sing
\underset{\textnormal{$q \neq \pm 1$}}{\overset{\eqref{tausingIDs}}{=}}
\left(\frac{q-q^{-1}}{\ii q^{1/2}}\right) \CCembedor\super{0}\sub{1,1} \big( \Basis_0\super{0} \big) .
\end{align}
(If $q=1$, the right side of~\eqref{TildeEmbedor} is not defined, 
but the map $\Trep_{2}^{0}( \Lgen_1 )$ is still well-defined: in this case,
it embeds the trivial $\mathfrak{sl}_2$-module into the tensor product of two of its fundamental modules, as discussed in appendix~\ref{ClassicalApp}.
In the present section, we mainly assume $q \neq \pm 1$, although most results apply verbatim to the case $q=1$.)

\begin{defn} \label{SingletBasisDefinition}
For any link pattern $\alpha \in \smash{\LP_n\super{s}}$, we recursively define the corresponding vector $\Sing_\alpha$ as follows:
we set $\Sing_\emptyset := 1 \in \bC$ for the empty link pattern $\emptyset \in \LP_0$, and we define the other vectors $\Sing_\alpha$ 
via the following recipe:
\begin{enumerate}
\itemcolor{red}
\item If $s = n$, then denoting by $\defects_n \in \smash{\LP_n\super{n}}$ the link pattern that consists of $n$ defects, 
\begin{align} \label{AllDefectsDiagram}
\defects_n \quad := \quad & \underbrace{\vcenter{\hbox{\includegraphics[scale=0.275]{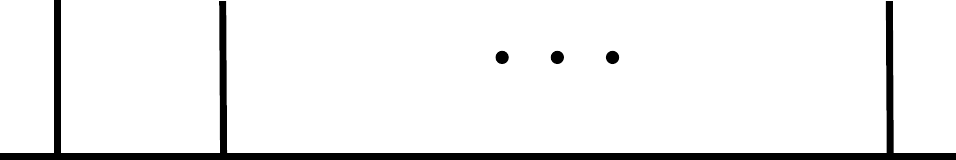}}}}_{\text{$n$ defects}} , \\ 
\label{SingletBasisDefAllDefects}
\textnormal{we set} \qquad\qquad
\Sing_{\scaleobj{0.85}{\defects_n}} \quad 
:= \quad & \underbrace{\FundBasis_0 \otimes \FundBasis_0 \otimes \dotsm \otimes \FundBasis_0}_{\text{$n$ tensorands}}
\overset{\eqref{MThwv}}{=} \MTbas_0\super{n} 
\quad \underset{\eqref{SimpleTensor}}{\overset{\eqref{OrientedDefects}}{=}} \quad
\underbrace{\vcenter{\hbox{\includegraphics[scale=0.275]{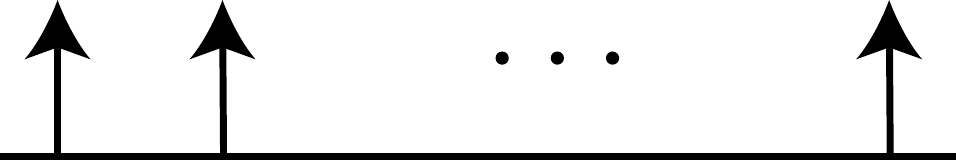}}}}_{\text{$n$ defects}} .
\end{align}

\item If $s < n$, then assuming that all of the vectors $\{ \Sing_\beta \, | \, \beta \in \smash{\LP_{m}\super{s}} \}$ with $1 \leq m \leq n-1$
have been defined, we define $\Sing_\alpha$ for any $\alpha \in \smash{\LP_n\super{s}}$ as follows. 
First, we choose a link joining two consecutive nodes
in $\alpha$, and we consider the vector $\Sing_{\hat{\alpha}}$ associated to the link pattern $\hat{\alpha} \in \smash{\LP_{n-2}\super{s}}$ obtained from
$\alpha$ by removing the chosen link. Suppose that the removed link joined the $k$:th and $(k+1)$:st nodes.
Then with the vector $\Sing_{\hat{\alpha}}$ already defined, we set
\begin{align} \label{SingletBasisDef}
\Sing_\alpha \quad := \quad
\Lgen_k \Sing_{\hat{\alpha}}
\quad \underset{\textnormal{$q \neq \pm 1$}}{\overset{\textnormal{(\ref{GenProj2-1}, \ref{GenProj2-1Exceptional})}}{=}} \quad
\left(\frac{q-q^{-1}}{\ii q^{1/2}}\right) 
\big(\id^{\otimes(k-1)} \otimes \CCembedor\super{0}\sub{1,1} \otimes \id^{\otimes(n-k-1)}\big) (\Sing_{\hat{\alpha}}) .
\end{align}
Induction on the number of links in $\alpha$ and the commutation relation for the left generators $\Lgen_k$ 
from~\eqref{MaxRelations} show that $\Sing_\alpha$ does not depend on the choice of the removed link.
\end{enumerate}
We also linearly extend this 
to a map $\alpha \mapsto \Sing_\alpha$ 
sending any link state $\alpha \in \smash{\LS_n\super{s}}$ to a corresponding vector $\Sing_\alpha \in \VecSp_n$.

We similarly define the link-pattern basis vectors $\SingBar_{\alphaBar}$ for all $\alphaBar \in \LPBar_n$,
\begin{align} 
\label{SingletBasisDefAllDefectsBar}
\Sing_{\scaleobj{0.85}{\defectsBar_n}} 
\quad & \overset{\hphantom{\eqref{singletVector}}}{:=} \quad
\underbrace{\FundBasisBar_0 \otimes \FundBasisBar_0 \otimes \dotsm \otimes \FundBasisBar_0}_{\text{$n$ tensorands}}
\quad \overset{\eqref{MThwvBar}}{=} \quad \MTbasBar_0\super{n} 
\quad \underset{\eqref{SimpleTensor}}{\overset{\eqref{OrientedDefects}}{=}} \quad
\overbrace{\raisebox{-10pt}{\hbox{\includegraphics[scale=0.275]{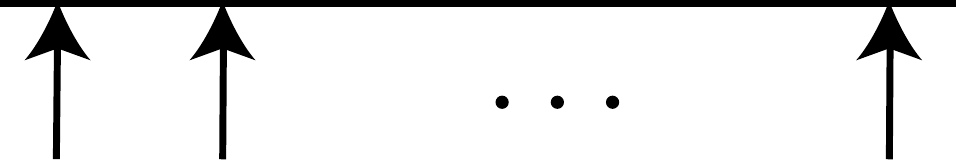}}}}^{\text{$n$ defects}} , \\ 
\label{SingletBasisDefBar}
\SingBar_{\alphaBar}
\quad & \overset{\hphantom{\eqref{singletVector}}}{:=} \quad
\SingBar_{\hat{\alphaBar}} \Rgen_k
\quad \underset{\textnormal{$q \neq \pm 1$}}{\overset{\textnormal{(\ref{GenProj2-1Bar}, \ref{GenProj2-1ExceptionalBar})}}{=}} \quad
\ii q^{1/2}(q-q^{-1}) \, 
\big(\id^{\otimes(k-1)} \otimes \CCembedorBar\super{0}\sub{1,1} \otimes \id^{\otimes(n-k-1)}\big) (\SingBar_{\hat{\alphaBar}}) .
\end{align}
\end{defn}

Recalling the graphical calculus from section~\ref{LSandBiformandSCGrapgSec}, 
we may identify each vector $\Sing_\alpha$ (resp.~$\SingBar_{\alphaBar}$) with the diagram 
of $\alpha$ (resp.~$\alphaBar$) whose defects have an upward orientation: 
for instance (see also~(\ref{VectorToLinkDiagramDefinitionEx1},~\ref{VectorToLinkDiagramDefinitionEx2})),
\begin{align} 
\label{SingAlphaDiagram}
\alpha \quad & = \quad \vcenter{\hbox{\includegraphics[scale=0.275]{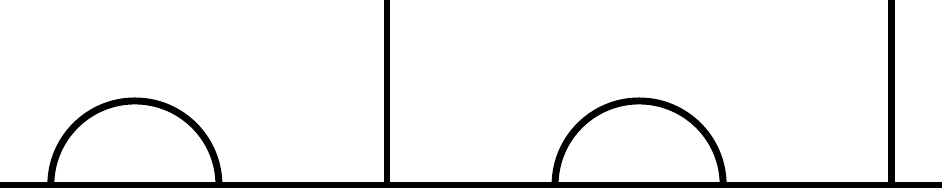}}} 
&&  \underset{\eqref{singletLinkDiagram}}{\overset{\eqref{OrientedDefects}}{\Longrightarrow}} \qquad \qquad
\Sing_\alpha \quad := \quad 
\sing \otimes \FundBasis_0 \otimes \sing \otimes \FundBasis_0 
\quad = \quad  \vcenter{\hbox{\includegraphics[scale=0.275]{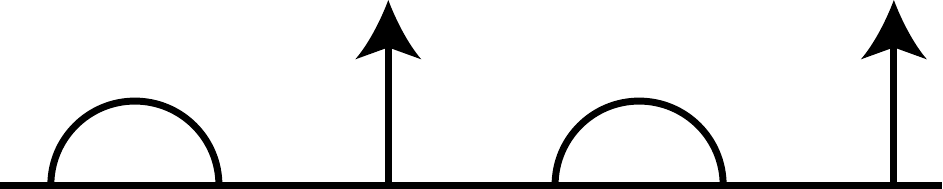} ,}}  \\[1em]
\label{SingAlphaDiagramBar}
\alphaBar \quad & = \quad \vcenter{\hbox{\includegraphics[scale=0.275]{Figures/e-Connectivities9_flipped.pdf}}} 
&& \underset{\eqref{singletLinkDiagram}}{\overset{\eqref{OrientedDefects}}{\Longrightarrow}}  \qquad \qquad
\SingBar_{\alphaBar} \quad := \quad 
\singBar \otimes \FundBasisBar_0 \otimes \singBar \otimes \FundBasisBar_0 
\quad = \quad  \vcenter{\hbox{\includegraphics[scale=0.275]{Figures/e-Connectivities9_defects_oriented_flipped.pdf} .}} 
\end{align}


Rules~(\ref{ExtendThis0}--\ref{ExtendThisBar2}) from section~\ref{DiacActTypeOneSec} amount 
to natural graphical rules for the $\smash{\TL_n^m}$-action on the vectors $\Sing_\alpha$ and $\SingBar_{\alphaBar}$,
given by the $\smash{\TL_n^m}$-action on the link states $\alpha$ and $\alphaBar$ themselves:

\begin{lem} \label{SmoothingLem} 
For all link states $\alpha \in \LS_m$ and $\alphaBar \in \LSBar_n$ 
and tangles $T \in \TL_n^m(\nu)$, we have
\begin{align} \label{LcommutationN} 
\Sing_{T \alpha} = T  \Sing_\alpha
\qquad \textnormal{and} \qquad 
\SingBar_{\alphaBar \, T} = \SingBar_{\alphaBar}  T .
\end{align}
\end{lem}

\emph{Remark}. This lemma is rather easy to verify using diagram representations for the link-state vectors 
via link states with upward-oriented defects (cf. section~\ref{LSandBiformandSCGrapgSec} and above)
and the usual tangle action on link states 
(cf. section~\ref{TLReviewSec}).
We invite the reader to check this. Here, we give a proof relying on the formal definition~\ref{SingletBasisDefinition}.

\begin{proof}
We prove the left equation of~\eqref{LcommutationN}; the right equation can be proven similarly.
Recalling from lemma~\ref{StdLem} that any tangle $T \in \smash{\TL_n^m(\nu)}$   
equals a polynomial in the generators $\Lgen_i$ and $\Rgen_j$, it suffices to show that
\begin{align} \label{ProveLRSingComm}
\Sing_{\Lgen_i \alpha} = \Lgen_i  \Sing_\alpha 
\qquad \qquad \text{and} \qquad \qquad 
\Sing_{\Rgen_j \alpha} = \Rgen_j  \Sing_\alpha 
\end{align}
for all link patterns $\alpha \in \LP_m$ and indices $i \in \{1,2,\ldots,m + 1\}$ and $j \in \{1,2,\ldots,m-1\}$, 
by linearity and item~\ref{HomoLem22} of lemma~\ref{HomoLem}.
The left equation of~\eqref{ProveLRSingComm} follows from definition~\ref{SingletBasisDefinition}
by choosing $k = i$ and $\alpha \mapsto \Lgen_i \alpha$, so $\smash{\widehat{\Lgen_i \alpha} = \alpha}$ and
\begin{align}
\Sing_{\Lgen_i \alpha} \overset{\eqref{SingletBasisDef}}{=}
\Lgen_i \Sing_{\widehat{\Lgen_i \alpha}} = \Lgen_i \Sing_\alpha .
\end{align}
We prove the right equation of~\eqref{ProveLRSingComm} by induction on $m \geq 2$. 
In the initial case with $m=2$, we necessarily have $k=j=1$, and we consider two cases as in definition~\ref{SingletBasisDefinition}:  
\begin{enumerate}[leftmargin=*]
\itemcolor{red}
\item If $\alpha = \defects_2 \in \smash{\LP_2\super{2}}$,  
then $\Rgen_1 \alpha = \Rgen_1 \defects_2 = 0$ by~(\ref{LgenForm},~\ref{TLTurnBack0},~\ref{AllDefectsDiagram}), 
so $\smash{\Sing_{\Rgen_1 \defects_2}} = 0 = \Rgen_1 \Sing_{\scaleobj{0.85}{\defects_2}}$ by~(\ref{ExtendThis2},~\ref{SingletBasisDefAllDefects}). 

\item If $\alpha = \smash{\vcenter{\hbox{\includegraphics[scale=0.185]{Figures/e-Sing5.pdf}}}} \in \smash{\LP_2\super{0}}$, 
then $\Rgen_1 \alpha = \Rgen_1 \smash{\vcenter{\hbox{\includegraphics[scale=0.185]{Figures/e-Sing5.pdf}}}} = \nu$ by~(\ref{LgenForm},~\ref{TLfugacity},~\ref{singletDiagram}), 
so $\smash{\Sing_{\Rgen_1 \vcenter{\hbox{\includegraphics[scale=0.125]{Figures/e-Sing5.pdf}}}}} 
= \nu = \Rgen_1 \smash{\Sing_{\vcenter{\hbox{\includegraphics[scale=0.125]{Figures/e-Sing5.pdf}}}}}$
by~(\ref{ExtendThis2},~\ref{singletDiagram}). 
\end{enumerate}
Next, we assume that $\smash{\Sing_{\Rgen_j \hat{\alpha}} = \Rgen_j  \Sing_{\hat{\alpha}}}$ for all
link patterns $\smash{\hat{\alpha}} \in \LP_n$ with $2 \leq n \leq m-1$, and we consider a link pattern $\alpha \in \LP_m$.
We again consider two cases as in definition~\ref{SingletBasisDefinition}:  
\begin{enumerate}[leftmargin=*]
\itemcolor{red}
\item If $\alpha = \defects_m \in \LP_m\super{m}$, then definition~\ref{SingletBasisDefinition} gives
\begin{align}
\Rgen_j \Sing_{\scaleobj{0.85}{\defects_m}} 
\overset{\eqref{SingletBasisDefAllDefects}}{=}
\Rgen_j \MTbas_0\super{n} 
\overset{\eqref{ExtendThis2}}{=}
0 
\overset{\eqref{TLTurnBack0}}{=}
\Sing_{\Rgen_j \defects_m} .
\end{align} 

\item If $\alpha \in \smash{\LP_m\super{s}}$ with $s < m$, then we choose an index $k \in \{1,2,\ldots,m-1\}$
such that a link joins the $k$:th and $(k+1)$:st nodes of $\alpha$, and let 
$\hat{\alpha} \in \smash{\LP_{m-1}\super{s}}$ denote the link pattern obtained from $\alpha$ by removing this link. 
Now, we have
\begin{align} \label{LRCommuted}
\Rgen_j \Sing_\alpha 
\overset{\eqref{SingletBasisDef}}{=}
\Rgen_j \Lgen_k \Sing_{\hat{\alpha}}
\overset{\eqref{MaxRelations}}{=}
\begin{cases} 
\Sing_{\hat{\alpha}} , & k = j \pm 1, \\ 
\nu \Sing_{\hat{\alpha}} , & k = j, \\ 
\Lgen_k \Rgen_{j-2} \Sing_{\hat{\alpha}} , &  k \leq j-2, \\ 
\Lgen_{k-2} \Rgen_j \Sing_{\hat{\alpha}} , &  j \leq k-2 ,
\end{cases} 
\end{align} 
recalling relations~\eqref{MaxRelations} of the left and right generator tangles.
Therefore, it remains to verify that the right side of~\eqref{LRCommuted} coincides with $\Sing_{\Rgen_j \alpha}$.
We consider these four cases:
\begin{enumerate}
\itemcolor{red}

\item[(a):] When $k = j \pm 1$, we have $\Rgen_j \alpha = \hat{\alpha}$, so indeed, $\Sing_{\Rgen_j \alpha} = \Sing_{\hat{\alpha}}$.

\item[(b):] When $k = j$, we have $\Rgen_j \alpha = \nu \hat{\alpha}$, so indeed, $\Sing_{\Rgen_j \alpha} = \nu \Sing_{\hat{\alpha}}$.

\item[(c):] When $k \leq j-2$, we have $\smash{\widehat{\Rgen_j \alpha}} = \Rgen_{j-2} \hat{\alpha}$ and the induction hypothesis applied to $\smash{\hat{\alpha}}$ shows that
\begin{align}
\Lgen_k \Rgen_{j-2} \Sing_{\hat{\alpha}} 
\overset{\eqref{ProveLRSingComm}}{=} 
\Lgen_k \Sing_{\Rgen_{j-2} \hat{\alpha}}
= \Lgen_k \Sing_{\widehat{\Rgen_j \alpha}}
\overset{\eqref{SingletBasisDef}}{=}
\Sing_{\Rgen_j \alpha} .
\end{align}

\item[(d):] When $j \leq k-2$, we have $\smash{\widehat{\Rgen_j \alpha}} = \Rgen_j \hat{\alpha}$ and the induction hypothesis applied to $\smash{\hat{\alpha}}$ shows that
\begin{align}
\Lgen_{k-2} \Rgen_j \Sing_{\hat{\alpha}}
\overset{\eqref{ProveLRSingComm}}{=} 
\Lgen_{k-2} \Sing_{\Rgen_j \hat{\alpha}}
= \Lgen_{k-2} \Sing_{\widehat{\Rgen_j \alpha}}
\overset{\eqref{SingletBasisDef}}{=}
\Sing_{\Rgen_j \alpha} .
\end{align}
\end{enumerate}
\end{enumerate}
This finishes the induction step for the right equation of~\eqref{ProveLRSingComm}, thus concluding the proof of the lemma. 
\end{proof}

We also obtain an explicit formula for the action of the Temperley-Lieb generators $\Gen_j$ on the link-pattern basis vectors.
To state this formula, we define for all $j \in \{ 1, 2, \ldots, n - 1 \}$
the $j$:th \emph{cutting map} $\smash{\tieOp_j \colon \LP_n\super{s} \to \LP_n\super{s}}$
by linear extension of the following action on link patterns $\alpha \in \smash{\LP_n\super{s}}$
(for a formal definition, see~\cite[section~\red{3.3}]{kp}): 
\begin{itemize}
\item if a link joins the $j$:th and $(j+1)$:st nodes of $\alpha$, then we set $\tieOp_j(\alpha) := \alpha$, 
\item otherwise, $\tieOp_j$ acts as illustrated below:
\begin{align} \label{CutMap} 
\vcenter{\hbox{\includegraphics[scale=0.275]{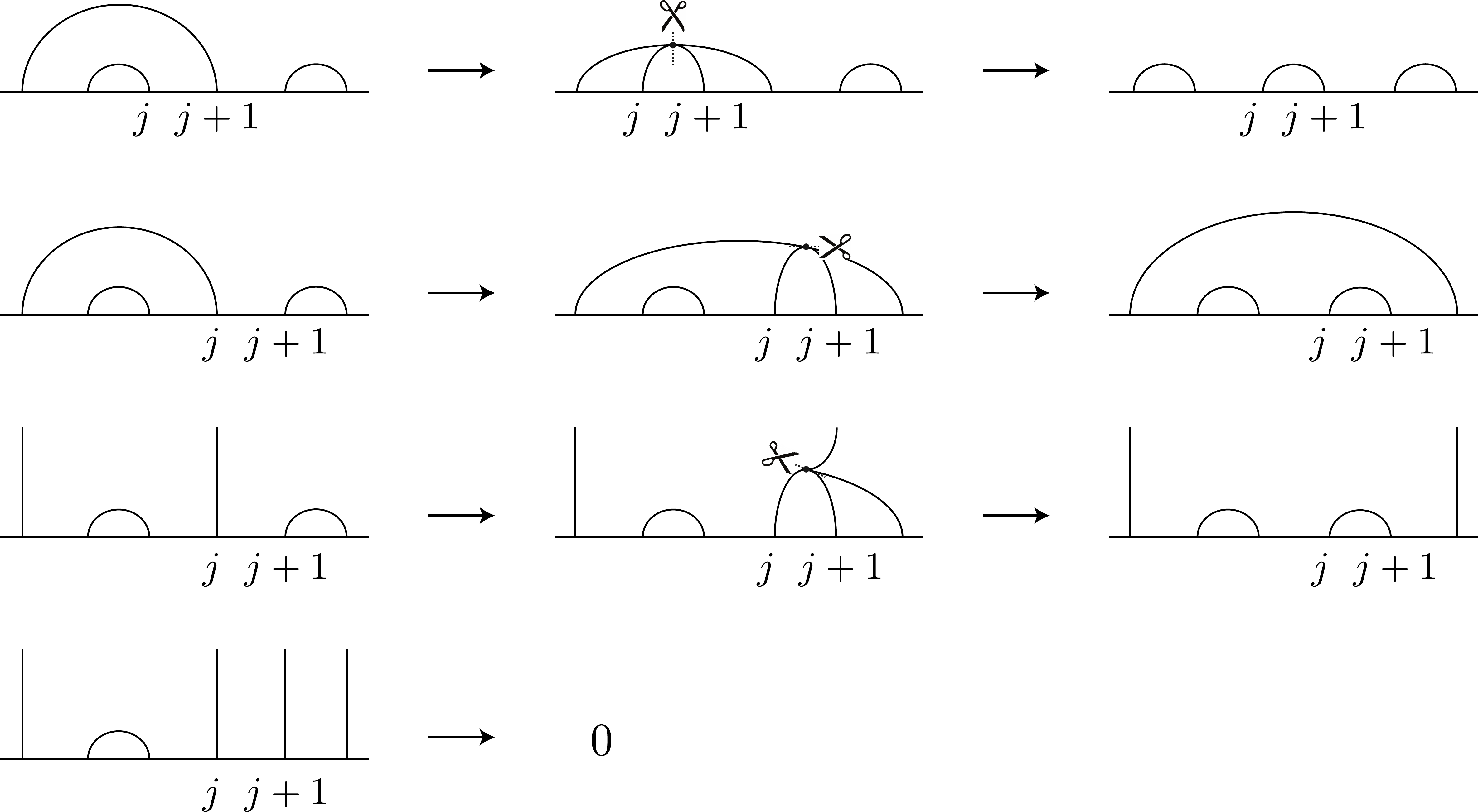}}} 
\end{align}
\end{itemize}

\begin{cor} \label{SingletBasisProjCor}
For all $\alpha \in \LP_n$ and for all $j \in \{1, 2, \ldots, n - 1\}$, we have
\begin{align} 
\label{SingletBasisUaction}
\hspace*{-5mm}
\Gen_j  \Sing_\alpha = 
\begin{cases} 
0 , & \textnormal{if defects touch both the $j$:th and $(j+1)$:st nodes in $\alpha$,} \\
\nu\Sing_\alpha , 
& \textnormal{if a link joins the $j$:th and $(j+1)$:st nodes of $\alpha$,} \\ 
 \Sing_{\tieOp_j(\alpha)} , 
& \textnormal{otherwise.} 
\end{cases} 
\end{align}
Similarly, this corollary holds for the right $\Gen_j$-action after the 
symbolic replacements $\alpha \mapsto \alphaBar$, $\LP \mapsto \LPBar$, and $\Sing_\alpha \mapsto \SingBar_{\alphaBar}$.
\end{cor}

\begin{proof}
This immediately follows from setting $T = U_j$ in lemma~\ref{SmoothingLem2}.
\end{proof}

Our next aim is to prove that the link-pattern basis vectors $\Sing_\alpha$ and $\SingBar_{\alphaBar}$
are highest-weight vectors (lemma~\ref{sGradingLem}) and
linearly independent (lemma~\ref{SingletBasisIsLinIndepLem}). 
We also collect other useful properties of these vectors
(lemma~\ref{SingletBasisExpandFormulaLem} and corollary~\ref{CorSingletBasisExpand}).

\begin{lem} \label{sGradingLem} 
Suppose $q \in \bC^\times \setminus \{\pm1\}$. We have 
\begin{align} \label{sGradingProperty} 
\big\{ \Sing_\alpha \, \big | \,\alpha \in \LS_n\super{s} \big\} \subset \HWsp_n\super{s} 
\qquad \qquad \textnormal{and} \qquad \qquad
\big\{ \SingBar_{\alphaBar} \, \big | \,\alphaBar \in \LSBar_n\super{s} \big\} \subset \HWspBar_n\super{s} .
\end{align}
\end{lem}

\begin{proof}
By linearity, we may assume that $\alpha$ is a link pattern.  We prove~\eqref{sGradingProperty} by induction on the number of links in $\alpha$. 
In the initial case $\alpha$ has no links, so $\smash{\alpha = \defects_s \in \LP_s\super{s}}$ for some $s \in \bZnn$,
and $\Sing_\alpha = \smash{\Sing_{\scaleobj{0.85}{\defects_s}}} = \smash{\MTbas_0\super{s}} \in \smash{\HWsp_s\super{s}}$ 
by~(\ref{MThwv},~\ref{SingletBasisDefAllDefects}).
Next, we let $\ell \in \bZpos$ and assume that~\eqref{sGradingProperty} holds for all link patterns with at most $\ell-1$ links,
and we let $\alpha \in \smash{\LP_n\super{s}}$ have $\ell$ links.
Now, if $\hat{\alpha}$ is a link pattern obtained from dropping a link from $\alpha$, 
then $\hat{\alpha}$ has $\ell-1$ links, so
\begin{align}
\alpha \in \LP_n\super{s} \qquad \Longrightarrow  \qquad 
\hat{\alpha} \in \LP_{n - 2}\super{s} 
\qquad \overset{\eqref{sGradingProperty}}{\Longrightarrow} \qquad 
\Sing_{\hat{\alpha}} \in \HWsp_{n - 2}\super{s} 
\qquad \underset{\eqref{SingletBasisDef}}{\overset{\eqref{NoShiftPropertyHWVN}}{\Longrightarrow}} \qquad 
\Sing_\alpha = \Lgen_k \Sing_{\hat{\alpha}}  \in \HWsp_n\super{s} ,
\end{align}
recalling from corollary~\ref{NoShiftPropertyCorHWVN} that the action of $\Lgen_k$ is a $\Uqsltwo$-homomorphism that respects the $s$-grading.
This proves the induction step and thus the first assertion of~\eqref{sGradingProperty}. The second assertion can be proven similarly.
\end{proof}

As a tool to show that $\Sing_\alpha$ and $\SingBar_{\alphaBar}$ indexed by link patterns are linearly independent, 
we consider walk representations for them 
and the standard basis vectors.
For each walk $\varrho = (r_1, r_2, \ldots, r_n)$ over $\OneVec{n}$ as in~\eqref{WalkHeights}, we set 
\begin{align} \label{FundBasisWalk}
\FundBasis_n^\varrho := \FundBasis_{\ell_1} \otimes \FundBasis_{\ell_2} \otimes \dotsm \otimes \FundBasis_{\ell_n} 
\qquad \qquad \text{and} \qquad \qquad
\FundBasisBar_n^\varrho := \FundBasisBar_{\ell_1} \otimes \FundBasisBar_{\ell_2} \otimes \dotsm \otimes \FundBasisBar_{\ell_n} 
\end{align} 
where $\ell_1, \ldots, \ell_n \in \{0,1\}$ are the unique indices such that the heights $r_j$ of the walk $\varrho$ for all $j \in \{1,2,\ldots,n\}$, are
\begin{align}
r_j = \sum_{i \, = \, 1}^j (1-2\ell_i) .
\end{align}
By lemma~\ref{BiFormDefLem}, these vectors are the useful orthogonality property
\begin{align}
\label{OrthoFormula}
\SPBiForm{\FundBasisBar_n^\varrho}{\FundBasis_n^{\varrho'}} \overset{\eqref{biformnormalization}}{=} \delta_{\varrho, \varrho'} .
\end{align}
The set of walks over $\OneVec{n}$ has a partial order defined for any two walks $\varrho = (r_1, r_2, \ldots, r_n)$ and $\varrho' = (r_1', r_2', \ldots, r_n')$ by
\begin{align} \label{PartialOrd} 
\varrho \DPleq \varrho' \qquad \text{if and only if} \qquad
r_i \leq r_i' \text{ for all $i \in \{0,1,\ldots,n\}$.}
\end{align}
We write $\varrho \DPle \varrho'$ if $\varrho \DPleq \varrho'$ and $\varrho \neq \varrho'$, and
we say that $\varrho$ and $\varrho'$ are incomparable if we have neither $\varrho \DPleq \varrho'$ nor $\varrho' \DPleq \varrho$.
For later use, we record the following obvious fact: 
\begin{align} \label{ord} 
\varrho \DPle \varrho' \text{ or $\varrho$ and $\varrho'$ are incomparable}
\qquad \qquad \Longrightarrow \qquad \qquad 
r_i < r_i' \quad \text{for some $i\in\{0,1,\ldots,n\}$}.
\end{align}
Next, we recall from~\cite[section~\red{4}]{fp3a} the notion of a \emph{walk representation} of a link pattern $\alpha \in \LP_n$ (or $\alphaBar \in \LPBar_n$): 
\begin{align} \label{alphaWalkNN} 
\varrho_\alpha := (r_1\super{\alpha}, r_2\super{\alpha}, \ldots, r_n\super{\alpha})
\end{align}
denotes the walk over $\OneVec{n}$ associated to $\alpha$, that is, the walk whose heights $\smash{r_j\super{\alpha}}$ 
are determined by the links and defects in $\alpha$ as follows: 
$\smash{r_1\super{\alpha}} = 1$, and for $j \in \{2,3,\ldots,n\}$,
\begin{align}
\label{InducedWalkDefn}
r_j\super{\alpha} = 
\begin{cases} 
r_{j-1}\super{\alpha} + 1 , 
& \quad \text{the $j$:th node in $\alpha$ has a defect or left endpoint of a link,} \\ 
r_{j-1}\super{\alpha} - 1 , 
& \quad \text{the $j$:th node in $\alpha$ has a right endpoint of a link.}
\end{cases} 
\end{align}

Item~\ref{SingletBasisExpandDefRelationItem} of the next result formalizes the representation of the singlet vectors $\sing$ and $\singBar$ 
as links~(\ref{singletDiagram},~\ref{singletDiagramBar})
and the nested singlet structure of $\Sing_\alpha$ and $\SingBar_{\alphaBar}$ in definition~\ref{SingletBasisDefinition} and~(\ref{VectorToLinkDiagramDefinitionEx1},~\ref{VectorToLinkDiagramDefinitionEx2},~\ref{SingAlphaDiagram},~\ref{SingAlphaDiagramBar}).
Item~\ref{SingletBasisExpandImpliesItem} is a comparison tool to obtain an upper-triangular structure
for the matrix consisting of pairings of $\Sing_\alpha$ and $\SingBar_{\alphaBar}$ 
with $\smash{\FundBasis_n^\varrho}$ and $\smash{\FundBasisBar_n^\varrho}$ in corollary~\ref{CorSingletBasisExpand}.

\begin{lem}  \label{SingletBasisExpandFormulaLem}
For each $\alpha \in \LP_n$, write the decomposition of $\Sing_\alpha$ over the standard basis as
\begin{align} \label{SingletBasisExpandDef}
\Sing_\alpha = \sum_{ \ell_1, \ell_2, \ldots, \ell_n \, \in \, \{0,1\}  } C_{\ell_1, \ell_2, \ldots, \ell_n}\super{\alpha} 
\FundBasis_{\ell_1} \otimes \FundBasis_{\ell_2} \otimes \dotsm \otimes \FundBasis_{\ell_n} 
\end{align}
for some coefficients $\smash{C_{\ell_1, \ell_2, \ldots, \ell_n}\super{\alpha}} \in \bC$.
Then, the following hold:
\begin{enumerate} 
\itemcolor{red}
\item \label{SingletBasisExpandDefRelationItem}
Suppose a link joins the $k$:th and $(k+1)$:st nodes in $\alpha$, and let $\hat{\alpha} \in \LP_{n-2}$ be the link pattern 
obtained from $\alpha$ by removing this link. Then, the coefficients of $\Sing_\alpha$ in~\eqref{SingletBasisExpandDef} are uniquely determined from those of $\Sing_{\hat{\alpha}}$ via
\begin{align} \label{SingletBasisExpandDefRelation}
\begin{cases}
C_{\ell_1, \ell_2, \ldots, \ell_{k-1}, 0, 1, \ell_{k+2}, \ldots, \ell_n}\super{\alpha}  
& = \hphantom{-} \ii q^{1/2} \hphantom{{}^-} \, C_{\ell_1, \ell_2, \ldots, \ell_{k-1}, \ell_{k+2}, \ldots, \ell_n}\super{\hat{\alpha}} , \\
C_{\ell_1, \ell_2, \ldots, \ell_{k-1}, 1, 0, \ell_{k+2}, \ldots, \ell_n}\super{\alpha}  
& = - \ii q^{-1/2} \, C_{\ell_1, \ell_2, \ldots, \ell_{k-1}, \ell_{k+2}, \ldots, \ell_n}\super{\hat{\alpha}} , \\
C_{\ell_1, \ell_2, \ldots, \ell_{k-1}, 0, 0, \ell_{k+2}, \ldots, \ell_n}\super{\alpha} & = \hphantom{-} C_{\ell_1, \ell_2, \ldots, \ell_{k-1}, 1, 1, \ell_{k+2}, \ldots, \ell_n}\super{\alpha} = 0 .
\end{cases}
\end{align}
\item \label{SingletBasisExpandImpliesItem}
We have
\begin{align} \label{SingletBasisExpandImplies}
C_{\ell_1, \ell_2, \ldots, \ell_n}\super{\alpha} \neq 0
\qquad \quad \Longrightarrow \qquad \quad
\sum_{i \, = \,1}^j (1-2\ell_i) \leq r_j\super{\alpha}  \quad \textnormal{for all $j \in \{1,2,\ldots,n\}$.}
\end{align}
\end{enumerate}
Similarly, this lemma holds after the symbolic replacements 
$\alpha \mapsto \alphaBar$, $\LP \mapsto \LPBar$, $\Sing_\alpha \mapsto \SingBar_{\alphaBar}$, 
and $\FundBasis \mapsto \FundBasisBar$.
\end{lem}

\begin{proof}
We prove items~\ref{SingletBasisExpandDefRelationItem} and~\ref{SingletBasisExpandImpliesItem} as follows:
\begin{enumerate}[leftmargin=*]
\itemcolor{red}

\item  We can always expand $\Sing_\alpha$ in the form~\eqref{SingletBasisExpandDef}, so the task is
check that the coefficients have the asserted form: 
\begin{align} 
\nonumber
\Sing_\alpha
\overset{\eqref{SingletBasisDef}}{=} \; & \Lgen_k \Sing_{\hat{\alpha}}
\underset{\eqref{SingletBasisExpandDef}}{\overset{\eqref{ExtendThis1}}{=}} 
\sum_{\substack{\ell_1, \ldots, \ell_{k-1}, \\ \ell_{k+2}, \ldots, \ell_n \, \in \, \{0,1\}}} 
C_{\ell_1, \ldots, \ell_{k-1}, \ell_{k+2}, \ldots, \ell_n}\super{\hat{\alpha}}  \;
\FundBasis_{\ell_1} \otimes \cdots \otimes \FundBasis_{\ell_{k-1}} \otimes \sing \otimes \FundBasis_{\ell_{k+2}} \otimes \dotsm \otimes \FundBasis_{\ell_n} .
\end{align}
By~(\ref{singletVector},~\ref{SingletBasisExpandDef}),
this proves asserted formulas~\eqref{SingletBasisExpandDefRelation} for the coefficients in~\eqref{SingletBasisExpandDef}. 

\item 
We prove~\eqref{SingletBasisExpandImplies} by induction on $n \in \bZpos$. The initial case with $n=1$ is trivial.
Assuming~\eqref{SingletBasisExpandImplies} holds for all link patterns $\beta \in \LP_{n-2}$, we prove that it holds for all link patterns $\alpha \in \LP_n$. If $\alpha$ has no links, then we trivially have
\begin{align}
\alpha = \defects_n \in \smash{\LP_n\super{n}} \qquad 
&\overset{\eqref{SingletBasisDefAllDefects}}{\Longrightarrow} \qquad 
\Sing_\alpha = \Sing_{\scaleobj{0.85}{\defects_n}} = \FundBasis_{0} \otimes \FundBasis_{0} \otimes \dotsm \otimes \FundBasis_{0}, \\
&\overset{\eqref{FundBasisWalk}}{\Longrightarrow} \qquad
\sum_{i \, = \, 1}^j (1-2\ell_i) = j = r_j\super{\alpha}  \quad \text{for all $j \in \{1,2,\ldots,n\}.$}
\end{align}
Therefore, we assume that $\alpha \in \smash{\LP_n\super{s}}$ with $s < n$, so $\alpha$ contains a link joining 
the $k$:th and $(k+1)$:st nodes for some $k \in \{1,2,\ldots,n - 1\}$. 
We let $\hat{\alpha}$ be the link pattern obtained by removing this link from $\alpha$. 
Then, item~\ref{SingletBasisExpandDefRelationItem} gives
\begin{align}
C_{\ell_1, \ldots, \ell_n}\super{\alpha} \neq 0 \qquad \qquad \overset{\eqref{SingletBasisExpandDefRelation}}{\Longrightarrow} \qquad \qquad
C_{\ell_1, \ldots, \ell_{k-1}, \ell_{k+2}, \ldots, \ell_n}\super{\hat{\alpha}} \neq 0 \qquad \text{and} \qquad (\ell_k,\ell_{k+1}) \in \{(0,1),(1,0)\} .
\end{align}
With $\hat{\alpha} \in \smash{\LP_{n-2}\super{s}}$, the induction hypothesis shows that 
\begin{align} \label{ImplOfSingletBasisExpand}
\sum_{\substack{1 \, \leq \, i \, \leq \, j \\ i \, \neq \, k,k+1}} (1-2\ell_i) \overset{\eqref{SingletBasisExpandImplies}}{\leq} 
\begin{cases} 
r_j\super{\hat{\alpha}}, 
& j \in \{1,2,\ldots,k-1\}, \\ 
r_{j-2}\super{\hat{\alpha}}, 
& j \in \{k+2,k+3,\ldots,n\}. 
\end{cases} 
\end{align}
Therefore, we have
\begin{align} \label{UpperBounds}
\sum_{i \, = \, 1}^j (1-2\ell_i) \overset{\eqref{ImplOfSingletBasisExpand}}{\leq} 
\begin{cases} 
r_j\super{\hat{\alpha}}, & j \in \{1,2,\ldots,k-1\}, \\ 
r_{k-1}\super{\hat{\alpha}} + (1-2\ell_k),  & j = k, \\ 
r_{k-1}\super{\hat{\alpha}} + (1-2\ell_k) + (1-2\ell_{k+1}),  & j = k+1, \\ 
r_{j-2}\super{\hat{\alpha}} + (1-2\ell_k) + (1-2\ell_{k+1}),  & j \in \{k+2,k+3,\ldots,n\}. 
\end{cases} 
\end{align}
On the other hand, for the walks $\varrho_\alpha$ and $\varrho_{\hat{\alpha}}$ associated to $\alpha$ and $\hat{\alpha}$, we have
\begin{align} \label{MoreUpperBounds}
r_j\super{\alpha} \overset{\eqref{InducedWalkDefn}}{=} 
\begin{cases} 
r_j\super{\hat{\alpha}},  & j \in \{1,2,\ldots,k-1\}, \\ 
r_{k-1}\super{\hat{\alpha}} + 1,  & j = k, \\ 
r_{k-1}\super{\hat{\alpha}},  & j = k+1, \\ 
r_{j-2}\super{\hat{\alpha}},  & j \in \{k+2,k+3,\ldots,n\}. 
\end{cases} 
\end{align}
With $(\ell_k,\ell_{k+1}) \in \{(0,1),(1,0)\}$, we have $(1-2\ell_k) \in \{\pm 1\}$ and
$(1-2\ell_k) + (1-2\ell_{k+1}) = 0$, so~(\ref{UpperBounds},~\ref{MoreUpperBounds}) imply
\begin{align}
\sum_{i \, = \,1}^j (1-2\ell_i) \underset{\eqref{MoreUpperBounds}}{\overset{\eqref{UpperBounds}}{\leq}}  
\begin{cases} 
r_j\super{\hat{\alpha}} = r_j\super{\alpha},  & j \in \{1,2,\ldots,k-1\}, \\ 
r_{k-1}\super{\hat{\alpha}} + 1 = r_k\super{\alpha},  & j = k, \\ 
r_{k-1}\super{\hat{\alpha}} = r_{k+1}\super{\alpha},  & j = k+1, \\ 
r_{j-2}\super{\hat{\alpha}} = r_j\super{\alpha},   & j \in \{k+2,k+3,\ldots,n\}. 
\end{cases} 
\end{align}
Hence,~\eqref{SingletBasisExpandImplies} holds for all $\alpha \in \smash{\LP_n\super{s}}$ with $s < n$ too.
This completes the induction step.
\end{enumerate}
The corresponding assertions for $ \SingBar_{\alphaBar}$ can be proven similarly. 
\end{proof}

\begin{cor} \label{CorSingletBasisExpand} 
Let $\alpha \in \LP_n\super{s}$. 
Then for any walk $\varrho$ over $\OneVec{n}$ such that 
$\varrho_\alpha$ and $\varrho$ are incomparable or $\varrho_\alpha \DPleq \varrho$, we have
\begin{align} \label{CorSingletBasisExpandClaim}
\SPBiForm{\FundBasisBar_n^\varrho}{\Sing_\alpha} = \delta_{\varrho_\alpha, \varrho} (\ii q^{1/2})^{n-s} .
\end{align}
Similarly, if $\alphaBar \in \LPBar_n\super{s}$, then for any walk $\varrho$ over $\OneVec{n}$ such that 
$\varrho_{\alphaBar}$ and $\varrho$ are incomparable or $\varrho_{\alphaBar} \DPleq \varrho$, we have
\begin{align} \label{CorSingletBasisExpandClaimBar}
\SPBiForm{\SingBar_{\alphaBar}}{\FundBasis_n^\varrho} = \delta_{\varrho_{\alphaBar}, \varrho} (\ii q^{1/2})^{n-s} .
\end{align}
\end{cor}

\begin{proof}
If $\varrho_\alpha \DPle \varrho$ or the walks $\varrho_\alpha$ and $\varrho$ are incomparable, 
then property~\eqref{ord} and lemma~\ref{SingletBasisExpandFormulaLem} imply that the coefficient of $\FundBasis_n^\varrho$
in decomposition~\eqref{SingletBasisExpandDef} of $\Sing_\alpha$ equals zero.
Claim~\eqref{CorSingletBasisExpandClaim} then follows from the orthogonality property~\eqref{OrthoFormula}.
On the other hand, if $\varrho = \varrho_\alpha$, then~\eqref{CorSingletBasisExpandClaim} is a straightforward calculation. 
Identity~\eqref{CorSingletBasisExpandClaimBar} can be proven similarly. 
\end{proof}

\begin{lem} \label{SingletBasisIsLinIndepLem} 
Suppose $q \in \bC^\times \setminus \{\pm1\}$. The following hold:
\begin{enumerate}
\itemcolor{red}
\item \label{SingletBasisIsLinIndepItem2} 
The collection $\{\Sing_\alpha \, | \, \alpha \in \LP_n\}$
is a linearly independent subset of $\HWsp_n$.
\item \label{SingletBasisIsLinIndepItem1}
For each $s \in \DefectSet_n$, the collection $ \{\Sing_\alpha \, | \, \alpha \in \LP_n\super{s} \}$
is a linearly independent subset of $\smash{\HWsp_n\super{s}}$.
\end{enumerate}
Similarly, this lemma holds after the symbolic replacements
$\Sing_\alpha \mapsto \SingBar_{\alphaBar}$, $\alpha \mapsto \alphaBar$, $\LP \mapsto \LPBar$, and $\HWsp \mapsto \HWspBar$.
\end{lem}

\begin{proof}
We prove the lemma for $\Sing_\alpha$, for the proof in the case of $\SingBar_{\alphaBar}$ is similar.
In light of lemma~\ref{sGradingLem}, it suffices to show that the collection $\{\Sing_\alpha \, | \, \alpha \in \LP_n \}$ is linearly independent.
For this purpose, we suppose that
\begin{align} 
\label{VanishingLinComb}
\sum_{\alpha \, \in \, \LP_n} c_\alpha \Sing_\alpha = 0 
\end{align}
is a vanishing linear combination of vectors in this set,
where $c_\alpha \in \bC$ are some constants. Taking the bilinear pairing of~\eqref{VanishingLinComb} 
with the vectors $\smash{\FundBasisBar_n^{\varrho}}$ for all walks $\varrho = \varrho_{\betaBar}$ having $\smash{\betaBar \in \LPBar_n}$,
we obtain the system of equations
\begin{align} 
\sum_{\alpha \, \in \, \LP_n} M_{\betaBar,\alpha} \, c_\alpha = 0 , \qquad \text{where} 
\quad M_{\betaBar,\alpha} := \SPBiForm{\FundBasisBar_n^{\varrho_{\betaBar}}}{\Sing_\alpha} .
\end{align}
By corollary~\ref{CorSingletBasisExpand}, the matrix $M = \smash{(M_{\betaBar,\alpha})}$ is upper-triangular 
when its column order respects the partial order $\DPleq$ of the walks over $\OneVec{n}$.
Therefore, we have $c_\alpha=0$ for all $\alpha \in \LP_n$, so
the collection $\{ \Sing_\alpha \, | \, \alpha \in \LP_n \}$ is linearly independent.
\end{proof}

\subsection{The link state -- highest-weight vector correspondence}
\label{subsec: link state hwv correspondence}

Next, we construct the \emph{(valenced) link-pattern basis vectors} $\Sing_\alpha$ for $\alpha \in \LP_\multii$ 
(resp.~$\SingBar_{\alphaBar}$ for $\alphaBar \in \LPBar_\multii$). 
In lemmas~\ref{sGradingLem2} and~\ref{SingletBasisIsLinIndepLem2}, 
we prove that these are linearly independent highest-weight vectors in $\Module{\VecSp_\multii}{\Uqsltwo}$ 
(resp.~in $\RModule{\VecSp_\multii}{\Uqsltwo}$).
We gather important properties of this ``link state -- highest-weight vector correspondence'' in proposition~\ref{HWspLem2}.

As a warm-up, let us consider the vectors $\Projectionhat_\multii ( \Sing_\alpha )$ obtained from $\Sing_\alpha$ via the projector 
$\Projectionhat_\multii \colon \VecSp_{\Summed_\multii} \longrightarrow \VecSp_\multii$, defined in~\eqref{Composites}.
Because it 
is a $\Uqsltwo$-homomorphism, 
the images of $\Sing_\alpha$ are highest-weight vectors (cf. lemmas~\ref{EmbProjLem} and~\ref{sGradingLem}). 
However, they are not all linearly independent, because the map $\Projectionhat_\multii$ has a nontrivial kernel.
In fact, the kernel can be understood explicitly, in terms of ``special link patterns''~\eqref{SPDefn} 
discussed in corollary~\ref{ProjImKerCor} below (see also~\cite[appendix~\red{B}]{fp3a}).
(We recall from item~\ref{It3} of lemma~\ref{EmbProjLem} that
$\ker \Projectionhat_\multii = \ker \Projection_\multii$; corollary~\ref{ProjImKerCor} actually concerns the latter).

To disregard this kernel, we utilize 
the link-state embedding $\WJEmb_\multii \colon \LS_\multii \longrightarrow \LS_{\Summed_\multii}$ 
and its reflection $\smash{\WJProjHat_\multii = \WJEmb_\multii^*}$ 
from (\ref{WJCompEmbAndProjHat},~\ref{LinkEmbDef},~\ref{CommDiagram}),
satisfying~\eqref{StarMapAlphaEmbProj}.
For any multiindex 
$\multii$ as in~(\ref{MultiindexNotation},~\ref{ndefn}),
we define the vectors $\Sing_\alpha$ and $\SingBar_{\alphaBar}$ as
\begin{align} \label{Lmap2} 
\Sing_\alpha := \Projectionhat_\multii(\Sing_{\WJEmb_\multii \alpha})
\qquad\qquad \textnormal{and} \qquad \qquad
\SingBar_{\alphaBar} := \ProjectionhatBar_\multii(\SingBar_{\alphaBar \WJProjHat_\multii}) 
\end{align}
for all valenced link states $\alpha \in \LS_\multii$ and $\alphaBar \in \LSBar_\multii$.
Recalling corollary~\ref{CompositeProjCorHatEmb},
we can illustrate the vectors $\Sing_\alpha$ 
as 
\begin{align} \label{SingAlphaDiagramVal}
\alpha \quad & = \quad \vcenter{\hbox{\includegraphics[scale=0.275]{Figures/e-LinkPattern3_valenced.pdf}}}
\qquad \qquad \underset{\eqref{SingAlphaDiagram}}{\overset{\textnormal{(\ref{WJCompEmbAndProjHat}, \ref{CompTwoProjsHatEmb})}}{\Longrightarrow}} \qquad \qquad 
\Sing_\alpha \quad = \quad  \vcenter{\hbox{\includegraphics[scale=0.275]{Figures/e-LinkPattern3_valenced_orient.pdf} .}} 
\end{align}
The vectors $\SingBar_{\alphaBar}$ have a similar illustration.

\begin{lem} \label{sGradingLem2} 
Suppose $\max \multii < \pmin(q)$. We have 
\begin{align} \label{sGradingLemProperty2} 
\big\{ \Sing_\alpha \, \big | \,\alpha \in \LS_\multii\super{s} \big\} \subset \HWsp_\multii\super{s} 
\qquad \qquad \textnormal{and} \qquad \qquad
\big\{ \SingBar_{\alphaBar} \, \big | \,\alphaBar \in \LSBar_\multii\super{s} \big\} \subset \HWspBar_\multii\super{s} .
\end{align}
\end{lem}

\begin{proof} 
Lemma~\ref{sGradingLem} constitutes the special case $\multii = \OneVec{n}$ for some $n \in \bZnn$.
From this, we get the general case $\multii \in \smash{\bZpos^\#}$:
\begin{align}
\alpha \in \LS_\multii\super{s} 
\quad \overset{\eqref{LinkEmbDef}}{\Longrightarrow} \quad 
\WJEmb_\multii \alpha \in \LS_{\Summed_\multii}\super{s} 
\quad \overset{\eqref{sGradingProperty}}{\Longrightarrow} 
\quad \quad \Sing_{\WJEmb_\multii \alpha} \in \HWsp_{\Summed_\multii}\super{s} 
\quad \overset{RespGradeProjhat}{\Longrightarrow} 
\quad \Sing_\alpha
\overset{\eqref{Lmap2}}{=}
\Projectionhat_\multii( \Sing_{\WJEmb_\multii \alpha} ) \in \HWsp_\multii\super{s} ,
\end{align}
as $\Projectionhat_\multii$ is a $\Uqsltwo$-homomorphism.
This proves the first assertion of~\eqref{sGradingLemProperty2}. The second assertion can be proven similarly.
\end{proof}

The next lemma implies
that the linear map $\alpha \mapsto \Sing_\alpha$ (resp.~$\alphaBar \mapsto \SingBar_{\alphaBar}$)
respects the action of the valenced Temperley-Lieb algebra 
on one hand,
on its link state module $\LS_\multii$ (resp.~$\LSBar_\multii$) and, 
on the other hand, on $\CModule{\VecSp_\multii}{\TL}$ (resp.~$\CRModule{\VecSpBar_\multii}{\TL}$). 
In particular, the diagram action on the valenced link states $\alpha$ and $\alphaBar$ agrees with 
the graphical rules for the $\smash{\TL_\multii^\multiii}$-action on the corresponding vectors $\Sing_\alpha$ and $\SingBar_{\alphaBar}$
represented as valenced link states with upward-oriented defects~\eqref{SingAlphaDiagramVal}.

\begin{lem} \label{SmoothingLem2} 
Suppose $\max(\multii,\multiii) < \pmin(q)$.   
For all valenced link states $\alpha \in \LS_\multiii$ and $\alphaBar \in \LSBar_\multii$ 
and valenced tangles $T \in \TL_\multii^\multiii$, we have
\begin{align} \label{Lcommutation} 
\Sing_{T \alpha} = T  \Sing_\alpha
\qquad \textnormal{and} \qquad 
\SingBar_{\alphaBar \, T} = \SingBar_{\alphaBar}  T
\end{align}
\end{lem}

\begin{proof} 
Lemma~\ref{SmoothingLem} gives the special case of $\multii = \OneVec{n}$ and $\multiii = \OneVec{m}$ for some $n,m \in \bZpos$.
For general $\multii , \multiii \in \smash{\bZpos^\#}$, we have
\begin{align} 
\nonumber 
\Sing_{T\alpha} 
\, \overset{\eqref{Lmap2}}{=} \, 
\Projectionhat_\multii(\Sing_{\WJEmb_\multii T\alpha})
&\, \overset{\eqref{IdCompAndWJPhatPEmb}}{=} \, 
\Projectionhat_\multii \big(\Sing_{\WJEmb_\multii T \WJProjHat_\multiii \WJEmb_\multiii \alpha} \big) \\
&\, \overset{\eqref{Lcommutation}}{=} \, 
\Projectionhat_\multii \big( \WJEmb_\multii T \smash{\WJProjHat}_\multiii  \Sing_{ \WJEmb_\multiii\alpha} \big)
\, \underset{\eqref{Lmap2}}{\overset{\eqref{CompTwoProjsHatEmb}}{=}} \, 
\Projectionhat_\multii \big( \WJEmb_\multii T  \Sing_\alpha \big)
\, \underset{\eqref{CompTwoProjsHatEmb}}{\overset{\eqref{IdCompAndWJPhatPEmb}}{=}} \, 
\Projectionhat_\multii ( \WJEmb_\multii T \smash{\WJProjHat}_\multiii  \Embedding_\multiii (\Sing_\alpha) )
\, \overset{\eqref{ImultHom}}{=} \,  
T  \Sing_\alpha 
\end{align}
for any valenced tangle $T \in \TL_\multii^\multiii$ and for any valenced link state $\alpha \in \LS_\multiii$.
This yields the first equation of~\eqref{Lcommutation}.
Similar work shows the second equation of~\eqref{Lcommutation}.
\end{proof}

To prove the linear independence of the valenced link-pattern basis vectors $\Sing_\alpha$ and $\SingBar_{\alphaBar}$
(lemma~\ref{SingletBasisIsLinIndepLem2}), we begin with an auxiliary observation. 

\begin{lem} \label{StableProjectionLem}
Suppose $\max \multii < \pmin(q)$. We have
\begin{align} \label{StableProjectionInclusion}
\{ \Sing_{\WJEmb_\multii \alpha} \, | \, \alpha \in \LS_\multii \} 
 = \{ \Sing_{\WJProj_\multii \alpha} \, | \, \alpha \in \LS_{\Summed_\multii} \} 
\subset \im \Projection_\multii,
\end{align}
and similarly,
\begin{align} \label{StableProjectionInclusionBar}
\{ \SingBar_{\alphaBar \, \WJProjHat_\multii} \, | \, \alphaBar \in \LSBar_\multii \} 
 = \{ \SingBar_{\alphaBar \, \WJProj_\multii} \, | \, \alphaBar \in \LSBar_{\Summed_\multii} \} 
\subset \im \ProjectionBar_\multii .
\end{align}
\end{lem}

\begin{proof}
Commuting diagram~\eqref{CommDiagram} shows that
$\im \WJEmb_\multii (\,\cdot\,) = \im \WJProj_\multii (\,\cdot\,)$, which implies the equality in~\eqref{StableProjectionInclusion}.
Using homomorphism-like property~\eqref{Lcommutation} from lemma~\ref{SmoothingLem2}
and corollary~\ref{CompositeProjCor}, we obtain  
\begin{align} \label{StableProjectionInProof}
\Projection_\multii ( \Sing_{\WJEmb_\multii \alpha} )
\underset{\eqref{Lcommutation}}{\overset{\eqref{CompTwoProjs}}{=}}
\Sing_{\WJProj_\multii \WJEmb_\multii \alpha}
\overset{\eqref{IdCompAndWJPhatPEmb}}{=}  
\Sing_{\WJEmb_\multii \alpha} .
\end{align}
This proves the the first asserted inclusion in~\eqref{StableProjectionInclusion}.
The second asserted formula~\eqref{StableProjectionInclusionBar} can be proven similarly.
\end{proof}

\begin{lem} \label{SingletBasisIsLinIndepLem2} 
Suppose $\max \multii < \pmin(q)$. The following hold:
\begin{enumerate}
\itemcolor{red}
\item \label{SingletBasisIsLinIndep2Item2} 
The collection $\{\Sing_\alpha \, | \, \alpha \in \LP_\multii\}$
is a linearly independent subset of $\HWsp_\multii$.
\item \label{SingletBasisIsLinIndep2Item1}
For each $s \in \DefectSet_\multii$, the collection $ \{\Sing_\alpha \, | \, \alpha \in \LP_\multii\super{s} \}$
is a linearly independent subset of $\smash{\HWsp_\multii\super{s}}$.
\end{enumerate}
Similarly, this lemma holds after the symbolic replacements 
$\Sing_\alpha \mapsto \SingBar_{\alphaBar}$, $\alpha \mapsto \alphaBar$, 
$\LP \mapsto \LPBar$, and $\HWsp \mapsto \HWspBar$.
\end{lem}

\begin{proof}
We prove the lemma for $\Sing_\alpha$, 
as $\SingBar_{\alphaBar}$ are similar.
By lemma~\ref{sGradingLem2}, it suffices to show that 
$\{\Sing_\alpha \, | \, \alpha \in \LP_\multii \}$ is linearly independent.
Lemma~\ref{SingletBasisIsLinIndepLem} gives the special case $\multii = \OneVec{n}$, and
the general case then follows using the following facts.
First, the map $\WJEmb_\multii (\,\cdot\,)$ is a linear isomorphism from 
$\LS_\multii$ to $\WJProj_\multii  \LS_{\Summed_\multii}$ by~\eqref{CommDiagram}.
Second, lemma~\ref{StableProjectionLem} gives 
$\Sing_{\WJEmb_\multii \alpha} \in \im \Projection_\multii$ for all $\alpha \in \LP_\multii$. 
Third, by lemma~\ref{EmbProjLem}, the map 
$\Projectionhat_\multii$ is a linear isomorphism from $\im \Projection_\multii$ to $\VecSp_\multii$.
Combining these facts with lemma~\ref{SingletBasisIsLinIndepLem},
we conclude that the map $\alpha \mapsto \Sing_\alpha$ defined in~\eqref{Lmap2}
injectively sends the linearly independent collection 
$\LP_\multii$ of valenced link patterns to the collection $\{ \Sing_\alpha \, | \, \alpha \in \LP_\multii \}$, 
which therefore is linearly independent too.
\end{proof}

We gather the main results obtained in this section into the following proposition:

\begin{prop} \label{HWspLem2}
\textnormal{(Link state -- highest-weight vector correspondence):} 
Suppose $\max \multii < \pmin(q)$. 
The map $\alpha \mapsto \Sing_\alpha$ from $\LS_\multii$ to $\HWsp_\multii$ has the following properties:
\begin{enumerate}
\itemcolor{red}

\item \label{Lmap2Item4} 
It is a homomorphism of $\TL_\multii(\nu)$-modules, i.e.,~\eqref{Lcommutation} holds
for all valenced tangles $T \in \TL_\multii(\nu)$.

\item \label{Lmap2Item2} 
It is a linear injection. Furthermore, if $\Summed_\multii < \pmin(q)$, then it is a linear isomorphism.

\item \label{Lmap2Item3} 
It respects the $s$-grading~\textnormal{(\ref{HWDirectSumPM},~\ref{LSDirSum})}, i.e., 
we have $\{ \Sing_\alpha \, | \, \alpha \in \smash{\LS_\multii\super{s}} \} \subset \smash{\HWsp_\multii\super{s}}$, 
with equality if $\Summed_\multii < \pmin(q)$.
\end{enumerate}
Similarly, items~\ref{Lmap2Item4}--\ref{Lmap2Item3} hold for the map 
$\alphaBar \mapsto \SingBar_{\alphaBar}$ after the symbolic replacements
\begin{align}
\alpha \mapsto \alphaBar , \qquad
\Sing_\alpha \mapsto \SingBar_{\alphaBar} , \qquad
\LS \mapsto \LSBar , 
\qquad \textnormal{and} \qquad 
\HWsp \mapsto \HWspBar .
\end{align} 
Finally, the maps $\alpha \mapsto \Sing_\alpha$
and $\alphaBar \mapsto \SingBar_{\alphaBar}$
together preserve the bilinear pairing and the map $(\, \cdot \,)^*$~\textnormal{(\ref{StarMap},~\ref{StarMapAlpha}):}
\begin{enumerate}
\setcounter{enumi}{3}
\itemcolor{red}
\item \label{Lmap2Item1} 
For all valenced link states $\alphaBar \in \LSBar_\multii$ and
$\beta \in \LS_\multii$, we have
\begin{align} \label{BilinFormPreserve}
\SPBiForm{\SingBar_{\alphaBar}}{\Sing_\beta} = \LSBiFormBar{\alphaBar}{\beta} .
\end{align}

\item For all valenced link states $\alpha \in \LS_\multii$ and $\betaBar \in \LSBar_\multii$, we have
\begin{align} \label{StarMapPreserve}
\Sing_\alpha^* = \SingBar_{\alpha^*}
\qquad \qquad \textnormal{and} \qquad \qquad 
\SingBar_{\betaBar}^* = \Sing_{\betaBar^*} .
\end{align}
\end{enumerate}
\end{prop}

\begin{proof}
Lemma~\ref{sGradingLem} implies that the image of the map $\alpha \mapsto \Sing_\alpha$ is a subset of $\HWsp_\multii$.
We prove items~\ref{Lmap2Item4}--\ref{Lmap2Item1} as follows:
\begin{enumerate}[leftmargin=*]
\itemcolor{red}

\item This is a restatement of lemma~\ref{SmoothingLem2} specialized to the case $\multiii = \multii$.

\item 
Lemma~\ref{SingletBasisIsLinIndepLem2}  
shows that the map $\alpha \mapsto \Sing_\alpha$ is a linear injection. 
Also, lemma~\ref{LSDimLem2}, 
corollary~\ref{CobloBasisCor}, 
and a dimension count imply that
if $\Summed_\multii < \pmin(q)$, then $\alpha \mapsto \Sing_\alpha$ is an isomorphism of vector spaces:
\begin{align}
\Summed_\multii < \pmin(q) 
\qquad \qquad \Longrightarrow \qquad \qquad 
\dim \LS_\multii \overset{\eqref{LSDim2}}{=}   
\Dim_\multii
\overset{\eqref{HWspDimension}}{=} 
\dim \HWsp_\multii.
\end{align}

\item This follows from lemma~\ref{sGradingLem2}, item~\ref{Lmap2Item2}, and a dimension count: for all $s \in \DefectSet_\multii$, we have
\begin{align}
\Summed_\multii < \pmin(q) 
\qquad \qquad \Longrightarrow \qquad \qquad 
\dim \LS_\multii\super{s} \overset{\eqref{LSDim2}}{=}   
\Dim_\multii\super{s}
\overset{\eqref{HWspDimension}}{=} 
\dim \HWsp_\multii\super{s}.
\end{align}
\end{enumerate}
The corresponding statements for the map $\alphaBar \mapsto \SingBar_{\alphaBar}$
can be proven similarly.
\begin{enumerate}[leftmargin=*]
\setcounter{enumi}{3}
\itemcolor{red}
\item We first prove the case $\multii = \OneVec{n}$ for some $n \in \bZpos$. 
To this end, definition~\ref{SingletBasisDefinition} and lemma~\ref{BiFormLem} combined with 
the graphical representation of the vectors $\SingBar_{\alphaBar}$ and $\Sing_\beta$ as link states with upward-oriented defects
imply that 
\begin{align} 
\SPBiForm{\SingBar_{\alphaBar}}{\Sing_\beta}
\underset{\eqref{WeightProd}}{\overset{\eqref{SingAlphaDiagram}}{=}}
(\, \SingBar_{\alphaBar} \BarAction \Sing_\beta \,)
\underset{\textnormal{(\ref{TLfugacity}--\ref{TLTurnBack0})}}{\overset{\textnormal{(\ref{ThroughPathWeight}--\ref{TurnBackWeightZero})}}{=}}
\LSBiFormBar{\alphaBar}{\beta} ,
\end{align}
by noticing that the evaluation of the oriented network $\SingBar_{\alphaBar} \BarAction \Sing_\beta$ via rules 
(\ref{ThroughPathWeight}--\ref{TurnBackWeightZero}) agrees with the evaluation of the non-oriented network
$\alphaBar \BarAction \beta$ via rules (\ref{TLfugacity}--\ref{TLTurnBack0}) because all defects in 
$\SingBar_{\alphaBar}$ and $\Sing_\beta$ have upward orientation.
For the general case $\multii \in \smash{\bZpos^\#}$, lemma~\ref{biformPropertyLem} and corollary~\ref{PmoveCor} imply that
\begin{align} \label{BilinFormPreserveAux}
\SPBiForm{\SingBar_{\alphaBar}}{\Sing_\beta}
\overset{\eqref{Lmap2}}{=}
\SPBiForm{\ProjectionhatBar_\multii(\SingBar_{\alphaBar \WJProjHat_\multii})}{\Projectionhat_\multii(\Sing_{\WJEmb_\multii \beta})}
\underset{\eqref{SPBiFormNewEmbed}}{\overset{\eqref{UQIdComp}}{=}}
\SPBiForm{\ProjectionBar_\multii (\SingBar_{\alphaBar \WJProjHat_\multii})}{\Projection_\multii(\Sing_{\WJEmb_\multii \beta})}
\underset{\eqref{StableProjectionInProof}}{\overset{\eqref{Pmove}}{=}}
\SPBiForm{\SingBar_{\alphaBar \WJProjHat_\multii}}{\Sing_{\WJEmb_\multii \beta}} 
\end{align}
for all valenced link states $\smash{\alphaBar \in \LSBar_\multii}$ and $\beta \in \LS_\multii$.
Then, using the already proven identity~\eqref{BilinFormPreserve} for the link states
$\alphaBar \smash{\WJProjHat_\multii \in \LSBar_{\Summed_\multii}}$ and 
$\WJEmb_\multii \beta \in \LS_{\Summed_\multii}$, we obtain the asserted identity:
\begin{align}
\SPBiForm{\SingBar_{\alphaBar}}{\Sing_\beta}
\overset{\eqref{BilinFormPreserveAux}}{=}
\SPBiForm{\SingBar_{\alphaBar \WJProjHat_\multii}}{\Sing_{\WJEmb_\multii \beta}}
\overset{\eqref{BilinFormPreserve}}{=}
\LSBiFormBar{\alphaBar \WJProjHat_\multii}{\WJEmb_\multii \beta} 
\overset{\eqref{LSBiFormExt}}{=} 
\LSBiFormBar{\alphaBar}{\beta} .
\end{align}

\item We first prove the case $\multii = \OneVec{n}$  for some $n \in \bZpos$.
By linearity, we may assume that $\alpha$ and $\betaBar$ are link patterns.  
We prove the left equation in~\eqref{StarMapPreserve} by induction on the number of links in $\alpha$;
the right equation in~\eqref{StarMapPreserve} can be proven similarly. 
In the initial case $\alpha$ has no links, so 
$\alpha = \defects_n \in \smash{\LP_n\super{n}}$ with $\smash{\defects_n^* = \defectsBar_n}$,
and definition~\ref{SingletBasisDefinition} gives 
\begin{align} 
\Sing_{\scaleobj{0.85}{\defects_n}}^* 
\overset{\eqref{SingletBasisDefAllDefects}}{=} 
\MTbas_0\superscr{(n) \, *}
\underset{\eqref{MThwv}}{\overset{(\text{\ref{StarMap}, \ref{StarDistribute}})}{=}}
\MTbasBar_0\super{n} 
\overset{\eqref{SingletBasisDefAllDefectsBar}}{=} 
\Sing_{\scaleobj{0.85}{\defectsBar_n}} .
\end{align}
Next, we let $\ell \in \bZpos$ and assume that the left equation in~\eqref{StarMapPreserve} holds for all link patterns with at most $\ell-1$ links.
For the induction step, we let $\alpha \in \smash{\LS_n\super{s}}$ have $\ell$ links.
Now, if $\hat{\alpha}$ is a link pattern obtained by dropping a link from $\alpha$ joining the $k$:th and $(k+1)$:st nodes,
then $\hat{\alpha}$ has $\ell-1$ links, so using the induction hypothesis, we find that
\begin{align}
\nonumber
\Sing_\alpha^* 
& \underset{\eqref{SingletBasisDef}}{\overset{\eqref{StarEmbed2x2}}{=} }
-q \left(\frac{q-q^{-1}}{\ii q^{1/2}}\right) 
\big(\id^{\otimes(k-1)} \otimes \CCembedorBar\super{0}\sub{1,1} \otimes \id^{\otimes(n-k-1)}\big) (\Sing_{\hat{\alpha}}^*) \\
\nonumber
& \underset{\hphantom{\eqref{StarEmbed2x2}}}{\overset{\eqref{StarMapPreserve}}{=} }
-q \left(\frac{q-q^{-1}}{\ii q^{1/2}}\right) 
\big(\id^{\otimes(k-1)} \otimes \CCembedorBar\super{0}\sub{1,1} \otimes \id^{\otimes(n-k-1)}\big) (\SingBar_{\hat{\alpha}^*}) \\
&  \underset{\hphantom{\eqref{StarEmbed2x2}}}{\overset{\eqref{SingletBasisDefBar}}{=} }
-q^{-1} (-q) \,  \SingBar_{\hat{\alpha}^*}
= \SingBar_{\hat{\alpha}^*} .
\end{align}
This proves the induction step and thus the left equation in~\eqref{StarMapPreserve} for $\multii = \OneVec{n}$.
For the general case $\multii \in \smash{\bZpos^\#}$, item~\ref{EmbeddingStarItem} of lemma~\ref{EmbProjLem} gives
\begin{align} 
\Sing_\alpha^* \overset{\eqref{Lmap2}}{=} 
\Projectionhat_\multii(\Sing_{\WJEmb_\multii \alpha})^*
\overset{\eqref{StarProjHat}}{=} 
\ProjectionhatBar_\multii \big( \Sing_{\WJEmb_\multii \alpha}^* \big) 
\underset{\eqref{StarMapPreserve}}{\overset{\eqref{StarMapAlphaEmbProj}}{=} }
\ProjectionhatBar_\multii \big( \SingBar_{\alpha^* \WJProjHat_\multii} \big) 
\overset{\eqref{Lmap2}}{=} \SingBar_{\alpha^*}
\end{align}
for all valenced link states $\alpha \in \LS_\multii$.
This proves the left equation in~\eqref{StarMapPreserve};
the right equation is similar. 
\end{enumerate}
The proof is complete.
\end{proof}

If $\Summed_\multii \geq \pmin(q)$, then
there can exist other $\Uqsltwo$-highest-weight vectors than the ones obtained from the embedding $\alpha \mapsto \Sing_\alpha$. 
This fact is related to degeneracies that arise when the algebra $\Uqsltwo$ is not semisimple and its representation theory
does not play well with the usual highest-weight theory of Lie algebras~\cite{Lusz, cp, ck}.
Indeed, for instance, one can check directly from definitions~\eqref{HopfRep} 
that the element $E^{\pmin(q)}$ acts as zero on any simple type-one module $\Wd\sub{s}$, 
recalling from~\eqref{Qinteger} that $[k \pmin(q)] = 0$ for any $k \in \bZnn$.
In particular, the basis vector $\smash{\Basis_{\ell}\super{s}}$ with $\ell = \pmin(q)$ is thus a highest-weight vector if $s \geq \pmin(q)$.
More generally, using this and formula~\ref{CoproductFormulas} from appendix~\ref{PreliApp} for the coproduct of powers of $E$,
it follows that $E^{\pmin(q)}$ acts as zero also on any type-one module $\Module{\VecSp_\multii}{\Uqsltwo}$.
(See also appendix~\ref{ExceptionalQSect}.)

We show in section~\ref{QuotientRadicalSect} that the highest-weight vectors which are not reachable via 
the link state -- highest-weight vector correspondence are orthogonal to the link-pattern basis vectors $\Sing_\alpha$.
In particular, we show that certain quotients of the $\Uqsltwo,\UqsltwoBar$-highest-weight vector spaces 
$\HWsp_\multii$ and $\HWspBar_\multii$ are isomorphic to simple $\TL_\multii(\nu)$-modules even if $\Summed_\multii \geq \pmin(q)$.

\subsection{Direct-sum decompositions}
\label{subsec: dual direct-sum decomposition}

The next result 
improves proposition~\ref{MoreGenDecompAndEmbProp},
with the conformal-block basis replaced by the link-pattern basis.

\begin{prop} \label{MoreGenDecompAndEmbProp2}
Suppose $\max \multii < \pmin(q)$.
There exists an embedding of left $\Uqsltwo$-modules
\begin{align} \label{DirectSumInclusion2}
\bigoplus_{\substack{s \, \in \, \DefectSet_\multii \\ s \, < \, \pmin(q) }} \Dim_\multii\super{s} \Wd\sub{s} 
\quad \lhook\joinrel\rightarrow \quad \Module{\VecSp_\multii}{\Uqsltwo}
\end{align}
such that the following hold:
\begin{enumerate}
\itemcolor{red}
\item \label{DirectSumInclusionLem2Item1}
For each valenced link pattern $\alpha \in \smash{\LP_\multii\super{s}}$, the collection
\begin{align} 
\label{BasisForm3} 
\big\{ F^\ell.\Sing_\alpha \, \big| \, 0 \leq \ell \leq s \big\}
\end{align}
is a basis for the image of a unique direct summand $\Wd\sub{s}$ in~\eqref{DirectSumInclusion2}.

\item  \label{DirectSumInclusionLem2Item2}
The image of each summand $\Wd\sub{s}$ has a unique basis of the form~\eqref{BasisForm3} with $\alpha \in \smash{\LP_\multii\super{s}}$.

\item  \label{DirectSumInclusionLem2Item3} 
If $\Summed_\multii < \pmin(q)$, then~\eqref{DirectSumInclusion2} is an isomorphism of left $\Uqsltwo$-modules,
\begin{align} \label{MoreGenDecomp2} 
\Module{\VecSp_\multii}{\Uqsltwo} \isom 
\bigoplus_{s \, \in \, \DefectSet_\multii} \Dim_\multii\super{s} \Wd\sub{s} .
\end{align}
\end{enumerate}
Similarly, this proposition holds for right $\Uqsltwo$-modules after the symbolic replacements
\begin{align}  \label{DirectSumInclusion2Replace}
\Wd \mapsto \WdBar, \qquad 
\Module{\VecSp_\multii}{\Uqsltwo} \mapsto \RModule{\VecSpBar_\multii}{\Uqsltwo} ,
\qquad  \alpha \mapsto \alphaBar ,
\qquad  \LP \mapsto \LPBar,
\qquad \textnormal{and} \qquad 
F^\ell.\Sing_\alpha \mapsto \SingBar_{\alphaBar}.E^\ell.
\end{align}
Finally, both the left-action and right-action versions of this proposition hold after replacing $\Uqsltwo \mapsto \UqsltwoBar$ in either.
\end{prop}

\begin{proof}
After replacing the conformal-block vectors $\HWvec^{\varrho}_{\multii}$
with the link-pattern basis vectors $\Sing_\alpha$ and 
instead of lemma~\ref{HWPropLem0} (resp.~lemma~\ref{ConfBlockLinIndepLem})
referencing lemma~\ref{sGradingLem2} (resp.~lemma~\ref{SingletBasisIsLinIndepLem2}),  
the proof of proposition~\ref{MoreGenDecompAndEmbProp} adapts almost verbatim. 
\end{proof}

\begin{cor} \label{SingletBasisIsBasisCor} 
Suppose $\Summed_\multii < \pmin(q)$. Then, the following hold:
\begin{enumerate}
\itemcolor{red}
\item \label{BasIt2} 
The collection $\{\Sing_\alpha \, | \, \alpha \in \LP_\multii\}$
is a basis for $\HWsp_\multii$.

\item \label{BasIt1}
For each $s \in \DefectSet_\multii$, the collection $ \{\Sing_\alpha \, | \, \alpha \in \LP_\multii\super{s} \}$
is a basis for $\smash{\HWsp_\multii\super{s}}$.
\end{enumerate}
Similarly, this corollary holds after the symbolic replacements $\Sing_\alpha \mapsto \SingBar_{\alphaBar}$, $\alpha \mapsto \alphaBar$,
$\LP \mapsto \LPBar$, and $\HWsp \mapsto \HWspBar$.
\end{cor}

\begin{proof}
This immediately follows from lemma~\ref{SingletBasisIsLinIndepLem2} and proposition~\ref{MoreGenDecompAndEmbProp2}.
\end{proof}

The above result gives the dimension of the space $\smash{\HWsp_\multii\super{s}}$ of highest-weight vectors:

\begin{cor} 
Suppose $\max \multii < \pmin(q)$. For all $s \in \DefectSet_\multii$, we have
\begin{align} 
\Dim_\multii\super{s} \leq \dim \HWsp_\multii\super{s} \leq B_\multii\super{s} 
\qquad \qquad \textnormal{and} \qquad \qquad
\Dim_\multii\super{s} \leq \dim \HWspBar_\multii\super{s} \leq B_\multii\super{s} .
\end{align}
Furthermore, if $\Summed_\multii < \pmin(q)$, then we have
$\smash{\dim \HWsp_\multii\super{s} = \dim \HWspBar_\multii\super{s}} = \Dim_\multii\super{s}$.
\end{cor}

\begin{proof}
The lower bound follows from lemma~\ref{LSDimLem2} 
with item~\ref{SingletBasisIsLinIndep2Item1} of lemma~\ref{SingletBasisIsLinIndepLem2}, 
and the upper bound from corollary~\ref{UpperBoundDimension2Coro}.
Corollary~\ref{SingletBasisIsBasisCor} with lemma~\ref{LSDimLem2} 
give the dimensions when $\Summed_\multii < \pmin(q)$.
\end{proof}

In proposition~\ref{HWspaceDecTL} in section~\ref{GenDiacActTypeOneSec}, 
we established a ``dual'' direct-sum decomposition to that~\eqref{MoreGenDecomp2} 
of proposition~\ref{MoreGenDecompAndEmbProp2} (or~\eqref{MoreGenDecomp} of proposition~\ref{MoreGenDecompAndEmbProp}):
decomposition of the left $\TL_\multii(\nu)$-module
$\CModule{\VecSp_\multii}{\TL}$ (resp.~$\CRModule{\VecSpBar_\multii}{\TL}$)
into a direct sum of $\TL_\multii(\nu)$-submodules when $\Summed_\multii < \pmin(q)$.
However, we did not prove that the summands in this decomposition are simple.
Next, we prove a refinement for proposition~\ref{HWspaceDecTL},
realizing the direct summands $\CModule{\HWsp_\multii\super{s}}{\TL}$
equivalently as the link state modules $\smash{\LS_\multii\super{s}}$. 
We note in particular that 
$\CModule{\HWsp_\multii\super{s}}{\TL} \cong \smash{\LS_\multii\super{s}}$
are simple $\TL_\multii(\nu)$-modules if $\Summed_\multii < \pmin(q)$.

\begin{prop} \label{HWspacePropEmbAndIso}
Suppose $\max \multii < \pmin(q)$.
There exists an embedding of left $\TL_\multii(\nu)$-modules
\begin{align} \label{ContainmentEmb}
\bigoplus_{ \substack{s \, \in \, \DefectSet_\multii \\ s \, < \, \pmin(q)} } (s + 1) \, \LS_\multii\super{s}
\quad \lhook\joinrel\rightarrow \quad \CModule{\VecSp_\multii}{\TL} 
\end{align}
such that the following hold:
\begin{enumerate}
\itemcolor{red}
\item \label{HWspacePropEmbAndIsoItem1}
For each integer $\ell \in \{0, 1, \ldots, s\}$, the collection
\begin{align} \label{BasisForm6} 
\big\{ F^\ell.\Sing_\alpha \, \big| \, \alpha \in \LP_\multii\super{s} \big\}
\end{align}
is a basis for the image of a unique direct summand $\smash{\LS_\multii\super{s}}$ in~\eqref{ContainmentEmb}.

\item \label{HWspacePropEmbAndIsoItem2}
The image of each summand $\smash{\LS_\multii\super{s}}$ has a unique basis of the form~\eqref{BasisForm6} with $\ell \in \{0, 1, \ldots, s\}$.

\item \label{HWspacePropEmbAndIsoItem3}
If $\Summed_\multii < \pmin(q)$, then~\eqref{ContainmentEmb} is an isomorphism of left $\TL_\multii(\nu)$-modules,
\begin{align} \label{WJVnDecomp2} 
\CModule{\VecSp_\multii}{\TL} \isom 
\bigoplus_{s \, \in \, \DefectSet_\multii} (s + 1) \, \LS_\multii\super{s} .
\end{align}
\end{enumerate}
Similarly, this proposition holds for right $\TL_\multii(\nu)$-modules after the symbolic replacements
\begin{align} 
\LS \mapsto \LSBar, \qquad  
\CModule{\VecSp_\multii}{\TL} \mapsto \CRModule{\VecSpBar_\multii}{\TL} , \qquad
F^\ell.\Sing_\alpha \mapsto \SingBar_{\alphaBar}.E^\ell ,\qquad  
\alpha \mapsto \alphaBar ,
\qquad \textnormal{and} \qquad 
\LP \mapsto \LPBar .
\end{align}
\end{prop}

\begin{proof}
The map $\alpha \mapsto \Sing_\alpha$ in proposition~\ref{HWspLem2} gives an embedding of left $\TL_\multii(\nu)$-modules: 
\begin{align}  \label{PreEmbedding}
\LS_\multii \overset{\eqref{LSDirSum}}{=}
\bigoplus_{s \, \in \, \DefectSet_\multii} \LS_\multii\super{s}
\quad \lhook\joinrel\rightarrow \quad 
\bigoplus_{s \, \in \, \DefectSet_\multii} \CModule{\HWsp_\multii\super{s}}{\TL}
\overset{\eqref{HWDirectSumPM}}{\subset}
\CModule{\HWsp_\multii}{\TL} .
\end{align}
By item~\ref{HWspaceDecTLItem3} of proposition~\ref{HWspaceDecTL}, for 
each $0\leq\ell\leq s < \pmin(q)$, the left $\TL_\multii(\nu)$-modules $\CModule{\HWsp_\multii\super{s}}{\TL}$ and
$\CModule{F^\ell.\HWsp_\multii\super{s}}{\TL}$ are isomorphic, 
and since the $K$-eigenvalues for different $\ell$ are distinct (for fixed $s$), their sum is direct.
Furthermore, by fact~\ref{HWVFact}, each of these modules has multiplicity $s+1 = \dim \Wd\sub{s}$.  
In conclusion, embedding~\eqref{PreEmbedding} yields 
\begin{align} 
\bigoplus_{ \substack{s \, \in \, \DefectSet_\multii \\ s \, < \, \pmin(q)} } (s + 1) \, \LS_\multii\super{s}
\quad \overset{\eqref{PreEmbedding}}{\lhook\joinrel\rightarrow} \quad 
\bigoplus_{ \substack{s \, \in \, \DefectSet_\multii \\ s \, < \, \pmin(q)} } 
\bigoplus_{\ell \, = \, 0}^s \CModule{F^\ell.\HWsp_\multii\super{s}}{\TL} 
\; \subset \; \CModule{\VecSp_\multii}{\TL} ,
\end{align}
which gives~\eqref{ContainmentEmb}. 
Items~\ref{HWspacePropEmbAndIsoItem1} and~\ref{HWspacePropEmbAndIsoItem2} 
are then immediate from the construction. 
Item~\ref{HWspacePropEmbAndIsoItem3} 
follows by corollary~\ref{SingletBasisIsBasisCor} and a dimension count,
for if $\Summed_\multii < \pmin(q)$,
then the dimensions of both sides of~\eqref{ContainmentEmb} are equal by 
item~\ref{DnumberItem3} of lemma~\ref{DnumberLem}.
\end{proof}

To finish, we gather some further implications of proposition~\ref{HWspLem2}.

\begin{cor} \label{ProjImageandKernelCor}
Suppose $\max \multii < \pmin(q)$. We have 
\begin{align} 
\label{ProjImageandKernel}
\Sing_\alpha \in \im \Projection_\multii
\quad \Longleftrightarrow \quad
\alpha \in \im \WJProj_\multii (\,\cdot\,)
\qquad \qquad \textnormal{and} \qquad \qquad
\Sing_\alpha \in \ker \Projection_\multii
\quad \Longleftrightarrow \quad
\alpha \in \ker \WJProj_\multii (\,\cdot\,) ,
\end{align}
Similarly, this corollary holds after the symbolic replacements 
$\Sing_\alpha \mapsto \SingBar_{\alphaBar}$, $\Projection \mapsto \ProjectionBar$, and  $\alpha \mapsto \alphaBar$.
\end{cor}

\begin{proof}
Because $\Projection_\multii$ and $\WJProj_\multii(\,\cdot\,)$ are projections
by lemma~\ref{EmbProjLem} and property~\eqref{ProjectorID2} of Jones-Wenzl projector, we have
\begin{align} \label{ProjProj}
v \in \im \Projection_\multii \quad \Longleftrightarrow \quad \Projection_\multii(v) = v 
\qquad \qquad \text{and} \qquad \qquad
\alpha \in \im \WJProj_\multii (\,\cdot\,) \quad \Longleftrightarrow \quad \WJProj_\multii \alpha = \alpha .
\end{align}
Using homomorphism-like property~\eqref{Lcommutation} from lemma~\ref{SmoothingLem2},
identifying the action of the tangle $\WJProj_\multii$ with the projection $\Projection_\multii$
via corollary~\ref{CompositeProjCor}, and recalling
that the map $\alpha \mapsto \Sing_\alpha$ is a linear injection by 
item~\ref{Lmap2Item2} of proposition~\ref{HWspLem2}, we obtain
\begin{alignat}{4}
\label{PLImagePf}
& \Sing_\alpha \in \im \Projection_\multii \quad
&& \overset{\eqref{ProjProj}}{\Longleftrightarrow} \quad 
\hphantom{0} \Sing_\alpha = \Projection_\multii (\Sing_\alpha)
\underset{\eqref{Lcommutation}}{\overset{\eqref{CompTwoProjs}}{=}}
\Sing_{\WJProj_\multii \alpha} \quad
&& \underset{\textnormal{item~\ref{Lmap2Item2}}}{\overset{\textnormal{prop.~\ref{HWspLem2},}}{\Longleftrightarrow}} 
\quad \WJProj_\multii \alpha = \alpha \quad
&& \overset{\eqref{ProjProj}}{\Longleftrightarrow} \quad 
\alpha \in \im \WJProj_\multii (\,\cdot\,), \\
\label{PLKernelPf}
& \Sing_\alpha \in \ker \Projection_\multii \quad
&& \overset{\hphantom{\eqref{ProjProj}}}{\Longleftrightarrow} \quad 
\hphantom{\Sing_\alpha} 0 = \Projection_\multii (\Sing_\alpha)
\underset{\eqref{Lcommutation}}{\overset{\eqref{CompTwoProjs}}{=}}
\Sing_{\WJProj_\multii \alpha} = 0 \quad
&& \underset{\textnormal{item~\ref{Lmap2Item2}}}{\overset{\textnormal{prop.~\ref{HWspLem2},}}{\Longleftrightarrow}} 
\quad
\WJProj_\multii \alpha = 0 \quad
&& \overset{\hphantom{\eqref{ProjProj}}}{\Longleftrightarrow} \quad 
\alpha \in \ker \WJProj_\multii (\,\cdot\,).
\end{alignat}
This proves the first assertion of~\eqref{ProjImageandKernel}, and
the second assertion can be proven similarly. 
\end{proof}

Next, we describe the image and kernel of the map $\Projection_\multii$. For this purpose, we introduce ``special link patterns."
We group the the $\Summed_\multii$ left nodes of a $\Summed_\multii$-link pattern into the bins of nodes
\begin{alignat}{4} \label{lblocks} 
&\text{left:} \quad &&\{ 1, 2, \ldots, \sIndex_1 \} , \quad && \{ \sIndex_1 + 1, \sIndex_1 + 2, \ldots, \sIndex_1 + \sIndex_2 \} , \quad 
&& \{ \sIndex_1 + \sIndex_2 + 1, \sIndex_1 + \sIndex_2 + 2, \ldots, \sIndex_1 + \sIndex_2 + \sIndex_3 \}, \quad \; \text{etc.} .
\end{alignat}
Then, we define a \emph{special link pattern} to be a link pattern in $\LP_{\Summed_\multii}$ that lacks a turn-back link joining two nodes 
in a common bin of~\eqref{lblocks}, and we denote 
\begin{align} \label{SPDefn}
\smash{\SpecialPattern_\multii\super{s}} 
:= \big\{ \text{special link patterns in $\LP_{\Summed_\multii}\super{s}$} \big\} , \qquad \qquad
\SpecialPattern_\multii := \bigcup_{s \, \in  \, \DefectSet_{\Summed_\multii}} \SpecialPattern_\multii\super{s}
= \big\{ \text{special link patterns in $\LP_{\Summed_\multii}$} \big\}  .
\end{align}
For example, below, the left figure is a special link pattern in 
$\SpecialDiagram_\multii^\multiii$ with $\multii = (2,2,3,2)$ and $\multiii = (2,3,2)$, but
the right figure is not such a link pattern: 
\begin{align} 
\vcenter{\hbox{\includegraphics[scale=0.275]{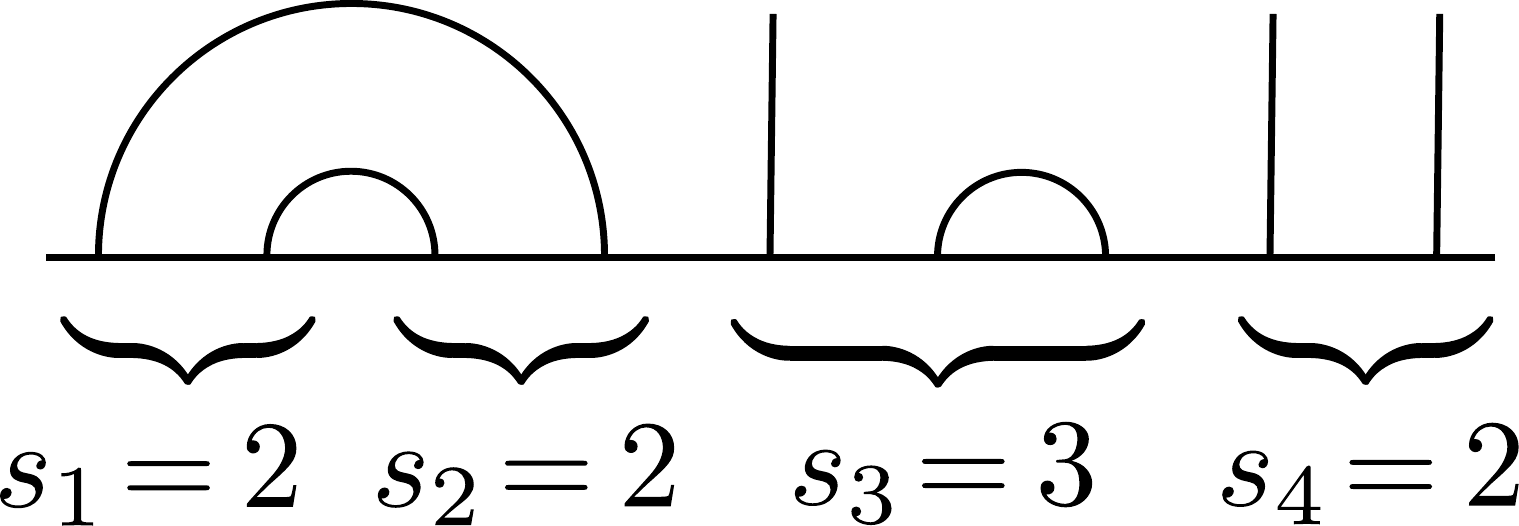}}} 
\qquad \qquad \qquad \qquad
\vcenter{\hbox{\includegraphics[scale=0.275]{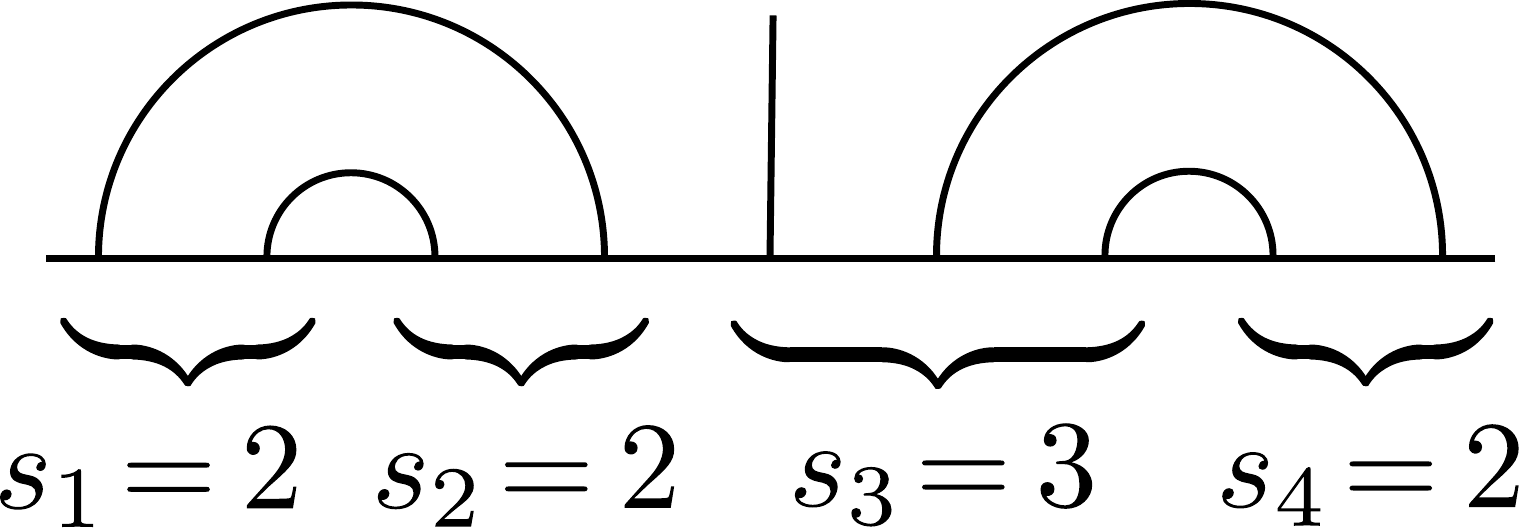} .}} 
\end{align}

\begin{cor} \label{ProjImKerCor}
Suppose $\max \multii < \pmin(q)$. The following hold:
\begin{enumerate}
\itemcolor{red}

\item \label{ProjImageItem}
The following collection is a linearly independent subset of $\im \Projection_\multii$\textnormal{:}
\begin{align} \label{PImageBasis}
\big\{ F^\ell.\Sing_{\alpha} \, \big| \, \textnormal{$\alpha \in \SpecialPattern_\multii\super{s}$, $s \in \DefectSet_\multii$, and $0\leq\ell\leq s  < \pmin(q)$} \big\} .
\end{align}
Furthermore, if $\Summed_\multii < \pmin(q)$, then this set is a basis for $\im \Projection_\multii$.

\item \label{ProjKernelItem}
The following collection is a linearly independent subset of $\ker \Projection_\multii$\textnormal{:}
\begin{align} \label{PKernelBasis}
\big\{ F^\ell.\Sing_\alpha \, \big| \, \textnormal{$\alpha \in \LP_{\Summed_\multii}\super{s} \setminus \SpecialPattern_\multii\super{s}$, $s \in \DefectSet_\multii$, and $0\leq\ell\leq s  < \pmin(q)$} \big\} .
\end{align}
Furthermore, if $\Summed_\multii < \pmin(q)$, then this set is a basis for $\ker \Projection_\multii$.
\end{enumerate}
Similarly, this corollary holds after the symbolic replacements
\begin{align}
\Projection \mapsto \ProjectionBar , \qquad 
F^\ell.\Sing_\alpha \mapsto \SingBar_{\alphaBar}.E^\ell ,\qquad  
\alpha \mapsto \alphaBar , \qquad
\SpecialPattern \mapsto \SpecialPatternBar , 
\qquad \textnormal{and} \qquad 
\LP \mapsto  \LPBar .
\end{align}
\end{cor}

\begin{proof} 
To begin, we collect three relevant facts.
\begin{itemize}[leftmargin=*]
\item[-]
Item~\ref{DirectSumInclusionLem2Item1} of proposition~\ref{MoreGenDecompAndEmbProp2} 
implies that the following collection is a linearly independent subset of $\VecSp_{\Summed_\multii}$:
\begin{align}
\label{LinIndCollection}
\smash{\{ F^\ell.\Sing_\alpha \,|\, \textnormal{$\alpha\in\LP_{\Summed_\multii}\super{s}$, $s \in \DefectSet_{\Summed_\multii}$, and $0\leq\ell\leq s < \pmin(q)$}\}}.
\end{align}
Moreover, proposition~\ref{MoreGenDecompAndEmbProp2}
implies that this set is a basis for $\VecSp_{\Summed_\multii}$ if $\Summed_\multii < \pmin(q)$.

\item[-]
By item~\ref{It1} of lemma~\ref{EmbProjLem}, the subspaces $\im \Projection_\multii$ and $\ker \Projection_\multii$ 
are closed under the $\Uqsltwo$-action on them.

\item[-] 
By~\cite[Lemma~\red{B.1}]{fp3a} the respective
collections $\SpecialPattern_\multii$ and 
$\LP_{\Summed_\multii} \setminus \SpecialPattern_\multii$ are 
bases for $\im \WJProj_\multii (\,\cdot\,)$ and $\ker \WJProj_\multii (\,\cdot\,)$. 
\end{itemize}
Now, corollary~\ref{ProjImageandKernelCor} shows that we have 
$\Sing_\alpha \in \im \Projection_\multii$ if and only if $\alpha \in \im \WJProj_\multii (\,\cdot\,) = \Span \SpecialPattern_\multii$,
and we have
$\Sing_\alpha \in \ker \Projection_\multii$ if and only if $\alpha \in \ker \WJProj_\multii (\,\cdot\,) = \Span \LP_{\Summed_\multii} \setminus \SpecialPattern_\multii$.
These facts imply items~\ref{ProjImageItem} and~\ref{ProjKernelItem}.
\end{proof}

\subsection{Graphical calculus for descendant vectors}
\label{DescGraphSec}

In this section, we collect explicit diagram formulas for $F$-descendants of the link-pattern basis vectors 
$\Sing_\alpha$.  
Lemma~\ref{DescLem} concerns the case of a tensor power of fundamental representations and in 
lemma~\ref{DescLem2}, we establish the general case.
To begin, we list simple $q$-identities.

\begin{lem} 
The following identities hold, for all $q \in \bC^\times$ and $k, \ell, m \in \bZ$\textnormal{:}
\begin{align}
\label{QintegerIdentity1}
[k] & = q^{k-1} + q^{k-3} + \cdots + q^{3-k} + q^{1-k} , \\
\label{QintegerIdentity2}
\qbin{m}{\ell} & = q^{\ell-m} \qbin{m-1}{\ell-1} + q^\ell \qbin{m-1}{\ell} .
\end{align}
\end{lem}

\begin{proof}
The asserted identities are straightforward to verify using definition~\eqref{Qinteger}. 
\end{proof}

The next observation (lemma~\ref{SwitchOrientLem}
and the consequent lemma~\ref{DiagLem} are useful tools in many diagram calculations.

\begin{lem} \label{SwitchOrientLem} 
\textnormal{\cite[page~\red{442}]{fk}} 
Suppose $s < \pmin(q)$. 
Then within any diagram,
we have 
\begin{align} \label{SwitchOrient} 
\vcenter{\hbox{\includegraphics[scale=0.275]{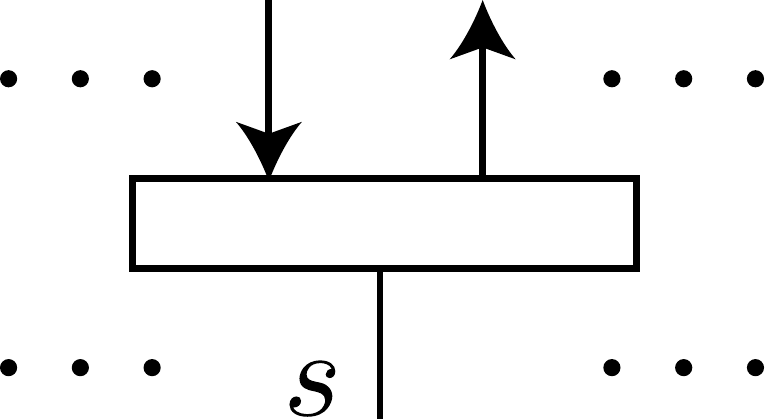}}} \quad
= \quad q \,\, \times \,\, \vcenter{\hbox{\includegraphics[scale=0.275]{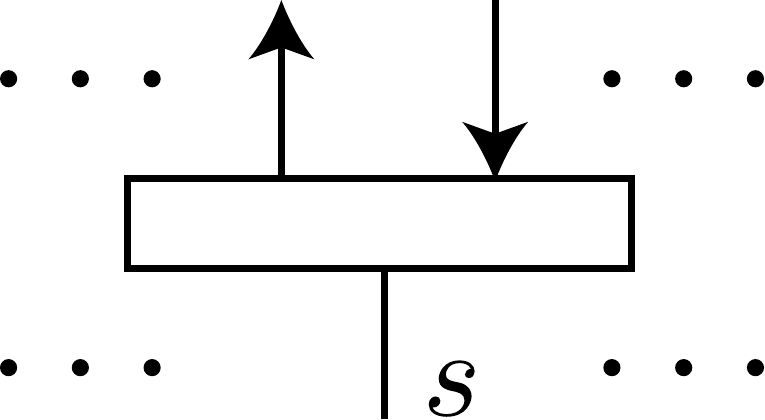} ,}} 
\end{align}
and similarly,
\begin{align} \label{SwitchOrientBar} 
\vcenter{\hbox{\includegraphics[scale=0.275]{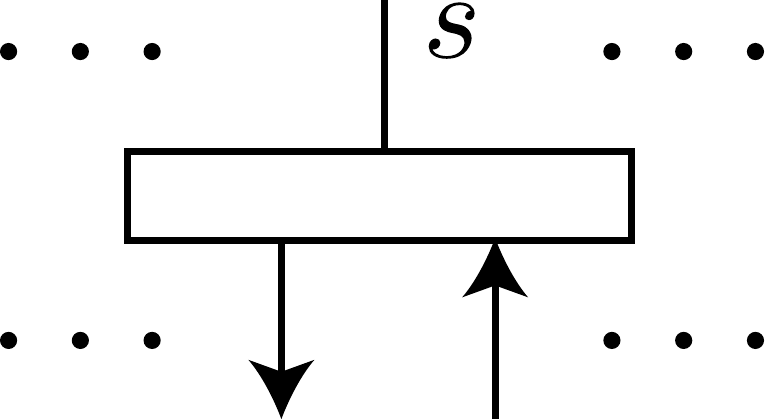}}} \quad
= \quad q \,\, \times \,\, \vcenter{\hbox{\includegraphics[scale=0.275]{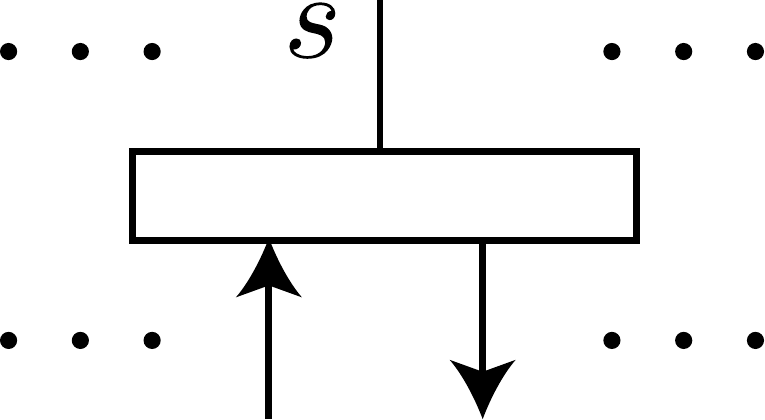} ,}} 
\end{align}
where the ellipses stand for unspecified parts of the link state which are the same on both sides.
\end{lem}

\begin{proof} 
By property~\eqref{ProjectorID2} of the Jones-Wenzl projector, we have
\begin{align}
\vcenter{\hbox{\includegraphics[scale=0.275]{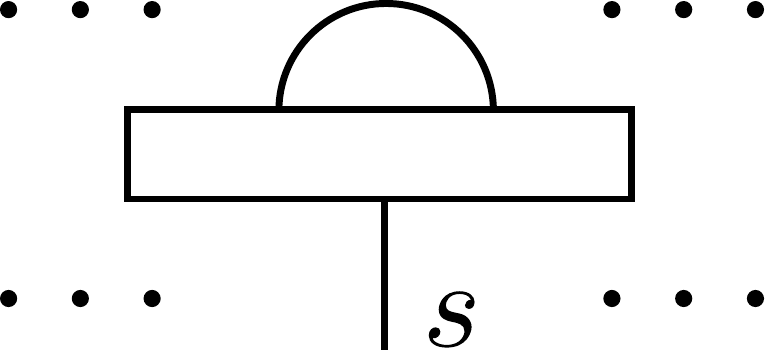}}} \quad = \quad 0 
\quad = \quad \vcenter{\hbox{\includegraphics[scale=0.275]{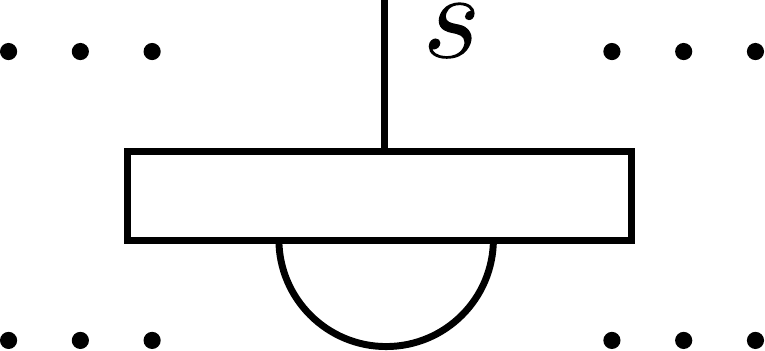} .}} 
\end{align}
After decomposing the link via~(\ref{singletDiagramNotation},~\ref{singletDiagramNotationBar}) and rearranging, 
we arrive with~\eqref{SwitchOrient} and~\eqref{SwitchOrientBar}.
\end{proof}

\begin{lem} \label{DiagLem} 
Suppose $s < \pmin(q)$. Then, for all $k,\ell \in \{0, 1,\ldots,s\}$, we have
\begin{align} \label{DiagResult0} 
\left( \; \vcenter{\hbox{\includegraphics[scale=0.275]{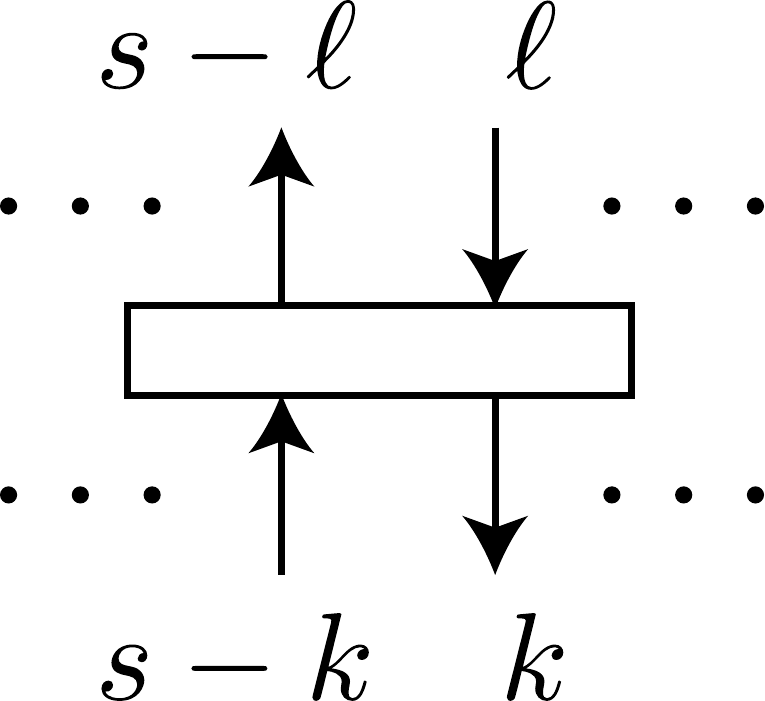}}} \; \right)
\quad = \quad \delta_{k,\ell} \, q^{k(k-s)} \qbin{s}{k}^{-1} \hphantom{.}
\end{align}
and similarly,
\begin{align} \label{DiagResult0Bar} 
\left( \; \vcenter{\hbox{\includegraphics[scale=0.275]{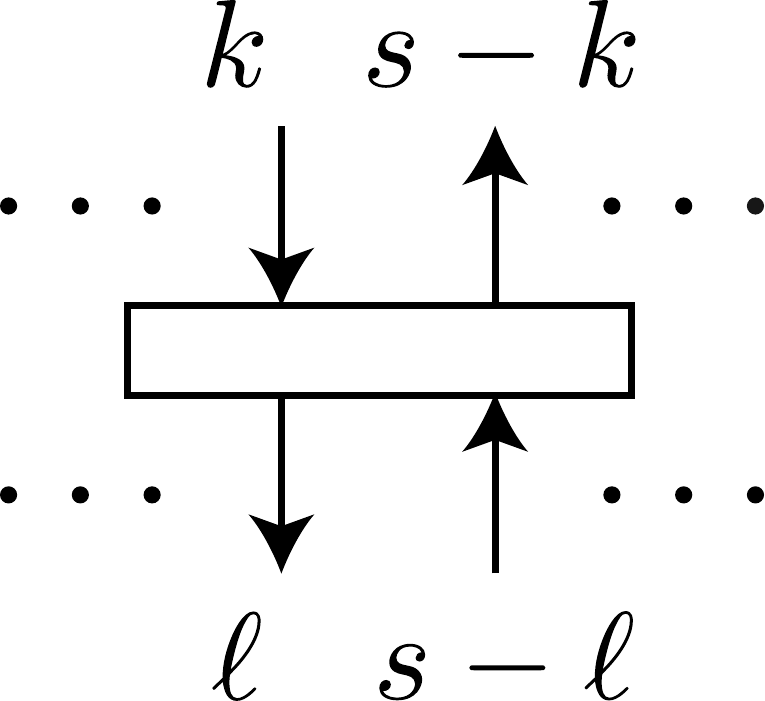}}} \; \right)
\quad = \quad \delta_{k,\ell} \, \qbin{s}{k}^{-1}  . \hphantom{q^{k(k-s)}} 
\end{align}
\end{lem}

%
%
%

\begin{proof} 
Without loss of generality, we assume that $k \geq \ell$.  To evaluate the network of~\eqref{DiagResult0}, 
we replace the projector box in it with the right side of the recursive identity 
\begin{align} \label{UsefulID} 
\vcenter{\hbox{\includegraphics[scale=0.275]{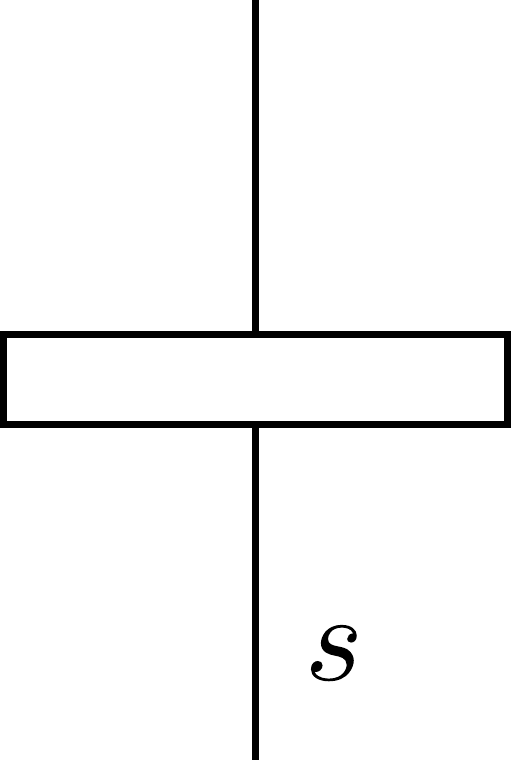}}} \quad
= \quad \vcenter{\hbox{\includegraphics[scale=0.275]{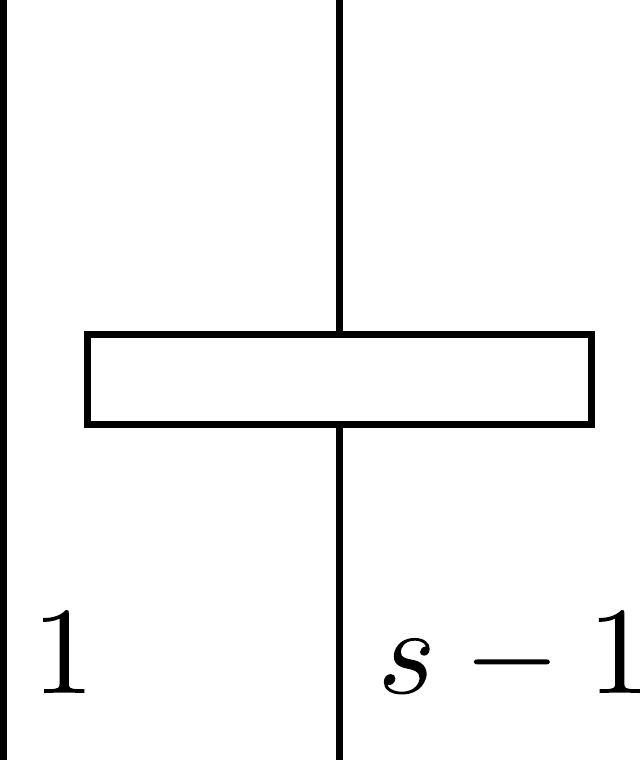}}} \quad
+ \quad \sum_{r \, = \, 0}^{s - 2} \frac{[s - r - 1]}{[s]} \,\, \times \,\,
\vcenter{\hbox{\includegraphics[scale=0.275]{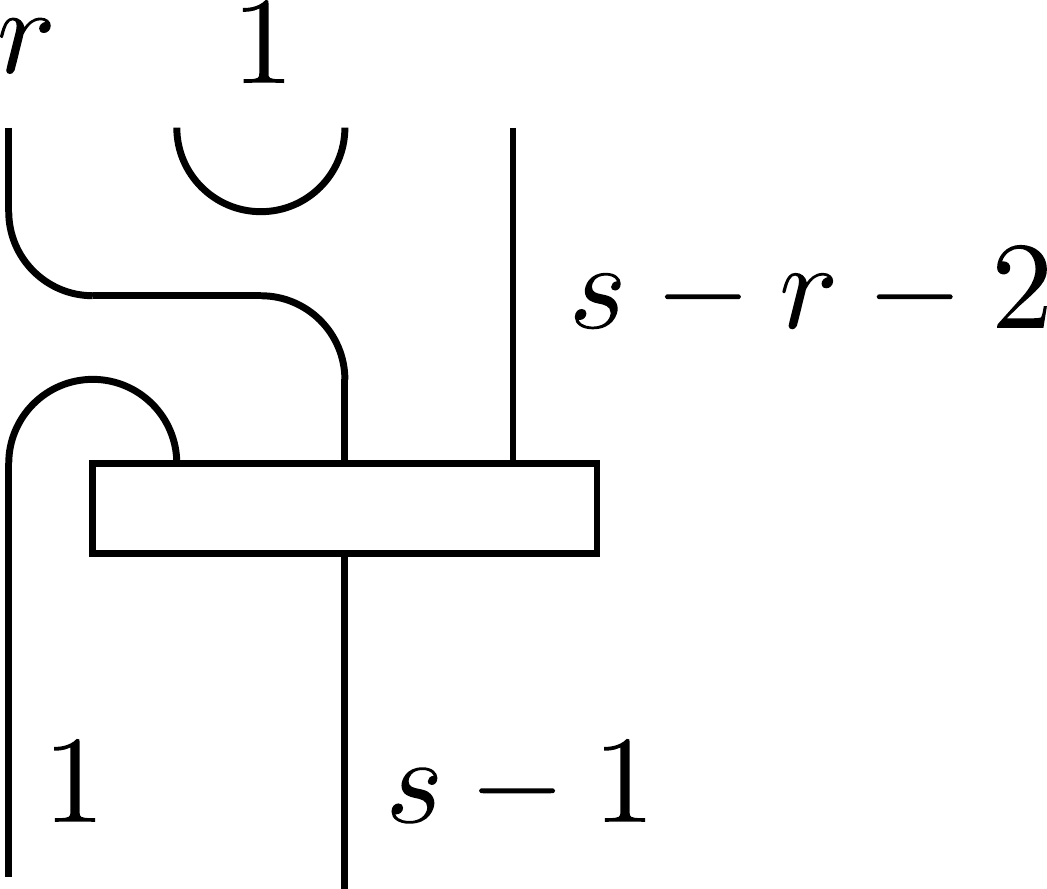} ,}} 
\end{align}
see~\cite[equation~(\red{2.59})]{mrr}.
After inserting this identity into~\eqref{DiagResult0}, 
all but two of the resulting terms have components with clashing orientations that cause them to vanish by rule~\eqref{TurnBackWeightZero}. 
With this observation, we find that
\begin{align} \label{MidStep} 
\vcenter{\hbox{\includegraphics[scale=0.275]{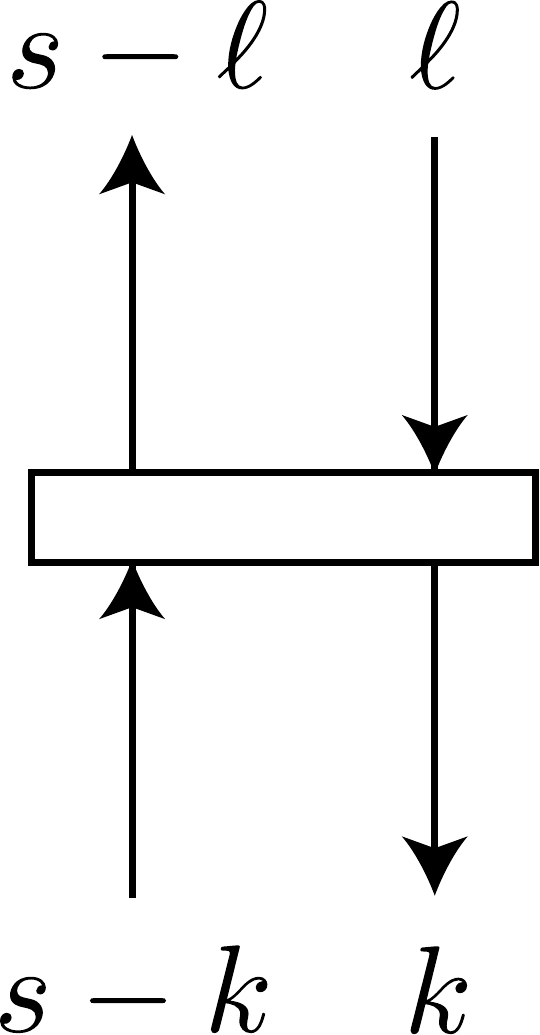}}} \quad
\underset{\eqref{UsefulID}}{\overset{\eqref{DiagResult0}}{=}} \quad \vcenter{\hbox{\includegraphics[scale=0.275]{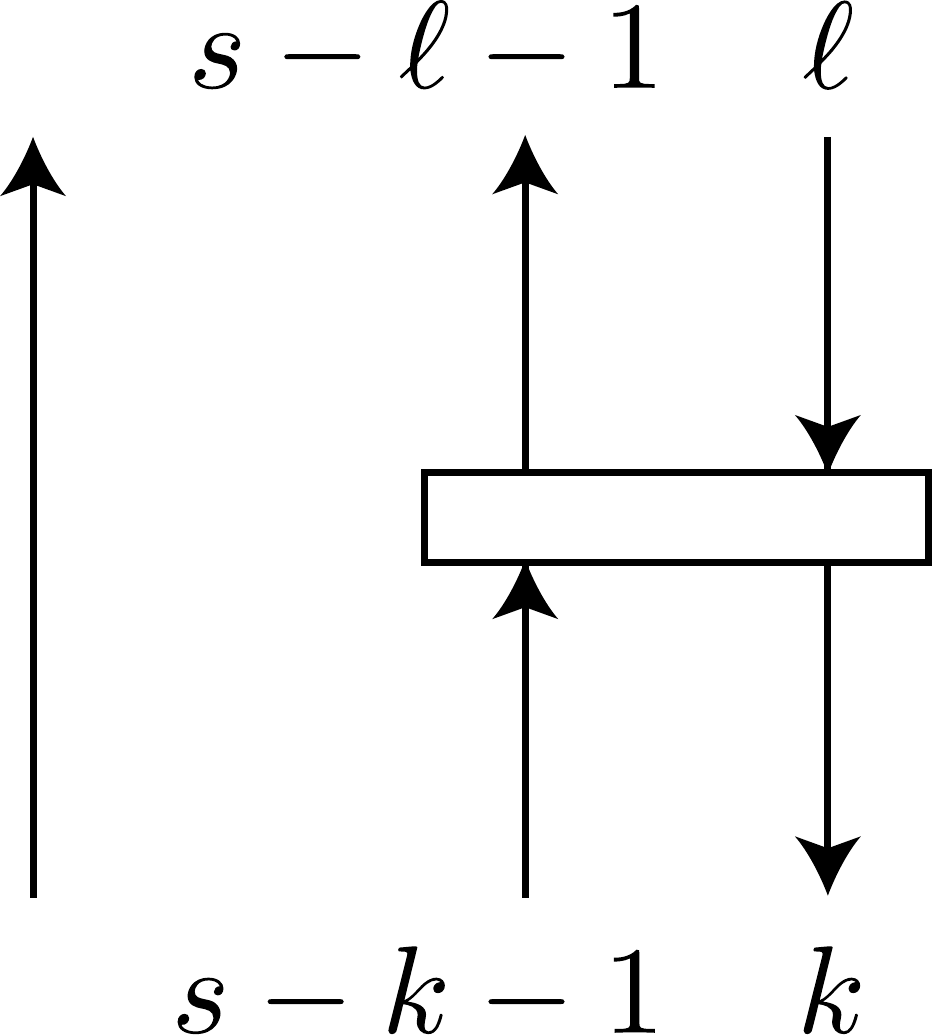}}} \quad
+ \quad \frac{[\ell]}{[s]} \,\, \times \vcenter{\hbox{\includegraphics[scale=0.275]{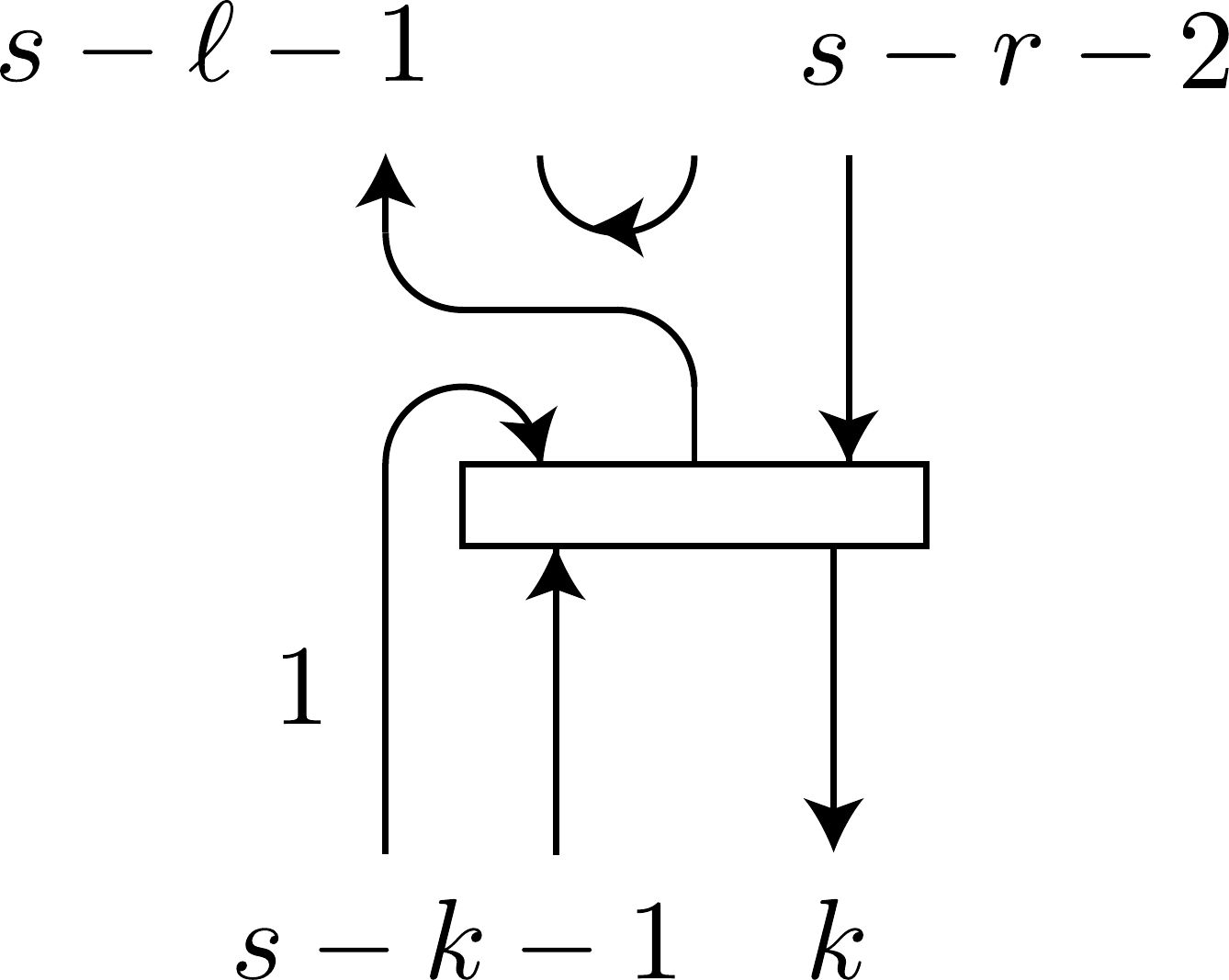} .}}
\end{align}
Replacing the disconnected components 
with their weights via~(\ref{ThroughPathWeight},~\ref{TurnBackWeightCW}) 
and straightening the leftmost defect below the projector box in the right network, 
thus producing a further factor of $\ii q^{1/2}$ according to~\eqref{TurnBackWeightCW},
we obtain
\begin{align} 
\label{MidStep2} 
\vcenter{\hbox{\includegraphics[scale=0.275]{Figures/e-RidoutID0.pdf}}} \quad
& \overset{\eqref{MidStep}}{=} 
\vcenter{\hbox{\includegraphics[scale=0.275]{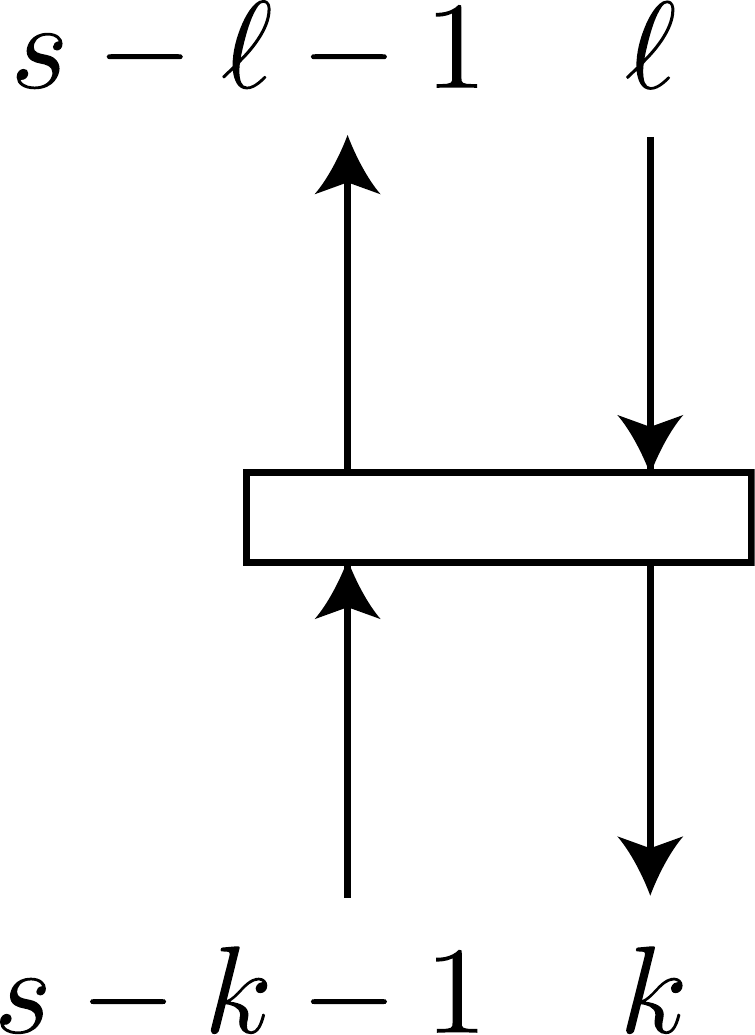}}} 
\quad + \quad  \frac{[\ell]}{[s]} (\ii q^{1/2})^2 \,\, \times \vcenter{\hbox{\includegraphics[scale=0.275]{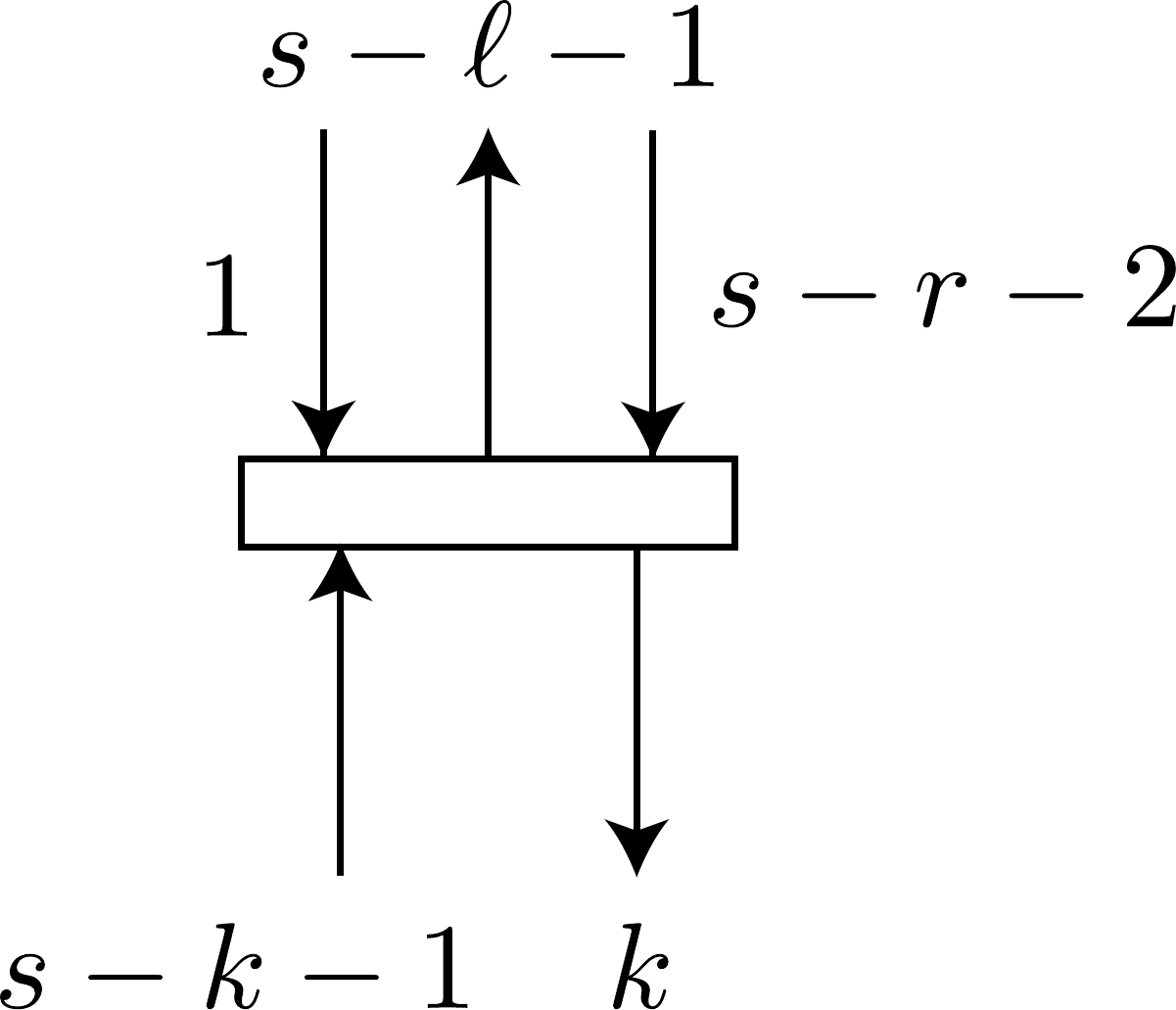}}} \\[1em] 
\label{MidStep3} 
& \overset{\eqref{SwitchOrient}}{=} \quad \left( 1 - \frac{[\ell]}{[s]} q^{s - \ell} \right)
\,\, \times \vcenter{\hbox{\includegraphics[scale=0.275]{Figures/e-RidoutID8.pdf}}} 
\quad = \quad \frac{[s-\ell]}{[s]}q^{-\ell} \,\, \times
\vcenter{\hbox{\includegraphics[scale=0.275]{Figures/e-RidoutID8.pdf} .}}
\end{align}
Repeating this process another $s-k-1$ times, we reduce the original network into one that clearly evaluates to $\delta_{k,l}$:
\begin{align} 
\label{DiagResult} 
\vcenter{\hbox{\includegraphics[scale=0.275]{Figures/e-RidoutID0.pdf}}} \quad
& \overset{\eqref{MidStep3}}{=} \quad \frac{[s - \ell]}{[s]} \, q^{-\ell} \, \frac{[s - \ell - 1]}{[s-1]} \, q^{-\ell} 
\, \dotsm \, \frac{[k - \ell]}{[k+1]} \, q^{-\ell} \,\, \times \,\,
\vcenter{\hbox{\includegraphics[scale=0.275]{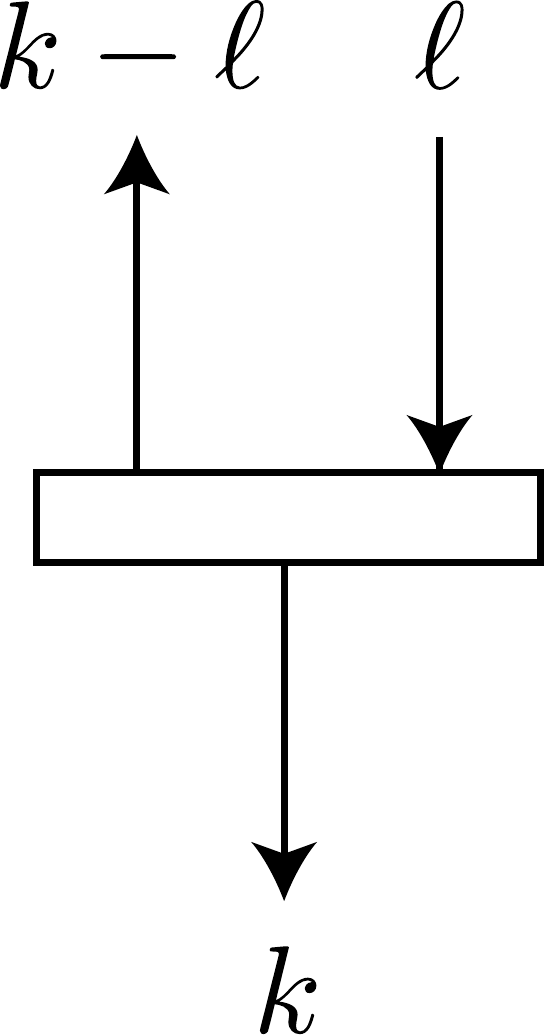}}} \quad
= \quad q^{k(k-s)} \qbin{k}{s}^{-1} \, \delta_{k,\ell} ,
\end{align}
This proves~\eqref{DiagResult0}. 
Using lemma~\ref{SwitchOrientLem} to switch defects with opposite orientations, we obtain~\eqref{DiagResult0Bar} from~\eqref{DiagResult0}.
\end{proof}

The descendants 
of the highest-weight vectors $\smash{\MTbas_0\super{s}}$ and $\smash{\MTbasBar_0\super{s}}$, defined in~(\ref{MTFActhwv},~\ref{MTFActhwvBar})
have explicit diagram formulas. These can be proven by induction.
Below, we give a proof using diagram calculations, which we find more illuminating.

\begin{lem} \label{DescLem} 
Suppose $s < \pmin(q)$.  Then, for all $\ell \in \{0,1,\ldots,s\}$, we have
\begin{align} \label{DescDiagram} 
\MTbas_\ell\super{s} := F^\ell. \MTbas_0\super{s} 
\quad = \quad \frac{[s]!}{[s-\ell]!} \,\, \times \,\, \vcenter{\hbox{\includegraphics[scale=0.275]{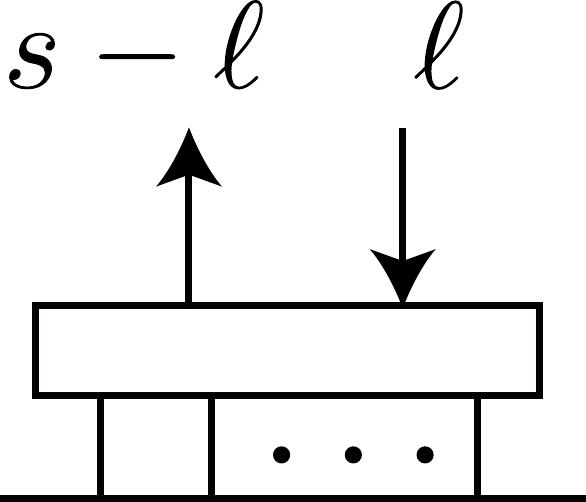} ,}} 
\qquad \qquad 
\FBar^\ell. \MTbas_0\super{s} = \quad \frac{[s]!}{[s-\ell]!} \,\, \times \,\, \vcenter{\hbox{\includegraphics[scale=0.275]{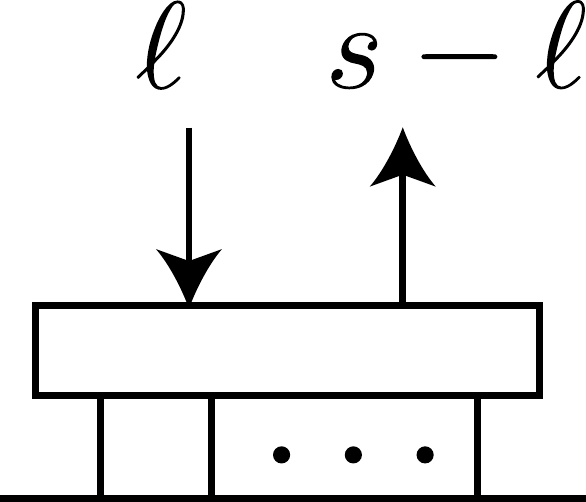} ,}} 
\end{align}
and similarly,
\begin{align} \label{DescDiagramBar} 
\qquad \quad \,
\MTbasBar_\ell\super{s} := \MTbasBar_0\super{s}. E^\ell 
\quad = \quad \frac{[s]!}{[s-\ell]!} \,\, \times \,\, 
\raisebox{-25pt}{\includegraphics[scale=0.275]{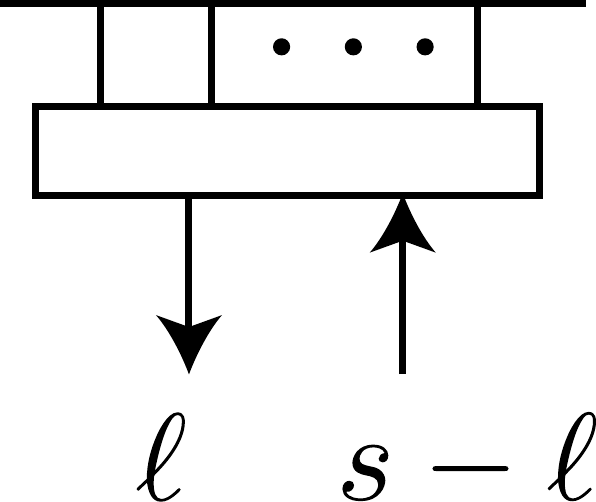} ,}
\qquad \qquad 
\MTbasBar_0\super{s}. \EBar^\ell = \quad \frac{[s]!}{[s-\ell]!} \,\, \times \,\, 
\raisebox{-25pt}{\includegraphics[scale=0.275]{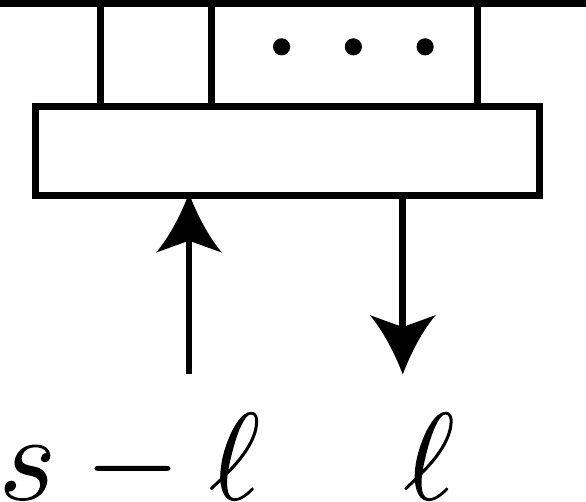} .}
\end{align}
\end{lem}

\begin{proof}
To begin, we show that $\smash{\MTbas_\ell\super{s}}$ and the left diagram in~\eqref{DescDiagram} are proportional.  
For this, we observe these facts:
\begin{itemize}[leftmargin=*]
\item[-] 
Thanks to the projector box, lemma~\ref{CompositeProjCor} implies that the left diagram in~\eqref{DescDiagram} 
lies in the image of $\Projection\sub{s}$.

\item[-] 
With $\ell$ downward-oriented defects and $s - \ell$ upward-oriented defects, the left diagram in~\eqref{DescDiagram} 
lies also in $\smash{\Ksp_s\super{s - 2\ell}}$.

\item[-] 
By definition~\eqref{EmbeddingDef} and item~\ref{It1} of lemma~\ref{EmbProjLem} with~\eqref{sGrading}, we have
\begin{align}
\im \Projection\sub{s} \cap \Ksp_s\super{s - 2\ell} 
\underset{\textnormal{item~\ref{It1}}}{\overset{\textnormal{lem.~\ref{EmbProjLem},}}{=}} 
\im \Embedding\sub{s} \cap \Ksp_s\super{s - 2\ell} 
\overset{\eqref{EmbeddingDef}}{=} \Span \big\{ \MTbas_0\super{s}, \MTbas_1\super{s}, \ldots, \MTbas_s\super{s} \big\} \cap \Ksp_s\super{s - 2\ell} \underset{\eqref{MTbasExplicit}}{\overset{\eqref{sGrading}}{=}} \Span \big\{ \MTbas_\ell\super{s} \big\}.
\end{align}
\end{itemize}
It follows that the left diagram in~\eqref{DescDiagram} equals $\smash{\lambda_\ell\super{s} \MTbas_\ell\super{s}}$ for some constant $\smash{\lambda_\ell\super{s} \in \bC}$. To compute the constant, we let 
\begin{align} 
\overbarStraight{v} = \underbrace{\FundBasisBar_0 \otimes \FundBasisBar_0 \otimes \cdots \otimes \FundBasisBar_0}_{\textnormal{$s - \ell$ times}} \otimes \underbrace{\FundBasisBar_1 \otimes \FundBasisBar_1 \otimes \cdots \otimes \FundBasisBar_1}_{\textnormal{$\ell$ times}},
\end{align}
and we use the explicit formula for $\smash{\MTbas_\ell\super{s}}$ from lemma~\ref{MTbasExplicitLem} to compute that 
\begin{align}
\label{FirstBiFormEval}
\SPBiForm{\overbarStraight{v}}{\lambda_\ell\super{s}\MTbas_\ell\super{s}}
\; \underset{\eqref{biformnormalization}}{\overset{\eqref{MTbasExplicit}}{=}} \; 
\lambda_\ell\super{s} q^{-\ell(s - \ell)} [\ell]!.
\end{align}
On the other hand, according to lemma~\ref{BiFormLem}, the bilinear pairing of the vector $\overbarStraight{v}$ 
with the left diagram in~\eqref{DescDiagram} equals
the evaluation of the network in~\eqref{DiagResult0} of lemma~\ref{DiagLem}. 
Equating this network evaluation with~\eqref{FirstBiFormEval} gives
\begin{align}
\lambda_\ell\super{s} q^{-\ell(s - \ell)} [\ell]! \overset{\eqref{DiagResult0}}{=} q^{-\ell(s - \ell)} \qbin{s}{\ell}^{-1} 
\qquad \qquad \Longrightarrow \qquad \qquad 
\lambda_\ell\super{s} = \frac{[s - \ell]!}{[s]!}.
\end{align}
This completes the proof of the first equality in~\eqref{DescDiagram}.  
We obtain the second equality of~\eqref{DescDiagram} by using lemmas~\ref{MergeLem} and~\ref{SwitchOrientLem} to rewrite the former. 
Identity~\eqref{DescDiagramBar} can be proven similarly. 
\end{proof}

\begin{lem} \label{DescLem2} 
Suppose $s < \pmin(q)$. 
Then, for all valenced link states $\alpha \in \smash{\LS_\multii\super{s}}$ and for all $\ell \in \{0, 1, \ldots, s\}$, we have
\begin{align} \label{DescDiagram2} 
F^\ell.\Sing_\alpha  
\quad = \quad \frac{[s]!}{[s-\ell]!} \,\, \times \,\, 
\vcenter{\hbox{\includegraphics[scale=0.275]{Figures/e-Desc_smaller_new.pdf} ,}} 
\qquad \qquad
\FBar^\ell.\Sing_\alpha  
\quad = \quad \frac{[s]!}{[s-\ell]!} \,\, \times \,\,
\vcenter{\hbox{\includegraphics[scale=0.275]{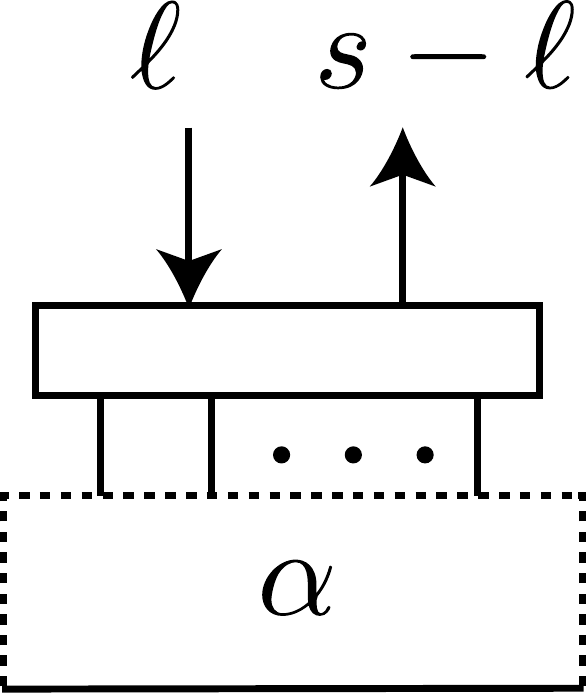} .}} 
\end{align}
Similarly, for all valenced link states $\alphaBar \in \smash{\LSBar_\multii\super{s}}$
and for all $\ell \in \{0, 1, \ldots, s\}$, we have
\begin{align} \label{DescDiagram2Bar} 
\SingBar_{\alphaBar} .E^\ell
\quad = \quad \frac{[s]!}{[s-\ell]!} \,\, \times \,\,
\raisebox{-30pt}{\includegraphics[scale=0.275]{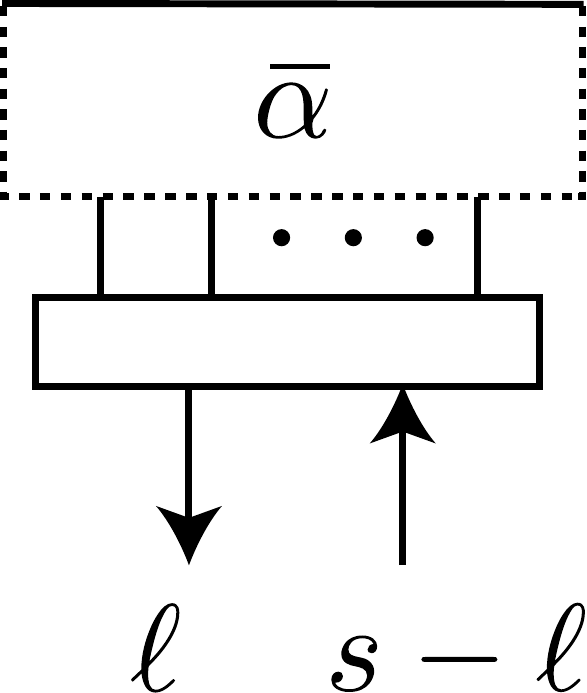} ,} 
\qquad \qquad
\SingBar_{\alphaBar}. \EBar^\ell
\quad = \quad \frac{[s]!}{[s-\ell]!} \,\, \times \,\,
\raisebox{-30pt}{\includegraphics[scale=0.275]{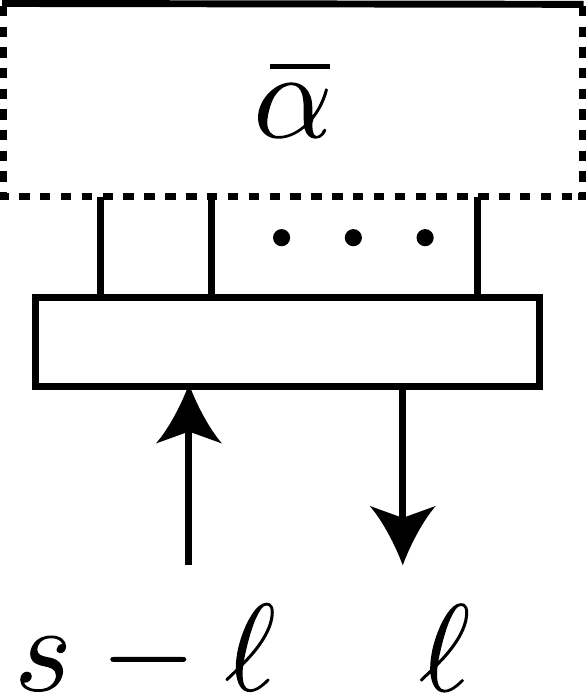} .}
\end{align}
\end{lem} 

\begin{proof} 
We prove the first equality in~\eqref{DescDiagram2}.
As before, we obtain the second equality by using lemmas~\ref{MergeLem} and~\ref{SwitchOrientLem}, 
and~\eqref{DescDiagram2Bar} can be proven similarly. 
First, we assume that $\multii = \OneVec{n}$ for some $n \in \bZpos$. 
By linearity, we may also assume that $\alpha \in \smash{\LP_n\super{s}}$.
Now, there are some indices $i_1,i_2,\ldots, i_k \in \{1,2,\ldots,n-1\}$ such that
\begin{align} 
\label{eqalpha3} 
\alpha & \overset{\hphantom{\eqref{eqalpha3}}}{=}  \Lgen_{i_k} \Lgen_{i_{k-1}} \dotsm \Lgen_{i_2} \Lgen_{i_1} \defects_s 
\qquad \underset{\eqref{SingletBasisDef}}{\overset{\eqref{SingletBasisDefAllDefects}}{=}} \qquad
\Sing_\alpha = \Lgen_{i_k} \Lgen_{i_{k-1}} \dotsm \Lgen_{i_2} \Lgen_{i_1}  \MTbas_0\super{s} ,
\end{align}
so
\begin{align}
\label{thesum2} 
F^\ell.\Sing_\alpha & \overset{\eqref{eqalpha3}}{=}  
F^\ell \big(\Lgen_{i_k} \Lgen_{i_{k-1}} \dotsm \Lgen_{i_2} \Lgen_{i_1}  \MTbas_0\super{s} \big) 
\underset{\eqref{HomProp2}}{\overset{\eqref{MTFActhwv}}{=}} 
\Lgen_{i_k} \Lgen_{i_{k-1}} \dotsm \Lgen_{i_2} \Lgen_{i_1}  \MTbas_\ell\super{s},
\end{align}
where the last equality follows from lemma~\ref{UqHomoLem2}. Then, lemma~\ref{DescLem} shows that
\begin{align} \label{DashedBox} 
F^\ell.\Sing_\alpha \overset{\eqref{thesum2}}{=}
\Lgen_{i_k} \Lgen_{i_{k-1}} \dotsm \Lgen_{i_2} \Lgen_{i_1}  \MTbas_\ell\super{s} 
 \overset{\eqref{DescDiagram}}{=}  
\frac{[s]!}{[s-\ell]!} \Lgen_{i_k} \Lgen_{i_{k-1}} \dotsm \Lgen_{i_2} \Lgen_{i_1} \,\, \times \,\, 
\vcenter{\hbox{\includegraphics[scale=0.275]{Figures/e-Des2_smaller.pdf} ,}}
\end{align}
and by~\eqref{eqalpha3}, 
the right side of~\eqref{DashedBox} indeed equals the right side of~\eqref{DescDiagram2}.

If $\multii \neq \OneVec{n}$, then with
$\smash{\WJEmb_\multii\alpha \in \LS_{\Summed_\multii}\super{s}}$ for any valenced link state $\smash{\alpha \in \LS_\multii\super{s}}$ by~\eqref{LinkEmbDef}, 
we use the above result to write
\begin{align}
\label{DescDiagram3} 
F^\ell.\Sing_\alpha & \; \overset{\eqref{Lmap2}}{=} \; 
F^\ell.\Projectionhat_\multii \big(\Sing_{\WJEmb_\multii \alpha} \big) 
\; \underset{\textnormal{item~\ref{It1}}}{\overset{\textnormal{lem.~\ref{EmbProjLem},}}{=}} \; 
\Projectionhat_\multii \big( F^\ell.\Sing_{\WJEmb_\multii \alpha} \big) \\[1em]
&\; \underset{\eqref{DescDiagram2}}{\overset{\eqref{CompTwoProjsHatEmb}}{=}} \; \;  \frac{[s]!}{[s-\ell]!} \,\, \times \,\, \WJProjHat_\multii \; \; \vcenter{\hbox{\includegraphics[scale=0.275]{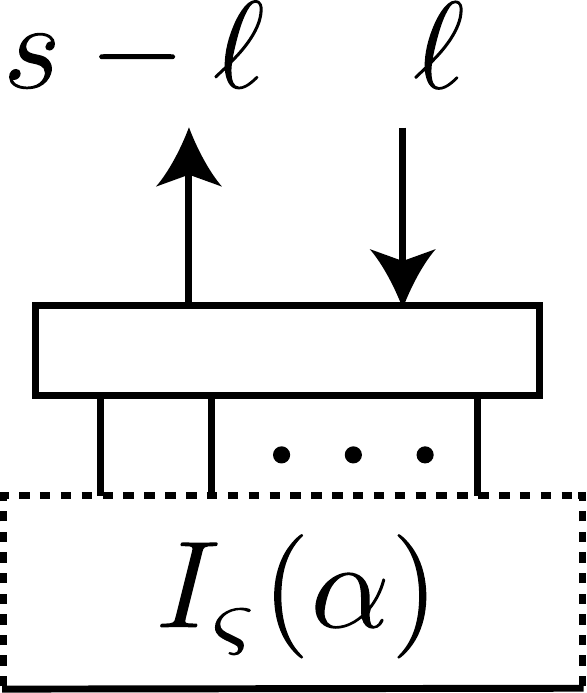}}} 
\; \underset{\hphantom{\eqref{DescDiagram2}}}{\overset{\eqref{IdCompAndWJPhatPEmb}}{=}} 
\; \frac{[s]!}{[s-\ell]!} \,\, \times \,\, \vcenter{\hbox{\includegraphics[scale=0.275]{Figures/e-Desc_smaller_new.pdf} .}}
\end{align} 
This finishes the proof.
\end{proof}

\section{Graphical calculus on type-one modules} 
\label{GraphTLSect}
This section concerns diagram identities and graphical calculus useful both for the analysis of
the $\Uqsltwo$-module structure of $\Module{\VecSp_\multii}{\Uqsltwo}$
and for the $\TL_\multii(\nu)$-module structure of $\CModule{\VecSp_\multii}{\TL}$.
In particular, by lemma~\ref{UqHomoLem2}, the latter diagram action gives rise to homomorphisms of $\Uqsltwo$-modules between different 
modules $\Module{\VecSp_\multii}{\Uqsltwo}$ and $\Module{\VecSp_\multiii}{\Uqsltwo}$.
The main results of the present section give explicit expressions for such homomorphisms
in lemma~\ref{ThisLemma} and proposition~\ref{TLProjectionLem3} (both in section~\ref{GraphicalProjSect}). 
We also find diagram presentations for the conformal-block highest-weight vectors (section~\ref{CoBloGraphical}),
establishing that they are orthogonal (lemma~\ref{CoBloOrthBasisLem}).
We also show in lemma~\ref{GeneratorLemTL} that certain $\Uqsltwo$-submodule projectors generate 
the image of the representation $\Trep_\multii$ of the valenced Temperley-Lieb algebra.
Lastly, sections~\ref{QuotientRadicalSect}--\ref{subsec: radical embedding} concern 
orthocomplements and quotients of the $\Uqsltwo,\UqsltwoBar$-highest-weight vector spaces 
$\HWsp_\multii$ and $\HWspBar_\multii$ with respect to the bilinear pairing~$\SPBiForm{\cdot}{\cdot}$.
In particular, we identify in proposition~\ref{QuotientProp} certain quotients 
with simple $\TL_\multii(\nu)$-modules via 
the link state -- highest-weight vector correspondence (proposition~\ref{HWspLem2}).

\subsection{Graphical calculus for conformal-block vectors}
\label{CoBloGraphical}

In the next auxiliary lemma, we find an explicit sum formula for a valenced link decomposed into oriented defects. 

\begin{lem} \label{SingletFormulaLem}  
Suppose $\max(r,t) < \pmin(q)$. Within any link state with oriented defects, we have
\begin{align} \label{SingletFormula} 
\vcenter{\hbox{\includegraphics[scale=0.275]{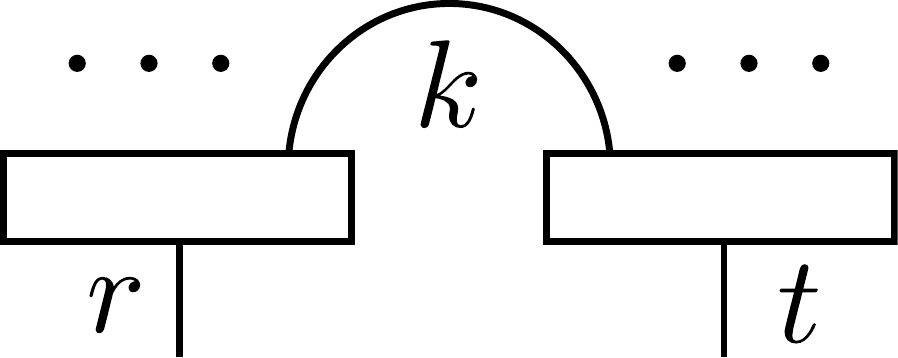}}}
\quad = \quad \sum_{i \, = \, 0}^{k} \sum_{j \, = \, 0}^{k}  \delta_{i + j,k} 
\frac{(-1)^j q^{j(k+1-j)}}{(\ii q^{1/2})^k} \qbin{k}{j} \,\, \times \,\, 
\vcenter{\hbox{\includegraphics[scale=0.275]{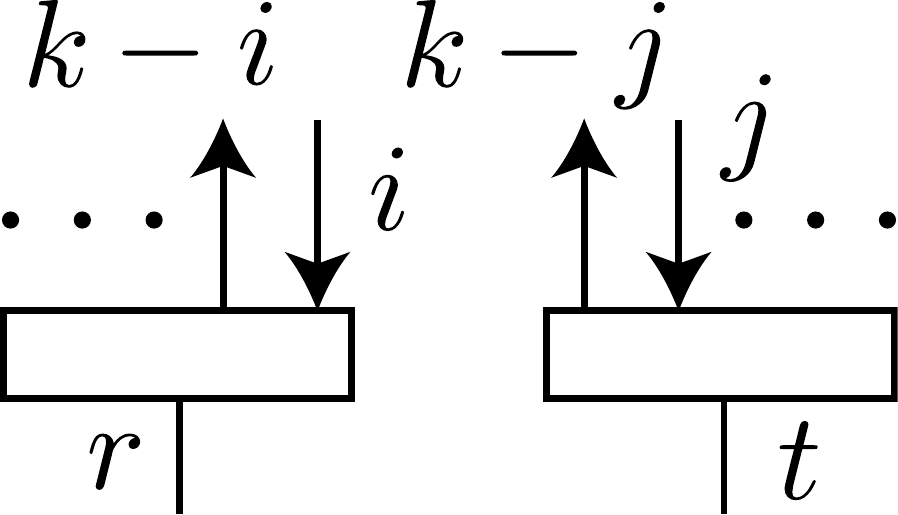}}} 
\end{align}
and similarly,
\begin{align} \label{SingletFormulaBar} 
\vcenter{\hbox{\includegraphics[scale=0.275]{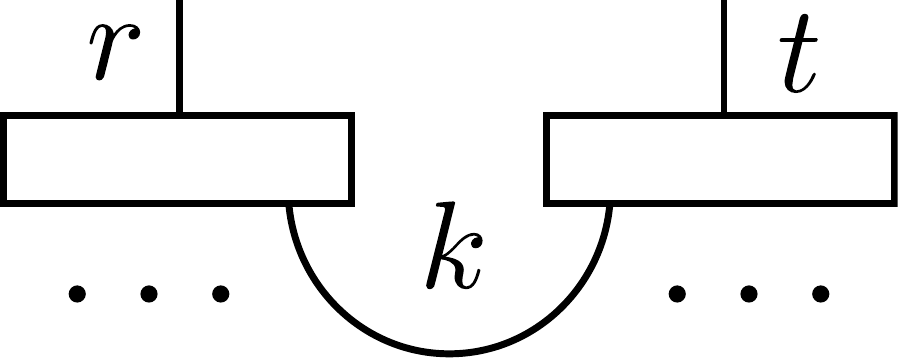}}}
\quad = \quad \sum_{i \, = \, 0}^{k} \sum_{j \, = \,0}^{k}  \delta_{i + j,k} 
\frac{(-1)^i q^{i(i-k-1)}}{(-\ii q^{-1/2})^k} \qbin{k}{i} \,\, \times \,\, 
\vcenter{\hbox{\includegraphics[scale=0.275]{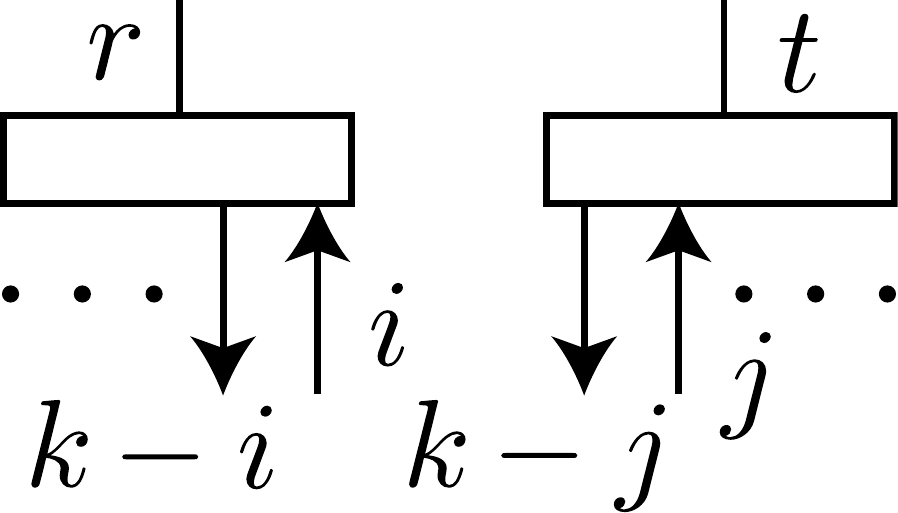}}} 
\end{align}
where the ellipses stand for unspecified parts of the link state which are the same on both sides.
\end{lem}

\begin{proof}
We prove formula~\eqref{SingletFormula} by induction on $k \in \bZnn$. The initial case $k = 0$ is trivial. 
Assuming that formula~\eqref{SingletFormula} for the diagram with $k-1$ links holds, 
we consider the left side of~\eqref{SingletFormula} with $k$ links. 
Decomposing those links gives
\begin{align} \label{DecomInd}
\vcenter{\hbox{\includegraphics[scale=0.275]{Figures/e-HWV17_singlet.pdf}}}
\quad \overset{\eqref{singletDiagramNotation}}{=} \quad 
\ii q^{1/2} \,\, \times \,\,
\vcenter{\hbox{\includegraphics[scale=0.275]{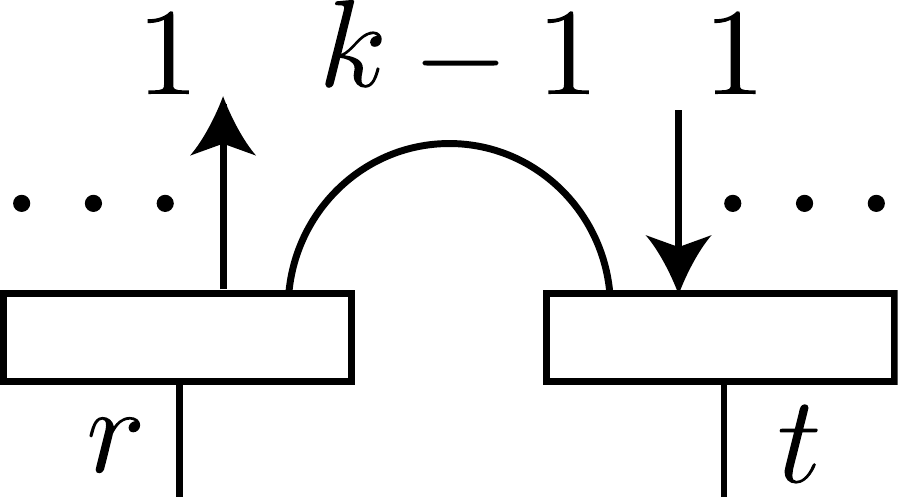}}}
\quad - \quad \ii q^{-1/2} \,\, \times \,\,
\vcenter{\hbox{\includegraphics[scale=0.275]{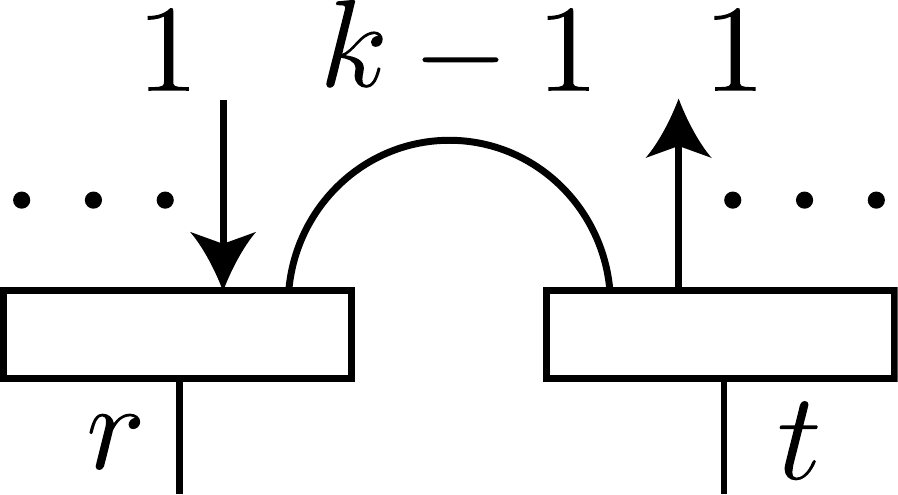} .}}
\end{align}
Next, after applying the induction hypothesis and using identity~\eqref{SwitchOrient}
from lemma~\ref{SwitchOrientLem} to commute the upward-oriented defects to the left of 
the downward-oriented defects in the second diagram of the result, we obtain 
\begin{align} 
\label{DecomIndSum}
\vcenter{\hbox{\includegraphics[scale=0.275]{Figures/e-HWV17_singlet.pdf}}} \quad 
\underset{\eqref{DecomInd}}{\overset{\eqref{SingletFormula}}{=}} & \quad 
\ii q^{1/2} \sum_{i \, = \, 0}^{k-1} \sum_{j \, = \, 0}^{k-1}  \delta_{i + j,k-1} 
\frac{(-1)^j q^{j(k-j)}}{(\ii q^{1/2})^{k-1}} \qbin{k-1}{j} \\[1em]
\nonumber
& \,\, \times \,\, \left( \; \vcenter{\hbox{\includegraphics[scale=0.275]{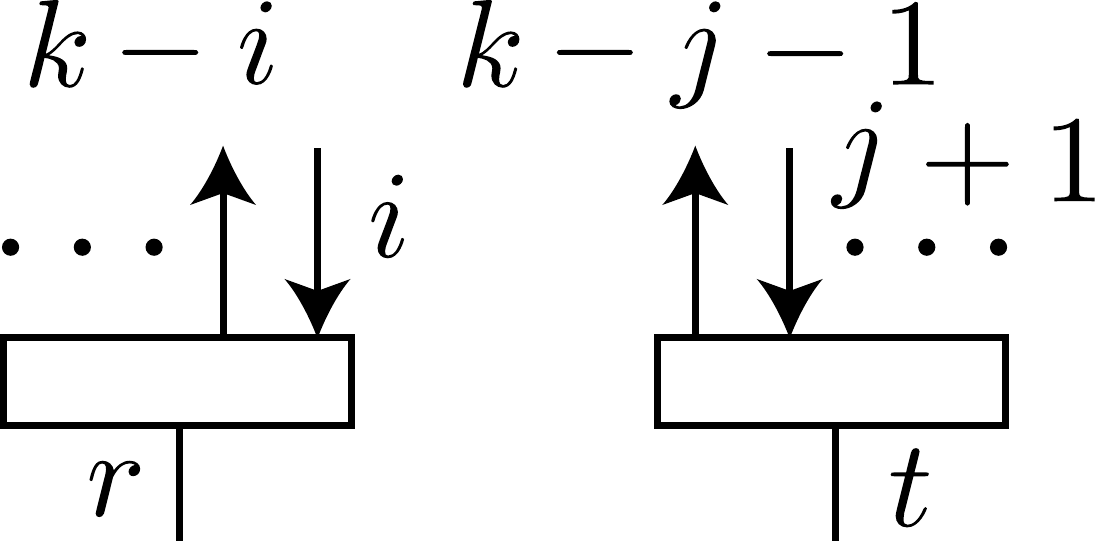}}}
\quad - \quad q^{-1} q^{k-1-i+j} \,\, \times \,\, \vcenter{\hbox{\includegraphics[scale=0.275]{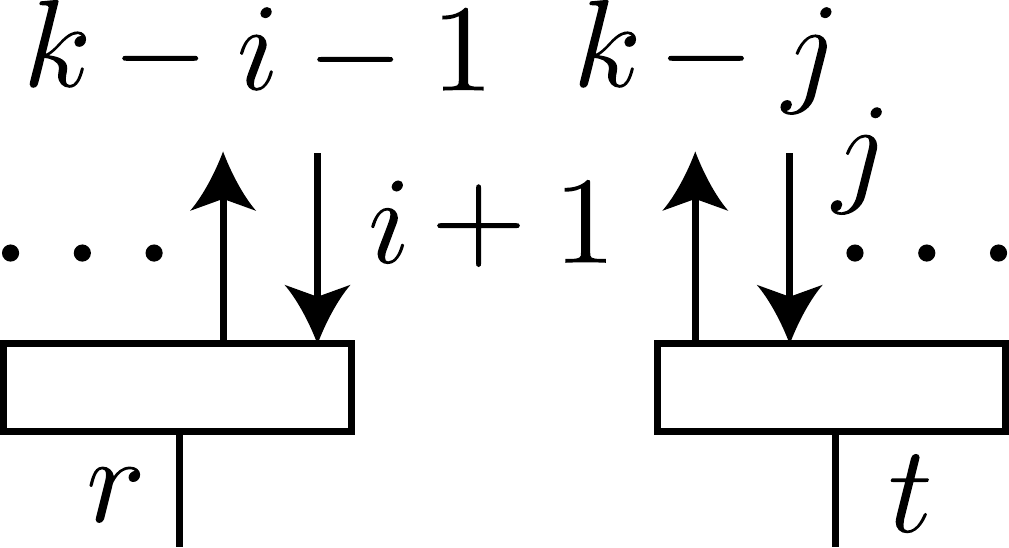}}} \; \right) 
\end{align}
Finally, after changing summation indices by $j+1 \mapsto j$ in the first term and $i+1 \mapsto i$ in the second term,
and using the convention that 
${k-1 \brack -1} = 0$ and ${k-1 \brack k} = 0$,
we simplify the right side of~\eqref{DecomIndSum} into the following form:
\begin{align} 
(\ii q^{1/2})^2 \frac{(-1)^{j-1} q^{j(k+1-j)}}{(\ii q^{1/2})^{k}} q^{-1}
\sum_{i \, = \, 0}^{k} \sum_{j \, = \, 0}^{k}  \delta_{i + j,k} 
\Bigg( q^{j-k} \qbin{k-1}{j-1} \; + \; q^{j} \qbin{k-1}{j} 
\Bigg) \,\, \times \,\, \vcenter{\hbox{\includegraphics[scale=0.275]{Figures/e-HWV17_singlet_twoparts.pdf} .}}
\end{align}
A simplification using identity~\eqref{QintegerIdentity2} now finishes the proof of~\eqref{SingletFormula}. 
Identity~\eqref{SingletFormulaBar} can be proven similarly.
\end{proof}

Using the \emph{oriented closed three-vertex} notation
\begin{align}
\label{BlackVertexOriented}
\text{for $s \in \DefectSet\sub{r,t}$} , \qquad 
& \vcenter{\hbox{\includegraphics[scale=0.275]{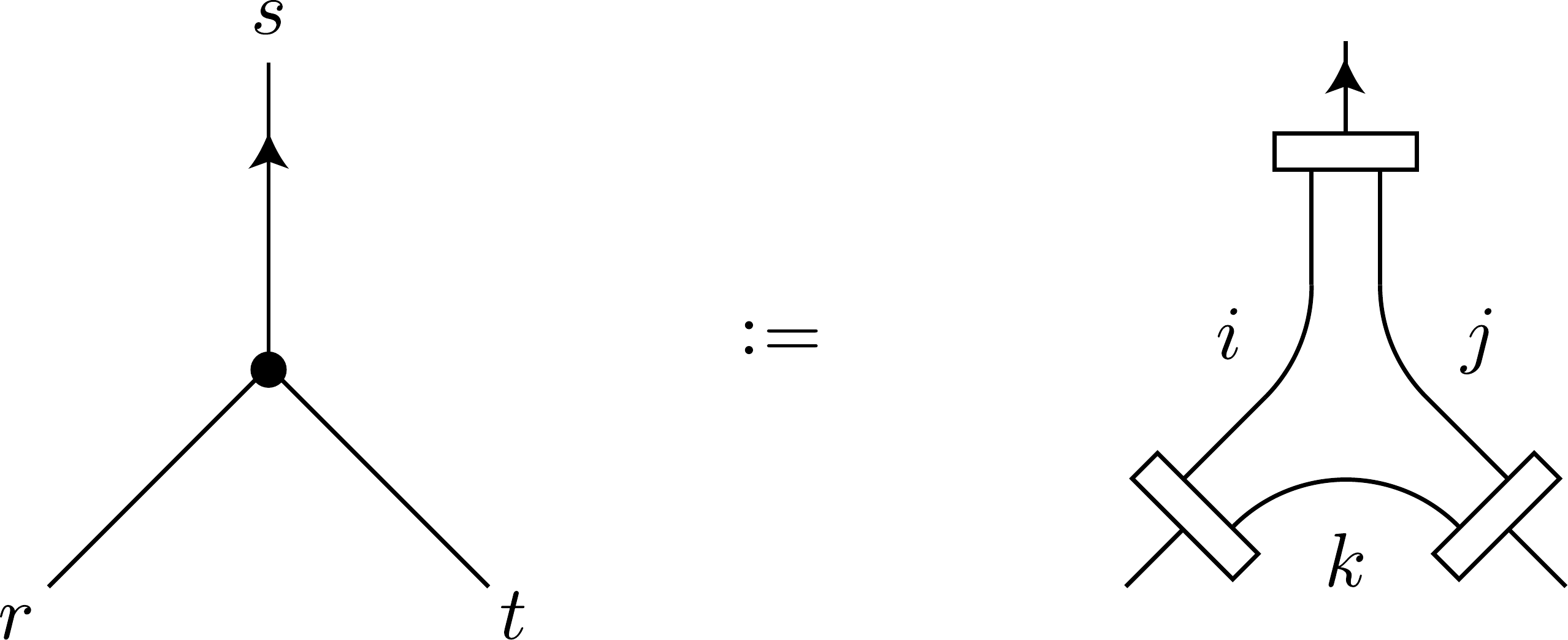} ,}} \qquad  \qquad \qquad
\begin{aligned} 
i & = \frac{r + s - t}{2} , \\[.7em] 
j & = \frac{s + t - r}{2} , \\[.7em] 
k & = \frac{t + r - s}{2} ,
\end{aligned} 
\end{align}
and lemma~\ref{SingletFormulaLem}, we express the conformal-block vectors $\smash{\HWvec\super{s}\sub{r,t}}$
and $\smash{\HWvecBar\super{s}\sub{r,t}}$, defined in~(\ref{tau},~\ref{taubar}), in diagram form.

\begin{lem} \label{HWDiagramLem} 
Suppose $\max(r,t) < \pmin(q)$.  We have 
\begin{align} \label{TWenzl}  
\HWvec\super{s}\sub{r,t} 
\quad = \quad \frac{(\ii q^{1/2})^{\frac{r + t - s}{2}}}{(q - q^{-1})^{\frac{r + t - s}{2}}[\frac{r + t - s}{2}]!} \,\, \times \,\,
\vcenter{\hbox{\includegraphics[scale=0.275]{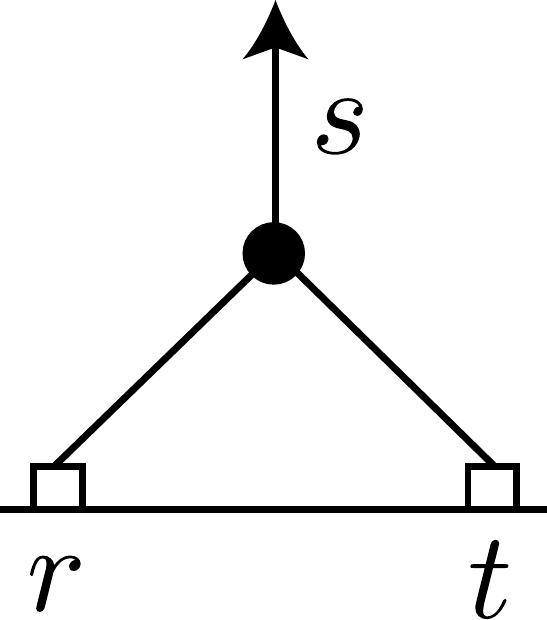}}}  
\end{align}
and
\begin{align} \label{TWenzlBar}  
\HWvecBar\super{s}\sub{r,t} 
\quad = \quad \frac{(-\ii q^{-1/2})^{\frac{r + t - s}{2}}}{(q - q^{-1})^{\frac{r + t - s}{2}}[\frac{r + t - s}{2}]!} \,\, \times \,\,
\raisebox{-20pt}{\includegraphics[scale=0.275]{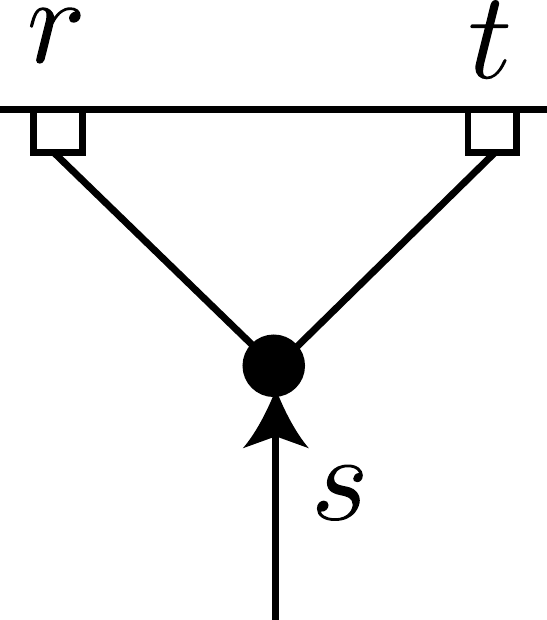}} \; .
\end{align}
\end{lem}

\begin{proof}
We only prove identity~\eqref{TWenzl}; identity~\eqref{TWenzlBar} can be proven similarly.
Definition~\eqref{tau} reads 
\begin{align} \label{TauAgain}
\HWvec\sub{r,t}\super{s} 
:= \sum_{i \, = \, 0}^{\frac{r + t - s}{2}} \sum_{j \, = \, 0}^{\frac{r + t - s}{2}} \delta_{i + j,\frac{r + t - s}{2}} 
\frac{(-1)^{j} q^{j(t + 1 - j)}}{(q - q^{-1})^{(r + t - s) / 2}} 
\frac{[r - i]![t - j]!}{[i]![j]![r]![t]!} F^i.\Basis_0\super{r} \otimes F^j.\Basis_0\super{t} ,
\end{align}
and lemma~\ref{DescLem2} gives diagram representations for the vectors $F^i.\Basis_0\super{r}$ and $F^j.\Basis_0\super{t}$:
\begin{align} \label{OneTermDiag}
F^i.\Basis_0\super{r} \otimes F^j.\Basis_0\super{t}
\quad = \quad \frac{[r]!}{[r-i]!} \frac{[t]!}{[t-j]!} \, \, \times \,\, 
\vcenter{\hbox{\includegraphics[scale=0.275]{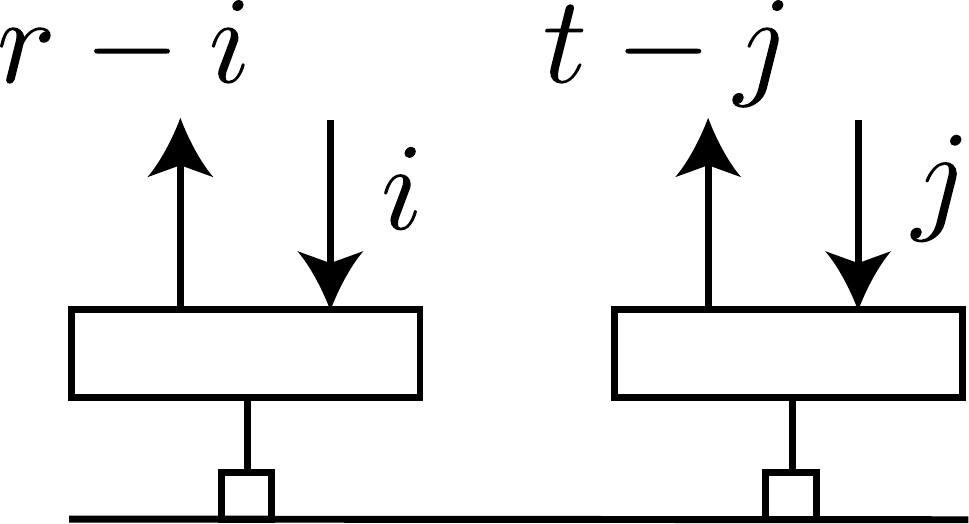} .}}
\end{align}
On the other hand, using~\eqref{BlackVertexOriented} 
we can write the diagram on the right side of~\eqref{TWenzl} in the following form:
\begin{align} \label{OneLinkPattern}
\vcenter{\hbox{\includegraphics[scale=0.275]{Figures/e-HWV16_valenced3.pdf}}} 
\quad \underset{\eqref{TurnBackWeightZero}}{\overset{\eqref{BlackVertexOriented}}{=}} \quad 
\vcenter{\hbox{\includegraphics[scale=0.275]{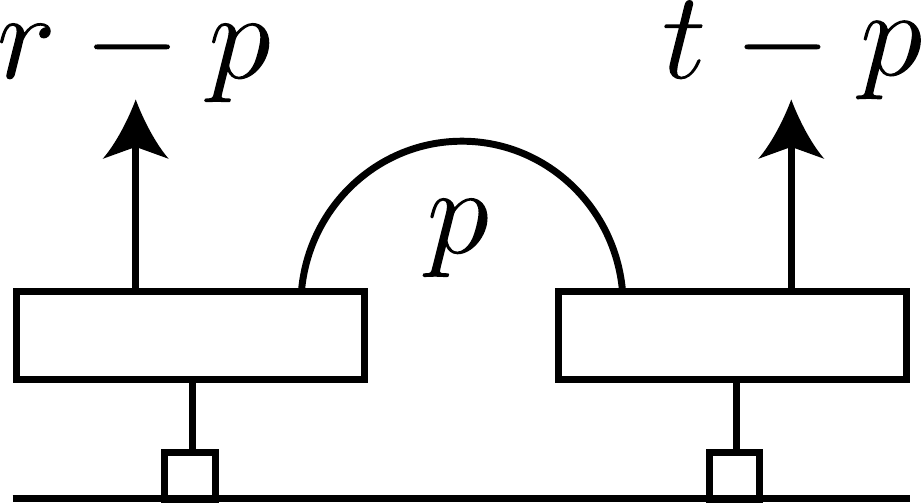} ,}} \qquad \qquad
\text{where} \quad p = \frac{r + t - s}{2} ,
\end{align}
where we removed the top projector box thanks to the following simple observation:
each internal link diagram of the top projector box with a turn-back link would have
clashing orientations and thus weight zero by~\eqref{TurnBackWeightZero}.

Next, we use lemma~\ref{SingletFormulaLem} to expand the right side of~\eqref{OneLinkPattern} into a sum of
link patterns only  comprising oriented defects, and lemma~\ref{SwitchOrientLem} 
to commute the upward-oriented defects to the left of downward-oriented defects, arriving with
\begin{align} \label{AllTermsDiag}
\vcenter{\hbox{\includegraphics[scale=0.275]{Figures/e-HWV17_valenced2_new.pdf}}}
\quad \underset{\eqref{SingletFormula}}{\overset{\eqref{SwitchOrient}}{=}} \quad 
\sum_{i \, = \, 0}^{p} \sum_{j \, = \, 0}^{p}  \delta_{i + j,p} 
\frac{(-1)^j q^{j(p+1-j) + j(t-p)}}{(\ii q^{1/2})^p} \qbin{p}{j} \,\, \times \,\, 
\vcenter{\hbox{\includegraphics[scale=0.275]{Figures/e-HWV17_valenced_twoparts.pdf} ,}}
\end{align}
Recalling definition~\eqref{Qinteger} of the $q$-binomial coefficients,
a comparison of~\eqref{AllTermsDiag} with~(\ref{TauAgain},~\ref{OneTermDiag}) gives~\eqref{TWenzl}.
\end{proof}

Recalling definition~\eqref{GenHWV}, we also record the following more general result.

\begin{lem} \label{FinalDiagLemma} 
Suppose $\max(\varrho, \vartheta, r, t) < \pmin(q)$.
Then, for all valenced link states $\alpha \in \smash{\LS_\varrho\super{r}}$ and $\beta \in \smash{\LS_\vartheta\super{t}}$, we have
\begin{align} \label{FinalDiag} 
\HWvecMap\sub{r,t}\super{s}(\Sing_\alpha \otimes \Sing_\beta)  
\quad = \quad \frac{(\ii q^{1/2})^{\frac{r + t - s}{2}}}{(q - q^{-1})^{\frac{r + t - s}{2}}[\frac{r + t - s}{2}]!} 
\,\, \times \,\, \vcenter{\hbox{\includegraphics[scale=0.275]{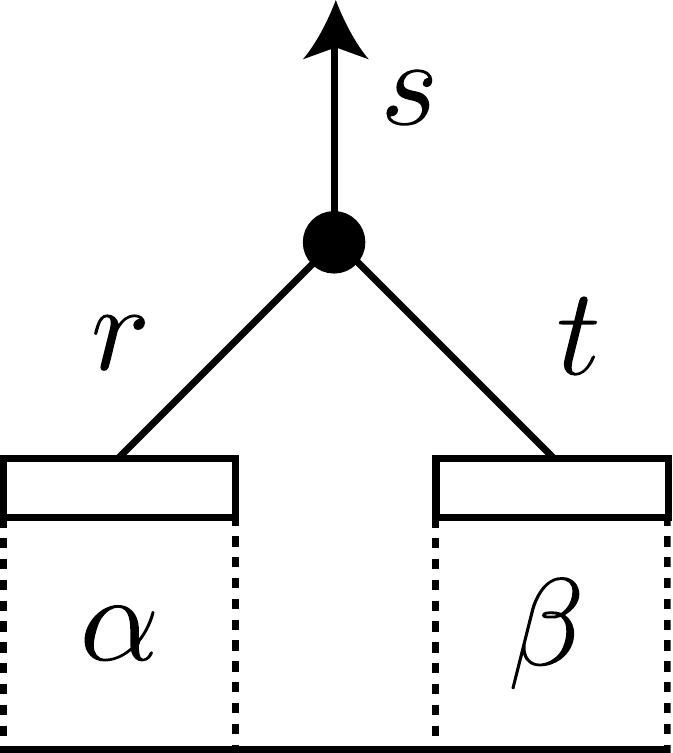} .}} 
\end{align}
Similarly, for all valenced link states $\alphaBar \in \smash{\LSBar_\varrho\super{r}}$ and
$\betaBar \in \smash{\LSBar_\vartheta\super{t}}$, we have
\begin{align} \label{FinalDiagBar} 
\HWvecMapBar\sub{r,t}\super{s}(\SingBar_{\alphaBar} \otimes \SingBar_{\betaBar})  
\quad = \quad \frac{(-\ii q^{-1/2})^{\frac{r + t - s}{2}}}{(q - q^{-1})^{\frac{r + t - s}{2}}[\frac{r + t - s}{2}]!} 
\,\, \times \,\, \raisebox{-40pt}{\includegraphics[scale=0.275]{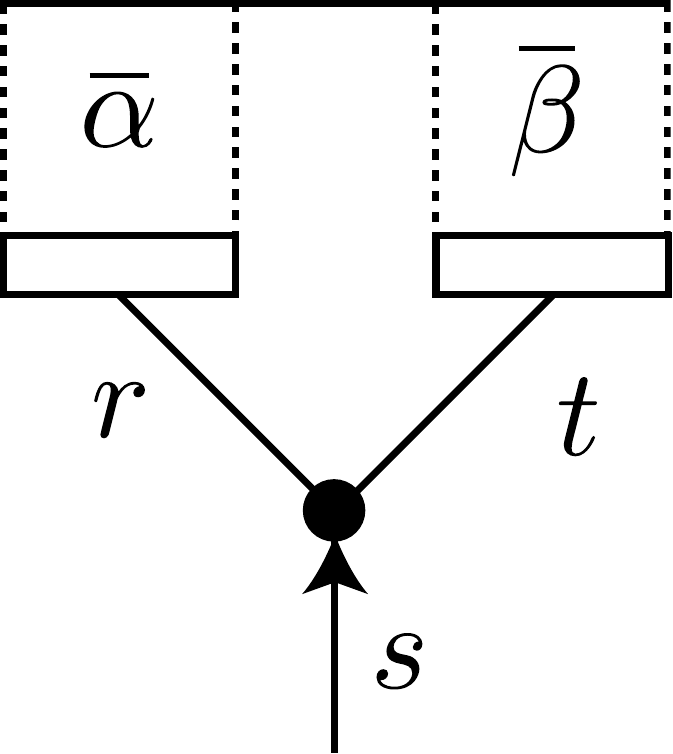}} \; \, .
\end{align}
\end{lem}

\begin{proof} 
The proof of this lemma is almost identical to the proof of lemma~\ref{HWDiagramLem}. Alternatively, one can use lemma~\ref{HWDiagramLem} 
to prove this lemma in a manner very similar to the use of lemma~\ref{DescLem} to prove lemma~\ref{DescLem2} in section~\ref{DescGraphSec}.
\end{proof}

For the general conformal-block vectors $\smash{\HWvec^{\varrho'}_{\multii}}$
and $\smash{\HWvecBar^{\varrho}_{\multii}}$
defined in~(\ref{ConfBlockDefn},~\ref{ConfBlockDefnBar}), we obtain explicit, simple diagram expressions.

\begin{lem} \label{EmbTheta2Lem} 
Suppose $\max \multii < \pmin(q)$.  For all walks $\varrho = (r_1, r_2, \ldots, r_{\np_\multii})$ over $\multii$ 
with $\max \hat{\varrho} = \max (r_1, \ldots, r_{\np_\multii-1})  < \pmin(q)$, we have
\begin{align} \label{EmbTheta2} 
\HWvec^{\varrho}_{\multii} \quad = \quad
\Bigg( \prod_{j \, = \, 1}^{\np_\multii - 1} \frac{(\ii q^{1/2})^{\frac{r_j + \sIndex_{j+1} - r_{j+1}}{2}}}{(q - q^{-1})^{\frac{r_j + \sIndex_{j+1} - r_{j+1}}{2}}[\frac{r_j + \sIndex_{j+1} - r_{j+1}}{2}]!} \Bigg)
\,\, \times
\hspace*{-5mm}
\vcenter{\hbox{\includegraphics[scale=0.275]{Figures/e-Trivalent_link_state2_with_arrow.pdf}}}  
\end{align}
and
\begin{align} \label{EmbTheta2Bar} 
\HWvecBar^{\varrho}_{\multii} \quad = \quad 
\Bigg( \prod_{j \, = \, 1}^{\np_\multii - 1} \frac{(-\ii q^{-1/2})^{\frac{r_j + \sIndex_{j+1} - r_{j+1}}{2}}}{(q - q^{-1})^{\frac{r_j + \sIndex_{j+1} - r_{j+1}}{2}}[\frac{r_j + \sIndex_{j+1} - r_{j+1}}{2}]!} \Bigg)
\,\, \times
\hspace*{-5mm}
\raisebox{-50pt}{\includegraphics[scale=0.275]{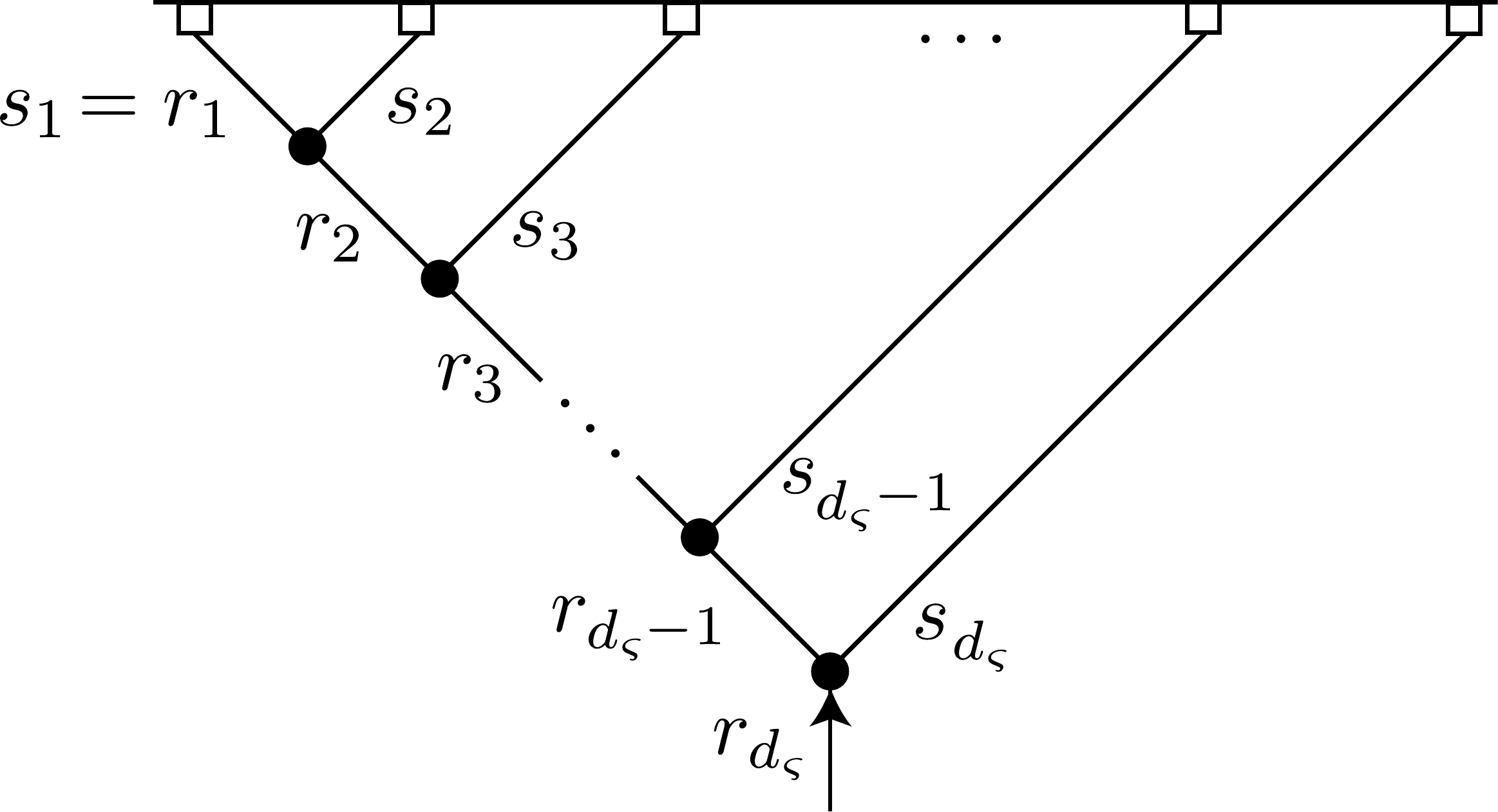}} \; .
\end{align}
\end{lem}

\begin{proof} 
We prove identity~\eqref{EmbTheta2} by induction on $\np_\multii \in \bZpos$;
identity~\eqref{EmbTheta2Bar} is similar.
The initial case $\np_\multii = 1$ is trivial.  
Next, we suppose~\eqref{EmbTheta2} holds for $\np_\multii = d - 1$ for some integer $d \geq 2$.  After inserting the substitutions
\begin{align}
\Sing_\alpha = \HWvec^{\hat{\varrho}}_{\hat{\multii}},\qquad 
\Sing_\beta = \Basis_0\super{\sIndex_{d}}, \qquad r = r_{d - 1}, \qquad t = \sIndex_{d}, \qquad s = r_{d} 
\end{align}
into~\eqref{FinalDiag}, with $\hat{\varrho} := (r_1,r_2,\ldots,r_{d-1})$ and $\hat{\multii} := (\sIndex_1,\sIndex_2,\ldots,\sIndex_{d-1})$
(so $\alpha$ is given by~\eqref{EmbTheta2} with $\np_\multii = d-1$ and 
$\beta = \defectsVal_{\sIndex_d}$ is the link pattern with $\sIndex_d$ defects attached to one node), 
we obtain~\eqref{EmbTheta2} for $\np_\multii = d$. 
\end{proof}

Using the diagram simplification formula
\begin{align}\label{LoopErasure1} 
 \left( \; \vcenter{\hbox{\includegraphics[scale=0.275]{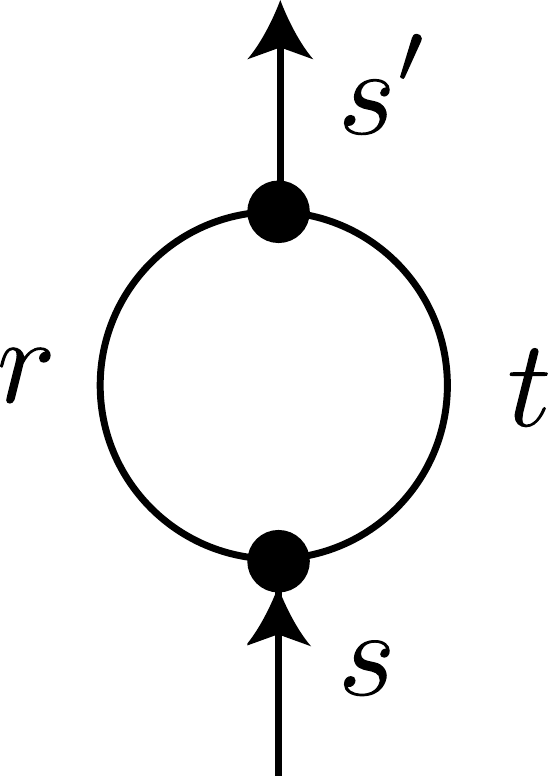}}} \; \right) \quad 
= \quad \delta_{s, s'} \frac{ \ThetaNet(r,s,t) }{(-1)^s [s+1]} ,
\end{align} 
we calculate the bilinear pairing of the conformal-block vectors $\smash{\HWvec^{\varrho'}_{\multii}}$
and $\smash{\HWvecBar^{\varrho}_{\multii}}$, finding that they are orthogonal:

\begin{lem} \label{CoBloOrthBasisLem}  
Suppose $\max \multii < \pmin(q)$.  
For any walks $\varrho = (r_1, r_2, \ldots, r_{\np_\multii})$ and $\varrho'$ over $\multii$ with $\max (\hat{\varrho}, \hat{\varrho}') < \pmin(q)$, we have
\begin{align} \label{WalkBiForm2} 
\SPBiForm{\HWvecBar^{\varrho}_{\multii}}{\HWvec^{\varrho'}_{\multii}} 
= \delta_{\varrho, \varrho'} \prod_{j \, = \, 1}^{\np_\multii - 1}  \frac{\ThetaNet( r_j, r_{j+1}, \sIndex_{j+1} )}{(q - q^{-1})^{r_j + \sIndex_{j+1} - r_{j+1}} \big[ \frac{r_j + \sIndex_{j+1} - r_{j+1}}{2} \big]!^2 \, [r_{j + 1}+1]} .
\end{align}
\end{lem}

\begin{proof} 
Using (\ref{EmbTheta2},~\ref{EmbTheta2Bar}) and lemma~\ref{BiFormLem}, we have
\begin{align} 
\hspace*{-6mm}
\SPBiForm{\HWvecBar^{\varrho}_{\multii}}{\HWvec^{\varrho'}_{\multii}} 
\underset{\textnormal{(\ref{EmbTheta2}, \ref{EmbTheta2Bar})}}{\overset{\eqref{WeightProd}}{=}}
\Bigg( \prod_{j \, = \, 1}^{\np_\multii - 1} \frac{(-\ii q^{-1/2})^{\frac{r_j + \sIndex_{j+1} - r_{j+1}}{2}}}{(q - q^{-1})^{\frac{r_j + \sIndex_{j+1} - r_{j+1}}{2}}[\frac{r_j + \sIndex_{j+1} - r_{j+1}}{2}]!} \Bigg)
\,\, \times \,\,
\vcenter{\hbox{\includegraphics[scale=0.275]{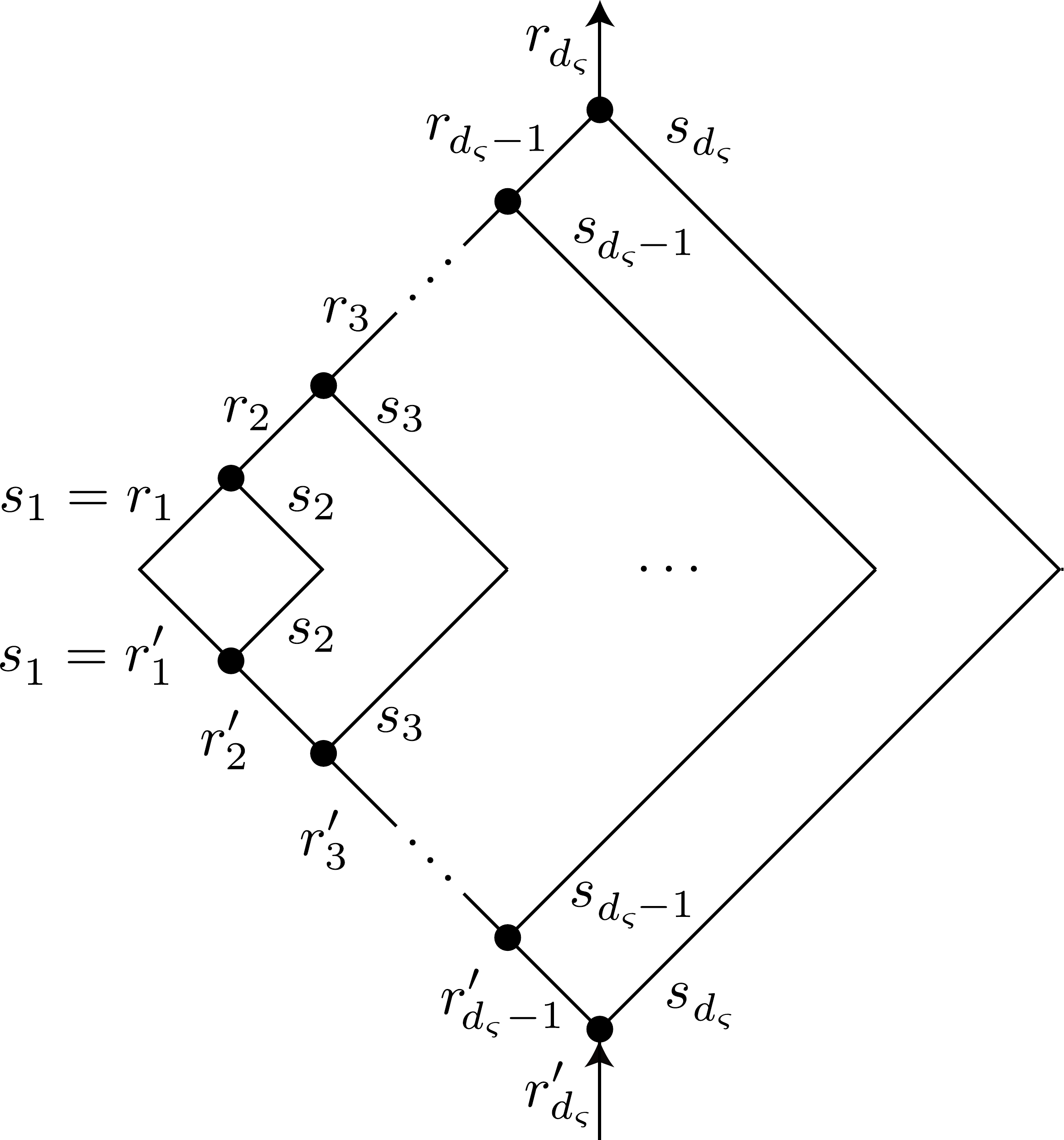} .}} 
\end{align}
Then, after using simplification~\eqref{LoopErasure1} $\np_\multii - 1$ times to delete each loop in the network, we arrive with~\eqref{WalkBiForm2}.
\end{proof}

Next, we state a corollary to lemma~\ref{EmbTheta2Lem}, relating the conformal-block vectors to the link-pattern basis vectors.
For this purpose, generalizing the walk representations of link patterns 
already encountered in~\eqref{alphaWalkNN} in section~\ref{subsec: link state hwv for n}, 
we recall from~\cite[section~\red{4}]{fp3a} that, for any valenced link pattern $\alpha \in \LP_\multii$,
the \emph{walk representation} of $\alpha$ is a walk
\begin{align}\label{alphaWalk} 
\varrho_\alpha = (r_1, r_2, \ldots, r_{\np_\multii})
\end{align}
over $\multii = (\sIndex_1, \sIndex_2, \ldots, \sIndex_{\np_\multii})$ as in~\eqref{WalkHeights}, 
whose heights $r_j$ 
are determined by the links and defects in $\alpha$ as in
\begin{align}\label{Generic} 
\alpha \quad = \quad 
\vcenter{\hbox{\includegraphics[scale=0.275]{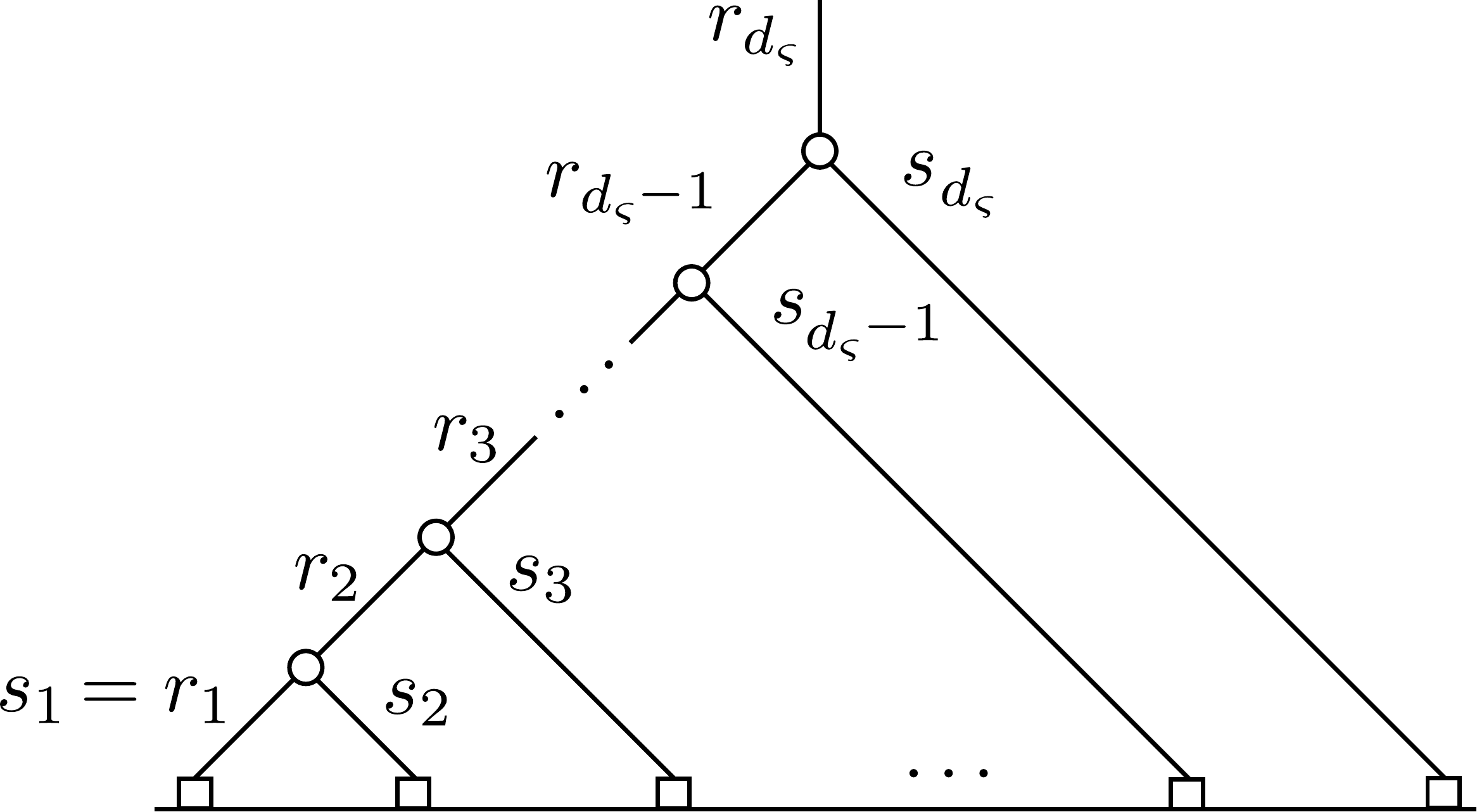} ,}}
\end{align}
where $\alpha$ is written in the generic form of a trivalent graph with \emph{open three-vertices}~\cite{kl} 
\begin{align}\label{3vertex2} 
\text{for $s \in \DefectSet\sub{r,t}$} , \qquad 
\vcenter{\hbox{\includegraphics[scale=0.275]{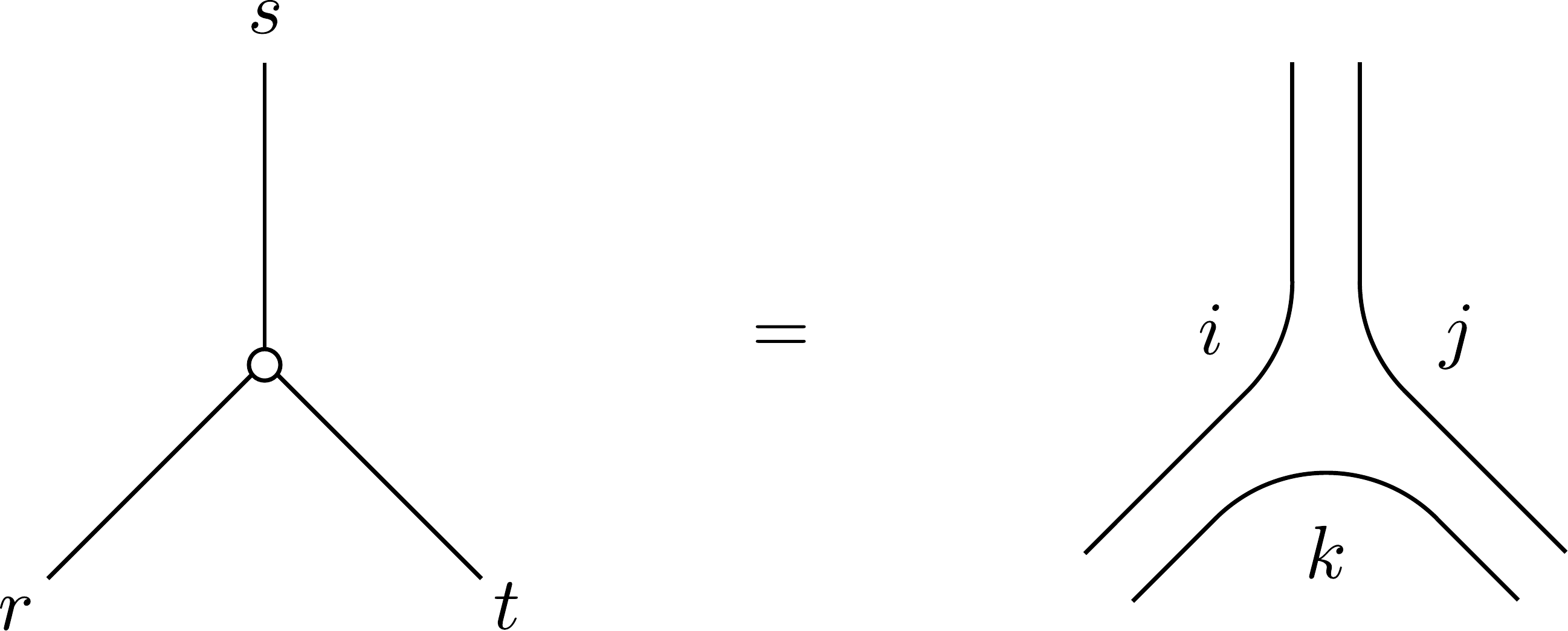} ,}} \qquad  \qquad \qquad
\begin{aligned} 
i & = \frac{r + s - t}{2} , \\[.7em] 
j & = \frac{s + t - r}{2} , \\[.7em] 
k & = \frac{t + r - s}{2} .
\end{aligned}
\end{align}
We also need the (non-oriented) \emph{closed three-vertex} notation~\eqref{3vertex1}
Now, we define the following operation $\alpha \mapsto \hcancel{\alpha}$ on valenced link patterns $\alpha \in \LP_\multii$.
Writing $\alpha$ in the generic form~\eqref{Generic}, we define the \emph{trivalent link state} $\hcancel{\alpha} \in \LS_\multii$ as
\begin{align} \label{ValencedLinkStateDef}
\hcancel{\alpha} \quad := \quad \vcenter{\hbox{\includegraphics[scale=0.275]{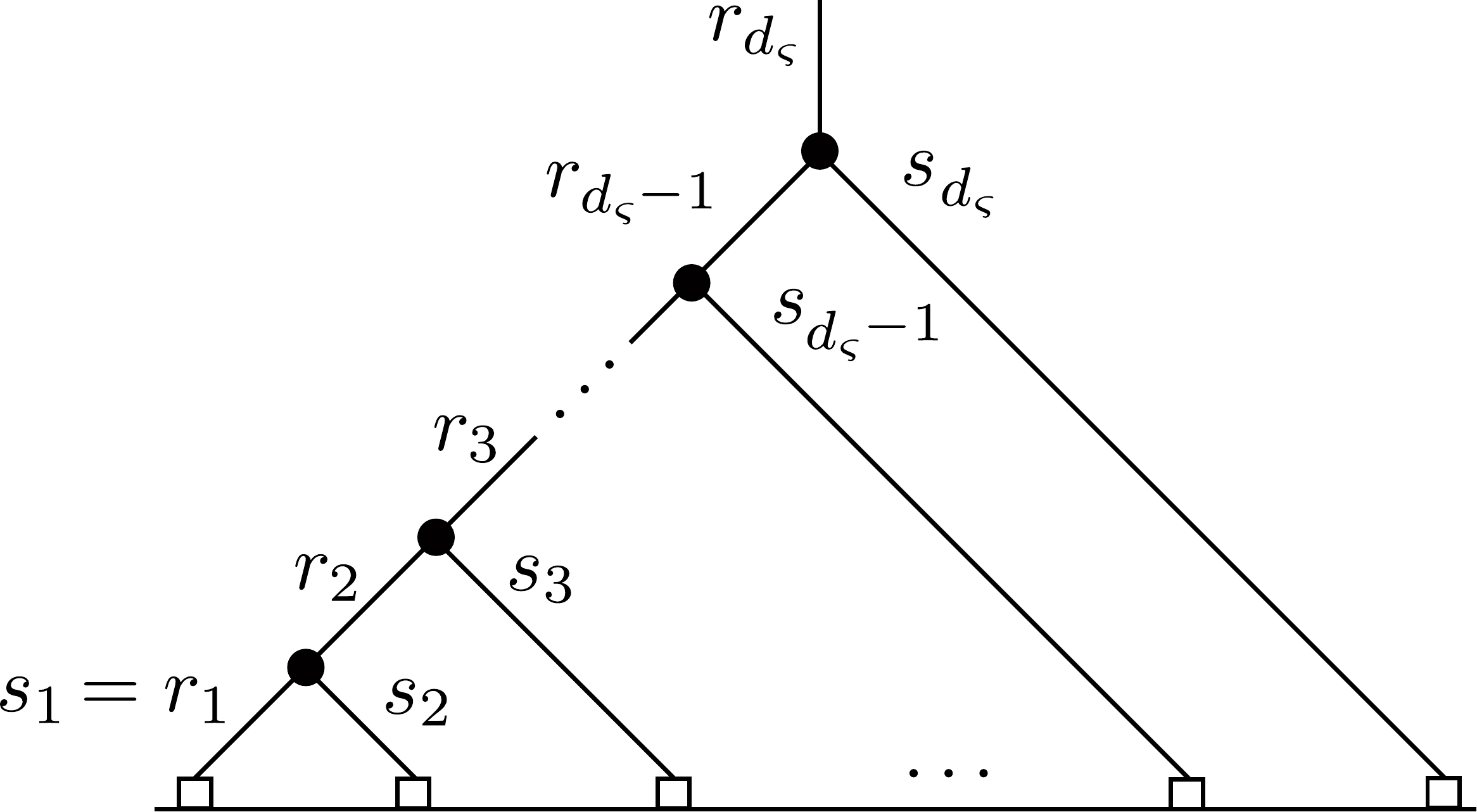} ,}}
\end{align}
that is, to obtain $\hcancel{\alpha}$ from $\alpha$, we replace the $j$:th open vertex in the walk representation~\eqref{Generic} 
of the valenced link pattern $\alpha$ with a closed vertex for each step $j \in \{1,2,\ldots,\np_{\multii}-1\}$ of the walk.
We note from~(\ref{3vertex2},~\ref{3vertex1}) that replacing the $j$:th open vertex with a closed vertex inserts three projector boxes on the appropriate cables,
which might not exist a priori. However, in~\cite[definition~\red{4.3}]{fp3a} it is shown how the definition makes sense for all $\max \multii < \pmin(q)$.

\begin{cor} \label{RelateThemCor} 
Suppose $\max \multii < \pmin(q)$. For all valenced link patterns 
$\alpha \in \smash{\LP_\multii}$ such that 
$\varrho_\alpha = (r_1, r_2, \ldots, r_{\np_\multii})$ and $\max \hat{\varrho}_\alpha < \pmin(q)$, we have
\begin{align} \label{WalkToLink} 
\HWvec^{\varrho_\alpha}_{\multii} = 
\Bigg( \prod_{j \, = \, 1}^{\np_\multii - 1} \frac{(\ii q^{1/2})^{\frac{r_j + \sIndex_{j+1} - r_{j+1}}{2}}}{(q - q^{-1})^{\frac{r_j + \sIndex_{j+1} - r_{j+1}}{2}}[\frac{r_j + \sIndex_{j+1} - r_{j+1}}{2}]!} \Bigg) \Sing_{\scaleobj{0.85}{\hcancel{\alpha}}} .
\end{align}
Similarly, for all valenced link patterns $\alphaBar \in \smash{\LPBar_\multii}$ such that 
$\varrho_{\alphaBar} = (r_1, r_2, \ldots, r_{\np_\multii})$ and $\max \hat{\varrho}_{\alphaBar} < \pmin(q)$, we have
\begin{align} \label{WalkToLinkBar} 
\HWvecBar^{\varrho_{\alphaBar}}_{\multii} = 
\Bigg( \prod_{j \, = \, 1}^{\np_\multii - 1} \frac{(-\ii q^{-1/2})^{\frac{r_j + \sIndex_{j+1} - r_{j+1}}{2}}}{(q - q^{-1})^{\frac{r_j + \sIndex_{j+1} - r_{j+1}}{2}}[\frac{r_j + \sIndex_{j+1} - r_{j+1}}{2}]!} \Bigg) \SingBar_{\scaleobj{0.85}{\overbarcal{\hcancel{\alpha}}}} .
\end{align}
\end{cor}

\begin{proof} 
This follows immediately from lemma~\ref{EmbTheta2Lem} and~\eqref{ValencedLinkStateDef}. 
\end{proof}

\subsection{Valenced tangle representations of projectors}
\label{GraphicalProjSect}

Our next aim is to interpret diagrams containing Jones-Wenzl projectors as certain $\Uqsltwo$-submodule projectors.
These results generalize lemma~\ref{WJprojLem} 
and corollaries~\ref{CompositeProjCor}--\ref{CompositeProjCorHatEmb} from sections~\ref{DiacActTypeOneSec}--\ref{GenDiacActTypeOneSec}. 
In lemma~\ref{ThisLemma}, we show how to use the Jones-Wenzl projectors to project onto any submodule 
of the tensor product $\Module{\VecSp_\multii}{\Uqsltwo}$ 
when $\Summed_\multii < \pmin(q)$
(and not just the submodule $\Wd\sub{n}$ with largest dimension as the Jones-Wenzl projectors do). 
In proposition~\ref{TLProjectionLem3}, we interpret the embedding $\smash{\CCembedor\super{s}\sub{r,t}}$
and the two projectors $\smash{\CCprojector\sub{r,t}\superscr{(r,t);(s)}}$ and $\smash{\CChatprojector\super{r,t}\sub{s}}$,
respectively defined in (\ref{EmbeddingDef2x2},~\ref{ProjectionDefn2x2},~\ref{ProjectioHatDefn2x2}), as tangles.

To begin, for each pair $\alpha \in \smash{\LP_\multii\super{s}}$,  
$\betaBar \in \smash{\LPBar_\multiii\super{s}}$ 
of valenced link patterns with the same number $s$ of defects, we define 
$\BarAction \alpha \quad \betaBar \BarAction$ to be the $(\multii, \multiii)$-valenced link diagram 
obtained by placing $\alpha$ to the left of $\betaBar$ in the plane, 
rotating both $\alpha$ and $\betaBar$ by $-\pi/2$ radians, and 
joining the $s$ defects of $\alpha$ and $\betaBar$ together pairwise top-to-bottom: for example, 
\begin{align}  \label{ExampleSadwich1}
\overbrace{\vcenter{\hbox{\includegraphics[scale=0.275]{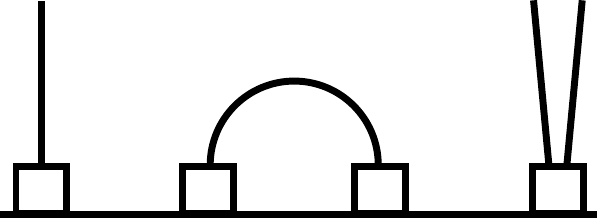}}}}^{\alpha \, \in \,\LP\super{3}\sub{1,1,1,2}} \qquad  \text{and} \qquad
\overbrace{\vcenter{\hbox{\scalebox{1}[-1]{\includegraphics[scale=0.275]{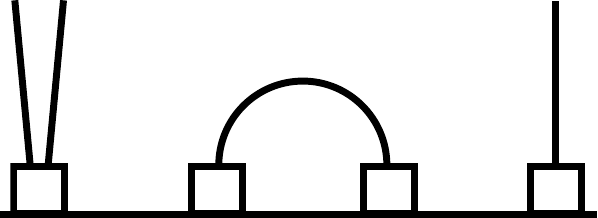}}}}}^{\overbarcal{\beta} \, \in \,\LPBar\super{3}\sub{2,1,1,1}} 
\qquad \qquad \longmapsto \qquad \qquad
\vcenter{\hbox{\includegraphics[scale=0.275]{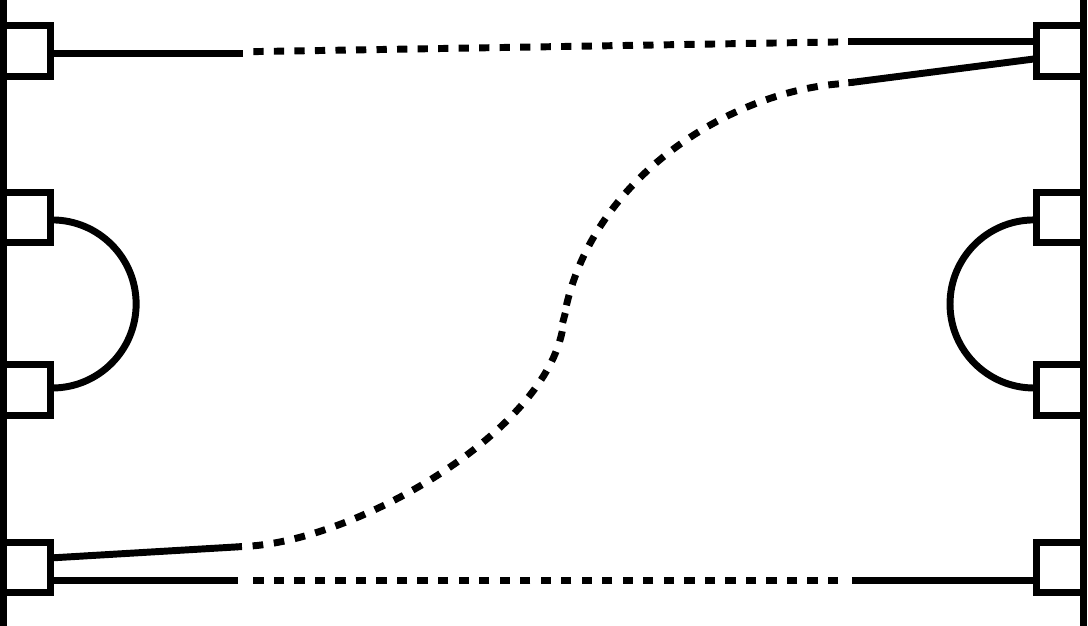}}} \quad 
= \quad \overbrace{\;\;\vcenter{\hbox{\includegraphics[scale=0.275]{Figures/e-TLSplit5_valenced2.pdf} .}}}^{\BarAction \alpha \quad  \overbarcal{\beta} \BarAction \, \in \, \TL\sub{1,1,1,2}\super{2,1,1,1}}
\end{align}
When $s < \pmin(q)$,
we also define $\BarAction \alpha \;\; \ProjBox \;\; \betaBar \BarAction$ to be the  $(\multii, \multiii)$-valenced tangle obtained from 
$\BarAction \alpha \quad \beta \BarAction$ by inserting a vertical projector box across all $s$ of its crossing links: for example, 
\begin{align}  \label{ExampleSadwich2}
\overbrace{\vcenter{\hbox{\includegraphics[scale=0.275]{Figures/e-LinkPattern5_valenced.pdf}}}}^{\alpha \, \in \,\LP\super{3}\sub{1,1,1,2}} \qquad  \text{and} \qquad
\overbrace{\vcenter{\hbox{\scalebox{1}[-1]{\includegraphics[scale=0.275]{Figures/e-LinkPattern7_valenced.pdf}}}}}^{\overbarcal{\beta} \, \in \,\LPBar\super{3}\sub{2,1,1,1}} 
\qquad \qquad \longmapsto \qquad \qquad
\vcenter{\hbox{\includegraphics[scale=0.275]{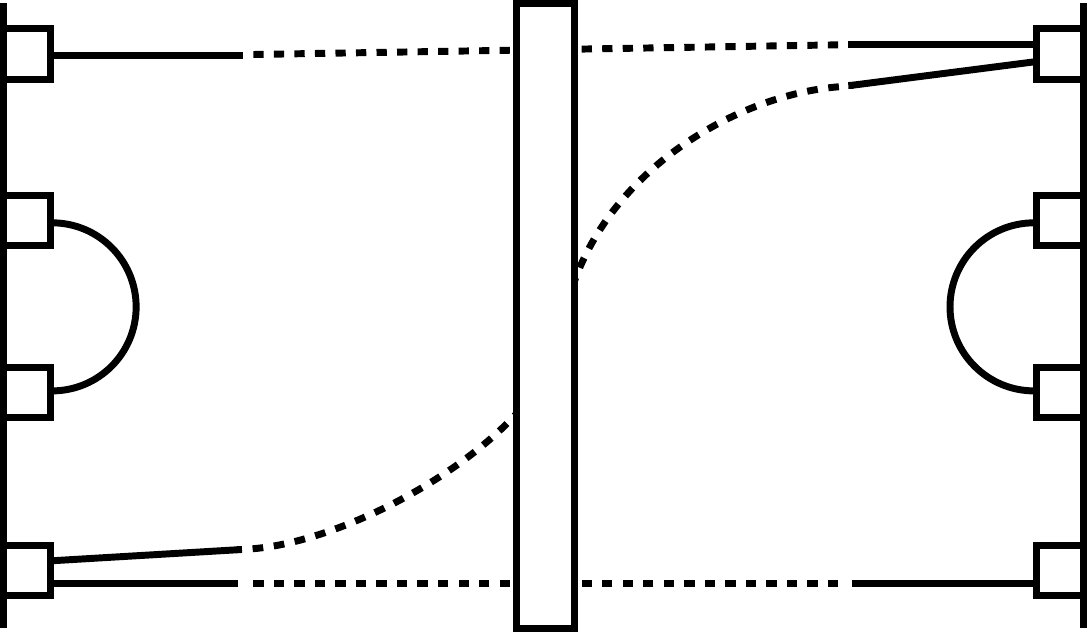}}} \quad 
= \quad \overbrace{\;\;\vcenter{\hbox{\includegraphics[scale=0.275]{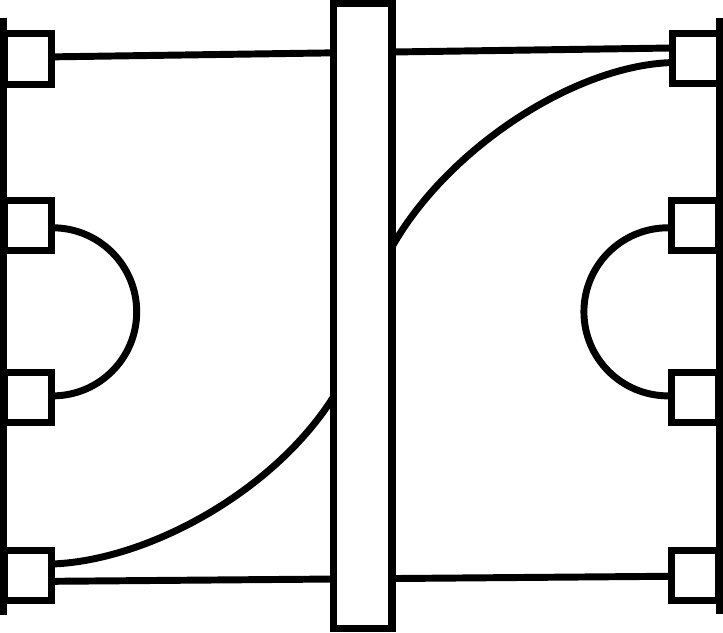} .}}}^{\BarAction \alpha \quad  \overbarcal{\beta} \BarAction \, \in \, \TL\sub{1,1,1,2}\super{2,1,1,1}}
\end{align}
We remark that tangles as the ones appearing on the right side of either~(\ref{ExampleSadwich1},~\ref{ExampleSadwich2}) 
form a basis for $\TL_\multii^\multiii$. 
The former assertion is clear, and to argue the latter, we see by recursion property~\eqref{wjrecursion} of the Jones-Wenzl projector that
the linear map defined by inserting the vertical projector box into each basis element~\eqref{ExampleSadwich1} 
has an upper unitriangular matrix representation (cf.~identity~\eqref{ProjDecBasis} below), thus being a bijection.  

\bigskip

The next lemma gives useful rules for calculating the diagram action in terms of the bilinear pairing $\SPBiForm{\cdot}{\cdot}$ on $\VecSp_\multii$.

\begin{lem} \label{NewRidoutLem} 
Suppose $\max(\multii, \multiii) < \pmin(q)$, and let $\alpha \in \smash{\LS_\multii\super{s}}$ and $\betaBar \in \smash{\LSBar_\multiii\super{s}}$.
Then, the following hold:
\begin{enumerate}
\itemcolor{red}

\item \label{NewRidoutIdIt1}
For all vectors $v \in \smash{\Ksp_\multiii\super{s}}$ and $\overbarStraight{v} \in \smash{\KspBar_\multii\super{s}}$, we have
\begin{align}
\label{NewRidoutId}
\BarAction \alpha \quad \betaBar \BarAction v = \SPBiForm{\SingBar_{\betaBar}}{v} \, \Sing_\alpha 
\qquad \qquad \textnormal{and} \qquad\qquad 
\overbarStraight{v} \BarAction \alpha \quad \betaBar \BarAction = \SPBiForm{\overbarStraight{v}}{\Sing_{\alpha}} \, \SingBar_{\betaBar}.
\end{align}

\item \label{NewRidoutIdIt2}
If $s < \pmin(q)$, then for all valenced link states $\gamma \in \smash{\LS_\multiii\super{s}}$ 
and $\overbarcal{\delta} \in \smash{\LSBar_\multii\super{s}}$, and for all $\ell \in \bZnn$, we have
\begin{align}
\label{TangleProjAct}
\BarAction \alpha \;\; \ProjBox \;\; \betaBar \BarAction F^\ell.\Sing_\gamma = \SPBiForm{\SingBar_{\betaBar}}{\Sing_\gamma} \, F^\ell.\Sing_\alpha 
\qquad \qquad \textnormal{and} \qquad\qquad 
\Sing_{\overbarcal{\delta}} .E^\ell \BarAction \alpha \;\; \ProjBox \;\; \betaBar \BarAction = \SPBiForm{\SingBar_{\overbarcal{\delta}}}{\Sing_{\alpha}} \, \SingBar_{\betaBar}. E^\ell.
\end{align}
\end{enumerate}
\end{lem}

\begin{proof} 
We prove items~\ref{NewRidoutIdIt1}--\ref{NewRidoutIdIt2} as follows
(only the left equations of~\eqref{NewRidoutId} and~\eqref{TangleProjAct}; the right ones are similar):
\begin{enumerate}[leftmargin=*]
\itemcolor{red}

\item 
To begin, we prove the left equation of~\eqref{NewRidoutId}
for the special case $\multii = \OneVec{n}$ and $\multiii = \OneVec{m}$. 
By linearity, we may assume that $\alpha \in \smash{\LP_n\super{s}}$ and $\betaBar \in \smash{\LPBar_m\super{s}}$ 
are link patterns and $v$ is a standard basis vector in $\smash{\VecSp_m\super{s}}$
of the form 
\begin{align} \label{StdBasisForm2} 
v = \FundBasis_{\ell_1} \otimes \FundBasis_{\ell_2} \otimes \dotsm \otimes \FundBasis_{\ell_m} 
\qquad \text{with} \quad \ell_1, \ell_2,\ldots,\ell_m \in \{0,1\} \quad \text{and} \quad m - 2(\ell_1 + \ell_2 + \dotsm + \ell_m) = s . 
\end{align}
Writing the link patterns $\alpha$ and $\betaBar$ in terms of left and right generators $\Lgen_i, \Rgen_j$~\eqref{LgenForm}, 
definition~\ref{SingletBasisDefinition} gives
\begin{alignat}{4}
\label{canonvec1} 
& \alpha = \Lgen_{i_k} \Lgen_{i_{k-1}} \dotsm \Lgen_{i_2} \Lgen_{i_1} \defects_s 
\qquad\qquad \, \overset{\eqref{SingletBasisDef}}{\Longrightarrow} \qquad \qquad
&& \Sing_\alpha = \Lgen_{i_k} \Lgen_{i_{k-1}} \dotsm \Lgen_{i_2} \Lgen_{i_1}  \MTbas_0\super{s} 
\\
\label{canonvec2} 
& \betaBar = \defects_s \Rgen_{j_1} \Rgen_{j_2} \dotsm \Rgen_{j_{l-1}} \Rgen_{j_l} 
\qquad\qquad \overset{\eqref{SingletBasisDefBar}}{\Longrightarrow} \qquad \qquad
&& \SingBar_{\betaBar} = \MTbas_0\super{s}  \Rgen_{j_1} \Rgen_{j_2} \dotsm \Rgen_{j_{l-1}} \Rgen_{j_l} , 
\end{alignat}
for some $i_1, i_2, \ldots, i_k, j_1, j_2, \ldots, j_l \in \bZpos$, where $2l + s = m$.
Also, we have 
\begin{align} \label{linkDiagFactor} 
\BarAction \alpha \quad \betaBar \BarAction \underset{\eqref{canonvec2}}{\overset{\eqref{canonvec1}}{=}} 
\Lgen_{i_k} \Lgen_{i_{k-1}} \dotsm \Lgen_{i_2} \Lgen_{i_1} \mathbf{1}_{\TL_s} \Rgen_{j_1} \Rgen_{j_2} \dotsm \Rgen_{j_{l-1}} \Rgen_{j_l}.
\end{align}
Now, with the vector $v$ of the form~\eqref{StdBasisForm2} with $\frac{m-s}{2}$ tensorands equal to $\FundBasis_1$ and
$\frac{m+s}{2}$ tensorands equal to $\FundBasis_0$, 
repeated application of rule~\eqref{ExtendThis2} shows that
\begin{align} \label{Rv} 
\Rgen_{j_1} \Rgen_{j_2} \dotsm \Rgen_{j_{l-1}} \Rgen_{j_l}  v \underset{\eqref{StdBasisForm2}}{\overset{\eqref{ExtendThis2}}{=}} r_{j_1} r_{j_2} \dotsm r_{j_l} \MTbas_0\super{s} \qquad \text{with $r_{j_1},$ $r_{j_2},\ldots, r_{j_l} \in \{\pm \ii q^{\pm1/2},0\}$}, 
\end{align}
where, for each $p$, the factor $r_{j_p}$ arises from the action of the right generator $R_{j_p}$. Thus, we arrive with 
\begin{align} 
\nonumber
\BarAction \alpha \quad \betaBar \BarAction v &\overset{\eqref{linkDiagFactor}}{=} \Lgen_{i_k} \Lgen_{i_{k-1}} \dotsm \Lgen_{i_2} \Lgen_{i_1} \mathbf{1}_{\TL_s} \Rgen_{j_1} \Rgen_{j_2} \dotsm \Rgen_{j_{l-1}} \Rgen_{j_l}  v \\
\label{sameR1} 
&\overset{\eqref{Rv}}{=} r_{j_1} r_{j_2} \dotsm r_{j_l} \Lgen_{i_k} \Lgen_{i_{k-1}} \dotsm \Lgen_{i_2} \Lgen_{i_1}  \MTbas_0\super{s} \underset{\eqref{canonvec2}}{\overset{\eqref{canonvec1}}{=}}  r_{j_1} r_{j_2} \dotsm r_{j_l} \Sing_\alpha. 
\end{align}
On the other hand, corollary~\ref{SwicthTCor} gives
\begin{align} 
\nonumber
\SPBiForm{\SingBar_{\betaBar}}{v} \, \Sing_\alpha & 
\underset{\eqref{canonvec2}}{\overset{\eqref{canonvec1}}{=}}  
\SPBiForm{ \MTbasBar_0\super{s} \Rgen_{j_1} \Rgen_{j_2} \dotsm \Rgen_{j_{l-1}} \Rgen_{j_l}}{v} \,  \Sing_\alpha 
\overset{\eqref{SwicthT}}{=} \SPBiForm{\MTbasBar_0\super{s}}{\Rgen_{j_1} \Rgen_{j_2} \dotsm \Rgen_{j_{l-1}} \Rgen_{j_l}  v} \,  \Sing_\alpha \\ 
\label{sameR2} 
& \overset{\eqref{Rv}}{=} r_{j_1} r_{j_2} \dotsm r_{j_l} 
\SPBiForm{\MTbasBar_0\super{s}}{\MTbas_0\super{s}} \Sing_\alpha 
\underset{\eqref{biformnormalization}}{\overset{\eqref{MThwv}}{=}} 
r_{j_1} r_{j_2} \dotsm r_{j_l} \Sing_\alpha
\overset{\eqref{sameR1}}{=} \BarAction \alpha \quad \betaBar \BarAction v.
\end{align}
This proves the left equation of~\eqref{NewRidoutId} for the case $\multii = \OneVec{n}$ and $\multiii = \OneVec{m}$.

Next, we use the result of the previous paragraph to prove the left equation of~\eqref{NewRidoutId} 
for general multiindices $\multii, \multiii \in \smash{\bZpos^\#}$. 
Indeed, using lemma~\ref{SmoothingLem2} and corollary~\ref{CompositeProjCorHatEmb}, we obtain
\begin{align} 
\nonumber 
\BarAction \alpha \quad \betaBar \BarAction v 
&\overset{\eqref{ImultHom}}{=} \Projectionhat_\multii \big( \BarAction \WJEmb_\multii \alpha \quad \betaBar \WJProjHat_\multiii \BarAction \Embedding_\multiii(v) \big) 
\overset{\eqref{NewRidoutId}}{=} \SPBiForm{\SingBar_{\betaBar\WJProjHat_\multiii}}{\Embedding_\multiii(v)} \, \Projectionhat_\multii(\Sing_{\WJEmb_\multii\alpha}) 
\overset{\eqref{Lcommutation}}{=} 
\SPBiForm{\SingBar_{\betaBar} \WJProjHat_\multiii}{\Embedding_\multiii(v)} \, \Projectionhat_\multii(\Sing_{\WJEmb_\multii\alpha}) \\
\label{RidoutGeneralArgument}
& \overset{\eqref{CompTwoProjsHatEmbBar}}{=} 
\SPBiForm{\EmbeddingBar_\multiii(\SingBar_{\betaBar})}{\Embedding_\multiii(v)} \, \Sing_\alpha 
\overset{\eqref{SPBiFormNewEmbed}}{=} \SPBiForm{\SingBar_{\betaBar}}{v} \, \Sing_\alpha ,
\end{align}
which proves the left equation of~\eqref{NewRidoutId} for the general case. 

\item 
Thanks to lemma~\ref{SmoothingLem2}, we may assume that $\ell = 0$.
To begin, we prove the left equation of~\eqref{TangleProjAct}
for the special case $\multii = \OneVec{n}$ and $\multiii = \OneVec{m}$. 
By linearity, we may assume that $\alpha \in \smash{\LP_n\super{s}}$, $\betaBar \in \smash{\LPBar_m\super{s}}$, and $\gamma \in \smash{\LP_m\super{s}}$
are link patterns. We decompose the projector box between $\alpha$ and $\betaBar$ via
its recursion property~\eqref{wjrecursion},
\begin{align} \label{ProjDecBasis} 
\BarAction \alpha \quad \ProjBox \quad \betaBar \BarAction 
\overset{\eqref{wjrecursion}}{=} \BarAction \alpha \quad \betaBar \BarAction 
+ \sum_{\substack{ r \, \in \, \DefectSet_n^m \\ r \, < \, s}} T_{\alpha,\betaBar}\super{r} ,
\end{align}
$\smash{T_{\alpha,\betaBar}\super{r}} \in \smash{\TL_n^m}$ being tangles with exactly $r$ crossing links 
and $\smash{\DefectSet_n^m}$ the set of all integers $r \geq 0$ for which such tangles exist.
Then, the action of tangle~\eqref{ProjDecBasis} on the vector $\Sing_\gamma$ reads
\begin{align}  \label{BarProdTang}
\BarAction \alpha \quad \ProjBox \quad \betaBar \BarAction \Sing_\gamma
\overset{\eqref{ProjDecBasis}}{=} 
\BarAction \alpha \quad \betaBar \BarAction \Sing_\gamma
+ \sum_{\substack{ r \, \in \, \DefectSet_n^m \\ r \, < \, s}} T_{\alpha,\betaBar}\super{r} \, \Sing_\gamma 
\underset{\eqref{NewRidoutId}}{\overset{\textnormal{(\ref{TurnBackWeightZero}, \ref{SingAlphaDiagram})}}{=}}
\SPBiForm{\SingBar_{\betaBar}}{\Sing_\gamma} \, \Sing_\alpha ,
\end{align}
as each tangle $\smash{T_{\alpha,\betaBar}\super{r}}$ appearing in the sum has $r < s$ crossing links, 
so all of these tangles necessarily join two defects of $\gamma$ together in the product $\smash{T_{\alpha,\betaBar}\super{r}}\gamma$ via turn-back links,
so $\smash{T_{\alpha,\betaBar}\super{r}}\gamma = 0$ by rules~(\ref{TurnBackWeightZero},~\ref{SingAlphaDiagram}).
This proves the left equation of~\eqref{NewRidoutId} for the case $\multii = \OneVec{n}$ and $\multiii = \OneVec{m}$.
We establish the general case via the same argument as in~\eqref{RidoutGeneralArgument}. 
\end{enumerate}
This concludes the proof.
\end{proof}

Assuming that $\rad \smash{\LS_\multii\super{s}} = \{0\} = \rad \smash{\LSBar_\multii\super{s}}$, and
choosing a basis $\smash{\mathsf{B}_\multii\super{s}} \subset \smash{\LS_\multii\super{s}}$ for the $\TL_\multii(\nu)$-standard module,
for each element $\alpha \in  \smash{\mathsf{B}_\multii\super{s}}$, we let $\alpha\superscr{\cheque} \in \smash{\LSBar_\multii\super{s}}$ 
denote the dual of $\alpha$ with respect to the bilinear pairing $\LSBiFormBar{\cdot}{\cdot}$, 
\begin{align} \label{LPdualBasis1}
\LSBiFormBar{\alpha\superscr{\cheque}}{\beta} = \delta_{\alpha,\beta}  \quad \text{for all $\alpha, \beta \in \mathsf{B}_\multii\super{s}$.} 
\end{align}
Similarly, 
choosing a basis $\smash{\overbarStraight{\mathsf{B}}_\multii\super{s}} \subset \smash{\LSBar_\multii\super{s}}$, 
for each element $\alphaBar \in \smash{\overbarStraight{\mathsf{B}}_\multii\super{s}}$, we define the dual
$\alphaBar\superscr{\cheque} \in \smash{\LS_\multii\super{s}}$ of $\alphaBar$ by the property
\begin{align} \label{LPdualBasis2}
\LSBiFormBar{\alphaBar}{\betaBar\superscr{\cheque}} = \delta_{\alphaBar,\betaBar}  \quad \text{for all $\alphaBar, \betaBar \in \overbarStraight{\mathsf{B}}_\multii\super{s}$.} 
\end{align}
Then, proposition~\ref{HWspLem2} implies that the sets
$\{\Sing_\alpha \, | \, \alpha \in \smash{\mathsf{B}_\multii\super{s}}\}$ and 
$\{\Sing_\alpha\superscr{\cheque} \, | \, \alpha \in \smash{\mathsf{B}_\multii\super{s}}\}$
are dual bases for $\smash{\HWsp_\multii\super{s}}$ and $\smash{\HWspBar_\multii\super{s}}$: 
\begin{align} 
\Sing_\alpha\superscr{\cheque} := \SingBar_{\alpha\superscr{\cheque}}
\qquad \qquad \Longrightarrow \qquad \qquad 
\SPBiForm{\Sing_\alpha\superscr{\cheque}}{\Sing_{\beta}}
= \SPBiForm{\SingBar_{\alpha\superscr{\cheque}}}{\Sing_{\beta}}
\overset{\eqref{BilinFormPreserve}}{=} 
\LSBiFormBar{\alpha\superscr{\cheque}}{\beta}
\overset{\eqref{LPdualBasis1}}{=}  
\delta_{\alpha,\beta} ,
\end{align}
and similarly, the sets
$\{\SingBar_{\alphaBar} \, | \, \alphaBar \in \smash{\overbarStraight{\mathsf{B}}_\multii\super{s}}\}$ and 
$\{\SingBar_{\alphaBar}\superscr{\cheque} \, | \, \alpha \in \smash{\overbarStraight{\mathsf{B}}_\multii\super{s}}\}$
are also dual bases for $\smash{\HWspBar_\multii\super{s}}$ and $\smash{\HWsp_\multii\super{s}}$: 
\begin{align} 
\SingBar_{\alphaBar}\superscr{\cheque} := \Sing_{\alphaBar\superscr{\cheque}}
\qquad \qquad \Longrightarrow \qquad \qquad 
\SPBiForm{\SingBar_{\betaBar}}{\SingBar_{\alphaBar}\superscr{\cheque}}
= \SPBiForm{\SingBar_{\betaBar}}{\Sing_{\alphaBar\superscr{\cheque}}}
\overset{\eqref{BilinFormPreserve}}{=} 
\LSBiFormBar{\betaBar}{\alphaBar\superscr{\cheque}}
\overset{\eqref{LPdualBasis2}}{=}  
\delta_{\alphaBar,\betaBar} .
\end{align}
Recalling the $s$-gradings of $\LS_\multii$, $\rad \LS_\multii$, and $\HWsp_\multii$ from~(\ref{LSDirSum},~\ref{LSRadical},~\ref{HWspace2}), 
assuming that $\rad \LS_\multii = \{0\} = \rad \LSBar_\multii$, 
any basis $\mathsf{B}_\multii$ for $\LS_\multii$ or any basis $\overbarStraight{\mathsf{B}}_\multii$ for $\LSBar_\multii$
give rise to the above type dual basis pairs for $\LS_\multii$ and $\LSBar_\multii$ as well as for $\HWsp_\multii$ and $\HWspBar_\multii$.

Recalling  from~\cite[corollary~\red{5.22} and equations~(\red{5.106}-\red{5.107})]{fp3a} that these assumptions are valid
when $\Summed_\multii < \pmin(q)$,
\begin{align}
\Summed_\multii < \pmin(q)
\qquad \underset{\textnormal{\cite[cor. \red{5.22}]{fp3a}}}{\overset{\textnormal{\cite[(\red{5.106}-\red{5.107})]{fp3a}}}{\Longrightarrow}} \qquad 
\rad \LS_\multii = \bigoplus_{s \, \in \, \DefectSet_\multii} \rad \LS_\multii\super{s} = \{0\} ,
\end{align}
and noting that $\rad \LSBar_\multii$ is isomorphic to $\rad \LS_\multii$ by definition~\eqref{LSRadical},
the above dual bases are always defined if $\Summed_\multii < \pmin(q)$.
In particular, assuming that $\max(\Summed_\multiii, \Summed_\multii) < \pmin(q)$, we define the projection operators
by homomorphic extensions of
\begin{align}
\label{ProjDefn}
\CCprojector^{\beta}_\alpha & \colon \VecSp_\multiii \longrightarrow \VecSp_\multii, \qquad 
\CCprojector_\alpha^\beta(v) := \SPBiForm{\Sing_{\beta}\superscr{\cheque}}{v} \Sing_\alpha \qquad \text{for $v \in \HWsp_\multiii$}, \\
\label{ProjDefnBar}
\CCprojectorBar^{\betaBar}_{\alphaBar} & \colon \VecSpBar_\multii \longrightarrow \VecSpBar_\multiii , \qquad 
\CCprojectorBar_{\alphaBar}^{\betaBar}(\overbarStraight{v}) := 
\SPBiForm{\overbarStraight{v}}{\SingBar_{\alphaBar}\superscr{\cheque}} \SingBar_{\betaBar}  \qquad \text{for $\overbarStraight{v} \in \HWspBar_\multii$}, 
\end{align}
so they are homomorphisms of $\Uqsltwo$- and $\UqsltwoBar$-modules. We also set 
\begin{align}
\CChatprojector^{\beta} = \CCprojector^{\beta}\sub{s} , \qquad 
\CCembedor_\alpha := \CCprojector\super{s}_\alpha , \qquad 
\CCprojector_\alpha := \CCprojector_\alpha^\alpha ,
\qquad\qquad \text{and} \qquad\qquad
\CChatprojectorBar^{\betaBar} = \CCprojectorBar^{\betaBar}\sub{s} , \qquad 
\CCembedorBar_{\alphaBar} := \CCprojectorBar\super{s}_{\alphaBar} , \qquad 
\CCprojectorBar_{\alphaBar} := \CCprojectorBar_{\alphaBar}^{\alphaBar} .
\end{align}

To state the next lemma, we also denote by $\defectsVal_s$ the valenced link pattern with $s$ defects attached to one node.

\begin{lem} \label{ThisLemma}
Suppose $\max(\Summed_\multiii, \Summed_\multii) < \pmin(q)$.  
Then, the following hold:
\begin{enumerate}
\itemcolor{red}

\item 
With respect to any bases $\mathsf{B}_\multiii$ and $\mathsf{B}_\multii$ for $\LS_\multiii$ and $\LS_\multii$, 
and for all valenced link states $\alpha \in \mathsf{B}_\multii \cap \smash{\LS_\multii\super{s}}$ 
and $\beta \in \mathsf{B}_\multiii \cap \smash{\LS_\multiii\super{s}}$,
the maps~\eqref{ImultHom} send the following valenced tangles to the following $\Uqsltwo$-homomorphisms\textnormal{:}
\begin{align}
\label{DiagProjHat}
& \BarAction \text{\rotatebox[origin=c]{-90}{$\defectsVal$}}_s \, \ProjBox \;\; \beta\superscr{\cheque} \BarAction
\qquad \overset{\Trep\sub{s}^\multiii}{\longmapsto} \qquad 
\CChatprojector^{\beta}, \\
\label{DiagEmbed}
& \BarAction \alpha \;\;\, \ProjBox \;\; \text{\raisebox{-.3ex}{\rotatebox[origin=c]{90}{$\defectsVal$}}}_s \BarAction
\qquad \overset{\Trep\super{s}_\multii}{\longmapsto} \qquad 
\CCembedor_\alpha, \\
\label{DiagProj}
& \BarAction \alpha \;\;\, \ProjBox \;\; \beta\superscr{\cheque} \BarAction
\qquad \overset{\Trep_\multii^\multiii}{\longmapsto} \qquad  \CCprojector_\alpha^{\beta} .
\end{align}

\item 
With respect to any bases $\overbarStraight{\mathsf{B}}_\multiii$ and $\overbarStraight{\mathsf{B}}_\multii$ for $\LSBar_\multiii$ and $\LSBar_\multii$, 
and for all valenced link states $\alphaBar \in \overbarStraight{\mathsf{B}}_\multiii \cap \smash{\LSBar_\multii\super{s}}$ 
and $\betaBar \in \overbarStraight{\mathsf{B}}_\multii \cap \smash{\LSBar_\multiii\super{s}}$,
the maps~\eqref{ImultHomBar} send the following valenced tangles to the following $\UqsltwoBar$-homomorphisms\textnormal{:}
\begin{align}
\label{DiagProjHatBar}
& \BarAction \alphaBar\superscr{\cheque \,} \; \ProjBox \;  \text{\raisebox{-.3ex}{\rotatebox[origin=c]{90}{$\defectsVal$}}}_s \BarAction
\qquad \overset{\TrepBar_\multii\super{s}}{\longmapsto} \qquad 
\CChatprojectorBar_{\alphaBar}, \\
\label{DiagEmbedBar}
& \BarAction \rotatebox[origin=c]{-90}{$\defectsVal$}_s \; \ProjBox \;\;\; \betaBar \BarAction
\qquad \overset{\TrepBar\sub{s}^\multiii}{\longmapsto} \qquad 
\CCembedorBar^{\betaBar}, \\
\label{DiagProjBar}
& \BarAction \alphaBar\superscr{\cheque} \;\, \ProjBox \;\;\; \betaBar \BarAction
\qquad \overset{\TrepBar_\multii^\multiii}{\longmapsto} \qquad  \CCprojectorBar_{\alphaBar}^{\betaBar} .
\end{align} 
\end{enumerate}
\end{lem}

\begin{proof} 
We only prove identities~(\ref{DiagProjHat}--\ref{DiagProj}), since~(\ref{DiagProjHatBar}--\ref{DiagProjBar}) can be proven similarly.
Also, as the first two claims (\ref{DiagProjHat},~\ref{DiagEmbed}) are special cases of the last~\eqref{DiagProj}, it suffices to prove only the latter. 
Furthermore, 
definition~\eqref{ProjDefn} of $\smash{\CCprojector^{\beta}_\alpha}$
and orthogonality property~\eqref{OrthoSubspaces} from lemma~\ref{biformPropertyLem} imply that 
if $v \notin \smash{\VecSp_\multiii\super{s}}$, then $\smash{\CCprojector^{\beta}_\alpha}(v) = 0$. 
Also, because the link-state vectors $\Sing_\gamma$ with $\gamma \in \LP_\multiii$ form a basis for $v \in \HWsp_\multiii$
by item~\ref{BasIt2} of corollary~\ref{SingletBasisIsBasisCor}, by linearity and since $\smash{\CCprojector^{\beta}_\alpha}$
is a homomorphism of $\Uqsltwo$-modules, it actually suffices to prove that
$\smash{\BarAction \alpha \;\; \ProjBox \;\; \beta\superscr{\cheque} \BarAction \Sing_\gamma} =  \smash{\CCprojector^{\beta}_\alpha}(\Sing_\gamma)$
for all $\gamma \in \LP_\multiii$: indeed,
\begin{align}
\label{ProjToBiForm}
\BarAction \alpha \;\; \ProjBox \;\; \beta\superscr{\cheque} \BarAction \Sing_\gamma \overset{\eqref{TangleProjAct}}{=} \SPBiForm{\Sing_\beta\superscr{\cheque}}{\Sing_\gamma}\Sing_\alpha \overset{\eqref{ProjDefn}}{=}  \CCprojector_\alpha^\beta(\Sing_\gamma) 
\end{align}
for any link pattern $\gamma \in \smash{\LP_\multiii\super{s}}$, 
by item~\ref{NewRidoutIdIt2} of lemma~\ref{NewRidoutLem}. This proves the lemma.
\end{proof}

The next proposition gives explicit, simple diagram presentations for the consecutive-tensorand
embedding $\smash{\CCembedor\super{s}\sub{r,t}}$ and 
projectors $\smash{\CCprojector\sub{r,t}\superscr{(r,t);(s)}}$ and $\smash{\CChatprojector\super{r,t}\sub{s}}$,
respectively defined in (\ref{EmbeddingDef2x2},~\ref{ProjectionDefn2x2},~\ref{ProjectioHatDefn2x2}).

\begin{prop} \label{TLProjectionLem3} 
Suppose $r + t < \pmin(q)$, and define
\begin{align} \label{ABConstants}
A\super{r,t}\sub{s} 
:= \frac{\ThetaNet(r, s, t) (\ii q^{1/2})^{\frac{r + t - s}{2}}}{(-1)^s [s + 1](q - q^{-1})^{\frac{r + t - s}{2}}[\frac{r + t - s}{2}]!}
\qquad \qquad \textnormal{and} \qquad \qquad
B\sub{r,t}\super{s} := \frac{(q - q^{-1})^{\frac{r + t - s}{2}}[\frac{r + t - s}{2}]!}{ (\ii q^{1/2})^{\frac{r + t - s}{2}}}.
\end{align}
Then, the maps~\eqref{ImultHom} send the following valenced tangles to the following $\Uqsltwo$-homomorphisms\textnormal{:}
\begin{align} 
\label{ProjDiagram3} 
\vcenter{\hbox{\includegraphics[scale=0.275]{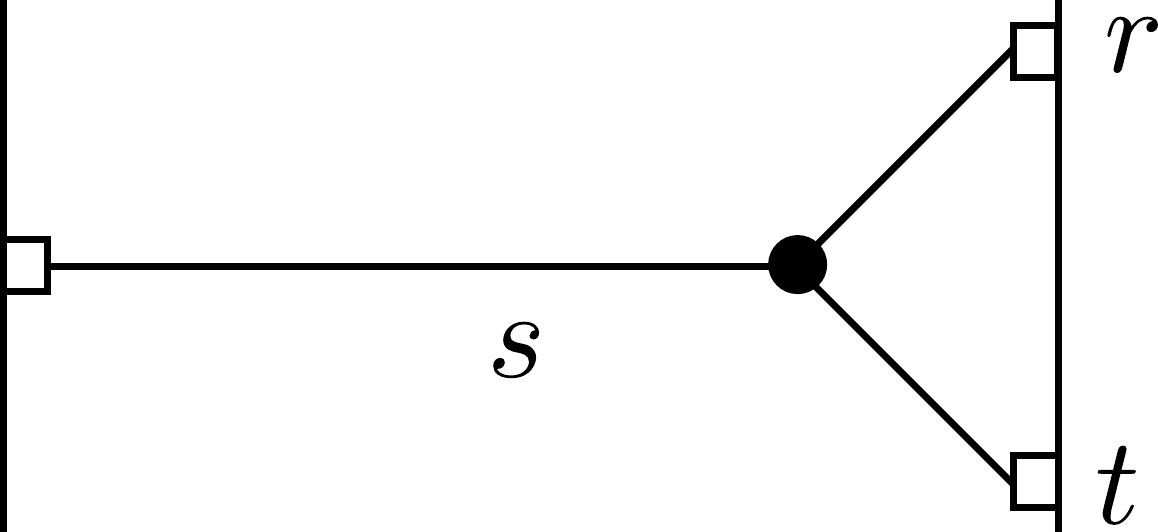} }} 
& \quad \overset{\Trep\sub{s}\super{r,t}}{\longmapsto} \quad A\super{r,t}\sub{s} \; \CChatprojector\super{r,t}\sub{s} , \\[1em]
\label{ProjDiagram4} 
 \vcenter{\hbox{\includegraphics[scale=0.275]{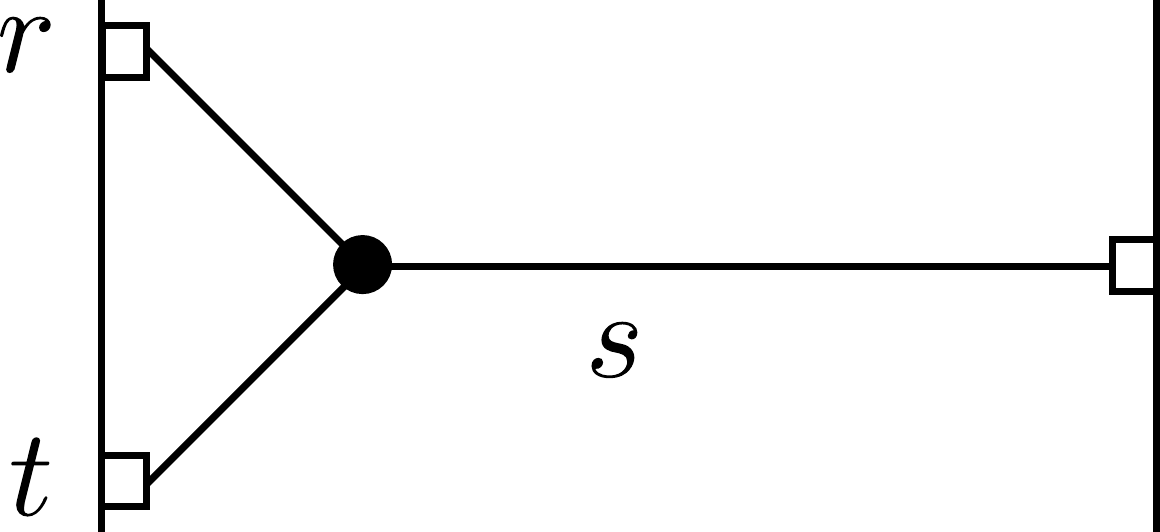} }} \; \; \,
& \quad \overset{\Trep\sub{r,t}\super{s}}{\longmapsto} \quad B\sub{r,t}\super{s} \; \CCembedor\super{s}\sub{r,t} , \\[1em]
\label{ProjDiagram} 
\vcenter{\hbox{\includegraphics[scale=0.275]{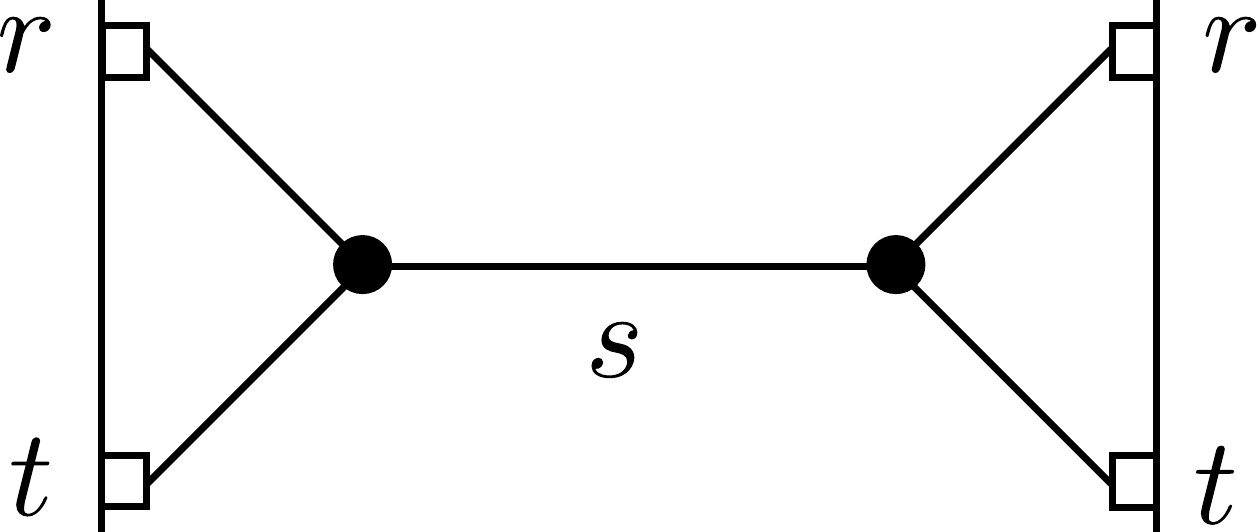} }}  
& \quad \overset{\Trep\sub{r,t}}{\longmapsto} \quad  B\sub{r,t}\super{s} A\super{r,t}\sub{s} \; \CCprojector\sub{r,t}\superscr{(r,t);(s)} .
\end{align}
Similarly, define 
\begin{align} 
\overbarcal{A}\super{r,t}\sub{s} 
:= \frac{\ThetaNet(r, s, t) (-\ii q^{-1/2})^{\frac{r + t - s}{2}}}{(-1)^s [s + 1](q - q^{-1})^{\frac{r + t - s}{2}}[\frac{r + t - s}{2}]!}
\qquad \qquad \textnormal{and} \qquad \qquad
\overbarcal{B}\sub{r,t}\super{s} := \frac{(q - q^{-1})^{\frac{r + t - s}{2}}[\frac{r + t - s}{2}]!}{ (-\ii q^{-1/2})^{\frac{r + t - s}{2}}}.
\end{align}
Then, the maps~\eqref{ImultHomBar} send the following valenced tangles to the following $\UqsltwoBar$-homomorphisms:
\begin{align} 
\label{ProjDiagram3Bar} 
\vcenter{\hbox{\includegraphics[scale=0.275]{Figures/e-Generators2Nodes_3Vertex_right_valenced.pdf} }} \; \; \,
& \quad \overset{\TrepBar\sub{s}\super{r,t}}{\longmapsto} \quad \overbarcal{A}\super{r,t}\sub{s} \; \CChatprojectorBar\super{r,t}\sub{s} , \\[1em]
\label{ProjDiagram4Bar} 
\vcenter{\hbox{\includegraphics[scale=0.275]{Figures/e-Generators2Nodes_3Vertex_left_valenced.pdf} }} 
& \quad \overset{\TrepBar\sub{r,t}\super{s}}{\longmapsto} \quad \overbarcal{B}\sub{r,t}\super{s} \; \CCembedorBar\super{s}\sub{r,t} , \\[1em]
\label{ProjDiagramBar}
\vcenter{\hbox{\includegraphics[scale=0.275]{Figures/e-Generators2Nodes_3Vertex_valenced.pdf} }}  
& \quad \overset{\TrepBar\sub{r,t}}{\longmapsto} \quad  \overbarcal{B}\sub{r,t}\super{s} \overbarcal{A}\super{r,t}\sub{s} \; \CCprojectorBar\sub{r,t}\superscr{(r,t);(s)} .
\end{align}
\end{prop}

\begin{proof}
To prove~\eqref{ProjDiagram3}, we recall from lemma~\ref{HWDiagramLem} that
\begin{align} 
\beta_s \quad = \quad \vcenter{\hbox{\includegraphics[scale=0.275]{Figures/e-HWV16_valenced3.pdf}}} 
\qquad\qquad \overset{\eqref{TWenzl}}{\Longrightarrow} \qquad\qquad
\Sing_{\beta_s} = \frac{(q - q^{-1})^{\frac{r + t - s}{2}}[\frac{r + t - s}{2}]!}{(\ii q^{1/2})^{\frac{r + t - s}{2}}} \, \HWvec\sub{r,t}\super{s}.
\end{align}
Also, we have the duality
\begin{align} 
\beta_s \quad = \quad \vcenter{\hbox{\includegraphics[scale=0.275]{Figures/e-HWV16_valenced3.pdf}}} 
\qquad\qquad \overset{\eqref{LoopErasure1}}{\Longrightarrow} \qquad\qquad
\beta_s\superscr{\cheque} \quad = \quad \frac{(-1)^s [s + 1]}{\ThetaNet(r,s,t)} \,\, \times \,\, \vcenter{\hbox{\includegraphics[scale=0.275]{Figures/e-HWV16_valenced3_flipped.pdf}}}.
\end{align}
Now, corollary~\ref{SingletBasisIsBasisCor} implies that the collection 
$\{ \Sing_{\beta_u} \,|\, u \in \DefectSet\sub{r,t} \}$ is a basis for $\HWsp\sub{r,t}$. 
Selecting an arbitrary vector $\Sing_{\beta_u}$ from this set and inserting 
$\beta_s$  into~\eqref{DiagProjHat}, lemma~\ref{ThisLemma} gives 
\begin{align}
\label{BigID}
\frac{(-1)^s[s + 1]}{\ThetaNet(r,s,t)} \,\, \times \,\,
\vcenter{\hbox{\includegraphics[scale=0.275]{Figures/e-Generators2Nodes_3Vertex_left_valenced.pdf} }}  \Sing_{\beta_u}  \quad
& \underset{\hphantom{\eqref{ProjectioHatDefn2x2}}}{\overset{\eqref{DiagProjHat}}{=}}  
\CChatprojector^{\beta_s}(\Sing_{\beta_u}) 
\overset{\eqref{ProjDefn}}{=}  
\SPBiForm{\Sing_{\beta_s\superscr{\cheque}}}{\Sing_{\beta_u}} \Basis_0\super{s} 
\overset{\eqref{BilinFormPreserve}}{=} \delta_{s,u} \Basis_0\super{s} \\
& \overset{\eqref{ProjectioHatDefn2x2}}{=}  
\CChatprojector\sub{s}\super{r,t} \big( \HWvec\sub{r,t}\super{u} \big) = \frac{(\ii q^{1/2})^{\frac{r + t - s}{2}}}{(q - q^{-1})^{\frac{r + t - s}{2}}[\frac{r + t - s}{2}]!}\,  \CChatprojector\sub{s}\super{r,t} (\Sing_{\beta_p}).
\end{align}
Furthermore, because the map $\smash{\CChatprojector\sub{s}\super{r,t}}$ and the diagram action on the left side of~\eqref{BigID} are $\Uqsltwo$-homomorphisms according to lemma~\ref{EmbProjLem} and item~\ref{UqHomo2It1} of lemma~\ref{UqHomoLem2} respectively,~\eqref{BigID} extends 
for all vectors $v \in \{F^\ell. \Sing_{\beta_u} \,|\, u \in \DefectSet\sub{r,t} \}$:
\begin{align}
\label{BigID2}
\vcenter{\hbox{\includegraphics[scale=0.275]{Figures/e-Generators2Nodes_3Vertex_left_valenced.pdf} }}  v 
\quad \overset{\eqref{BigID}}{=} \quad 
\frac{\ThetaNet(r, s, t) (\ii q^{1/2})^{\frac{r + t - s}{2}}}{(-1)^s [s + 1](q - q^{-1})^{\frac{r + t - s}{2}}[\frac{r + t - s}{2}]!} \, \CChatprojector\sub{s}\super{r,t} (v) .
\end{align}
Proposition~\ref{HWspacePropEmbAndIso} says that this set is a basis for $\VecSp\sub{r,t}$ if $r + t < \pmin(q)$,
so~\eqref{BigID2} holds for all $v \in \VecSp\sub{r,t}$, and~\eqref{ProjDiagram3} follows.

To finish, identity~\eqref{ProjDiagram4} can be proven similarly as above,~\eqref{ProjDiagram} then follows from (\ref{ProjDiagram3},~\ref{ProjDiagram4}) and relation~\eqref{2ndIt3IdentityTrivial},
and identities (\ref{ProjDiagram3Bar}--\ref{ProjDiagramBar}) can be proven similarly. 
\end{proof}


The next lemma shows that the image of the representation 
$\Trep_\multii$ of the valenced Temperley-Lieb algebra $\TL_\multii(\nu)$ is generated by 
the projectors on $\VecSp_\multii$ acting on consecutive tensorands. 
In section~\ref{KerImSubSec}, we show that the representation $\Trep_\multii$ is in fact faithful, which thus identifies $\TL_\multii(\nu)$ with its image.

\begin{lem}  \label{GeneratorLemTL} 
Suppose $\Summed_\multii < \pmin(q)$.  Then, the image of the representation 
$\Trep_\multii \colon \TL_\multii(\nu) \longrightarrow \End \VecSp_\multii$ 
is generated by the collection of all submodule projectors that act strictly on consecutive pairs of 
tensorands of vectors in  $\VecSp_\multii$\textnormal{:}
\begin{align} \label{GeneratorsTLImageAux} 
\Trep_\multii ( \TL_\multii(\nu) )
= \big\langle \CCprojector\sub{\sIndex_i,\sIndex_{i+1}}\superscr{(\sIndex_i,\sIndex_{i+1});(s)} \, \; \big| \, \;
s \in \DefectSet\sub{\sIndex_i,\sIndex_{i+1}}, \, i \in \{1, 2, \ldots, \np_\multii - 1\} \big\rangle .
\end{align} 
Similarly, this proposition holds after the symbolic replacements $\Trep \mapsto \TrepBar$,  $\VecSp \mapsto \VecSpBar$, and
$\CCprojector \mapsto  \CCprojectorBar$.
\end{lem}

\begin{proof}
Identity~\eqref{ProjDiagram} (resp.~\eqref{ProjDiagramBar}) 
of proposition~\ref{TLProjectionLem3} and lemma~\ref{TensorLem}
together imply that 
\begin{align} \label{Sent} 
\frac{(-1)^s [s + 1]}{\ThetaNet(\sIndex_i, \sIndex_{i+1}, s)} \,\, \times \,\, \vcenter{\hbox{\includegraphics[scale=0.275]{Figures/e-Generators_3Vertex_valenced.pdf}}} 
\quad \qquad \qquad \underset{\eqref{ProjDiagram}}{\overset{\eqref{TensorID}}{\longmapsto}} \qquad \qquad 
\CCprojector\sub{\sIndex_i,\sIndex_{i+1}}\superscr{(\sIndex_i,\sIndex_{i+1});(s)} ,
\end{align}
(resp.~for $\smash{\CCprojector\sub{\sIndex_i,\sIndex_{i+1}}\superscr{(\sIndex_i,\sIndex_{i+1});(s)}}$), 
for all $s \in \DefectSet\sub{\sIndex_i,\sIndex_{i+1}}$ and $i \in \{1, 2, \ldots, \np_\multii - 1\}$.
With $\Summed_\multii < \pmin(q)$, the prefactor in~\eqref{Sent} is finite and nonvanishing. 
By~\cite[proposition~\red{2.10}]{fp3a} (see also~\cite[theorem~\red{1.1}]{fp3b}), 
we know that the collection of all tangles on the left side of~\eqref{Sent} generates $\TL_\multii(\nu)$.  
Because $\Trep_\multii$ is a homomorphism of algebras, the claim follows.
\end{proof}

%

Using the previous results, we also obtain a diagrammatic proof for the commuting diagram in~\cite[Lemma~\red{3.10}]{ep2},
which gives means to calculate compositions of projectors explicitly. 

\begin{lem} \label{EveProjLem} 
Suppose $r + t < \pmin(q)$, and denote 
$\CChatprojector := \CChatprojector\super{1,1}\sub{0}$ and $k := \frac{r + t - s}{2}$. 
Then, we have
\begin{align} \label{OpenUp} 
\left( \frac{[r - k]![t - k]![r + t - k + 1]!}{[2]^{k} [r]! [s+1]! [t]!} \right) 
\CChatprojector\super{r,t}\sub{s}
= \Projectionhat\sub{s} \circ \big( \CChatprojector_{r - k + 1,r - k + 2} 
\circ \CChatprojector_{r - k + 2,r - k + 3} \circ \dotsm \circ 
\CChatprojector_{r - 1, r} \circ \CChatprojector_{r, r + 1} \big) 
\circ \Embedding\sub{r,t}. 
\end{align}
Equivalently, the following diagram commutes:
\begin{equation} \label{EveDiag}
\begin{tikzcd}[column sep=2.5cm, row sep=1.5cm]
\VecSp\sub{r} \otimes \VecSp\sub{t} \arrow{ddddd}[swap]{\CChatprojector\super{r,t}\sub{s}}
\arrow{r}{\Embedding\sub{r,t}} 
& \VecSp_r\otimes\VecSp_t \arrow{d}{\CChatprojector_{r,r+1}} \\ 
& \VecSp_{r-1} \otimes \VecSp_{t-1}  \arrow{d}{\CChatprojector_{r-1,r}} \\ 
& \vdots \arrow{d} \\
& \VecSp_{r-k+1} \otimes \VecSp_{t-k+1} \arrow{d}{\CChatprojector_{r-k+1,r-k+2}} \\
& \VecSp_{r-k} \otimes \VecSp_{t-k} = \VecSp_s  \arrow{d}{\Projectionhat\sub{s}} \\
\VecSp\sub{s} \arrow{r}{\frac{[r - k]![t - k]![r + t - k + 1]!}{[2]^{k} [r]! [s]! [t]!} \times \id}
& \VecSp\sub{s} 
\end{tikzcd}
\end{equation}
Similarly, with $\CChatprojectorBar := \CChatprojectorBar\super{1,1}\sub{0}$, 
we have
\begin{align} \label{OpenUpBar} 
\left( \frac{[r - k]![t - k]![r + t - k + 1]!}{[2]^{k} [r]! [s+1]! [t]!} \right) 
\CChatprojectorBar\super{r,t}\sub{s}
= \ProjectionhatBar\sub{s} \circ \big( \CChatprojectorBar_{r - k + 1,r - k + 2} 
\circ \CChatprojectorBar_{r - k + 2,r - k + 3} \circ \dotsm \circ 
\CChatprojectorBar_{r - 1, r} \circ \CChatprojectorBar_{r, r + 1} \big) 
\circ \EmbeddingBar\sub{r,t} .
\end{align}
\end{lem}

\begin{proof} 
Proposition~\ref{TLProjectionLem3} and corollary~\ref{CompositeProjCorHatEmb} show that
$\smash{\CChatprojector\super{r,t}\sub{s}}$ is the image of the following valenced tangle under 
$\smash{\Trep\sub{s}\super{r,t}}$:
\begin{align} \label{Penult} 
\frac{1}{A\super{r,t}\sub{s}} \,\, \vcenter{\hbox{\includegraphics[scale=0.275]{Figures/e-Generators2Nodes_3Vertex_left_valenced.pdf} }}
\quad \underset{\eqref{WJCompEmbAndProjHat}}{\overset{\eqref{ProjectorID1}}{=}}  \quad  
\frac{1}{A\super{r,t}\sub{s}} \,\, \WJProjHat\sub{s} \,\,
\vcenter{\hbox{\includegraphics[scale=0.275]{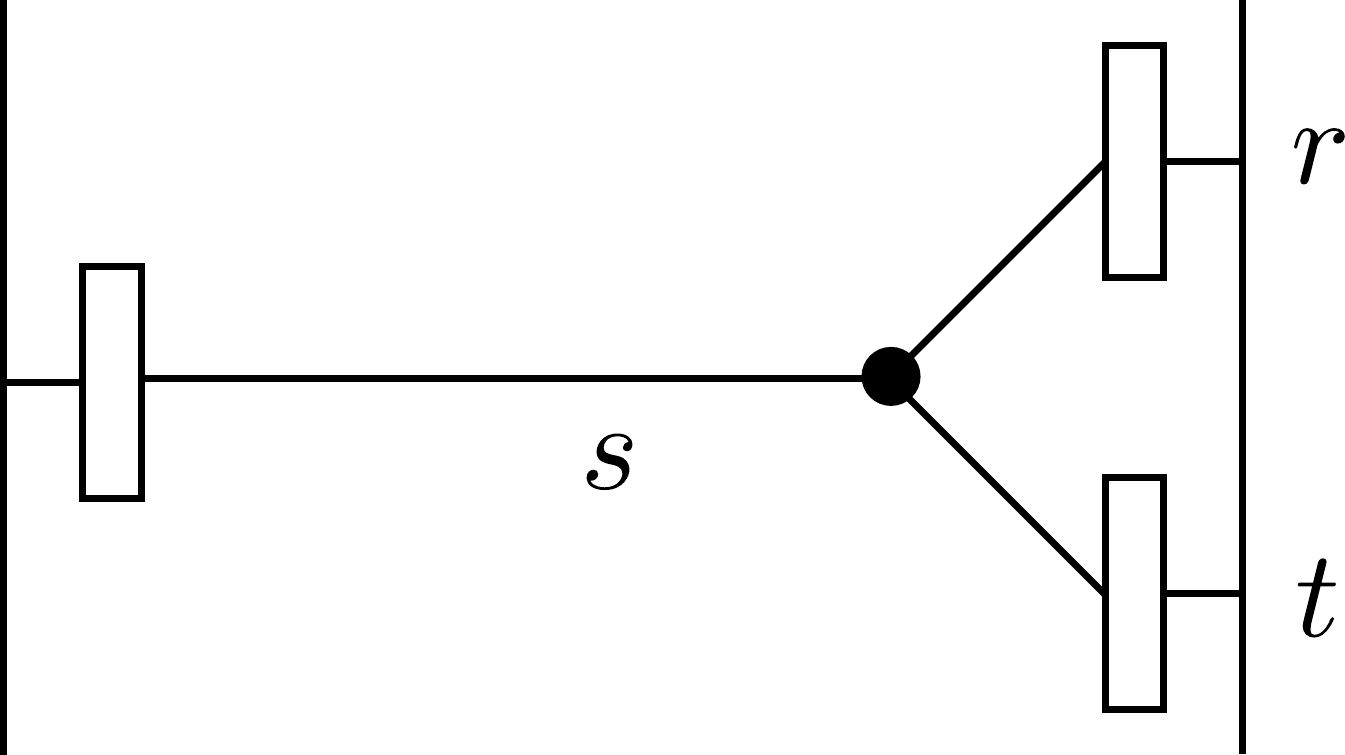}}} \,\, \WJEmb\sub{r,t} .
\end{align}
Decomposing the three-vertex in this tangle, we have
\begin{align} \label{OpenItUp} 
\vcenter{\hbox{\includegraphics[scale=0.275]{Figures/e-Generators2Nodes_3Vertex_left.pdf}}} 
\quad \overset{\eqref{3vertex1}}{=} \quad 
\vcenter{\hbox{\includegraphics[scale=0.275]{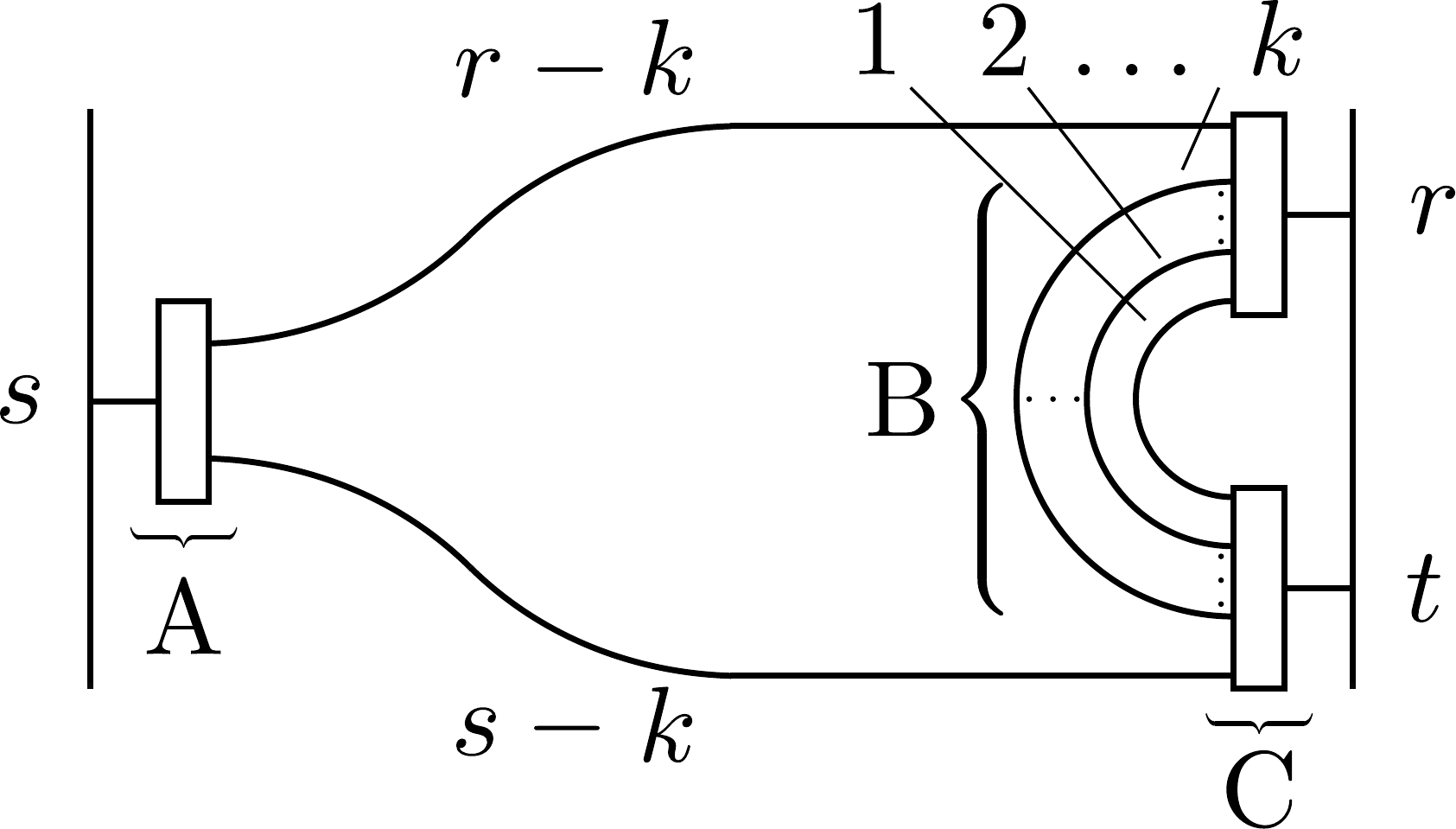} .}} 
\end{align}
We consider the image of the valenced tangle~\eqref{OpenItUp} under the map $\smash{\Trep\sub{s}\super{r,t}}$.
We use \eqref{CompTwoProjs} of
corollary~\ref{CompositeProjCor} to interpret part C of~\eqref{OpenItUp} as 
the projector 
$\Projection\sub{r,t} = \Projection\sub{r} \otimes \Projection\sub{t}$,
\eqref{GenProj2-0} of lemma~\ref{TLprojLemNew} to interpret each link in part B of~\eqref{OpenItUp} as 
\begin{align} \label{UseIt} 
\left(\frac{\ii \nu q^{1/2}}{q-q^{-1}}\right) 
\big(\id^{\otimes(j-1)} \otimes \CChatprojector\super{1,1}\sub{0} \otimes \id^{\otimes(n-j-1)}\big) 
= \left(\frac{\ii \nu q^{1/2}}{q-q^{-1}}\right) \CChatprojector_{j,j+1}
\end{align}
for some appropriate index $j \in \{1, 2, \ldots, r + t\}$,
and \eqref{Recover} of lemma~\ref{WJprojLem} to interpret part A 
of~\eqref{OpenItUp} as the projector $\Projection\sub{s}$.  After combining these three parts, we obtain
\begin{align} \label{OpenItUp2} 
\text{\eqref{OpenItUp}}
\quad \overset{\Trep\sub{s}\super{r,t}}{\longmapsto} \quad  
\overbrace{\Projection\sub{s}}^{\text{A}} \circ
\underbrace{ \left( \left( \frac{\ii \nu q^{1/2}}{q-q^{-1}} \right) \CChatprojector_{r - k + 1,r - k + 2} 
\circ \dotsm \circ \left( \frac{\ii \nu q^{1/2}}{q-q^{-1}} \right) \CChatprojector_{r - 1,r} 
\circ \left( \frac{\ii \nu q^{1/2}}{q-q^{-1}} \right) \CChatprojector_{r,r + 1} \right)}_{\text{B}}
\circ \overbrace{\Projection\sub{r,t}}^{\text{C}} ,
\end{align}
with labeled parts in~\eqref{OpenItUp2} matching those in~\eqref{OpenItUp}.  
To finish, using lemma~\ref{EmbProjLem} and corollary~\ref{CompositeProjCorHatEmb}, 
we identify the two representations~\eqref{Penult} and~\eqref{OpenItUp2} 
for the image of $\smash{\CChatprojector\super{r,t}\sub{s}}$ under the map 
$\smash{\Trep\sub{s}\super{r,t}}$, to obtain
\begin{align}  
\CChatprojector\super{r,t}\sub{s} 
\underset{\textnormal{(\ref{Penult}-\ref{OpenItUp2})}}{\overset{\eqref{CompTwoProjsHatEmb}}{=}} & \; 
\frac{1}{A\super{r,t}\sub{s}} \, \left( \frac{\ii \nu q^{1/2}}{q-q^{-1}} \right)^{k} \, 
\Projectionhat\sub{s} \circ \Projection\sub{s} \circ 
\left( \CChatprojector_{r - k + 1,r - k + 2} \circ \dotsm \circ \CChatprojector_{r - 1,r} \circ \CChatprojector_{r,r + 1} \right)
\circ \Projection\sub{r,t} \circ \Embedding\sub{r,t} \\
\underset{\hphantom{\textnormal{(\ref{Penult}-\ref{OpenItUp2})}}}{\overset{\eqref{PhatP}}{=}} & \; 
\frac{1}{A\super{r,t}\sub{s}} \, \left( \frac{\ii \nu q^{1/2}}{q-q^{-1}} \right)^{k} \, 
\Projectionhat\sub{s} \circ
\left( \CChatprojector_{r - k + 1,r - k + 2} \circ \dotsm \circ \CChatprojector_{r - 1,r} \circ \CChatprojector_{r,r + 1} \right)
\circ \Embedding\sub{r,t} .
\end{align}
Simplifying the result to~\eqref{OpenUp} via formula~\eqref{ThetaFormula} for the Theta network, we obtain asserted identity~\eqref{OpenUp}.
\end{proof}

\subsection{Semisimple quotient modules of highest-weight vectors}
\label{QuotientRadicalSect}

In this section, we 
identify quotients of the $\Uqsltwo,\UqsltwoBar$-highest-weight vector spaces 
$\HWsp_\multii$ and $\HWspBar_\multii$ with simple $\TL_\multii(\nu)$-modules via 
the link state -- highest-weight vector correspondence (proposition~\ref{HWspLem2}). 
In general, all simple $\TL_\multii(\nu)$-modules are obtained from the standard modules $\smash{\LS_\multii\super{s}}$ by taking 
their quotients by their radicals $\rad \smash{\LS_\multii\super{s}}$ with respect to the link state bilinear pairing $\LSBiFormBar{\cdot}{\cdot}$
(recall~(\ref{LSRadical}--\ref{QuoSp}) from section~\ref{LSandBiformandSCGrapgSec} and~\cite[proposition~\red{6.7}]{fp3a}).
We verify in lemma~\ref{RadicalInclusionLem} that this radical 
embeds into the radical of $\HWsp_\multii$ with respect to the bilinear pairing~$\SPBiForm{\cdot}{\cdot}$.
Then, in proposition~\ref{QuotientProp} we prove that such a quotient respects the link state -- highest-weight vector correspondence.

First, using the bilinear pairing $\SPBiForm{\cdot}{\cdot}$ 
we can define a notion of orthogonality within $\HWspBar_\multii \times \HWsp_\multii$. (See also appendix~\ref{LinAlgAppAux}.)
In~particular, we define the orthocomplement of 
$\{ \Sing_\alpha \, | \, \alpha \in \LS_\multii \}$ in $\HWsp_\multii$ to be 
\begin{align} 
\label{LSPerpDefn}
\{ \Sing_\alpha \, | \, \alpha \in \LS_\multii \}\superscr{\perp}
& := \{ v \in \HWsp_\multii \, | \,  \SPBiForm{\SingBar_{\alphaBar}}{v}  = 0 \textnormal{ for all } \alphaBar \in \LSBar_\multii \} 
\; \subset \; \HWsp_\multii .
\end{align}
We similarly define the orthocomplement $\smash{\{ \SingBar_{\alphaBar} \, | \, \alphaBar \in \LSBar_\multii \}\superscr{\perp} \subset \HWspBar_\multii}$.
Also, for each $s \in \smash{\DefectSet_{\Summed_\multii}\superscr{\pm}}$, we define
\begin{align}
\label{LSPerpSDefn}
\big\{ \Sing_\alpha \, \big| \, \alpha \in \LS\super{s}_\multii \big\}\superscr{\perp}
& := \big\{ v \in \HWsp\super{s}_\multii \, \big| \,  \SPBiForm{\SingBar_{\alphaBar}}{v}  = 0 \textnormal{ for all } \alphaBar \in \LSBar\super{s}_\multii \big\} 
\; \subset \; \HWsp\super{s}_\multii ,
\end{align}
and we similarly define
$\smash{\{ \SingBar_{\alphaBar} \, | \, \alphaBar \in \LSBar\super{s}_\multii \}\superscr{\perp} \subset  \HWspBar\super{s}_\multii}$,
where we use the convention that $\smash{\LS\super{s}_\multii}, \smash{\LSBar\super{s}_\multii} = \{0\}$ if $s \notin \DefectSet_\multii$.


\begin{lem} \label{LSPerpLem}
Suppose $\max \multii < \pmin(q)$. 
\begin{enumerate}
\itemcolor{red}

\item \label{LSPerpItem1}
The vector space $\smash{\{ \Sing_\alpha \, | \, \alpha \in \LS_\multii \}\superscr{\perp}}$ is closed under the left $\TL_\multii(\nu)$-action on it. 

\item \label{LSPerpItem2} 
For each $s \in \DefectSet_\multii$, the vector space $\smash{\{ \Sing_\alpha \, | \, \alpha \in \LS\super{s}_\multii \}\superscr{\perp}}$
is closed under the left $\TL_\multii(\nu)$-action on it. 

\item \label{LSPerpItem3} 
We have the following direct-sum decomposition
of $\TL_\multii(\nu)$-submodules of $\CModule{\HWsp_\multii}{\TL}$\textnormal{:}
\begin{align} \label{LSPerpDirSum}
\CModule{\{ \Sing_\alpha \, | \, \alpha \in \LS_\multii \}\superscr{\perp}}{\TL}
= 
\Bigg(\bigoplus_{s \, \in \, \DefectSet_\multii} \CModule{\big\{ \Sing_\alpha \, \big| \, \alpha \in \LS\super{s}_\multii \big\}\superscr{\perp}}{\TL} \Bigg)
\oplus
\Bigg( \bigoplus_{s \, \in \, \DefectSet_{\Summed_\multii}\superscr{\pm} \setminus \DefectSet_\multii}
\CModule{\big\{ \Sing_\alpha \, \big| \, \alpha \in \LS\super{s}_\multii \big\}\superscr{\perp}}{\TL} \Bigg) .
\end{align}
\end{enumerate}
Similarly, this lemma holds for the right $\TL_\multii(\nu)$-action on $\HWspBar_\multii$
after the symbolic replacements
\begin{align}
\Sing_\alpha \mapsto \SingBar_{\alphaBar} , \qquad 
\alpha \mapsto \alphaBar, \qquad 
\LS \mapsto \LSBar, 
\qquad \textnormal{and} \qquad 
\HWsp \mapsto \HWspBar .
\end{align}
\end{lem}

\begin{proof}
It follows from corollary~\ref{SwicthTCor} that, for all vectors 
$v \in \smash{\{ \Sing_\alpha \, | \, \alpha \in \LS_\multii \}\superscr{\perp}}$ and valenced tangles $T \in \TL_\multii(\nu)$, we have
\begin{align}
\label{LSPerpSubmod1}
v \in \{ \Sing_\alpha \, | \, \alpha \in \LS_\multii \}\superscr{\perp} 
\qquad & \overset{\eqref{LSPerpDefn}}{\Longrightarrow} \qquad
\SPBiForm{\SingBar_{\alphaBar}}{v} = 0 && \textnormal{for all } \alphaBar \in \LSBar_\multii \\
\label{LSPerpSubmod2}
& \overset{\hphantom{\eqref{LSPerpDefn}}}{\Longrightarrow} \qquad
\SPBiForm{\SingBar_{\alphaBar}}{T v}
\overset{\eqref{SwicthT}}{=} 
\SPBiForm{\SingBar_{\alphaBar} T}{v}
\overset{\eqref{Lcommutation}}{=} 
\SPBiForm{\SingBar_{\alphaBar \, T}}{v}
\overset{\eqref{LSPerpSubmod1}}{=} 0
&& \textnormal{for all } \alphaBar \in \LSBar_\multii \\[.2em]
& \overset{\hphantom{\eqref{LSPerpDefn}}}{\Longrightarrow} \qquad
T v \in \{ \Sing_\alpha \, | \, \alpha \in \LS_\multii \}\superscr{\perp}  .
\end{align}
This proves item~\ref{LSPerpItem1}, and a similar argument proves item~\ref{LSPerpItem2}.
To prove item~\ref{LSPerpItem3}, we first note that by lemma~\ref{HWDirectSumLem},
any vector $v \in \HWsp_\multii$ can be written in the form 
\begin{align} \label{VExpansion}
v \overset{\eqref{HWDirectSumPM}}{=} 
\sum_{s \, \in \, \DefectSet_{\Summed_\multii}\superscr{\pm}} v\super{s} , \qquad
 \text{where} \quad v\super{s} \in \HWsp_\multii\super{s} .
\end{align}
Now, using lemma~\ref{sGradingLem2}, direct-sum decomposition~\eqref{LSDirSum} of $\LSBar_\multii$,
and the orthogonality of the spaces $\smash{\HWspBar_\multii\super{s}}$ and $\smash{\HWsp_\multii\super{t}}$ for $s \neq t$
from item~\ref{biformitem3} of lemma~\ref{biformPropertyLem}, we have 
\begin{align}
\label{LSPerpCalc1}
v \in \{ \Sing_\alpha \, | \, \alpha \in \LS_\multii \}\superscr{\perp} 
\qquad & \quad \underset{\eqref{VExpansion}}{\overset{\eqref{LSPerpDefn}}{\Longleftrightarrow}} \qquad
\sum_{s \, \in \, \DefectSet_{\Summed_\multii}\superscr{\pm}} \SPBiForm{\SingBar_{\alphaBar}}{v\super{s}} = 0 && \textnormal{for all } \alphaBar \in \LSBar_\multii \\
\label{LSPerpCalc2}
\qquad & \underset{\textnormal{(\ref{LSDirSum}, \ref{sGradingLemProperty2})}}{\overset{\eqref{FactBiformId}}{\Longleftrightarrow}} \qquad
\SPBiForm{\SingBar_{\alphaBar}}{v\super{s}}= 0 && \textnormal{for all } \alphaBar \in \LSBar_\multii\super{s} \textnormal{ and for all } s \in \DefectSet_{\Summed_\multii}\superscr{\pm} \\
\label{LSPerpCalc3}
\qquad & \underset{\hphantom{\textnormal{(\ref{LSDirSum}, \ref{sGradingLemProperty2})}}}{\overset{\eqref{LSPerpSDefn}}{\Longleftrightarrow}} \qquad
v\super{s} \in \{ \Sing_\alpha \, | \, \alpha \in \LS_\multii\super{s} \}\superscr{\perp} && \textnormal{for all } s \in \DefectSet_{\Summed_\multii}\superscr{\pm} .
\end{align}
It remains to note that the sum on the right side of~\eqref{LSPerpDirSum} is direct thanks to~(\ref{HWDirectSumPM},~\ref{LSPerpSDefn}).
This proves~\eqref{LSPerpDirSum}.
The assertions for the right $\TL_\multii(\nu)$-action on $\HWspBar_\multii$ can be proven similarly.
\end{proof}

%

Using the bilinear pairing, we 
also define radicals of subspaces of $\HWspBar_\multii$ and $\HWsp_\multii$. In particular, we define
\begin{align} 
\label{LSEmbRadDefn}
\rad \{ \Sing_\alpha \, | \, \alpha \in \LS_\multii \}
\underset{\hphantom{\eqref{BilinFormPreserve}}}{\overset{\hphantom{\eqref{LSRadical}}}{:=}}
& \;
 \{ \Sing_\alpha \, | \, \alpha \in \LS_\multii \text{ and } \SPBiForm{\SingBar_{\betaBar}}{\Sing_\alpha}  = 0 \textnormal{ for all } \betaBar \in \LSBar_\multii \} \\
\label{LSEmbRad}
\underset{\eqref{BilinFormPreserve}}{\overset{\eqref{LSRadical}}{=}}
& \;
\{ \Sing_\alpha \,|\, \alpha \in \rad \LS_\multii \} \; \subset \; \HWsp_\multii .
\end{align}
We similarly define
$\smash{\rad \{ \SingBar_{\alphaBar} \, | \, \alphaBar \in \LSBar_\multii \} \subset \HWspBar_\multii}$.
Also, for each $s \in \DefectSet_\multii$, we define
\begin{align}
\label{LSEmbRadSDefn}
\rad \big\{ \Sing_\alpha \, \big| \, \alpha \in \LS_\multii\super{s} \big\}
\underset{\hphantom{\eqref{BilinFormPreserve}}}{\overset{\hphantom{\eqref{LSRadicalS}}}{:=}} 
& \;
\big\{ \Sing_\alpha \, \big| \, \alpha \in \LS_\multii\super{s} \text{ and }  \SPBiForm{\SingBar_{\betaBar}}{\Sing_\alpha}  = 0 \textnormal{ for all } \betaBar \in \LSBar_\multii\super{s} \big\} \\
\label{LSEmbRadS}
\underset{\eqref{BilinFormPreserve}}{\overset{\eqref{LSRadicalS}}{=}}
& \;
 \big\{ \Sing_\alpha \,\big|\, \alpha \in \rad \LS_\multii\super{s} \big\} 
\; \subset \; \HWsp_\multii\super{s} ,
\end{align} 
and we similarly define 
$\smash{\rad \{ \SingBar_{\alphaBar} \, | \, \alphaBar \in \LSBar_\multii\super{s} \} \subset \HWspBar\super{s}_\multii}$.
As in the proof of lemma~\ref{LSPerpLem}, corollary~\ref{SwicthTCor} implies that these spaces are closed under their $\TL_\multii(\nu)$-actions, 
and linearity and similar arguments as in~(\ref{VExpansion}--\ref{LSPerpCalc3}) show that 
\begin{align} 
\label{LSradicalDirSum}
\CModule{\rad\{ \Sing_\alpha \, | \, \alpha \in \LS_\multii \}}{\TL}
& = \bigoplus_{s \, \in \, \DefectSet_\multii} \CModule{\rad \big\{ \Sing_\alpha \, \big| \, \alpha \in \LS_\multii\super{s} \big\}}{\TL} , \\
\label{LSradicalDirSumBar}
\CRModule{\rad \{ \SingBar_{\alphaBar} \, | \, \alphaBar \in \LSBar_\multii \}}{\TL}
& = \bigoplus_{s \, \in \, \DefectSet_\multii} \CRModule{\rad \big\{ \SingBar_{\alphaBar} \,\big|\, \alphaBar \in \LSBar_\multii\super{s} \big\}}{\TL} .
\end{align}
Finally, we define the radical of $\HWsp_\multii$ as 
\begin{align} 
\rad \HWsp_\multii 
& := \{  v \in \HWsp_\multii \, | \, \SPBiForm{\overbarStraight{w}}{v}  = 0 \textnormal{ for all } \overbarStraight{w} \in \HWspBar_\multii \}
\; \subset \; \HWsp_\multii ,
\end{align} 
we define $\rad \HWspBar_\multii  \subset \HWspBar_\multii$ similarly, and for each $s \in \smash{\DefectSet_{\Summed_\multii}\superscr{\pm}}$, we define
\begin{align} 
\rad \HWsp_\multii\super{s} 
& := \big\{  v \in \HWsp_\multii\super{s}  \, \big| \, \SPBiForm{\overbarStraight{w}}{v}  = 0 \textnormal{ for all } \overbarStraight{w} \in \HWspBar_\multii\super{s}  \big\} 
\; \subset \; \HWsp_\multii\super{s} ,
\end{align}
and similarly $\smash{\rad \HWspBar_\multii\super{s} \subset \HWspBar_\multii\super{s}}$.
As above, corollary~\ref{SwicthTCor} implies that these spaces are closed under their $\TL_\multii(\nu)$-actions, 
and linearity and similar arguments as in~(\ref{VExpansion}--\ref{LSPerpCalc3}) show that 
\begin{align} \label{HradicalDirSum}
\CModule{\rad \HWsp_\multii}{\TL} 
& = \bigoplus_{s \, \in \, \DefectSet_{\Summed_\multii}\superscr{\pm}} \CModule{\rad \HWsp_\multii\super{s}}{\TL}  
\qquad\qquad \text{and} \qquad\qquad
\CRModule{\rad \HWspBar_\multii}{\TL}
= \bigoplus_{s \, \in \, \DefectSet_{\Summed_\multii}\superscr{\pm}} \CRModule{\rad \HWspBar_\multii\super{s}}{\TL} .
\end{align}

The next lemma says that the radical of $\LS_\multii$ with respect to the link state bilinear pairing  $\LSBiFormBar{\cdot}{\cdot}$
embeds into the radical of $\HWsp_\multii$ with respect to the bilinear pairing~$\SPBiForm{\cdot}{\cdot}$.
This fact implies that the complementary subspace of $\LS_\multii$ inside $\HWsp_\multii$ is orthogonal to 
$\LS_\multii$, as stated in~(\ref{DecompositionOfLandLperp}--\ref{DecompositionOfHwithLandLperp}) below. 

\begin{restatable}{lem}{RadicalInclusionLem}  \label{RadicalInclusionLem}
Suppose $\max \multii < \pmin(q)$. 
\begin{enumerate}
\itemcolor{red}

\item \label{RadicalInclusionLemItem1}
We have $\rad \{ \Sing_\alpha \, | \, \alpha \in \LS_\multii \} \subset \rad \HWsp_\multii$.

\item \label{RadicalInclusionLemItem2}
For each $s \in \DefectSet_\multii$, we have $\smash{\rad \{ \Sing_\alpha \, | \, \alpha \in \LS_\multii\super{s} \}\subset \rad \HWsp_\multii\super{s}}$.
\end{enumerate}
Similarly, this lemma holds after the symbolic replacements
$\Sing_\alpha \mapsto \SingBar_{\alphaBar}$, $\alpha \mapsto \alphaBar$, $\LS \mapsto \LSBar$, and $\HWsp \mapsto \HWspBar$.
\end{restatable}

\begin{proof}
We postpone the proof of this lemma to the next section~\ref{subsec: radical embedding}.
\end{proof}

\begin{prop} \label{QuotientProp}
Suppose $\max \multii < \pmin(q)$. 
The following hold:
\begin{enumerate}
\itemcolor{red}

\item \label{QuotientCorItem1}
The map $\alpha \mapsto \Sing_\alpha$ 
induces the following isomorphism of $\TL_\multii(\nu)$-modules:
\begin{align} \label{QuotientCorIsom}
\Quo_\multii \isom 
\frac{\CModule{\HWsp_\multii}{\TL}}{\CModule{\{ \Sing_\alpha \, | \, \alpha \in \LS_\multii \}\superscr{\perp}}{\TL}} .
\end{align}

\item \label{QuotientCorItem2}
For each $s \in \DefectSet_\multii$, the map $\alpha \mapsto \Sing_\alpha$ 
induces the following isomorphism of $\TL_\multii(\nu)$-modules:
\begin{align}
\Quo_\multii\super{s} \isom 
\frac{\CModule{\HWsp_\multii\super{s}}{\TL}}{\CModule{\{ \Sing_\alpha \, | \, \alpha \in \LS\super{s}_\multii \}\superscr{\perp}}{\TL}} .
\end{align}

\item \label{QuotientCorItem3} 
The induced isomorphism of~\eqref{QuotientCorIsom} 
respects $s$-grading~\textnormal{(\ref{HWDirectSumPM},~\ref{QuoSp},~\ref{LSPerpDirSum})}, i.e., 
we have the following direct-sum decomposition of $\TL_\multii(\nu)$-modules:
\begin{align}  \label{QuotientCorItem3DirSum}
\Quo_\multii \isom 
\frac{\CModule{\HWsp_\multii}{\TL}}{\CModule{\{ \Sing_\alpha \, | \, \alpha \in \LS_\multii \}\superscr{\perp}}{\TL}}
= \bigoplus_{s \, \in \, \DefectSet_\multii}
\frac{\CModule{\HWsp_\multii\super{s}}{\TL}}{\CModule{\{ \Sing_\alpha \, | \, \alpha \in \LS\super{s}_\multii \}\superscr{\perp}}{\TL}} 
\isom  \bigoplus_{s \, \in \, \DefectSet_\multii} \Quo_\multii\super{s} .
\end{align}
\end{enumerate}
Similarly, this proposition holds after the symbolic replacements
\begin{align} \label{QuotientCorReplace}
\alpha \mapsto \alphaBar,  \qquad
\Sing_\alpha \mapsto \SingBar_{\alphaBar} , \qquad 
\Quo \mapsto \QuoBar, \qquad 
\HWsp \mapsto \HWspBar, 
\qquad \textnormal{and} \qquad 
\LS \mapsto \LSBar .
\end{align}
\end{prop}

\begin{proof}
We only consider the case of $\Quo_\multii$; the assertions for $\QuoBar_\multii$ can be proven similarly.
The desired induced map is $[\alpha] \mapsto [\Sing_\alpha]$. 
To see that this map is well-defined, using item~\ref{Lmap2Item1} of proposition~\ref{HWspLem2}, we observe that 
\begin{align} 
[\beta] = [\gamma]
\qquad \overset{\eqref{QuoSp}}{\Longleftrightarrow} \qquad
\beta - \gamma \in \rad \LS_\multii
\qquad \underset{\eqref{LSEmbRad}}{\overset{\eqref{LSPerpDefn}}{\Longleftrightarrow}} \qquad
\Sing_\beta - \Sing_\gamma \in \{ \Sing_\alpha \, | \, \alpha \in \LS_\multii \}\superscr{\perp}
\qquad \Longleftrightarrow \qquad
[\Sing_\beta] = [\Sing_\gamma] .
\end{align}
To prove item~\ref{QuotientCorItem1}, 
we must show that the map $[\alpha] \mapsto [\Sing_\alpha]$ is an isomorphism of $\TL_\multii(\nu)$-modules
from $\LS_\multii$ to the right side of~\eqref{QuotientCorIsom}. 
First, it is a homomorphism of $\TL_\multii(\nu)$-modules by construction.
Second, it is a linear injection because
\begin{align}
[\Sing_\beta] = 0
\qquad \Longleftrightarrow \qquad
\Sing_\beta \in \{ \Sing_\alpha \, | \, \alpha \in \LS_\multii \}\superscr{\perp}
\qquad \underset{\eqref{LSEmbRad}}{\overset{\eqref{LSPerpDefn}}{\Longleftrightarrow}} \qquad
\beta \in \rad \LS_\multii
\qquad \overset{\eqref{QuoSp}}{\Longleftrightarrow} \qquad
[\beta] = 0 .
\end{align}
To prove that the map $[\alpha] \mapsto [\Sing_\alpha]$ is surjective, 
we use lemma~\ref{BigLinAlgLem} from appendix~\ref{LinAlgAppAux} 
(whose assumptions are guaranteed by item~\ref{RadicalInclusionLemItem1} of lemma~\ref{RadicalInclusionLem})
to the subspaces $\{ \Sing_\alpha \, | \, \alpha \in \LS_\multii \} \subset  \HWsp_\multii$ 
and $\{ \SingBar_{\alphaBar} \, | \, \alphaBar \in \LSBar_\multii \} \subset  \HWspBar_\multii$ 
to deduce that there exist subspaces $\mathsf{W}_1, \mathsf{W}_2 \subset \HWsp_\multii$ 
such that
\begin{align} \label{DecompositionOfLandLperp}
\{ \Sing_\alpha \, | \, \alpha \in \LS_\multii \} = \mathsf{W}_1 \oplus \rad \{ \Sing_\alpha \, | \, \alpha \in \LS_\multii \} 
\qquad \qquad 
\{ \Sing_\alpha \, | \, \alpha \in \LS_\multii \}\superscr{\perp} =  \rad \{ \Sing_\alpha \, | \, \alpha \in \LS_\multii \} \oplus \mathsf{W}_2 ,
\end{align}
and in particular,
\begin{align} \label{DecompositionOfHwithLandLperp}
\HWsp_\multii = \mathsf{W}_1 \oplus \rad \{ \Sing_\alpha \, | \, \alpha \in \LS_\multii \} \oplus \mathsf{W}_2 .
\end{align}
From~(\ref{DecompositionOfLandLperp},~\ref{DecompositionOfHwithLandLperp}), we now see that
for each vector 
\begin{align}
[v] \in \frac{\HWsp_\multii}{\{ \Sing_\alpha \, | \, \alpha \in \LS_\multii \}\superscr{\perp}} ,
\end{align}
there exists a valenced link state $\alpha \in \LS_\multii$ such that $[v] = [\Sing_\alpha] \in \Quo_\multii$.
This shows that the map $[\alpha] \mapsto [\Sing_\alpha]$ is surjective from $\LS_\multii$ to the right side of~\eqref{QuotientCorIsom}
finishing the proof of item~\ref{QuotientCorItem1}. 
Similar work shows item~\ref{QuotientCorItem2}, and 
the $s$-grading preservation asserted in item~\ref{QuotientCorItem3} then
immediately follows from items~\ref{QuotientCorItem1}--\ref{QuotientCorItem2} combined with
item~\ref{Lmap2Item3} of proposition~\ref{HWspLem2} and $s$-grading~\eqref{HWDirectSumPM} of $\HWsp_\multii$.
The corresponding statements with replacements~\eqref{QuotientCorReplace} and be proven similarly.
\end{proof}

\subsection{Link state radical inside the highest-weight vector space}
\label{subsec: radical embedding}

In this section, we consider the following images of the highest-weight vector spaces
under the embedding $\Embedding_\multii$~\eqref{Composites}:  
\begin{align} 
\mathsf{N}\super{s}_\multii := 
\big\{ \Embedding_\multii(v) \,\big|\, v \in \HWsp_\multii\super{s} \big\} 
\qquad\qquad \text{and} \qquad\qquad
\overbarStraight{\mathsf{N}}\super{s}_\multii := 
\big\{ \EmbeddingBar_\multii(\overbarStraight{v}) \,\big|\, \overbarStraight{v} \in \HWspBar_\multii\super{s} \big\}  .
\end{align}
By item~\ref{biformitem4} of lemma~\ref{biformPropertyLem},
the embedding $\Embedding_\multii$ preserves the bilinear pairing, so we have 
\begin{align} 
\rad \mathsf{N}\super{s}_\multii
\overset{\eqref{SPBiFormNewEmbed}}{=} 
\big\{ \Embedding_\multii(v) \,\big|\, v \in \rad \HWsp_\multii\super{s} \big\} .
\end{align}

\begin{lem} \label{EmbProjectLem} 
Suppose $\max \multii < \pmin(q)$.  For each $s \in \DefectSet_\Summed$, we have
\begin{align} \label{EmbProject}
\rad \mathsf{N}\super{s}_\multii = \mathsf{N}\super{s}_\multii \cap \rad \HWsp_{\Summed_\multii}\super{s} 
= \big\{ \Projection_\multii(v) \,\big|\, v \in \rad \HWsp_{\Summed_\multii}\super{s} \big\}. 
\end{align}
Similarly, this lemma holds after the symbolic replacements $\mathsf{N} \mapsto \overbarStraight{\mathsf{N}}$, $\HWsp \mapsto \HWspBar$, 
and $\Projection \mapsto \ProjectionBar$.
\end{lem}

\begin{proof} 
This lemma can be proven similarly as~\cite[lemma~\red{B.4}]{fp3a}, by making the symbolic replacements 
$\PS \mapsto \mathsf{N}$, $\LS \mapsto \HWsp$, and $\WJProj \mapsto \Projection$; 
or $\PS \mapsto \overbarStraight{\mathsf{N}}$, 
$\LS \mapsto \HWspBar$, and $\WJProj \mapsto \ProjectionBar$.
We leave the details to the reader.
\end{proof}

\begin{cor} \label{RadicalCoro}
Suppose $\max \multii < \pmin(q)$.  For each $s \in \DefectSet_n$, we have 
\begin{align} \label{RadicalCoroIdentity}
\rad \HWsp_\multii\super{s} 
= \big\{ \Projectionhat_\multii(v) \,\big|\, v \in \rad \HWsp_{n_\multii}\super{s} \big\}. 
\end{align}
Similarly, this corollary holds after the symbolic replacements $\HWsp \mapsto \HWspBar$ and $\Projectionhat \mapsto \ProjectionhatBar$.
\end{cor}

\begin{proof} 
This corollary can be proven similarly as~\cite[corollary~\red{B.5}]{fp3a}, by making the symbolic replacements 
$\LS \mapsto \HWsp$ and $\WJProjHat \mapsto \Projectionhat$;
or $\LS \mapsto \HWspBar$ and $\WJProjHat \mapsto \ProjectionhatBar$.
We leave the details to the reader.
\end{proof}

\begin{lem} \label{RadicalInclusionLemN} \
\begin{enumerate}
\itemcolor{red}

\item \label{RadicalInclusionLemNItem1}
We have $\rad \{ \Sing_\alpha \, | \, \alpha \in \LS_n \} \subset \rad \HWsp_n$.

\item \label{RadicalInclusionLemNItem2}
For each $s \in \DefectSet_n$, we have $\smash{\rad \{ \Sing_\alpha \, | \, \alpha \in \LS_n\super{s} \}\subset \rad \HWsp_n\super{s}}$.

\end{enumerate}
Similarly, this lemma holds after the symbolic replacements
$\Sing \mapsto \SingBar$, $\alpha \mapsto \alphaBar$, $\LS \mapsto \LSBar$, and $\HWsp \mapsto \HWspBar$.
\end{lem}

\begin{proof}
Item~\ref{RadicalInclusionLemNItem1} follows from item~\ref{RadicalInclusionLemNItem2}
with direct-sum decompositions~(\ref{LSradicalDirSum},~\ref{HradicalDirSum}).
Therefore, it suffices to prove item~\ref{RadicalInclusionLemNItem2}. 
To begin, we note that 
\begin{align} \label{RadRAD}
\big\{ \Sing_{\alpha} \,\big|\, \alpha \in \rad \LS_n\super{s} \big\} 
\overset{\eqref{LSEmbRad}}{=} 
\rad \big\{ \Sing_{\alpha} \,\big|\, \alpha \in \LS_n\super{s} \big\} ,
\end{align}
for any $s \in \DefectSet_n$.
We see from~\eqref{RadRAD} that if $\smash{\rad \LS_n\super{s}} = \{0\}$, then 
the right side of~\eqref{RadRAD} equals zero, 
so the assertion in item~\ref{RadicalInclusionLemNItem2} is trivial.
Thus, without loss of generality, we assume that 
$\smash{\rad \LS_n\super{s}} \neq \{0\}$.
In this case, we have 
\begin{align} \label{SBelongsTo}
\rad \LS_n\super{s} \neq \{0\}
& \qquad \qquad \Longrightarrow \qquad \qquad
\begin{cases}
\text{$1 < \pmin(q) < \infty$} & \text{by~\cite[corollary~\red{5.2}]{fp3a},} \\
\text{$\pmin(q) \nmid (s+1)$} & \text{by~\cite[lemma~\red{5.3}]{fp3a}.}
\end{cases}
\end{align}

We prove item~\ref{RadicalInclusionLemNItem2} by
induction on $n \in \bZpos$. In the initial case $n=1$, there is nothing to prove, 
because both radicals $\smash{\rad \{ \Sing_{\alpha} \,|\, \alpha \in \LS_1\super{1} \}}$
and $\smash{\rad \HWsp_1\super{1}}$ 
are trivial. 
Thus, we assume that, for some $n \geq 2$ and for all $r \in \DefectSet_{n-1}$, 
\begin{align} \label{RadRadIndAss}
\rad \big\{ \Sing_{\beta} \,\big|\, \beta \in \LS_{n-1}\super{r} \big\}  \subset \rad \HWsp_{n-1}\super{r} .
\end{align}
Recalling from~\cite[definition~\red{4.3}]{fp3a}  
(see also~\eqref{ValencedLinkStateDef} in section~\ref{CoBloGraphical})
the definition of the trivalent link state 
$\smash{\hcancel{\alpha} \in \LS_n\super{s}}$ 
associated to the link pattern $\alpha \in \smash{\LP_n\super{s}}$,
we know by~\cite[proposition~\red{5.7}]{fp3a} that 
$\smash{\rad \LS_n\super{s}}$ has a basis of the form  
\begin{align} \label{BigTail} 
\big\{ \hcancel{\alpha} \,\big|\, \alpha \in \smash{\LP_n\super{s}}, \, \textnormal{tail}(\alpha) \in \smash{\mathsf{R}_n\super{s}} \big\} ,
\end{align}
where $\textnormal{tail}(\alpha)$ is defined in~\cite[equations~(\red{4.46}--\red{4.47}),~section~\red{4.A}]{fp3a}
and the set $\smash{\mathsf{R}_n\super{s}}$ of ``radical tails" is defined in~\cite[beneath equation~(\red{5.2}), section~\red{5.A}]{fp3a}.
(We refer to~\cite[sections~\red{4}--\red{5}]{fp3a} for more details.)
In particular, because the map $\alpha \mapsto \Sing_\alpha$ is a linear injection by item~\ref{Lmap2Item2} of proposition~\ref{HWspLem2},
the radical~\eqref{RadRAD} has a basis of the form  
\begin{align} \label{BigTailForVectors} 
\big\{ \Sing_{\scaleobj{0.85}{\hcancel{\alpha}}} \,\big|\, \alpha \in \smash{\LP_n\super{s}}, \, \textnormal{tail}(\alpha) \in \smash{\mathsf{R}_n\super{s}} \big\} .
\end{align}
Hence, to finish the induction step,
it suffices to show that all of the basis vectors $\Sing_{\scaleobj{0.85}{\hcancel{\alpha}}}$ in~\eqref{BigTailForVectors}
belong to $\smash{\rad \HWsp_n\super{s}}$.

To establish this, we will derive a recursive formula for 
$\SPBiForm{\overbarStraight{v}}{\Sing_{\scaleobj{0.85}{\hcancel{\alpha}}}}$ 
for arbitrary vectors $\smash{\overbarStraight{v} \in \HWspBar_n\super{s}}$ 
and $\Sing_{\scaleobj{0.85}{\hcancel{\alpha}}} $ in the basis~\eqref{BigTailForVectors}.
To begin, using item~\ref{HWform1GenBarItem3} of lemma~\ref{HWform1GenLem}, we write $\overbarStraight{v}$ in the form
\begin{align}  \label{GenFormForApp}
\overbarStraight{v} \overset{\eqref{HWform1GenBar3}}{=} 
\overbarStraight{v}_0 \otimes \FundBasisBar_1 + 
\overbarStraight{v}_1 \otimes \FundBasisBar_0 ,
\qquad \textnormal{where} \qquad 
\begin{cases}
\overbarStraight{v}_0 \in \HWspBar_{n-1}\super{s + 1} , \\
\overbarStraight{v}_1 \in \KspBar_{n-1}\super{s - 1} , \\
\end{cases}
\qquad \textnormal{and} \qquad 
\overbarStraight{v}_1 .F = - q^{-s - 1} \, \overbarStraight{v}_0 .
\end{align}
Next, we derive a similar recursive formula for the vector $\Sing_{\scaleobj{0.85}{\hcancel{\alpha}}}$ in the set~\eqref{BigTailForVectors}.
In general, by lemma~\ref{DescLem2}, we have 
\begin{align} \label{TrivLS}
\hcancel{\alpha} \quad = \quad \vcenter{\hbox{\includegraphics[scale=0.275]{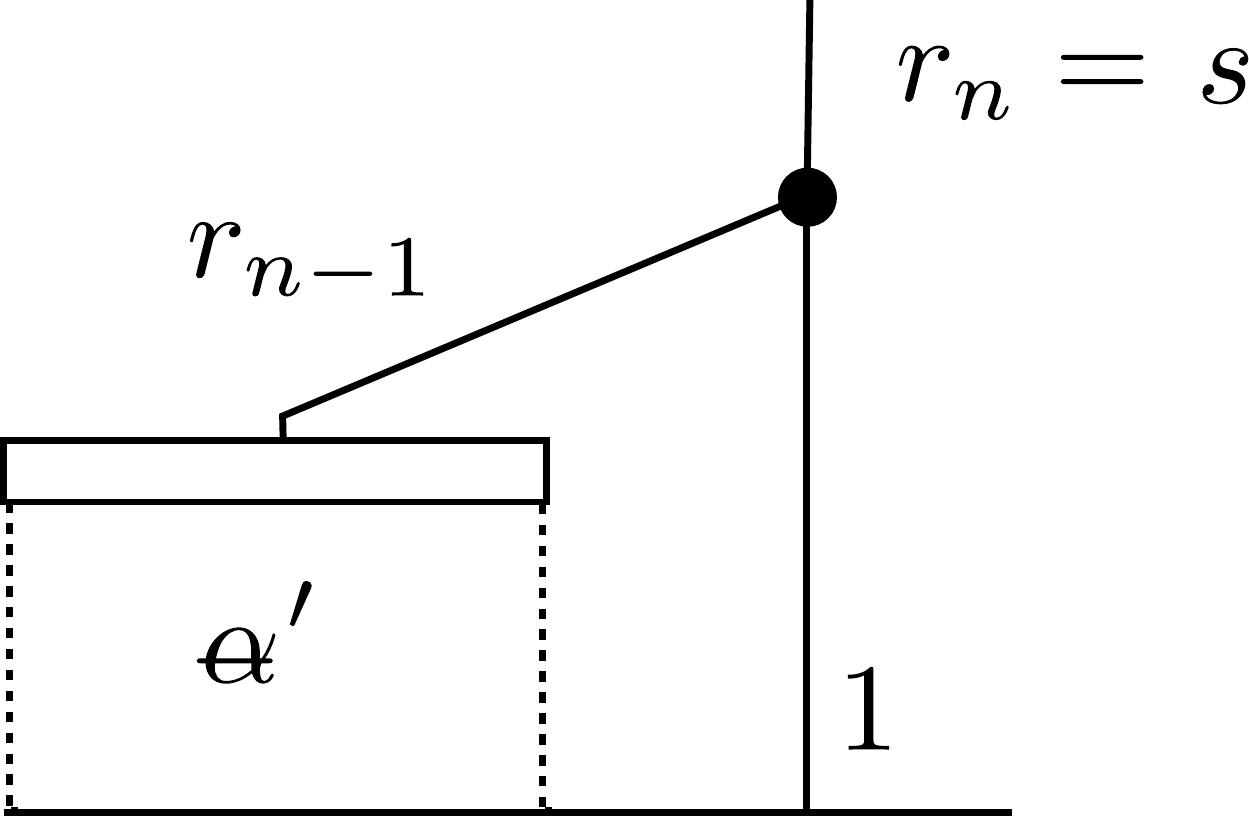}}} 
\qquad \qquad \overset{\eqref{DescDiagram2}}{\Longrightarrow} \qquad \qquad
\Sing_{\scaleobj{0.85}{\hcancel{\alpha}}} \quad = \quad
\vcenter{\hbox{\includegraphics[scale=0.275]{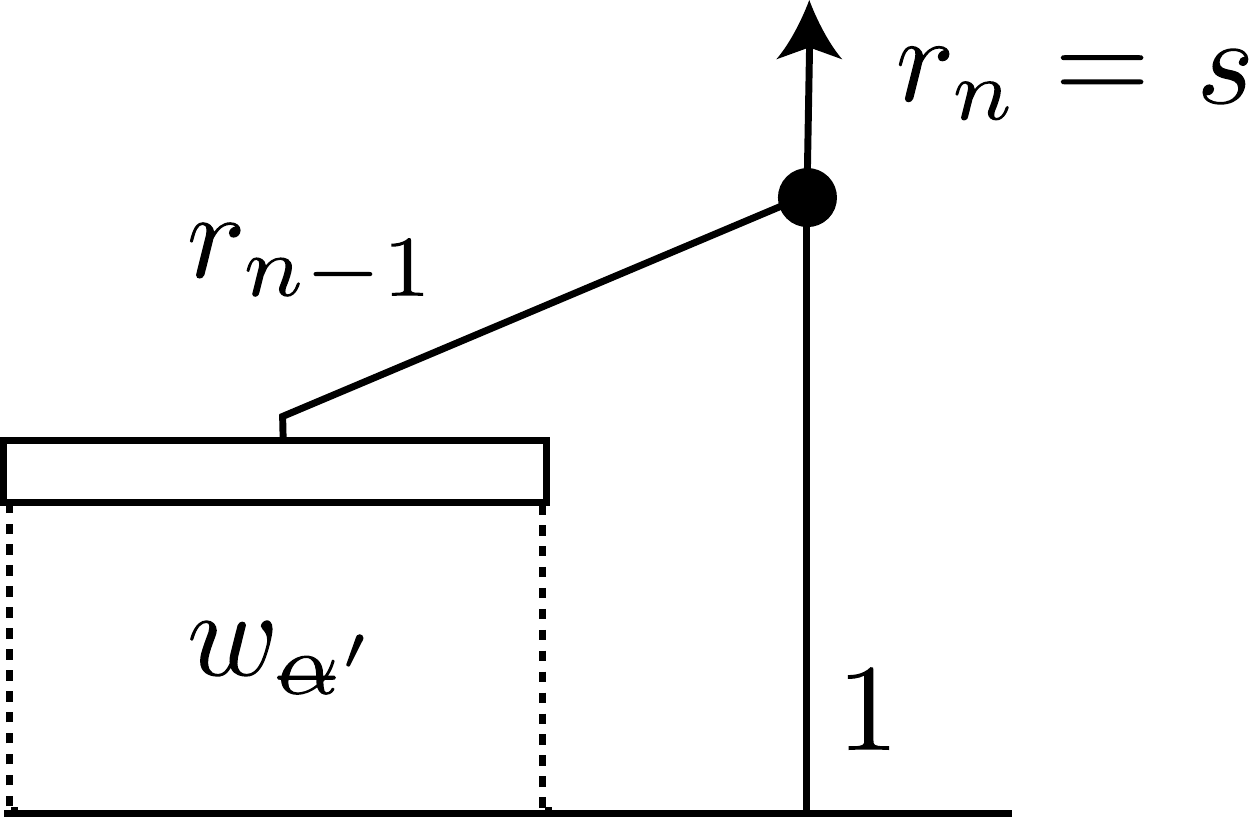}}} 
\end{align} 
for some link pattern $\alpha' \in \smash{\LS_{n-1}\super{r_{n-1}}}$. 
Furthermore, by~\eqref{WalkHeights}, the penultimate height in~\eqref{TrivLS} equals 
\begin{align}
r_{n-1} \overset{\eqref{WalkHeights}}{\in} \DefectSet\sub{1,s} \overset{\eqref{SpecialDefSet}}{=} \{s-1, s+1\} .
\end{align}
We consider these two cases separately:
\begin{enumerate}[leftmargin=*]
\itemcolor{red}
\item[(a):] When $r_n = s = r_{n-1} - 1$, the rightmost closed three-vertex~\eqref{3vertex1} in $\hcancel{\alpha}$ reads
\begin{align} 
\quad  & \vcenter{\hbox{\includegraphics[scale=0.275]{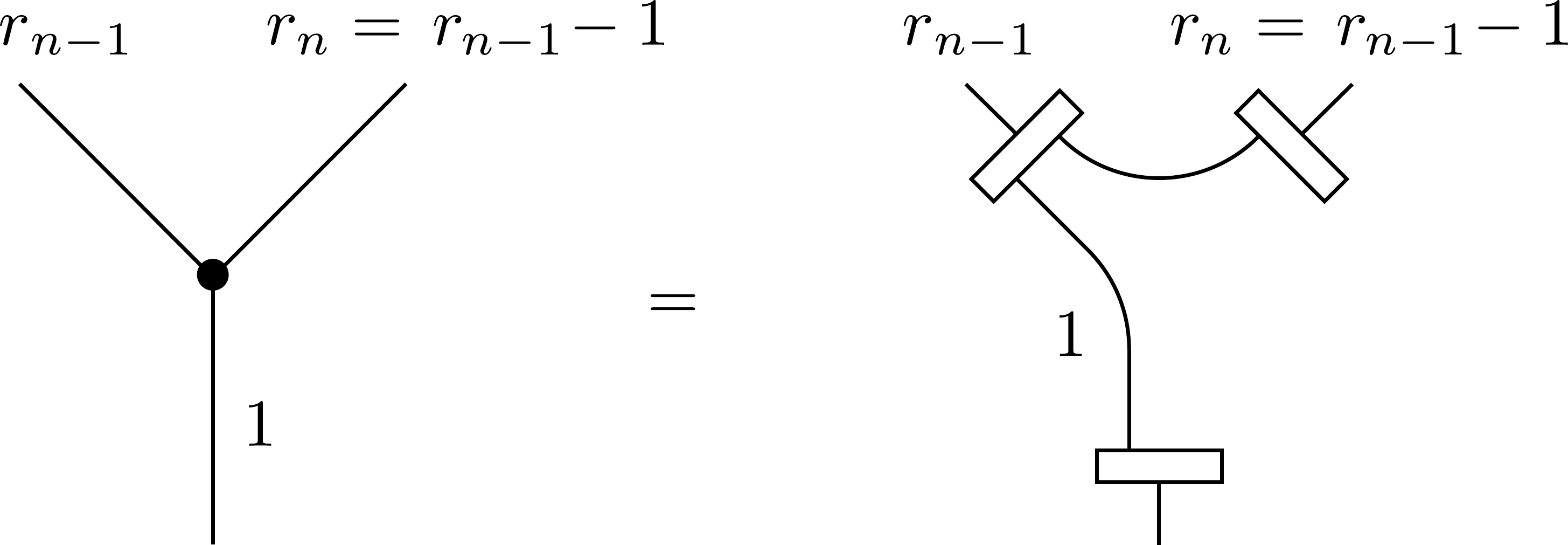} .}} 
\end{align} 
Using property~\eqref{ProjectorID1}, we absorb the top projector box of 
size $r_n$ into the top projector box of size $r_{n-1}$.
Bending the resulting diagram, orienting its defects, and using
formulas~(\ref{singletDiagramNotation},~\ref{singletLinkDiagram}) 
and lemma~\ref{DescLem2}, we arrive with
\begin{align} 
\nonumber
\quad &  \vcenter{\hbox{\includegraphics[scale=0.275]{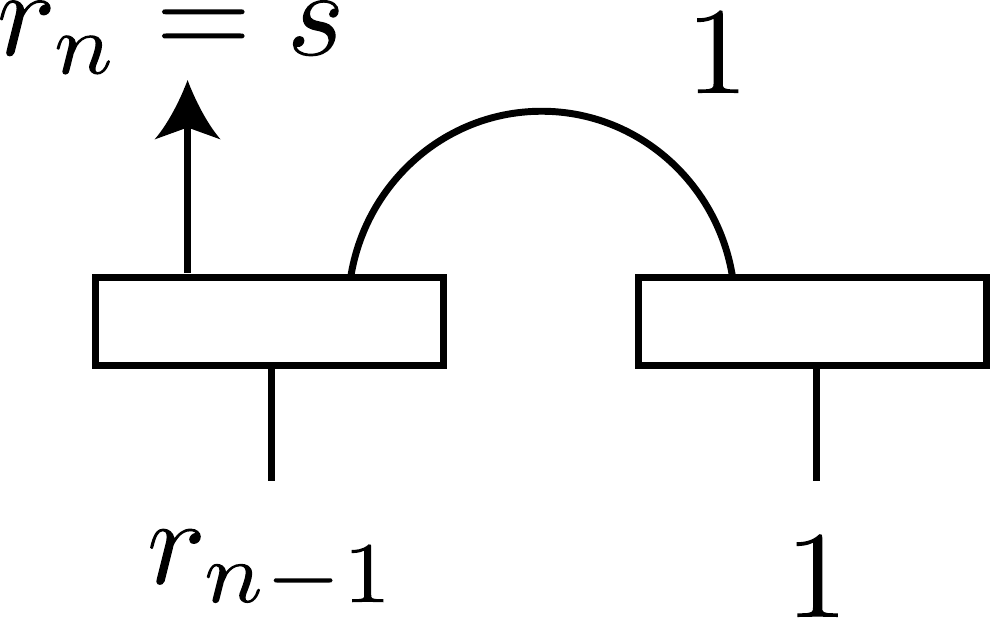}}} \\[1em]
\nonumber
\quad &  \underset{\eqref{DescDiagram2}}{\overset{\textnormal{(\ref{singletDiagramNotation}, \ref{singletLinkDiagram})}}{=}} \quad
\ii q^{1/2} \,\, \times \,\, 
\vcenter{\hbox{\includegraphics[scale=0.275]{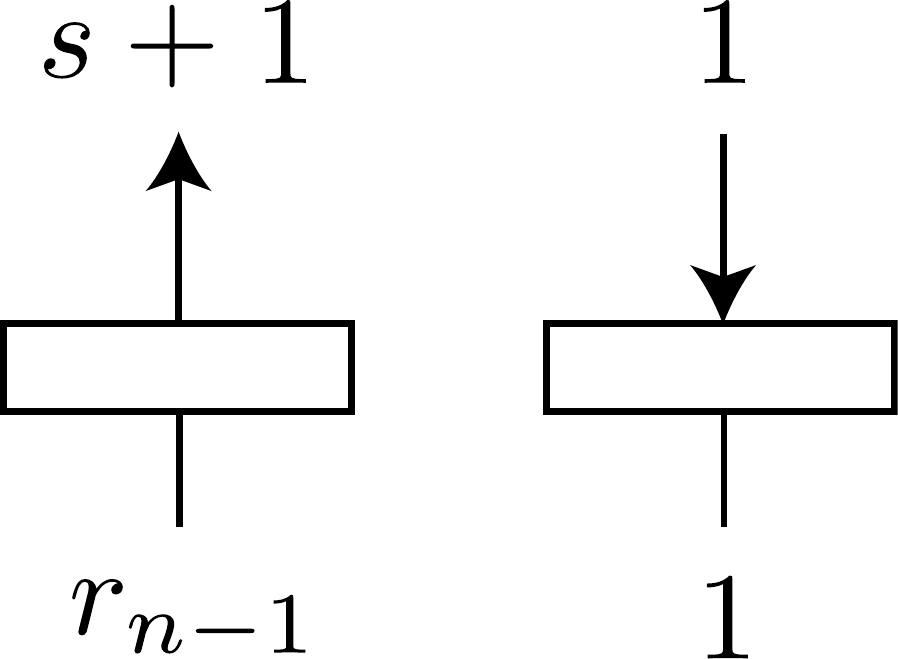} }} 
\quad - \quad 
\ii q^{-1/2} \,\, \times \,\, 
\vcenter{\hbox{\includegraphics[scale=0.275]{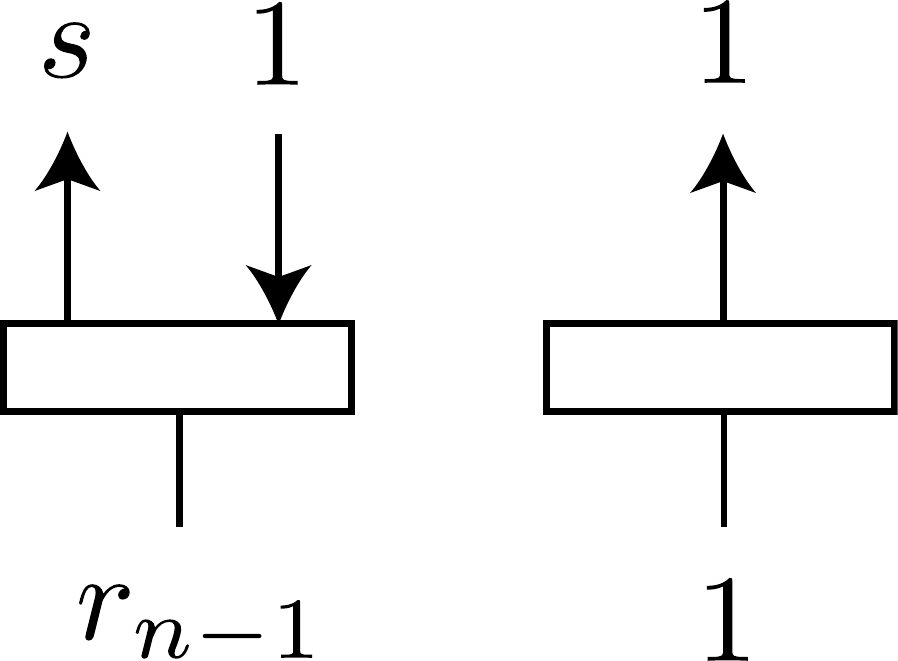} }}  \\[1em]
\label{WalphaTail2}
\quad & \underset{\hphantom{\textnormal{(\ref{singletDiagramNotation}, \ref{singletLinkDiagram})}}}{\overset{\eqref{DescDiagram2}}{=}} \quad
\ii q^{1/2} \,\, \times \,\, 
\vcenter{\hbox{\includegraphics[scale=0.275]{Figures/e-TrivLSVec1.pdf} }} 
\quad - \quad 
\ii q^{-1/2} \frac{[s]!}{[s+1]!}  \,\, \times \,\, 
\left( F . \,\,  \vcenter{\hbox{\includegraphics[scale=0.275]{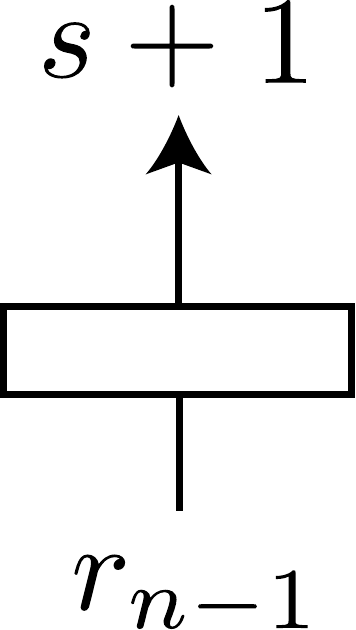}}} \; \right) \; \otimes \; 
\; \vcenter{\hbox{\includegraphics[scale=0.275]{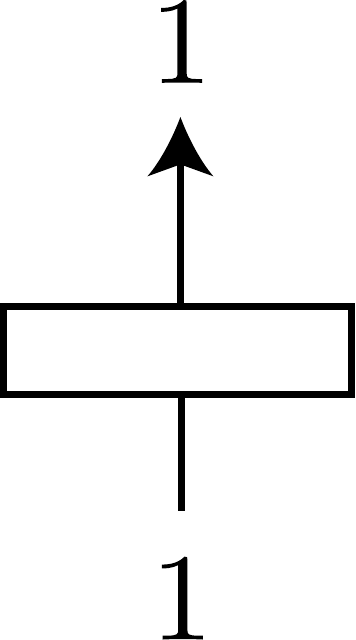} .}}   
\end{align}
After inserting~\eqref{WalphaTail2}
into the diagram~\eqref{TrivLS} for $\Sing_{\scaleobj{0.85}{\hcancel{\alpha}}}$, 
we arrive with the following recursive formula  for $\Sing_{\scaleobj{0.85}{\hcancel{\alpha}}}$:
\begin{align} \label{GenFormForAlphaApp}
\Sing_{\scaleobj{0.85}{\hcancel{\alpha}}}
\underset{\eqref{WalphaTail2}}{\overset{\eqref{OrientedDefects}}{=}}
 \ii q^{1/2}  \left( \Sing_{\scaleobj{0.85}{\hcancel{\alpha}'}} \otimes \FundBasis_1 
+ \left( - \frac{q^{-1}}{[s+1]} F . \Sing_{\scaleobj{0.85}{\hcancel{\alpha}'}} \right) \otimes \FundBasis_0 \right) ,
\qquad \textnormal{where} \qquad 
\Sing_{\scaleobj{0.85}{\hcancel{\alpha}'}} \in \HWsp_{n-1}\super{s + 1} . 
\end{align}
Using lemmas~\ref{BiFormDefLem} and~\ref{biformPropertyLem}, 
we now calculate the value of $\SPBiForm{\overbarStraight{v}}{\Sing_{\scaleobj{0.85}{\hcancel{\alpha}}}}$:
\begin{align}
\nonumber
- \ii q^{-1/2} \SPBiForm{\overbarStraight{v}}{\Sing_{\scaleobj{0.85}{\hcancel{\alpha}}}}
\underset{\eqref{GenFormForAlphaApp}}{\overset{\eqref{GenFormForApp}}{=}} 
\; & \; \hphantom{+} 
\SPBiForm{\overbarStraight{v}_0 \otimes \FundBasisBar_1}{\Sing_{\scaleobj{0.85}{\hcancel{\alpha}'}}\otimes \FundBasis_1}
- \frac{q^{-1}}{[s+1]} \, \SPBiForm{\overbarStraight{v}_0 \otimes \FundBasisBar_1}{F . \Sing_{\scaleobj{0.85}{\hcancel{\alpha}'}} \otimes \FundBasis_0} \\
\nonumber
\; & + \SPBiForm{\overbarStraight{v}_1 \otimes \FundBasisBar_0}{\Sing_{\scaleobj{0.85}{\hcancel{\alpha}'}} \otimes \FundBasis_1}
- \frac{q^{-1}}{[s+1]} \SPBiForm{\overbarStraight{v}_1 \otimes \FundBasisBar_0}{F . \Sing_{\scaleobj{0.85}{\hcancel{\alpha}'}} \otimes \FundBasis_0} \\
\nonumber
\underset{\eqref{biformfactorize}}{\overset{\eqref{biformnormalization}}{=}} 
\; & \SPBiForm{\overbarStraight{v}_0}{\Sing_{\scaleobj{0.85}{\hcancel{\alpha}'}}}
- \frac{q^{-1}}{[s+1]} \SPBiForm{\overbarStraight{v}_1}{F . \Sing_{\scaleobj{0.85}{\hcancel{\alpha}'}}} \\ 
\nonumber
\overset{\eqref{biformEquivitem1form2}}{=}
\; & \SPBiForm{\overbarStraight{v}_0}{\Sing_{\scaleobj{0.85}{\hcancel{\alpha}'}}}
- \frac{q^{-1}}{[s+1]} \SPBiForm{\overbarStraight{v}_1 . F}{\Sing_{\scaleobj{0.85}{\hcancel{\alpha}'}}} \\ 
\label{Biformshouldbezero}
\overset{\eqref{GenFormForApp}}{=}
\; & \SPBiForm{\overbarStraight{v}_0}{\Sing_{\scaleobj{0.85}{\hcancel{\alpha}'}}}
+ \frac{q^{-2-s}}{[s+1]} \, \SPBiForm{\overbarStraight{v}_0}{\Sing_{\scaleobj{0.85}{\hcancel{\alpha}'}}} 
= \left( 1 + \frac{q^{-2-s}}{[s+1]} \right)
\SPBiForm{\overbarStraight{v}_0}{\Sing_{\scaleobj{0.85}{\hcancel{\alpha}'}}} .
\end{align}

\item[(b):] When $r_n = s = r_{n-1} + 1$, the rightmost closed three-vertex~\eqref{3vertex1} in $\hcancel{\alpha}$ reads
\begin{align} \label{WalphaTail4}
\quad  & \vcenter{\hbox{\includegraphics[scale=0.275]{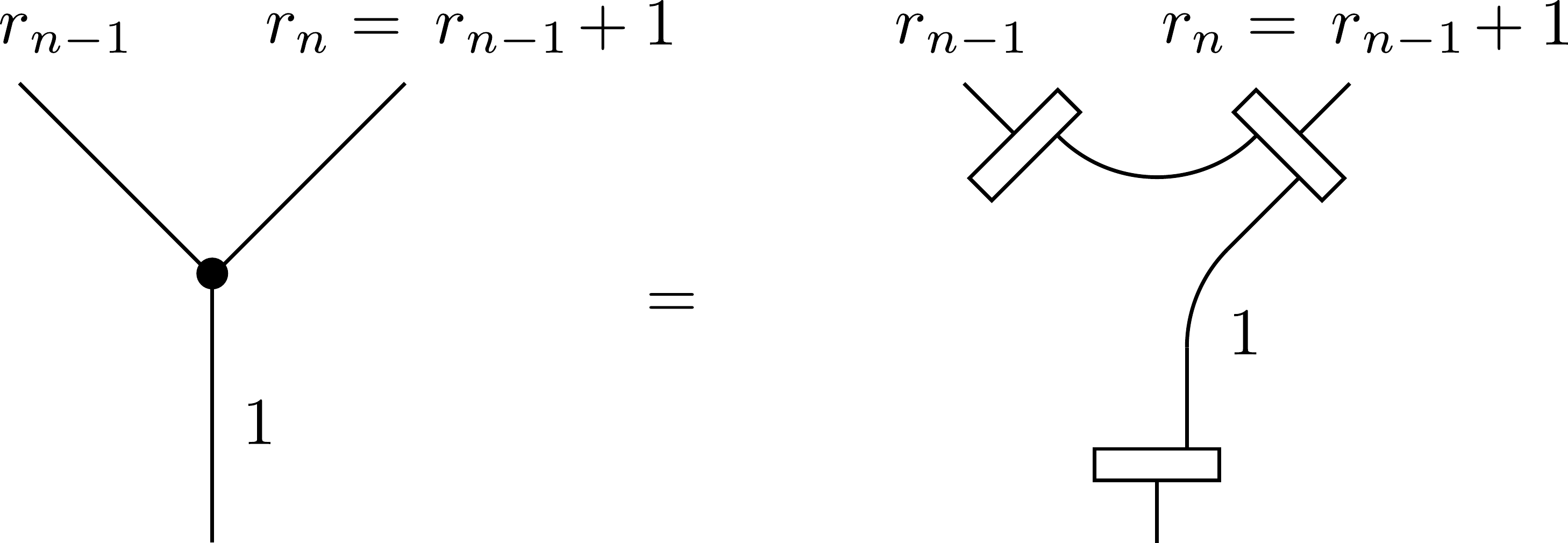}}} 
= \quad  
\vcenter{\hbox{\includegraphics[scale=0.275]{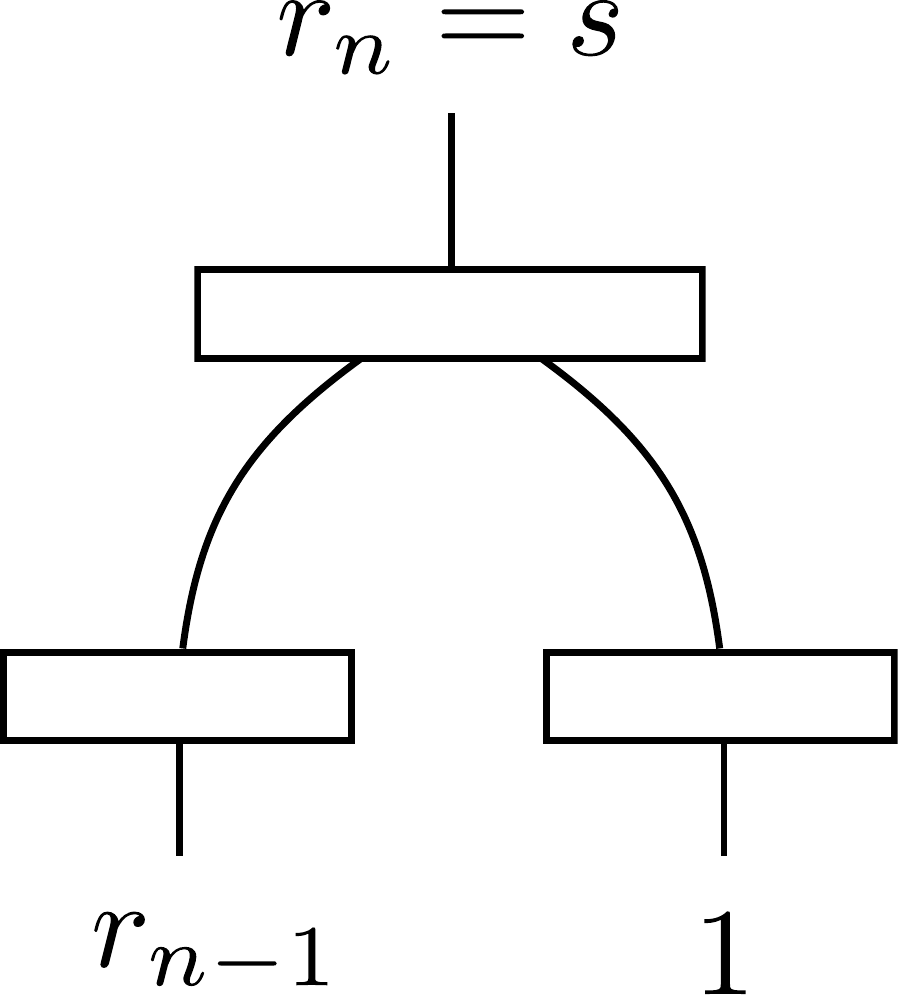}}} 
\quad  = \qquad  
\vcenter{\hbox{\includegraphics[scale=0.275]{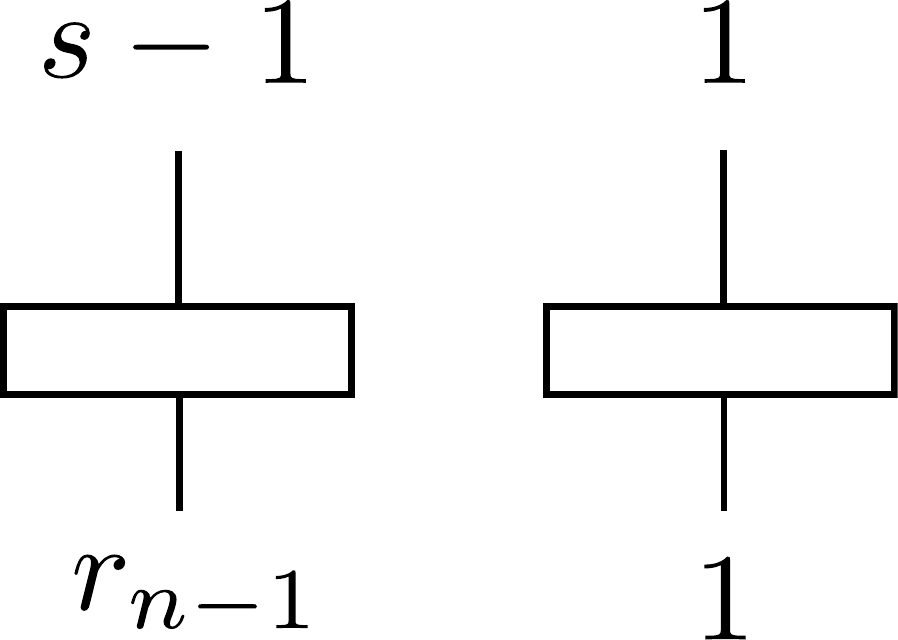} ,}} 
\end{align} 
where we used~\cite[lemma~\red{A.2}]{fp3a} to remove the top projector box.
Then, after orienting the defects in~\eqref{WalphaTail4} and inserting it into the diagram~\eqref{TrivLS} for $\Sing_{\scaleobj{0.85}{\hcancel{\alpha}}}$, 
we arrive with the following recursive formula  for $\Sing_{\scaleobj{0.85}{\hcancel{\alpha}}}$: 
\begin{align} \label{GenFormForAlphaApp2}
\Sing_{\scaleobj{0.85}{\hcancel{\alpha}}}
\underset{\eqref{WalphaTail4}}{\overset{\eqref{OrientedDefects}}{=}}
\Sing_{\scaleobj{0.85}{\hcancel{\alpha}'}} \otimes \FundBasis_0 .
\end{align}
Using lemmas~\ref{BiFormDefLem} and~\ref{biformPropertyLem}, 
we calculate $\SPBiForm{\overbarStraight{v}}{\Sing_{\scaleobj{0.85}{\hcancel{\alpha}}}}$:
\begin{align}
\label{Biformshouldbezero2}
\SPBiForm{\overbarStraight{v}}{\Sing_{\scaleobj{0.85}{\hcancel{\alpha}}}}
\underset{\eqref{GenFormForAlphaApp2}}{\overset{\eqref{GenFormForApp}}{=}} 
\SPBiForm{\overbarStraight{v}_0 \otimes \FundBasisBar_1}{\Sing_{\scaleobj{0.85}{\hcancel{\alpha}'}} \otimes \FundBasis_0}
+ \SPBiForm{\overbarStraight{v}_1  \otimes \FundBasisBar_0 }{\Sing_{\scaleobj{0.85}{\hcancel{\alpha}'}} \otimes \FundBasis_0}
\underset{\eqref{biformfactorize}}{\overset{\eqref{biformnormalization}}{=}} 
\SPBiForm{\overbarStraight{v}_1}{\Sing_{\scaleobj{0.85}{\hcancel{\alpha}'}}} .
\end{align}
\end{enumerate}

To finish, we use induction hypothesis~\eqref{RadRadIndAss} to 
show that~\eqref{Biformshouldbezero} and~\eqref{Biformshouldbezero2} vanish. 
We consider two cases:
\begin{itemize}[leftmargin=*]
\item If $\textnormal{tail}(\alpha') \in \smash{\mathsf{R}_{n-1}\super{s\pm1}}$, then
$\smash{\Sing_{\scaleobj{0.85}{\hcancel{\alpha}'}}}$ belongs to the radical 
\begin{align} 
\Span \big\{ \Sing_{\scaleobj{0.85}{\hcancel{\beta}}} \,\big|\, \beta \in \smash{\LP_n\super{s}}, \, \textnormal{tail}(\beta) \in \smash{\mathsf{R}_{n-1}\super{s\pm1}} \big\} 
\overset{\eqref{BigTailForVectors}}{=} 
\big\{ \Sing_{\beta} \,\big|\, \beta \in \rad \LS_{n-1}\super{s\pm1} \big\} 
\; \underset{\eqref{RadRadIndAss}}{\overset{\eqref{LSEmbRad}}{\subset}} \; 
\rad \HWsp_{n-1}\super{s\pm1} .
\end{align}
Hence, by induction hypothesis~\eqref{RadRadIndAss}, 
we have $\SPBiForm{\overbarStraight{v}_0}{\Sing_{\scaleobj{0.85}{\hcancel{\alpha}'}}} = 0$, 
which implies that
$\SPBiForm{\overbarStraight{v}}{\Sing_{\scaleobj{0.85}{\hcancel{\alpha}}}} = 0$
in~\eqref{Biformshouldbezero} or~\eqref{Biformshouldbezero2}.

\item If $\textnormal{tail}(\alpha') \notin \smash{\mathsf{R}_{n-1}\super{s\pm1}}$, 
then we recall the definition of radical tails from~\cite[section~\red{5.A}]{fp3a}, 
which says that stopping condition \red{1} of~\cite[definition~\red{4.3}]{fp3a} occurs 
when forming the trivalent link state $\hcancel{\alpha}$ from $\alpha$ (as illustrated in~\cite[figure~\red{4.4}]{fp3a}).
Together with~\eqref{SBelongsTo}, this implies that we necessarily have 
\begin{align} \label{valueofs}
r_n = s = (k+1) \, \pmin(q) - 2 \qquad\qquad \text{and}  \qquad\qquad  r_{n-1} = r_n + 1 = (k+1) \, \pmin(q) - 1 .
\end{align}
In other words, the above scenario~\red{(a)} necessarily occurs.
Thus, using the fact that $q^{k \, \pmin(q)} = q^{-k \, \pmin(q)}$, we obtain
\begin{align}
\nonumber
- \ii q^{-1/2} \SPBiForm{\overbarStraight{v}}{\Sing_{\scaleobj{0.85}{\hcancel{\alpha}}}}
\underset{\eqref{valueofs}}{\overset{\eqref{Biformshouldbezero}}{=}}  \; & 
\left( 1 + \frac{q^{-(k+1) \, \pmin(q)}}{[(k+1) \, \pmin(q) - 1]} \right)
\SPBiForm{\overbarStraight{v}_0}{\Sing_{\scaleobj{0.85}{\hcancel{\alpha}'}}} 
\overset{\eqref{Qinteger}}{=} 
\frac{q^{k \, \pmin(q) - 1} - q^{k \, \pmin(q) - 1} }{q^{k \, \pmin(q) - 1} - q^{-k \, \pmin(q) + 1}} \,
\SPBiForm{\overbarStraight{v}_0}{\Sing_{\scaleobj{0.85}{\hcancel{\alpha}'}}} 
= 0.
\end{align}
\end{itemize}
In conclusion, we have shown that any basis vector $\Sing_{\scaleobj{0.85}{\hcancel{\alpha}}}$ in the set~\eqref{BigTailForVectors} 
belongs to the radical $\smash{\rad \HWsp_n\super{s}}$. By linearity,
this completes the induction step, thus proving item~\ref{RadicalInclusionLemNItem2}, 
and item~\ref{RadicalInclusionLemNItem1} follows.

The statements 
with $\Sing \mapsto \SingBar$, $\alpha \mapsto \alphaBar$, $\LS \mapsto \LSBar$, and $\HWsp \mapsto \HWspBar$
can be proven similarly.
\end{proof}

\begin{lem}  \label{AnotherInclusionLem}
Suppose $\max \multii < \pmin(q)$.  For each $s \in \DefectSet_\Summed$, we have
\begin{align} \label{AnotherInclusion}
\big\{ \Embedding_\multii(\Sing_\alpha) \,\big|\, \alpha \in \rad \LS_\multii\super{s} \big\} 
= \big\{ \Sing_{\WJProj_\multii \alpha} \,\big|\, \alpha \in \rad \LS_{\Summed_\multii}\super{s} \big\} ,
\end{align}
and similarly, 
\begin{align} \label{AnotherInclusionBar}
\big\{ \EmbeddingBar_\multii(\SingBar_{\alphaBar}) \,\big|\, \alphaBar \in \rad \LSBar_\multii\super{s} \big\} 
= \big\{ \SingBar_{\alphaBar \WJProj_\multii} \,\big|\, \alphaBar \in \rad \LSBar_{\Summed_\multii}\super{s} \big\} .
\end{align}
\end{lem}

\begin{proof}
\cite[corollary~\red{B.5}]{fp3a} says that
\begin{align} \label{EmbProj22} 
\rad \LS_\multii\super{s} = \WJProjHat_\multii \rad \LS_{\Summed_\multii}\super{s} ,
\end{align} 
which together with lemma~\ref{SmoothingLem2}, corollary~\ref{CompositeProjCorHatEmb},
lemma~\ref{EmbProjLem}, and corollary~\ref{CompositeProjCor}
implies that
\begin{align}
\big\{ \Embedding_\multii(\Sing_\alpha) \,\big|\, \alpha \in \rad \LS_\multii\super{s} \big\} 
& \overset{\eqref{EmbProj22}}{=} \big\{ \Embedding_\multii(\Sing_{\WJProjHat_\multii \alpha}) \,\big|\, \alpha \in \rad \LS_{\Summed_\multii}\super{s} \big\} 
\underset{\eqref{Lcommutation}}{\overset{\eqref{CompTwoProjsHatEmb}}{=}} 
\big\{ (\Embedding_\multii \circ \Projectionhat_\multii) (\Sing_{\alpha}) \,\big|\, \alpha \in \rad \LS_{\Summed_\multii}\super{s} \big\}  \\
& \overset{\eqref{UQIdComp}}{=} 
\big\{ \Projection_\multii(\Sing_{\alpha}) \,\big|\, \alpha \in \rad \LS_{\Summed_\multii}\super{s} \big\} 
\underset{\eqref{Lcommutation}}{\overset{\eqref{CompTwoProjs}}{=}} 
\big\{ \Sing_{\WJProj_\multii \alpha} \,\big|\, \alpha \in \rad \LS_{\Summed_\multii}\super{s} \big\} .
\end{align}
This proves~\eqref{AnotherInclusion}. Identity~\eqref{AnotherInclusionBar} can be proven similarly.
\end{proof}

\RadicalInclusionLem*

\begin{proof}
Item~\ref{RadicalInclusionLemItem1} follows from item~\ref{RadicalInclusionLemItem2}
with direct-sum decompositions~(\ref{LSradicalDirSum},~\ref{HradicalDirSum}).
Therefore, it suffices to prove item~\ref{RadicalInclusionLemItem2}.
Lemma~\ref{RadicalInclusionLemN} already gives the case of $\multii = n$.
Combining it with lemma~\ref{AnotherInclusionLem}, we obtain 
\begin{align}  \label{subset}
\big\{ \Embedding_\multii(\Sing_\alpha) \,\big|\, \alpha \in \rad \LS_\multii\super{s} \big\} 
& \overset{\eqref{AnotherInclusion}}{=} 
\big\{ \Sing_{\WJProj_\multii \alpha} \,\big|\, \alpha \in \rad \LS_{\Summed_\multii}\super{s} \big\} \\
& \overset{\hphantom{\eqref{AnotherInclusion}}}{\subset} 
\big\{ \Sing_{\alpha} \,\big|\, \alpha \in \rad \LS_{\Summed_\multii}\super{s} \big\} 
\overset{\eqref{LSEmbRad}}{=} 
\rad \big\{ \Sing_{\alpha} \,\big|\, \alpha \in \LS_{\Summed_\multii}\super{s} \big\} 
\overset{\text{lem.~\ref{AnotherInclusionLem}}}{\subset}
\rad \HWsp_{\Summed_\multii}\super{s} ,
\end{align}
where to pass to the second line, we used the observation that,
because $\smash{\rad \LS_{\Summed_\multii}\super{s}}$ is a $\TL_{\Summed_\multii}(\nu)$-submodule of $\smash{\LS_{\Summed_\multii}\super{s}}$,
$\smash{\alpha \in \rad \LS_{\Summed_\multii}\super{s}}$ implies
$\smash{\WJProj_\multii \alpha \in \rad \LS_{\Summed_\multii}\super{s}}$.
Thanks to~\eqref{UQIdComp}, applying $\Projectionhat_\multii$ to both sides
and using corollary~\ref{RadicalCoro}, we obtain 
\begin{align} 
\big\{ \Sing_\alpha \,\big|\, \alpha \in \rad \LS_\multii\super{s} \big\} 
\overset{\eqref{subset}}{\subset} 
\big\{ \Projectionhat_\multii(v) \,\big|\, v \in \rad \HWsp_{\Summed_\multii}\super{s} \big\}
\overset{\eqref{RadicalCoroIdentity}}{=} 
\rad \HWsp_\multii\super{s} .
\end{align}
This proves item~\ref{RadicalInclusionLemItem2}.
The case of the subsets of $\rad \HWspBar_\multii$ can be proven similarly.
\end{proof}

\section{Higher-spin quantum Schur-Weyl duality}  
\label{HigherQSchurWeylSect}
In this final section, we prove the quantum Schur-Weyl duality theorems~\ref{HighQSchurWeylThm2} and~\ref{GeneralCommutantThm};
the latter in section~\ref{DCProofSec} and the former in section~\ref{QSWProofSec}.
To identify the commutant algebra with the valenced Temperley-Lieb algebra $\TL_\multii(\nu)$,
we first show (as a special case of proposition~\ref{PreFaithfulPropGen} concerning $\smash{\TL_\multii^\multiii}$) that the representation 
of $\TL_\multii(\nu)$ is faithful whenever defined. 
Lastly, in section~\ref{GeneratorThmCommSubSec} we prove proposition~\ref{GeneratorThmComm} concerning a generating set 
for the commutant algebra of $\Uqsltwo$: we show that it is obtained from 
$\Uqsltwo$-submodule projectors acting on consecutive tensor components of $\VecSp_\multii$.
(These operators are also closely related to the R-matrix of $\Uqsltwo$~\cite{ck, cp}, not however discussed in the present work.)

\subsection{On kernels and images of representations} \label{KerImSubSec}

In this section, we prove that the representation 
$\Trep_\multii \colon \TL_\multii(\nu) \longrightarrow \End \VecSp_\multii$
of the valenced Temperley-Lieb algebra $\TL_\multii(\nu)$ 
on the tensor product $\VecSp_\multii$ is faithful whenever defined (i.e., for all $\max \multii < \pmin(q)$).
This follows as a special case of the next proposition. 
Taking $\multii = \OneVec{n}$, 
we thus also recover a result proved for the Temperley-Lieb algebra $\TL_n(\nu)$ independently 
by P.~Martin~\cite[theorem~\red{1}]{ppm} and F.~Goodman and H.~Wenzl~\cite[theorem~\red{2.4}]{gwe}.

\begin{prop} \label{PreFaithfulPropGen} 
Suppose $\max (\multii,\multiii) < \pmin(q)$.  
The following maps are injective\textnormal{:}
\begin{align}
\Trep_\multii^{\multiii} \colon \TL_\multii^{\multiii}(\nu) \longrightarrow  \Hom ( \VecSp_\multiii, \VecSp_\multii )
\qquad \qquad \textnormal{and} \qquad \qquad
\TrepBar_\multii^{\multiii} \colon \TL_\multii^{\multiii}(\nu) \longrightarrow \Hom (\VecSpBar_\multii, \VecSpBar_\multiii ) .
\end{align}
\end{prop}

\begin{proof}
We prove the assertion for $\smash{\Trep_\multii^{\multiii}}$; the case of $\smash{\TrepBar_\multii^{\multiii}}$ is similar.
To begin, we specialize to the case of $\multii = \OneVec{n}$ and $\multiii = \OneVec{m}$ for some $m,n \in \bZpos$. 
By linearity, to show that the map $\smash{\Trep_n^m}$ is injective, 
it suffices to show that $\smash{\Trep_n^m(T)} = 0$ implies $T=0$.
For this purpose, we expand each tangle $T \in \smash{\TL_n^m(\nu)}$ according to the number $r$ of crossing links, 
\begin{align} \label{TExpanded} 
T 
= \sum_{r \, \in \, \DefectSet_n^m} T\super{r} ,
\qquad \text{where} \quad
T\super{r} 
= \sum_{\substack{\alpha  \, \in \, \LP_{n}\super{r} \\ \betaBar \, \in \, \LPBar_m\super{r}}} 
c_{\alpha,\betaBar}\super{r} \BarAction \alpha \quad \betaBar \BarAction ,
\end{align} 
$\smash{T\super{r}} \in \smash{\TL_n^m}$ being tangles with exactly $r$ crossing links, 
$\smash{\DefectSet_n^m}$ the set of all integers $r \geq 0$ for which such tangles exist,
and $\smash{c_{\alpha,\betaBar}\super{r}} \in \bC$ some coefficients. 
We suppose, towards a contradiction, that $\smash{\Trep_n^m(T)} = 0$ but $T \neq 0$, and choose $s \in \smash{\DefectSet_n^m}$ 
to be the largest number such that $\smash{c_{\alpha,\betaBar}\super{s}} \neq 0$ in~\eqref{TExpanded}  
for some pair of link patterns 
$\alpha \in \smash{\LP_n\super{s}}$ and $\betaBar \in \smash{\LPBar_m\super{s}}$. 
By the assumption that $\smash{\Trep_n^m(T)} = 0$, we have
\begin{align} \label{LongFormula1}
0 = T v
\overset{\eqref{TExpanded}}{=}
\sum_{\substack{r \, \in \, \DefectSet_n^m \\ r \, \leq \, s}} T\super{r} v
\end{align}
for any vector $v \in \VecSp_m$. 
By considering the action of $T$ on standard basis vectors 
$\FundBasis_m^\varrho \in \smash{\Ksp_m\super{s}}$ with $\varrho = \varrho_\gamma$, we have
\begin{align} \label{SmallerRAnnihilate0}
\begin{cases}
\gamma \in \smash{\LP_m\super{s}} , \\
\varrho = \varrho_\gamma
\end{cases}
\qquad \qquad 
\underset{\eqref{TExpanded}}{\overset{\eqref{TurnBackWeightZero}}{\Longrightarrow}} 
\qquad\qquad 
T\super{r} \FundBasis_m^\varrho = 0 \quad \text{ for all $r < s$,}
\end{align}
because if $r < s$, then for any link pattern $\betaBar \in \smash{\LPBar_m\super{r}}$, 
the oriented network $\betaBar \BarAction \FundBasis_m^\varrho$ has a turn-back link of $\betaBar$ joining 
two defects of $\FundBasis_m^\varrho$ with identical orientations. 
Using this observation and item~\ref{NewRidoutIdIt1} of lemma~\ref{NewRidoutLem}, we thus obtain
\begin{align} \label{LongFormula2}
0 \overset{\eqref{LongFormula1}}{=} T \FundBasis_m^\varrho 
\underset{\eqref{SmallerRAnnihilate0}}{\overset{\eqref{LongFormula1}}{=}}
T\super{s} \FundBasis_m^\varrho 
\overset{\eqref{TExpanded}}{=}
\sum_{\substack{\alpha  \, \in \, \LP_{n}\super{s} \\ \betaBar \, \in \, \LPBar_m\super{s}}} 
c_{\alpha,\betaBar}\super{s} \BarAction \alpha \quad \betaBar \BarAction \FundBasis_m^\varrho 
\overset{\eqref{NewRidoutId}}{=}
\sum_{\substack{\alpha  \, \in \, \LP_{n}\super{s} \\ \betaBar \, \in \, \LPBar_m\super{s}}} 
c_{\alpha,\betaBar}\super{s} \, \SPBiForm{\SingBar_{\betaBar}}{\FundBasis_m^\varrho} \Sing_\alpha .
\end{align}
Now, by lemma~\ref{SingletBasisIsLinIndepLem}, the collection $\{\Sing_\alpha \,|\, \alpha \in \smash{\LP_n\super{s}} \}$ is linearly independent, 
which implies the following system of equations for all link patterns 
$\alpha \in \smash{\LP_n\super{s}}$ and $\gamma \in \smash{\LP_m\super{s}}$:
\begin{align} \label{SysS}
\sum_{\betaBar \, \in \, \LPBar_m\super{s}} M_{\betaBar, \gamma} c_{\alpha,\betaBar}\super{s} = 0,
\qquad \text{where} \quad M_{\betaBar,\gamma} := \SPBiForm{\SingBar_{\betaBar}}{\FundBasis_m^{\varrho}} \text{ and } \varrho = \varrho_\gamma .
\end{align}
Corollary~\ref{CorSingletBasisExpand} implies that we can arrange the columns of the matrix 
$\smash{M_{\betaBar,\gamma}}$ in such a way that it is upper-triangular with non-vanishing diagonal elements. Therefore, we conclude that 
\begin{align}
\text{\eqref{SysS}} \qquad & \Longrightarrow \qquad 
c_{\alpha,\betaBar}\super{s} = 0 \quad
\text{for all link patterns } \alpha \in \LP_n\super{s} \text{ and } \betaBar \in \LPBar_m\super{s} .
\end{align}
But this contradicts the choice of $s$. 
Therefore, we have $T = 0$, so the map $\smash{\Trep_n^m}$ is injective.

Next, to prove that the map $\smash{\Trep_\multii^{\multiii}}$ is injective for general multiindices $\multii, \multiii \in \smash{\bZpos^\#}$ too,
we again show that $\smash{\Trep_\multii^\multiii(T)} = 0$ implies $T = 0$. 
First, by corollary~\ref{CompositeProjCorHatEmb}, 
we have
\begin{align} \label{SetToZero}
\Trep_{\Summed_\multii}^{\Summed_\multiii}(\WJEmb_\multii T \WJProjHat_\multiii)(w)
= \WJEmb_\multii T \WJProjHat_\multiii w 
\overset{\eqref{CompTwoProjsHatEmb}}{=} \Embedding_\multii(T\WJProjHat_\multiii w) \; \in \; \im \Embedding_\multii 
\end{align}
for all vectors $w \in \VecSp_{\Summed_\multiii}$. 
Because $\Projectionhat_\multii = \Embedding_\multii^{-1}$ on $\im \Embedding_\multii$ by  
item~\ref{It3} of lemma~\ref{EmbProjLem}, recalling definition~\eqref{ImultHom}, we obtain
\begin{align} \label{IsZero}
\Trep_\multii^\multiii(T) = 0 \qquad \underset{\eqref{SetToZero}}{\overset{\eqref{ImultHom}}{\Longrightarrow}} \qquad 
(\WJEmb_\multii T \WJProjHat_\multiii) \Embedding_\multiii(v) = 0 \quad \text{for all vectors } v \in \VecSp_\multiii .
\end{align}
Item~\ref{It3} of lemma~\ref{EmbProjLem} also implies that for each vector $w \in \VecSp_{\Summed_\multiii}$, 
there exists a vector $v \in \VecSp_\multiii$ for which we have $\Projection_\multiii(w) = \Embedding_\multiii(v)$. 
Furthermore, by corollary~\ref{CompositeProjCor}, this pair $v, w$ satisfies 
\begin{align}
\label{vwPair}
\WJProjHat_\multiii\Embedding_\multiii(v) = \WJProjHat_\multiii\Projection_\multiii(w) 
\overset{\eqref{CompTwoProjs}}{=} \WJProjHat_\multiii\WJProj_\multiii w 
\overset{\eqref{IdCompAndWJPhatPEmb}}{=}  
\WJProjHat_\multiii w .
\end{align}
After inserting this into~\eqref{IsZero} 
and recalling that we have already shown the injectivity of that $\smash{\Trep_{\Summed_\multii}^{\Summed_\multiii}}$, 
we conclude that
\begin{align}
\label{IsZero2}
\Trep_\multii^\multiii(T) = 0 \qquad \underset{\eqref{vwPair}}{\overset{\eqref{IsZero}}{\Longrightarrow}} \qquad 
\Trep_{\Summed_\multii}^{\Summed_\multiii} (\WJEmb_\multii T \WJProjHat_\multiii) = 0 
\qquad \Longrightarrow \qquad \WJEmb_\multii T \WJProjHat_\multiii = 0 .
\end{align}
To finish, we recall from \cite[lemma~\red{B.1}]{fp3a} that the map
$\smash{T \mapsto \WJEmb_\multii T \WJProjHat_\multiii}$
for valenced tangles $T \in \smash{\TL_\multii^\multiii}$ is a linear injection. 
Combining this with~\eqref{IsZero2}, we see that the property $\Trep_\multii^\multiii(T) = 0$ indeed implies that $T = 0$.
\end{proof}

The above result also holds for $q \in \{\pm 1\}$, with the same proof.

\bigskip

We remark that if $\Summed_\multii < \pmin(q)$, then it alternatively follows from
the faithfulness results in~\cite[corollaries~\red{3.8} and~\red{5.22}]{fp3a}
combined with proposition~\ref{HWspLem2} of the present article 
that the representation $\Trep_\multii$ 
for $\Module{\VecSp_\multii}{\TL}$ is faithful:
\begin{align}
\Summed_\multii < \pmin(q)
\qquad \underset{\textnormal{\cite[cor. \red{5.22}]{fp3a}}}{\overset{\textnormal{\cite[(\red{5.106}-\red{5.107})]{fp3a}}}{\Longrightarrow}} \qquad 
\rad \LS_\multii = \{0\}
& \qquad \overset{\textnormal{\cite[cor. \red{3.8}]{fp3a}}}{\Longrightarrow} \qquad 
\Module{\LS_\multii}{\TL} \text{ faithful} \\
& \qquad \underset{\hphantom{\textnormal{\cite[cor. \red{3.8}]{fp3a}}}}{\overset{\textnormal{prop. \ref{HWspLem2}}}{\Longrightarrow}} \qquad 
\Module{\VecSp_\multii}{\TL} \text{ faithful.}
\end{align}

\bigskip

We recall from lemma~\ref{UqHomoLem2} that the image of $\smash{\Trep_\multii^{\multiii}}$ lies in fact in 
the commutant space $\HomMod{\Uqsltwo} ( \VecSp_\multiii, \VecSp_\multii )$.
In the next section, we show that if $\max (\Summed_\multii, \Summed_\multiii) < \pmin(q)$, 
then this image fills the whole commutant space (theorem~\ref{GeneralCommutantThm}).

Conversely, we consider the left representation 
\begin{align} \label{LeftRegRep}
\LeftRegRep_\multii \colon \Uqsltwo \longrightarrow \End{\VecSp_\multii} ,
\qquad \qquad 
\LeftRegRep_\multii := (\LeftRegRep\sub{\sIndex_1} \otimes \LeftRegRep\sub{\sIndex_2} \otimes \dotsm \otimes \LeftRegRep\sub{\sIndex_{\np_\multii}}) \circ \Delta\super{\np_\multii} 
\end{align}
associated to the module $\Module{\VecSp_\multii}{\Uqsltwo}$. 
We similarly define the right representation $\RightRegRep_\multii$ 
associated to the module $\RModule{\VecSpBar_\multii}{\Uqsltwo}$, 
and the corresponding left and right representations $\LeftTwistRep_\multii$ and $\RightTwistRep_\multii$
for $\Module{\VecSp_\multii}{\UqsltwoBar}$ and $\RModule{\VecSpBar_\multii}{\UqsltwoBar}$,

Theorem~\ref{GeneralCommutantThm} also shows that if $\max \Summed_\multii < \pmin(q)$, 
then the image of the representation $\LeftRegRep_\multii \colon \Uqsltwo \longrightarrow \End \VecSp_\multii$ 
constitutes all operators which commute with the $\TL_\multii(\nu)$-action on $\Module{\VecSp_\multii}{\TL}$. 
However, the representation 
$\LeftRegRep_\multii$ is not faithful (it cannot be, as $\Uqsltwo$ is infinite-dimensional),
and its image 
is isomorphic to a finite-dimensional 
quotient of the quantum group $\Uqsltwo$, called a $q$-Schur algebra.
In the case of $\multii = \OneVec{n}$, the $q$-Schur algebra $\LeftRegRep_n(\Uqsltwo)$ appears,
e.g., in~\cite{dj}.


When $\multii = (s)$, it is straightforward to find explicit formulas for the image and kernel of $\LeftRegRep\sub{s}$.
We collect these formulas here. 
For this purpose, we define the projection operators $\smash{G\super{s}_\ell}$, for $\ell \in \{0,1,\ldots,s\}$, 
\begin{align} 
\label{GradeElement}
G\super{s}_\ell := \prod_{\substack{0 \, \leq \, j \, \leq \, s \\ j \, \neq \, \ell}} \frac{K - q^{s - 2j}}{q^{s - 2\ell} - q^{s - 2j}} ,
\qquad \qquad 
G\super{s}_\ell \Basis_k\super{s} = \delta_{k,\ell} \Basis_\ell\super{s} .
\end{align}
Also, for all $\ell \in \{0, 1, \ldots, s\}$, $k, m \in \bZnn$ with $m \leq k$ and $n \in \bZ$, we define the constants
\begin{align}
A_{\ell, k,n,m}\super{s} & :=  \one\{ k \leq \ell + m \leq s\} \, q^{-n(s - 2\ell - 2m)} \, \frac{ [k - m]!^2}{[k]!^2} \frac{ \qbin{\ell}{\ell - k + m } \qbin{s - \ell + k - m}{s - \ell}}{ \qbin{\ell + m}{\ell + m - k} \qbin{s -\ell - m + k}{ s - \ell - m} }, \\
B_{\ell, k,n,m}\super{s} & := \one\{ m \leq \ell + k \leq s\} \, q^{-n(s - 2\ell - 2k)} \, \frac{1}{[m]!^2} \frac{1}{\qbin{\ell + k}{ \ell + k - m} \qbin{s - \ell - k + m}{s - \ell - k}}.
\end{align}

\begin{lem} \label{KernelImageLem}
Suppose $s < \pmin(q)$. Then, the following hold:
\begin{enumerate} 
\itemcolor{red}
\item \label{RhoIt2}
The following set is a basis for the image 
$\im \LeftRegRep\sub{s}$\textnormal{:}
\begin{align}
\label{BigImg}
\bigcup_{0 \, \leq \, \ell \, \leq \, s } \big\{ \LeftRegRep\sub{s}(F^k G\super{s}_\ell) \,|\, 0 \leq k \leq s - \ell \big\} 
\cup \big\{ \LeftRegRep\sub{s}(E^k G\super{s}_\ell) \,|\, 1 \leq k \leq \ell \big\},
\end{align}
and similarly, the following set is a basis for the image 
$\im \RightRegRep\sub{s}$\textnormal{:}
\begin{align}
\label{BigImg2}
\bigcup_{0 \, \leq \, \ell \, \leq \, s } \big\{ \RightRegRep\sub{s}(G\super{s}_\ell F^k) \,\big|\, 0 \leq k \leq s - \ell \big\} 
\cup \big\{ \RightRegRep\sub{s}(G\super{s}_\ell E^k) \,\big|\, 1 \leq k \leq \ell \big\}.
\end{align}

\item \label{RhoIt4}
The following set spans the kernel $\ker \LeftRegRep\sub{s} \subset \Uqsltwo$\textnormal{:}
\begin{align} \label{BigKerSpan}
\bigcup_{\substack{0 \, \leq \, \ell \, \leq \, s \\ k \, \in \, \bZnn , \; n \, \in \, \bZ \\ 0 \, \leq \, m \, \leq \, k}}  
\left \{ \begin{array}{ll} (E^{k - m} - A_{\ell, k,n,m}\super{s} E^kK^nF^m) G\super{s}_\ell, \quad & E^kK^nF^m(G_0\super{s} + G_1\super{s} + \dotsm + G_s\super{s} - \mathbf{1}_{\Uqsltwo}) \\[5pt]
(F^{k - m} - B\super{s}_{\ell, k,n,m}E^mK^nF^k) G\super{s}_\ell, \quad & E^mK^nF^k(G_0\super{s} + G_1\super{s} + \dotsm + G_s\super{s} - \mathbf{1}_{\Uqsltwo})
\end{array} \right\},
\end{align}
and similarly, the following set spans the kernel $\ker \RightRegRep\sub{s} \subset \Uqsltwo$\textnormal{:}
\begin{align} \label{BigKerSpan2}
\bigcup_{\substack{0 \, \leq \, \ell \, \leq \, s \\ k \, \in \, \bZnn , \; n \, \in \, \bZ \\ 0 \, \leq \, m \, \leq \, k}}  
\left \{ \begin{array}{ll} G\super{s}_\ell (F^{k - m} - A_{\ell, k,n,m}\super{s} F^kK^nE^m), \quad & (G_0\super{s} + G_1\super{s} + \dotsm + G_s\super{s} - \mathbf{1}_{\Uqsltwo}) F^kK^nE^m \\[5pt]
G\super{s}_\ell (E^{k - m} - B\super{s}_{\ell, k,n,m}F^mK^nE^k), \quad & (G_0\super{s} + G_1\super{s} + \dotsm + G_s\super{s} - \mathbf{1}_{\Uqsltwo}) F^mK^nE^k
\end{array} \right\}.
\end{align}
\end{enumerate}
Similarly, this lemma holds for the left and right representations $\LeftTwistRep\sub{s}$ and $\RightTwistRep\sub{s}$ of $\UqsltwoBar$.
\end{lem}

\begin{proof} 
We prove the assertions concerning $\LeftRegRep\sub{s}$; the case of $\RightRegRep\sub{s}$ is similar.
For each element $x \in \Uqsltwo$, we let $[\LeftRegRep\sub{s}(x)]_{i,j}$ 
denote the entries of the matrix representation
with respect to the 
basis
$\smash{\{ \Basis_0\super{s}, \Basis_1\super{s}, \ldots, \Basis_s\super{s} \}} \subset \Wd\sub{s}$ .
Then, we have
\begin{align}
\label{obs1}
0 \leq k \leq s - \ell & \qquad \Longrightarrow \qquad 
\big[ \LeftRegRep\sub{s}(F^k G\super{s}_\ell) \big]_{i,j} = \delta_{i,\ell + k} \delta_{j,\ell}, 
\\
\label{obs2}
0 \leq k \leq \ell & \qquad \Longrightarrow \qquad \big[ \LeftRegRep\sub{s}(E^k G\super{s}_\ell) \big]_{i,j} = 
\frac{[\ell]! [s - \ell + 1]!}{[\ell-k]! [s - \ell + k - 1]!}
\delta_{i,\ell - k} \delta_{j,\ell} , 
\end{align}
for all $\ell \in \{0,1,\ldots,s\}$. 
Equations~(\ref{obs1},~\ref{obs2}) imply that~\eqref{BigImg} is a basis for $\LeftRegRep\sub{s}(\Uqsltwo)$
and that $\LeftRegRep\sub{s}(\Uqsltwo) = \End \VecSp\sub{s}$.
To prove item~\ref{RhoIt4}, we let $\mathsf{W}'$ denote the subspace of $\Uqsltwo$ on the right side of~\eqref{BigKerSpan}, and 
\begin{align}
\mathsf{W} := \bigcup_{0 \, \leq \, \ell \, \leq \, s } \big\{ F^k G\super{s}_\ell \,|\, 0 \leq k \leq s - \ell \big\} \cup \big\{ E^k G\super{s}_\ell \,|\, 0 \leq k \leq \ell \big\}.
\end{align}
First, a straightforward calculation shows that
\begin{align}
\label{AKer}
\mathsf{W}' \subset \ker \LeftRegRep\sub{s}.
\end{align}
It is also straightforward to show that the span of $\mathsf{W}' \cup \mathsf{W}$ 
contains the basis $\{ E^k K^m F^\ell \,|\, k ,\ell \in \bZnn, m \in \bZ \}$ of $\Uqsltwo$,  
and is therefore a spanning set for this algebra. 
As such, we can write any element $x \in \ker \LeftRegRep\sub{s} \subset \Uqsltwo$ in the form
\begin{align}
\label{WriteAs}
x = y + \sum_i c_i z_i
\end{align}
for some 
$y \in \mathsf{W}'$, 
$z_i \in \mathsf{W}$, and 
$c_i \in \bC$.  
After acting on both sides of~\eqref{WriteAs} with the representation $\LeftRegRep\sub{s}$, we obtain
\begin{align}
\label{WriteAs2}
x \in \ker \LeftRegRep\sub{s} \qquad \Longrightarrow \qquad 0 
= \LeftRegRep\sub{s}(x) \overset{\eqref{WriteAs}}{=} 
\LeftRegRep\sub{s}(y) + \sum_i c_i \LeftRegRep\sub{s}(z_i) \overset{\eqref{AKer}}{=} \sum_i c_i \LeftRegRep\sub{s}(z_i).
\end{align}
By item~\ref{RhoIt2}, the set $\{ \LeftRegRep\sub{s}(z) \,|\, z \in \mathsf{W} \}$ is linearly independent. 
Thus, identity~\eqref{WriteAs2} implies that $c_i = 0$ for all $i$, so $x = y \in \mathsf{W}'$. 
We infer that $\ker \LeftRegRep\sub{s} \subset \mathsf{W}'$, and the equality in~\eqref{BigKerSpan} then follows from this fact combined with~\eqref{AKer}.

Finally, the statements~\ref{RhoIt2}--\ref{RhoIt4} 
for the left and right representations $\LeftTwistRep\sub{s}$ and $\RightTwistRep\sub{s}$ of $\UqsltwoBar$
can also be proven similarly. 
\end{proof}

\subsection{Double-commutant property for quantum group and Temperley-Lieb actions} \label{DCProofSec}

Throughout the rest of this section, we assume that $\Summed_\multii < \pmin(q)$. 
Using the general double-commutant property from proposition~\ref{DoubleCommGenProp} in appendix~\ref{DCApp}
and the injectivity of the map $\smash{\Trep_\multii^{\multiii}}$
from proposition~\ref{PreFaithfulPropGen}, we establish one of the key results of this section:
the $\Uqsltwo$-action and the Temperley-Lieb action are each others' commutants.

\GeneralCommutantThm*

\noindent
\textit{Similarly, this theorem holds after the symbolic replacements
$\VecSp \mapsto \VecSpBar$, $\LeftRegRep \mapsto \overbarcal{\LeftRegRep}$, $\Trep \mapsto \TrepBar$, 
and $\CCprojector \mapsto \CCprojectorBar$.
Finally, the 
analogue of this theorem holds for the left and right representations $\LeftTwistRep_\multii$ and $\RightTwistRep_\multii$ of $\UqsltwoBar$.}

\begin{proof} 
Lemma~\ref{UqHomoLem2} gives the ``if" part of item~\ref{GeneralCommutantThmItem1}. 
For the ``only if" part, assuming diagram~\eqref{1stThmCommute} commutes, 
taking $x = 1 \in \Uqsltwo$ implies that necessarily $L = R \in \HomMod{\Uqsltwo} (\VecSp_\multiii, \VecSp_\multii)$. 
Proposition~\ref{MoreGenDecompAndEmbProp2}
and lemma~\ref{StructureOfCommutantLem} together imply that 
\begin{align} \label{DimensionCount}
\dim \HomMod{\Uqsltwo} (\VecSp_\multiii, \VecSp_\multii)
\underset{\eqref{BasisForHomSpace}}{\overset{\eqref{MoreGenDecomp2}}{=}}
\sum_{s \, \in \, \DefectSet_\multii \cap \, \DefectSet_\multiii}  \Dim_\multiii\super{s}  \Dim_\multii\super{s} 
\overset{\eqref{TLdim}}{=} \dim \TL_\multii^\multiii 
\overset{\textnormal{prop.~\ref{PreFaithfulPropGen}}}{=} \dim \Trep_\multii^{\multiii} \big( \TL_\multii^\multiii \big) ,
\end{align}
using the injectivity of $\smash{\Trep_\multii^{\multiii}}$
from proposition~\ref{PreFaithfulPropGen} for the last equality.
We conclude that $L = R = \Trep_\multii^\multiii(T)$ for some valenced tangle $T \in \TL_\multii^\multiii$.
This proves item~\ref{GeneralCommutantThmItem1}. 
Item~\ref{GeneralCommutantThmItem2} then follows from
item~\ref{DoubleCommLemGenHomItem2} of proposition~\ref{DoubleCommGenProp}.
Finally, the corresponding statements 
for the representations $\RightRegRep_\multii$, $\LeftTwistRep_\multii$, and $\RightTwistRep_\multii$ can be proven similarly. 
\end{proof}

\subsection{Higher-spin quantum Schur-Weyl duality} \label{QSWProofSec}

The next key theorem gives the complete information about the (bi)-module 
structure on the vector space $\VecSp_\multii$ 
carrying both the left $\Uqsltwo$-action defined via~(\ref{HopfRep},~\ref{TensProdModules},~\ref{IteratedCoProd},~\ref{UqTensorAction}) 
and the left $\TL_\multii(\nu)$-action defined via lemma~\ref{HomoLem2},
\begin{align}
\Wd_\multii := \BIModule{\VecSp_\multii}{\Uqsltwo}{\TL} .
\end{align}
This result contains the well-known quantum Schur-Weyl duality property \`a la Jimbo~\cite{mj2} as a special case.

\SecondHighQSchurWeylThm*

\noindent
\textit{Similarly, this theorem holds after the symbolic replacements
\begin{align} \label{HighQSchurWeylThm2Replace}
\Trep \mapsto \TrepBar , \qquad
\VecSp \mapsto \VecSpBar , \qquad
\LeftRegRep \mapsto \overbarcal{\LeftRegRep}, \qquad
\Wd \mapsto \WdBar , \qquad
\LS \mapsto \LSBar , 
\qquad \textnormal{and} \qquad 
F^\ell.\Sing_\alpha \mapsto \SingBar_{\alphaBar}.E^\ell.
\end{align}
Finally, the 
analogue of this theorem holds for the left and right representations $\LeftTwistRep_\multii$ and $\RightTwistRep_\multii$ of $\UqsltwoBar$.}

\begin{proof}
Because the representation $\Trep_\multii$ is faithful by proposition~\ref{PreFaithfulPropGen}, we have
$\TL_\multii(\nu) \cong \Trep_\multii ( \TL_\multii(\nu) )$. 
Lemma~\ref{UqHomoLem2} gives $\Trep_\multii ( \TL_\multii(\nu) ) \subset \EndMod{\Uqsltwo} \VecSp_\multii$,
and dimension count~\eqref{DimensionCount} with $\multiii = \multii$ gives the left equalities in~\eqref{TLandUQdoubleCommutants}.
The other statements in items~\ref{HQsw21} and~\ref{HQsw22} then follow 
from double-commutant theorem~\ref{DoubleMainTheorem}
from appendix~\ref{DCApp} 
with $\LeftRegRep_\multii(\Uqsltwo) \subset \End \VecSp_\multii$.
Propositions~\ref{MoreGenDecompAndEmbProp2} 
and~\ref{HWspacePropEmbAndIso} imply item~\ref{HQsw33}.
The 
statements with replacements~\eqref{HighQSchurWeylThm2Replace}, $\LeftTwistRep_\multii$, or $\RightTwistRep_\multii$ 
follow similarly. 
\end{proof}

\subsection{Consecutive tensorand projectors as generators for the commutant}
\label{GeneratorThmCommSubSec}

The next proposition shows that the commutant algebra $\EndMod{\Uqsltwo} \VecSp_\multii$,
isomorphic to 
the valenced Temperley-Lieb algebra $\TL_\multii(\nu)$
by theorem~\ref{HighQSchurWeylThm2},
is generated by the projectors on $\VecSp_\multii$ acting on consecutive tensorands.

\GeneratorThmComm*

\noindent
\textit{Similarly, this proposition holds after the symbolic replacements
$\Uqsltwo \mapsto \UqsltwoBar$, $\VecSp \mapsto \VecSpBar$, $\Trep \mapsto \TrepBar$,
and $\CCprojector \mapsto \CCprojectorBar$.}

\begin{proof}
By theorem~\ref{HighQSchurWeylThm2}, the map 
$\Trep_\multii \colon \TL_\multii(\nu) \longrightarrow \EndMod{\Uqsltwo} \VecSp_\multii$ is an isomorphism of algebras, so
lemma~\ref{GeneratorLemTL} implies that its image is generated by the right side of~\eqref{GeneratorsComm}. 
The corresponding statement for $\UqsltwoBar$ follows similarly.
\end{proof}

\bigskip
\bigskip

\appendixpage

\begin{appendices}
\renewcommand{\thesection}{\Alph{section}}
\renewcommand{\thesubsection}{\arabic{subsection}}
\renewcommand{\thesubsubsection}{\Alph{subsubsection}}

\section{Quantum group --- variants, calculations, and auxiliary results}
\label{PreliApp}
In this appendix, we gather easy auxiliary results and calculations 
and some definitions, for usage throughout the present article as well as for future applications.


\begin{lem} \label{CoproductFormulasLem}
Suppose $q \in \bC^\times \setminus \{\pm1\}$. 
For all integers $ 0 \leq k,\ell \leq \pmin(q)$ 
and $m \in \bZ$, we have
\begin{align} \label{CoproductFormulas}
\Delta(E^k K^m F^\ell) = \sum_{i \, = \,0}^k \sum_{j \, = \,0}^\ell \, q^{i(k - i) - j(\ell - j)} \qbin{k}{i} \qbin{\ell}{j}
E^{k - i} K^{m - \ell + j} F^j \otimes E^i K^{m + k - i} F^{\ell - j},
\end{align}
and similarly,
\begin{align} \label{CoproductFormulasBar}
\DeltaBar(\FBar^k \KBar^m \EBar^\ell) = \sum_{i \, = \, 0}^k \sum_{j \, = \, 0}^\ell \, q^{i(k - i) - j(\ell - j)} \qbin{k}{i} \qbin{\ell}{j}
\FBar^{k - i} \KBar^{m - \ell + j} \EBar^j \otimes \FBar^i \KBar^{m + k - i} \EBar^{\ell - j}.
\end{align}
\end{lem}

\begin{proof}
This follows by relatively straightforward calculations, see, e.g.,
\cite[proposition~\red{VII.1.3}]{ck}.
\end{proof}

\subsection{Properties of submodule projectors and embeddings}

Next, we gather useful properties of the mappings 
$\Embedding_\multii$, $\Projection_\multii$, and $\Projectionhat_\multii$ defined in~\eqref{Composites} 
in section~\ref{EmbAndProjSec} and used repeatedly throughout this article.

\begin{lem} \label{EmbProjLem}
Suppose $\max \multii < \pmin(q)$.  The following hold:
\begin{enumerate}
\itemcolor{red}
\item \label{It1} 
The maps $\Embedding_\multii$, $\Projection_\multii$, and $\Projectionhat_\multii$ are 
$\Uqsltwo$- and $\UqsltwoBar$-homomorphisms.

\item \label{It2} 
$\Embedding_\multii$ is a linear injection, 
$\Projectionhat_\multii$ is a linear surjection, and $\Projection_\multii$ is a linear projection.

\item \label{It3} 
We have $\im \Embedding_\multii = \im \Projection_\multii$, $\ker \Projectionhat_\multii = \ker \Projection_\multii$, 
\begin{align}\label{UQIdComp} 
\Projectionhat_\multii \circ \Embedding_\multii = \id_{\VecSp_\multii} ,
\qquad \textnormal{and} \qquad
\Embedding_\multii \circ \Projectionhat_\multii = \Projection_\multii.
\end{align}
In other words, the following diagram commutes:
\begin{equation} \label{UqcdBig}
\begin{tikzcd}[column sep=2cm, row sep=1.5cm]
& \arrow{ld}[swap]{\Projectionhat_\multii} \arrow{d}{\Projection_\multii}
\VecSp_{\Summed_\multii}
= \VecSp_{\sIndex_1} \otimes \VecSp_{\sIndex_2} \otimes \dotsm \otimes \VecSp_{\sIndex_{\np_\multii}} \\ 
\VecSp_\multii = \VecSp\sub{\sIndex_1} \otimes \VecSp\sub{\sIndex_2} \otimes \dotsm \otimes \VecSp\sub{\sIndex_{\np_\multii}} 
\arrow{r}{\Embedding_\multii}
& \im \Embedding_\multii = \im \Projection_\multii \subset \VecSp_{\Summed_\multii}
\end{tikzcd}
\end{equation}
Also, the map $\Projectionhat_\multii$ restricted to $\im \Projection_\multii$ is 
a $\Uqsltwo$- and $\UqsltwoBar$-module isomorphism, 
with inverse $\Embedding_\multii$. 

\item \label{It3b} 
We have 
\begin{align} \label{PhatP} 
\Projectionhat_\multii \circ \Projection_\multii = \Projectionhat_\multii 
\qquad \textnormal{and} \qquad
\Projection_\multii \circ \Embedding_\multii = \Embedding_\multii .
\end{align}

\end{enumerate}
Similarly, items~\ref{It1}--\ref{It3b} hold
for right $\Uqsltwo$ and $\UqsltwoBar$-modules, after the symbolic replacements
\begin{align}\label{EmbProjLemReplace}
\Embedding \mapsto \EmbeddingBar , \qquad
\Projection \mapsto \ProjectionBar , \qquad
\Projectionhat \mapsto \ProjectionhatBar , \qquad 
\VecSp \mapsto \VecSpBar ,
\qquad \textnormal{and} \qquad
\HWsp \mapsto \HWspBar .
\end{align}
\begin{enumerate}
\setcounter{enumi}{5}
\itemcolor{red}
\item \label{EmbeddingStarItem}
For all vectors $v \in \VecSp_\multii$ and $\overbarStraight{v} \in \VecSpBar_\multii$, we have
\begin{align} \label{StarEmbed}
\Embedding_\multii(v)^* = \EmbeddingBar_\multii(v^*) \qquad \qquad \textnormal{and} \qquad \qquad
\EmbeddingBar_\multii(\overbarStraight{v})^* = \Embedding_\multii(\overbarStraight{v}^*).
\end{align} 
Similarly, for all vectors $w \in \VecSp_{\Summed_\multii}$ and $\overbarStraight{w} \in \VecSpBar_{\Summed_\multii}$, we have
\begin{align}
\label{StarProj}
\Projection_\multii(w)^* = \ProjectionBar_\multii(w^*) \qquad \qquad \textnormal{and} \qquad \qquad
\ProjectionBar_\multii(\overbarStraight{w})^* = \Projection_\multii(\overbarStraight{w}^*) , \\
\label{StarProjHat}
\Projectionhat_\multii(w)^* = \ProjectionhatBar_\multii(w^*) \qquad \qquad \textnormal{and} \qquad \qquad
\ProjectionhatBar_\multii(\overbarStraight{w})^* = \Projectionhat_\multii(\overbarStraight{w}^*) .
\end{align} 
\end{enumerate}
\end{lem}

\begin{proof} 
Items~\ref{It1}--\ref{It3} are immediate from the definitions of the maps in the assertion.
Then, item~\ref{It3} gives item~\ref{It3b}:
\begin{align}
\Projectionhat_\multii \circ \Projection_\multii 
\overset{\eqref{UQIdComp}}{=}
\Projectionhat_\multii \circ \Embedding_\multii \circ \Projectionhat_\multii
\overset{\eqref{UQIdComp}}{=} \Projectionhat_\multii 
\qquad \qquad \textnormal{and} \qquad \qquad
\Projection_\multii \circ \Embedding_\multii 
\overset{\eqref{UQIdComp}}{=}
\Embedding_\multii \circ \Projectionhat_\multii \circ \Embedding_\multii 
\overset{\eqref{UQIdComp}}{=} \Embedding_\multii .
\end{align}
Items~\ref{It1}--\ref{It3b} with replacements~\eqref{EmbProjLemReplace} can be proven similarly.

To finish, we prove item~\ref{EmbeddingStarItem}. 
First, because the maps $\Embedding_\multii, \EmbeddingBar_\multii$, 
and $( \, \cdot \,)^*$ are linear and factor over the tensorands of $\VecSp_\multii$ or $\VecSpBar_\multii$, 
we may take $\multii = (s)$ and $v = \smash{\Basis_\ell\super{s}}$ without loss of generality. 
Then, we have
\begin{align}
\Embedding\sub{s}\big(\Basis_\ell\super{s}\big)^* 
\underset{\eqref{EmbeddingDef}}{\overset{\textnormal{(\ref{BasisStar}, \ref{StarAnti1})}}{=}}
\MTbas_0\superscr{(s) \, *} . \EBar^\ell 
\underset{\eqref{MThwv}}{\overset{\textnormal{(\ref{StarMap}, \ref{StarDistribute})}}{=}}
\MTbasBar_0\super{s} .\EBar^\ell 
\underset{\eqref{EmbeddingDef}}{\overset{\eqref{MTFActhwvBar}}{=}} 
q^{-\ell(s - \ell)} \EmbeddingBar\sub{s}\big(\BasisBar_\ell\super{s}\big)
\overset{\eqref{StarMap}}{=} \EmbeddingBar\sub{s}\big(\Basis_\ell\superscr{(s) \, *}\big) ,
\end{align}
which proves the left equation of~\eqref{StarEmbed}. 
The right equation of~\eqref{StarEmbed} can be proven similarly.
Next, because the maps $\Projection_\multii, \ProjectionBar_\multii, \Projectionhat_\multii, \ProjectionhatBar_\multii$, 
and $( \, \cdot \,)^*$ are linear and factor over the tensorands of $\VecSp_{\Summed_\multii}$ or $\VecSpBar_{\Summed_\multii}$, 
we may take 
$w= \smash{\MTbas_\ell\super{s}}$, so
\begin{align}
\Projection\sub{s}\big(\MTbas_\ell\super{s}\big)^* 
\overset{\eqref{EmbeddingDef}}{=} 
\Projection\sub{s}\big(\Embedding\sub{s}\big(\Basis_\ell\super{s}\big)\big)^* 
\overset{\eqref{PhatP}}{=} 
\Embedding\sub{s}\big(\Basis_\ell\super{s}\big)^* 
\overset{\eqref{StarEmbed}}{=} 
\EmbeddingBar\sub{s}\big(\Basis_\ell\superscr{(s) \, *}\big) 
\overset{\eqref{PhatP}}{=} 
\ProjectionBar\sub{s}\big(\EmbeddingBar\sub{s}\big(\Basis_\ell\superscr{(s) \, *}\big) \big)
\underset{\eqref{StarEmbed}}{\overset{\eqref{EmbeddingDef}}{=}}
\ProjectionBar\sub{s}\big(\MTbas_\ell\superscr{(s) \, *} \big) ,
\end{align}
which proves the left equation of~\eqref{StarProj}. The other equations in~(\ref{StarProj},~\ref{StarProjHat}) can be proven similarly.
\end{proof}

The mappings 
$\smash{\CCembedor\super{s}\sub{r,t}}$, $\smash{\CCprojector\superscr{(r,t);(s)}\sub{r,t}}$, and $\smash{\CChatprojector\sub{s}\super{r,t}}$ 
defined in~(\ref{EmbeddingDef2x2},~\ref{ProjectionDefn2x2},~\ref{ProjectioHatDefn2x2}) in section~\ref{EmbAndProjSec} have
similar properties. 

\begin{lem} \label{EmbProjLem2}
Suppose $r + t < \pmin(q)$.  Then for each $s \in \DefectSet\sub{r,t}$, the following hold:
\begin{enumerate}
\itemcolor{red}
\item \label{2ndIt1} 
The maps $\smash{\CCembedor\super{s}\sub{r,t}}$, $\smash{\CCprojector\superscr{(r,t);(s)}\sub{r,t}}$, and $\smash{\CChatprojector\sub{s}\super{r,t}}$ 
are homomorphisms of left $\Uqsltwo$ and $\UqsltwoBar$-modules.

\item \label{2ndIt2} 
$\smash{\CCembedor\super{s}\sub{r,t}}$ is a linear injection, 
$\smash{\CChatprojector\sub{s}\super{r,t}}$ is a linear surjection, 
and $\smash{\CCprojector\superscr{(r,t);(s)}\sub{r,t}}$ is a linear projection.

\item \label{2ndIt3} 
We have $\im \smash{\CCembedor\super{s}\sub{r,t}} = \im \smash{\CCprojector\superscr{(r,t);(s)}\sub{r,t}}$, 
$\ker \smash{\CChatprojector\sub{s}\super{r,t}} = \ker \smash{\CCprojector\superscr{(r,t);(s)}\sub{r,t}}$, 
\begin{align} \label{2ndIt3IdentityTrivial}
\smash{\CChatprojector\sub{s}\super{r,t}} \circ \smash{\CCembedor\super{s}\sub{r,t}} = \id_{\VecSp\sub{r,t}} ,
\qquad \textnormal{and} \qquad
\smash{\CCembedor\super{s}\sub{r,t}} \circ \smash{\CChatprojector\sub{s}\super{r,t}} = \smash{\CCprojector\superscr{(r,t);(s)}\sub{r,t}}.
\end{align}
Thus, the following diagram commutes:
\begin{equation} \label{2MultiiTriangleDiagram}
\begin{tikzcd}[column sep=2cm, row sep=1.5cm]
& \arrow{ld}[swap]{ \CChatprojector\sub{s}\super{r,t} } \arrow{d}{ \CCprojector\superscr{(r,t);(s)}\sub{r,t} }
\VecSp\sub{r,t}
\\ 
\VecSp\sub{s} 
\arrow{r}{ \CCembedor\super{s}\sub{r,t} }
& \im  \CCembedor\super{s}\sub{r,t}  = \im  \CCprojector\superscr{(r,t);(s)}\sub{r,t}  \subset \VecSp\sub{r,t}
\end{tikzcd}
\end{equation}
and the map $\smash{\CChatprojector\sub{s}\super{r,t}}$ restricted to $\im \smash{\CCprojector\superscr{(r,t);(s)}\sub{r,t}}$ is 
an isomorphism of left $\Uqsltwo$ and $\UqsltwoBar$-modules,
with inverse $\smash{\CCembedor\super{s}\sub{r,t}}$. 

\item \label{2ndIt3b} 
We have 
\begin{align} \label{2ndIt3bEqn}
\CChatprojector\sub{s}\super{r,t} \circ \CCprojector\superscr{(r,t);(s)}\sub{r,t} = \CChatprojector\sub{s}\super{r,t}
\qquad \textnormal{and} \qquad
\CCprojector\superscr{(r,t);(s)}\sub{r,t} \circ \CCembedor\super{s}\sub{r,t} = \CCembedor\super{s}\sub{r,t} .
\end{align}
\end{enumerate}
Similarly, items~\ref{2ndIt1}--\ref{2ndIt3b} hold
for right $\Uqsltwo$ and $\UqsltwoBar$-modules, after the symbolic replacements
\begin{align}
\CCembedor \mapsto \CCembedorBar , \qquad
\CCprojector \mapsto \CCprojectorBar , \qquad
\CChatprojector \mapsto \CChatprojectorBar , \qquad 
\VecSp \mapsto \VecSpBar ,
\qquad \textnormal{and} \qquad
\HWsp \mapsto \HWspBar .
\end{align}
\begin{enumerate}
\setcounter{enumi}{4}
\itemcolor{red}
\item \label{2ndEmbeddingStarItem}
For all vectors $v \in \VecSp\sub{s}$ and $\overbarStraight{v} \in \VecSpBar\sub{s}$, we have
\begin{align} \label{StarEmbed2x2}
\CCembedor\super{s}\sub{r,t}(v)^* = (-q)^{(r + t - s) / 2} \,
\CCembedorBar\super{s}\sub{r,t}(v^*) 
\qquad \qquad \textnormal{and} \qquad \qquad
\CCembedorBar\super{s}\sub{r,t}(\overbarStraight{v})^* = 
(-q)^{-(r + t - s) / 2} \, \CCembedor\super{s}\sub{r,t}(\overbarStraight{v}^*).
\end{align} 
Similarly, for all vectors $w \in \VecSp\sub{r,t}$ and $\overbarStraight{w} \in \VecSpBar\sub{r,t}$, we have
\begin{align}
\label{StarProj2x2}
& \CCprojector\superscr{(r,t);(s)}\sub{r,t}(w)^* 
= \CCprojectorBar\superscr{(r,t);(s)}\sub{r,t}(w^*) 
&&  \qquad \textnormal{and}  \qquad
&& \CCprojectorBar\superscr{(r,t);(s)}\sub{r,t}(\overbarStraight{w})^* 
= \CCprojector\superscr{(r,t);(s)}\sub{r,t}(\overbarStraight{w}^*) , \\
\label{StarProjHat2x2}
& \hphantom{{}\superscr{;(s)}} \CChatprojector\sub{s}\super{r,t}(w)^*
= (-q)^{(r + t - s) / 2} \, \CChatprojectorBar\sub{s}\super{r,t}(w^*) 
&&  \qquad \textnormal{and}  \qquad
&&  \hphantom{{}\superscr{;(s)}} \CChatprojectorBar\sub{s}\super{r,t}(\overbarStraight{w})^* 
= (-q)^{-(r + t - s) / 2} \, \CChatprojector\sub{s}\super{r,t}(\overbarStraight{w}^*) . 
\end{align} 
\end{enumerate}
\end{lem}

\begin{proof}
This can be proven similarly as lemma~\ref{EmbProjLem}, 
except that identities~(\ref{StarEmbed2x2}--\ref{StarProjHat2x2}) are slightly different due to our normalization choice for the embedding maps.
For the latter, we first observe that 
\begin{align} \label{CoBloStar}
\HWvec\sub{r,t}\superscr{(s) \, *}
& \underset{\textnormal{(\ref{StarMap}, \ref{StarDistribute})}}{\overset{\textnormal{(\ref{tau}, \ref{taubar})}}{=}} 
(-q)^{(r + t - s) / 2} \, \HWvecBar\sub{r,t}\super{s} .
\end{align}
Because the maps 
$\smash{\CCembedor\super{s}\sub{r,t}}, \smash{\CCembedorBar\super{s}\sub{r,t}}$, 
and $( \, \cdot \,)^*$ are linear and factor over the tensorands of $\VecSp\sub{s}$ or $\VecSpBar\sub{s}$, it suffices to calculate
\begin{align}
\CCembedor\super{s}\sub{r,t} \big( \Basis_\ell\super{s} \big)^*
\overset{\eqref{EmbeddingDef2x2}}{=} \big( F^\ell.\HWvec\sub{r,t}\super{s} \big)^*
\underset{\eqref{StarAnti1}}{\overset{\eqref{BasisStar}}{=}} 
\HWvec\sub{r,t}\superscr{(s) \, *} . \EBar^\ell
\overset{\eqref{CoBloStar}}{=}
(-q)^{(r + t - s) / 2} \, \HWvecBar\sub{r,t}\super{s} . \EBar^\ell 
\underset{\eqref{EmbeddingDef2x2}}{\overset{\textnormal{(\ref{BarToNoneRight}, \ref{StarMap})}}{=}}
(-q)^{(r + t - s) / 2} \, \CCembedorBar\super{s}\sub{r,t}  \big(  \Basis_\ell\superscr{(s) \, *} \big) ,
\end{align}
which proves the left equation of~\eqref{StarEmbed2x2}. The other equations in~(\ref{StarEmbed2x2}--\ref{StarProjHat2x2}) can be proven similarly.
\end{proof}

\subsection{The star involution}

We defined two analogous bialgebras, $\Uqsltwo$ and $\UqsltwoBar$ in section~\ref{UqSect}.
There is a simple relation between them:
\begin{align}\label{BasisStar} 
E^k K^m F^\ell 
\quad \longmapsto \quad 
(E^k K^m F^\ell)^* := \EBar^\ell \KBar^m \FBar^k ,
\end{align}
extended linearly. 
This map is an involution in the sense that its inverse is given by the linear extension of 
\begin{align}\label{BasisStarInverse}
\EBar^k \KBar^m \FBar^\ell 
\quad \longmapsto \quad 
(\EBar^k \KBar^m \FBar^\ell)^* := E^\ell K^m F^k .
\end{align}


\begin{lem} \label{UqStarMapLemma}
Suppose $q \in \bC^\times \setminus \{\pm1\}$. The map 
$(\, \cdot \,)^* \colon \Uqsltwo \longrightarrow \UqsltwoBar$ 
is an anti-isomorphim of associative, unital algebras, 
as well as an isomorphism of coassociative, counital coalgebras. 
\end{lem}

\begin{proof}
This follows from definition~\eqref{BasisStar}, its inverse~\eqref{BasisStarInverse}, and
relations (\ref{AlgRelations},~\ref{CoProd},~\ref{CoUnit},~\ref{AlgRelationsBar},~\ref{CoProdBar},~\ref{CoUnitBar}).
\end{proof}

\begin{lem}  \label{MergeLem}
Suppose $q \in \bC^\times \setminus \{\pm1\}$. 
If $v \in \smash{\Ksp_\multii\super{s}}$, then we have
\begin{align} \label{vShift}
E.v, \EBar.v \in \Ksp_\multii\super{s + 2}, \qquad 
F.v, \FBar.v \in \Ksp_\multii\super{s - 2}, \qquad 
K^{\pm1}.v, \KBar^{\pm1}.v \in \Ksp_\multii\super{s}, 
\end{align}
and
\begin{align} \label{BarToNoneLeft}
(\EBar^k \KBar^m \FBar^\ell) .v 
= q^{-k(s - 2\ell + k) + m(s - 2\ell) + \ell(s - \ell)} \, (E^k K^m F^\ell) .v .
\end{align}
Similarly, if $\overbarStraight{v} \in \smash{\KspBar_\multii\super{s}}$, then we have
\begin{align} \label{vShiftBar}
\overbarStraight{v}. E, \overbarStraight{v}.\EBar \in \KspBar_\multii\super{s - 2}, \qquad 
\overbarStraight{v}.F, \overbarStraight{v}.\FBar \in \KspBar_\multii\super{s + 2}, \qquad  
\overbarStraight{v}.K^{\pm1}, \overbarStraight{v}.\KBar^{\pm1} \in \KspBar_\multii\super{s} , 
\end{align}
and
\begin{align} \label{BarToNoneRight}
\overbarStraight{v}. (\EBar^k \KBar^m \FBar^\ell)
= q^{-k(s - k) + m(s - 2k) - \ell(s - 2k + \ell)} \, \overbarStraight{v}. (E^k K^m F^\ell).
\end{align}
\end{lem}


\begin{proof} 
Formulas~(\ref{vShift},~\ref{vShiftBar}) follow from definitions~(\ref{HopfRep},~\ref{HopfRepRightBar}) 
together with coproduct formulas~(\ref{CoProd},~\ref{CoProdBar}).
To prove~\eqref{BarToNoneLeft}, it suffices to check the three special cases $(m,n,p) \in \{(1,0,0), (0,1,0), (0,0,1)\}$, 
as the general result~\eqref{BarToNoneLeft} then follows by using formulas~\eqref{vShift}. 
We prove the three special cases by induction on $\np_\multii \in \bZpos$. 
In the initial case $\np_\multii = 1$, we have $\multii = (t)$ for some $t \in \bZpos$, and
\begin{align}
\EBar.\Basis\super{t}_{(t - s)/2} = q^{-1-s} E.\Basis\super{t}_{(t - s)/2} , \qquad
\FBar.\Basis\super{t}_{(t - s)/2} = q^{-1+s} F.\Basis\super{t}_{(t - s)/2} , 
\qquad \text{and} \qquad
\KBar.\Basis\super{t}_{(t - s)/2} = K.\Basis\super{t}_{(t - s)/2} ,
\end{align}
and~\eqref{BarToNoneLeft} follows because any vector $\smash{v \in \VecSp_\multii\super{s}}$ is proportional to
$\smash{\Basis\super{t}_{(t - s)/2}}$ by~\eqref{sGrading}. 
Next, we let $d \geq 2$ and assume that for any multiindex $\smash{\lds \in \bZpos^{d - 1}}$, 
for any index $s \in \DefectSet_{\lds}\superscr{\pm}$, and for any vector $\smash{v \in \Ksp_{\lds}\super{s}}$, we have 
\begin{align} \label{BarToNoneEFK} 
\EBar.v = q^{-1-s} E.v  , \qquad
\FBar.v = q^{-1+s} F.v  , 
\qquad \text{and} \qquad
 \KBar.v = K.v .
\end{align}
Then, for any $\multii = \lds \oplus (t)$ with $t \in \bZpos$ as in~\eqref{hats}, writing a generic vector in the form
\begin{align} \label{KspForm} 
v \in \VecSp_\multii\super{s} \qquad \qquad \overset{\eqref{sGrading}}{\Longrightarrow} \qquad \qquad v = \sum_{\ell \, = \, 0}^{t} v_\ell \otimes \Basis_{t - \ell}\super{t} , 
\qquad \textnormal{where} \qquad 
v_\ell \in \Ksp_{\lds}\super{s + t - 2 \ell} ,
\end{align}
it is straightforward to verify~\eqref{BarToNoneEFK} for $v \in \smash{\VecSp_\multii\super{s}}$ using the induction hypothesis: for example, 
\begin{align}
\EBar.v & \underset{\eqref{KspForm}}{\overset{\eqref{CoProdBar}}{=}} \sum_{\ell \, = \, 0}^t (\KBar^{-1}.v_\ell \otimes \EBar.\Basis_{t - \ell}\super{t} + \EBar.v_\ell \otimes \Basis_{t - \ell}\super{t}) \\
& \underset{\eqref{BarToNoneEFK}}{\overset{\eqref{HopfRepBar}}{=}} \sum_{\ell \, = \, 0}^t (q^{-(s + t - 2 \ell)}v_\ell \otimes q^{-1 - (2\ell - t)}E.\Basis_{t - \ell}\super{t} + q^{-1 - (s + t - 2\ell)} E.v_\ell \otimes \Basis_{t - \ell}\super{t}) \\
& \underset{\hphantom{\eqref{KspForm}}}{\overset{\eqref{HopfRep}}{=}} \sum_{\ell \, = \, 0}^t q^{-1 - s}(E.v_\ell \otimes K.\Basis_{t - \ell}\super{t} + v_\ell \otimes E.\Basis_{t - \ell}\super{t}) 
\underset{\eqref{KspForm}}{\overset{\eqref{CoProd}}{=}} q^{-1 - s}E.v 
\end{align}
and $\FBar.v$ and $\KBar.v$ are similar.
This completes the induction step and, combined with formulas~\eqref{vShift}, finishes the proof of~\eqref{BarToNoneLeft}. 
Identity~\eqref{BarToNoneRight} can be proven similarly, by writing a generic vector 
with $\multii =  (t) \oplus \fds$ instead in the form
\begin{align} \label{KspBarForm} 
\overbarStraight{v} \in \smash{\VecSpBar_\multii\super{s}} \qquad \qquad \overset{\eqref{sGrading}}{\Longrightarrow} \qquad \qquad 
\overbarStraight{v} = \sum_{\ell \, = \, 0}^{t} \BasisBar_{t - \ell}\super{t} \otimes \overbarStraight{v}_\ell , 
\qquad \textnormal{where} \qquad 
\overbarStraight{v}_\ell \in \Ksp_{\fds}\super{s + t - 2 \ell} .
\end{align}
This concludes the proof.
\end{proof}

Next, we define a linear isomorphism
$( \, \cdot \,)^* \colon \VecSp\sub{s} \longrightarrow \VecSpBar\sub{s}$ by linear extension of
\begin{align}
\label{StarMap}
\Basis_\ell\superscr{(s) \, *} := q^{-\ell(s - \ell)} \BasisBar_\ell\super{s} 
\qquad \qquad \text{and} \qquad \qquad 
\BasisBar_\ell\superscr{(s) \, *} := q^{\ell(s - \ell)} \Basis_\ell\super{s} ,
\end{align}
and as shown, we also denote its inverse map 
$( \, \cdot \,)^* \colon \smash{\VecSpBar\sub{s}} \longrightarrow \smash{\VecSp\sub{s}}$ by the same symbol.
We extend this definition to a linear isomorphism 
$\smash{( \, \cdot \,)^* \colon \VecSp_\multii \longrightarrow \VecSpBar_\multii}$
on the tensor product~\eqref{VecSpTensProd}, with inverse  
$\smash{( \, \cdot \,)^* \colon \VecSpBar_\multii \longrightarrow \VecSp_\multii}$, via
\begin{align}
\label{StarDistribute}
(v \otimes w)^* := v^* \otimes w^* \qquad \qquad \text{and} \qquad \qquad (\overbarStraight{v} \otimes \overbarStraight{w})^* := \overbarStraight{v}^* \otimes \overbarStraight{w}^*.
\end{align}
This map defines $\Uqsltwo,\UqsltwoBar$-homomorphisms in the following sense.

\begin{lem}
\label{StarLem}
Suppose $q \in \bC^\times \setminus \{\pm1\}$. 
For all elements $x \in \Uqsltwo$ and $\bar{x} \in \UqsltwoBar$ and vectors $v \in \VecSp_\multii$ and $\overbarStraight{v} \in \VecSp_\multii$, we have
\begin{align}
\label{StarAnti1}
(x.v)^* = v^*.x^*, \qquad \qquad (\bar{x}.v)^* = v^*.\bar{x}^*,
\end{align}
and similarly,
\begin{align}
\label{StarAnti2}
(\overbarStraight{v}.x)^* = x^*.\overbarStraight{v}^*, \qquad \qquad (\overbarStraight{v}.\bar{x})^* = \bar{x}^*.\overbarStraight{v}^*.
\end{align}
\end{lem}


\begin{proof} 
As the map $(\, \cdot \,)^* \colon \Uqsltwo \longrightarrow \UqsltwoBar$ of~\eqref{BasisStar} is an antihomomorphism of algebras, 
it suffices to prove the lemma for $x \in \{E, F, K^{\pm1}\}$. 
Also, we only prove the left equation in~\eqref{StarAnti1}; the other equations in (\ref{StarAnti1},~\ref{StarAnti2}) can be proven similarly.
We proceed by induction on $\np_\multii \in \bZpos$. 
In the initial case $\np_\multii = 1$, we have $\multii = (t)$ for some $t \in \bZpos$, 
and
\begin{align}
\big(F.\Basis_\ell\super{t}\big)^* 
&\overset{\eqref{HopfRep}}{=} \Basis_{\ell + 1}\superscr{(t) \, *}
\overset{\eqref{StarMap}}{=} q^{-(\ell + 1)(t - \ell - 1)} \BasisBar_{\ell + 1}\super{t}
\underset{\eqref{HopfRepRightBar}}{\overset{\eqref{HopfRepRight}}{=}} q^{-\ell(t - \ell)} \BasisBar_\ell\super{t} . \EBar 
\underset{\eqref{StarMap}}{\overset{\eqref{BasisStar}}{=}} \Basis_\ell\superscr{(t) \, *} .F^*, \\
\big(E.\Basis_\ell\super{t}\big)^* 
&\overset{\eqref{HopfRep}}{=} [\ell][t - \ell + 1] \Basis_{\ell - 1}\superscr{(t) \, *} 
\overset{\eqref{StarMap}}{=} q^{-(\ell - 1)(t - \ell + 1)} [\ell][t - \ell + 1] \BasisBar_{\ell - 1}\super{t} 
\underset{\eqref{HopfRepRightBar}}{\overset{\eqref{HopfRepRight}}{=}} q^{-\ell(t - \ell)} \BasisBar_\ell\super{t} . \FBar 
\underset{\eqref{StarMap}}{\overset{\eqref{BasisStar}}{=}} \Basis_\ell\superscr{(t) \, *}.E^*, \\
\big(K^{\pm1}.\Basis_\ell\super{t}\big)^* 
&\overset{\eqref{HopfRep}}{=} q^{\pm(t - 2\ell)} \Basis_\ell\superscr{(t) \, *} 
\overset{\eqref{StarMap}}{=} q^{\pm(t - 2\ell) - \ell(t - \ell)} \BasisBar_\ell\super{t} 
\underset{\eqref{HopfRepRightBar}}{\overset{\eqref{HopfRepRight}}{=}} q^{-\ell(t - \ell)} \BasisBar_\ell\super{t} . \KBar^{\pm1} 
\underset{\eqref{StarMap}}{\overset{\eqref{BasisStar}}{=}} \Basis_\ell\superscr{(t) \, *} .K^{\pm1*}.
\end{align}
The left equation of~\eqref{StarAnti1} then follows by linearity.
Next, we let $d \geq 2$ and assume that the left equation in~\eqref{StarAnti1} holds 
for any multiindex $\smash{\lds \in \bZpos^{d - 1}}$. 
Then, with $\multii = \smash{\lds \oplus (t)}$, by linearity it suffices to consider
\begin{align}
v &= u \otimes w \in \VecSp_\multii, \qquad \qquad 
\text{where} \quad u \in \VecSp_{\lds} \text{ and } w \in \VecSp\sub{t} .
\end{align}
Now, by applying the induction hypothesis to $u$ and $w$, we obtain
\begin{align}
(E.(u \otimes w))^* & \overset{\eqref{CoProd}}{=}
(E.u \otimes K.w + u \otimes E.w)^*
\overset{\eqref{StarDistribute}}{\underset{\eqref{StarAnti1}}{=}} 
u^*.E^* \otimes w^*.K^* + u^* \otimes w^*.E^* \\
& \overset{\eqref{BasisStar}}{=} u^*. \FBar \otimes w^*. \KBar + u^* \otimes w^*. \FBar 
\underset{\eqref{StarDistribute}}{\overset{\eqref{CoProdBar}}{=}} (u \otimes w)^*. \FBar 
\overset{\eqref{BasisStar}}{=} (u \otimes w)^* . E^*.
\end{align}
This proves the left equation of~\eqref{StarAnti1} for $x = E$. The cases $x \in \{F, K^{\pm 1}\}$ can be proven similarly.
\end{proof} 

\subsection{Mapping $q \mapsto q^{-1}$}


We define the maps $(\, \cdot \,)^\mathrm{op} \colon \PlusMinusUqSLTwo \longrightarrow \AntiMinusPlusUqSLTwo$ and $(\, \cdot \,)^\mathrm{op} \colon \AntiPlusMinusUqSLTwo \longrightarrow \MinusPlusUqSLTwo$ by linear extensions of 
\begin{align} \label{BasisOp} 
(\PlusMinusEGen^k \PlusMinusKGen^m \PlusMinusFGen^\ell )^\mathrm{op} := \MinusPlusFGen^\ell \MinusPlusKGen^m \MinusPlusEGen^k 
\qquad\qquad \text{and} \qquad\qquad
(\AntiPlusMinusEGen^k \AntiPlusMinusKGen^m \AntiPlusMinusFGen^\ell )^\mathrm{op} := \AntiMinusPlusFGen^\ell \AntiMinusPlusKGen^m \AntiMinusPlusEGen^k.
\end{align}
These maps are involutions in the sense that, for all elements $x \in \PlusMinusUqSLTwo$ and $\bar{x} \in \AntiPlusMinusUqSLTwo$, we have
\begin{align} \label{OpInv}
(x^\mathrm{op})^\mathrm{op} = x
\qquad \qquad \text{and} \qquad \qquad 
(\bar{x}^\mathrm{op})^\mathrm{op} = \bar{x}.
\end{align}

\begin{lem} \label{SharpAndNonopLem}
Suppose $q \in \bC^\times \setminus \{\pm1\}$. The maps
$(\, \cdot \,)^\mathrm{op} \colon \PlusMinusUqSLTwo \longrightarrow \AntiMinusPlusUqSLTwo$ and $(\, \cdot \,)^\mathrm{op} \colon \AntiPlusMinusUqSLTwo \longrightarrow \MinusPlusUqSLTwo$
are anti-isomorphims of associative, unital algebras, 
as well as isomorphisms of coassociative, counital coalgebras. 
\end{lem}

\begin{proof}
This follows from definition~\eqref{BasisOp}, its inverse~\eqref{OpInv}, and
relations (\ref{AlgRelations},~\ref{CoProd},~\ref{CoUnit},~\ref{AlgRelationsBar},~\ref{CoProdBar},~\ref{CoUnitBar}).
\end{proof}


Assuming that $s < \pmin(q)$, 
we define the maps $(\, \cdot \,)^\mathrm{op} \colon \VecSp\sub{t} \longrightarrow \VecSpBar\sub{t}$ and $(\, \cdot \,)^\mathrm{op} \colon \VecSpBar\sub{t} \longrightarrow \VecSp\sub{t}$ by linear extensions of 
\begin{align} \label{OpMapVec}
\Basis_\ell\super{s}\,^\mathrm{op} := \frac{[\ell]!}{[s]![s - \ell]!} \BasisBar_{s - \ell}\super{s},
\qquad \qquad \text{and} \qquad \qquad
\BasisBar_\ell\super{s}\,^\mathrm{op} := \frac{[\ell]![s]!}{[s - \ell]!} \Basis_\ell\super{s},
\end{align}
and we extend these maps to tensor products of vectors via
\begin{align} \label{OpDistribute}
(v \otimes w)^\mathrm{op} = v^\mathrm{op} \otimes w^\mathrm{op} 
\qquad \qquad \text{and} \qquad \qquad 
(\bar{v} \otimes \bar{w})^\mathrm{op} = \bar{v}^\mathrm{op} \otimes \bar{w}^\mathrm{op}.
\end{align}
Again, we observe that these maps are involutions in the sense that
\begin{align} \label{OpInvolution}
(v^\mathrm{op})^\mathrm{op} = v
\qquad \text{and} \qquad
(\bar{v}^\mathrm{op})^\mathrm{op} = \bar{v}.
\end{align}
Furthermore, these maps send the grade subspaces $\smash{\VecSp_\multii\super{s}}$ and $\smash{\VecSpBar_\multii\super{s}}$ 
to the subspaces with opposite grade value:
\begin{align}
\label{OpGrade}
v \in \VecSp_\multii\super{s} 
\qquad \qquad
\underset{\textnormal{(\ref{OpMapVec}, \ref{OpDistribute})}}{\overset{\eqref{sGrading}}{\Longrightarrow}} 
\qquad\qquad
v^{\mathrm{op}} \in \VecSp_\multii\super{-s}, \\
\label{OpGradeBar}
\overbarStraight{v} \in \VecSpBar_\multii\super{s} 
\qquad \qquad
\underset{\textnormal{(\ref{OpMapVec}, \ref{OpDistribute})}}{\overset{\eqref{sGrading}}{\Longrightarrow}} 
\qquad\qquad
\overbarStraight{v}^{\mathrm{op}} \in \VecSpBar_\multii\super{-s}.
\end{align}

\begin{lem} \label{OpLem}
Suppose $q \in \bC^\times \setminus \{\pm1\}$. For all elements $x \in \PlusMinusUqSLTwo$, $\bar{x} \in \AntiPlusMinusUqSLTwo$ and vectors $v \in \VecSp_\multii$, $\bar{v} \in \VecSpBar_\multii$, we have
\begin{align}
\label{OpAnti1}
(x.v)^\mathrm{op} = v^\mathrm{op}.x^\mathrm{op}, \qquad \qquad (\bar{x}.v)^\mathrm{op} = v^\mathrm{op}.\bar{x}^\mathrm{op},
\end{align}
and similarly,
\begin{align}
\label{OpAnti2}
(\overbarStraight{v}.x)^\mathrm{op} = x^\mathrm{op}.\overbarStraight{v}^\mathrm{op}, \qquad \qquad (\overbarStraight{v}.\bar{x})^\mathrm{op} = \bar{x}^\mathrm{op}.\overbarStraight{v}^\mathrm{op}.
\end{align}
\end{lem}

\begin{proof}
Checking all four cases is a straightforward computation.
\end{proof}

Next, assuming $s < \pmin(q)$, we define the vectors
\begin{align}
\MTbas_{\ell\pm}\super{s} := \PlusMinusFGen^\ell.\MTbas_0\super{s}, \qquad \MTbasBar_{\ell\pm}\super{s} := \MTbasBar_0\super{s}. \PlusMinusEGen^\ell,
\end{align}
the embeddings 
$\Embedding_{\scaleobj{0.85}{(s)} \, \pm} \colon \VecSp\sub{s} \lhook\joinrel\rightarrow \VecSp_s$ 
and $\EmbeddingBar_{\scaleobj{0.85}{(s)}\, \pm} \colon \VecSpBar\sub{s} \lhook\joinrel\rightarrow \VecSpBar_s$ by linear extensions of
\begin{align}
\label{SmallEmbeddingPMDefn}
\Embedding_{\scaleobj{0.85}{(s)}\pm}\big(\Basis_\ell\super{s}) = \MTbas_{\ell\pm}\super{s},
\qquad
\EmbeddingBar_{\scaleobj{0.85}{(s)}\pm}\big(\BasisBar_\ell\super{s}) = \MTbasBar_{\ell\pm}\super{s},
\end{align}
and the embeddings $\Embedding_{\multii \, \pm} \colon \VecSp_\multii \lhook\joinrel\rightarrow \VecSp_{\Summed_\multii}$ 
and $\EmbeddingBar_{\multii \, \pm} \colon \VecSpBar_\multii \lhook\joinrel\rightarrow \VecSpBar_{\Summed_\multii}$ by
\begin{align}
\label{EmbeddingPMDefn}
\Embedding_{\multii \, \pm} := \Embedding_{\scaleobj{0.85}{(s_1)} \, \pm} \otimes \Embedding_{\scaleobj{0.85}{(s_2)} \, \pm} \otimes \dotsm \otimes \Embedding_{\scaleobj{0.85}{(s_{\np_\multii})} \, \pm} ,
\qquad \qquad
\EmbeddingBar_{\multii \, \pm} := 
\EmbeddingBar_{\scaleobj{0.85}{(s_1)} \, \pm} \otimes \EmbeddingBar_{\scaleobj{0.85}{(s_2)} \, \pm} \otimes \dotsm \otimes \EmbeddingBar_{\scaleobj{0.85}{(s_{\np_\multii})} \, \pm}.
\end{align}

\begin{lem}
\label{CircleLem}
Suppose $\max \multii < \pmin(q)$. We have
\begin{align}
\label{EmbedOp}
\Embedding_{\multii \, \pm} (\bar{v}^\mathrm{op}) = [s_1]! \, [s_2]! \, \dotsm \, [s_{\np_\multii}]! \, \EmbeddingBar_{\multii \, \mp}(\bar{v})^\mathrm{op}
\qquad\qquad
\EmbeddingBar_{\multii \, \pm} (v^\mathrm{op}) = [s_1]! \, [s_2]! \, \dotsm \, [s_{\np_\multii}]! \, \Embedding_{\multii \, \mp}(v)^\mathrm{op}.
\end{align}
\end{lem}

\begin{proof}
In light of~\eqref{EmbeddingPMDefn}, it suffices to consider the case $\multii = (s)$ for some $s \in \bZpos$ (so $\np_\multii = 1$). 
We will prove the left equation of~\eqref{EmbedOp}, for the right one is similar.   
To this end, we observe that the map $f \in \End \VecSp\sub{s}$ given by
\begin{align}
\label{FoldDefn}
f(v) := \EmbeddingBar_{\scaleobj{0.85}{(s)}\mp}^{-1}\big(\Embedding_{\scaleobj{0.85}{(s)}\pm}(v^\mathrm{op})^\mathrm{op}\big)
\end{align}
is a $\Uqsltwo$-homomorphism: 
by item~\ref{It1} of lemma~\ref{EmbProjLem} and properties of the $(\,\cdot\,)^\mathrm{op}$ map, we have
\begin{align}
\nonumber
f(x.v) = \EmbeddingBar_{\scaleobj{0.85}{(s)}\mp}^{-1}\big(\Embedding_{\scaleobj{0.85}{(s)}\pm}((x.v)^\mathrm{op})^\mathrm{op}\big)
& \overset{\eqref{OpAnti1}}{=} \EmbeddingBar_{\scaleobj{0.85}{(s)}\mp}^{-1}\big((\Embedding_{\scaleobj{0.85}{(s)}\pm}(v^\mathrm{op}).x^\mathrm{op})^\mathrm{op}\big) \\
& \underset{\eqref{OpAnti1}}{\overset{\eqref{OpInv}}{=}} 
\EmbeddingBar_{\scaleobj{0.85}{(s)}\mp}^{-1}\big(x.\Embedding_{\scaleobj{0.85}{(s)}\pm}(v^\mathrm{op})^\mathrm{op}\big) = x.\EmbeddingBar_{\scaleobj{0.85}{(s)}\mp}^{-1}\big(\Embedding_{\scaleobj{0.85}{(s)}\pm}(v^\mathrm{op})^\mathrm{op}\big) = x.f(v).
\end{align}
As such, by Schur's lemma, there exists a constant $\lambda \in \bC$ such that $f(v) = \lambda v$ for all $v \in \VecSp\sub{s}$, or by~\eqref{FoldDefn}, 
\begin{align}
\label{Unfold}
\Embedding_{\scaleobj{0.85}{(s)}\pm}(v^\mathrm{op}) = \lambda \, \EmbeddingBar_{\scaleobj{0.85}{(s)}\mp}(v)^\mathrm{op}.
\end{align}
After inserting $\bar{v} = \smash{\BasisBar_0\super{s}}$ into~\eqref{Unfold}, we finally arrive with
\begin{align}
\nonumber
\Embedding_{\scaleobj{0.85}{(s)}\pm}(\BasisBar_0\super{s}\,^\mathrm{op}) \overset{\eqref{Unfold}}{=} \lambda \, \EmbeddingBar_{\scaleobj{0.85}{(s)}\mp}\big(\BasisBar_0\super{s}\big)^\mathrm{op} & \underset{\hphantom{\eqref{MThwv}}}{\overset{\eqref{SmallEmbeddingPMDefn}}{=}} \lambda \, \MTbasBar_{0\, \mp}\super{s}\,^\mathrm{op} \\
\nonumber
&\overset{\eqref{MThwvBar}}{=} \lambda \, (\FundBasisBar_0 \otimes \FundBasisBar_0 \otimes \dotsm \otimes \FundBasisBar_0)^\mathrm{op} \\
&\underset{\hphantom{\eqref{MThwv}}}{\overset{\eqref{OpMapVec}}{=}} \lambda \, (\FundBasis_1 \otimes \FundBasis_1 \otimes \dotsm \otimes \FundBasis_1) \\
\nonumber
&\overset{\eqref{MTbashighestK}}{=} \lambda \, [s]!^{-1} \, \MTbas_{s \,\pm }\super{s} \overset{\eqref{SmallEmbeddingPMDefn}}{=} \lambda \, [s]!^{-1} \, \Embedding_{\scaleobj{0.85}{(s)}\pm}\big(\Basis_s\super{s}\big) 
\overset{\eqref{OpMapVec}}{=} \lambda \, [s]!^{-1} \, \Embedding_{\scaleobj{0.85}{(s)}\pm}\big(\Basis_0\super{s}\,^\mathrm{op}\big).
\end{align}
Hence, we have $\lambda = [s]!$ in~\eqref{Unfold}, finishing the proof.
\end{proof}

After inserting $\multii = (s)$ and $\bar{v} = \smash{\BasisBar_\ell\super{s}}$ into~\eqref{EmbedOp} and using~\eqref{OpMapVec} to simplify the result, 
we  obtain
\begin{align}
\MTbasBar_{\ell \, \pm}\super{s}\,^\mathrm{op} = \frac{[\ell]!}{[s - \ell]!} \MTbas_{s - \ell \, \mp}\super{s} .
\end{align}

\subsection{Bilinear forms from the bilinear pairing}

Using the maps $(\, \cdot \,)^\mathrm{op}$, we define bilinear forms 
$\SPBiFormNewInv{\,\cdot\,}{\,\cdot\,} \colon \VecSp_\multii \times \VecSp_\multii \longrightarrow \bC$ and 
$\SPBiFormNewInv{\,\cdot\,}{\,\cdot\,} \colon \VecSpBar_\multii \times \VecSpBar_\multii \longrightarrow \bC$ respectively by
\begin{align}
\label{AltBiFormDefn}
\SPBiFormNewInv{v}{w} := \LSBiFormBar{v^{\mathrm{op}}}{w}
\qquad \text{and} \qquad
\SPBiFormNewInv{\overbarStraight{v}}{\overbarStraight{w}} := \LSBiFormBar{\overbarStraight{v}}{\overbarStraight{w}^{\mathrm{op}}}.
\end{align}

\begin{lem} \label{2ndBiFormPropertyLem}
Suppose $\max \multii < \pmin(q)$.  The following hold:
\begin{enumerate}
\itemcolor{red}
\item \label{2biformitem0}
We have
\begin{align} \label{BiformDefnBasisvec}
\SPBiFormNewInv{\Basis_{\ell_1}\super{\sIndex_1} \otimes \Basis_{\ell_2}\super{\sIndex_2} \otimes \dotsm \otimes \Basis_{\ell_{\np_\multii}}\super{\sIndex_{\np_\multii}}}{\Basis_{m_1}\super{\sIndex_1} \otimes \Basis_{m_2}\super{\sIndex_2} \otimes \dotsm \otimes \Basis_{m_{\np_\multii}}\super{\sIndex_{\np_\multii}}}
&= \prod_{k \, = \, 1}^{\np_\multii} \delta_{\ell_k + m_k, \sIndex_k}
\end{align}
and similarly, 
\begin{align}
\label{AntiBiformDefnBasisvec}
\SPBiFormNewInv{\BasisBar_{\ell_1}\super{\sIndex_1} \otimes \BasisBar_{\ell_2}\super{\sIndex_2} \otimes \dotsm \otimes \BasisBar_{\ell_{\np_\multii}}\super{\sIndex_{\np_\multii}}}{\BasisBar_{m_1}\super{\sIndex_1} \otimes \BasisBar_{m_2}\super{\sIndex_2} \otimes \dotsm \otimes \BasisBar_{m_{\np_\multii}}\super{\sIndex_{\np_\multii}}}
&= \prod_{k \, = \, 1}^{\np_\multii} \delta_{\ell_k + m_k, \sIndex_k}.
\end{align}
In particular, the bilinear form is symmetric: 
for all vectors $\overbarStraight{v}, \overbarStraight{w} \in \VecSpBar_\multii$ and $v, w \in \VecSp_\multii$, we have
\begin{align}
\label{SymBiFormProp}
\SPBiFormNewInv{v}{w} = \SPBiFormNewInv{w}{v} 
\qquad \textnormal{and} \qquad
\SPBiFormNewInv{\overbarStraight{v}}{\overbarStraight{w}} = \SPBiFormNewInv{\overbarStraight{w}}{\overbarStraight{v}}.
\end{align}
\item \label{2biformitem1}
For all vectors 
$\overbarStraight{v}_j, \overbarStraight{w}_j \in \VecSpBar\sub{\sIndex_j}$ and $v_j, w_j \in \VecSp\sub{\sIndex_j}$,
with $j \in \{1,2,\ldots,\np_\multii\}$, we have the factorizations 
\begin{align} 
\label{biformfactorize2}
\SPBiFormNewInv{v_1 \otimes v_2 \otimes \dotsm \otimes v_{\np_\multii}}{
w_1 \otimes w_2 \otimes \dotsm \otimes w_{\np_\multii}}
= \; & \prod_{k \, = \, 1}^{\np_\multii} \SPBiFormNewInv{v_k}{w_k} , \\[.3em]
\label{biformfactorize3}
\SPBiFormNewInv{\overbarStraight{v}_1 \otimes \overbarStraight{v}_2 \otimes \dotsm \otimes \overbarStraight{v}_{\np_\multii}}{
\overbarStraight{w}_1 \otimes \overbarStraight{w}_2 \otimes \dotsm \otimes \overbarStraight{w}_{\np_\multii}}
= \; & \prod_{k \, = \, 1}^{\np_\multii} \SPBiFormNewInv{\overbarStraight{v}_k}{\overbarStraight{w}_k}.
\end{align}

\item \label{2biformitem2}
For all elements $x_\pm \in \PlusMinusUqSLTwo$ and $\bar{x}_\pm \in \AntiPlusMinusUqSLTwo$, or $x_\pm \in \PlusMinusUqSLTwo^{\otimes \np_\multii}$ and $\bar{x}_\pm \in \AntiPlusMinusUqSLTwo^{\otimes \np_\multii}$, and vectors
$v, w \in \VecSpBar_\multii$, we have
\begin{align} \label{biformEquivitem1form22}
\SPBiFormNewInv{v}{x_\pm.w} = \SPBiFormNewInv{x_\pm^{\mathrm{op}}.v}{w}
\qquad \qquad \textnormal{and} \qquad \qquad
\SPBiFormNewInv{v}{\bar{x}_\pm.w} = \SPBiFormNewInv{\bar{x}_\pm^{\mathrm{op}}.v}{w}, \\
\label{biformEquivitem1form23}
\SPBiFormNewInv{x_\pm.v}{w} = \SPBiFormNewInv{v}{x_\pm^{\mathrm{op}}.w}
\qquad \qquad \textnormal{and} \qquad \qquad
\SPBiFormNewInv{\bar{x}_\pm.v}{w} = \SPBiFormNewInv{v}{\bar{x}_\pm^{\mathrm{op}}.w} , \\
\label{biformEquivitem1form24}
\SPBiFormNewInv{\overbarStraight{v}}{\overbarStraight{w}.x_\pm} = \SPBiFormNewInv{\overbarStraight{v}.x_\pm^{\mathrm{op}}}{\overbarStraight{w}}
\qquad \qquad \textnormal{and} \qquad \qquad
\SPBiFormNewInv{\overbarStraight{v}}{\overbarStraight{w}.\bar{x}_\pm} = \SPBiFormNewInv{\overbarStraight{v}.\bar{x}_\pm^{\mathrm{op}}}{\overbarStraight{w}}, \\
\label{biformEquivitem1form25}
\SPBiFormNewInv{\overbarStraight{v}.x_\pm}{\overbarStraight{w}} = \SPBiFormNewInv{\overbarStraight{v}}{\overbarStraight{w}.x_\pm^{\mathrm{op}}}
\qquad \qquad \textnormal{and} \qquad \qquad
\SPBiFormNewInv{\overbarStraight{v}.\bar{x}_\pm}{\overbarStraight{w}} = \SPBiFormNewInv{\overbarStraight{v}}{\overbarStraight{w}.\bar{x}_\pm^{\mathrm{op}}}.
\end{align}

\item \label{2biformitem32}
The pairs $\smash{\VecSp_\multii\super{s}}$ and $\smash{\VecSp_\multii\super{t}}$,
or $\smash{\VecSpBar_\multii\super{s}}$ and $\smash{\VecSpBar_\multii\super{t}}$, 
are respectively orthogonal\textnormal{:}
\begin{align} 
\label{QInvOrtho}
\SPBiFormNewInv{v}{w} = 0 \qquad
\textnormal{for all} \quad 
v \in \smash{\VecSp_\multii\super{s}} \textnormal{ and } 
w \in \smash{\VecSp_\multii\super{t}}
\textnormal{ with } s \neq -t , \\
\label{QInvOrthoBar}
\SPBiFormNewInv{\overbarStraight{v}}{\overbarStraight{w}} = 0 \qquad
\textnormal{for all} \quad 
\overbarStraight{v} \in \smash{\VecSpBar_\multii\super{s}} \textnormal{ and } 
\overbarStraight{w} \in \smash{\VecSpBar_\multii\super{t}} 
\textnormal{ with } s \neq -t .
\end{align}
Also, for all vectors 
$v \in \smash{\VecSp_\multii\super{s}}$ with $\PlusMinusFGen.v = \AntiPlusMinusFGen.v = 0$ and $w \in \smash{\VecSp_\multii\super{t}}$ with $\MinusPlusEGen.w = \AntiMinusPlusEGen.w= 0$, we have
\begin{align}
\label{FactBiformId22}
\SPBiFormNewInv{\PlusMinusEGen^\ell.v}{\MinusPlusFGen^m.w} = \delta_{s + t, 0} \, \delta_{\ell, m} [\ell]!^2 \qbin{s}{\ell} \SPBiFormNewInv{v}{w} = \SPBiFormNewInv{\AntiPlusMinusEGen^\ell.v}{\AntiMinusPlusFGen^m.w},
\end{align}
and for all vectors $\overbarStraight{v} \in \smash{\VecSpBar_\multii\super{s}}$ with $\overbarStraight{v}.\PlusMinusFGen = \overbarStraight{v}.\AntiPlusMinusFGen = 0$ and $\overbarStraight{w} \in \smash{\VecSpBar_\multii\super{t}}$ with $\overbarStraight{w}.\MinusPlusEGen = \overbarStraight{w}.\AntiMinusPlusEGen = 0$, we have
\begin{align}
\label{FactBiformId222}
\SPBiFormNewInv{\overbarStraight{v}.\PlusMinusEGen^\ell}{\overbarStraight{w}.\MinusPlusFGen^m} = \delta_{s + t, 0} \, \delta_{\ell, m} [\ell]!^2 \qbin{s}{\ell} \SPBiFormNewInv{\overbarStraight{v}}{\overbarStraight{w}} = \SPBiFormNewInv{\overbarStraight{v}.\AntiPlusMinusEGen^\ell}{\overbarStraight{w}.\AntiMinusPlusFGen^m}.
\end{align}
 
\item \label{2biformitem4}
For all vectors 
$\overbarStraight{v}, \overbarStraight{w} \in \VecSpBar_\multii$ and $v, w \in \VecSpBar_\multii$, we have
\begin{align} \label{SPBiFormNewEmbed2}
\SPBiFormNewInv{v}{w} 
&= [\sIndex_1]! \, [\sIndex_2]! \, \dotsm [\sIndex_{\np_\multii}]!
\SPBiFormNewInv{\Embedding_{\multii \, \mp}(v)}{\Embedding_{\multii \, \pm}(w)}
\end{align}
and
\begin{align}
\label{SPBiFormNewEmbed2Bar}
\SPBiFormNewInv{\overbarStraight{v}}{\overbarStraight{w}} 
&= [\sIndex_1]! \, [\sIndex_2]! \, \dotsm [\sIndex_{\np_\multii}]!
\SPBiFormNewInv{\EmbeddingBar_{\multii \, \mp}(\overbarStraight{v})}{\EmbeddingBar_{\multii \, \pm}(\overbarStraight{w})}.
\end{align}

\item \label{2biformitem5}
For all vectors $\overbarStraight{v}, \overbarStraight{w} \in \VecSpBar_\multii$ and $v, w \in \VecSp_\multii$, we have
\begin{align}
\label{OpNoOp}
\SPBiFormNewInv{v}{w} = \SPBiFormNewInv{v^{\mathrm{op}}}{w^{\mathrm{op}}} \qquad \textnormal{and} \qquad \SPBiFormNewInv{\overbarStraight{v}}{\overbarStraight{w}} = \SPBiFormNewInv{\overbarStraight{v}^{\mathrm{op}}}{\overbarStraight{w}^{\mathrm{op}}}.
\end{align}

\end{enumerate}
\end{lem}

\begin{proof}
We prove these properties as follows.
\begin{enumerate}[leftmargin=*]
\itemcolor{red}
\item 
Properties (\ref{BiformDefnBasisvec},~\ref{AntiBiformDefnBasisvec}) follow from 
definition~\eqref{AltBiFormDefn} of the bilinear form combined with~\eqref{biformfactorize} of 
item~\ref{biformitem1} in lemma~\ref{biformPropertyLem}, and distributive property~\eqref{OpMapVec} of the map $\smash{(\,\cdot\,)^\mathrm{op}}$.
Then,~\eqref{SymBiFormProp} follows from this.

\item
Properties (\ref{biformfactorize2},~\ref{biformfactorize3}) follow from (\ref{BiformDefnBasisvec},~\ref{AntiBiformDefnBasisvec}) of item~\ref{2biformitem0} and linearity.

\item We have
\begin{align}
\SPBiFormNewInv{v}{x_\pm.w} \overset{\eqref{AltBiFormDefn}}{=} \LSBiFormBar{v^{\mathrm{op}}}{x_\pm.w} \overset{\eqref{biformEquivitem1form2}}{=} \LSBiFormBar{v^{\mathrm{op}}.x_\pm}{w} \underset{\eqref{OpAnti2}}{\overset{\eqref{OpInvolution}}{=}} \LSBiFormBar{(x_\pm^\mathrm{op}.v)^{\mathrm{op}}}{w} \overset{\eqref{AltBiFormDefn}}{=} \SPBiFormNewInv{x_\pm^{\mathrm{op}}.v}{w} ,
\end{align}
which proves the left equality of~\eqref{biformEquivitem1form22}. 
The right equality of~\eqref{biformEquivitem1form22} and 
(\ref{biformEquivitem1form23}--\ref{biformEquivitem1form25}) can be proven similarly.

\item (\ref{QInvOrtho},~\ref{QInvOrthoBar}) follow from bilinear form definition~\eqref{AltBiFormDefn}, (\ref{OpGrade},~\ref{OpGradeBar}), and~\eqref{OrthoSubspaces} in item~\ref{biformitem3} of lemma~\ref{biformPropertyLem}. 
Properties (\ref{FactBiformId22},~\ref{FactBiformId222}) then follow from this,~\eqref{BasisOp}, (\ref{OpAnti1},~\ref{OpAnti2}) of lemma~\ref{OpLem}, and~\eqref{FactBiformId} in item~\ref{biformitem3} of lemma~\ref{biformPropertyLem}.

\item
Properties (\ref{SPBiFormNewEmbed2},~\ref{SPBiFormNewEmbed2Bar}) follow from~\eqref{EmbedOp} of lemma~\ref{CircleLem}, 
bilinear form definition~\eqref{AltBiFormDefn}, and~\eqref{SPBiFormNewEmbed} from item~\ref{biformitem4} in lemma~\ref{biformPropertyLem}.

\item Property~\eqref{OpNoOp} follows from bilinear form definition~\eqref{AltBiFormDefn} and involution property~\eqref{OpInvolution}.
\end{enumerate}
This completes the proof.
\end{proof}

\section{Exceptional case: $q = \pm \ii$}  
\label{ExceptionalQSect}
In lemma~\ref{TLprojLemNew} and corollary~\ref{RepCor},  
if $q \in \{\pm \ii\}$, then $\nu = 0$, so~(\ref{GenProj2-0},~\ref{GenProj2-0Bar},~\ref{GenProj},~\ref{GenProjBar}) are identically zero.
On the other hand,~(\ref{GenProj2-1},~\ref{GenProj2-1Bar}) are still well-defined, and we can also 
make sense of analogues of~(\ref{GenProj2-0},~\ref{GenProj2-0Bar},~\ref{GenProj},~\ref{GenProjBar}), 
as we show in this appendix.
To begin, we note that if $q = \pm \ii$, then the vectors in~(\ref{M02Basis},~\ref{M02BasisBar}) 
fail to form bases for the left and right $\Uqsltwo$-modules $\Module{\VecSp_2}{\Uqsltwo}$ and $\RModule{\VecSpBar_2}{\Uqsltwo}$.
For example, the following two vectors are proportional to each other:
\begin{align}
F^2 . \HWvec\sub{1,1}\super{2} = 0 \qquad \qquad \text{and} \qquad \qquad
\HWvec\sub{1,1}\super{0} = \mp \frac{\ii}{2} \; \smash{F . \HWvec\sub{1,1}\super{2}} .
\end{align}
However, we can choose for instance the following different bases for $\Module{\VecSp_2}{\Uqsltwo}$ and $\RModule{\VecSpBar_2}{\Uqsltwo}$:
\begin{align} 
\label{M02BasisNonSS} 
& \Big\{ \HWvec := \HWvec\sub{1,1}\super{0}, \quad 
\MTbas 
:= \HWvec\sub{1,1}\super{2}  = \FundBasis_0 \otimes \FundBasis_0, \quad 
\zeta := - \frac{1}{4}(1 \pm \ii) \left( \FundBasis_0 \otimes \FundBasis_1 + \FundBasis_1 \otimes \FundBasis_0 \right) , \quad 
\mu := \FundBasis_1 \otimes \FundBasis_1 \Big\} \quad \subset \quad \Module{\VecSp_2}{\Uqsltwo} , \\
\label{M02BasisNonSSBar} 
& \Big\{ \HWvecBar := \HWvecBar\sub{1,1}\super{0}, \quad 
\MTbasBar := \HWvecBar\sub{1,1}\super{2}  = \FundBasisBar_0 \otimes \FundBasisBar_0, \quad 
\overbarcal{\zeta} := \;\;\, \frac{1}{4}(1 \mp \ii) \left( \FundBasisBar_0 \otimes \FundBasisBar_1 + \FundBasisBar_1 \otimes \FundBasisBar_0 \right) , \quad 
\overbarcal{\mu} := \FundBasisBar_1 \otimes \FundBasisBar_1 \Big\} \quad \subset \quad \RModule{\VecSpBar_2}{\Uqsltwo} .
\end{align}
Then, the left module $\Module{\VecSp_2}{\Uqsltwo}$ has the following structure:
\begin{itemize}[leftmargin=*]
\item $\HWvec$ generates a trivial $\Uqsltwo$-submodule 
\begin{align}
\mathsf{S} := \Module{\Span \{ \HWvec \}}{\Uqsltwo} \cong \Wd\sub{0} , 
\qquad \text{with} \qquad
K . \HWvec =  \HWvec , \qquad 
E . \HWvec = 0 , \qquad 
F . \HWvec = 0 ,
\end{align}

\item $\MTbas$ and $\mu$ both generate two-dimensional $\Uqsltwo$-submodules 
$\Module{\Span \{\MTbas, \HWvec \}}{\Uqsltwo}$ and $\Module{\Span \{\mu, \HWvec \}}{\Uqsltwo}$,
respectively, which are non-isomorphic, indecomposable, but not simple, and whose intersection is 
$\mathsf{S}$:
\begin{align}
K . \MTbas = - \MTbas , \qquad
E . \MTbas = 0 , \qquad 
F . \MTbas = \pm 2 \ii \HWvec ,
\qquad \qquad \qquad 
K . \mu = - \mu , \qquad
E . \mu = \pm 2 \ii \HWvec ,
\qquad 
F . \mu = 0 ,
\end{align}

\item $\zeta$ generates the whole module $\Module{\VecSp_2}{\Uqsltwo}$, which is indecomposable but not simple:
\begin{align}
K . \zeta = \zeta , \qquad 
E . \zeta = \mp \frac{\ii}{2} \, \MTbas , \qquad 
F . \zeta = \mp \frac{\ii}{2} \mu , \qquad 
F E . \zeta =  E F . \zeta = \HWvec .
\end{align}
\end{itemize}

Similarly, the right module $\RModule{\VecSpBar_2}{\Uqsltwo}$ has the following structure:
\begin{itemize}[leftmargin=*]
\item $\HWvecBar$ generates a trivial $\Uqsltwo$-submodule  
\begin{align}
\overbarStraight{\mathsf{S}} :=  \RModule{\Span \{ \HWvecBar \}}{\Uqsltwo}  \cong \WdBar\sub{0} , 
\qquad \text{with} \qquad
\HWvecBar . K =  \HWvecBar , \qquad 
\HWvecBar . E = 0 , \qquad 
\HWvecBar . F = 0 ,
\end{align}

\item $\MTbasBar$ and $\overbarcal{\mu}$ both generate two-dimensional $\Uqsltwo$-submodules 
$\RModule{\Span \{\MTbasBar, \HWvecBar \}}{\Uqsltwo}$ and $\RModule{\Span \{\overbarcal{\mu}, \HWvecBar \}}{\Uqsltwo}$,
respectively, which are non-isomorphic, indecomposable, but not simple, and whose intersection is 
$\smash{\overbarStraight{\mathsf{S}}}$:
\begin{align}
\MTbasBar . K = - \MTbasBar , \qquad
\MTbasBar . E = \pm 2 \ii \HWvecBar , \qquad
\MTbasBar . F = 0 ,  
\qquad \qquad \qquad 
\overbarcal{\mu} . K = - \overbarcal{\mu} , 
\qquad \overbarcal{\mu} . E = 0 , \qquad
\overbarcal{\mu} . F = \pm 2 \ii \HWvecBar ,
\end{align}

\item $\overbarcal{\zeta}$ generates the whole module $\RModule{\VecSpBar_2}{\Uqsltwo}$, which is indecomposable but not simple:
\begin{align}
\overbarcal{\zeta} . K = \overbarcal{\zeta} , \qquad 
\overbarcal{\zeta} . E = \mp \frac{\ii}{2} \, \overbarcal{\mu} , \qquad 
\overbarcal{\zeta} . F = \mp \frac{\ii}{2} \MTbasBar , \qquad 
\overbarcal{\zeta} . F E =  \overbarcal{\zeta} . E F = \HWvecBar .
\end{align}
\end{itemize}

\begin{remark}
In particular, we see that the quantum Schur-Weyl duality decomposition~\eqref{GenDecompWJ} 
in theorem~\ref{HighQSchurWeylThm2} cannot hold for this case.
However, if $q \in \{\pm \ii\}$, then the structure~\eqref{M02BasisNonSS} and a direct calculation shows that any 
$\Uqsltwo$-homomorphism  $T \in \EndMod{\Uqsltwo} \VecSp_2$ 
must have the form
\begin{align}
T(\HWvec) = c \HWvec, \qquad T(\MTbas) = c \MTbas, \qquad T(\mu) = c \mu , \qquad 
T(\zeta) = a \HWvec + c \zeta \qquad \text{for some } a,c \in \bC .
\end{align}
Hence, the space $\EndMod{\Uqsltwo} \VecSp_2$ is two-dimensional, spanned by the identity map and the map 
$\smash{\CCunprojector\sub{1,1}\superscr{(1,1); (0)}} \colon \VecSp_2 \longrightarrow \VecSp_2$, 
\begin{align}
\label{UnProjector}
& \CCunprojector\sub{1,1}\superscr{(1,1); (0)} ( \HWvec ) := 0 , \qquad 
\CCunprojector\sub{1,1}\superscr{(1,1); (0)} ( \MTbas ) := 0 , \qquad
\CCunprojector\sub{1,1}\superscr{(1,1); (0)} ( \mu ) := 0 , \qquad
\CCunprojector\sub{1,1}\superscr{(1,1); (0)} ( \zeta ) := \HWvec .
\end{align}
In particular, the Temperley-Lieb algebra is isomorphic to the commutant algebra  $\EndMod{\Uqsltwo} \VecSp_2$:
we obtain a representation of the Temperley-Lieb algebra on $\CModule{\VecSp_2}{\TL}$,
by setting the generator $\Gen_1$ to act as the map~\eqref{UnProjector}, see corollary~\ref{RepCorExceptional}.
\end{remark}

Similarly to~\eqref{UnProjector}, we define
$\smash{\CCunprojectorBar\sub{1,1}\superscr{(1,1); (0)}} \colon \VecSpBar_2 \longrightarrow \VecSpBar_2$,
\begin{align}
\label{UnProjectorBar}
& \CCunprojectorBar\sub{1,1}\superscr{(1,1); (0)} ( \HWvecBar ) := 0 , \qquad 
\CCunprojectorBar\sub{1,1}\superscr{(1,1); (0)} ( \MTbasBar ) := 0 , \qquad
\CCunprojectorBar\sub{1,1}\superscr{(1,1); (0)} ( \overbarcal{\mu} ) := 0 , \qquad
\CCunprojectorBar\sub{1,1}\superscr{(1,1); (0)} ( \overbarcal{\zeta} ) := \HWvecBar .
\end{align}
We also extend definition~\eqref{EmbeddingDef2x2} of the two embeddings 
$\smash{\CCembedor\super{0}\sub{1,1} \colon \VecSp_0 \longrightarrow \VecSp_2}$ and
$\smash{\CCembedorBar\super{0}\sub{1,1} \colon \VecSpBar_0 \longrightarrow \VecSpBar_2}$
from the range of parameter values $\{ q \in \bC^\times \, | \, \pmin(q) > 2 \}  = \bC^\times \setminus \{\pm 1, \pm \ii\}$
to the range $\{ q \in \bC^\times \, | \, \pmin(q) \geq 2 \} = \bC^\times \setminus \{\pm 1\}$
by linear extensions of
\begin{align} 
\CCembedor\super{0}\sub{1,1} \big( \Basis_0\super{0} \big) := \HWvec\sub{1,1}\super{0} = \HWvec 
\qquad\qquad \text{and} \qquad\qquad
\CCembedorBar\super{0}\sub{1,1} \big( \BasisBar_0\super{0} \big) := \HWvecBar\sub{1,1}\super{0} = \HWvecBar .
\end{align}
and we define 
the maps $\smash{\CChatunprojector\super{1,1}\sub{0} \colon \VecSp_2 \longrightarrow \VecSp_0}$
and $\smash{\CChatunprojectorBar\super{1,1}\sub{0} \colon \VecSpBar_2 \longrightarrow \VecSpBar_0}$,
analogous to~\eqref{ProjectioHatDefn2x2}, by linear extensions of
\begin{align}
\label{UnProjectorHatBar}
& \CChatunprojectorBar\super{1,1}\sub{0} ( \HWvecBar ) := 0 , \qquad 
\CChatunprojectorBar\super{1,1}\sub{0} ( \MTbasBar ) := 0 , \qquad
\CChatunprojectorBar\super{1,1}\sub{0} ( \overbarcal{\mu} ) := 0 , \qquad
\CChatunprojectorBar\super{1,1}\sub{0} ( \overbarcal{\zeta} ) :=  \BasisBar_0\super{0} , \\
\label{UnProjectorHat}
& \CChatunprojector\super{1,1}\sub{0} ( \HWvec ) := 0 , \qquad 
\CChatunprojector\super{1,1}\sub{0} ( \MTbas ) := 0 , \qquad
\CChatunprojector\super{1,1}\sub{0} ( \mu ) := 0 , \qquad
\CChatunprojector\super{1,1}\sub{0} ( \zeta ) :=  \Basis_0\super{0} .
\end{align}

We record the following obvious properties of these maps, vaguely analogous to lemma~\ref{EmbProjLem2}:
\begin{lem} \label{EmbProjLem2Exc}
Suppose $q \in \{\pm \ii\}$ $($i.e., $\pmin(q) = 2)$. Then, the following hold:
\begin{enumerate}
\itemcolor{red}
\item \label{2ndIt1Exc} 
The maps $\smash{\CCembedor\super{0}\sub{1,1}}$, $\smash{\CCunprojector\sub{1,1}\superscr{(1,1); (0)}}$, and $\smash{\CChatunprojector\super{1,1}\sub{0}}$ 
are homomorphisms of left $\Uqsltwo$ and $\UqsltwoBar$-modules.

\item \label{2ndIt2Exc} 
$\smash{\CCembedor\super{0}\sub{1,1}}$ is a linear injection, 
$\smash{\CChatunprojector\super{1,1}\sub{0}}$ is a linear surjection, 
but $\smash{\CCunprojector\sub{1,1}\superscr{(1,1); (0)}}$ is not a linear projection, as
\begin{align}
\CCunprojector\sub{1,1}\superscr{(1,1); (0)} \circ\CCunprojector\sub{1,1}\superscr{(1,1); (0)} = 0 .
\end{align}

\item \label{2ndIt3Exc} 
We have $\im \smash{\CCembedor\super{0}\sub{1,1}} = \im \smash{\CCunprojector\sub{1,1}\superscr{(1,1); (0)}} = \mathsf{S}$, 
$\smash{\im \CChatunprojector\super{1,1}\sub{0} = \VecSp_0}$,
$\ker \smash{\CChatunprojector\super{1,1}\sub{0}} = \ker \smash{\CCunprojector\sub{1,1}\superscr{(1,1); (0)}}$,  and
\begin{align} 
\label{EmbProjExcept}
\CCembedor\super{0}\sub{1,1} \circ \CChatunprojector\super{1,1}\sub{0} = \CCunprojector\sub{1,1}\superscr{(1,1); (0)}
\qquad \qquad \textnormal{but} \qquad \qquad \CChatunprojector\super{1,1}\sub{0} \circ \CCembedor\super{0}\sub{1,1} = 0 .
\end{align}
Thus, the following diagram commutes:
\begin{equation} 
\label{2MultiiTriangleDiagramExc}
\begin{tikzcd}[column sep=2cm, row sep=1.5cm]
& \arrow{ld}[swap]{ \CChatunprojector\super{1,1}\sub{0} } \arrow{d}{ \CCunprojector\sub{1,1}\superscr{(1,1); (0)} }
\VecSp_2
\\ 
\VecSp_0 
\arrow{r}{ \CCembedor\super{0}\sub{1,1} }
& \im  \CCembedor\super{0}\sub{1,1}  = \im  \CCunprojector\sub{1,1}\superscr{(1,1); (0)} = \mathsf{S} \subset \VecSp_2
\end{tikzcd}
\end{equation}
\end{enumerate}
Similarly, 
items~\ref{2ndIt1Exc}--\ref{2ndIt3Exc} hold
for right $\Uqsltwo$ and $\UqsltwoBar$-modules, after the symbolic replacements
\begin{align} \label{EmbProjLemReplaceExc}
\CCembedor \mapsto \CCembedorBar , \qquad
\CCunprojector \mapsto \CCunprojectorBar , \qquad
\CChatunprojector \mapsto \CChatunprojectorBar , \qquad 
\mathsf{S} \mapsto \overbarStraight{\mathsf{S}} ,  
\qquad \textnormal{and} \qquad
\VecSp \mapsto \VecSpBar .
\end{align}
\end{lem}

\begin{proof}
In item~\ref{2ndIt1Exc}, the $\Uqsltwo$-homomorphism property follows 
from the definitions of the maps in the assertion,
and the $\UqsltwoBar$-homomorphism property can be verified using this and identity~\eqref{BarToNoneLeft} from lemma~\ref{MergeLem}.
Items~\ref{2ndIt2Exc}--\ref{2ndIt3Exc} are immediate from the definitions of the maps in the assertion.
The assertions with replacements~\eqref{EmbProjLemReplaceExc} can be proven similarly.
\end{proof}

\begin{lem}  \label{TLprojLemNewExceptional} 
Suppose $q \in \{\pm \ii\}$ $($i.e., $\pmin(q) = 2)$. 
Then, for all integers $n \geq 2$ and $i , j \in \{1,2,\ldots,n-1\}$, we have
\begin{align} 
\label{GenProj2-1Exceptional}
\Trep_{n}^{n-2}(\Lgen_i) & = \left(\frac{q-q^{-1}}{\ii q^{1/2}}\right) 
\big(\id^{\otimes(i-1)} \otimes \CCembedor\super{0}\sub{1,1} \otimes \id^{\otimes(n-i-1)}\big) , \\
\label{GenProj2-0Exceptional} 
\Trep_{n-2}^n(\Rgen_j) & = \left(\frac{\ii q^{1/2}}{q-q^{-1}}\right) 
\big(\id^{\otimes(j-1)} \otimes \CChatunprojector\super{1,1}\sub{0} \otimes \id^{\otimes(n-j-1)}\big) ,
\end{align}
and similarly, 
\begin{align} 
\label{GenProj2-1ExceptionalBar}
\TrepBar_{n-2}^{n}(\Rgen_j) & =  \ii q^{1/2} (q-q^{-1})
\big(\id^{\otimes(j-1)} \otimes \CCembedorBar\super{0}\sub{1,1} \otimes \id^{\otimes(n-j-1)}\big) , \\
\label{GenProj2-0ExceptionalBar} 
\TrepBar_n^{n-2}(\Lgen_i) & = \left(\frac{1}{\ii q^{1/2} (q-q^{-1})}\right) 
\big(\id^{\otimes(i-1)} \otimes \CChatunprojectorBar\super{1,1}\sub{0} \otimes \id^{\otimes(n-i-1)}\big) .
\end{align}
\end{lem}
\begin{proof}
Assertions~(\ref{GenProj2-1Exceptional},~\ref{GenProj2-1ExceptionalBar}) are the same as~(\ref{GenProj2-1},~\ref{GenProj2-1Bar}) in lemma~\ref{TLprojLemNew}.
Assertions~(\ref{GenProj2-0Exceptional},~\ref{GenProj2-0ExceptionalBar}) can be proven via direct calculations, 
using~(\ref{ExtendThis2},~\ref{ExtendThisBar2},~\ref{UnProjectorHat},~\ref{UnProjectorHatBar}),
and the assumption that $q \in \{\pm \ii\}$, so $q^{-1} = -q$.
\end{proof}

\begin{cor} \label{RepCorExceptional} 
Suppose $q \in \{\pm \ii\}$ $($i.e., $\pmin(q) = 2)$. Then, 
$\Trep_n \colon \TL_n(0) \longrightarrow \End \VecSp_n$ and 
$\TrepBar_n \colon \TL_n(0) \longrightarrow \EndOp \VecSpBar_n$ 
are respectively left and right representations,
and for all $j \in \{1,2,\ldots,n-1\}$, we have
\begin{align}
\label{GenProjExceptional} 
\Trep_n(\Gen_j) = 
\big(\id^{\otimes(j-1)} \otimes \CCunprojector\sub{1,1}\superscr{(1,1); (0)} \otimes \id^{\otimes(n-j-1)}\big),
\end{align}
and similarly,
\begin{align}
\label{GenProjExceptionalBar} 
\TrepBar_n(\Gen_j) = 
\big(\id^{\otimes(j-1)} \otimes \CCunprojectorBar\sub{1,1}\superscr{(1,1); (0)} \otimes \id^{\otimes(n-j-1)}\big).
\end{align}
\end{cor}

\begin{proof}
This can be proven similarly as corollary~\ref{RepCor}, 
but using lemmas~\ref{EmbProjLem2Exc} and~\ref{TLprojLemNewExceptional}
instead of lemmas~\ref{EmbProjLem2} and~\ref{TLprojLemNew}. 
\end{proof}

\begin{remark}
In agreement with~\cite{psa}, the decomposition of $\VecSp_2$ as a $\TL_2(0)$-module is
\begin{align}
\CModule{\VecSp_2}{\TL} \isom \mathsf{P}_2\super{2} \oplus 2 \, \LS_2\super{2} ,
\end{align}
where $\smash{\LS_2\super{2}}$ is a simple one-dimensional $\TL_2(0)$-module
and $\smash{\mathsf{P}_2\super{2}}$ is the two-dimensional principal indecomposable $\TL_2(0)$-module.
As a $\TL_2(0)$-module, $\smash{\mathsf{P}_2\super{2}}$ is isomorphic to $\TL_2(0)$ itself, with left multiplication. 
Explicitly, we have
\begin{align}
\CModule{ \Span \{\MTbas\} }{\TL} \isom \LS_2\super{2} , \qquad\qquad
\CModule{\Span \{\mu\}}{\TL} \isom \LS_2\super{2} , \qquad \qquad
\CModule{\Span \{\HWvec,\zeta\}}{\TL} \isom \mathsf{P}_2\super{2} ,
\end{align}
and the map $\HWvec \mapsto \Gen_1$, $\zeta \mapsto \mathbf{1}_{\TL_2}$ gives the isomorphism from
$\CModule{\Span \{\HWvec,\zeta\}}{\TL}$ to $\CModule{\TL_2(0)}{\TL}$. 

A direct calculation shows that the commutant algebra $\EndMod{\TL} \VecSp_2$ is ten-dimensional, as in the generic case.
However, not all elements in it are obtained from the image of the $\Uqsltwo$-action on $\Module{\VecSp_2}{\Uqsltwo}$.
Namely, the elements $E^2$ and $F^2$ act as zero on $\VecSp_2$, and $K^2,K^{-2}$ act as the identity on $\VecSp_2$,
so $K$ and $K^{-1}$ act as the same element. 
Therefore, only the following images of the basis elements~\eqref{PBWBasis} of $\Uqsltwo$ 
give non-zero operators in $\EndMod{\TL} \VecSp_2$: 
\begin{align}
\{ 1, E, K , F , EK, EF, KF, EKF \} .
\end{align}
As was proven by P.~Martin~\cite{ppm} (see also~\cite{gv, psa}),
the remaining two elements in $\EndMod{\TL} \VecSp_2$ can be obtained by enlarging $\Uqsltwo$ by the additional generators
(also known as Lusztig's divided powers~\cite{Lus89})
\begin{align}
\lim_{q' \to q} \frac{E^2}{[2]_q} \qquad \text{and} \qquad \lim_{q' \to q} \frac{F^2}{[2]_q} ,
\end{align}
with $q = \pm \ii$ and the limit $q' \to q$ taken along a sequence not containing roots on unity. 
In this case, $\TL_2(0)$ still remains as the commutant algebra of this larger algebra on $\VecSp_2$.
\end{remark}

\section{Classical case: $q = 1$}  
\label{ClassicalApp}
Theorem~\ref{HighCSchurWeylThm2} is a classical version of 
a ``higher-spin Schur-Weyl duality," which can be thought of as
the ``$q \rightarrow 1$" limit of theorem~\ref{HighQSchurWeylThm2}.
However, 
such a limit is heuristic rather than literal; 
we cannot simply set $q = 1$ in the relations~\eqref{AlgRelations} that define the algebra $\Uqsltwo$,
for the commutator $[E,F]$ in~\eqref{AlgRelations} is not defined at $q=1$.  
(We discuss how to make sense of this limit in the end of this appendix.)
Instead, we regard $\mathsf{U}_1 = \Usltwo := U(\mathfrak{sl}_2)$ of as the universal enveloping algebra
of the Lie algebra $\mathfrak{sl}_2$, with generators $E$, $F$, and $H$ and relations
\begin{align} \label{URelationsClass} 
[H, E] = 2 E , \qquad 
[H, F] = -2 F , \qquad 
[E ,F ] = H .
\end{align} 
This algebra has the following coproduct and counit:
\begin{align} \label{ClassCP1} 
\Delta(E) & = E \otimes 1 + 1 \otimes E , \qquad
\Delta(F) = F \otimes 1 + 1 \otimes F , \qquad
\Delta(H) = H \otimes 1 + 1 \otimes H , \\
\label{ClassCU1} 
\epsilon(E)& = \epsilon(F)= \epsilon(H) = 0 .
\end{align}

The algebra $\Usltwo$ is semisimple, and its representation theory is analogous to that of $\Uqsltwo$ for $q \in \bC^\times$ not a root of unity. 
Fixing terminology, as in section~\ref{UqSect}, we say that a vector $v \in \mathsf{V} \setminus \{0\}$ in a $\Usltwo$-module $\mathsf{V}$ 
has \emph{weight} $\lambda \in \bC$, if we have $H.v = \lambda v$.
If in addition $v$ satisfies $E.v = 0$ and $K.v = \lambda v$,
then we call $v$ a \emph{highest-weight vector}, and we call 
the $\Usltwo$-module that $v$ generates a \emph{highest-weight module}. 
One can show from the definitions that any non-zero $\Usltwo$-module contains a highest-weight vector 
(see, e.g.,~\cite[proposition~\red{V.4.2}]{ck}).
As usual, we use the counit to define a $\Usltwo$-action on the ground field $\bC$ 
by $x.\lambda = \epsilon(x)\lambda$ for all $x \in \Usltwo$ and $\lambda \in \bC$, 
and we use the counit $\Delta$ to define tensor products of modules.
Finally, we have the following facts, analogous to~(\ref{HWVFact0}--\ref{HWVFact}):


\begin{enumerate}
\itemcolor{red}
\item \textnormal{\cite[lemma~\red{V.4.3}]{ck}:}
Let $v_0$ be a highest-weight vector of weight $\lambda \in \bC$ and $v_\ell := F^\ell. v_0$. Then, for all $\ell \in \bZnn$,
\begin{align} \label{ClassRep0}
H.v_\ell = (s - 2 \ell) v_\ell , \qquad
E . v_\ell = \ell (s-\ell+1) v_{\ell-1} , \qquad
F . v_\ell = v_{\ell+1} .
\end{align}

\item \textnormal{\cite[Theorem~\red{V.4.4}]{ck}}: 
Let $\mathsf{N}$ be a $\Usltwo$-module generated by a highest-weight vector $v_0$ of weight $\lambda \in \bC$, and let $v_\ell := F^\ell. v_0$. 
If $0 < \dim \mathsf{N} = s + 1$, then
\begin{enumerate}
\itemcolor{red}
\item[(a):] 
we have $\lambda = s$,

\item[(b):] 
we have $v_\ell = 0$ for $\ell > s$, and $\{v_0, v_1, \ldots, v_s\}$ is a basis for $\mathsf{N}$,

\item[(c):]
we have $H.v_\ell = (s - 2 \ell) v_\ell$ for all $\ell \in \{0, 1, \ldots, s\}$,

\item[(d):]
any other highest-weight vector of $\mathsf{N}$ is a scalar multiple of $v_0$, 

\item[(e):] 
$\mathsf{N}$ is simple, and 

\item[(f):] 
any finite-dimensional simple $\Usltwo$-module of dimension $s + 1$ is isomorphic to $\mathsf{N}$,
having basis $\{v_0, v_1, \ldots, v_s\}$ and $\Usltwo$-action~\eqref{ClassRep0}.
\end{enumerate}
\end{enumerate}

On the generic vector space 
$\VecSp\sub{s} := \Span \smash{\{ \Basis_0\super{s}, \Basis_1\super{s}, \ldots, \Basis_s\super{s} \}}$,
we define a left $\Usltwo$-module structure via the rules
\begin{align} \label{ClassRep}
F.\Basis_\ell\super{s} := 
\begin{cases}
\Basis_{\ell+1}\super{s}, & 0 \leq \ell \leq s - 1 , \\ 
0, & \ell = s ,
\end{cases}
\qquad 
E.\Basis_\ell\super{s} := 
\begin{cases}
 \ell (s-\ell+1) \Basis_{\ell-1}\super{s} , & 1 \leq \ell \leq s , \\ 
0, & \ell = 0 ,
\end{cases}
\qquad 
H.\Basis_\ell\super{s} := (s - 2 \ell) \Basis_\ell\super{s} ,
\end{align}
and we denote 
the corresponding representation by $\LeftRegRep\sub{s} \colon \Usltwo \longrightarrow \End{\VecSp\sub{s}}$ and
the resulting simple $\Usltwo$-module by 
\begin{align} \label{ClassRepNotation}
\Wd\sub{s} := \Module{\VecSp\sub{s}}{\Usltwo} .
\end{align}
We define a left $\Usltwo$-module structure on the tensor product $\VecSp_\multii$~\eqref{VecSpTensProd} 
by using the coproduct~\eqref{ClassCP1}, 
\begin{align} \label{ClassTensProd}
\Module{\VecSp_\multii}{\Usltwo}
:= \Wd\sub{\sIndex_1} \otimes \Wd\sub{\sIndex_2} \otimes \dotsm \otimes \Wd\sub{\sIndex_{\np_\multii}} ,
\end{align}
and we denote the corresponding representation by $\smash{\LeftRegRep_\multii \colon \Usltwo \longrightarrow \End{\VecSp_\multii}}$,
defined by~\eqref{LeftRegRep}.  
We also define a right $\Usltwo$-module structure on the generic vector space
$\VecSpBar\sub{s} := \Span \smash{\{ \BasisBar_0\super{s}, \BasisBar_1\super{s}, \ldots, \BasisBar_s\super{s} \}}$ 
via the rules
\begin{align} \label{ClassRepRight}
\BasisBar_\ell\super{s}  . E:= 
\begin{cases}
\BasisBar_{\ell+1}\super{s}, & 0 \leq \ell \leq s - 1 , \\ 
0, & \ell = s ,
\end{cases}
\qquad 
\BasisBar_\ell\super{s} . F := 
\begin{cases}
 \ell (s-\ell+1) \BasisBar_{\ell-1}\super{s} , & 1 \leq \ell \leq s , \\ 
0, & \ell = 0 ,
\end{cases}
\qquad 
\BasisBar_\ell\super{s} . H := (s - 2 \ell) \BasisBar_\ell\super{s} ,
\end{align}
and we define a right $\Usltwo$-module structure on the tensor product $\VecSp_\multii$~\eqref{VecSpTensProd} 
using the coproduct~\eqref{ClassCP1},  
\begin{align} \label{ClassTensProdRight}
\WdBar\sub{s} := \RModule{\VecSpBar\sub{s}}{\Usltwo} ,
\qquad \qquad 
\RModule{\VecSpBar_\multii}{\Usltwo}
:= \WdBar\sub{\sIndex_1} \otimes \WdBar\sub{\sIndex_2} \otimes \dotsm \otimes \WdBar\sub{\sIndex_{\np_\multii}}  .
\end{align}
We denote the corresponding representation by $\RightRegRep_\multii \colon \Usltwo \longrightarrow \EndOp{\VecSpBar_\multii}$.


We define the spaces of highest-weight vectors in these $\Usltwo$-type-one modules as in~\eqref{HWsp}.
These spaces are graded as in~(\ref{Kgraded},~\ref{sGrading}), 
with $H$-eigenvalues of the form $s$, for integers $s \in \smash{\DefectSet_\multii\superscr{\pm}}$:
\begin{align}
\HWsp_\multii\super{s} := \; & \HWsp_\multii \cap \Ksp _\multii\super{s}
= \big\{ v \in \VecSp_\multii \; \big| \, E.v = 0, \, H.v = s.v \big\}, \\
\HWspBar_\multii\super{s} := \; & \HWspBar_\multii \cap \KspBar _\multii\super{s}
= \big\{ \overbarStraight{v} \in \VecSpBar_\multii \, \big| \, \overbarStraight{v}.F = 0, \, \overbarStraight{v} . H = s. \overbarStraight{v} \big\} .
\end{align}

Notably, $\Usltwo$ is a semisimple algebra and item~\ref{DirectSumInclusionItem3} of 
proposition~\ref{MoreGenDecompAndEmbProp} holds for it without restriction on the magnitude of $\Summed_\multii$.
Indeed, we have the direct-sum decomposition of~\eqref{ClassTensProd} into simple left $\Usltwo$-modules as 
\begin{align} \label{ClassMoreGenDecomp} 
\Module{\VecSp_\multii}{\Usltwo}
\isom 
\bigoplus_{s \, \in \, \DefectSet_\multii} \Dim_\multii\super{s} \Wd\sub{s} .
\end{align}

%

\subsection{Classical higher-spin Schur-Weyl duality}

Next, we extend the results of sections~\ref{RepTheorySect}--\ref{HigherQSchurWeylSect} to the case $q = 1$. 
The key point why we can extend our previous results to the case $q = 1$
is the fact that, thanks to the identity
\begin{align} \label{SpecialLim} 
\lim_{q \rightarrow 1} \, [k] = \lim_{q \rightarrow 1} \frac{q^k-q^{-k}}{q-q^{-1}} = k  \qquad 
\text{for all $k \in \bZ$} . 
\end{align}
Jones-Wenzl projectors of all sizes exist at $q = 1$ 
and they are given by the recursion~\eqref{wjrecursion} with $[k] \mapsto k$ for all quantum integers $[k]$.
Therefore, the valenced Temperley-Lieb algebra $\TL_\multii(-2)$ is well-defined also with $q = 1$ and $\nu = -2$.
Furthermore, 
we have the direct-sum decomposition from proposition~\ref{HWspacePropEmbAndIso} with $\TL_\multii(-2)$:
\begin{align} \label{ClassWJVnDecomp2} 
\CModule{\VecSp_\multii}{\TL} \isom 
\bigoplus_{s \, \in \, \DefectSet_\multii} (s + 1) \LS_\multii\super{s} , 
\end{align}
without restriction on the magnitude of $\Summed_\multii$.
The (higher-spin) Schur-Weyl duality theorem~\ref{HighCSchurWeylThm2} relates these two decompositions to each other.

\begin{lem} \label{ReplaceLem} 
All lemmas, propositions, and corollaries of sections~\ref{RepTheorySect}--\ref{HigherQSchurWeylSect} 
hold for the algebras $\Usltwo$ and $\TL_\multii(-2)$
after making the following replacements:
\begin{enumerate}
\itemcolor{red}
\item We replace $q \in \bC^\times$ by $q = 1$ \textnormal{(}and $\nu = -q - q^{-1}$ by $\nu = -2$\textnormal{)}.

\item We replace $\pmin(q)$ by $\infty$. 

\item We replace the $q$-integer $[k]$ by the integer $k$.


\item We replace $\Uqsltwo = U_q(\mathfrak{sl}_2)$ by $\Usltwo = U(\mathfrak{sl}_2)$, and
we consider its modules 
$\RModule{\VecSpBar_\multii}{\Usltwo}$
and 
$\Module{\VecSp_\multii}{\Usltwo}$.

\item We replace $\UqsltwoBar$ by $\Usltwo$ and $\overbarStraight{x} \in \UqsltwoBar$ by $x \in \Usltwo$.


\item We omit all factors of the form $(q-q^{-1})^k$. 



\item We have $\smash{\rad \LS_\multii\super{s}} = \{0\}$ for all $s \in \DefectSet_\multii$ 
in~\eqref{LSRadical} and $\smash{\Quo_\multii\super{s} = \LS_\multii\super{s}}$ for all $s \in \DefectSet_\multii$ in~\eqref{QuoSp},
by~\textnormal{\cite[corollary~\red{5.2}]{fp3a}}. 
\end{enumerate}
\end{lem}

\begin{proof}
These replacements do not essentially affect 
the proofs of the lemmas, propositions, 
and corollaries of sections~\ref{RepTheorySect}--\ref{HigherQSchurWeylSect}, except for the ones using 
formula~\eqref{CoProd} for the coproduct for $\Uqsltwo$, different from that~\eqref{ClassCP1} for $\Usltwo$.
Let us summarize:
\begin{itemize}[leftmargin=*]
\item For lemmas~\ref{UqStarMapLemma}--\ref{MergeLem}, proved in appendix~\ref{PreliApp},
we replace the generator $K$ by the generator $H$ (and omit $K^{-1}$), 
relations~(\ref{AlgRelations},~\ref{CoProd},~\ref{CoUnit}) by relations~(\ref{URelationsClass},~\ref{ClassCP1}), 
and formulas~(\ref{CoproductFormulas},~\ref{CoproductFormulasBar}) in lemma~\ref{CoproductFormulasLem} by formula
\begin{align} \label{pFoldCorProd} 
\Delta\super{\np}(x) = \sum_{j=1}^\np 1^{\otimes(j-1)} \otimes x \otimes 1^{\otimes \np-j} 
\qquad \text{for all } x \in \mathfrak{sl}_2 .
\end{align} 
Lemma~\ref{HWform1GenLem} readily adapts to the case $q = 1$, 
replacing formula~\eqref{CoProd} for the coproduct for $\Uqsltwo$ by that~\eqref{ClassCP1} for $\Usltwo$.
The other results in section~\ref{RepTheorySect} readily adapt to the case $q = 1$.  

%




\item All lemmas, propositions, and corollaries of sections~\ref{DiagAlgSect},~\ref{GraphTLSect},~and~\ref{HigherQSchurWeylSect},
readily adapt to the case $q = 1$.  


\item All lemmas, propositions, and corollaries of section~\ref{GraphUQSect} readily adapt to the case $q = 1$.  
In particular, proposition~\ref{HWspLem2}, now with $q = 1$, says that
the map $\alpha \mapsto \Sing_\alpha$ 
is always an isomorphism from $\LS_\multii$ to $\HWsp_\multii$.

%
\end{itemize}
\end{proof}

\SecondHighCSchurWeylThm*

\noindent
\textit{The 
analogue of this theorem holds for the right representation $\RightRegRep_\multii$ of $\Usltwo$.
after the symbolic replacements
\begin{align} 
\Trep \mapsto \TrepBar , \qquad
\VecSp \mapsto \VecSpBar , \qquad
\LeftRegRep \mapsto \overbarcal{\LeftRegRep}, \qquad
\Wd \mapsto \WdBar , \qquad
\LS \mapsto \LSBar , 
\qquad \textnormal{and} \qquad 
F^\ell.\Sing_\alpha \mapsto \SingBar_{\alphaBar}.E^\ell.
\end{align}}

\begin{proof}
This follows from the results leading to theorem~\ref{HighQSchurWeylThm2} by using lemma~\ref{ReplaceLem}.
\end{proof}

%

\subsection{The limit $q \to 1$}

As in~\cite[chapter~\red{XVII}]{ck}, we can make sense of the ``$q \rightarrow 1$" limit of $\Uqsltwo$ by considering instead an appropriate topological algebra 
over the ring $\bC[[h]]$ of formal power series in an indeterminate $h$ with coefficients 
in $\bC$. This algebra $\mathsf{U}_h$ is generated by the unit and three generators $E$, $F$, and $H$, subject to the relations
\begin{align} \label{UhRelations} 
[H, E] = 2 E , \qquad 
[H, F] = -2 F , \qquad 
[E, F] = \frac{q^H - q^{-H}}{q - q^{-1}},  
\qquad\qquad \text{where $q := e^h$}
\end{align}
and where the expression $q^X = e^{h X}$ 
is defined by the limit 
$\smash{e^{h X} := \underset{N \to \infty}{\lim} \sum_{j \, = \, 0}^N \frac{X^j}{j!} h^j}$ 
in the $h$-adic topology. 
Via the embedding 
\begin{align} \label{Mapping} 
E \mapsto E q^{H / 2} , \qquad 
F \mapsto q^{-H / 2} F , \qquad 
K \mapsto q^{H / 2} , \qquad 
K^{-1} \mapsto q^{-H / 2}, 
\qquad  \qquad \text{with $q := e^h$},
\end{align} 
the quantum group $\Uqsltwo$ can be regarded as a sub-Hopf algebra of $\mathsf{U}_h$~\cite[Theorem~\red{XVII.4.1}]{ck}.  
In the $h \rightarrow 0$ (i.e., $q \rightarrow 1$) limit, we recover the universal enveloping algebra $\Usltwo := U(\mathfrak{sl}_2)$,
of the classical Lie algebra $\mathfrak{sl}_2$, with generators $E$, $F$, and $H$ and relations~\eqref{URelationsClass}. 

The Hopf algebra structure of $\Uqsltwo$ in the informal limit $q \rightarrow 1$ endows $\Usltwo$ with 
a Hopf algebra structure. 
Indeed, setting $q \mapsto 1$ in~\eqref{Mapping}, the coproduct~\eqref{CoProd} and counit~\eqref{CoUnit} become 
\begin{align} 
\label{ClassCP1Again} 
\Delta(E) & = E \otimes 1 + 1 \otimes E , \qquad
\Delta(F) = F \otimes 1 + 1 \otimes F , \qquad
\Delta(1) = 1 \otimes 1 , \\
\label{ClassCU1Again} 
\epsilon(E)& = \epsilon(F) = 0, \qquad \qquad \quad \, \epsilon(1) = 1 ,
\end{align} 
and for the generator $H = [E,F]$, we obtain
\begin{align} \label{ClassCP2Again}
\Delta(H) &= \Delta([E,F]) = [\Delta(E) , \Delta(F)]  
= H \otimes 1 + 1 \otimes H , \\
\label{ClassCU2Again} 
\epsilon(H) & = \epsilon([E,F]) = 0 .
\end{align}
This agrees with~(\ref{ClassCP1}--\ref{ClassCU1}).
Finally, taking the informal limit $q \rightarrow 1$ of the coproduct $\DeltaBar$~\eqref{CoProdBar} instead, 
we obtain the same structure~(\ref{ClassCP1Again},~\ref{ClassCP2Again}). 
Therefore, the Hopf algebra $\overbarStraight{\Usltwo}$ obtained by the above ``$q \rightarrow 1$'' limit of $\UqsltwoBar$ 
is the same as $\Usltwo$.  

\bigskip

\section{On radicals and orthocomplements}
\label{LinAlgAppAux}

The purpose of this appendix is to gather simple facts from linear algebra, used in the proof of
proposition~\ref{QuotientProp} in section~\ref{subsec: radical embedding}. 

Throughout, we let $\mathsf{V}$ and $\overbarStraight{\mathsf{V}}$ be two vector spaces of the same finite dimension $n \in \bZpos$.
We suppose that $\SPBiForm{\cdot}{\cdot} \colon \mathsf{V} \times \overbarStraight{\mathsf{V}} \longrightarrow \bC$
is a bilinear pairing of $\mathsf{V}$ and $\overbarStraight{\mathsf{V}}$.
We say that the bases 
$\mathsf{B} = \{e_1, e_2, \ldots, e_n\} \subset \mathsf{V}$ and 
$\overbarStraight{\mathsf{B}} = \{\overbarStraight{e}_1, \overbarStraight{e}_2, \ldots, \overbarStraight{e}_n\} \subset \overbarStraight{\mathsf{V}}$ 
are orthogonal if
\begin{align} \label{OrthogonalityDefn}
\SPBiForm{\overbarStraight{e}_i}{e_j} & = 0 \quad \textnormal{ for all }
i \neq j , \textnormal{ with } i, j \in \{1, 2, \ldots, n\} .
\end{align}

\begin{lem} \label{diagonalizationLem}
There exist orthogonal bases $\{e_1, e_2, \ldots, e_n\}$ and $\{\overbarStraight{e}_1, \overbarStraight{e}_2, \ldots, \overbarStraight{e}_n\}$ 
for $\mathsf{V}$ and $\overbarStraight{\mathsf{V}}$.
\end{lem}

\begin{proof}
We perform induction on the dimension $n = \dim \mathsf{V} = \dim \overbarStraight{\mathsf{V}}$.
The initial case $n=1$ is clear. 
Thus, we assume that $n \geq 2$ and that the assertion holds for all pairs of vector spaces of dimension $1 \leq m < n$.
We consider two cases:
\begin{enumerate}[leftmargin=*]
\itemcolor{red}
\item If there exists a pair $e \in \mathsf{V}$, $\overbarStraight{e} \in \overbarStraight{\mathsf{V}}$ of vectors such that 
$\SPBiForm{\overbarStraight{e}}{e} \neq 0$, then we can write
\begin{align}
\mathsf{V} & = \Span \{e\} \oplus \mathsf{V}' ,
\qquad \text{where} \qquad 
\mathsf{V}' := \big\{ v \in \mathsf{V} \, \big| \,  \SPBiForm{\overbarStraight{e}}{v}  = 0\big\} , \\
\overbarStraight{\mathsf{V}} & = \Span \{\overbarStraight{e}\} \oplus \overbarStraight{\mathsf{V}}' , 
\qquad \text{where} \qquad 
\overbarStraight{\mathsf{V}}' := \big\{ \overbarStraight{v} \in \overbarStraight{\mathsf{V}} \, \big| \,  \SPBiForm{\overbarStraight{v}}{e}  = 0\big\} .
\end{align}
By the induction hypothesis, $\mathsf{V}'$ and $\overbarStraight{\mathsf{V}}'$ have 
orthogonal bases $\{e_1, e_2, \ldots, e_{n-1}\}$ and $\{\overbarStraight{e}_1, \overbarStraight{e}_2, \ldots, \overbarStraight{e}_{n-1}\}$. 
Choosing $e_n = e$ and $\overbarStraight{e}_n = \overbarStraight{e}$, we obtain the asserted bases.

\item If $\SPBiForm{v}{\overbarStraight{w}} = 0$ for all $v \in \mathsf{V}$ and $\overbarStraight{w} \in \overbarStraight{\mathsf{V}}$,
then any bases $\{e_1, e_2, \ldots, e_n\}$ and $\{\overbarStraight{e}_1, \overbarStraight{e}_2, \ldots, \overbarStraight{e}_n\}$ for 
$\mathsf{V}$ and $\overbarStraight{\mathsf{V}}$ are orthogonal.
\end{enumerate}
This finishes the induction step. 
\end{proof}

For all pairs $\mathsf{W} \subset \mathsf{V}$ and $\overbarStraight{\mathsf{W}} \subset \overbarStraight{\mathsf{V}}$ of subspaces, 
we define their orthocomplements as
\begin{align}
\label{orthocomplement}
\mathsf{W}\superscr{\perp} 
& := \big\{ v \in \mathsf{V} \, \big| \,  \SPBiForm{\overbarStraight{w}}{v}  = 0  \textnormal{ for all } \overbarStraight{w} \in \overbarStraight{\mathsf{W}} \big\} 
\; \subset \; \mathsf{V} , \\
\label{orthocomplementBar}
\overbarStraight{\mathsf{W}}\superscr{\perp} 
& := \big\{ \overbarStraight{v} \in \overbarStraight{\mathsf{V}} \, \big| \,  \SPBiForm{\overbarStraight{v}}{w}  = 0  \textnormal{ for all } w \in \mathsf{W} \big\} 
\; \subset \; \overbarStraight{\mathsf{V}} ,
\end{align}
and their radicals as
\begin{align}
\label{radical}
\rad \mathsf{W} 
& := \big\{ w \in \mathsf{W} \, \big| \,  \SPBiForm{\overbarStraight{u}}{w}  = 0  \textnormal{ for all } \overbarStraight{u} \in \overbarStraight{\mathsf{W}} \big\} 
\; \subset \; \mathsf{W} , \\
\label{radicalBar}
\rad \overbarStraight{\mathsf{W}}
& := \big\{ \overbarStraight{w} \in \overbarStraight{\mathsf{W}} \, \big| \,  \SPBiForm{\overbarStraight{w}}{u}  = 0  \textnormal{ for all } u \in \mathsf{W} \big\} 
\; \subset \; \overbarStraight{\mathsf{W}} .
\end{align}
These definitions readily imply that
\begin{align} \label{IntersectionOfRads}
\rad \mathsf{W} = \mathsf{W} \cap \mathsf{W}\superscr{\perp}
\qquad\qquad \text{and} \qquad\qquad
\rad \overbarStraight{\mathsf{W}} = \overbarStraight{\mathsf{W}} \cap \overbarStraight{\mathsf{W}}\superscr{\perp}  .
\end{align}

\begin{lem} \label{trivialVecSpLemCombined}
Let $\mathsf{W} \subset \mathsf{V}$ and $\overbarStraight{\mathsf{W}} \subset \overbarStraight{\mathsf{V}}$ 
have orthogonal bases 
$\{e_1, e_2, \ldots, e_m\} \subset \mathsf{W} $ and 
$\{\overbarStraight{e}_1, \overbarStraight{e}_2, \ldots, \overbarStraight{e}_m\} \subset \overbarStraight{\mathsf{W}} $ 
such that 
\begin{align} 
\label{AllNonZero}
\SPBiForm{\overbarStraight{e}_j}{e_j} & \neq 0 
\qquad \textnormal{ for all } 1 \leq j \leq k , \\
\label{AllZero}
\SPBiForm{\overbarStraight{e}_j}{e_j} & = 0
\qquad  \textnormal{ for all } k+1 \leq j \leq m ,
\end{align}
for some $1 \leq k \leq m$.
Then, the following statements hold:
\begin{enumerate}
\itemcolor{red}

\item \label{trivialVecSpLemCombinedItem2}
The collections $\{e_{k+1}, e_{k+2}, \ldots, e_m\}$ and $\{\overbarStraight{e}_{k+1}, \overbarStraight{e}_{k+2}, \ldots, \overbarStraight{e}_m\}$
are respectively bases for $\rad \mathsf{W}$ and  $\rad \overbarStraight{\mathsf{W}}$.

\item \label{trivialVecSpLemCombinedItem3} 
If $k=m$, then we have 
$\smash{\mathsf{V} = \mathsf{W} \oplus \mathsf{W}\superscr{\perp} }$
and $\smash{\overbarStraight{\mathsf{V}} = \overbarStraight{\mathsf{W}} \oplus \overbarStraight{\mathsf{W}}\superscr{\perp}}$.
\end{enumerate}
\end{lem}

\begin{proof}
We prove items~\ref{trivialVecSpLemCombinedItem2}--\ref{trivialVecSpLemCombinedItem3} as follows:
\begin{enumerate}[leftmargin=*]
\itemcolor{red}

\item We have $\Span \{e_{k+1}, e_{k+2} , \ldots, e_m\} \subset \rad \mathsf{W} = \smash{\mathsf{W} \cap \mathsf{W}\superscr{\perp}}$  
by~(\ref{OrthogonalityDefn},~\ref{IntersectionOfRads},~\ref{AllZero}). Conversely, if $c_1, c_2, \ldots, c_m \in \bC$ are arbitrary and
\begin{align}
v = \sum_{j \, = \, 1}^m c_j e_j \in \mathsf{W} \cap \mathsf{W}\superscr{\perp} 
\qquad \qquad \text{and} \qquad \qquad
\overbarStraight{w} = \sum_{j\, = \,1}^m c_j \overbarStraight{e}_j \in \overbarStraight{\mathsf{W}} ,
\end{align}
then, we have $\SPBiForm{\overbarStraight{w}}{v} = 0$, which is equivalent to
\begin{align}
\sum_{j\, = \,1}^m c_j^2 \SPBiForm{\overbarStraight{e}_j }{e_j} = 0 
& \qquad \Longleftrightarrow \qquad 
c_j = 0  \textnormal{ for all } 1 \leq j  \leq k 
\qquad \Longleftrightarrow \qquad 
v \in \Span \{e_{k+1}, e_{k+2} , \ldots, e_m\} .
\end{align}
This shows that $\{e_{k+1}, e_{k+2} , \ldots, e_m\}$ is a basis for 
$\rad \mathsf{W}$.
The case of $\rad \overbarStraight{\mathsf{W}}$ 
can be proven similarly. 

\item If $k = m$, then item~\ref{trivialVecSpLemCombinedItem2} shows that 
$\smash{\mathsf{W} \cap \mathsf{W}\superscr{\perp}} = \{0\}$, 
so the sum $\smash{\mathsf{W} \oplus \mathsf{W}\superscr{\perp}}$ is direct.
An analogue of the dimension theorem shows that this inclusion is actually an equality:
$\smash{\mathsf{W} \oplus \mathsf{W}\superscr{\perp} = \mathsf{V}}$.
Similarly, we have
$\smash{\overbarStraight{\mathsf{W}} \oplus \overbarStraight{\mathsf{W}}\superscr{\perp} = \overbarStraight{\mathsf{V}}}$. 
\end{enumerate}
This concludes the proof. 
\end{proof}

\begin{lem}  \label{BigLinAlgLem}
Let $\mathsf{W} \subset \mathsf{V}$ and $\overbarStraight{\mathsf{W}} \subset \overbarStraight{\mathsf{V}}$ 
be subspaces of the same dimension, such that
\begin{align} \label{RadicalInclusion}
\rad \mathsf{W} \subset \rad \mathsf{V}
\qquad \qquad \textnormal{and} \qquad \qquad 
\rad \overbarStraight{\mathsf{W}} \subset \rad \overbarStraight{\mathsf{V}} .
\end{align}
Then, there exist subspaces $\mathsf{W}_1, \mathsf{W}_2 \subset \mathsf{V}$ 
such that the following hold:
\begin{enumerate}
\itemcolor{red}

\item \label{DirSumItem2}
$\mathsf{W} = \mathsf{W}_1 \oplus \rad \mathsf{W}$, 

\item \label{DirSumItem3}
$\smash{\mathsf{W}\superscr{\perp}} = \rad \mathsf{W} \oplus \mathsf{W}_2$,
and

\item \label{DirSumItem1}
$\mathsf{V} = \mathsf{W}_1 \oplus \rad \mathsf{W} \oplus \mathsf{W}_2$.

\end{enumerate}
Similarly, there exist subspaces $\overbarStraight{\mathsf{W}}_1, \overbarStraight{\mathsf{W}}_2 \subset \overbarStraight{\mathsf{V}}$ 
with properties~\ref{DirSumItem2}--\ref{DirSumItem1}, 
after the symbolic replacements $\mathsf{V} \mapsto \overbarStraight{\mathsf{V}}$ and 
$\mathsf{W} \mapsto \overbarStraight{\mathsf{W}}$.
\end{lem}

\begin{proof}
By lemma~\ref{diagonalizationLem} and item~\ref{trivialVecSpLemCombinedItem2} of lemma~\ref{trivialVecSpLemCombined}, we have the direct-sum decompositions
\begin{align} 
\label{DirSumOfUU}
\mathsf{W} = \mathsf{W}_1 \oplus \rad \mathsf{W}
\qquad\qquad \text{and} \qquad\qquad
\overbarStraight{\mathsf{W}} = \overbarStraight{\mathsf{W}}_1 \oplus \rad \overbarStraight{\mathsf{W}} ,
\end{align}
where $\mathsf{W}_1$ and $\overbarStraight{\mathsf{W}}_1$ respectively have orthogonal bases
$\{e_1, e_2, \ldots, e_k\}$ and $\{\overbarStraight{e}_1, \overbarStraight{e}_2, \ldots, \overbarStraight{e}_k\}$, and
\begin{align} \label{Wnonzero}
\SPBiForm{\overbarStraight{e}_j}{e_j} & \neq 0 \qquad \textnormal{ for all } 1 \leq j \leq k .
\end{align}
As such, we may apply item~\ref{trivialVecSpLemCombinedItem3} of lemma~\ref{trivialVecSpLemCombined}
to the spaces $\mathsf{W}_1$ and $\overbarStraight{\mathsf{W}}_1$ to arrive with the orthogonal decompositions
\begin{align} \label{TotalDirSum}
\mathsf{V} =  \mathsf{W}_1 \oplus \mathsf{W}_1\superscr{\perp}
\qquad\qquad \text{and} \qquad\qquad
\overbarStraight{\mathsf{V}} =  \overbarStraight{\mathsf{W}}_1 \oplus \overbarStraight{\mathsf{W}}_1\superscr{\perp} .
\end{align}
Also, by definition~\eqref{orthocomplement} (resp.~\eqref{orthocomplementBar})
and direct-sum decomposition~\eqref{DirSumOfUU}, 
the set $\rad \mathsf{W}$ (resp.~$\rad \overbarStraight{\mathsf{W}}$)
is a subspace of $\smash{\mathsf{W}_1\superscr{\perp}}$
(resp.~$\smash{\overbarStraight{\mathsf{W}}_1\superscr{\perp}}$), so
there exists a subspace $\mathsf{W}_2 \subset \smash{\mathsf{W}_1\superscr{\perp}}$ 
(resp.~$\overbarStraight{\mathsf{W}}_2 \subset \smash{\overbarStraight{\mathsf{W}}_1\superscr{\perp}}$) 
such that 
\begin{align} \label{RemainingClaim}
\mathsf{W}_1\superscr{\perp} = \rad \mathsf{W} \oplus \mathsf{W}_2 
\qquad\qquad \text{and} \qquad\qquad
\overbarStraight{\mathsf{W}}_1\superscr{\perp} = \rad \overbarStraight{\mathsf{W}} \oplus \overbarStraight{\mathsf{W}}_2 .
\end{align}
Hence, item~\ref{DirSumItem2} follows from~\eqref{DirSumOfUU}, 
item~\ref{DirSumItem3} from~\eqref{RemainingClaim},  
and item~\ref{DirSumItem1} then follows from~(\ref{TotalDirSum},~\ref{RemainingClaim}),
once we prove that 
\begin{align}  \label{WequalsW1}
\mathsf{W}\superscr{\perp} = \mathsf{W}_1\superscr{\perp} 
\qquad\qquad \text{and} \qquad\qquad
\overbarStraight{\mathsf{W}}\superscr{\perp} = \overbarStraight{\mathsf{W}}_1\superscr{\perp} .
\end{align}
By lemma~\ref{diagonalizationLem}, we find orthogonal 
bases $\{e_{k+1}, e_{k+2}, \ldots, e_n\}$ and $\{\overbarStraight{e}_{k+1}, \overbarStraight{e}_{k+2}, \ldots, \overbarStraight{e}_n\}$
for the orthocomplements $\smash{\mathsf{W}_1\superscr{\perp}}$ and  $\smash{\overbarStraight{\mathsf{W}}_1\superscr{\perp}}$
in~\eqref{TotalDirSum}.
In order to prove the left equality in~\eqref{WequalsW1}, 
we show that $\{e_{k+1}, e_{k+2}, \ldots, e_n\}$ is a basis for $\smash{\mathsf{W}\superscr{\perp}}$.

To begin, we prove that $\{e_{k+1}, e_{k+2}, \ldots, e_n\} \subset \smash{\mathsf{W}\superscr{\perp}}$.
Towards a contradiction, we take 
\begin{align} \label{contra}
& e_j \notin \smash{\mathsf{W}\superscr{\perp}} \qquad \text{for some } j \in \{k+1, k+2, \ldots,n\} 
\qquad\qquad \overset{\eqref{orthocomplement}}{\Longrightarrow} \qquad \qquad
\SPBiForm{\overbarStraight{w}}{e_j} \neq 0 \qquad \text{for some } \overbarStraight{w} \in \overbarStraight{\mathsf{W}} .
\end{align}
On the other hand,~\eqref{DirSumOfUU} shows that
$\overbarStraight{w} = \overbarStraight{w}_1 + \overbarStraight{w}_2$ for some 
$\overbarStraight{w}_1 \in \overbarStraight{\mathsf{W}}_1$ and 
$\overbarStraight{w}_2 \in \rad \overbarStraight{\mathsf{W}}$.
Because $e_j \in \mathsf{W}_1\superscr{\perp}$, we may assume that 
$\overbarStraight{w} \in \rad \overbarStraight{\mathsf{W}}$. 
But then, by assumption~\eqref{RadicalInclusion}, we have
\begin{align}
e_j \notin \smash{\mathsf{W}\superscr{\perp}}
\qquad\qquad \overset{\eqref{contra}}{\Longrightarrow} \qquad \qquad &
\SPBiForm{\overbarStraight{w}}{e_j} \neq 0
\qquad \text{for some} \quad 
\overbarStraight{w} \in 
\rad \overbarStraight{\mathsf{W}} 
\overset{\eqref{RadicalInclusion}}{\subset} \rad \overbarStraight{\mathsf{V}} ,
\end{align}
which contradicts definition~\eqref{radicalBar} of $\rad \overbarStraight{\mathsf{V}}$. 
Therefore, $\{e_{k+1}, e_{k+2}, \ldots, e_n\} \subset \smash{\mathsf{W}\superscr{\perp}}$
is a linearly independent subset.

To finish, we prove that $\smash{\mathsf{W}\superscr{\perp}} = \Span \{e_{k+1}, e_{k+2}, \ldots, e_n\}$.
Indeed, we have
\begin{align}
\begin{cases}
\mathsf{W}\superscr{\perp} \cap \mathsf{W}_1 
\overset{\eqref{DirSumOfUU}}{\subset}
\mathsf{W}\superscr{\perp} \cap \mathsf{W}
\overset{\eqref{IntersectionOfRads}}{=}
\rad \mathsf{W} , \\
\mathsf{W}\superscr{\perp} \cap \mathsf{W}_1 
\overset{\hphantom{\eqref{DirSumOfUU}}}{\subset} \mathsf{W}_1 
\end{cases}
\qquad\qquad \Longrightarrow \qquad\qquad
\mathsf{W}\superscr{\perp} \cap \mathsf{W}_1 
\subset \rad \mathsf{W} \cap \mathsf{W}_1 
\overset{\eqref{DirSumOfUU}}{=}
= \{0\} ,
\end{align}
and hence, because $\{e_{1}, e_{2}, \ldots, e_n\}$ is a basis for $\mathsf{V}$ by~(\ref{Wnonzero},~\ref{TotalDirSum}),
the collection $\{e_{k+1}, e_{k+2}, \ldots, e_n\}$ is a basis for $\smash{\mathsf{W}\superscr{\perp}}$.
This proves the left equality in~\eqref{WequalsW1}, and the right equality can be proven similarly.
\end{proof}

\section{Double-commutant properties}
\label{DCApp}
The purpose of this appendix is to show that the double-commutant property and duality discussed in sections~\ref{DCProofSec}--\ref{QSWProofSec}
follow from very  basic results 
on 
representations of associative algebras.  
Here, we only consider finite-dimensional left representations;
analogous results also hold for finite-dimensional right representations. 
Most of the material in this appendix is standard and can be found in one form of another in many textbooks and lecture notes.

Throughout, we consider an associative unital $\bC$-algebra $\alg$.
We recall that a (left) \emph{representation} of $\alg$ is a homomorphism 
$\rho \colon \alg \longrightarrow \End \mathsf{V}$ of algebras from $\alg$ to 
the algebra $\End \mathsf{V}$ of endomorphisms of some finite-dimensional vector space $\mathsf{V}$. 
We call $\Module{\mathsf{V}}{\alg}$ 
an $\alg$-\emph{module}, emphasizing the action in the notation.

We call $\Module{\mathsf{V}}{\alg}$ \emph{simple}, and $\rho$ \emph{irreducible}, 
if it is not zero and it 
contains no non-zero proper submodules.
One can show by induction 
that any finite-dimensional $\alg$-module has a \emph{composition series}, i.e., a strictly increasing sequence
\begin{align}
\Module{\{0\}}{\alg} = \Module{\mathsf{V}_0}{\alg} \; \subsetneq \; \Module{\mathsf{V}_1}{\alg} \; \subsetneq \; \Module{\mathsf{V}_2}{\alg} 
\; \subsetneq \; \cdots \; \subsetneq \; \Module{\mathsf{V}_m}{\alg} = \Module{\mathsf{V}}{\alg}
\end{align}
of submodules such that 
the quotient modules $\Module{(\mathsf{V}_j / \mathsf{V}_{j-1})}{\alg}$, called \emph{composition factors} of $\Module{\mathsf{V}}{\alg}$,
are simple for all $j=1,2,\ldots,m$. 
It also follows by induction 
that 
any two composition series of $\Module{\mathsf{V}}{\alg}$ have 
the same length $m$ and the same composition factors up to isomorphism and permutation.
This fact is known as the Jordan-H\"older theorem.

Next, we consider the module $\Module{\alg}{\alg}$ with action on itself given by multiplication.
The associated representation is called the \emph{regular representation} of $\alg$.
It follows immediately from the definitions that a subspace $\mathsf{J} \subset \alg$ gives an $\alg$-submodule if and only if it is a left ideal. 
Now, if $\Module{\mathsf{V}}{\alg}$ is a simple $\alg$-module, 
we claim that  it corresponds to a quotient module $\Module{(\alg / \mathsf{J})}{\alg}$ for some maximal left ideal $\mathsf{J}$. 
Indeed, to see this, we fix a vector $v \in \mathsf{V}$ and define a map
\begin{align}
\phi \colon \alg \longrightarrow \mathsf{V} , \qquad\qquad
\phi(a) := \rho(a)(v) ,
\end{align}
where $\rho \colon \alg \longrightarrow \End \mathsf{V}$ is the representation corresponding to $\Module{\mathsf{V}}{\alg}$.
This map $\phi$ is a non-zero homomorphism of $\alg$-modules, and since $\Module{\mathsf{V}}{\alg}$ is simple, we see that
\begin{align}
\Module{\mathsf{V}}{\alg} \isom \Module{\im \phi}{\alg} \isom \Module{(\alg / \ker \phi)}{\alg} ,
\end{align}
where $\ker \phi \subset \alg$ is the sought maximal left ideal. 
In particular, all simple $\alg$-modules appear as composition factors of $\Module{\alg}{\alg}$, 
so the Jordan-H\"older theorem shows that simple $\alg$-modules are classified by maximal left ideals of $\alg$.

We call $\Module{\mathsf{V}}{\alg}$ \emph{semisimple}, and $\rho$ \emph{completely reducible}, 
if $\Module{\mathsf{V}}{\alg}$ is isomorphic to a direct sum of simple $\alg$-modules.
We recall that the $\alg$-action on a direct sum module is diagonal:
\begin{align} \label{AVDirectDecompNonSimple}
\Module{\mathsf{V}}{\alg} = \bigoplus_{\lambda} \Module{\mathsf{V}_\lambda}{\alg} 
\qquad\qquad \Longrightarrow \qquad\qquad 
\rho = \bigoplus_{\lambda} \rho_\lambda ,
\end{align}
where $\rho_\lambda \colon \alg \longrightarrow \End \mathsf{V}_\lambda$ are the representations associated to the summands.

If some of the $\alg$-modules $\smash{\Module{ \mathsf{V}_\lambda}{\alg}}$ in~\eqref{AVDirectDecompNonSimple} are isomorphic,
it is convenient to index them by the same $\lambda$ and write $\smash{\Module{\mathsf{V}}{\alg}}$
as a direct sum of mutually non-isomorphic $\alg$-modules $\smash{\Module{ \mathsf{V}_\lambda}{\alg}}$
with multiplicities $m_\lambda \in \bZpos$: 
\begin{align} \label{AVDirectDecompNonSimpleWithMult}
\Module{\mathsf{V}}{\alg} \isom \bigoplus_{\lambda} \Module{ m_\lambda\mathsf{V}_\lambda}{\alg} ,
\qquad \qquad
\text{where } \qquad
\Module{ m_\lambda\mathsf{V}_\lambda}{\alg} 
:= \bigoplus_{j \, = \, 1}^{m_\lambda} \Module{ \mathsf{V}_\lambda}{\alg} ,
\end{align}
and where we have implicitly chosen some representatives of the non-isomorphic 
$\alg$-modules $\smash{\Module{ \mathsf{V}_\lambda}{\alg}}$.

Of particular interest is the decomposition of a semisimple module $\Module{\mathsf{V}}{\alg}$ into 
a direct sum of simple $\alg$-modules. 
We cautiously note that such a decomposition 
is by no means canonical if some submodule has multiplicity 
greater than one.
Indeed, each choice of basis for $\mathsf{V}$ in~\eqref{AVDirectDecompNonSimple}
and choices of representative modules in~\eqref{AVDirectDecompNonSimpleWithMult} 
gives another decomposition, all of which are isomorphic but not equal.

Another important general question 
is the classification of simple modules.
For endomorphism algebras, this is completely understood. 
If $\dim \mathsf{V} = n$, the elements of $\algB = \End \mathsf{V}$ can be viewed as ($n \times n$)-matrices with entries in $\bC$.
In particular, we know that $\algB$ has exactly one maximal left ideal,
and exactly one simple module, 
namely $\Module{\mathsf{V}}{\algB}$. 

\begin{lem} \label{MatrixAlgLem}
Suppose $\{\mathsf{S}_\lambda\}$ is a finite collection of finite-dimensional vector spaces, and
let $\algB = \smash{\underset{\lambda}{\bigoplus} \, \End \mathsf{S}_\lambda}$. 
Then, 
$\{\Module{\mathsf{S}_\lambda}{\algB}\}$ is the complete set of non-isomorphic simple $\algB$-modules, with $\algB$-action defined as
\begin{align}
\bigoplus_{\lambda'} L_{\lambda'} (v_\lambda) := L_\lambda (v_\lambda) 
\qquad \textnormal{for all elements } \bigoplus_{\lambda'} L_{\lambda'} \in \algB \textnormal{ and vectors } v_\lambda \in \mathsf{S}_\lambda .
\end{align}
\end{lem}

\begin{proof}
For any index $\lambda$ and vectors $v_\lambda, w_\lambda \in \mathsf{S}_\lambda \setminus \{0\}$,
there exists an element $L_\lambda \in \End \mathsf{S}_\lambda$ such that $L_\lambda(v_\lambda) = w_\lambda$.
Hence, the modules $\Module{\mathsf{S}_\lambda}{\algB}$ are simple, and they are non-isomorphic by construction. 
Because the simple $\algB$-modules are classified by maximal left ideals in 
$\algB = \smash{\underset{\lambda}{\bigoplus} \, \End \mathsf{S}_\lambda}$, we see that 
$\{\Module{\mathsf{S}_\lambda}{\algB}\}$ constitute all of the simple $\algB$-modules.
\end{proof}

The main aim of this appendix is to analyze the space $\HomMod{\alg} ( \mathsf{V} , \mathsf{W} )$ of $\alg$-homomorphisms
between two finite-dimensional modules $\Module{\mathsf{V}}{\alg}$ and $\Module{\mathsf{W}}{\alg}$.
Our presentation partially follows~\cite[chapter~\red{4}]{gw2}, with statements and proofs presented here tailored to the purposes of the present work.
We also recommend, e.g,~\cite{lam, etingof}.

All results in this appendix are consequences of complete reducibility and Schur's lemma,
which, being almost obvious, can be regarded as one of the cornerstones of the representation theory of associative algebras.

\begin{lem} \label{SchurLem}
\textnormal{(Schur's lemma):} 
If $\Module{\mathsf{V}}{\alg}$ and $\Module{\mathsf{W}}{\alg}$ are finite-dimensional simple $\alg$-modules,
then we have
\begin{align}
\dim \HomMod{\alg} ( \mathsf{V} , \mathsf{W} ) =
\begin{cases}
1 , & \textnormal{$\mathsf{V}$ and $\mathsf{W}$ are isomorphic,} \\
0 , & \textnormal{otherwise}.
\end{cases}
\end{align}
In particular, if $\mathsf{V}$ is a finite-dimensional simple $\alg$-module, then 
the identity map $\id_{\mathsf{V}}$ spans $\EndMod{\alg} \mathsf{V}$.
\end{lem}

\begin{proof}
See, e.g.,~\cite[lemma~\red{4.1.4}]{gw2}: the key argument is that both the kernel and the image of any 
homomorphism of $\alg$-modules is also an $\alg$-module (and for the latter statement, we consider the algebraically closed field $\bC$). 
\end{proof}

For any $\alg$-module $\Module{\mathsf{V}}{\alg}$, 
the endomorphism algebra $\mathsf{C} := \EndMod{\alg} \mathsf{V}$ 
is called the \emph{commutant} (i.e.,~\emph{centralizer}) of $\alg$ on $\mathsf{V}$.
We denote by $\CModule{\mathsf{V}}{\mathsf{C}}$ the natural $\mathsf{C}$-module structure on $\mathsf{V}$.
We can also define 
a representation $\varrho_{\mathsf{V}} \colon \alg \otimes \mathsf{C} \longrightarrow \End \mathsf{V}$,
\begin{align} \label{rAction}
\varrho_{\mathsf{V}}(a \otimes C) := \rho_{\mathsf{V}}(a) \circ C = C \circ \rho_{\mathsf{V}}(a) \quad 
\text{for all $a \in \alg$ and $C \in \mathsf{C}$,}
\end{align}
since the algebra $\mathsf{C}$ commutes with the action of $\alg$ on $\mathsf{V}$.
We denote the obtained $(\alg \otimes \mathsf{C})$-(bi)module by $\BIModule{\mathsf{V}}{\alg}{\mathsf{C}}$.
In theorem~\ref{DoubleMainTheorem}, we give a direct-sum decomposition of the $(\alg \otimes \mathsf{C})$-module,
when $\Module{\mathsf{V}}{\alg}$ is semisimple. 

Conversely, one might ask what is the commutant algebra of $\mathsf{C}$ on $\mathsf{V}$.
Clearly, we have $\rho(\alg) \subset \EndMod{\mathsf{C}} {\mathsf{V}}$,
and corollary~\ref{DoubleCommCor} below shows that if $\Module{\mathsf{V}}{\alg}$ is semisimple,
then the converse inclusion holds too.

Before proving these results, we note that
Schur's lemma gives a canonical way to decompose a semisimple module $\Module{\mathsf{V}}{\alg}$ into a direct sum of simple submodules,
thus improving the arbitrary decomposition~\eqref{AVAVDirectDecomp}.
The multiplicity 
\begin{align}
m_\lambda^{\mathsf{V}} = \dim \HomMod{\alg} ( \mathsf{S}_\lambda , \mathsf{V} ) \; \in \; \bZnn 
\end{align}
of simple submodules isomorphic to $\Module{\mathsf{S}_\lambda}{\alg}$ in $\Module{\mathsf{V}}{\alg}$ 
is governed by the space $\HomMod{\alg} ( \mathsf{S}_\lambda , \mathsf{V} )$, 
with $\Module{\mathsf{S}_\lambda}{\alg}$ representatives of non-isomorphic simple modules.
Indeed, using Schur's lemma~\ref{SchurLem}, it is straightforward to check that the map
\begin{align} \label{varphiMap}
& \varphi \colon \bigoplus_{\lambda} 
\big( \Module{\mathsf{S}_\lambda}{\alg} \otimes \HomMod{\alg} ( \mathsf{S}_\lambda , \mathsf{V} ) \big) 
\longrightarrow \Module{\mathsf{V}}{\alg} , \qquad\qquad
\varphi := \bigoplus_{\lambda}  
\varphi_\lambda , \\[1em]
\nonumber
\hspace*{-5mm}
\text{where} \qquad \qquad
& \varphi_\lambda \colon  \Module{\mathsf{S}_\lambda}{\alg} \otimes \HomMod{\alg} ( \mathsf{S}_\lambda , \mathsf{V} ) \longrightarrow \Module{\mathsf{V}}{\alg} ,
\qquad\qquad
\varphi_\lambda (v \otimes L) := L(v) ,
\end{align}
gives a natural isomorphism of $\alg$-modules.
See, e.g.,~\cite[proposition~\red{4.1.15}]{gw2} for details.

So far, the multiplicity spaces $\HomMod{\alg} ( \mathsf{S}_\lambda , \mathsf{V} )$ have no more structure than that of a vector space.
However, they do carry a natural action of the commutant algebra $\mathsf{C}$ via left multiplication,
and we actually have a similar direct-sum decomposition for $\CModule{\mathsf{V}}{\mathsf{C}}$
(cf. equation~\eqref{AVDirectDecomp0CanonicalCommutant}). 
Furthermore, 
the $\mathsf{C}$-modules 
$\CModule{\HomMod{\alg} ( \mathsf{S}_\lambda , \mathsf{V} )}{\mathsf{C}}$ are also simple, which shows that 
$\CModule{\mathsf{V}}{\mathsf{C}}$ is also semisimple. 
This observation is key to establish double-commutant theorem~\ref{DoubleMainTheorem}.

\subsection{Structure of commutants}

To begin, we identify 
the commutant space $\HomMod{\alg} ( \mathsf{V} , \mathsf{W} )$ 
between two finite-dimensional semisimple modules 
\begin{align} \label{AVAVDirectDecomp}
\Module{\mathsf{V}}{\alg} \isom \bigoplus_{\lambda \, \in \, \DefectSet_{\mathsf{V}}} \Module{ m_\lambda^{\mathsf{V}} \, \mathsf{S}_\lambda}{\alg} 
\qquad\qquad \textnormal{and} \qquad\qquad
\Module{\mathsf{W}}{\alg} \isom \bigoplus_{\lambda \, \in \, \DefectSet_{\mathsf{W}}} \Module{ m_\lambda^{\mathsf{W}} \, \mathsf{S}_\lambda}{\alg} , 
\end{align}
where $\smash{\Module{\mathsf{S}_\lambda}{\alg}}$ are non-isomorphic simple modules 
and $m_\lambda^{\mathsf{V}}, m_\lambda^{\mathsf{W}} \in \bZpos$ non-negative.
Only Schur's lemma~\ref{SchurLem} is needed. 

\begin{lem} \label{StructureOfCommutantLem} 
If 
$\Module{\mathsf{V}}{\alg}$ and $\Module{\mathsf{W}}{\alg}$ are semisimple as in~\eqref{AVAVDirectDecomp}, then
the following set is a basis for $\HomMod{\alg} ( \mathsf{V} , \mathsf{W} )$\textnormal{:}
\begin{align} \label{BasisForHomSpace}
\big\{ \almostId_{\mathsf{W}}\super{k , \lambda} \circ \pi_{\ell , \lambda}^{\mathsf{V}} \, \big| \, \lambda \in \DefectSet_{\mathsf{V}} \cap \DefectSet_{\mathsf{W}} , \, 
1 \leq k \leq m_\lambda^{\mathsf{W}} , \textnormal{ and } 1 \leq \ell \leq m_\lambda^{\mathsf{V}} \big\} ,
\end{align}
where 
\begin{align} \label{ProjectorSchur}
\pi_{\ell , \lambda}^{\mathsf{V}} (v) 
:=
\begin{cases}
( \almostId_{\mathsf{V}}\super{\ell , \lambda} )^{-1} (v)  , & v \in \almostId_{\mathsf{V}}\super{\ell , \lambda} ( \mathsf{S}_\lambda ) , \\
0 & \textnormal{otherwise}
\end{cases}
\qquad\qquad \textnormal{for all } v \in \mathsf{V} ,
\end{align}
and where $\smash{\{ \almostId_{\mathsf{V}}\super{\ell , \lambda} \, | \, 1 \leq \ell \leq m_\lambda^{\mathsf{V}}  \}}$ 
and $\smash{\{ \almostId_{\mathsf{W}}\super{k , \lambda} \, | \, 1 \leq k \leq m_\lambda^{\mathsf{W}} \}}$ 
are bases for $\HomMod{\alg} (\mathsf{S}_\lambda, \mathsf{V})$ and $\HomMod{\alg} (\mathsf{S}_\lambda, \mathsf{W})$, respectively. 
\end{lem}

\begin{proof}
To begin, we prove that for all $\lambda \in \DefectSet_{\mathsf{V}}$ and 
$\ell, \ell' \in \{1,2,\ldots, m_\lambda^{\mathsf{V}}\}$ such that $\ell' \neq \ell$, we have
\begin{align} \label{IntersectZero}
\almostId_{\mathsf{V}}\super{\ell , \lambda} ( \mathsf{S}_\lambda ) \cap \almostId_{\mathsf{V}}\super{\ell' , \lambda} ( \mathsf{S}_\lambda ) = \{0\} .
\end{align}
Indeed, 
Schur's lemma~\ref{SchurLem} implies that each map in the basis 
$\smash{\{ \almostId_{\mathsf{V}}\super{\ell , \lambda} \, | \, 1 \leq \ell \leq m_\lambda^{\mathsf{V}}  \}}$ 
is an isomorphism onto its image, which is a submodule of $\Module{\mathsf{V}}{\alg}$
isomorphic to $\Module{\mathsf{S}_\lambda}{\alg}$.
We form the sum module 
$\smash{\almostId_{\mathsf{V}}\super{\ell , \lambda}}(\Module{\mathsf{S}_\lambda}{\alg}) 
+ \smash{\almostId_{\mathsf{V}}\super{\ell' , \lambda}}(\Module{\mathsf{S}_\lambda}{\alg})$, 
and note that the intersection
$\smash{\almostId_{\mathsf{V}}\super{\ell , \lambda}}(\Module{\mathsf{S}_\lambda}{\alg}) 
\cap \smash{\almostId_{\mathsf{V}}\super{\ell' , \lambda}}(\Module{\mathsf{S}_\lambda}{\alg})$
is an $\alg$-submodule of either simple module
$\smash{\almostId_{\mathsf{V}}\super{\ell , \lambda}}(\Module{\mathsf{S}_\lambda}{\alg})$ 
and $\smash{\almostId_{\mathsf{V}}\super{\ell' , \lambda}}(\Module{\mathsf{S}_\lambda}{\alg})$.
Hence, because both of these submodule are simple, we have
\begin{align}
\text{either} \qquad
\almostId_{\mathsf{V}}\super{\ell , \lambda} ( \mathsf{S}_\lambda ) \cap \almostId_{\mathsf{V}}\super{\ell' , \lambda} ( \mathsf{S}_\lambda ) = \{0\} ,
\qquad \text{or} \qquad
\almostId_{\mathsf{V}}\super{\ell , \lambda} ( \mathsf{S}_\lambda ) \cap \almostId_{\mathsf{V}}\super{\ell' , \lambda} ( \mathsf{S}_\lambda ) 
= \almostId_{\mathsf{V}}\super{\ell , \lambda} ( \mathsf{S}_\lambda )
= \almostId_{\mathsf{V}}\super{\ell' , \lambda} ( \mathsf{S}_\lambda ) .
\end{align}
In the latter case, Schur's lemma again shows that the composed map 
$\smash{(\almostId_{\mathsf{V}}\super{\ell' , \lambda})^{-1} \circ \almostId_{\mathsf{V}}\super{\ell , \lambda}}$
is a multiple of the identity map $\id_{\mathsf{S}_\lambda}$, which implies that 
$\smash{\almostId_{\mathsf{V}}\super{\ell' , \lambda}}$ and $\smash{\almostId_{\mathsf{V}}\super{\ell , \lambda}}$
are proportional to each other, violating the fact that they are distinct basis elements in $\HomMod{\alg} (\mathsf{S}_\lambda, \mathsf{V})$.
We conclude that~\eqref{IntersectZero} indeed holds.

Now, any vector $v \in \mathsf{V}$ 
which belongs to a submodule of 
$\Module{\mathsf{V}}{\alg}$ isomorphic to $\Module{\mathsf{S}_\lambda}{\alg}$ lies in the image 
$\almostId ( \mathsf{S}_\lambda )$ for some isomorphism $\almostId \in \HomMod{\alg} ( \mathsf{S}_\lambda, \mathsf{V})$. 
Writing $v = \almostId(u)$ for some $u \in \mathsf{S}_\lambda$
and expanding $\almostId$ with some $\smash{b_{\ell}\super{\lambda}} \in \bC$, 
we have
\begin{align} 
\label{Expandv}
& \; v = 
\almostId(u) 
= \sum_{\ell \, = \, 1}^{m_\lambda^{\mathsf{V}} } b_{\ell}\super{\lambda} \, \almostId_{\mathsf{V}}\super{\ell , \lambda} (u) \\
\label{Expandv2}
\Longrightarrow \qquad
& \; \sum_{\lambda' \, \in \, \DefectSet_{\mathsf{V}}} \sum_{\ell' \, = \, 1}^{m_{\lambda'}^{\mathsf{V}} } 
\big( \almostId_{\mathsf{V}}\super{\ell' , \lambda'} \circ \pi_{\ell' , \lambda'}^{\mathsf{V}} \big) (v)
\underset{\eqref{Expandv}}{\overset{\textnormal{(\ref{ProjectorSchur}, \ref{IntersectZero})}}{=}}
\sum_{\ell \, = \, 1}^{m_\lambda^{\mathsf{V}} } b_{\ell}\super{\lambda} \, \almostId_{\mathsf{V}}\super{\ell , \lambda} (u) 
\overset{\eqref{Expandv}}{=} v .
\end{align}
Therefore, direct-sum decomposition~\eqref{AVAVDirectDecomp} combined with a similar analysis for $\mathsf{W}$ implies that we can write
\begin{align} \label{IdentitiesDecomposed}
\id_{\mathsf{V}} \underset{\eqref{Expandv2}}{\overset{\eqref{AVAVDirectDecomp}}{=}}
\sum_{\lambda \, \in \, \DefectSet_{\mathsf{V}}} \sum_{\ell \, = \, 1}^{m_\lambda^{\mathsf{V}} } \almostId_{\mathsf{V}}\super{\ell , \lambda}  \circ \pi_{\ell , \lambda}^{\mathsf{V}} 
\qquad\qquad \text{and} \qquad\qquad
\id_{\mathsf{W}} \underset{\eqref{Expandv2}}{\overset{\eqref{AVAVDirectDecomp}}{=}}
\sum_{\lambda' \, \in \, \DefectSet_{\mathsf{W}}} \sum_{k \, = \, 1}^{m_{\lambda'}^{\mathsf{W}} } \almostId_{\mathsf{W}}\super{k , \lambda'}  \circ \pi_{k , \lambda'}^{\mathsf{W}} .
\end{align}
Now, we fix an arbitrary map $L \in \HomMod{\alg} ( \mathsf{V} , \mathsf{W} )$ 
and consider the composed map
$\smash{\pi_{k , \lambda'}^{\mathsf{W}} \circ L \circ \almostId_{\mathsf{V}}\super{\ell , \lambda} } \in \HomMod{\alg} ( \mathsf{S}_\lambda, \mathsf{S}_{\lambda'} )$, 
with $\smash{\pi_{k , \lambda'}^{\mathsf{W}} := ( \almostId_{\mathsf{W}}\super{k , \lambda'} )^{-1}}$ defined as in~\eqref{ProjectorSchur}.
As $\Module{\mathsf{S}_\lambda}{\alg}$ are simple and non-isomorphic, 
Schur's lemma~\ref{SchurLem} shows that
\begin{align} \label{SchurImplies}
\pi_{k , \lambda'}^{\mathsf{W}} \circ L \circ \almostId_{\mathsf{V}}\super{\ell , \lambda} = 
\begin{cases}
c_{k, \ell}\super{\lambda}(L) \, \id_{\mathsf{S}_\lambda} , & \lambda' = \lambda , \\
0 , & \lambda' \neq \lambda ,
\end{cases}
\end{align}
for some constants $\smash{c_{k, \ell}\super{\lambda}(L)} \in \bC$.  Therefore, we obtain
\begin{align} \label{Lexpanded}
L \underset{\eqref{IdentitiesDecomposed}}{\overset{\eqref{SchurImplies}}{=}}
\sum_{\lambda \, \in \, \DefectSet_{\mathsf{V}} \cap \, \DefectSet_{\mathsf{W}}} 
c_{k, \ell}\super{\lambda}(L) \, \almostId_{\mathsf{W}}\super{k , \lambda} \circ \pi_{\ell , \lambda}^{\mathsf{V}} .
\end{align}
This proves that collection~\eqref{BasisForHomSpace} is a spanning set for $\HomMod{\alg} ( \mathsf{V} , \mathsf{W} )$.
Furthermore, we see from~(\ref{SchurImplies},~\ref{Lexpanded}) that 
\begin{align}
c_{k, \ell}\super{\lambda}(L) = 0 \quad \text{for all } \lambda \in \DefectSet_{\mathsf{V}} \cap \DefectSet_{\mathsf{W}} , \, 
1 \leq k \leq m_\lambda^{\mathsf{W}} , \textnormal{ and } 1 \leq \ell \leq m_\lambda^{\mathsf{V}} 
\qquad\qquad \Longleftrightarrow \qquad\qquad 
L = 0 ,
\end{align}
so the maps in collection~\eqref{BasisForHomSpace} are also linearly independent, thus forming a basis for $\HomMod{\alg} ( \mathsf{V} , \mathsf{W} )$.
\end{proof}

\begin{cor} \label{SemiSmipleCorComm1}
If 
$\Module{\mathsf{V}}{\alg}$ is semisimple as in~\eqref{AVAVDirectDecomp}, then we have 
\begin{align}  \label{CommutantAlgebra1}
\mathsf{C} := \EndMod{\alg} \mathsf{V}  \isom 
\bigoplus_{\lambda \, \in \, \DefectSet_{\mathsf{V}}} 
\End \big( \HomMod{\alg} ( \mathsf{S}_\lambda , \mathsf{V} ) \big) ,
\end{align}
and the collection
$\{\CModule{\HomMod{\alg} ( \mathsf{S}_\lambda , \mathsf{V} )}{\mathsf{C}} \,|\, \lambda \in \DefectSet_{\mathsf{V}}\}$  
is the complete set of non-isomorphic simple $\mathsf{C}$-modules. 
\end{cor}

\begin{proof}
Using basis~\eqref{BasisForHomSpace} for $\mathsf{C}$ from lemma~\ref{StructureOfCommutantLem}, 
we define a homomorphism
$\psi \colon \mathsf{C} \longrightarrow \smash{\underset{\lambda \, \in \, \DefectSet_{\mathsf{V}}}{\bigoplus}} \End ( \HomMod{\alg} ( \mathsf{S}_\lambda , \mathsf{V} ) )$
of algebras by homomorphic extension of the rule 
\begin{align} \label{PsiComposition}
\psi := \sum_{\lambda \, \in \, \DefectSet_{\mathsf{V}}} \psi_\lambda , \qquad \qquad
\psi_\lambda \big( \almostId_{\mathsf{V}}\super{k , \lambda'} \circ \pi_{\ell , \lambda'}^{\mathsf{V}} \big) (f)
:= \almostId_{\mathsf{V}}\super{k , \lambda'} \circ \pi_{\ell , \lambda'}^{\mathsf{V}} \circ f 
\; \in \; 
\begin{cases}
\HomMod{\alg} ( \mathsf{S}_\lambda , \mathsf{V} )  , & \lambda' = \lambda , \\
\{ 0 \}, & \lambda' \neq \lambda ,
\end{cases}
\end{align}
for each element $f \in \HomMod{\alg} ( \mathsf{S}_\lambda , \mathsf{V} )$
and indices $\lambda' \in \DefectSet_{\mathsf{V}}$ and $k, \ell \in \{1,2,\ldots, m_{\lambda'}^{\mathsf{V}}\}$.
Indeed, $\psi$ is a well-defined injective homomorphism of algebras, because the action~\eqref{PsiComposition} of it images on 
the basis elements $\smash{\almostId_{\mathsf{V}}\super{j , \lambda}} \in  \HomMod{\alg} ( \mathsf{S}_\lambda , \mathsf{V} )$ 
reads
\begin{align} 
& \; \psi \big( \almostId_{\mathsf{V}}\super{k , \lambda'} \circ \pi_{\ell , \lambda'}^{\mathsf{V}} \big) \big(  \almostId_{\mathsf{V}}\super{j , \lambda} \big)
\overset{\eqref{PsiComposition}}{=}
\almostId_{\mathsf{V}}\super{k , \lambda'} \circ \pi_{\ell , \lambda'}^{\mathsf{V}}  \circ \almostId_{\mathsf{V}}\super{j , \lambda} 
\underset{\eqref{IntersectZero}}{\overset{\eqref{ProjectorSchur}}{=}}
\delta_{\lambda, \lambda'} \delta_{\ell, j} \almostId_{\mathsf{V}}\super{k , \lambda} .
\end{align} 
Furthermore, $\psi$ is surjective, because the dimensions of its domain and codomain coincide.
Finally, lemma~\ref{MatrixAlgLem} then implies that the collection
$\{\CModule{\HomMod{\alg} ( \mathsf{S}_\lambda , \mathsf{V} )}{\mathsf{C}} \,|\, \lambda \in \DefectSet_{\mathsf{V}}\}$  
is the complete set of non-isomorphic simple $\mathsf{C}$-modules. 
\end{proof}

\subsection{General double-commutant property}

The next key result is a generalization of~\cite[theorem~\red{4.1.13}]{gw2}
(in the case of associative algebras), 
and a part of its proof is a straightforward adaptation of the proof given in~\cite{gw2}.
The latter is called the double-commutant theorem, 
or 
the double-commutant (double-centralizer) property.
We prefer to use the term ``double-commutant theorem" for the stronger statement 
in theorem~\ref{DoubleMainTheorem} below, also including duality of the algebra and its commutant. 

%

\begin{prop} \label{DoubleCommGenProp}  
Let $\rho_{\mathsf{V}} \colon \alg \longrightarrow \End \mathsf{V}$ 
and $\rho_{\mathsf{W}} \colon \alg \longrightarrow \End \mathsf{W}$
be finite-dimensional representations of $\alg$, and
denote the commutant by $\mathsf{C} := \HomMod{\alg} ( \mathsf{V} , \mathsf{W} )$.
Let $L \in \End \mathsf{V}$ and $R \in \End \mathsf{W}$.
Then, the following hold: 
\begin{enumerate}
\itemcolor{red}
\item \label{DoubleCommLemGenHomItem1}
Suppose the module $\Module{\mathsf{V}}{\alg}$ is semisimple.
Then, the diagram
\begin{equation} \label{GeneralCommute}
\begin{tikzcd}[column sep=2cm, row sep=1.5cm]
\mathsf{V} \arrow{r}{C} \arrow{d}[swap]{L}
& \arrow{d}{R} \mathsf{W} \\ 
\mathsf{V} \arrow{r}{C}
& \mathsf{W}
\end{tikzcd}
\end{equation}
commutes for all elements $C \in \mathsf{C}$
if and only if there exists an element $a \in \alg$ such that,
for all 
$C \in \mathsf{C}$, we have
\begin{align} \label{GeneralCommuteForm}
C \circ \rho_{\mathsf{V}}(a) = C \circ L = R \circ C = \rho_{\mathsf{W}}(a) \circ C .
\end{align}

\item \label{DoubleCommLemGenHomItem2}
Suppose both modules $\Module{\mathsf{V}}{\alg}$ and $\Module{\mathsf{W}}{\alg}$ 
are semisimple, with direct-sum decompositions as in~\eqref{AVAVDirectDecomp}. 
Then, diagram~\eqref{GeneralCommute}
commutes for all elements $C \in \mathsf{C}$
if and only if there exists an element $a \in \alg$ such that 
\begin{align} \label{GeneralCommuteForm2}
L \,\, = \,\, \rho_{\mathsf{V}}(a) \,\, + \sum_{\lambda \, \in \, \DefectSet_{\mathsf{V}} \, \setminus \, \DefectSet_{\mathsf{W}}} \pi_\lambda^{\mathsf{V}} \circ L'
\qquad\qquad \textnormal{and} \qquad\qquad
R \,\, = \,\, \rho_{\mathsf{W}}(a) \,\, + \sum_{\lambda \, \in \, \DefectSet_{\mathsf{W}} \, \setminus \, \DefectSet_{\mathsf{V}}} R' \circ \pi_\lambda^{\mathsf{W}}
\end{align}
for some 
endomorphisms $L' \in \End \mathsf{V}$ and $R' \in \End \mathsf{W}$,
where $\smash{\pi_\lambda^{\mathsf{V}}} \in \EndMod{\alg} \mathsf{V}$
and $\smash{\pi_\lambda^{\mathsf{W}}} \in \EndMod{\alg} \mathsf{W}$ 
denote the projections onto the $\lambda$:th summands in~\eqref{AVAVDirectDecomp}.
\end{enumerate}
\end{prop}

\begin{proof}
We prove items~\ref{DoubleCommLemGenHomItem1} and~\ref{DoubleCommLemGenHomItem2} as follows:
\begin{enumerate}[leftmargin=*] 
\itemcolor{red}

\item (See also~\cite[theorem~\red{4.1.13}]{gw2}.)
The ``if" part holds by definition of $\mathsf{C} := \HomMod{\alg} ( \mathsf{V} , \mathsf{W} )$. 
To prove the ``only if" part, we fix a basis
$\{ e_1, e_2, \ldots, e_n \}$ for $\mathsf{V}$, denoting $n = \dim \mathsf{V}$,
and consider the fixed vector
\begin{align} \label{SumOfAllVectors}
w := e_1 \oplus e_2 \oplus \cdots \oplus e_n  
\; \in \;  n \mathsf{V} = \mathsf{V}^{\oplus n} .
\end{align}
Let $\smash{\pi \in \EndMod{\alg} {\mathsf{V}^{\oplus n} }}$ be the projection 
onto the submodule 
$\Module{\{ \rho_{\mathsf{V}}^{\oplus n}(a)(w) \, | \, a \in \alg \}}{\alg}$ 
generated by $w$, where
$\rho_{\mathsf{V}}^{\oplus n}$ denotes the diagonal action~\eqref{AVDirectDecompNonSimple} on 
$n \mathsf{V} = \smash{\mathsf{V}^{\oplus n}}$.
Its domain being the direct sum $\mathsf{V}^{\oplus n}$, 
the map $\pi$ has a matrix representation with entries $\smash{\pi_{i,j}} \in \EndMod{\alg} \mathsf{V}$,
with $i,j \in \{1,2,\ldots,n\}$. Furthermore, because $\pi(w) = w$ by definition, 
we obtain
\begin{align} \label{ProjPreservesW}
e_j = \pi_{j,1}(e_1) + \pi_{j,2}(e_2) + \cdots + \pi_{j,n}(e_n) 
\qquad \text{for all } j \in \{1,2,\ldots,n\} .
\end{align}
Denoting 
$C^{\oplus n}(v_1 \oplus v_2 \oplus \cdots \oplus v_n) := C(v_1) \oplus C(v_2) \oplus \cdots \oplus C(v_n)$
for all 
$v_1, \ldots, v_n \in \mathsf{V}$ and 
$C \in \End \mathsf{V}$, we have
\begin{align}
\nonumber
(C^{\oplus n} \circ \pi \circ L^{\oplus n})(w) 
\overset{\eqref{SumOfAllVectors}}{=} \; &
(C^{\oplus n} \circ \pi \circ L^{\oplus n}) ( e_1 \oplus e_2 \oplus \cdots \oplus e_n )  \\
\nonumber
\overset{\hphantom{\eqref{ProjPreservesW}}}{=} \; &
\bigoplus_{j \, = \, 1}^n
\big( (C \circ \pi_{j,1} \circ L)(e_1) + (C \circ \pi_{j,2} \circ L)(e_2) + \cdots + (C \circ \pi_{j,n} \circ L)(e_n) \big) \\
\nonumber
\overset{\eqref{GeneralCommute}}{=} \; &
\bigoplus_{j \, = \, 1}^n (R \circ C) (\pi_{j,1}(e_1) + \pi_{j,2}(e_2) + \cdots + \pi_{j,n}(e_n)) \\
\label{LRCommuteProject}
\overset{\eqref{ProjPreservesW}}{=} \; &
(R \circ C)^{\oplus n} ( e_1 \oplus e_2 \oplus \cdots \oplus e_n ) 
\overset{\eqref{SumOfAllVectors}}{=} 
(R \circ C)^{\oplus n} (w) .
\end{align}
Using this observation and recalling that the range of the submodule projector $\smash{\pi}$ is
$\Module{\{ \rho_{\mathsf{V}}^{\oplus n}(a)(w) \, | \, a \in \alg \}}{\alg}$, we see that
there exists an element $a \in \alg$ such that 
\begin{align}
\nonumber
(R \circ C)^{\oplus n} (w)
\overset{\eqref{LRCommuteProject}}{=} \; & 
(C^{\oplus n} \circ \pi )( L^{\oplus n}(w) ) 
\overset{\eqref{SumOfAllVectors}}{=} 
(C^{\oplus n} \circ \rho_{\mathsf{V}}^{\oplus n}(a)) (w) \\
\label{GeneralCommuteForm0}
\overset{\eqref{SumOfAllVectors}}{\Longrightarrow} \qquad  \qquad 
(C \circ \rho_{\mathsf{V}}(a)) (e_j) \overset{\hphantom{\eqref{SumOfAllVectors}}}{=} \; & (R \circ C)(e_j) \quad \text{for all } j \in \{1,2,\ldots,n\} .
\end{align}
In conclusion, because $\{ e_1, e_2, \ldots, e_n \}$ is a basis for $\mathsf{V}$ and $C \in \mathsf{C}$ is arbitrary, 
we obtain 
asserted equalities~\eqref{GeneralCommuteForm}:
\begin{align} 
C \circ L 
\overset{\eqref{GeneralCommute}}{=} 
R \circ C 
\overset{\eqref{GeneralCommuteForm0}}{=} 
C \circ \rho_{\mathsf{V}}(a) = \rho_{\mathsf{W}}(a) \circ C , 
\end{align}
recalling in the last equality that $C \in \mathsf{C} := \HomMod{\alg} ( \mathsf{V} , \mathsf{W} )$
commutes with the $\alg$-action by definition.

\item 
For the ``if" part, Schur's lemma~\ref{SchurLem} shows that 
$\smash{\pi_\lambda^{\mathsf{W}}} \circ C$ and $C \circ \smash{\pi_{\lambda'}^{\mathsf{V}}}$ equal zero 
for all 
$\lambda \in \DefectSet_{\mathsf{W}} \setminus \DefectSet_{\mathsf{V}}$
and $\lambda' \in \DefectSet_{\mathsf{V}} \setminus \DefectSet_{\mathsf{W}}$, 
so we readily deduce that diagram~\eqref{GeneralCommute} commutes:
\begin{align} 
C \circ L 
 \,\, \overset{\eqref{GeneralCommuteForm2}}{=} \; &  \,\,
C \circ \rho_{\mathsf{V}}(a)  \,\, +  \,\, 
\sum_{\lambda \, \in \, \DefectSet_{\mathsf{V}} \, \setminus \, \DefectSet_{\mathsf{W}}} C \circ \pi_\lambda^{\mathsf{V}} \circ L'
 \,\, =  \,\, \rho_{\mathsf{W}}(a) \circ C  \,\, +  \,\, 
\sum_{\lambda \, \in \, \DefectSet_{\mathsf{V}} \, \setminus \, \DefectSet_{\mathsf{W}}} 
C \circ \pi_\lambda^{\mathsf{V}} \circ L' \\
 \,\, \overset{\eqref{GeneralCommuteForm2}}{=} \; &  \,\,
R \circ C 
 \,\, -  \,\, \sum_{\lambda \, \in \, \DefectSet_{\mathsf{W}} \, \setminus \, \DefectSet_{\mathsf{V}}} 
R' \circ \pi_\lambda^{\mathsf{W}} \circ C 
 \,\, +  \,\, \sum_{\lambda \, \in \, \DefectSet_{\mathsf{V}} \, \setminus \, \DefectSet_{\mathsf{W}}} 
C \circ \pi_\lambda^{\mathsf{V}} \circ L' 
 \,\, =  \,\, R \circ C ,
\end{align}
For the ``only if" part, it suffices to show that~\eqref{GeneralCommuteForm} in 
item~\ref{DoubleCommLemGenHomItem1}
implies~\eqref{GeneralCommuteForm2} with the same element $a \in \alg$. 
To this end, we first note that for all elements $C \in \mathsf{C}$, we have
\begin{align}  \label{Negligibles}
0 \overset{\eqref{GeneralCommuteForm}}{=} C \circ ( L - \rho_{\mathsf{V}}(a) )
\qquad \qquad \text{and} \qquad \qquad 
0 \overset{\eqref{GeneralCommuteForm}}{=} ( R - \rho_{\mathsf{W}}(a)) \circ C ,
\end{align}
which implies that the image of $L - \rho_{\mathsf{V}}(a)$ 
belongs to the kernel of every map $C \in \mathsf{C}$,
and the image of every map $C \in \mathsf{C}$
belongs to the kernel of $R - \rho_{\mathsf{W}}(a)$.
To solve for the possibilities for $L - \rho_{\mathsf{V}}(a) \in \End \mathsf{V}$ and $R - \rho_{\mathsf{W}}(a) \in \End \mathsf{W}$, 
we write arbitrary linear maps $L'' \in \End \mathsf{V}$ and $R'' \in \End \mathsf{W}$ 
via direct-sum decompositions~\eqref{AVAVDirectDecomp} in the form 
\begin{align}
L'' \overset{\eqref{AVAVDirectDecomp}}{=} 
\sum_{\lambda \, \in \, \DefectSet_{\mathsf{V}}} \pi_\lambda^{\mathsf{V}} \circ L'' 
\qquad \qquad \text{and} \qquad \qquad 
R'' \overset{\eqref{AVAVDirectDecomp}}{=} 
\sum_{\lambda \, \in \, \DefectSet_{\mathsf{W}}} R'' \circ \pi_\lambda^{\mathsf{W}} .
\end{align}
Recalling basis~\eqref{BasisForHomSpace} of $\mathsf{C}$ 
from lemma~\eqref{StructureOfCommutantLem}, consisting of homomorphisms of $\alg$-modules 
sending submodules of type $\Module{\mathsf{S}_\lambda}{\alg}$ in $\Module{\mathsf{V}}{\alg}$ 
isomorphically to submodules of $\Module{\mathsf{W}}{\alg}$ of the same type,
we see that
\begin{align}
\label{Lpp}
\im L'' \in \; & \ker C \hphantom{''} \qquad \text{for all } C \in \mathsf{C}
\qquad\qquad \Longrightarrow \qquad\qquad
\pi_\lambda^{\mathsf{V}} \circ L'' = 0 \qquad \text{for all } \lambda \in \DefectSet_{\mathsf{V}} \cap \DefectSet_{\mathsf{W}} , \\
\label{Rpp}
\im C \in \; & \ker R'' \qquad \text{for all } C \in \mathsf{C}
\qquad\qquad \Longrightarrow \qquad\qquad
R'' \circ \pi_\lambda^{\mathsf{W}} = 0 \qquad \text{for all } \lambda \in \DefectSet_{\mathsf{V}} \cap \DefectSet_{\mathsf{W}} .
\end{align}
By applying~(\ref{Lpp},~\ref{Rpp}) to $L'' = L - \rho_{\mathsf{V}}(a)$ and $R'' = R - \rho_{\mathsf{W}}(a)$, we conclude that 
\begin{align}
L - \rho_{\mathsf{V}}(a) 
\,\, \underset{\eqref{Lpp}}{\overset{\eqref{Negligibles}}{=}}  \,\,
\sum_{\lambda \, \in \, \DefectSet_{\mathsf{V}} \, \setminus \, \DefectSet_{\mathsf{W}}} \pi_\lambda^{\mathsf{V}} \circ L'
\qquad\qquad \text{and} \qquad\qquad
R - \rho_{\mathsf{W}}(a) 
\,\, \underset{\eqref{Rpp}}{\overset{\eqref{Negligibles}}{=}}  \,\,
\sum_{\lambda \, \in \, \DefectSet_{\mathsf{W}} \, \setminus \, \DefectSet_{\mathsf{V}}} R' \circ \pi_\lambda^{\mathsf{W}} ,
\end{align}
which proves~\eqref{GeneralCommuteForm2} with $L' = L - \rho_{\mathsf{V}}(a)$ and $R' = R - \rho_{\mathsf{W}}(a)$.
\end{enumerate}
This finishes the proof. 
\end{proof}

We recover the double-commutant property~\cite[theorem~\red{4.1.13}]{gw2} as a special case.

\begin{cor} \label{DoubleCommCor} 
\textnormal{\cite[theorem~\red{4.1.13}]{gw2}}
Suppose $\Module{\mathsf{V}}{\alg}$ is 
semisimple and $\mathsf{C} := \EndMod{\alg} \mathsf{V}$.
Then $\rho(\alg) = \EndMod{\mathsf{C}} {\mathsf{V}}$. 
\end{cor}

\begin{proof}
This follows from proposition~\ref{DoubleCommGenProp} by specializing to $\mathsf{W} = \mathsf{V}$.
\end{proof}

\begin{lem} \label{SemiSmipleLem1}
Suppose the module $\Module{\mathsf{V}}{\alg}$ is semisimple with direct-sum decomposition as in~\eqref{AVAVDirectDecomp}.
Then, we have
\begin{align} \label{Algebra1}
\rho(\alg) \isom \bigoplus_{\lambda \, \in \, \DefectSet_{\mathsf{V}}} \End \mathsf{S}_\lambda ,
\end{align}
and the collection $\{\mathsf{S}_\lambda \,|\, \lambda \in \DefectSet_{\mathsf{V}}\}$  
is the complete set of non-isomorphic simple $\rho(\alg)$-modules. 
\end{lem}

\begin{proof}
We show that the following representation gives rise to the sought isomorphism~\eqref{Algebra1} of algebras:
\begin{align} \label{LinearMapComplRed1}
\bigoplus_{\lambda \, \in \, \DefectSet_{\mathsf{V}}} \rho_\lambda 
\colon \alg \longrightarrow \bigoplus_{\lambda \, \in \, \DefectSet_{\mathsf{V}}} \End \mathsf{S}_\lambda .
\end{align}
First,~\eqref{LinearMapComplRed1} is injective by~\eqref{AVAVDirectDecomp}.
Second, to prove that~\eqref{LinearMapComplRed1} is surjective, we consider the commutant algebra 
\begin{align}
\mathsf{C} := \EndMod{\alg} \Big( \bigoplus_{\lambda \, \in \, \DefectSet_{\mathsf{V}}} \mathsf{S}_\lambda \Big) .
\end{align}
Because all summands $\Module{\mathsf{S}_\lambda}{\alg}$ are simple and non-isomorphic, 
Schur's lemma~\ref{SchurLem} shows that each element $C \in \mathsf{C}$ 
must preserve each summand $\mathsf{S}_\lambda$, acting as a scalar on it:
\begin{align} \label{SimpleActionOnSummands1}
C \in \mathsf{C} \qquad \Longrightarrow \qquad
\begin{cases}
C\big|_{\mathsf{S}_\lambda} \colon \mathsf{S}_\lambda \longrightarrow \mathsf{S}_\lambda , \\
C\big|_{\mathsf{S}_\lambda} = c_\lambda(C) \, \id_{\mathsf{S}_\lambda}
\end{cases}
\end{align}
for some constants $c_\lambda(C) \in \bC$ and for all $\lambda \in \DefectSet_{\mathsf{V}}$.
On the other hand,  using corollary~\ref{DoubleCommCor}, we obtain
\begin{align} 
\bigoplus_{\lambda \, \in \, \DefectSet_{\mathsf{V}}} \End \mathsf{S}_\lambda
\overset{\eqref{SimpleActionOnSummands1}}{\subset} 
\EndMod{\mathsf{C}} \Big( \bigoplus_{\lambda \, \in \, \DefectSet_{\mathsf{V}}} \mathsf{S}_\lambda \Big)
\overset{\textnormal{cor.~\ref{DoubleCommCor}}}{=} 
\bigoplus_{\lambda \, \in \, \DefectSet_{\mathsf{V}}} \rho_\lambda (\alg)
\overset{\eqref{LinearMapComplRed1}}{\subset} 
\bigoplus_{\lambda \, \in \, \DefectSet_{\mathsf{V}}} \End \mathsf{S}_\lambda .
\end{align}
Lastly, lemma~\ref{MatrixAlgLem} implies that $\{\mathsf{S}_\lambda \,|\, \lambda \in \DefectSet_{\mathsf{V}}\}$  
constitute the complete set of non-isomorphic simple $\rho(\alg)$-modules. 
\end{proof}

\subsection{Double-commutant theorem}

The \emph{Jacobson radical} 
of $\alg$ is the intersection of all of the maximal ideals in $\alg$.
We say that $\alg$ is \emph{semisimple} if its Jacobson radical is trivial.
Basic examples of semisimple algebras are direct sums of matrix algebras as in lemma~\ref{MatrixAlgLem}:

\begin{lem} \label{MatrixAlgIsSSLem}
If $\algB = \smash{\underset{\lambda}{\bigoplus} \, \End \mathsf{S}_\lambda}$, with
$\{\mathsf{S}_\lambda\}$ a finite collection of finite-dimensional vector spaces, then
$\algB$ is semisimple.
\end{lem}

\begin{proof}
See, e.g.,~\cite[example~\red{(7)},~page~\red{60}]{lam}. 
\end{proof}

In fact, the celebrated Wedderburn's structure theorem~\cite[theorem~\red{(3.5)}]{lam}
says that these are all of the finite-dimensional semisimple algebras. 
Also, there are numerous equivalent notions of semisimplicity. 
For instance, $\alg$ is semisimple if and only if all $\alg$-modules are 
semisimple~\cite[theorems~\red{(2.5)}~and~\red{(4.14)}]{lam}. 
We refer to, e.g.,~\cite[chapters~\red{1}--\red{2}]{lam} for more background 
and history. 
For our purposes, semisimplicity of $\alg$ is only an additional observation.

Without loss of generality, we state the main theorem~\ref{DoubleMainTheorem} only for 
a unital subalgebra of $\End \mathsf{V}$. 
In general, we can take it to be the image of a finite-dimensional representation of an associative unital algebra.

\begin{restatable}{theorem}{DoubleMainTheorem} \label{DoubleMainTheorem}
\textnormal{(Double-commutant theorem):}
Let $\algB \subset \End \mathsf{V}$ be a unital subalgebra of
the endomorphism algebra of a finite-dimensional vector space $\mathsf{V}$. 
If the module $\smash{\Module{\mathsf{V}}{\algB}}$ is semisimple, then the following hold: 
\begin{enumerate}
\itemcolor{red}
\item \label{Double1}
All three algebras $\mathsf{C} := \EndMod{\algB} \mathsf{V}$, 
$\algB = \EndMod{\mathsf{C}} {\mathsf{V}}$, and $\algB \otimes \mathsf{C}$ are semisimple.

\item \label{Double2} 
There exists a natural bijection between the collection 
$\{\Module{\mathsf{S}_\lambda}{\algB} \, | \, \lambda \in \DefectSet_{\mathsf{V}}\}$ of all non-isomorphic simple $\algB$-modules 
and the collection 
$\{\CModule{\mathsf{L}^\lambda}{\mathsf{C}} \, | \, \lambda \in \DefectSet_{\mathsf{V}}\}$ of all non-isomorphic simple $\mathsf{C}$-modules,
such that
\begin{align}
\Module{\mathsf{S}_\lambda}{\algB} \isom \Module{\HomMod{\mathsf{C}} ( \mathsf{L}^\lambda , \mathsf{V} )}{\algB}
\qquad\qquad \textnormal{and}\qquad\qquad 
\CModule{\mathsf{L}^\lambda}{\mathsf{C}} \isom \CModule{\HomMod{\algB} ( \mathsf{S}_\lambda , \mathsf{V} )}{\mathsf{C}}
\end{align}
and we have
\begin{align} \label{DoubleComm} 
\BIModule{\mathsf{V}}{\algB}{\mathsf{C}}
\isom 
\bigoplus_{\lambda \, \in \, \DefectSet_{\mathsf{V}}} \; \Module{\mathsf{S}_\lambda}{\algB} \otimes \CModule{\mathsf{L}^\lambda}{\mathsf{C}} .
\end{align} 
Furthermore, $\{ \Module{\mathsf{S}_\lambda}{\alg} \otimes \CModule{\mathsf{L}^{\lambda'}}{\mathsf{C}} \, | \, \lambda, \lambda' \in \DefectSet_{\mathsf{V}} \}$ 
is the complete set of non-isomorphic simple $( \algB \otimes \mathsf{C} )$-modules.
\end{enumerate}
\end{restatable}

\begin{proof}
$\mathsf{C}$ is semisimple by corollary~\ref{SemiSmipleCorComm1} and lemma~\ref{MatrixAlgIsSSLem}.
Corollary~\ref{DoubleCommCor} gives $\algB = \EndMod{\mathsf{C}} {\mathsf{V}}$,
which is semisimple by lemmas~\ref{SemiSmipleLem1} and~\ref{MatrixAlgIsSSLem}.
By lemma~\ref{MatrixAlgIsSSLem}, it then also follows that $\algB \otimes \mathsf{C}$ is semisimple, 
\begin{align}
\algB \otimes \mathsf{C} 
\underset{\eqref{Algebra1}}{\overset{\eqref{CommutantAlgebra1}}{\isom}}
\bigoplus_{\lambda, \, \lambda' } 
(\End \mathsf{S}_\lambda) \otimes (\End \mathsf{L}^{\lambda'})
= \bigoplus_{\lambda, \, \lambda'} 
\End (\mathsf{S}_\lambda \otimes \mathsf{L}^{\lambda'}) ,
\end{align} 
and lemma~\ref{MatrixAlgLem} then shows that 
$\{ \Module{\mathsf{S}_\lambda}{\algB} \otimes \CModule{\mathsf{L}^{\lambda'}}{\mathsf{C}} \, | \, \lambda, \lambda' \in \DefectSet_{\mathsf{V}} \}$ 
is the complete set of non-isomorphic simple $( \algB \otimes \mathsf{C} )$-modules.
Finally, 
we note that 
the map $\varphi$ in~\eqref{varphiMap} gives rise to a natural isomorphism of $\algB$-modules 
\begin{align} \label{AVDirectDecomp0Canonical}
\Module{\mathsf{V}}{\algB} \isom 
\bigoplus_{\lambda \, \in \, \DefectSet_{\mathsf{V}}} 
\big( \Module{\mathsf{S}_\lambda}{\algB} \otimes \HomMod{\algB} ( \mathsf{S}_\lambda , \mathsf{V} ) \big)  ,
\end{align}
and a natural isomorphism of $\mathsf{C}$-modules 
\begin{align} \label{AVDirectDecomp0CanonicalCommutant}
\CModule{\mathsf{V}}{\mathsf{C}} \isom 
\bigoplus_{\lambda \, \in \, \DefectSet_{\mathsf{V}}} 
\big( \mathsf{S}_\lambda \otimes \CModule{\HomMod{\algB} ( \mathsf{S}_\lambda , \mathsf{V} )}{\mathsf{C}}  \big) ,
\end{align}
Thus, direct-sum decomposition~\eqref{DoubleComm} 
follows by combining~(\ref{AVDirectDecomp0Canonical},~\ref{AVDirectDecomp0CanonicalCommutant}). 
This concludes the proof.
\end{proof}

\end{appendices}

\endgroup

\newcommand{\etalchar}[1]{$^{#1}$}

\renewcommand{\bibnumfmt}[1]{\makebox[5.3em][l]{[#1]}}

\end{document}